\documentclass{book}
\usepackage{amssymb}
\usepackage{amsfonts}
\usepackage{amsmath}
\usepackage{appendix}
\usepackage{float}
\usepackage{cite}

\newtheorem{theorem}{Theorem}[section]

\newtheorem{corollary}[theorem]{Corollary}

\newtheorem{definition}[theorem]{Definition}
\newtheorem{example}[theorem]{Example}

\newtheorem{lemma}[theorem]{Lemma}

\newtheorem{proposition}[theorem]{Proposition}

\newenvironment{proof}[1][Proof]{\noindent\textbf{#1.} }{\ \rule{0.5em}{0.5em}}
\begin{document}

\frontmatter
\title{Symmetries of Differential equations and Applications in
Relativistic Physics.}
\author{Andronikos Paliathanasis}
\date{Athens, 2014}
\maketitle

\chapter*
\newpage
\section*{Preface}
This thesis is part of the PhD program of the Department of Astronomy,
Astrophysics and Mechanics of the Faculty of Physics of the University of Athens, Greece.%
\clearpage
\thispagestyle{plain}

\vspace*{.35\textheight}{\centering In memory of my grandmother Amalia}

\newpage
\section*{Abstract}
In this thesis, we study the one parameter point transformations which leave invariant the differential equations.
In particular we study the Lie and the Noether point symmetries of second order differential equations.
We establish a new geometric method which relates the point symmetries of the differential equations
with the collineations of the underlying manifold where the motion occurs.
This geometric method is applied in order the two and three dimensional
Newtonian dynamical systems to be classified in relation to the point symmetries;
to generalize the Newtonian Kepler-Ermakov system in Riemannian spaces;
to study the symmetries between classical and quantum systems and to investigate the geometric
origin of the Type II hidden symmetries for the homogeneous heat equation and for the Laplace equation
in Riemannian spaces. At last but not least, we apply this geometric approach in order to determine
the dark energy models by use the Noether symmetries as a geometric criterion in modified theories of gravity.%
\newpage
\section*{Acknowledgement}

This doctoral thesis would not have been possible without the guidance and support
of numerous people and I wish to express my gratitude here.\\\\
Professor Peter G.L. Leach deserves a special mention. It was your research that inspired me to wade into the
symmetries of differential equations and your encouragement
when I first set about my own research proved invaluable.\\\\
I am grateful for my collaboration with my Supervisor, Professor Michael Tsamparlis.
Your faith in my abilities and your advice, on both my research and my career,
throughout the years, has allowed me to grow to the scientist I am today.
You introduced me to the subject of differential geometry and taught me the geometric thought process.\\\\
I would especially like to thank both my advisors Dr. Spyros Basilakos, who opened up to
me the subject of cosmology, and Dr. Christos Efthimiopoulos; the discussions we had on the
topic have been priceless. Also, I feel fortunate and grateful for my examination committee, Prof. S. Capozziello,
Prof. P. J. Ioannou, Prof T. Apostolatos and Prof A. H. Kara. Without your important comments
I could not have completed this thesis.\\ \\
Above all, I feel grateful for my friends and family. Stelios, Sifis, Venia and Thanos your
friendship and encouragement all these years have been my most valuable support mechanism.\\\\
Last but not least, I would not be who I am today without my family.
My parents Maria and Sotiris, and my grandfather Ioannis, you have always believed in and you
have supported me in all the important decisions of my life. Marianthi, your patience,
your kindness and your help on all levels, is more than any brother deserves.%

\tableofcontents
\listoftables%

\mainmatter%

\part{Introduction}

\chapter{Introduction \label{summary}}

\section{Summary}

In this thesis, we study the geometric properties of the Lie and the Noether
point symmetries of second order differential equations. In particular, we
find a connection between the point symmetries of some class of second order
differential equations with the collineations of the underlying manifold
where the "motion" occurs.

The novelty here is that we provide a geometrical method to determine the
symmetries of dynamical systems. The importance of Lie and Noether
symmetries is that they offer invariant functions which can be used to find
analytic solutions of the dynamical system.

The above mentioned geometric method is applied in order the two and three
dimensional Newtonian dynamical systems to be classified in relation to the
point symmetries; to generalize the Newtonian Kepler-Ermakov system in
Riemannian spaces; to study the symmetries between classical and quantum
systems and to investigate the geometric origin of the Type II hidden
symmetries for the homogeneous heat equation and for the Laplace equation in
Riemannian spaces. At last but not least, we apply this geometric approach
in order to determine the dark energy models by use the Noether symmetries
as a geometric criterion in modified theories of gravity.

The plan of the thesis is as follows.

\subsection{Summary of Part I:\ Introduction}

In Part I we give the basic properties and definitions of one parameter
point transformations.

\bigskip

In Chapter \ref{chapter1}, we study the geometry of the one parameter point
transformation, the properties of Lie algebras and the invariant functions.
Moreover, the Lie and the Noether symmetries of ordinary and partial
differential equations are analyzed and two schemes for using Lie symmetries
to construct solutions are presented. Furthermore, we study the action of
point transformation on linear differential geometry object.

\subsection{Summary of Part II: Symmetries of ODEs}

In Part II we study the geometric origin of the Lie and the Noether point
symmetries of second order ordinary differential equations.

\bigskip

In Chapter \ref{LieSymGECh}, we consider the set of autoparallels - not
necessarily affinely parameterized - of a symmetric connection. We find that
the major symmetry condition relates the Lie symmetries with the special
projective symmetries of the connection. We derive the Lie symmetry
conditions for a general system of second order ODE polynomial in the first
derivatives and we apply these conditions in the special case of geodesic
equations of Riemannian spaces. Furthermore we give the generic Lie symmetry
vector of the geodesic equations in terms of the special projective
collineations of the metric and their degenerates and the generic Noether
symmetry vector of the geodesic Lagrangian in terms of the homothetic
algebra of the Riemannian space. Finally we apply the results to various
cases and eventually we give the Lie symmetries, the Noether symmetries and
the associated conserved quantities of Einstein spaces, the G\"{o}del
spacetime, the Taub spacetime and the Friedman Robertson Walker spacetimes.

\bigskip

In Chapter \ref{chapter3}, we generalize the results of the previous chapter
in the case of the equations of motion of a particle moving in a Riemannian
space under the action of a general force $F^{i}$. We apply these results in
order to determine all two dimensional and and all three dimensional
Newtonian dynamical systems which admit Lie and\ Noether point symmetries.
We demonstrate the use of the results in two cases. The non-conservative
Kepler - Ermakov system and the case of the H\`{e}non Heiles type potentials.

\bigskip

In Chapter \ref{chapter4}, we generalize the two-dimensional autonomous
Hamiltonian--Kepler--Ermakov dynamical system to three dimensions using the $%
sl(2,R)$ invariance of Noether symmetries and determine all
three-dimensional autonomous Hamiltonian--Kepler--Ermakov dynamical systems
which are Liouville integrable via Noether symmetries. Subsequently, we
generalize the autonomous Kepler--Ermakov system in a Riemannian space which
admits a gradient homothetic vector by the requirements (a) that it admits a
first integral (the Riemannian Ermakov invariant) and (b) it has $sl(2,R)$
invariance. We consider both the non-Hamiltonian and the Hamiltonian
systems. In each case, we compute the Riemannian--Ermakov invariant and the
equations defining the dynamical system. We apply the results in general
relativity and determine the autonomous
Hamiltonian-Riemannian--Kepler--Ermakov system in the spatially flat
Friedman Robertson Walker spacetime. We consider a locally rotational
symmetric spacetime of class A and discuss two cosmological models. The
first cosmological model consists of a scalar field with an exponential
potential and a perfect fluid with a stiff equation of state. The second
cosmological model is the$f(R)$-modified gravity model of $\Lambda _{bc}$%
CDM. It is shown that in both applications the gravitational field equations
reduce to those of the generalized autonomous Riemannian--Kepler--Ermakov
dynamical system which is Liouville integrable via Noether integrals.

\subsection{Summary of Part III: Symmetries of PDEs}

In Part III we study the geometric origin of the Lie and the Noether point
symmetries of second order ordinary differential equations.

\bigskip

In Chapter \ref{chapter5}, we attempt to extend this correspondence of point
symmetries and collineations of the space to the case of second order
partial differential equations. We examine the PDE of the form $%
A^{ij}u_{ij}-F(x^{i},u,u_{i})=0$,where $u=u\left( x^{i}\right) $ and $u_{ij}$
stands for the second partial derivative. We find that if the coefficients $%
A_{ij}$ are independent of $u$ then the Lie point symmetries of the PDE form
a subgroup of the conformal symmetries of the metric defined by the
coefficients $A_{ij}$. We specialize the study to linear forms of $%
F(x^{i},u,u_{i})~$and write the Lie symmetry conditions for this case. We
apply this result to two cases. The Poisson/Yamabe equation for which we
derive the Lie symmetry vectors. Subsequently we consider the heat equation
with a flux in an n-dimensional Riemannian space and show that the Lie
symmetry algebra is a subalgebra of the homothetic algebra of the space. We
discuss this result in the case of de Sitter space time and in flat space.

\bigskip

In Chapter \ref{chapter6}, we determine the Lie point symmetries of the Schr%
\"{o}dinger and the Klein Gordon equations in a general Riemannian space. It
is shown that these symmetries are related with the homothetic and the
conformal algebra of the metric of the space respectively. We consider the
kinematic metric defined by the classical Lagrangian and show how the Lie
point symmetries of the Schr\"{o}dinger equation and the Klein Gordon
equation are related with the Noether point symmetries of this Lagrangian.
The general results are applied to two practical problems a. The
classification of all two and three dimensional potentials in a Euclidian
space for which the Schr\"{o}dinger equation and the Klein Gordon equation
admit Lie point symmetries and b. The application of Lie point symmetries of
the Klein Gordon equation in the exterior Schwarzschild spacetime and the
determination of the metric by means of conformally related Lagrangians.

\bigskip

In Chapter \ref{chapter7}, we study the geometric origin of Type II hidden
symmetries for the Laplace equation and for the homogeneous heat equation in
certain Riemannian spaces. As concerns the homogeneous heat equation, we
study the reduction of the heat equation in Riemannian spaces which admit a
gradient Killing vector, a gradient homothetic vector and in Petrov Type D,
N, II and Type III spacetimes. In each reduction we identify the source of
the Type II hidden symmetries. More specifically we find that (a) if we
reduce the heat equation by the symmetries generated by the gradient KV the
reduced equation is a linear heat equation in the nondecomposable space. (b)
If we reduce the heat equation via the symmetries generated by the gradient
HV the reduced equation is a Laplace equation for an appropriate metric. In
this case the Type II hidden symmetries are generated from the proper CKVs.
(c) In the Petrov space--times the reduction of the heat equation by the
symmetry generated from the nongradient HV gives PDEs which inherit the Lie
symmetries hence no Type II hidden symmetries appear. \ For the reduction of
the Laplace equation we consider Riemannian spaces which admit a gradient
Killing vector, a gradient Homothetic vector and a special Conformal Killing
vector. In each reduction we identify the source of Type II hidden
symmetries. We find that in general the Type II hidden symmetries of the
Laplace equation are directly related to the transition of the CKVs from the
space where the original equation is defined to the space where the reduced
equation resides. In particular we consider the reduction of the Laplace
equation (i.e., the wave equation) in the Minkowski space and obtain the
results of all previous studies in a straightforward manner. We consider the
reduction of Laplace equation in spaces which admit Lie point symmetries
generated from a non-gradient HV and a proper CKV and we show that the
reduction with these vectors does not produce Type II hidden symmetries. We
apply the results to general relativity and consider the reduction of
Laplace equation in locally rotational symmetric space times (LRS) and in
algebraically special vacuum solutions of Einstein's equations which admit a
homothetic algebra acting simply transitively. In each case we determine the
Type II hidden symmetries. We apply the general results to cases in which
the initial metric is specified.

\subsection{Summary of Part IV: Noether symmetries and theories of gravity}

In Part IV, we apply the Noether symmetry approach as a geometric criterion,
in order to probe the nature of dark energy in modified theories of gravity.

\bigskip

In Chapter \ref{chapter8}, we discuss the conformal equivalence of
Lagrangians for scalar fields in a Riemannian space of dimension $4$ and $n$
respectively.\ In particular we enunciate a theorem which proves that the
field equations for a non-minimally coupled scalar field are the same at the
conformal level with the field equations of the minimally coupled scalar
field. The necessity to preserve Einstein's equations in the context of
Friedmann Robertson Walker spacetime leads us to apply, the current general
analysis to the scalar field (quintessence or phantom) in spatially flat FRW
cosmologies. Furthermore, we apply the Noether symmetry approach in non
minimally coupled scalar field in a spatially flat FRW spacetime and by
using the Noether invariants we determine analytical solutions for the field
equations. Moreover we apply the same procedure for a minimally coupled
scalar field in a spatially flat FRW spacetime and in Biachi Class A
homogeneous spacetimes.

\bigskip

In Chapter \ref{chapter9}, a detailed study of the modified gravity, $f(R)~$%
models is performed, using that the Noether point symmetries of these models
are geometric symmetries of the mini superspace of the theory. It is shown
that the requirement that the field equations admit Noether point symmetries
selects definite models in a self-consistent way. As an application in
Cosmology we consider the Friedman -Robertson-Walker spacetime and show that
the only cosmological model which is integrable via Noether point symmetries
is the $\Lambda _{bc}$CDM model, which generalizes the Lambda Cosmology.
Furthermore using the corresponding Noether integrals we compute the
analytic form of the main cosmological functions.

\bigskip

In Chapter \ref{chapter10}, we apply the Noether symmetry approach in the $%
f\left( T\right) $ modified theory of gravity in a spatially flat FRW
spacetime and in static spherically symmetric spacetime. \ First, we present
a full set of Noether symmetries for some minisuperspace models and we find
that only the $f\left( T\right) =T^{n}$ model admits extra Noether
symmetries. The existence of extra Noether integrals can be used in order to
simplify the system of differential equations as well as to determine the
integrability of the model. Then, we compute analytical solutions and find
that spherically symmetric solutions in $f(T)$ gravity can be recast in
terms of Schwarzschild-like solutions modified by a distortion function
depending on a characteristic radius.

\bigskip

Finally, in Chapter \ref{Discussion} we discuss our results.

\chapter{Point transformations and Invariant functions \label{chapter1}}

\section{Introduction}

Lie symmetry of a differential equation is a one parameter point
transformation which leaves the differential equation invariant. Lie
symmetries\footnote{%
In the following sections, by Lie symmetry we mean point symmetry. There are
also generalized Lie symmetries which are not point symmetries.} is the main
tool to study nonlinear differential equations. Indeed Lie symmetries
provide invariant functions which can be used to construct analytic
solutions for a differential equation. These solutions we call invariant
solutions. One such example is the solution $u\left( x,y\right) =e^{k\left(
t-x\right) }~$of the wave equation%
\begin{equation*}
u_{xx}-u_{tt}=0
\end{equation*}%
which is found by applying the Lie symmetry $X=\partial _{x}+ku\partial
_{u}. $

The structure of the chapter is as follows. In section \ref{PTran}, we study
the geometry of the one parameter point transformation, the properties of
Lie algebras and the invariant functions. Invariant functions are functions
which remain unchanged under the action of a point transformation. In
section \ref{LieSym}, the Lie symmetries of ordinary and partial
differential equations are analyzed and two schemes for using Lie symmetries
to construct solutions are presented. In section \ref{NoetherS}, Noether
symmetries, a special class of Lie symmetries, are discussed. Noether
symmetries are admited only by systems whose equation of motion result from
a variational principle. Noether symmetries are important because they
produce conservation laws and can be used to simplify the differential
equations.

In section \ref{geomObj}, the action of point transformation on linear
differential geometry object is examined. Collineations are point
transformations which do not leave necessary invariant a geometric object.
In particular we study the collineations of the metric (Conformal motions)
and of the Christoffel symbols (Projective collineations) of a Riemannian
space.

\section{Point Transformations}

\label{PTran}

Let $M~$be a manifold of class$~~C^{p}~$with$~p\succeq 2$ and let $U$ be a
neighborhood in $M.$ Consider two points $P,Q\in U~$with coordinates $\left(
x_{P},y_{P}\right) $ and \ $\left( x_{Q}^{\prime },y_{Q}^{\prime }\right) $
respectively. A point transformation on $U$ is a relation among the
coordinates of the points $P,Q\in U$ which is defined by the transformation
equations%
\begin{equation*}
x_{Q}^{\prime }=x^{\prime }\left( x_{P},y_{P}\right) ~~,~~y_{Q}^{\prime
}=y^{\prime }\left( x_{P},y_{P}\right)
\end{equation*}%
where the functions $x^{\prime }\left( x,y\right) ,~y^{\prime }\left(
x,y\right) ~$are $C^{p-1}$ and
\begin{equation}
\det \left\vert \frac{\partial \left( x^{\prime },y^{\prime }\right) }{%
\partial \left( x,y\right) }\right\vert \neq 0.  \label{PT0}
\end{equation}%
Condition (\ref{PT0} ) means that the functions $x^{\prime }\left(
x,y\right) ,~y^{\prime }\left( x,y\right) $ are independent. A special class
of point transformations are the one parameter point transformations defined
as follows \cite{StephaniB}.

\begin{definition}
\label{1ppt}The one parameter point transformations are point
transformations that depend on one arbitrary parameter as follows%
\begin{equation}
x^{\prime }=y^{\prime }\left( x,y,\varepsilon \right) ~,~y^{\prime
}=y^{\prime }\left( x,y,\varepsilon \right)  \label{PT1}
\end{equation}%
where $\varepsilon \in
\mathbb{R}
$ and the transformation satisfies the following conditions.

a) They are well defined, that is, that if~%
\begin{equation*}
x^{\prime }\left( x_{1},y_{1},\varepsilon \right) =x^{\prime }\left(
x_{2},y_{2},\varepsilon \right) ,~y^{\prime }\left( x_{1},y_{1},\varepsilon
\right) =y^{\prime }\left( x_{2},y_{2},\varepsilon \right)
\end{equation*}%
then $x_{1}=y_{1}$ and $y_{1}=y_{2}.$

b) They can be composed, that is, that if~%
\begin{equation*}
x^{\prime }=y^{\prime }\left( x,y,\varepsilon \right) ~~,~~y^{\prime
}=y^{\prime }\left( x,y,\varepsilon \right)
\end{equation*}%
and~%
\begin{equation*}
x^{\prime \prime }=x^{\prime \prime }\left( x^{\prime },y^{\prime
},\varepsilon ^{\prime }\right) ~,~~y^{\prime \prime }=y^{\prime \prime
}\left( x^{\prime },y^{\prime },\varepsilon ^{\prime }\right)
\end{equation*}%
are two successive one parameter point transformations, there is a one
parameter point transformation parametrized by the real parameter $%
\varepsilon ^{\prime \prime }=\left( \varepsilon ,\varepsilon ^{\prime
}\right) $ so that~$x^{\prime \prime }=x^{\prime \prime }\left(
x,y,\varepsilon ^{\prime \prime }\right) ,~~y^{\prime \prime }=y^{\prime
\prime }\left( x,y,\varepsilon ^{\prime \prime }\right) .$

c) They are invertible, that is, for each one parameter point transformation%
\begin{equation*}
x^{\prime }=y^{\prime }\left( x,y,\varepsilon \right) ~~,~~y^{\prime
}=y^{\prime }\left( x,y,\varepsilon \right)
\end{equation*}%
there exists the inverse transformation%
\begin{equation*}
x^{\prime \prime }\left( x^{\prime },y^{\prime },\varepsilon _{inv}\right)
=x~~,~~y^{\prime \prime }\left( x^{\prime },y^{\prime },\varepsilon
_{inv}\right) =y
\end{equation*}

d)There is the identity transformation defined by the value\footnote{%
Without abandoning generality, we can take $\varepsilon _{0}=0$.} $%
\varepsilon =\varepsilon _{0},~$that is~%
\begin{equation*}
x^{\prime }\left( x,y,\varepsilon _{0}\right) =x~~,~~y^{\prime }\left(
x,y,\varepsilon _{0}\right) =y.
\end{equation*}
\end{definition}

From the above it follows that the one parameter point transformations form
a group. A group of one parameter point transformation defines a family of
curves in $M,$ which are parametrized by the parameter $\varepsilon $ and
are called the \textbf{orbits of the group of transformations}. These curves
may be viewed as the integral curves of a differentiable vector field $X\in
M.$

\subsection{Infinitesimal Transformations}

Let $\bar{x}\left( x,y,\varepsilon \right) ,~\bar{y}\left( x,y,\varepsilon
\right) $ be the parametric equations of a group orbit through the point $%
P\left( x,y,0\right) $. The tangent vector at the point $P=P\left(
x,y,0\right) $ is given by%
\begin{equation*}
X_{P}=\frac{\partial \bar{x}}{\partial \varepsilon }|_{\varepsilon
\rightarrow 0}\partial _{x}|_{P}+\frac{\partial \bar{y}}{\partial
\varepsilon }|_{\varepsilon \rightarrow 0}\partial _{y}|_{P}.
\end{equation*}%
The vector $X_{P}$ defines near the point $P\left( x,y,0\right) $ a point
transformation%
\begin{equation}
\bar{x}=x+\varepsilon \xi _{P}~~,~~\bar{y}=y+\varepsilon \eta _{P}
\label{PT.02}
\end{equation}%
where we have set%
\begin{equation}
\xi _{P}=\frac{\partial \bar{x}}{\partial \varepsilon }|_{\varepsilon
\rightarrow 0}~~,\text{~~}\eta =\frac{\partial \bar{y}}{\partial \varepsilon
}|_{\varepsilon \rightarrow 0}.
\end{equation}%
The point transformation (\ref{PT.02}) is a one parameter a point
transformation which is called an \textbf{infinitesimal point transformation}%
. The vector field $X_{P}$ is called the \textbf{generator} of the
infinitesimal transformation (\ref{PT.02}) along the orbit through the point
$P$. Evidently, the infinitesimal transformation \textit{moves} a point
along the orbit of the group through that point.

\begin{example}
Compute the generator of the infinitesimal transformation for the one
parameter point transformation%
\begin{eqnarray*}
\bar{x} &=&x\cos \varepsilon -y\sin \varepsilon \\
\bar{y} &=&x\sin \varepsilon +y\cos \varepsilon .
\end{eqnarray*}

Solution: We have
\begin{eqnarray*}
\xi \left( x,y\right) &=&\frac{\partial \bar{x}}{\partial \varepsilon }%
|_{\varepsilon \rightarrow 0}=-\left( x\sin \varepsilon +y\cos \varepsilon
\right) |_{\varepsilon \rightarrow 0}=-y \\
\eta \left( x,y\right) &=&\frac{\partial \bar{y}}{\partial \varepsilon }%
|_{\varepsilon \rightarrow 0}=\left( x\cos \varepsilon -y\sin \varepsilon
\right) |_{\varepsilon \rightarrow 0}=x
\end{eqnarray*}%
from which follows that the generator of the infinitesimal transformation is~%
$X=-y\partial _{x}+x\partial _{y}.$
\end{example}

It has been showed that a one parameter point transformation fixes an
infinitesimal generator up to a constant depending on the parametrization of
the group orbit. In the following section, it will be shown that the
converse holds true, that is,\ for an infinitesimal generator there always
exists a unique one parameter point transformation.

\subsubsection{Integral curves}

Consider a differentiable vector field $X\in M$ given by$~X=X^{i}\partial
_{i}$. At each point of $P\in M$, $X$ determines a smooth curve $\gamma
_{X}\left( \varepsilon ,P\right) =c_{X}^{i}\left( \varepsilon ,P\right)
\partial _{i}$, where $\varepsilon \in J_{\varepsilon }~$and $J_{\varepsilon
}~$is an open intevral of $%
\mathbb{R}
,$ as follows
\begin{equation}
\frac{dc_{X}^{i}\left( \varepsilon ,P\right) }{d\varepsilon }=X^{i}\left(
c_{X}^{i}\left( \varepsilon ,P\right) \right) ~,~~c_{X}^{i}\left( 0,P\right)
=x^{i}\left( P\right) .  \label{PT.01}
\end{equation}%
The curve $\gamma _{X}\left( J_{\varepsilon },P\right) $ is called the%
\textbf{\ integral curve }of $X$ through $P$.

Equation (\ref{PT.01}) defines an autonomous system of ordinary differential
equations (ODEs) with solutions $\gamma _{X}\left( t,P\right) $ subject to
the initial conditions $c_{X}^{i}\left( 0,P\right) =x^{i}\left( P\right) $.
The existence and uniqueness of integral curves is given by, the following
theorem~\cite{SergeBook}.

\begin{theorem}
\label{IntC}Let $\gamma _{1}\left( J_{1},P\right) ~$and $\gamma _{2}\left(
J_{2},P\right) $ be two integral curves of the vector field $X$ on $M,$ with
the same initial condition $x^{i}\left( P\right) $.~Then $\gamma _{1}\left(
J_{1},P\right) ~$and $\gamma _{2}\left( J_{2},P\right) $ are equal on $%
J_{1}\cap J_{2}$,~where $J_{1},J_{2}$ are two open intervals of $%
\mathbb{R}
.$
\end{theorem}

In the case where $\gamma _{X}\left( J_{\varepsilon },P\right) $ defines a
one parameter point transformation, the system (\ref{PT.01}) defines the
generator of the infinitesimal transformation. Due to the uniqueness of the
solution there exists only one parameter point transformation for each
vector field $X$.

\begin{example}
Consider the space $%
\mathbb{R}
^{2}$ with coordinates $\left( x,y\right) $ and the vector field $%
X=y\partial _{x}+x\partial _{y}$. Find the integral curve of $X~$through the
point $P=\left( x_{0},y_{0}\right) .$

Solution: Let $\gamma _{P}\left( \varepsilon \right) =\left( x\left(
\varepsilon \right) ,y\left( \varepsilon \right) \right) ~$be the integral
curve of $X$ passing through $P$. The system of autonomous first order
equations defining the integral curves are%
\begin{equation*}
\frac{dx}{d\varepsilon }=y~~,~~\frac{dy}{d\varepsilon }=x
\end{equation*}%
with the initial condition $x\left( 0\right) =x_{0}~,~y\left( 0\right)
=y_{0} $. The solution of this system is%
\begin{eqnarray}
x\left( \varepsilon \right) &=&x_{0}\cosh \varepsilon +y_{0}\sinh \varepsilon
\label{PT.01a} \\
y\left( \varepsilon \right) &=&x_{0}\sinh \varepsilon +y_{0}\cosh
\varepsilon .  \label{PT.01b}
\end{eqnarray}%
Equations (\ref{PT.01a}),(\ref{PT.01b}) define the rotation in the
hyperbolic space.
\end{example}

\subsection{Invariant Functions}

\label{Fun}

Let $F\left( x,y\right) $ be a function in $M$. Under the one parameter
point transformation~%
\begin{equation*}
\bar{x}=\bar{x}\left( x,y,\varepsilon \right) ,~\bar{y}=\bar{y}\left(
x,y,\varepsilon \right)
\end{equation*}
the function becomes $\bar{F}\left( \bar{x},\bar{y}\right) .$

\begin{definition}
\label{invariantFunctions}The function $F$ is invariant under the one
parameter point transformation if and only if $\bar{F}\left( \bar{x},\bar{y}%
\right) =0$ when$~F\left( x,y\right) =0$ at all points where the one
parameter point transformation acts. Equivalently, the generator $X$ of the
point transformation $\ $ is a symmetry of the function $F$ if
\begin{equation}
X\left( F\right) =0~~,~~modF=0  \label{InF.01}
\end{equation}
\end{definition}

The symmetry condition (\ref{InF.01}) is equivalent to the first order
partial differential equation (PDE)
\begin{equation}
\xi \frac{\partial F}{\partial x}+\eta \frac{\partial F}{\partial y}=0.
\label{InF.02}
\end{equation}

In order to determine all functions which are invariant under the
infinitesimal generator $X$ one has to solve the associated Lagrange system%
\begin{equation*}
\frac{dx}{\xi \left( x,y\right) }=\frac{dy}{\eta \left( x,y\right) }.
\end{equation*}

The characteristic function or zero order invariant $W$ of $X$ is defined as
follows
\begin{equation}
dW=\frac{dx}{\xi \left( x,y\right) }-\frac{dy}{\eta \left( x,y\right) }.
\end{equation}%
The zero order invariant is indeed invariant under the $X,$ that is $X\left(
W\right) =0.$ Therefore, any function of the form $F=F\left( W\right) ,$
where $W$ is the zero order invariant satisfies (\ref{InF.02}) and it is
invariant under the one parameter point transformation with generator $X.$

\subsection{Lie Algebras}

\label{LieAlgebra}

In section \ref{PTran} we considered the one parameter point transformation
which depend on one parameter $\varepsilon $. However transformations can
depend on more that one parameter, as follows%
\begin{equation}
\bar{x}=\bar{x}\left( x,y,\mathbf{E}\right) ~~,~~\bar{y}=\bar{y}\left( x,y,%
\mathbf{E}\right)  \label{LA.00}
\end{equation}%
where $\mathbf{E}=\varepsilon ^{\beta }\partial _{\beta },$~is a vector
field in the $%
\mathbb{R}
^{\kappa },~\beta =1...\kappa $ with the same properties of definition \ref%
{1ppt}, is a multi parameter point transformation,

For every parameter $\varepsilon ^{\beta }~$of the multi parameter point
transformation (\ref{LA.00}) an infinitesimal generator can be defined
\begin{equation*}
X_{\beta }=\xi _{\beta }\left( x,y\right) \partial _{x}+\eta _{\beta }\left(
x,y\right) \partial _{y}
\end{equation*}%
where
\begin{equation*}
\xi _{\beta }=\frac{\partial \bar{x}}{\partial \varepsilon _{\beta }}%
|_{\varepsilon _{\beta }\rightarrow 0}~~,~~\eta _{\beta }=\frac{\partial
\bar{y}}{\partial \varepsilon _{\beta }}|_{\varepsilon _{\beta }\rightarrow
0}.
\end{equation*}

Let $F\left( x,y\right) $ be a function in $M$ which is invariant under a
multiparameter point transformation. Since the multi parameter
transformation can be described as $m$ one parameter point transformations, $%
F$ is invariant under $m$ infinitesimal generators.

\begin{definition}
A Lie algebra is a finite dimensional linear space $G,$ in which a binary
operator, denoted $\left[ ~,~\right] $ has been defined which has the
following properties

i) $\left[ X,X\right] =0$ for all $X\in G$

ii) $\left[ X,\left[ Y,Z\right] \right] +\left[ Y,\left[ Z,X\right] \right] +%
\left[ Z,\left[ X,Y\right] \right] =0$ for all $X,Y,Z\in G$.

iii)If $X_{A},X_{B}\in G$ then$~\left[ X_{A},X_{B}\right] =C_{AB}^{C}X_{C}$,
$X_{C}\in G$. The quantities $C_{AB}^{C}$ are constants and are called the
\textbf{structure constants} of the Lie algebra.
\end{definition}

The operator $\left[ ~,~\right] $ is called the \textbf{commutator} and it
is defined by the following expression%
\begin{equation}
\left[ X_{A},X_{B}\right] =X_{A}X_{B}-X_{B}X_{A}=-\left[ X_{B},X_{N}\right] .
\label{LA.01}
\end{equation}

From the definition of the commutator (\ref{LA.01}) and from the
requirements of the Lie algebra, it follows that the structure constants are
antisymmetric in the two lower indices, i.e.%
\begin{equation}
C_{AB}^{C}+C_{BA}^{C}=0  \label{LA.02}
\end{equation}%
and they have to satisfy the Jacobi identity%
\begin{equation}
C_{AB}^{E}C_{DE}^{C}+C_{BD}^{E}C_{AE}^{C}+C_{DA}^{E}C_{BE}^{C}=0.
\label{LA.03}
\end{equation}%
The structure constants characterize the Lie algebra because every set of
constants $C_{AB}^{C}$ which satisfy (\ref{LA.02}) and (\ref{LA.03}) defines
\emph{locally} a unique Lie group. An important property of the structure
constants is that they do not change under a coordinate transformation. The
structure constants do change under a transformation of the basis; this
property is useful because it can be used to simplify the structure
constants of the given group.

\begin{definition}
Let $G,H$ be closed Lie algebras with elements $\left\{ X_{A}\right\}
,~\left\{ Y_{a}\right\} $ respectively. If $\dim H\leq \dim G$ and $Y_{a}\in
G$ then $H$ is called a Lie subalgebra of $G.$
\end{definition}

Suppose \ that the vector fields $X,Y$ leave invariant a function $F=F\left(
x,y\right) $. If $\left[ X,Y\right] =Z~$with~$Z\neq X,Y$, i.e. the
generators $X,Y$ do not form a closed Lie algebra, then $F$ is also
invariant under the action of $Z$. This process can be used to find extra
symmetries.

\begin{example}
\label{exso3}The vector fields
\begin{equation*}
X_{1}=\sin \theta \partial _{\phi }+\cos \theta \cot \phi \partial _{\theta
}~\ ,~X_{2}=\cos \theta \partial _{\phi }-\sin \theta \cot \phi \partial
_{\theta }~~,~~X_{3}=\partial _{\theta }
\end{equation*}%
span the $so\left( 3\right) $ Lie algebra with\ structure constants~$%
C_{12}^{3}=C_{31}^{2}=C_{23}^{1}=1,$ i.e. the commutators are%
\begin{equation*}
\left[ X_{1},X_{2}\right] =X_{3}~~,~~\left[ X_{3},X_{1}\right] =X_{2}~~,~~%
\left[ X_{2},X_{3}\right] =X_{1}
\end{equation*}%
If the function $F=F\left( \theta ,\phi \right) $ is invariant under the
infinitesimal generators $X_{1},X_{3}$, then it is also invariant under the
action of $X_{2}$. In that case, it is easy to see that $F=F_{0},$ where $%
F_{0}$ is a constant.
\end{example}

\section{Lie symmetries of differential equations}

\label{LieSym}

Previously, we studied the case when a function $F\in M~$is invariant under
the action of a one parameter point transformation. In the following we
consider the case of differential equations (DEs) which are invariant under
a group of one parameter point transformations.

\subsection{Prolongation of point transformations}

\label{inF}

In order to study the action of a point transformation to a differential
equation~$H\left( x,y,y^{\prime },...,y^{\left( n\right) }\right) ~$where~$%
y=y\left( x\right) $, we have to prolong the point transformation to the
derivatives $y^{\left( n\right) }$. The infinitesimal transformation (\ref%
{PT.02}) in the jet space $B_{M}=\left\{ x,y,y^{\prime },...,y^{\left(
n\right) }\right\} $ is%
\begin{eqnarray*}
\bar{x} &=&x+\varepsilon \xi \\
\bar{y} &=&x+\varepsilon \eta \\
\bar{y}^{\left( 1\right) } &=&y^{\left( 1\right) }+\varepsilon \eta ^{\left[
1\right] } \\
&&... \\
\bar{y}^{\left( n\right) } &=&y^{\left( n\right) }+\varepsilon \eta ^{\left[
n\right] }
\end{eqnarray*}%
where $y^{\left( n\right) }=\frac{d^{n}y}{dx^{n}},~\bar{y}^{\left( n\right)
}=\frac{d^{n}\bar{y}}{d\bar{x}^{n}}$ and
\begin{equation*}
\eta ^{\left[ 1\right] }=\frac{\partial \bar{y}^{\left( 1\right) }}{\partial
\varepsilon }~~,...,~~\eta ^{\left[ n\right] }=\frac{\partial \bar{y}%
^{\left( n\right) }}{\partial \varepsilon }.
\end{equation*}%
That means that the variation equals the difference of the derivatives
before and after the action of the one parameter transformation. For the
first prolongation function $\eta ^{\left[ 1\right] }~$we have%
\begin{equation*}
\eta ^{\left[ 1\right] }\equiv \lim_{\varepsilon \rightarrow 0}\left[ \frac{1%
}{\varepsilon }\left( \bar{y}^{\left( 1\right) }-y^{\left( 1\right) }\right) %
\right] =\frac{d\eta }{dx}-y^{\left( 1\right) }\frac{d\xi }{dx}.
\end{equation*}%
Similarly for$~\eta ^{\left[ n\right] }$ we have the expression%
\begin{equation}
\eta ^{\left[ n\right] }=\frac{d\eta ^{n-1}}{dx}-y^{\left( n\right) }\frac{%
d\xi }{dx}=\frac{d^{n}}{dx^{n}}\left( \eta -y^{\left( 1\right) }\xi \right)
+y^{\left( n+1\right) }\xi .  \label{PP.011}
\end{equation}

Finally, the extension of the infinitesimal generator in the jet space $%
B_{M}~$is
\begin{equation*}
X^{\left[ n\right] }=X+\eta ^{\left[ 1\right] }\partial _{y^{\left( 1\right)
}}+...+\eta ^{\left[ n\right] }\partial _{y^{\left[ n\right] }}.
\end{equation*}%
The field $X^{\left[ n\right] }\in B_{M}$ is called the \textbf{nth
prolongation} of the generator $X,~$where
\begin{equation*}
X=\xi \left( x,y\right) \partial _{x}+\eta \left( x,y\right) \partial _{y}
\end{equation*}%
is the infinitesimal point generator in the space $\left\{ x,y\right\} $.

It is possible to write the prolongation coefficients in terms of the
partial derivatives of the components $\xi \left( x,y\right) ,~\eta \left(
x,y\right) $. The first and the second prolongation are expressed as follows%
\begin{equation*}
X^{\left[ 1\right] }=X+\left[ \eta _{,x}+y^{\left( 1\right) }\left( \eta
_{,y}-\xi _{,x}\right) -y^{\left( 1\right) ^{2}}\xi _{,y}\right] \partial
_{y^{\left( 1\right) }}
\end{equation*}%
\begin{equation}
X^{\left[ 2\right] }=X^{\left[ 1\right] }+\left[
\begin{array}{c}
\eta _{,xx}+2\left( \eta _{,xy}-\xi _{,xx}\right) y^{\left( 1\right)
}+\left( \eta _{,yy}-2\xi _{,xy}\right) y^{\left( 1\right) ^{2}}+ \\
-y^{\left( 1\right) ^{3}}\xi _{,yy}+\left( \eta _{,y}-2\xi _{,x}-3\xi
_{,y}y^{\left( 1\right) }\right) y^{\left( 2\right) }%
\end{array}%
\right] \partial _{y^{\left( 2\right) }}  \label{PP.01D}
\end{equation}%
where the comma $","$ denotes partial derivative.

Some important observations \cite{BlumanB} for the prolongation coefficient $%
\eta ^{\left[ n\right] }$ are:

(a) $\eta ^{\left[ n\right] }$ is linear in $y^{\left( n\right) }$

(b) $\eta ^{\left[ n\right] }$ is a polynomial in the derivatives $y^{\left(
1\right) },...,y^{\left( n\right) }$ whose coefficients are linear
homogeneous in the functions $\xi \left( x,y\right) ,$ $\eta \left(
x,y\right) $ up to nth order partial derivatives.

\subsubsection{Multiparameter prolongation}

In the case the differential equation $H$ depends on $n$ independent$~$%
variables~$\left\{ x^{i}:i=1..n\right\} $ and $m$ dependent variables~$%
\left\{ u^{A}:A=1...m\right\} $, i.e. $H=H\left(
x^{i},u^{A},u_{,i}^{A},u_{,ij}^{A},..\right) ,$ we consider the one
parameter point transformation%
\begin{equation*}
\bar{x}^{i}=\Xi ^{i}\left( x^{i},u^{A},\varepsilon \right) ~~,~~\bar{u}%
^{A}=\Phi ^{A}\left( x^{i},u^{A},\varepsilon \right) .
\end{equation*}%
In this case the generating vector is
\begin{equation}
X=\xi ^{i}\left( x^{k},u^{A}\right) \partial _{i}+\eta ^{A}\left(
x^{k},u^{A}\right) \partial _{A}  \label{MP.01}
\end{equation}%
where%
\begin{equation*}
\xi ^{i}\left( x^{k},u^{A}\right) =\frac{\partial \Xi ^{i}\left(
x^{i},u^{A},\varepsilon \right) }{\partial \varepsilon }|_{\varepsilon
\rightarrow 0}~~,~~\eta ^{A}\left( x^{k},u^{A}\right) =\frac{\partial \Phi
\left( x^{i},u^{A},\varepsilon \right) }{\partial \varepsilon }%
|_{\varepsilon \rightarrow 0}.
\end{equation*}

To extend the generator vector in the jet space $\bar{B}_{\bar{M}}=\left\{
x^{i},u^{A},u_{,i}^{A},u_{,ij}^{A},..,u_{ij...i_{n}}^{A}\right\} $ we apply
the same procedure as in section \ref{inF}. Therefore the vector field $X^{%
\left[ n\right] }\in \bar{B}_{\bar{M}}$
\begin{equation*}
X^{\left[ n\right] }=X+\eta _{i}^{A}\partial _{u_{i}}+...+\eta
_{ij..i_{n}}^{A}\partial _{u_{ij..i_{n}}}
\end{equation*}%
is defined as the nth prolongation of the generator (\ref{MP.01}), where%
\footnote{%
Where $D_{i}=\frac{\partial }{\partial x^{i}}+u_{i}^{A}\frac{\partial }{%
\partial u^{A}}+u_{ij}^{A}\frac{\partial }{\partial u_{j}^{A}}%
+...+u_{ij..i_{n}}^{A}\frac{\partial }{\partial u_{jk..i_{n}}^{A}}.$}%
\begin{equation}
\eta _{i}^{A}=D_{i}\eta ^{A}-u_{,j}^{A}D_{i}\xi ^{j}  \label{MP.02A}
\end{equation}%
\begin{equation}
\eta _{ij..i_{n}}^{A}=D_{i_{n}}\eta _{ij..i_{n-1}}^{A}-u_{ij..k}D_{i_{n}}\xi
^{k}.  \label{MP.02}
\end{equation}

In terms of the partial derivatives of the components $\xi ^{i}\left(
x^{k},u^{A}\right) ,~\eta ^{A}\left( x^{k},u^{A}\right) $ of the generator
vector (\ref{MP.01}), the first and the second prologations of (\ref{MP.01})
are expressed as follows%
\begin{equation}
X^{\left[ 1\right] }=X+\left( \eta _{,i}^{A}+u_{,i}^{B}\eta _{,B}^{A}-\xi
_{,i}^{j}u_{,j}^{A}-u_{,i}^{A}u_{,j}^{B}\xi _{,B}^{j}\right) \partial
_{u_{i}^{A}}  \label{MP.03}
\end{equation}%
\begin{equation}
X^{\left[ 2\right] }=X^{\left[ 1\right] }+\left[
\begin{array}{c}
\eta _{,ij}^{A}+2\eta _{,B(i}^{A}u_{,j)}^{B}-\xi _{,ij}^{k}u_{,k}^{A}+\eta
_{,BC}^{A}u_{,i}^{B}u_{,j}^{C}-2\xi _{,(i\left\vert B\right\vert
}^{k}u_{j)}^{B}u_{,k}^{A}+ \\
-\xi _{,BC}^{k}u_{,i}^{B}u_{,j}^{A}u_{,k}^{A}+\eta _{,B}^{A}u_{,ij}^{B}-2\xi
_{,(j}^{k}u_{,i)k}^{A}+-\xi _{,B}^{k}\left(
u_{,k}^{A}u_{,ij}^{B}+2u_{(,j}^{B}u_{,i)k}^{A}\right)%
\end{array}%
\right] \partial _{u_{ij}}.  \label{MP.04}
\end{equation}

\subsection{Lie symmetries of ODEs}

In the previous sections we analyzed the invariance of functions under the
action of a point transformation. In the following we define the invariance
of ordinary differential equations (ODEs) under a one parameter point
transformation.

Consider the $N-$dimensional system of ODEs\footnote{%
In the following equations,~$t$ is the independent parameter and $%
x^{i}=x^{i}\left( t\right) $ the dependent parameters.}%
\begin{equation}
x^{\left( n\right) i}=\omega ^{i}\left( t,x^{k},\dot{x}^{k},\ddot{x}%
^{k},...,x^{\left( n-1\right) i}\right)  \label{Ls.01}
\end{equation}%
where $\dot{x}^{i}=\frac{dx^{i}}{dt}$ , $x^{\left( n\right) }=\frac{d^{n}x}{%
dt^{n}}$ and the infinitesimal point transformation with infinitesimal
generator $X$ is
\begin{eqnarray}
\bar{t} &=&t+\varepsilon \xi \left( t,x^{k}\right)  \label{Ls.02} \\
\bar{x}^{i} &=&x+\varepsilon \eta ^{i}\left( t,x^{k}\right)  \label{Ls.03}
\end{eqnarray}

\begin{theorem}
Let
\begin{equation}
X=\xi \left( t,x^{k}\right) \partial _{t}+\eta ^{i}\left( t,x^{k}\right)
\partial _{i}  \label{Ls.04}
\end{equation}%
be the infinitesimal generator of point transformation (\ref{Ls.02})-(\ref%
{Ls.03}) and
\begin{equation*}
X^{\left[ n\right] }=X+\eta _{\left[ 1\right] }^{i}\partial _{\dot{x}%
^{i}}+...+\eta _{\left[ n\right] }^{i}\partial _{x^{\left( n\right) i}}
\end{equation*}%
be the nth prolongation of $X$, where $\eta _{\left[ n\right] }^{i}$ is
given by (\ref{MP.02}). We shall say that the $N-$dimensional system of ODEs
(\ref{Ls.01}) is invariant under the point transformation (\ref{Ls.02}), (%
\ref{Ls.03}) if and only if there exists a function $\lambda $ such as the
following condition holds%
\begin{equation}
\left[ X^{\left[ n\right] },A\right] =\lambda A  \label{Ls.05}
\end{equation}%
where\footnote{%
For Hamiltonian systems, the operator $\mathbf{A}$ is called \textbf{%
Hamiltonian vector field}.}
\begin{equation*}
A=\frac{\partial }{\partial t}+\dot{x}^{i}\frac{\partial }{\partial x^{i}}%
+...+\omega ^{i}\left( t,x^{k},\dot{x}^{k},\ddot{x}^{k},...,x^{\left(
n-1\right) i}\right) \frac{\partial }{\partial x^{\left( n\right) i}}.
\end{equation*}%
In that case we say that $X$ is a \textbf{Lie point symmetry }of the\textbf{%
\ }$N-$dimensional system of ODEs (\ref{Ls.01}).
\end{theorem}

If $f^{i}$ is a solution of the system (\ref{Ls.01}), i.e. $Af^{i}=0,$ then
condition (\ref{Ls.05}) becomes $X\left( Af^{i}\right) =0,~$that is,
\begin{equation}
\eta _{\left[ n\right] }^{i}=X^{\left[ n-1\right] }\omega ^{i}\left( t,x^{k},%
\dot{x}^{k},\ddot{x}^{k},...,x^{\left( n-1\right) i}\right) .  \label{Ls.06}
\end{equation}%
Equations (\ref{Ls.06}) are called the determining equations. The solution
of the determining equations (\ref{Ls.06}) gives the infinitesimal
generators of the transformation (\ref{Ls.02})-(\ref{Ls.03}).

In general, a function $H\left( t,\dot{x}^{k},\ddot{x}^{k},...,x^{\left(
n\right) k}\right) =0$ is invariant under the transformation (\ref{Ls.02})-(%
\ref{Ls.03}) if and only if%
\begin{equation}
X^{\left[ n\right] }\left( H\right) =\lambda H~,~modH=0.
\label{Ip.01a}
\end{equation}%
where $\lambda $ is a function to be determined~\cite{IbragB}.

Below we give an example in which the Lie symmetries are calculated using
the symmetry condition (\ref{Ls.06}).

\begin{example}
\label{ExLs}Find the Lie symmetries of the ODE~$\ddot{x}=0.$

Solution: Condition (\ref{Ls.06}) gives $\eta _{\left[ 2\right] }=0.~$From (%
\ref{PP.01D}) the following condition is found
\begin{equation}
\eta _{,tt}+2\left( \eta _{,tx}-\xi _{,tt}\right) \dot{x}+\left( \eta
_{,xx}-2\xi _{,tx}\right) \dot{x}^{2}-\dot{x}^{3}\xi _{,xx}=0  \label{Ip.03}
\end{equation}%
since $\ddot{x}=0$. Functions $\xi ,\eta $ are dependent only on the
variables $\left\{ t,x\right\} $, hence equation (\ref{Ip.03}) is a
polynomial of $\dot{x}.$ This polynomial must vanish identically hence the
coefficients of all powers of $\dot{x}$ must vanish. Therefore, we have the
following determining equations%
\begin{eqnarray*}
\left( \dot{x}\right) ^{0} &:&\eta _{,tt}=0 \\
\left( \dot{x}\right) ^{1} &:&\eta _{,xy}-\xi _{,tt}=0 \\
\left( \dot{x}\right) ^{2} &:&\eta _{,xx}-2\xi _{,tx}=0 \\
\left( \dot{x}\right) ^{3} &:&\xi _{,xx}=0.
\end{eqnarray*}%
whosw solution is%
\begin{eqnarray*}
\xi \left( t,x\right) &=&a_{1}+a_{2}t+a_{3}t^{2}+a_{4}x+a_{5}tx \\
\eta \left( t,x\right) &=&a_{6}+a_{7}t+a_{8}x+a_{3}tx+a_{5}x^{2}.
\end{eqnarray*}%
We conclude that the second order ODE $\ddot{x}=0$ has eight Lie pont
symmetries generated by the generic vector field\footnote{%
As many as the unspecified constants in the expression of the generic
symmetric vector.}%
\begin{equation}
X=\left( a_{1}+a_{2}t+a_{3}t^{2}+a_{4}x+a_{5}tx\right) \partial _{t}+\left(
a_{6}+a_{7}t+a_{8}x+a_{3}tx+a_{5}x^{2}\right) \partial _{x}.  \label{ex.00}
\end{equation}%
These vectors are the generators of the projective algebra $sl\left(
3,R\right) $ of the 2-d Euclidian plane. Furthermore, this is the maximum
number of symmetries that a single second order ODE (in one variable!) can
have.
\end{example}

Lie symmetries can be used to find invariants or to reduce the order of an
ODE. Furthermore an ODE is characterized by the admitted algabra of Lie
symmetries. For second order ODEs we have the following theorem.

\begin{theorem}
\label{LieTheor}If a second order ODE$~$admits as Lie point symmetries the
eight \ of $sl\left( 3,R\right) $, then, there exists a transformation which
brings the equation to the form $x^{\ast \prime \prime }=0$ and vice versa.
\end{theorem}

For instance the equations
\begin{equation*}
~\ddot{x}+\frac{1}{x}\dot{x}^{2}=0
\end{equation*}%
\begin{equation*}
\ddot{x}+3x\dot{x}+x^{3}=0
\end{equation*}%
\begin{equation*}
\ddot{x}+\omega _{0}^{2}x+\sin t=0
\end{equation*}%
\begin{equation*}
\left( 4+x^{2}+t^{2}\right) \ddot{x}-2t\dot{x}^{3}+2x\dot{x}^{2}-2t\dot{x}%
+2x=0
\end{equation*}%
are equivalent to the equation of motion of a free particle \cite%
{Leach80a,Prince01,Karasu09}, because they are invariant under the action of
the Lie algebra $sl\left( 3,R\right) .$ The transformation where the
admitted $sl\left( 3,R\right) $ algebra is written in the form (\ref{ex.00})
is a coordinate transformation on $M:$ $\left( t,x\right) \rightarrow \left(
\tau ,y\right) $ which transforms the ODE to the form $\frac{d^{2}y}{d\tau
^{2}}=0.$ This procedure is called linearization process (see \cite%
{Mahq07,Mahq08} and references therein).

Below we present two methods where the use of Lie symmetries reduces the
order of an ODE.

\subsubsection{Canonical coordinates}

\label{Cancoor}

\begin{definition}
Let $X~$be a vector field with coordinates $X=\xi \left( t,x\right) \partial
_{t}+\eta \left( t,x\right) \partial _{x}$. If under the coordinate
transformation $\left\{ t,x\right\} \rightarrow \left\{ r,s\right\} $ holds
that%
\begin{equation}
Xr=0~~,~~Xs=1  \label{cc.00}
\end{equation}%
then, we say that $\{r,s\}$ are the \textbf{canonical coordinates} of $X,$
i.e. $X=\partial _{s}$.
\end{definition}

Canonical coordinates can be used to reduce by one the order of an ODE. \
Consider the $nth$~order ODE $(n\succeq 2)$
\begin{equation}
\frac{d^{n}s}{dr^{2}}=\bar{\omega}\left( r,s,\frac{ds}{dr},...,\frac{d^{n-1}s%
}{dr^{n-1}}\right) .  \label{cc.00a}
\end{equation}%
Let $X_{C}=\partial _{s}$ be a Lie symmetry of (\ref{cc.00a}) written in
canonical coordinates.~The nth prolongation of the symmetry vector is $%
X_{C}^{\left[ n\right] }=X_{C},~$so condition (\ref{Ls.06}) sets the
constraint%
\begin{equation*}
\frac{\partial }{\partial s}\bar{\omega}\left( r,s,\frac{ds}{dr},...,\frac{%
d^{n-1}s}{dr^{n-1}}\right) =0.
\end{equation*}%
This implies that the function $\bar{\omega}$ is independent of $s,$
consequently (\ref{cc.00a}) can be written in the form%
\begin{equation}
\frac{d^{n-1}S}{dr^{n-1}}=\bar{\omega}\left( r,S,\frac{dS}{dr},...,\frac{%
d^{n-2}S}{dr}\right)  \label{cc.00c}
\end{equation}%
which is~a $\left( n-1\right) $ order ODE, where $S~$is defined~by $S=\frac{%
ds}{dr}$.

\begin{example}
Use the canonical coordinates of the Lie symmetry $X=x\partial
_{x}-t\partial _{t}$ to reduce the order of the ODE
\begin{equation}
\ddot{x}+x\dot{x}=0  \label{cc.000}
\end{equation}%
Solution: The canonical coordinates of $X$ are%
\begin{equation*}
x=e^{s}~~,~~t=re^{-s}.
\end{equation*}%
For the first and the second derivative of $x\left( t\right) $ we compute
\begin{equation*}
\dot{x}=\frac{e^{2s}s^{\prime }}{1-rs^{\prime }}
\end{equation*}%
\begin{equation*}
\ddot{x}=\frac{e^{3s}}{\left( 1-rs^{\prime }\right) ^{3}}\left[ \left(
2s^{\prime 2}+s^{\prime \prime }\right) \left( 1-rs^{\prime }\right)
+s^{\prime }\left( s^{\prime }+rs^{\prime \prime }\right) \right]
\end{equation*}%
where $s^{\prime }=\frac{ds}{dr}\,$.

Replacing in \ref{cc.000}, we find the reduced equation%
\begin{equation}
\left( 2S^{2}+S^{\prime }\right) \left( 1-rS\right) +S\left( S+rS^{\prime
}\right) +S\left( 1-rS\right) ^{2}=0  \label{cc.00d}
\end{equation}%
where we have set $S=s^{\prime }$. Equation (\ref{cc.00d}) is an Abel
equation of the first kind \cite{PolyaninB} invariant under $X=\frac{d}{ds}.$
\end{example}

\subsubsection{Invariants}

As in the case of functions in $M~$(see section \ref{Fun}), condition (\ref%
{Ip.01a}) is equivalent to the following Lagrange system%
\begin{equation}
\frac{dt}{\xi }=\frac{dx}{\eta }=\frac{d\dot{x}}{\eta _{\left[ 1\right] }}%
=...=\frac{dx^{\left( n\right) }}{\eta _{\left[ n\right] }}.  \label{Ip.02}
\end{equation}%
The Lagrange system (\ref{Ip.02}) provides us with characteristic functions
\begin{equation*}
W^{\left[ 0\right] }\left( t,x\right) ,~W^{\left[ 1\right] i}\left( t,x,\dot{%
x}\right) ,~W^{\left[ n\right] }\left( t,x,\dot{x},...,\dot{x}^{\left(
n\right) }\right)
\end{equation*}%
where $W^{\left[ n\right] }$ is called the \textbf{nth order invariant }of
the Lie symmetry vector. If (\ref{Ls.04}) is a Lie symmetry for the ODE%
\begin{equation}
x^{\left( n\right) }=\omega \left( t,\dot{x},\ddot{x},...,x^{\left(
n-1\right) }\right)  \label{Ip.03a}
\end{equation}%
it follows that (\ref{Ip.03a}) can be written as a function of the
characteristic functions $W^{\left[ 1\right] },...,W^{\left[ n\right] },~$of
(\ref{Ls.04}).

Let $u=W^{\left[ 0\right] }~,~v=W^{\left[ 1\right] },$ where $W^{\left[ 0%
\right] },~W^{\left[ 1\right] }$ are the zero and the first order invariants
of a Lie symmetry repetitively. From the invariants $u,~v,we$ define the
differential invariants
\begin{equation}
\frac{dv}{du}~,...,\frac{d^{n-1}v}{du^{n-1}}.  \label{Ip.04}
\end{equation}%
where
\begin{equation*}
\frac{dv}{du}=\frac{\frac{\partial v}{\partial t}+\frac{\partial v}{\partial
x}\dot{x}+...+\frac{\partial v}{\partial \dot{x}}\ddot{x}}{\frac{\partial v}{%
\partial t}+\frac{\partial v}{\partial x}\dot{x}}.
\end{equation*}

The differential invariants are functions of $x^{\left( n\right) }$, hence,
it is feasible that (\ref{Ip.03a}) may be written in terms of the
differential invariants (\ref{Ip.04}), i.e.%
\begin{equation}
\frac{d^{n-1}v}{du^{n-1}}=\Omega \left( u,v,\frac{dv}{du}~,...\frac{d^{n-1}v%
}{du^{n-1}}\right)
\end{equation}%
which is a $\left( n-1\right) $ order ODE.

\begin{example}
Use the first order invariants of the Lie symmetry $X=x\partial
_{x}-nt\partial _{t}$ to reduce the order of the Lane-Emden equation%
\begin{equation*}
\ddot{x}+\frac{2}{t}\dot{x}+x^{2n+1}=0
\end{equation*}%
where $n\neq -\frac{1}{2},0$. The Lane-Emden equation arises in the study of
equilibrium configurations of a spherical gas cloud acting under the mutual
attractions of its molecules and subject to the laws of thermodynamics \cite%
{LELeach,Khalique}.

Solution: The first prolongation of $X$ is
\begin{equation*}
X^{\left[ 1\right] }=x\partial _{x}-nt\partial _{t}+\left( n+1\right) \dot{x}%
\partial _{\dot{x}}
\end{equation*}%
hence, the corresponding Lagrange system is%
\begin{equation*}
\frac{dt}{-nt}=\frac{dx}{x}=\frac{d\dot{x}}{\left( n+1\right) \dot{x}}.
\end{equation*}%
The zero and the first order invariants are found to be
\begin{equation*}
u=xt^{\frac{1}{n}}~~,~~v=\dot{x}t^{1+\frac{1}{n}}.
\end{equation*}%
The differential invariant is defined as%
\begin{equation*}
\frac{dv}{du}=\frac{\left( 1+\frac{1}{n}\right) \dot{x}+t\ddot{x}}{\frac{1}{n%
}t^{-1}x+\dot{x}}.
\end{equation*}%
Substituting in the Lane-Emden equation we obtain the first order ODE%
\begin{equation}
v^{\prime }\left( u+nv\right) -\left( 1-n\right) v+nu^{2n+1}=0.  \label{Abel}
\end{equation}%
where $v^{\prime }=\frac{dv}{du}$. Equation (\ref{Abel}) is an Abel equation
of the second kind \cite{PolyaninB}.
\end{example}

An interesting application of the Lie invariants to the classical Kepler
system can be found in \cite{PrinceKepler} where the authors derive the
Runge-Lenz vector using the first order invariants of a point transformation.

It is possible an ODE to admit many symmetries which span a Lie algebra $%
G_{m}$ of dimension $m>1$. The following theorem relates the Lie point
symmetries of the reduced and of the original equation.

\begin{theorem}
\label{RedTh}Consider an ODE which admits the Lie point symmetries $%
X_{1},~X_{2}$ which are such that $\left[ X_{1},X_{2}\right]
=C_{12}^{1}X_{1}.~$Then, the reduction by $X_{1}$ leads to a reduced
equation which admits $X_{2}$ as a Lie symmetry whereas reduction of the ODE
by $X_{2}$ leads to a reduced equation, which does not admit $X_{1}$ as a
Lie symmetry.
\end{theorem}

In case the generators $X_{1},X_{2}$ \ form an Abelian Lie algebra, i.e. $%
\left[ X_{1},X_{2}\right] =0,$ then, the reduction preserves the Lie
symmetries. Theorem \ref{RedTh} is important because it gives a hint as to
which generator the reduction of an ODE must start in order to continue the
reduction process with the reduced equation.

Since the reduced equation is different from the original equation, it is
possible the reduced equation to admit extra Lie symmetries, which are not
Lie symmetries of the original equation. These new Lie symmetries have been
named \textbf{Type II hidden symmetries}. Type II hidden symmetries can be
used to continue the reduction process and - if it is feasible - to find a
solution of the original equation \cite{LeGovAb99}.

There are other methods to apply Lie symmetries; a quite intriguing
application of Lie symmetries is to produce integrals or Lagrangian
functions for a system of ODEs by the \ method of Jacobi's last multiplier
(see \cite{Nucci08,Nucci09,Leach2012} and references therein).

\section{Lie symmetries of PDEs}

A partial differential equation (PDE) is a function $H=H\left(
x^{i},u^{A},u_{,i}^{A},u_{,ij}^{A},..,u_{ij...i_{n}}^{A}\right) $ in the jet
space $\bar{B}_{\bar{M}},~$where $x^{i}$ are the independent variables and $%
u^{A}$ are the dependent variables. As in the case of ODEs we define the
invariance of a PDE under the infinitesimal point transformation%
\begin{eqnarray}
\bar{x}^{i} &=&x^{i}+\varepsilon \xi ^{i}\left( x^{k},u^{B}\right)
\label{pde1} \\
\bar{u}^{A} &=&\bar{u}^{A}+\varepsilon \eta ^{A}\left( x^{k},u^{B}\right)
\label{pde2}
\end{eqnarray}%
with generator $X$ as follows.

\begin{definition}
Let
\begin{equation}
X=\xi ^{i}\left( x^{k},u^{B}\right) \partial _{t}+\eta ^{A}\left(
x^{k},u^{B}\right) \partial _{B}  \label{pde3a}
\end{equation}%
be the generator of the infinitesimal point transformation (\ref{pde1})-(\ref%
{pde2}) and
\begin{equation*}
X^{\left[ n\right] }=X+\eta _{\left[ i\right] }^{A}\partial _{\dot{x}%
^{i}}+...+\eta _{\left[ ij...i_{n}\right] }^{A}\partial _{u_{ij...i_{n}}}
\end{equation*}%
be the nth prolongation vector, where $\eta _{\left[ ij...i_{n}\right] }^{A}$
is given from (\ref{MP.02}). Then, the transformation (\ref{pde1})-(\ref%
{pde2}) leaves invariant the PDE
\begin{equation}
H\left( x^{i},u^{A},u_{,i}^{A},u_{,ij}^{A},..\right) =0  \label{pde3}
\end{equation}%
if there exist a function $\lambda $ where the following condition holds,%
\begin{equation}
X^{\left[ n\right] }\left( H\right) =\lambda H~~,~~modH=0.
\label{pde4}
\end{equation}%
The infinitesimal generator $X$ is called a Lie point symmetry of the PDE (%
\ref{pde3}).
\end{definition}

In case the PDE (\ref{pde3}) can we written in solved form, i.e.%
\begin{equation*}
u_{ij...i_{n}}^{A}=h^{A}\left(
x^{i},u^{B},u_{,i}^{B},u_{,ij}^{B},..,u_{ij...i_{n-1}}^{B}\right)
\end{equation*}%
then the Lie condition (\ref{pde4}) is equivalent to the following system of
equations
\begin{equation}
\eta _{\left[ ij...i_{n}\right] }^{A}=X^{\left[ n-1\right] }h^{A}\left(
x^{i},u^{B},u_{,i}^{B},u_{,ij}^{B},..,u_{ij...i_{n-1}}^{B}\right) .
\label{pde5}
\end{equation}

Consequently the family of all solutions of (\ref{pde3}) is invariant under (%
\ref{pde1})-(\ref{pde2}) if condition (\ref{pde4}) holds. In the example
below we compute the Lie symmetries of the homogeneous heat equation.

\begin{example}
Find the Lie symmetries of the 1+1 heat equation
\begin{equation}
H\left( u_{,t},u_{,xx}\right) :u_{,xx}-u_{,t}=0.  \label{pde6}
\end{equation}

Solution: Let
\begin{equation*}
X=\xi ^{t}\left( t,x,u\right) \partial _{t}+\xi ^{x}\left( t,x,u\right)
\partial _{x}+\eta \left( t,x,u\right) \partial _{u}
\end{equation*}%
be a Lie symmetry of (\ref{pde6}). Then condition (\ref{pde5}) gives%
\begin{equation*}
\eta _{\left[ xx\right] }-\eta _{\left[ t\right] }=0.
\end{equation*}%
Substituting $\eta _{\lbrack t]},~\eta _{\left[ xx\right] }$ ~from (\ref%
{MP.03}),(\ref{MP.04}) and collecting terms of derivatives of $u\left(
t,x\right) $ we find the following determining equations%
\begin{equation*}
\xi _{,u}^{t}=0~~,~~\xi _{,x}^{t}=0~~~,~~\xi _{,u}^{x}=0~~,~~\xi
_{,ttt}^{t}=0
\end{equation*}%
\begin{equation*}
\xi _{,t}^{x}=-2\eta _{,xu}~~,~~\xi _{,x}^{x}=\frac{1}{2}\xi
_{,t}^{t}~~,~~\eta _{,uu}=0
\end{equation*}%
\begin{equation*}
\eta _{,xx}=\eta _{,t}~~,~~4\eta _{,tu}+\xi _{,tt}=0.
\end{equation*}%
The solution of the system of equations is
\begin{eqnarray*}
\xi ^{t}\left( t,x,u\right) &=&a_{1}t^{2}+2a_{2}t+a_{3} \\
\xi ^{x}\left( t,x,u\right) &=&a_{1}tx+a_{2}x+a_{4}t+a_{5} \\
\eta \left( t,x,u\right) &=&-\frac{a_{1}}{2}\left( t+\frac{1}{2}x^{2}\right)
u-\frac{a_{4}}{2}xu+a_{6}u+b\left( t,x\right)
\end{eqnarray*}%
where $b\left( t,x\right) $ is a solution of the heat equation. The generic
vector of the infinitesimal transformation which leaves invariant the heat
equation is
\begin{eqnarray*}
X &=&\left( a_{1}t^{2}+2a_{2}t+a_{3}\right) \partial _{t}+\left(
a_{1}tx+a_{2}x+a_{4}t+a_{5}\right) \partial _{x}+ \\
&&+\left[ -\frac{a_{1}}{2}\left( t+\frac{1}{2}x^{2}\right) u-\frac{a_{4}}{2}%
xu+a_{6}u+b\left( t,x\right) \right] \partial _{u}.
\end{eqnarray*}
\end{example}

\subsubsection{Invariant solutions}

In section \ref{Cancoor}, it has been explained how to use the canonical
coordinates to reduce the order of an ODE. The infinitesimal generator $%
X_{3}=\partial _{t}$ is a Lie symmetry of the heat equation (\ref{pde6}),
written in canonical coordinates. Nevertheless (\ref{pde6}) cannot be
reduced as it happened with the case of ODEs. However, one can calculate
from the Lie symmetry $X_{3}$ the zero order invariants, which are $\
y=x~,~w=u$. We select $y$ to be the independent variable and $w=w\left(
y\right) .$ Substituting the invariants in (\ref{pde6}) we get the ODE~$%
w_{yy}=0.$ Then, the following solution can be found easily%
\begin{equation}
u\left( x\right) =c_{1}x+c_{2}  \label{pde7}
\end{equation}%
which is independent on the variable $t.$ That is, the use of the Lie
symmetry $X_{3}~$ reduces by one the number of independent variables but not
the order of the PDE. The solution (\ref{pde7}) is called \textbf{an
invariant solution}.

\begin{definition}
The function $u=U\left( x^{i}\right) $ is an invariant solution
corresponding to the infinitesimal generator (\ref{pde3a}), if and only if ~$%
U\left( x^{i}\right) ~$is an invariant of the infinitesimal generator, i.e. $%
X\left( U\right) =0$ $\ $and solves the PDE~(\ref{pde3}).
\end{definition}

By definition, a Lie symmetry maps a solution onto a solution. The action of
the point transformation with infinitesimal generator $X_{3}$ on a solution $%
u$ of (\ref{pde6}) is $u\left( \bar{t},\bar{x}\right) =u\left( t+\varepsilon
,x\right) $, that is, $X_{3}\bar{u}=0$ or $\bar{u}_{,t}=0$, hence the use of
Lie point symmetries gives a constraint equation. However, if a solution of
a PDE is already known we can apply the point transformation to obtain a
family of solutions. This new solutions will depend on at most as many new
parameters as there are in the symmetry transformation used.

\begin{theorem}
Assume that the PDE (\ref{pde3}) admits a Lie point symmetry and let $%
u=u\left( x^{i}\right) $ be a solution of (\ref{pde3}) which is not
invariant under this Lie symmetry. Then, under the transformation (\ref{pde1}%
)-(\ref{pde2}) the solution $u\left( x^{i}\right) $ defines a family of
solutions of the PDE.
\end{theorem}

For instance, the point transformation corresponding to the Lie symmetry
\begin{equation*}
X_{1}=t^{2}\partial _{t}+tx\partial _{x}-\frac{1}{2}\left( t+\frac{1}{2}%
x^{2}\right) u\partial _{u}
\end{equation*}%
of the heat equation (\ref{pde6}) is%
\begin{equation*}
\bar{t}=t\left( 1-\varepsilon t\right) ^{-1}~~,~~\bar{x}=x\left(
1-\varepsilon t\right) ^{-1}
\end{equation*}%
\begin{equation*}
\bar{u}=u\sqrt{1-\varepsilon t}e^{-\frac{\varepsilon x^{2}}{4\left(
1-\varepsilon t\right) }}.
\end{equation*}%
Then, if $\bar{u}_{c}=\bar{u}_{c}\left( \bar{t},\bar{x}\right) $ is a
solution of (\ref{pde6}), which is not an invariant solution, the function
\begin{equation*}
u_{c}\left( x,t\right) =\frac{1}{\sqrt{1-\varepsilon t}}e^{\frac{\varepsilon
x^{2}}{4\left( 1-\varepsilon t\right) }}\bar{u}_{c}\left( \frac{x}{%
1-\varepsilon t},\frac{t}{1-\varepsilon t}\right)
\end{equation*}%
is also a solution of (\ref{pde6}).

\section{Lie B\"{a}cklund symmetries}

So far we have considered point transformations which depend on the
variables of the base manifold only. However there exist point
transformations who are defined in the jet space and depend also on the
derivatives.

The infinitesimal transformation
\begin{eqnarray}
\bar{x}^{i} &=&x^{i}+\varepsilon \xi ^{i}\left(
x^{i},u,u_{,i},u_{,ij}...\right)  \label{LB.01} \\
\bar{u} &=&u+\varepsilon \eta \left( x^{i},u,u_{,i},u_{,ij}...\right)
\label{LB.02}
\end{eqnarray}%
is called \textbf{Lie B\"{a}cklund transformation}. This point
transformation does not concern the present work but for the completeness we
we give the basic definitions.

\begin{definition}
The generator $X=\xi \left( x^{i},u,u_{,i},u_{,ij}...\right) \partial
_{i}+\eta \left( x^{i},u,u_{,i},u_{,ij}...\right) \partial _{u}$ generates a
\textbf{Lie B\"{a}cklund symmetry} for the DE $H\left(
x^{i},u,u_{,i},u_{,ij}...\right) =0,$ if and only if there exist a function $%
\lambda \left( x^{i},u,u_{,i},u_{,ij}...\right) $ such as the following
condition holds.
\begin{equation*}
\left[ X,H\right] =\lambda H~\ ,~modH=0.
\end{equation*}
\end{definition}

It follows from the above definition that a Lie B\"{a}cklund symmetry
preserves the set of solutions $u$ of the DE, i.e. $X\left( u\right) =Cu,$
where $C$ is a constant. A special class of Lie B\"{a}cklund symmetries are
the contact symmetries. \textbf{Contact symmetry} is a Lie B\"{a}cklund
symmetry where the generator depends only on the first derivative of $u_{,i}$%
, i.e. the generator of transformation (\ref{LB.01})-(\ref{LB.01}) has the
form $X=\xi \left( x^{i},u,u_{,i}\right) \partial _{i}+\eta \left(
x^{i},u,u_{,i}\right) \partial _{u}$.

\begin{proposition}
The generator $\bar{X}=\left( \eta -\xi ^{k}u_{,k}\right) \partial _{u}$ is
the canonical form of the operator $X=\xi \partial _{i}+\eta \partial _{u}.$
\end{proposition}

The operator $D_{i}=\partial _{i}+u_{,i}\partial _{u}+u_{,ij}\partial
_{u_{i}}+...$ is always a Lie B\"{a}cklund symmetry (the trivial one$,$
since $D_{i}$ in the canonical form is $\bar{D}_{i}=\left(
u_{,i}-u_{,i}\right) \partial _{u}+...=0$), hence $f^{i}\left(
x^{i},u,u_{,i},u_{,ij}...\right) D_{i}$ is also a Lie B\"{a}cklund symmetry.
If $X=\xi \partial _{i}+\eta \partial _{u}$ is a Lie B\"{a}cklund symmetry,
then the generator $\bar{X},~$
\begin{equation*}
\bar{X}=X-fD_{i}=\left( \xi ^{k}-f^{k}\right) \partial _{k}+\left( \eta
-f^{k}u_{,k}\right) \partial _{u}+...
\end{equation*}%
is also a Lie B\"{a}cklund symmetry. Since $f^{k}$ is an arbitrary function
we set $f^{k}=\xi ^{k}$ and obtain%
\begin{equation*}
\bar{X}=\left( \eta -\xi ^{k}u_{,k}\right) \partial _{u}.
\end{equation*}

We can always absorb the term $\xi ^{k}u_{k}$ inside the $\eta $ and have
the final result that $\bar{X}=Z\left( x^{i},u,u_{,i},u_{,ij}...\right)
\partial _{u}$ is the generator of a Lie B\"{a}cklund symmetry.

For a PDE the generator of a Lie point symmetry is
\begin{equation*}
X=\xi \left( x^{i},u\right) \partial _{i}+\eta \left( x^{i},u\right)
\partial _{u}
\end{equation*}%
then in this case the vector $\bar{X}$ can be written as%
\begin{equation*}
\bar{X}=\left( \eta \left( x^{i},u\right) -\xi ^{k}\left( x^{i},u\right)
u_{,k}\right) \partial _{u}
\end{equation*}%
which is linear in the first derivative; that is, for PDEs, Lie symmetries
are equivalent to contact symmetries that are linear in the first derivative
$u_{k}$ (\emph{that property does not hold for ODEs}).~For example, the Lie
point symmetry $X=\frac{\partial }{\partial x}$ of a PDE is equivalent to
the Lie B\"{a}cklund symmetry $\bar{X}=-u_{,x}\frac{\partial }{\partial u}$.

\section{Noether point symmetries}

\label{NoetherS}

In the following sections we consider the Lie symmetries and the
conservation laws of systems admitting a Lagrangian function, i.e. systems
whose equations of motion follow from a variational principle (e.g. Hamilton
principle).

\subsection{Noether symmetries of ODEs}

In Analytical Mechanics the Lagrangian $L=L\left( t,x^{k},\dot{x}^{k}\right)
$ is a function describing the dynamics of a system. The equations of motion
of the dynamical system follow from the action of the Euler Lagrange vector $%
E_{i}$ on the Lagragian $L$, i.e.%
\begin{equation}
E_{i}\left( L\right) =0.  \label{Ns.01}
\end{equation}%
where the \textbf{Euler Lagrange vector}
\begin{equation}
E_{i}=\frac{d}{dt}\frac{\partial }{\partial \dot{x}^{i}}-\frac{\partial }{%
\partial x^{i}}
\end{equation}

If the Lagrangian is invariant under the action of the transformation (\ref%
{Ls.02})-(\ref{Ls.03}), then, it is easy to observe that the Euler Lagrange
equations (\ref{Ns.01}) are invariant under the transformation (\ref{Ls.02}%
)-(\ref{Ls.03}). E. Noether proved that if the action of a one parameter
point transformation leaves invariant the Euler Lagrange equations (\ref%
{Ns.01}) then there exists a conserved quantity corresponding to the point
transformation.

\begin{theorem}
\label{NotTheor1}Let
\begin{equation}
X=\xi \left( t,x^{k}\right) \partial _{t}+\eta ^{i}\left( t,x^{k}\right)
\partial _{i}  \label{Ns.02}
\end{equation}%
be the infinitesimal generator of transformation (\ref{Ls.02})-(\ref{Ls.03})
and
\begin{equation}
L=L\left( t,x^{k},\dot{x}^{k}\right)  \label{Ns.02a}
\end{equation}%
be a Lagrangian describing the dynamical system (\ref{Ns.01}). The action of
the transformation (\ref{Ls.02})-(\ref{Ls.03}) on (\ref{Ns.02a})\footnote{%
For the application of Noether theorem in higher order Lagrangians see \cite%
{Miron95}} leaves the Euler Lagrange equations (\ref{Ns.01}) invariant, if
and only if there exist a function $f=f\left( t,x^{k}\right) $ such that the
following condition holds
\begin{equation}
X^{\left[ 1\right] }L+L\frac{d\xi }{dt}=\frac{df}{dt}  \label{Ns.03}
\end{equation}%
where $X^{\left[ 1\right] }$ is the first prolongation of (\ref{Ns.02}).
\end{theorem}

If the generator (\ref{Ns.02}) satisfies (\ref{Ns.03}), the generator (\ref%
{Ns.02}) is a Noether symmetry of the dynamical system described by the
Lagrangian (\ref{Ns.02a}). Noether symmetries form a Lie algebra called the
\textbf{Noether algebra}. If the dynamical system (\ref{Ns.01}) admits Lie
symmetries which span a\ Lie algebra $G_{m}$ of dimension $m\succeq 1$ then
the Noether symmetries of (\ref{Ns.01}) form a Lie algebra $G_{h}~,~h\geq 0,$
which is a subalgebra of $G_{m},$ $G_{h}\subseteq G_{m}$.

\begin{theorem}
\label{NotTheor2}For any Noether point symmetry (\ref{Ns.02}) of the
Lagrangian (\ref{Ns.02a}) there corresponds a function $\phi \left( t,x^{k},%
\dot{x}^{k}\right) $
\begin{equation}
\phi =\xi \left( \dot{x}^{i}\frac{\partial L}{\partial \dot{x}^{i}}-L\right)
-\eta ^{i}\frac{\partial L}{\partial x^{i}}+f  \label{Ns.04}
\end{equation}%
which is a first integral i.e. $\frac{d\phi }{dt}=0$ of the equations of
motion. The function (\ref{Ns.04}) is called a \textbf{Noether integral}
(first integral) of (\ref{Ns.01}).
\end{theorem}

Since a\ Noether symmetry leaves the DEs (\ref{Ns.01}) invariant, it is also
a Lie symmetry of (\ref{Ns.01}) and from (\ref{Ls.05}) we can say that~(\ref%
{Ns.04}) satisfies the condition $X\left( \phi \right) =0$, that is, Noether
integrals are invariant functions of the Noether symmetry vector $X.~$

The existence and the number of Noether symmetries characterize a dynamical
system. If a dynamical system~(\ref{Ns.01}) on $n$ degrees of freedom admits
(at least) $n$ linear independent first integrals\footnote{%
Not necessarilly Noether point integrals.} $\phi _{J},~J=1...n,$ which are
in involution, i.e.
\begin{equation}
\left\{ \phi _{J},\phi _{K}\right\} =0  \label{Ns.LI}
\end{equation}%
where $\left\{ ~,~\right\} $ is the Poisson bracket, then the solution of
the dynamical system can be obtained by quadratures.

We calculate the Noether symmetries of the free particle and of the harmonic
oscillator.

\begin{example}
Find the Noether symmetries of the Lagrangian $L=\frac{1}{2}\dot{x}^{2}.$

Solution: The Lagrangian corresponds to the equation of motion of the free
particle~$\ddot{x}=0$. In example \ref{ExLs}, the Lie symmetries of the free
particle were calculated. The generic Lie symmetry vector is applied to
condition (\ref{Ns.03}) and it is found that the generic Noether symmetry of
the free particle is%
\begin{equation}
X=\left( a_{1}+2a_{4}t+a_{5}t^{2}\right) \partial _{t}+\left(
a_{2}+a_{3}t+a_{4}x+a_{5}tx\right) \partial _{x}.  \label{Ns.05}
\end{equation}%
From the generic vector (\ref{Ns.05}) and (\ref{Ns.04}), it is found that
the corresponding Noether integral is%
\begin{equation*}
\phi =a_{1}\dot{x}^{2}+a_{2}\dot{x}+a_{3}\left( t\dot{x}-x\right)
+a_{4}\left( t\dot{x}^{2}-x\dot{x}\right) +a_{5}\left( t^{2}\dot{x}^{2}-2tx%
\dot{x}+x^{2}\right)
\end{equation*}
\end{example}

\begin{example}
Find the Noether symmetries of the one dimensional harmonic oscillator with
Lagrangian
\begin{equation}
L=\frac{1}{2}\dot{x}^{2}-\frac{1}{2}x^{2}.  \label{oNs.01}
\end{equation}

Solution: To derive the Noether symmetries of (\ref{oNs.01}) we apply
theorem \ref{NotTheor1}. Let $X=\xi \left( t,x\right) \partial _{t}+\eta
\left( t,x\right) \partial _{x}$ be the infinitesimal generator. The first
prolongation $X^{\left[ 1\right] }~$is
\begin{equation*}
X^{\left[ 1\right] }=\xi \partial _{t}+\eta \partial _{x}+\eta _{\left[ 1%
\right] }\partial _{\dot{x}}.
\end{equation*}%
where $\eta _{\left[ 1\right] }=\left[ \eta _{,t}+\dot{x}\left( \eta
_{,x}-\xi _{,t}\right) -\dot{x}^{2}\xi _{,x}\right] .~$For the terms of~(\ref%
{Ns.03}) we have%
\begin{equation*}
\frac{df}{dt}=f_{,t}+\dot{x}f_{,x}.
\end{equation*}%
\begin{equation*}
X^{\left[ 1\right] }L=\dot{x}\eta _{,t}+\dot{x}^{2}\left( \eta _{,x}-\xi
_{,t}\right) -\dot{x}^{3}\xi _{,x}-x\eta
\end{equation*}%
\begin{equation*}
L\frac{d\xi }{dt}=\frac{1}{2}\dot{x}^{2}\xi _{,t}-\frac{1}{2}x^{2}\xi _{,t}+%
\frac{1}{2}\dot{x}^{3}\xi _{,x}-\frac{1}{2}x^{2}\dot{x}\xi _{,x}.
\end{equation*}%
Replacing in (\ref{Ns.03}) we find:%
\begin{eqnarray*}
0 &=&-\left[ x\eta +\frac{1}{2}x^{2}\xi _{,t}+f_{,t}\right] +\dot{x}\left[
\eta _{,t}-\frac{1}{2}x^{2}\xi _{,x}-f_{,x}\right] \\
&&+\dot{x}^{2}\left[ \eta _{,x}-\frac{1}{2}\xi _{,t}\right] +\dot{x}^{3}%
\left[ \xi _{,x}\right]
\end{eqnarray*}%
The determining equations are%
\begin{equation}
\xi _{,x}=0~~,~~\eta _{,x}-\frac{1}{2}\xi _{,t}=0  \label{oNs.02}
\end{equation}%
\begin{equation}
\eta _{,t}-\frac{1}{2}x^{2}\xi _{,x}-f_{,x}=0~~,~~x\eta +\frac{1}{2}x^{2}\xi
_{,t}+f_{,t}=0.  \label{oNs.03}
\end{equation}%
The solution of the system of equations (\ref{oNs.02})-(\ref{oNs.03}) gives
the generic Noether symmetry \cite{WW76}
\begin{eqnarray}
X &=&\left( a_{1}+a_{4}\cos 2t+a_{5}\sin 2t\right) \partial _{t}+  \notag \\
&&+\left( a_{2}\sin t+a_{3}\cos t-a_{4}x\sin 2t+a_{5}x\cos 2t\right)
\partial _{x}  \label{oNs.04}
\end{eqnarray}%
with the corresponding gauge function%
\begin{equation*}
f\left( t,x\right) =a_{2}x\cos t-a_{3}\sin t-a_{3}x^{2}\cos
2t-a_{5}x^{2}\sin 2t.
\end{equation*}
\end{example}

It should be underlined that the free particle and the harmonic oscillator
have the same number of Noether symmetries. It can be easily noticed that
the Lie algebras (\ref{Ns.05}) and (\ref{oNs.04}) are the same Lie algebra
in different representations.

\subsection{Noether point symmetries of PDEs}

In case of PDEs\
\begin{equation}
H\left( x^{i},u,u_{i},u_{ij}\right) =0  \label{Ns.06b}
\end{equation}%
which arise from a variational principle the following theorem holds.

\begin{theorem}
The action of the transformation (\ref{pde1})-(\ref{pde2}) on the Lagrangian
\begin{equation}
L_{P}=L_{P}\left( x^{k},u,u_{k}\right)  \label{Ns.06a}
\end{equation}%
leaves (\ref{Ns.06b}) invariant if there exists a vector field $%
F^{i}=F^{i}\left( x^{i},u\right) $ such that
\begin{equation*}
X^{\left[ 1\right] }L_{P}+L_{P}D_{i}\xi ^{i}=D_{i}F^{i}.
\end{equation*}%
The generator of the point transformation (\ref{pde1})-(\ref{pde2}) is
called \textbf{Noether symmetry}. The corresponding\textbf{\ Noether flow} is%
\begin{equation}
\Phi ^{i}=\xi ^{k}\left( u_{k}\frac{\partial L}{\partial u_{i}}-L\right)
-\eta \frac{\partial L}{\partial u_{i}}+F^{i}  \label{Ns.06}
\end{equation}%
and satisfies the condition $D_{i}\Phi ^{i}.$
\end{theorem}

The implementation of Noether flows in PDEs differs from that in the ODEs.
Conservation flow (\ref{Ns.06}) is used to reduce the order of (\ref{Ns.06b}%
) by defining a new dependent variable $v\left( x^{i}\right) $. \ It has
been proved that the solution of the system
\begin{equation}
v_{,i}\left( x^{k}\right) =\Phi _{i}\left( x^{k},u,u_{k}\right)
\label{Ns.07}
\end{equation}%
is also a solution of (\ref{Ns.06b}). Furthermore, it is feasible new
symmetries to arise from the system (\ref{Ns.07}) which are not symmetries
of (\ref{Ns.06b}). These new symmetries have been called \textbf{potential
symmetries} \cite{Bluman88}.

\section{Collineations of Riemannian manifolds}

\label{geomObj}

Previously, we studied the case when a function is invariant under a point
transformation. In the following subsections we consider the cases in which
a function is invariant under a point transformation. In order to do this we
shall need the concept of the Lie derivative. In particular, the geometric
objects which will be considered the metric tensor $g_{ij}$ and the
connection coefficients $\Gamma _{jk}^{i}$ in a Riemannian space.

\begin{definition}
Consider an n dimensional space $A_{n}$ of class $C^{p}$. An object is
called a\textbf{\ geometric object }of class $r~\left( r\leq p\right) $ if
it has the following properties $.$

i) In each coordinate system $\left\{ x^{i}\right\} ,$ it has a well
determined set of components $\Omega ^{a}\left( x^{k}\right) $.

ii) Under a coordinate transformation $x^{i^{\prime }}=J^{i}\left(
x^{k}\right) $ the new components $\Omega ^{a^{\prime }}$ of the object in
the new coordinates $\left\{ x^{i^{\prime }}\right\} $ are represented as
well determined functions of class $r^{\prime }=p-r$ of the old components $%
\Omega ^{a}$ in the old coordinates $\left\{ x^{i}\right\} $, of the
functions $J^{i}$ and of their sth derivatives $\left( 1\leq s\leq p\right) $%
, that is, the new components $\Omega ^{i^{\prime }}$ of the object can be
represented by equations of the form%
\begin{equation}
\Omega ^{a^{\prime }}=\Phi ^{a}\left( \Omega ^{k},x^{k},x^{k^{\prime
}}\right) .  \label{go.01}
\end{equation}

iii) The functions $\Phi ^{a}$ have the group properties, that is, they
satisfy \ the following relations%
\begin{equation}
\Phi ^{a}\left( \Phi \left( \Omega ,x^{k},x^{k^{\prime }}\right)
,x^{k},x^{k^{\prime }}\right) =\Phi ^{a}\left( \Omega ,x^{k},x^{k^{\prime
}}\right)
\end{equation}%
\begin{equation}
\Phi ^{a}\left( F\left( \Omega ,x^{k},x^{k^{\prime }}\right) ,x^{k^{\prime
}},x^{k}\right) =\Omega ^{a}
\end{equation}
\end{definition}

The coordinate transformation law $\Phi \left( \Omega ,x^{k},x^{k^{\prime
}}\right) $ characterizes the geometric object. In case that the function $%
\Phi $ contains only $\Omega $ and the partial derivatives of $J^{k}$ with
respect to $x^{k},\,\ $the geometric object is said to be a \textbf{%
differential geometric object}.

Furthermore, we say that a geometric object is linear if for the
transformation law $\Phi \left( \Omega ,x^{k},x^{k^{\prime }}\right) $ it
holds%
\begin{equation}
\Phi ^{a}\left( \Omega ,x^{k},x^{k^{\prime }}\right) =J_{b}^{a}\left(
x^{k},x^{k^{\prime }}\right) \Omega ^{b}+C\left( x^{k},x^{k^{\prime
}}\right) .
\end{equation}%
When the transformation law is
\begin{equation}
\Phi ^{a}\left( \Omega ,x^{k},x^{k^{\prime }}\right) =J_{b}^{a}\left(
x^{k},x^{k^{\prime }}\right) \Omega ^{b}
\end{equation}
we say that the geometric object is a \textbf{linear homogeneous geometric
object}. One important calss of linear homogeneous geometric objects are the%
\textbf{\ tensors.}

\subsection{Lie Derivative}

Let $\Omega $ be a geometric object in $A_{n}$ with transformation law (\ref%
{go.01}) and consider the infinitesimal transformation
\begin{equation}
\bar{x}^{i}=x^{i}+\varepsilon \xi ^{i}\left( x^{k}\right) .  \label{go.02}
\end{equation}%
where $\xi =\xi ^{i}\left( x^{k}\right) $ are the components of the
infinitesimal generator. From the transformation law (\ref{go.01}), the
geometric object in the coordinate system $x^{i^{\prime }}=\bar{x}^{i}$ is $%
\Phi \left( \Omega ^{k},x^{k},x^{k^{\prime }}\right) $. The Lie derivative
of the geometric object $\Omega $ with respect to $\xi $ is defined as
follows
\begin{equation}
L_{X}\Omega =\lim_{\varepsilon \rightarrow 0}\frac{1}{\varepsilon }\left[
\Phi \left( \Omega ^{k},x^{k},x^{k^{\prime }}\right) -\Omega \right] .
\label{go.03}
\end{equation}

In order for the Lie derivative of a geometric object to be again a
geometric object (not necessarily of the same type)\ it is necessary and
sufficient that the geometric object be linear~\cite{Yano}. \

By definition, the Lie derivative of a geometric object depends on the
transformation law. The transformation law (\ref{go.01}) for functions is $%
f^{\prime }\left( x^{i^{\prime }}\right) =f\left( x^{i^{\prime }}\right) $,
hence, under the point transformation (\ref{go.02}) we have%
\begin{equation*}
f^{\prime }\left( x^{i^{\prime }}\right) =f\left( x^{i}+\varepsilon \xi
^{i}\right) =f\left( x^{i}\right) +\varepsilon f_{,i}\xi ^{i}+O\left(
\varepsilon ^{2}\right) .
\end{equation*}%
From (\ref{go.03}) the Lie derivative of functions $f$ along the vector
field $\xi $ is computed to be%
\begin{equation*}
L_{\xi }f=f_{,i}\xi ^{i}=\xi \left( f\right) .
\end{equation*}

The transformation law for a tensor field~$T$ of rank $\left( r,s\right) $ is%
\begin{equation*}
T_{~~j_{i}^{\prime }...j_{s}^{\prime }}^{i_{1}^{\prime }...i_{r}^{\prime
}}=J_{i_{1}}^{i_{1}^{\prime }}...J_{i_{r}}^{i_{r}^{\prime }}J_{j_{1}^{\prime
}}^{j_{1}}...J_{j_{s}}^{j_{s}^{\prime }}T_{~~j_{i}...j_{s}}^{i_{1}...i_{r}}
\end{equation*}%
where $J$ is the Jacobian of the transformation. Thus, from (\ref{go.03}),
the Lie derivative of $T$ with respect to $\xi $ is%
\begin{eqnarray}
L_{\xi }T_{~~j_{i}...j_{s}}^{i_{1}...i_{r}} &=&\xi
^{k}T_{~~j_{i}...j_{s,k}}^{i_{1}...i_{r}}-T_{~~j_{i}...j_{s}}^{m...i_{r}}\xi
_{,m}^{i_{1}}-T_{~~j_{i}...j_{s}}^{i_{1}m...i_{r}}\xi _{m}^{i_{2}}+...
\notag \\
&&...+T_{~~m...j_{s}}^{i_{1}...i_{r}}\xi
_{,j_{1}}^{m}+T_{~~j_{i}m...j_{s}}^{i_{1}...i_{r}}\xi _{j_{2}}^{m}+....
\label{go.03a}
\end{eqnarray}%
In case of vector fields $X$ and 1-forms $\omega $ expression (\ref{go.03a})
gives%
\begin{equation}
L_{\xi }X=\xi ^{k}X_{,k}^{i}-X^{k}\xi _{,k}^{i}=\left[ \xi ,X\right]
\end{equation}%
\begin{equation}
L_{\xi }\omega =\xi ^{k}\omega _{i,k}-\omega _{k}\xi _{,i}^{k}.
\end{equation}%
For a second order tensor $T_{ij}$, expression (\ref{go.03a}) becomes%
\begin{equation}
L_{\xi }T_{ij}=T_{ij,k}\xi ^{k}+T_{ik}\xi _{,j}^{k}+T_{kj}\xi _{,i}^{k}.
\label{go.03b}
\end{equation}

The connection coefficients $\Gamma _{jk}^{i}$ are linear differential
geometric objects with transformation law%
\begin{equation}
\Gamma _{j^{\prime }k^{\prime }}^{i^{\prime }}=J_{i}^{i^{\prime
}}J_{j^{\prime }}^{j}J_{k^{\prime }}^{k}\Gamma _{jk}^{i}+\frac{\partial
x^{i^{\prime }}}{\partial x^{r}}\frac{\partial ^{2}x^{r}}{\partial
x^{j^{\prime }}\partial x^{k^{\prime }}}.
\end{equation}%
Connection coefficients have different transformation law from tensor
fields, that is, the Lie derivative $L_{\xi }\Gamma _{jk}^{i}$ will be
different from that of tensors. Applying (\ref{go.03}),we find that the Lie
derivative $L_{\xi }\Gamma _{jk}^{i}$ is expressed as follows%
\begin{equation}
L_{\xi }\Gamma _{jk}^{i}=\xi _{,jk}^{i}+\Gamma _{jk,r}^{i}\xi ^{r}-\xi
_{,r}^{i}\Gamma _{jk}^{r}+\xi _{,j}^{s}\Gamma _{sk}^{i}+\xi _{,k}^{s}\Gamma
_{js}^{i}.  \label{go.04}
\end{equation}%
In case of symmetric connection $\Gamma _{jk}^{i}=\Gamma _{kj}^{i},$
condition (\ref{go.04}) can be written in the equivalent form%
\begin{equation}
L_{\xi }\Gamma _{jk}^{i}=\xi _{;jk}^{i}-R_{jkl}^{i}\xi ^{l}
\end{equation}%
where $R_{jkl}^{i}$ is the curvature tensor and the semicolon~\thinspace $%
";" $ means covariant derivative.

\subsubsection{Collineations}

In section \ref{Fun} we gave the conditions under which a function is
invariant under a one parameter point transformation. Similarly for linear
(homogeneous) differential geometric objects there is the following
definition.

\begin{definition}
A linear differential geometric object $~\Omega \left( x^{i}\right) $ is
invariant under a one parameter point transformation
\begin{equation}
\bar{x}^{i}=\bar{x}^{i}\left( x^{k},\varepsilon \right)  \label{go.05}
\end{equation}%
if and only if $~\bar{\Omega}\left( \bar{x}^{i}\right) =\Omega \left(
x^{i}\right) $ at all points where the one parameter point transformation
acts. Equivalently, the Lie derivative of $~$the geometric object $\Omega $
with respect to the infinitesimal generator of (\ref{go.05}) vanishes, that
is, $L_{\xi }\Omega =0.$
\end{definition}

A direct result which arises from the definition of the Lie derivative and
the transformation law of linear differential geometric objects is that if a
linear differential geometric object $\Omega $ is invariant under the
transformation (\ref{go.05}), then there exist a coordinate system with
respect to which the components of $\Omega $ are independent of one of the
coordinates.

One can generalize the concept of symmetry in the sense that one does not
require the Lie derivative to be equal to zero, but to another tensor. That
is, the Lie derivative of the linear differential geometric object $\Omega $
with respect to the infinitesimal generator $\xi $ is%
\begin{equation*}
L_{\xi }\Omega =\Psi
\end{equation*}%
where $\Psi $ has the same numbers of components and symmetries of the
indices with $\Omega .$ In this case the infinitesimal generator $\xi $ is
said to be a \textbf{collineation }of $\Omega ,$ the type of collineations
being defined by $\Psi $. Collineations are a powerful tool in the study of
the geometric properties of Riemannian manifolds.

In Riemannian geometry, the geometric object $\Omega $ can be the metric
tensor $g_{ij}$ or any other geometric object defined from it (e.g. the
connection coefficients).

\begin{definition}
All collineations involving geometric objects $\Omega $ derived from the
metric $g_{ij}$ of a Riemannian space shall be called \textbf{geometric
collineations}. In particular, the collineation defined by the metric $%
L_{\xi }g_{ij}$ is called the \textbf{generic collineation} because any
other geometric collineation can be written in terms of it. Furthermore, the
geometric collineations can be written in terms of the irreducible parts $%
\psi ,H_{ij}$ as follows%
\begin{equation}
L_{\xi }g_{ij}=2\psi g_{ij}+2H_{ij}  \label{cc.001}
\end{equation}%
where the function $\psi $ is called the \textbf{conformal factor} and $%
H_{ij}$ is a \textbf{symmetric traceless tensor}.
\end{definition}

The role of the quantities $\psi ,~H_{ij}$ is important because they can be
used as the variables in terms of which one can study any geometric
collineation. To do that, one has to express the Lie derivative of any
metric tensor in terms of the generic symmetry variables $\psi ,~H_{ij}$ and
their derivatives.

In the following we shall be interested in geometric collineations,
particularly in the collineations of the metric and of the connection
coefficients of a Riemannian space.

\subsection{Motions of Riemannian spaces}

Consider an n dimensional Riemannian space $V^{n}$ with line element%
\begin{equation}
ds^{2}=g_{ij}dx^{i}dx^{j}.  \label{mr.01}
\end{equation}%
where $g_{ij}$ is the metric tensor.

\begin{definition}
The point transformation (\ref{go.05}) is called a motion of $V^{n},$~if and
only if the line element is invariant under the action of (\ref{go.05}).
Equivalently, the Lie derivative of the metric tensor $g_{ij}$ with respect
to the infinitesimal generator $\xi $ of (\ref{go.05}) vanishes, i.e.%
\begin{equation}
L_{\xi }g_{ij}=0.  \label{mr.02}
\end{equation}
\end{definition}

The point transformations (\ref{go.05}) which are motions of $V^{n}$ form a
group named the \textbf{group of motions}. Since $g_{ij}$ is a metric,
condition (\ref{mr.02}) can be written in the equivalent form%
\begin{equation}
L_{\xi }g_{ij}=2\xi _{\left( i;j\right) }=0.  \label{mr.03}
\end{equation}%
This equation is known as \textbf{Killing equation} and $\xi $ is called an
\textbf{isometry }or \textbf{Killing Vector }(KV). The KVs of a metric form
a Lie algebra, which is called the \textbf{Killing algebra}.

Motions are important in physics. For instance, the Euclidian space\ admits
as motions the group of translations and the group of rotations~$T\left(
3\right) \otimes so\left( 3\right) $ and this implies the conservation of
linear and angular momentum respectively. As a second example in Cosmology
the assumption that the universe is isotropic and homogeneous about all
points leads to the Friedmann~Robertson~Walker (FRW) spacetime.

The maximum dimension of the Killing algebra that $V^{n}$ can admit is given
in the following theorem.

\begin{theorem}
\label{maxss}If an n dimensional Riemannian space $V^{n}$ admits a Killing
algebra $G_{KV},$ then, ~$0\leq \dim G_{KV}\leq \frac{1}{2}n\left(
n+1\right) .$
\end{theorem}

A Riemannian space which admits a Killing algebra of dimension $\frac{1}{2}%
n\left( n+1\right) $ is~called a\textbf{\ maximally symmetric space}. For
example, the Euclidian space $E^{3}$ and the Minkowski spacetime $M^{4}~$are
maximally symmetric spaces.

A special class of KVs are the gradient KVs. A KV $\xi $ is called \textbf{%
gradient} iff $\xi _{i;j}=0$,~that is,
\begin{equation*}
\xi _{\left( i;j\right) }=0~~\text{and~~}\xi _{\left[ i;j\right] }=0.
\end{equation*}%
For every gradient KV $\xi ,~$there exists a function $S$ so $%
S_{,k}g^{ik}=\xi ^{i}$ and $S_{;ij}=0.$

\begin{theorem}
\label{decSp}If $V^{n}$ admits $p~$gradient KVs $\left( \text{where~}p\leq
n\right) $ then, $V^{n}$ is called a $p~$\textbf{decomposable space }and in
this case there exists a coordinate system in which the line element (\ref%
{mr.01}) can be written in the form%
\begin{equation*}
ds^{2}=M_{\alpha \beta }dz^{\alpha }dz^{\beta }+h^{AB}\left( y^{A}\right)
dx_{A}dx_{B}
\end{equation*}%
where $\alpha ,\beta =1,2..,p$, $A,B=p+1,...n$ and $M_{a\beta }=diag\left(
c_{1},c_{2},...,c_{p}\right) $,~$c_{1},c_{2},...,c_{p}$ are constants.
\end{theorem}

\begin{example}
\label{Kvs2dsphere}Compute the KVs of the Euclidian sphere with line element%
\begin{equation}
ds^{2}=d\phi ^{2}+\sin \phi ^{2}d\theta ^{2}.  \label{mr.04}
\end{equation}

Solution: In order to find the KVs we solve the Killing equation (\ref{mr.03}%
).\ This leads to the system of equations%
\begin{equation*}
\xi _{,\theta }^{\theta }=0
\end{equation*}%
\begin{equation*}
2\xi _{,\phi }^{\phi }+2\xi ^{\theta }\sin \theta \cos \theta =0
\end{equation*}%
\begin{equation*}
\xi _{,\phi }^{\theta }+\xi _{,\theta }^{\phi }-2\xi ^{\phi }\cot \theta =0.
\end{equation*}%
whose solution are the elements of the $so\left( 3\right) $ Lie algebra (see
example \ref{exso3}).
\end{example}

The 2D Euclidian sphere (\ref{mr.04}) admits a three dimensional Killing
algebra, hence, it is a maximally symmetric space. Moreover, all spaces of
constant curvature are maximally symmetric spaces.

\subsubsection{Conformal motions}

When the point transformation (\ref{go.05}) does not change the angle
between two directions at a point, the transformation (\ref{go.05}) is
called a \textbf{conformal motion}. Technically we have the following
definition.

\begin{definition}
The infinitesimal generator $\xi $ of the point transformation (\ref{go.05})
is called \textbf{Conformal Killing Vector} (CKV) if the Lie derivative of
the metric $g_{ij}$ with respect to $\xi $ is a multiple of $g_{ij}$. That
is if,
\begin{equation}
L_{\xi }g_{ij}=2\psi g_{ij}
\end{equation}%
where $\psi =\frac{1}{n}\xi _{;k}^{k}$. In the case where $\psi _{;ij}=0,$ $%
\xi $ is a special CKV (sp.CKV)\footnote{%
For the conformal factor of a sp.CKV holds $\psi _{;ij}=0,$ that is, $\psi
_{,i}$ is a gradient KV. A Riemannian space admits a sp.CKV if and only if
admits a gradient KV and a gradient HV \cite{HallspCKV}.} and if $\psi =$%
constant, $\xi $ is a Homothetic Killing Vector (HV).
\end{definition}

The CKVs\ of a metric form a Lie algebra, which is called the c\textbf{%
onformal algebra}, $G_{CV}$. Obviously the KVs and the homothetic vector are
elements of the conformal algebra $G_{CV}$. If $G_{HV}$ is the algebra of
HVs (including the algebra $G_{KV}$ of KVs), then the following theorem
applies.

\begin{theorem}
Let $V^{n}$ be an n dimensional Riemannian space,~$n\geq 2,$ which admits a
conformal algebra $G_{CK}$, a homothetic algebra~$G_{HV}$ and a Killing
algebra~$G_{KV},$ then

i) $G_{KV}\subseteq G_{HV}\subseteq G_{CV}$.

ii) for arbitrary $n$, $G_{H-K}=G_{HV}-G_{HV}\cap G_{KV},$ then $0\leq \dim
G_{H-K}\leq 1$; that is, $V^{n}$ admits at most one HV.

iii) $V^{2}$ admits an infinite dimentional conformal algebra $G_{CV}$,

iv) for $n>2$, $0\leq \dim G_{CV}\leq \frac{1}{2}\left( n+1\right) \left(
n+2\right) \,.$
\end{theorem}

In the following by the term \emph{proper conformal algebra} we mean the
algebra of CKVs which are not HV/KVs. A generalization of theorem \ref{decSp}
for spaces which admit a gradient CKV is the following.

\begin{theorem}
If $V^{n}$ admits a gradient CKV then there exists a coordinate system in
which the line element can be written as follows%
\begin{equation*}
ds^{2}=f^{2}\left( x^{n}\right) \left[ dx^{n}+h^{AB}\left( x^{A}\right)
dx^{A}dx^{B}\right]
\end{equation*}%
where $A,B=1,2,...n-1$. In these coordinates the CKV is $X_{C}=\partial
_{x^{n}}$ with conformal factor $\psi =\frac{f_{,x^{n}}}{f}$. In the case in
which $f\left( x^{n}\right) =\exp \left( x^{n}\right) ,$ then $\psi =1$,
hence, $X_{C}$ becomes HV and if $f_{,x^{n}}=0$, $X_{C}$ becomes KV.
\end{theorem}

Two metrics $g_{ij},~\bar{g}_{ij}$ are \textbf{conformally related} if $\bar{%
g}_{ij}=N^{2}g_{ij}$ where the function ~$N^{2}$ is the \textbf{conformal
factor}. If $\xi $ is a CKV\ of the metric $\bar{g}_{ij}$ so that $L_{\xi }%
\bar{g}_{ij}=2\bar{\psi}\bar{g}_{ij},$ then $\xi $ is also a CKV of the
metric $g_{ij}$, that is $L_{\xi }g_{ij}=2\psi g_{ij}$ with conformal factor
$\psi \left( x^{k}\right) $
\begin{equation}
\psi =\bar{\psi}N^{2}-NN_{,i}\xi ^{i}.
\end{equation}%
The last relation implies that two conformally related metrics have the same
conformal algebra, but with different subalgebras; that is, a KV for one may
be proper CKV for the other. This is an important observation which shall be
useful later. A special class of conformally related metrics are the
conformally flat metrics. A space $V^{n}$ is conformally flat if the metric $%
g_{ij}$ of $V^{n}$ satisfies the relation $g_{ij}=N^{2}s_{ij}~$where $%
s_{ij}\,$~is the metric of a flat space which has the same signature with $%
g_{ij}$. \newline

For conformally flat spaces the following proposition aplies.

\begin{proposition}
Let $V^{n}$ be an n-dimensional Riemannian space, $n\geq 2,$

i) If $V^{n},~n>2$ is conformally flat then $V^{n}$ admits a conformal
algebra of dimension $\frac{1}{2}\left( n+1\right) \left( n+2\right) ~$.

ii) A three dimensional space is conformally flat if the Cotton-York tensor%
\begin{equation*}
C^{ij}=2\varepsilon ^{ikr}\left( R_{k}^{j}-\frac{1}{4}\delta _{k}^{j}\right)
_{;r}\newline
\end{equation*}%
vanishes.

iii) $V^{n},$ $n\geq 4\ $\ is conformally flat if the Weyl tensor%
\begin{equation*}
R_{ijkr}=C_{ijkr}+\frac{2}{n-2}\left( g_{i[k}R_{r]j}-g_{j[k}R_{r]i}\right) -%
\frac{2}{\left( n-1\right) \left( n-2\right) }Rg_{i[k}g_{r]j}
\end{equation*}

vanishes.

iv) If $V^{n},$ $n>3$ is a maximally symmetric space then $V^{n}$ is
conformally flat.

v) All two dimensional spaces are conformally flat .
\end{proposition}

A\ result which will be used in subsequent sections is the following.

\begin{example}[The conformal algebra of the flat space.]
\label{ExCAflat}Consider a flat space of dimension $n>2~$ with metric
\begin{equation*}
ds^{2}=\varepsilon dt^{2}+\delta _{AB}dy^{A}dy^{B}~,~\varepsilon =\pm 1.
\end{equation*}%
The conformal algebra of the space consists of the following vectors\newline
$n-$ gradient KVs,
\begin{equation*}
K_{G}^{1}=\partial _{t}~,~K_{G}^{A}=\partial _{A}\newline
\end{equation*}%
$\frac{n\left( n-1\right) }{2}~$ non~gradient KVs (rotations). $~$%
\begin{equation*}
X_{R}^{1A}=y^{A}\partial _{t}-\varepsilon t\partial _{A}
\end{equation*}%
\begin{equation*}
~X_{R}^{AB}=y^{B}\partial _{A}-y^{A}\partial _{B}
\end{equation*}%
one gradient HV$~$%
\begin{equation*}
H=t\partial _{t}+\sum\limits_{A}y^{A}\partial _{A}
\end{equation*}%
$n~$special CKVs%
\begin{equation*}
X_{C}^{1}=\frac{1}{2}\left( t^{2}-\varepsilon \sum\limits_{A}\left(
y^{A}\right) ^{2}\right) \partial _{t}+t\sum\limits_{A}y^{A}\partial _{A}
\end{equation*}%
\begin{equation*}
X_{C}^{A}=ty^{A}\partial _{t}+\frac{1}{2}\left( \varepsilon t^{2}+\left(
y^{A}\right) ^{2}-\sum\limits_{B\neq A}\left( y^{B}\right) ^{2}\right)
\partial _{A}+y^{A}\sum\limits_{B\neq A}y^{B}\partial _{B}
\end{equation*}%
where $y^{A}=1...n-1~$with conformal factor $\psi _{C}^{1}=t~$and $\psi
_{C}^{A}=y^{A}.$

For $n>2$ the flat space does not admit proper (non special) CKVs.

For $n=2$ the vector field%
\begin{equation*}
X=\left[ f_{1}\left( t+i\sqrt{\varepsilon }x\right) -f_{2}\left( t-i\sqrt{%
\varepsilon }x\right) +c_{0}\right] \partial _{t}+i\sqrt{\varepsilon }\left[
f_{1}\left( t+i\sqrt{\varepsilon }x\right) +f_{2}\left( t-i\sqrt{\varepsilon
}x\right) \right] \partial _{x}
\end{equation*}%
is the generic CKV, that is, it includes the KVs, the HV, and the CKVs.
\end{example}

\subsection{Symmetries of the connection}

Let $\xi $ be the generator of an infinitesimal transformation\ of (\ref%
{go.05}). In a Riemannian space with metric $g_{ij}$ we have the identity%
\begin{equation}
L_{\xi }\Gamma _{jk}^{i}=g^{ir}\left[ \left( L_{\xi }g_{rk}\right)
_{;j}+\left( L_{\xi }g_{rj}\right) _{;k}-\left( L_{\xi }g_{jk}\right) _{;r}%
\right] .  \label{scc.01}
\end{equation}

If $\xi $ is a HV or KV then from (\ref{scc.01}) follows that $L_{\xi
}\Gamma _{jk}^{i}$ vanishes, which implies that the $\Gamma _{jk}^{i}$ are
invariant under the action of transformation (\ref{go.05}).

\begin{definition}
The infinitesimal generator $\xi $ of the point transformation (\ref{go.05})
caries a geodesic into a geodesic and also preserves the affine parameter
iff the Lie derivative of connection coefficients $\Gamma _{jk}^{i}$ with
respect to $\xi $ vanishes, that is iff%
\begin{equation}
L_{\xi }\Gamma _{jk}^{i}=0  \label{scc.02}
\end{equation}%
The infinitesimal generator $\xi $ is called an \textbf{Affine Killing}
vector or \textbf{Affine collineation} (AC).
\end{definition}

ACs of $V^{n}$ form a Lie algebra, which is called \textbf{Affine algebra}, $%
G_{AC}$.\ Obviously the homothetic algebra\footnote{%
Note that the proper CKVs do not satisfy condition (\ref{scc.02}).} $%
G_{HV},~ $is a subalgebra of $G_{AC}$, $G_{HV}\subseteq G_{AC}.~$We shall
say that a spacetime admits proper ACs when $\dim G_{HV}\prec \dim G_{AC}$.

For instance, in the case of flat space condition (\ref{scc.02}) becomes~$%
\xi _{,jk}^{i}=0,~$therefore, the general solution is~$\xi
^{i}=A_{j}^{i}x^{j}+B^{i}~$where $A_{j}^{i},B^{i}$ are~$n\left( n+1\right) ~$
constants. Therefore the flat space admits a $n\left( n+1\right) $
dimensional Affine algebra (including the KVs and the HV). We have the
inverse result.

\begin{theorem}
If an n dimensional Riemannian space $V^{n}$ admits an Affine algebra $%
G_{AC} $ and\ $\dim G_{AC}=n\left( n+1\right) $, then $V^{n}$ is a flat
space.
\end{theorem}

A generalization of affine symmetry which is of interest is the Projective
collineation.

\begin{definition}
The infinitesimal generator $\xi $ of the point transformation (\ref{go.05})
is called a \textbf{Projective Collineation} (PC) if there exists a function%
\footnote{%
In general, $\xi $ is PC if there exists a one form $\omega _{i}$ such~that$%
~L_{\xi }\Gamma _{jk}^{i}=2\omega _{(j}\delta _{k)}^{i}.$ \ In \ a
Riemannian space, the one form $\omega _{i}$ is necessarily closed, that is,
there exist a function $\phi $ such as $\omega _{i}=\phi _{,i}.$} $\phi
\left( x^{k}\right) $ such that
\begin{equation}
L_{\xi }\Gamma _{jk}^{i}=\phi _{,j}\delta _{k}^{i}+\phi _{,k}\delta _{j}^{i}
\label{scc.03}
\end{equation}%
or equivalently%
\begin{equation*}
\xi _{\left( i;j\right) ;k}=2g_{ij}\phi _{,k}+2g_{k(i}\phi _{,j).}
\end{equation*}%
The function $\phi $ is called the \textbf{projective function}. If the
projective function satisfies the condition $\phi _{;ij}=0,$ then we say
that $\xi ~$is a \textbf{special PC} (sp.PC). Projective transformations
transform the system of geodesics (auto parallel curves) of $V^{n}$ into the
same system but they do not preserve the affine parameter.
\end{definition}

The PCs of $V^{n}$ form a Lie algebra which is called \textbf{Projective
algebra}, $G_{PC}$. The affine algebra $G_{AC}$ is a subalgebra of $G_{PC}$,
$G_{AC}\subseteq G_{PC}.$

In flat space condition (\ref{scc.03}) gives the generic projective
collineation
\begin{equation*}
\xi _{i}=A_{ij}x^{j}+\left( B_{j}x^{j}\right) x_{i}+C_{i}
\end{equation*}%
where $A_{ij},B_{j}$,$C_{i}$ are $n\left( n+2\right) $ constants.

\begin{theorem}
If an n dimensional Riemannian space $V^{n}~$admits a projective\ algebra $%
G_{PC},$ then$~\dim G_{PC}\leq n(n+2$). In case $\dim G_{PC}=n\left(
n+2\right) $ $\ $then $V^{n}$ is a maximally symmetric space \cite{Barnes}.
\end{theorem}

Some useful propositions for the existence of sp. PCs are the following \cite%
{Barnes,HalldaCosta,HallR}.

\begin{proposition}
Let $V^{n}$ be an n dimensional Riemannian space, then

i) If $V^{n}$ admits a $p\leq n$ dimensional Lie algebra of sp.PCs then also
admits $p$ gradient KVs and a gradient HV and if $p=n$ the space is flat
(the reverse also holds true).

ii) A maximally symmetric space which admits a proper AC or a sp.PC is a
flat space.
\end{proposition}

A Riemannian space is possible to admit more general collineations, e.g.
Curvature collineations. A full classification of the collineations of a
Riemannian space (with definite or indefinite metric) can be found in \cite%
{Katz69}. A summary of the above definitiosn is given in Table \ref{Table1}.

\begin{table}[tcp] \centering%
\caption{Collineations of a Riemannian space }%
\begin{tabular}{ccc}
\hline\hline
\textbf{Collineation~}$\mathcal{L}_{\xi }\mathbf{A}=\mathbf{B}$ & \textbf{\ }%
$\mathbf{A}$ & $\mathbf{B}$ \\ \hline
Killing Vector (KV) & $g_{ij}$ & $0$ \\
Homothetic vector (HV) & $g_{ij}$ & $2\psi g_{ij},~\psi _{,i}=0$ \\
Conformal Killing vector (CKV) & $g_{ij}$ & $2\psi g_{ij},~\psi ,_{i}\neq 0$
\\
Affine Collineation (AC) & $\Gamma _{jk}^{i}$ & $0$ \\
Proj. Collineations (PC) & $\Gamma _{jk}^{i}$ & $2\phi _{(,j}\delta
_{k)}^{i},$ $\phi ,_{i}\neq 0$ \\
Sp. Proj. collineation (sp.PC) & $\Gamma _{jk}^{i}$ & $2\phi _{(,j}\delta
_{k)}^{i},$ $\phi _{;jk}=0$ \\ \hline\hline
\end{tabular}%
\label{Table1}%
\end{table}%

We note that the Lie symmetries $\left( sl\left( 3,R\right) \right) $\ of
the free particle form the projective algebra of the two dimensional
Euclidian space. Therefore, the natural question which arises is the
following:

\begin{center}
\textit{Is there any connection between the Lie symmetries of differential
equations of second order and collineations? }
\end{center}

In the following chapters this will be confirmed and will be used to apply
the Lie symmetries of DEs in various areas of Geometry and Physics.

\part{Symmetries of ODEs}

\chapter{Lie symmetries of geodesic equations\label{LieSymGECh}}

\section{Introduction}

In a Riemannian space the affinely parameterized geodesics are determined
uniquely by the metric. Therefore one should expect a close relation between
the geodesics as a set of homogeneous ordinary differential equations (ODEs)
linear in the highest order term and quadratic non-linear in first order
terms, and the metric as a second order symmetric tensor. A system of such
ODEs is characterized (perhaps not fully) by its Lie symmetries and a metric
by its collineations. Therefore it is reasonable to expect that the Lie
symmetries of the system of geodesics of a metric will be closely related
with the collineations of the metric. That such a relation exists it is easy
to see by the following simple example. Consider on the Euclidian plane a
family of straight lines parallel to the $x-$axis. These curves can be
considered either as the integral curves of the ODE\ $\frac{d^{2}y}{dx^{2}}%
=0 $ or as the geodesics of the Euclidian metric $dx^{2}+dy^{2}$.
Subsequently consider a symmetry operation defined by a reshuffling of these
lines without preserving necessarily their parametrization. According to the
first interpretation this symmetry operation is a Lie symmetry of the ODE\ \
$\ddot{y}=0$ and according to the second interpretation it is a (special)
projective symmetry of the Euclidian two dimensional metric.

What has been said for a Riemannian space can be generalized to a space in
which there is only a linear connection. In this case the geodesics are
called autoparallels (or paths) and they comprise again a system of ODEs
linear in the highest order term and quadratic non-linear in the first order
terms. In this case one is looking for relations between the Lie symmetries
of the autoparallels and the projective or affine collineations of the
connection.

The above matters have been discussed extensively in a series of interesting
papers. Classic is the work of Katzin and Levin \cite%
{KatzinL74,KatzinL76,KatzinL81}. Important contributions have also been done
by Aminova \cite{Aminova94,Aminova95,Aminova00,Aminova2006,Aminova2010},
Prince and Crampin \cite{Prince84} and many others. More recent is the work
of Feroze et al \cite{FMkvs}. In \cite{FMkvs} they have considered the KVs
of the metric and their relation to the Lie symmetries of the system of
affinely parameterized geodesics of maximally symmetric spaces of low
dimension. In the same paper a conjecture is made, which essentially says
that the maximally symmetric spaces of non-vanishing curvature do not admit
further Lie symmetries.

In the following we consider the set of autoparallels - not necessarily
affinely parameterized - of a symmetric connection. We find that the major
symmetry condition relates the Lie symmetries with the special projective
symmetries of the connection. A\ similar result has been obtained by Prince
and Crampin in \cite{Prince84} using the bundle formulation of second order
ODEs.

Furthermore, because the geodesic equations follow from the variation of the
geodesic Lagrangian defined by the metric and due to the fact that the
Noether symmetries are a subgroup of the Lie group of Lie symmetries of
these equations, one should expect a relation / identification of the
Noether symmetries of this Lagrangian with the projective collineations of
the metric or with its degenerates. Recent work in this direction has been
by Bokhari et all \cite{Bok06,Bok07} in which the relation of the Noether
symmetries with the KVs of some special spacetimes is discussed.

In section \ref{TheSymCon} we derive the Lie symmetry conditions for a
general system of second order ODE polynomial in the first derivatives. In
section \ref{Liegeod} we apply these conditions in the special case of
Riemannian spaces and in Theorem \ref{TheoremGEs} we give the Lie symmetry
vectors in terms of the special projective collineations of the metric and
their degenerates. In section \ref{NSSS} we give the second result of this
work, that is Theorem \ref{TheoremGE2}, which relates the Noether symmetries
of the geodesic Lagrangian defined by the metric with the homothetic algebra
of the metric and comment on the results obtained so far in the literature.
Finally in section \ref{GEsAppl} we apply the results to various cases and
eventually we give the Lie symmetries, the Noether symmetries and the
associated conserved quantities of Einstein spaces, the G\"{o}del spacetime,
the Taub spacetime and the Friedman Robertson Walker spacetimes.

\section{The Lie symmetry conditions in an affine space}

\label{TheSymCon}

We consider the system of ODEs:
\begin{equation}
\ddot{x}^{i}+\Gamma _{jk}^{i}\dot{x}^{j}\dot{x}^{k}+\sum%
\limits_{m=0}^{n}P_{j_{1}...j_{m}}^{i}\dot{x}^{j_{1}}\ldots \dot{x}^{j_{m}}=0
\label{de.1}
\end{equation}%
where $\Gamma _{jk}^{i}$ are the connection coefficients of the space and $%
P_{j_{1}...j_{m}}^{i}(t,x^{i})$ are smooth functions completely symmetric in
the lower indices and derive the Lie point symmetry conditions in geometric
form using the standard approach. Equation (\ref{de.1}) is quite general and
covers most of the standard cases autonomous and non autonomous equations.
For instance for~all $\mathbf{P}=0$ equation (\ref{de.1}) becomes%
\begin{equation}
\ddot{x}+\Gamma _{jk}^{i}\dot{x}^{j}\dot{x}^{k}=0  \label{de.1a}
\end{equation}%
which are the geodesic equations. In case $P^{i}\mathbf{=}F^{i}$ and $%
P_{j_{1}...j_{m}}^{i}=0$ equation (\ref{de.1}) becomes
\begin{equation}
\ddot{x}+\Gamma _{jk}^{i}\dot{x}^{j}\dot{x}^{k}+F^{i}\left( t,x^{i}\right) =0
\label{de.1b}
\end{equation}%
which are the equations of motions of a particle in a curved space under the
action of the force $\mathbf{F}$. Furthermore because the $\Gamma _{jk}^{i}$%
's are not assumed to be symmetric, the results are valid in a space with
torsion. Obviously they hold in a Riemannian space provided the connection
coefficients are given in terms of the Christofell symbols.

Following, the standard procedure (see e.g. \cite{StephaniB,BlumanB}) we
find that the Lie symmetry conditions for the values of $m\leq 4$ are%
\footnote{%
The detailed calculations are given in Appendix \ref{appendixLieCon}.}

\begin{align}
L_{\eta }P^{i}+2\xi ,_{t}P^{i}+\xi P^{i},_{t}+\eta ^{i},_{tt}+\eta
^{j},_{t}P_{.j}^{i}& =0  \label{de.02} \\
L_{\eta }P_{j}^{i}+\xi ,_{t}P_{j}^{i}+\xi P_{j}^{i},_{t}+\left( \xi
,_{k}\delta _{j}^{i}+2\xi ,_{j}\delta _{k}^{i}\right) P^{k}+2\eta
^{i},_{t|j}-\xi ,_{tt}\delta _{k}^{i}+2\eta ^{k},_{t}P_{.jk}^{i}& =0
\label{de.03} \\
L_{\eta }P_{jk}^{i}+L_{\eta }\Gamma _{jk}^{i}+\left( \xi ,_{d}\delta
_{(k}^{i}+\xi ,_{(k}\delta _{|d|}^{i}\right) P_{.j)}^{d}+\xi
P_{.kj,t}^{i}-2\xi ,_{t(j}\delta _{k)}^{i}+3\eta ^{d},_{t}P_{.dkj}^{i}& =0
\label{de.04} \\
L_{\eta }P_{.jkd}^{i}-\xi ,_{t}P_{.jkd}^{i}+\xi ,_{e}\delta
_{(k}^{i}P_{.dj)}^{e}+\xi P_{.jkd,t}^{i}+4\eta ^{e},_{t}P_{.jkde}^{i}-\xi
_{(,j|k}\delta _{d)}^{i}& =0~  \label{de.05}
\end{align}%
and for $\left( \dot{x}\right) ^{l},l~\geq 4$
\begin{align}
& L_{\eta }P_{j_{1}...j_{m}}^{i}+P_{j_{1}...j_{m}~,t}^{i}\xi +\left(
2-m\right) \xi _{,t}P_{j_{1}...j_{m}}^{i}+  \notag \\
& +\xi _{,r}\left( 2-\left( m-1\right) \right) P_{j_{1}...j_{m-1}}^{i}\delta
_{j_{m}}^{r}+\left( m+1\right) P_{j_{1}...j_{m+1}}^{i}\eta
_{,t}^{j_{m+1}}+\xi _{,j}P_{j_{1}...j_{m-1}}^{j}\delta _{j_{m}}^{i}=0.
\label{de.06}
\end{align}

From the above general relations it is possible to extract the Lie symmetry
conditions for the various values of functions $\mathbf{P}$. For example the
Lie symmetry conditions for the geodesic equations (\ref{de.1a}) are as
follows.
\begin{align}
\eta ^{i},_{tt}& =0  \label{de.09} \\
2\eta ^{i},_{t|j}-\xi ,_{tt}\delta _{k}^{i}& =0  \label{de.10} \\
L_{\eta }\Gamma _{jk}^{i}-2\xi ,_{t(j}\delta _{k)}^{i}& =0  \label{de.11} \\
\xi _{(,j|k}\delta _{d)}^{i}& =0.  \label{de.12}
\end{align}%
Moreover, the symmetry conditions for the system (\ref{de.1b}) are

\begin{align}
L_{\eta }F^{i}+2\xi ,_{t}F^{i}+\xi F^{i},_{t}+\eta ^{i},_{tt}& =0
\label{de.13} \\
\left( \xi ,_{k}\delta _{j}^{i}+2\xi ,_{j}\delta _{k}^{i}\right) F^{k}+2\eta
^{i},_{t|j}-\xi ,_{tt}\delta _{k}^{i}& =0  \label{de.14} \\
L_{\eta }\Gamma _{jk}^{i}-2\xi ,_{t(j}\delta _{k)}^{i}& =0  \label{de.15} \\
\xi _{(,j|k}\delta _{d)}^{i}& =0.  \label{de.16}
\end{align}

In the same manner we work for more terms of $\mathbf{P}.$ We note the
appearance of the term $L_{\eta }\Gamma _{jk}^{i}$ in these expressions.

\section{Lie symmetries of autoparallel equation}

\label{autop}

Consider a $C^{p}$ manifold $M$ of dimension $n$, endowed with a $\Gamma
_{jk}^{i}$ symmetric\footnote{%
The coefficients $\Gamma _{jk}^{i}$ in general are not symmetric in the
lower indices. In the autoparallel equation (\ref{PCA.1}) the antisymmetric
part of $\Gamma _{\lbrack jk]}^{i}$ (the torsion) does not play a role.}
connection. In a local coordinate system $\{x^{i}:i=1,\ldots ,n\}$ the
connection $\Gamma _{jk}^{i}\partial _{i}=\nabla _{j}\partial _{k}$ and the
\textbf{autoparallels} of the connection are defined by the requirement
\begin{equation}
\ddot{x}^{i}+\Gamma _{jk}^{i}\dot{x}^{j}\dot{x}^{k}+\phi \left( t\right)
\dot{x}^{i}=0  \label{PCA.1}
\end{equation}%
where $t$ is a parameter along the paths. When $\phi \left( t\right) $
vanishes, we say that the autoparallels are \textbf{affinely parameterized}
and in this case $t$ is called an\textbf{\ affine parameter}, that is one
has (\ref{de.1a}). Consider the infinitesimal transformation%
\begin{equation}
\bar{t}=t+\varepsilon \xi \left( t,x^{k}\right) ~,~~\bar{x}%
^{i}=x^{i}+\varepsilon \eta ^{i}\left( t,x^{k}\right)  \label{PCA.1a}
\end{equation}%
with infinitesimal generator
\begin{equation}
X=\xi \left( t,x^{k}\right) \partial _{t}+\eta ^{i}\left( t,x^{k}\right)
\partial _{t}.  \label{PCA.1B}
\end{equation}%
\ The autoparallels (\ref{PCA.1}) are invariant under transformation (\ref%
{PCA.1a}) if
\begin{equation}
X^{\left[ 2\right] }\left( \ddot{x}^{i}+\Gamma _{jk}^{i}\dot{x}^{j}\dot{x}%
^{k}+\phi \left( t\right) \dot{x}^{i}\right) =0\;  \label{PCA.1C}
\end{equation}%
where $X^{\left[ 2\right] }$ is the second prolongation of (\ref{PCA.1B}%
).~For $P_{j_{1}...j_{m}}^{i}=0~$for $m\neq 0$ and $P_{j_{1}}^{i}=\phi
\left( t\right) ,~~\Gamma _{jk}^{i}\left( x^{k}\right) +P_{jk}^{i}\left(
t,x^{k}\right) =\Gamma _{jk}^{i}\left( t,x^{k}\right) ~$and from conditions (%
\ref{de.02})-(\ref{de.06})~we have the Lie symmetry conditions for the
autoparallel equations (\ref{PCA.1})
\begin{equation}
\eta _{,tt}^{i}+\eta _{,t}^{i}\phi =0  \label{PCA.31}
\end{equation}%
\begin{equation}
\xi _{,tt}\delta _{j}^{i}-\xi \phi _{,t}\delta _{j}^{i}-2\left[ \eta
_{,tj}^{i}+\eta _{,t}^{k}\Gamma _{(kj)}^{i}\right] -\phi \xi _{,t}\delta
_{j}^{i}=0  \label{PCA.32}
\end{equation}

\begin{equation}
\mathcal{L}_{\mathbf{\eta }}\Gamma _{(jk)}^{i}=-2\phi \xi _{(,j}\delta
_{k)}^{i}+\xi \Gamma _{(kj),t}^{i}+2\xi _{,t(j}\delta _{k)}^{i}
\label{PCA.34}
\end{equation}
\begin{equation}
\xi _{(,j|k}\delta _{d)}^{i}=0\;.  \label{PCA.35}
\end{equation}

Define the quantity
\begin{equation}
\Phi =\xi _{,t}-\phi \xi .  \label{PCA.36}
\end{equation}%
Then condition (\ref{PCA.34}) is written $\left( \text{note that }\phi
,_{i}=0\right) $:%
\begin{equation}
L_{\eta }\Gamma _{(jk)}^{i}=2\Phi _{(,j}\delta _{k)}^{i}-\xi \Gamma
_{(kj),t}^{i}.  \label{PCA.37}
\end{equation}%
If we consider the vector\textbf{\ }$\mathbb{\xi }\mathbf{=}\xi \partial
_{t} $ (which does not have components along $\partial _{i}$) we compute%
\begin{equation*}
L_{\mathbf{\xi }}\Gamma _{(jk)}^{i}=\xi \Gamma _{(kj),t}^{i}
\end{equation*}%
hence (\ref{PCA.37}) is written
\begin{equation}
L_{X}\Gamma _{(jk)}^{i}=2\Phi _{(,j}\delta _{k)}^{i}.  \label{PCA.38}
\end{equation}

We observe that this condition is the condition~for a projective
collineation of the connection $\Gamma _{(jk)}^{i}$ along the symmetry
vector $X$ and with projecting function $\Phi $. Concerning the other
conditions we note that (\ref{PCA.32}) can be written in covariant form as
follows:
\begin{equation}
\Phi _{,t}\delta _{j}^{i}-2\eta _{,t\mid j}^{i}=0  \label{PCA.39}
\end{equation}%
where $\eta _{,t\mid j}^{i}=\eta _{,tj}^{i}+\eta _{,t}^{k}\Gamma _{(kj)}^{i}$
is the covariant derivative with respect to $\Gamma _{(kj)}^{i}$ of the
vector $\eta _{,t}^{i}.$ Condition (\ref{PCA.35}) implies that $\xi ,_{{i}}$
is a gradient KV of the metric of the space $x^{i}$. Condition (\ref{PCA.31}%
) is obviously in covariant form with respect to the index $i.$

We arrive at the following conclusion.

a) The conditions for the Lie symmetries of the autoparallel equations (\ref%
{PCA.1}) are covariant equations because if we consider the connection in
the augmented $n+1$ space $\{x^{i},t\}$, all components of $\Gamma $ which
contain an index along the direction of $t$ vanish, therefore the partial
derivatives with respect to $t$ can be replaced with covariant derivatives
with respect to $t$.

b) Equation (\ref{PCA.31}) gives the functional dependence of $\eta ^{i}$ on
$t$ and the non-affine parametrization function $\phi (t)$.

c) Equation (\ref{PCA.35}) gives that the vector $\xi _{,i}$ is a gradient
Killing vector of the $n-$dimensional space $\{x^{i}\}$.

d) Equation (\ref{PCA.32}) relates the functional dependence of $\eta ^{i}$
and $\xi $ in terms of $t$.

e) Equation (\ref{PCA.34}) is the most important equation for our purpose,
because it states that \emph{the symmetry vector (\ref{PCA.1a})\ is an
affine collineation in the jet space }$\{t,x^{i}\}$ because it preserves
both the geodesics and their affine parameter. In the space $\{x^{i}\}$ the
vectors $\eta ^{i}\emph{(t,x)}\partial _{x^{i}}$ \emph{are projective
collineations} because they preserve the geodesics but not necessarily their
parametrization.

In the following we restrict our considerations to the case of Riemannian
connections that is the $\Gamma _{jk}^{i}$ are symmetric and the covariant
derivative of the metric vanishes.

\section{Lie and Noether symmetries of geodesic equations}

\label{Liegeod}

We compute the Lie symmetry vectors of geodesics equations (\ref{de.1a})
with affine parametrization; that is, we assume $\phi =0$ and $\Gamma
_{jk,t}^{i}=0.$ The later is a reasonable assumption because the $\Gamma
_{jk}^{i\prime }$'s are computed in terms of the metric which does not
depend on the parameter $t$. Under these assumptions the symmetry conditions
for the geodesics equations (\ref{de.1a}) are (\ref{de.09})-(\ref{de.12}).
We proceed with the solution of this system of equations\footnote{%
See also \cite{Prince84} Table II.}.

Equation (\ref{de.09}) implies%
\begin{equation}
\eta ^{i}(t,x)=A^{i}(x)t+B^{i}(x)  \label{PCA.59}
\end{equation}%
where $A^{i}(x),B^{i}(x)$ are arbitrary differentiable vector fields.

The solution of equation (\ref{de.12}) is%
\begin{equation}
\xi (t,x)=C_{J}(t)S^{J}(x)+D(t)  \label{PCA.60}
\end{equation}%
where $C_{J}(t),D(t)$ are arbitrary functions of the affine parameter $t$
and $S^{J}(x)$ is a function whose gradient is a gradient KV, i.e.
\begin{equation}
S^{J}(x)_{|(i,j)}=0.  \label{PCA.61}
\end{equation}%
The index $J$ runs through the number of gradient KVs of the metric.
Condition (\ref{de.10}) gives%
\begin{equation}
2A(x)_{|j}^{i}=\left[ C_{J}(t),_{tt}S^{J}(x)+D(t),_{tt}\right] \delta
_{j}^{i}.  \label{PCA.61a}
\end{equation}%
Because the left hand side is a function of $x$ only, we must have%
\begin{align}
D(t),_{tt}& =M\Rightarrow D(t)=\frac{1}{2}Mt^{2}+Kt+L\text{ \ where }M,K,L%
\text{ constants}  \label{PCA.62} \\
C_{J}(t),_{tt}& =G_{J}=\text{constant }\Rightarrow C_{J}(t)=\frac{1}{2}%
G_{J}t^{2}+E_{J}t+F_{J}\text{\ where }G_{J},E_{J},F_{J}\text{ constants.}
\label{PCA.63}
\end{align}%
Replacing in (\ref{PCA.61a}) we find%
\begin{equation}
2A(x)_{;j}^{i}=\left( G_{J}S^{J}(x)+M\right) \delta _{j}^{i}\Rightarrow
A(x)_{i;j}=\frac{1}{2}\left( G_{J}S^{J}(x)+M\right) g_{ij}  \label{PCA.64}
\end{equation}%
where we have lowered the index because the connection is metric (i.e. $%
g_{ij|k}=0)$. The last equation implies that the vector $A(x)^{i}$ is a
conformal Killing vector with conformal factor $\psi =\frac{1}{2}%
(G_{J}S^{J}(x)+M)$. Because $A(x)_{[i;j]}=0$ \ this vector is a gradient CKV.

We continue with condition (\ref{de.11}) and replace $\eta ^{i}(t,x)$ from (%
\ref{PCA.59})%
\begin{align}
L_{A}\Gamma _{jk}^{i}t+L_{B}\Gamma _{jk}^{i}& =2\xi ,_{t}{}_{(,j}\delta
_{k)}^{i}=2\left[ \left( G_{J}t+E_{J}\right) S^{J}(x)+Mt+K\right]
_{|(j}\delta _{k)}^{i}=2\left( G_{J}t+E_{J}\right) S^{J}(x)_{|(j}\delta
_{k)}^{i}\Rightarrow  \notag \\
L_{A}\Gamma _{jk}^{i}& =2G_{J}S^{J}(x)_{,(j}\delta _{k)}^{i}  \label{PCA.65}
\\
L_{B}\Gamma _{jk}^{i}& =2E_{J}S^{J}(x)_{,(j}\delta _{k)}^{i}.  \label{PCA.66}
\end{align}

The last two equations imply that the vectors $A^{i}(x),$ $B^{i}(x)$ \ are
\emph{special projective collineations} or \emph{affine collineations} of
the metric - or one of their specializations - with projective functions $%
G_{J}S^{J}(x)$ and $E_{J}S^{J}(x)$ or zero respectively. Note that relations
(\ref{PCA.65}), (\ref{PCA.66}) remain true if we add a KV to the vectors $%
A^{i}(x),$ $B^{i}(x)$ , therefore these vectors are determined up to a KV%
\footnote{%
Because $A^{i}$ is a projective collineation and a CKV it must be a HV.}.

It is well known that in a Riemannian space a CKV\ $K^{i}$ with conformal
factor $\psi (x^{i})$ satisfies the identity:%
\begin{equation}
L_{K}\Gamma _{(jk)}^{i}=g^{is}\left[ \psi _{,j}g_{ks}+\psi _{,k}g_{js}-\psi
_{,s}g_{jk}\right] .  \label{PCA.67}
\end{equation}%
Applying this identity to the CKV $A^{i}$ we find:%
\begin{equation}
G_{J}S^{J}(x)_{,k}=0\Rightarrow G_{J}S^{J}(x)=2\rho =\text{constant.}
\label{PCA.68}
\end{equation}

This implies that $A^{i}$\emph{\ is a gradient HV (not necessarily proper)
with homothetic factor} $\rho +\frac{1}{2}M$ . Furthermore (\ref{PCA.64})
implies:%
\begin{align}
2A^{i}& =(2\rho +M)x^{i}+2L^{i}\Rightarrow  \notag \\
A^{i}& =(\rho +\frac{1}{2}M)x^{i}+L^{i}  \label{PCA.68a}
\end{align}%
where $L^{i}$ is a non-gradient KV.

We continue with the special projective collineation vector $B^{i}$. For
this vector we use the property that for a symmetric connection the
following identity, holds%
\begin{equation*}
\mathcal{L}_{\mathbf{B}}\Gamma _{(jk)}^{i}=B_{;jk}^{i}-R_{jkl}^{i}B^{l}.
\end{equation*}%
Replacing the left hand side from (\ref{PCA.66}) we find%
\begin{equation}
B_{;jk}^{i}-R_{jkl}^{i}B^{l}=2E_{J}S^{J}(x)_{,(j}\delta _{k)}^{i}.
\label{PCA.69}
\end{equation}%
Contracting the indices $i,j$ we find%
\begin{equation}
\left( B_{;i}^{i}-(n+1)E_{J}S^{J}(x)\right) _{;k}=0  \label{PCA.70}
\end{equation}%
which implies
\begin{equation}
B_{;i}^{i}=(n+1)E_{J}S^{J}(x)+2b  \label{PCA.70a}
\end{equation}%
where $b=$constant. In case this vector is an affine collineation then $%
B_{;i}^{i}=2b.$ Using the above results we find for $\xi (t,x)$%
\begin{align}
\xi (t,x)& =C_{J}(t)S^{J}+D(t)  \notag \\
& =\left( \frac{1}{2}G_{J}t^{2}+E_{J}t+F_{J}\right) S^{J}+\frac{1}{2}%
Mt^{2}+Kt+L  \notag \\
& =\frac{1}{2}(G_{J}S^{J}+M)t^{2}+(E_{J}S^{J}+K)t+F_{J}S^{J}+L  \notag
\end{align}

We summarize the above results in the following Theorem.

\begin{theorem}
\label{TheoremGEs} The Lie symmetry vector
\begin{equation*}
X=\xi (t,x)\partial _{t}+\eta ^{i}(t,x)\partial _{x^{i}}
\end{equation*}%
of the equation of geodesics (\ref{de.1a}) in a Riemannian space is
generated from the elements of the special projective algebra as follows.

Case A. The metric admits gradient KVs. Then

a. The function%
\begin{equation}
\xi (t,x)=\frac{1}{2}\left( G_{J}S^{J}+M\right) t^{2}+\left[ E_{J}S^{J}+K%
\right] t+F_{J}S^{J}+L,  \label{PCA.70b}
\end{equation}%
where $G_{J},M,b,K,F_{J}$ and $L$ \ are constants and the index $J\ $runs
along the number of gradient KVs

b. The vector%
\begin{equation}
\eta ^{i}(t,x)=A^{i}(x)t+B^{i}(x)+D^{i}(x)  \label{PCA.70d}
\end{equation}%
where the vector $A^{i}(x)$ is a gradient HV with conformal factor $\psi =%
\frac{1}{2}\left( G_{J}S^{J}+M\right) $ (if it exists), $D^{i}(x)$ is a
non-gradient KV of the metric and $B^{i}(x)$ is either a special projective
collineation with projection function $E_{J}S^{J}(x)$ or an AC and $E_{J}=0$
in (\ref{PCA.70b}).

Case \noindent B. The metric does not admit gradient KVs. Then

a. The function%
\begin{equation}
\xi (t,x)=\frac{1}{2}Mt^{2}+Kt+L  \label{PCA.70cc}
\end{equation}

b. The vector%
\begin{equation}
\eta ^{i}(t,x)=A^{i}(x)t+B^{i}(x)+D^{i}(x),  \label{PCA.70e}
\end{equation}%
where $A^{i}(x)$ is a gradient HV with conformal factor $\psi =\frac{1}{2}M,$
$D^{i}(x)$ is a non-gradient KV\ of the metric and $B^{i}(x)$ is an AC. If
in addition the metric does not admit a gradient HV, then%
\begin{align}
\xi (t)& =Kt+L  \label{PCA.70g} \\
\eta ^{i}(x)& =B^{i}(x)+D^{i}(x).  \label{PCA.70f}
\end{align}
\end{theorem}

\subsection{Noether symmetries and conservation laws}

\label{NSSS}

In a Riemannian space the equations of geodesics (\ref{de.1a}) are produced
from the geodesic Lagrangian:%
\begin{equation}
L=\frac{1}{2}g_{ij}\dot{x}^{i}\dot{x}^{j}~.  \label{NS.6}
\end{equation}%
The infinitesimal generator
\begin{equation}
X=\xi \left( t,x^{k}\right) \partial _{t}+\eta ^{i}\left( t,x^{k}\right)
\partial _{x^{i}}  \label{NS.6a}
\end{equation}%
is a Noether symmetry of Lagrangian (\ref{NS.6}) if there exists a smooth
function $f(t,x^{i})$ such that%
\begin{equation}
X^{\left[ 1\right] }L+\frac{d\xi }{dt}L=\frac{df}{dt}  \label{NS.7}
\end{equation}%
where
\begin{equation*}
X^{\left[ 1\right] }=\xi \left( t,x^{k}\right) \partial _{t}+\eta ^{i}\left(
t,x^{k}\right) \partial _{x^{i}}+\left( \eta _{\left[ 1\right] }^{i}\right)
\partial _{\dot{x}^{i}}
\end{equation*}%
is the first prolongation of $\mathbf{X}.$ We compute%
\begin{equation*}
X^{\left[ 1\right] }L=\frac{1}{2}\left( \eta ^{k}g_{ij,k}\dot{x}^{i}\dot{x}%
^{j}+2\frac{d\eta ^{k}}{dt}g_{ik}\dot{x}^{i}-2\dot{x}^{i}\dot{x}^{j}\frac{%
d\xi }{dt}g_{ij}\right) .
\end{equation*}%
Replacing the total derivatives in the rhs
\begin{align*}
\frac{d\xi }{dt}& =\xi _{,t}+\dot{x}^{k}\xi _{,k} \\
\frac{d\eta ^{i}}{dt}& =\eta _{,t}^{i}+\dot{x}^{k}\eta _{,k}^{i}
\end{align*}%
we find that%
\begin{equation*}
X^{\left[ 1\right] }L=\frac{1}{2}\left(
\begin{array}{c}
\eta ^{k}g_{ij,k}\dot{x}^{i}\dot{x}^{j}+2\eta _{,t}^{i}g_{ij}\dot{x}%
^{j}+\eta _{,r}^{i}g_{ik}\dot{x}^{k}\dot{x}^{r}+ \\
+\eta _{,r}^{i}g_{kj}\dot{x}^{k}\dot{x}^{r}-2\xi _{,t}g_{ij}\dot{x}^{i}\dot{x%
}^{j}-2\xi _{,k}g_{ij}\dot{x}^{i}\dot{x}^{j}\dot{x}^{k}%
\end{array}%
\right) .
\end{equation*}%
The term%
\begin{equation*}
\frac{d\xi }{dt}L=\frac{1}{2}\left( \xi _{,t}+\dot{x}^{k}\xi _{,k}\right)
g_{ij}\dot{x}^{i}\dot{x}^{j}.
\end{equation*}%
Finally the Noether symmetry condition (\ref{NS.7}) is

\begin{equation*}
-2f_{,t}+\left[ 2\eta _{,t}^{i}g_{ij}-2f_{,i}\right] \dot{x}^{j}-\xi
_{,k}g_{ij}\dot{x}^{i}\dot{x}^{j}\dot{x}^{k}+\left[ \eta ^{k}g_{ij,k}+\eta
_{,i}^{k}g_{ik}+\eta _{,i}^{k}g_{kj}-g_{ij}\xi _{,t}\right] \dot{x}^{i}\dot{x%
}^{j}=0.
\end{equation*}

This relation is an identity hence the coefficient of each power of $\dot{x}%
^{j}$ must vanish. This results in the equations:
\begin{align}
\xi _{,k}& =0  \label{NS.8} \\
L_{\eta }g_{ij}& =2\left( \frac{1}{2}\xi _{,t}\right) g_{ij}  \label{NS.9} \\
\eta _{,t}^{,i}g_{ij}& =f_{,i}  \label{NS.10} \\
f_{,t}& =0  \label{NS.11}
\end{align}

\bigskip Condition (\ref{NS.8}) gives $\xi_{,k}=0 \Rightarrow \xi=\xi\left(
t\right) $.

Condition (\ref{NS.11}) implies $f\left( x^{k}\right) $ and then condition (%
\ref{NS.8}) gives that $\eta ^{i}$ is of the form:%
\begin{equation}
\eta _{i}=f_{,i}t+K_{i}(x^{j}).  \label{NS.12}
\end{equation}%
Then from (\ref{NS.9}) follows that $\xi _{,t}$ must be at most linear in $%
t. $ Hence $\xi (t)$ must be at most a function of $t^{2}.$ Furthermore from
(\ref{NS.9}) follows that $\eta ^{i}$ is at most a CKV with conformal factor
$\psi _{H}=\frac{1}{2}(At+B),$ where $A,B$ are constants. We consider
various cases.

Case 1: Suppose $\xi =$constant=$C_{1}$. Then $\eta ^{i}$ is a KV\ of the
metric which is independent of $t.$ This implies that either $f_{,i}=0$ and $%
f=$constant $=A=0$ or that $f_{,i}$ is a gradient KV. In this case the
Noether symmetry vector is:%
\begin{equation*}
X^{i}=C_{1}\partial _{t}+g^{ij}\left( f_{,j}t+K_{j}(x^{r})\right) ,
\end{equation*}%
where $K^{i}$ is a non-gradient KV of $g_{ij}.$

Case 2: Suppose $\xi =2t.$ Then $\eta ^{i}$ is a HV of the metric $g_{ij}$
with homothetic factor $1$. Then $\eta _{i}=H_{i}(x^{j})$ , $%
f_{,i}=0\Rightarrow f=$constant$=0$ where $H^{i}$ is a HV of $g_{ij}$ with
homothetic factor $\psi ,$ not necessarily a gradient HV. In this case the
Noether symmetry vector is:%
\begin{equation*}
X^{i}=2\psi t\partial _{t}+H^{i}(x^{r}).
\end{equation*}

Case 3: $\xi (t)=t^{2}.$ Then $\eta ^{i}$ is a HV of the metric $g_{ij}$
(the variable $t$ cancels) with homothetic factor $1.$ Again $f_{,i}$ is a
gradient HV with homothetic factor $\psi $ and the Noether symmetry vector is%
\begin{equation*}
X^{i}=\psi t^{2}\partial _{t}+g^{ij}f_{,j}t.
\end{equation*}
Therefore we have the result.

\begin{theorem}
\label{TheoremGE2}The Noether Symmetries of the geodesic Lagrangian follow
from the KVs and the HV of the metric $g_{ij}$ as follows:%
\begin{eqnarray}
X &=&\left( C_{3}\psi t^{2}+2C_{2}\psi t+C_{1}\right) \partial _{t}+  \notag
\\
&&+\left[ C_{J}S^{J,i}+C_{I}KV^{Ii}+C_{IJ}tS^{J,i}+C_{2}H^{i}+C_{3}t(GHV)^{i}%
\right] \partial _{i}  \label{NS.14a}
\end{eqnarray}%
with corresponding gauge function%
\begin{equation}
f(x^{k})=C_{1}+C_{2}+C_{I}+C_{J}+\left[ C_{IJ}S^{J}\right] +C_{3}\left[ GHV%
\right] ,  \label{NS.15}
\end{equation}%
where $S^{J,i}$ are the $C_{J}$ gradient KVs, $KV^{Ii}$ are the $C_{I}$
non-gradient KVs, $H^{i}$ is a HV\ not necessarily gradient and $(GHV)^{i}$
is the gradient HV (if it exists) of the metric $g_{ij}$.
\end{theorem}

The importance of Theorems \ref{TheoremGEs} and \ref{TheoremGE2} are that
one is able to compute the Lie symmetries and the Noether symmetries of the
geodesic equations in a Riemannian space by computing the corresponding
collineation vectors avoiding the cumbersome formulation of the Lie symmetry
method. It is also possible to use the inverse approach and prove that a
space does not admit KVs, HVs, ACs and special PCs by using the
calculational approach of the Lie symmetry method (assisted with algebraic
manipulation programmes)\ and avoid the hard approach of Differential
Geometric methods. In Section \ref{GEsAppl} we demonstrate the use of the
above results.

\subsubsection{Noether Integrals of geodesic equations}

We know that, if the infinitesimal generator (\ref{NS.6a}) is a Noether
symmetry with Noether function $f,$ then the quantity:
\begin{equation}
\phi =\xi \left( \dot{x}^{i}\frac{\partial L}{\partial \dot{x}^{i}}-L\right)
-\eta ^{i}\frac{\partial L}{\partial \dot{x}^{i}}+f  \label{NS.16}
\end{equation}%
is a First Integral of $L$ which satisfies $X\phi $\thinspace $=0,~\frac{%
d\phi }{dt}=0$. For the Lagrangian defined by the metric $g_{ij},$ i.e. $L=%
\frac{1}{2}g_{ij}\dot{x}^{i}\dot{x}^{j},$ we compute:
\begin{equation}
\phi =\frac{1}{2}\xi g_{ij}\dot{x}^{i}\dot{x}^{j}-g_{ij}\eta ^{i}\dot{x}%
^{j}+f.  \label{NS.17}
\end{equation}%
In (\ref{NS.14a}) we have computed the generic form of the Noether symmetry
and the associated Noether function for this Lagrangian. Substituting into (%
\ref{NS.17}) we find the following expression for the generic First Integral
of the geodesic equations:%
\begin{align}
\phi & =\frac{1}{2}\left[ C_{3}\psi t^{2}+2C_{2}\psi t+C_{1}\right] g_{ij}%
\dot{x}^{i}\dot{x}^{j}  \notag \\
& +\left[
C_{J}S^{J,i}+C_{I}KV^{Ii}+C_{IJ}tS^{J,i}+C_{2}H^{i}(x^{r})+C_{3}t(GHV)^{,i}%
\right] g_{ij}\dot{x}^{j}  \notag \\
& +C_{1}+C_{2}+C_{I}+C_{J}+\left[ C_{IJ}S^{J}\right] +C_{3}\left[ GHV\right]
.  \label{NS.21}
\end{align}

From the generic expression we obtain the following first integrals\footnote{%
GHV stands for gradient HV}%
\begin{eqnarray}
C_{I} &\neq &0:\phi _{C_{I}}=KV_{i}^{I}\dot{x}^{i}-C_{I}  \label{NS.22} \\
C_{J} &\neq &0:g_{ij}S^{J,i}\dot{x}^{j}-C_{J}  \label{NS.23} \\
C_{IJ} &\neq &0:tg_{ij}S^{J,i}\dot{x}^{j}-S^{J}  \label{NS.24} \\
C_{1} &\neq &0:\phi _{C_{1}}=\frac{1}{2}g_{ij}\dot{x}^{i}\dot{x}^{j}
\label{NS.25} \\
C_{2} &\neq &0:\phi _{C_{2}}=t\psi g_{ij}\dot{x}^{i}\dot{x}^{j}-g_{ij}H^{i}%
\dot{x}^{j}+C_{2}  \label{NS.26} \\
C_{3} &\neq &0:\phi _{C_{3}}=\frac{1}{2}t^{2}\psi g_{ij}\dot{x}^{i}\dot{x}%
^{j}-t(GHV)_{,i}\dot{x}^{i}+\left[ GHV\right] .  \label{NS.27}
\end{eqnarray}

We conclude that the first Integrals of the Noether symmetry vectors of the
geodesic equations are:

a) linear, the $\phi _{I},\phi _{J},\phi _{IJ}\,\ $

b) quadratic, the $\phi _{c1},\phi _{c2},\phi _{3}.~$

These results are compatible with the corresponding results of Katzin and
Levine \cite{KatzinL81}.

In a number of recent papers \ \cite{Feroze2011,Hussain2010,Feroze2010}, the
authors study the relation between the Noether symmetries of the geodesic
Lagrangian. They also make a conjecture concerning the relation between the
Noether symmetries and the conformal algebra of spacetime and concentrate
especially on conformally flat spacetimes. In \cite{Feroze2010} it is also
claimed that the author has found new conserved quantities for spaces of
different curvatures, which seem to be of non Noetherian character.
Obviously due to the above results the conjecture/results in these papers
should be revised and the word `conformal' should be replaced with the word
`homothetic'.

\section{Applications}

\label{GEsAppl}

We apply the general Theorems \ref{TheoremGEs} and \ref{TheoremGE2} in
various curved spaces where we determine the Lie and the Noether symmetries
of the corresponding geodesic equations.

\subsection{The geodesic symmetries of Einstein spaces}

Suppose $Y$ is a projective collineation with projection function $\phi
(x^{k}),$ such that
\begin{equation*}
L_{Y}\Gamma _{jk}^{i}=\phi _{,j}\delta _{k}^{i}+\phi _{,k}\delta _{j}^{i}.
\end{equation*}%
For a proper Einstein space $(R\neq 0)$ we have $R_{ab}={\frac{R}{n}g_{ab}}$
~from which follows \cite{StephaniE}%
\begin{equation}
L_{Y}g_{ab}=\frac{n(1-n)}{R}\phi _{;ab}-L_{Y}(\ln R)g_{ab}.  \label{NPP.2}
\end{equation}

Using the contracted Bianchi identity
\begin{equation*}
\left[ R^{ij}-\frac{1}{2}Rg^{ij}\right] _{;j}=0
\end{equation*}
it follows that in an Einstein space of dimension\footnote{%
Recall that all two dimensional spaces are Einstein spaces.} $n>2$ the
curvature scalar $R=$constant and (\ref{NPP.2}) reduces to%
\begin{equation*}
L_{Y}g_{ab}=\frac{n(1-n)}{R}\phi _{;ab}.
\end{equation*}

It follows that if $Y^{i}$ generates either an affine or a special
projective collineation, then $\phi _{;ab}=0.$ Hence $X^{a}$ reduces to a
KV. This means that proper Einstein spaces do not admit HV, ACs, special PCs
and gradient KVs (\cite{Yano,Barnes})

The above results and Theorem \ref{TheoremGEs} lead to the following result.

\begin{proposition}
\label{ThES}The Lie symmetries of the geodesic equations in a proper
Einstein space with curvature scalar $R$ $\neq 0$ are given by the vectors
\begin{equation*}
X=\left( Kt+L\right) \partial _{t}+D^{i}\left( x\right) \partial _{i}~
\end{equation*}%
where $D^{i}(x)~$is a nongradient KV and $K,L~$are constants
\end{proposition}

For the Noether symmetries of Einstein spaces we have the following

\begin{proposition}
\label{NES}The Noether symmetries of the geodesic equations in a proper
Einstein space with curvature scalar $R$\ $\neq 0$\ are given by the vectors%
\newline
\begin{equation*}
X=L\partial _{t}+D^{i}\left( x\right) \partial _{i}~~,~f=\,\text{constant}~
\end{equation*}
\end{proposition}

Proposition \ref{ThES} extends and amends the conjecture of \cite{FMkvs} to
the more general case of Einstein spaces.

We apply the results to the maximally symmetric space of Euclidian 2d sphere.

\subsubsection{Euclidian 2d sphere}

The geodesic Lagrangian of the Euclidian 2d sphere is%
\begin{equation*}
L\left( \phi ,\dot{\phi},\theta ,\dot{\theta}\right) =\frac{1}{2}\dot{\phi}%
^{2}+\sin ^{2}\phi ~\dot{\theta}^{2}
\end{equation*}%
and the geodesic equations are
\begin{equation*}
\ddot{\phi}-\frac{1}{2}\sin 2\phi ~\dot{\theta}^{2}=0
\end{equation*}%
\begin{equation*}
\ddot{\theta}+\cot \phi ~\dot{\phi}\dot{\theta}=0
\end{equation*}

The Euclidian 2d sphere is an Einstein space with curvature scalar $R=2$,
therefore propositions (\ref{ThES}) and (\ref{NES}) apply. The KVs of the
Euclidian 2d sphere are the elements of the $so\left( 3\right) $ Lie algebra
(See example \ref{Kvs2dsphere})
\begin{equation*}
X_{1}=\sin \theta \partial _{\phi }+\cos \theta \cot \phi \partial _{\theta
}~\ ,~X_{2}=\cos \theta \partial _{\phi }-\sin \theta \cot \phi \partial
_{\theta }~~,~~X_{3}=\partial _{\theta }.
\end{equation*}

Consequently the Lie symmetries of geodesic equations are the elements of
the $so\left( 3\right) $ plus the vectors $\partial _{t},~t\partial _{t}$. $%
\ $Likewise$~$the Noether symmetries are the elements of $so\left( 3\right) $
plus the vector $\partial _{t}$. The corresponding Noether integrals are%
\begin{align*}
\phi _{1}& =\dot{\theta}\sin ^{2}\phi ~ \\
\phi _{2}& =\dot{\phi}\sin \theta +\frac{1}{2}\dot{\theta}\sin 2\phi \cos
\theta ~ \\
\phi _{3}& =\dot{\phi}\cos \theta -\frac{1}{2}\dot{\theta}\sin 2\phi \sin
\theta ~~
\end{align*}%
and the Hamiltonian constant.

\subsection{The geodesic symmetries of G\"{o}del spacetime}

The G\"{o}del metric in Cartesian coordinates is

\begin{equation*}
ds^{2}=-dt^{2}-2e^{ax}dtdy+dx^{2}-\frac{1}{2}e^{2ax}dy^{2}+dz^{2}.
\end{equation*}%
The geodesic Lagrangian is
\begin{equation}
L=\frac{1}{2}\left( -t^{\prime 2}-2e^{ax}t^{\prime }y^{\prime }+dx^{2}-\frac{%
1}{2}e^{2ax}y^{\prime 2}+z^{\prime 2}\right)  \label{gLa}
\end{equation}%
where $^{\prime }$ means $\frac{d}{ds}$ where $"s"$ is an affine parameter.
The geodesic equations are%
\begin{align*}
t^{\prime \prime }+2at^{\prime }x^{\prime }+ae^{ax}x^{\prime }y^{\prime }& =0
\\
x^{\prime \prime }+ae^{ax}t^{\prime }y^{\prime }+\frac{1}{2}%
ae^{2ax}y^{\prime 2}& =0 \\
y^{\prime \prime }-2ae^{-ax}t^{\prime }x^{\prime }& =0 \\
z^{\prime \prime }& =0.
\end{align*}

The special projective algebra of the G\"{o}del metric has as follows:%
\begin{equation*}
Y^{1}=\partial _{z}~,~Y^{3}=\partial _{x}-ay\partial _{y}~,~Y^{4}=\partial
_{t}~,~Y^{5}=\partial _{y}~,~Y^{6}=z\partial _{z}
\end{equation*}%
\begin{equation*}
Y^{2}=\left( -\frac{2}{a}e^{-ax}\right) \partial _{t}+y\partial _{x}+\left(
\frac{2e^{-2ar}-a^{2}y^{2}}{2a}\right) \partial _{y}
\end{equation*}%
where $Y^{1}$ is a gradient KV $\left( S_{1}=z\right) $, $Y^{2-5}$ are non
gradient KVs and $Y^{6}$ is a proper AC. The G\"{o}del spacetime does not
admit proper sp.PC \cite{HalldaCosta}.

Applying theorem (\ref{TheoremGEs}) we find that, the G\"{o}del spacetime
admits ten Lie point symmetries as follows

\begin{equation*}
X_{1}=\partial _{s}~~,~~X_{2}=s\partial _{s}~~,~~X_{3}=z\partial
_{s}~~,~~X_{4}=Y^{4}
\end{equation*}%
\begin{equation*}
~X_{5}=Y^{2}~~,~~X_{6}=Y^{3}~~,~~X_{7}=Y^{5}~
\end{equation*}%
\begin{equation*}
X_{8}=Y^{1}~,~X_{9}=sY^{1}~~,~~X_{10}=Y^{6}
\end{equation*}%
whose Lie algebra is given in Table \ref{LieGodel}.

\begin{table}[tbp] \centering%
\caption{Lie algebra of the Gödel geodesic symmetries}%
\begin{tabular}{ccccccccccc}
\hline\hline
$\left[ X_{I},X_{J}\right] $ & $X_{1}$ & $X_{2}$ & $X_{3}$ & $X_{4}$ & $%
X_{5} $ & $X_{6}$ & $X_{7}$ & $X_{8}$ & $X_{9}$ & $X_{10}$ \\ \hline
$X_{1}$ & $0$ & $X_{1}$ & $0$ & $0$ & $0$ & $0$ & $0$ & $0$ & $X_{8}$ & $0$
\\
$X_{2}$ & $-X_{1}$ & $0$ & $-X_{3}$ & $0$ & $0$ & $0$ & $0$ & $0$ & $-X_{9}$
& $0$ \\
$X_{3}$ & $0$ & $X_{3}$ & $0$ & $0$ & $0$ & $0$ & $0$ & $-X_{1}$ & $%
X_{10}-X_{2}$ & $-X_{3}$ \\
$X_{4}$ & $0$ & $0$ & $0$ & $0$ & $0$ & $0$ & $0$ & $0$ & $0$ & $0$ \\
$X_{5}$ & $0$ & $0$ & $0$ & $0$ & $0$ & $X_{5}$ & $-X_{7}$ & $0$ & $0$ & $0$
\\
$X_{6}$ & $0$ & $0$ & $0$ & $0$ & $-X_{5}$ & $0$ & $aX_{7}$ & $0$ & $0$ & $0$
\\
$X_{7}$ & $0$ & $0$ & $0$ & $0$ & $X_{7}$ & $-aX_{7}$ & $0$ & $0$ & $0$ & $0$
\\
$X_{8}$ & $0$ & $0$ & $X_{1}$ & $0$ & $0$ & $0$ & $0$ & $0$ & $0$ & $X_{8}$
\\
$X_{9}$ & $-X_{8}$ & $X_{9}$ & $X_{2}-X_{10}$ & $0$ & $0$ & $0$ & $0$ & $0$
& $0$ & $X_{9}$ \\
$X_{10}$ & $0$ & $0$ & $X_{3}$ & $0$ & $0$ & $0$ & $0$ & $-X_{8}$ & $-X_{19}$
& $0$ \\ \hline\hline
\end{tabular}%
\label{LieGodel}%
\end{table}%

There are two Lie subalgebras. One spanned by the vectors $\left\{
X_{1},X_{4},X_{5},X_{6},X_{7},X_{8},X_{9}\right\} $ and a second spanned by
the vectors $\left\{ X_{2},X_{3},X_{10}\right\} $.~It can be shown that the
first subalgebra consists of the Noether symmetries of the Lagrangian (\ref%
{gLa}). The corresponding Noether integrals are%
\begin{eqnarray*}
\phi _{8} &=&z^{\prime } \\
\phi _{9} &=&sz^{\prime }-z \\
\phi _{4} &=&t^{\prime }+y^{\prime }e^{ax}~ \\
\phi _{6} &=&x^{\prime }+ay\phi _{7} \\
\phi _{7} &=&e^{\alpha x}\left( t^{\prime }+\frac{1}{2}y^{\prime
}e^{ax}\right) \\
\phi _{5} &=&\frac{2}{a}e^{-ax}t^{\prime }+\frac{3}{a}y^{\prime
}+2yx^{\prime }+ay^{2}\phi _{7}.
\end{eqnarray*}%
The Noether constant corresponding to the Noether symmetry $X_{1}=\partial
_{s}$ is the total energy i.e. the Hamiltonian.

\subsection{The geodesic symmetries of Taub spacetime}

Consider the Taub spacetime with line element

\begin{equation}
ds^{2}=x^{-\frac{1}{2}}\left( -dt^{2}+dx^{2}\right) +x\left(
dy^{2}+dz^{2}\right) .  \label{NGRG.24}
\end{equation}%
The geodesic Lagrangian is
\begin{equation}
L=\frac{1}{2}\left( x^{-\frac{1}{2}}\left( -t^{\prime 2}+x^{\prime 2}\right)
+x\left( y^{\prime 2}+z^{\prime 2}\right) \right)  \label{NGRG.24a}
\end{equation}%
and the geodesic equations are%
\begin{eqnarray*}
t^{\prime \prime }-\frac{1}{2x}t^{\prime }x^{\prime } &=&0 \\
y^{\prime \prime }+\frac{1}{x}y^{\prime }x^{\prime } &=&0 \\
z^{\prime \prime }+\frac{1}{x}z^{\prime }x^{\prime } &=&0 \\
\ddot{x}-\frac{1}{4x}\left( t^{\prime 2}+x^{\prime 2}\right) -\frac{x^{\frac{%
1}{2}}}{2}\left( y^{\prime 2}+z^{\prime 2}\right) &=&0.
\end{eqnarray*}

In order to find the Lie symmetries of the geodesic equations for the Taub
spacetime (\ref{NGRG.24}) we need to have the special projective algebra of (%
\ref{NGRG.24}). The spacetime (\ref{NGRG.24}) admits a five dimensional
special projective algebra\footnote{%
The spacetime (\ref{NGRG.24}) does not admit proper AC or sp.PC.} which
consists from four non gradient KVs $Y^{1-4}$ and a non gradient HV~$Y^{5}~$%
\cite{Taub96}.
\begin{eqnarray*}
Y_{1} &=&\partial _{t}~,~Y_{2}=z\partial _{y}-y\partial _{z}~ \\
Y_{3} &=&\partial _{y}~,~Y_{4}=\partial _{z}~
\end{eqnarray*}%
\begin{equation*}
Y_{5}=\frac{4}{3}t\partial _{t}+\frac{4}{3}x\partial _{x}+\frac{1}{3}%
y\partial _{y}+\frac{1}{3}z\partial _{z}~,~\psi =1.
\end{equation*}

Applying theorem (\ref{TheoremGEs}) we find that the geodesic equations of (%
\ref{NGRG.24}) admit seven Lie symmetries%
\begin{eqnarray*}
X_{1} &=&\partial _{s}~\ ,~X_{2}=s\partial _{s}~,~X_{3}=\partial
_{t}~~,~X_{4}=z\partial _{y}-y\partial _{z} \\
X_{5} &=&\partial _{y}~~,~~X_{6}=\partial _{z}~~,~X_{7}=\frac{4}{3}t\partial
_{t}+\frac{4}{3}x\partial _{x}+\frac{1}{3}y\partial _{y}+\frac{1}{3}%
z\partial _{z}
\end{eqnarray*}%
with Lie algebra is given in Table \ref{LieGodel}.

\begin{table}[tbp] \centering%
\caption{Lie algebra of the Taub geodesic symmetries}%
\begin{tabular}{cccccccc}
\hline\hline
$\left[ X_{I},X_{J}\right] $ & \thinspace \thinspace $X_{1}$ & $X_{2}$ & $%
X_{3}$ & $X_{4}$ & $X_{5}$ & $X_{6}$ & \thinspace $X_{7}$ \\ \hline
$X_{1}$ & $0$ & \thinspace $X_{1}$ & $0$ & $0$ & $0$ & $0$ & $0$ \\
$X_{2}$ & \thinspace $-X_{1}$ & $0$ & $0$ & $0$ & $0$ & $0$ & $0$ \\
$X_{3}$ & $0$ & $0$ & $0$ & $0$ & $0$ & $0$ & $\frac{4}{3}X_{3}$ \\
$X_{4}$ & $0$ & $0$ & $0$ & $0$ & \thinspace $X_{6}$ & \thinspace \thinspace
\thinspace $-X_{5}$ & $0$ \\
$X_{5}$ & $0$ & $0$ & $0$ & $-X_{6}$ & $0$ & $0$ & $\frac{1}{3}X_{5}$ \\
$X_{6}$ & $0$ & $0$ & $0$ & $X_{5}$ & $0$ & $0$ & $\frac{1}{3}X_{6}$ \\
$X_{7}$ & $0$ & $0$ & $-\frac{4}{3}X_{7}$ & $0$ & \thinspace $-$\thinspace
\thinspace \thinspace $\frac{1}{3}X_{5}$ & \thinspace \thinspace \thinspace $%
-\frac{1}{3}X_{6}$ & $0$ \\ \hline\hline
\end{tabular}%
\label{TaubL}%
\end{table}%

Similarly from Theorem \ref{TheoremGE2} we have that the geodesic Lagrangian
(\ref{NGRG.24a}) admits a six dimensional Noether algebra, with elements
\begin{equation*}
X_{1},~X_{3},~X_{4}~,~X_{5}~,~X_{6}~,~X_{8}=2X_{2}+X_{7}
\end{equation*}%
and correspoding Noether integrals%
\begin{eqnarray*}
\phi _{3} &=&x^{-\frac{1}{2}}t^{\prime }~,~\phi _{5}=xy^{\prime } \\
\phi _{6} &=&xz^{\prime }~,~\phi _{4}=x\left( zy^{\prime }-yz^{\prime
}\right) \\
\phi _{1} &=&x^{-\frac{1}{2}}\left( -t^{\prime 2}+x^{\prime 2}\right)
+x\left( y^{\prime 2}+z^{\prime 2}\right) \\
\phi _{8} &=&s\phi _{1}-\frac{1}{3}x^{-\frac{1}{2}}\left[ 4xx^{\prime }+x^{%
\frac{3}{2}}\left( yy^{\prime }+zz^{\prime }\right) -4tt^{\prime }\right] .
\end{eqnarray*}

\subsection{The geodesic symmetries of a 1+3 decomposable spacetime metric}

We consider next the metric which is a 1+3 decomposable, that is it has the
form:\
\begin{equation}
ds_{4}=-d\tau ^{2}+U^{2}\delta _{\alpha \beta }dx^{\alpha }dx^{\beta }
\label{RW.2}
\end{equation}%
where Greek indices take the values $1,2,3.$ It is well known \cite{TNA}
that this metric admits 15 CKVs. Seven of these vectors are KVs (the six
nongradient KVs of the 3-metric $\mathbf{r}_{\mu \nu },\mathbf{I}_{\mu }$
plus the gradient KV $\partial _{\tau })$ and nine proper CKVs. The vectors
of this conformal algebra are shown in Table \ref{C13me}.

\begin{table}[tbp] \centering%
\caption{The conformal algebra of the 1+3 metric}{\small $%
\begin{tabular}{cccccc}
\hline\hline
$\mathbf{K}$ & \textbf{CKVs of }$ds_{3}^{2}$ & $\mathbf{\psi }_{3}$ &
\textbf{\#} & \textbf{CKVs of }$ds_{1+3}^{2}$ & $\mathbf{\psi }_{1+3}$ \\
\hline
$1$ & $H=x^{a}\partial _{a}$ & $\psi _{+}(H)=U\left( 1-\frac{1}{4}x^{\alpha
}x_{\alpha }\right) $ & 1 & $H_{1}^{+}=-\psi _{+}(H)\cos \tau \partial
_{\tau }+H\sin \tau $ & $\psi _{+}(H)\sin \tau $ \\
$1$ & $H=x^{a}\partial _{a}$ & $\psi _{+}(H)=U\left( 1-\frac{1}{4}x^{\alpha
}x_{\alpha }\right) $ & 1 & $H_{2}^{+}=\psi _{+}(H)\sin \tau \partial _{\tau
}+H\cos \tau $ & $\psi _{+}(H)\cos \tau $ \\
$1$ & $C_{\mu }=\left( \delta _{\mu }^{a}-\frac{1}{2}Ux_{\mu }x^{a}\right)
\partial _{\alpha }$ & $\psi _{+}(C_{\mu })=-Ux^{\mu }$ & 3 & $Q_{\mu
}^{+}=-\psi _{+}(C_{\mu })\cos \tau \partial _{\tau }+C_{\mu }\sin \tau $ & $%
\psi _{+}(C_{\mu })\sin \tau $ \\
$1$ & $C_{\mu }=\left( \delta _{\mu }^{\alpha }-\frac{1}{2}Ux_{\mu
}x^{a}\right) \partial _{\alpha }$ & $\psi _{+}(C_{\mu })=-Ux^{\mu }$ & 3 & $%
Q_{\mu +3}^{+}=\psi _{+}(C_{\mu })\sin \tau \partial _{\tau }+C_{\mu }\cos
\tau $ & $\psi _{+}(C_{\mu })\cos \tau $ \\
$-1$ & $H=x^{\alpha }\partial _{\alpha }$ & $\psi _{-}(H)=U\left( 1+\frac{1}{%
4}x^{\alpha }x_{\alpha }\right) $ & 1 & $H_{1}^{-}=\psi _{-}(H)\cosh \tau
\partial _{\tau }+H\sinh \tau $ & $\psi _{-}(H)\sinh \tau $ \\
$-1$ & $H=x^{\alpha }\partial _{\alpha }$ & $\psi _{-}(H)=U\left( 1+\frac{1}{%
4}x^{\alpha }x_{\alpha }\right) $ & 1 & $H_{2}^{-}=\psi _{-}(H)\sinh \tau
\partial _{\tau }+H\cosh \tau $ & $\psi _{-}(H)\cosh \tau $ \\
$-1$ & $C_{\mu }=\left( \delta _{\mu }^{\alpha }+\frac{1}{2}Ux_{\mu
}x^{a}\right) \partial _{\alpha }$ & $\psi _{-}(C_{\mu })=Ux^{\mu }$ & 3 & $%
Q_{\mu }^{-}=\psi _{-}(C_{\mu })\cosh \tau \partial _{\tau }+C_{\mu }\sinh
\tau $ & $\psi _{-}(C_{\mu })\sinh \tau $ \\
$-1$ & $C_{\mu }=\left( \delta _{\mu }^{\alpha }+\frac{1}{2}Ux_{\mu
}x^{a}\right) \partial _{\alpha }$ & $\psi _{-}(C_{\mu })=Ux^{\mu }$ & 3 & $%
Q_{\mu +3}^{-}=\psi _{-}(C_{\mu })\sinh \tau \partial _{\tau }+C_{\mu }\cosh
\tau $ & $\psi _{-}(C_{\mu })\cosh \tau $ \\ \hline\hline
\end{tabular}%
$}\label{C13me}%
\end{table}%

According to Theorem \ref{TheoremGE2} this metric admits the following
Noether symmetries
\begin{equation*}
\partial _{s}\,~,~~s\partial _{\tau }~~,~~\mathbf{r}_{\mu \nu }~~,~~\mathbf{I%
}_{\mu }~~,~~\partial _{\tau }~
\end{equation*}%
with first integrals%
\begin{align*}
\phi _{s}& =\frac{1}{2}g_{ij}x^{\prime i}x^{\prime j}~ \\
\phi _{\tau }& =\tau ^{\prime }~~,~\ \phi _{\tau +1}=s\tau ^{\prime }-\tau \\
\phi _{I}& =\mathbf{I}_{i}^{I}x^{\prime i}~~,~~\phi _{r}=~\mathbf{r}_{\left(
AB\right) _{i}}x^{\prime j}.
\end{align*}%
The Lie symmetries of the geodesic equations of (\ref{RW.2}) are the Noether
symmetries plus the vectors $s\partial _{s}~,~\tau \partial _{s}$.

\subsection{The geodesic symmetries of the FRW metrics}

In a recent paper Bokhari and Kara \cite{Bok07} studied the Lie symmetries
of the conformally flat Friedman Robertson Walker (FRW) metric with the view
to understand how Noether symmetries compare with conformal Killing vectors.
More specifically they considered the conformally flat FRW\ metric\footnote{%
The second metric $ds^{2}=-t^{-\frac{4}{3}}dt^{2}+dx^{2}+dy^{2}+dz^{2}$ they
consider is the Minkowski metric whose Lie and Noether symmetries are well
known.}

\begin{equation*}
ds^{2}=-dt^{2}+t^{\frac{4}{3}}\left( dx^{2}+dy^{2}+dz^{2}\right)
\end{equation*}%
and found that the Noether symmetries are the seven vectors%
\begin{equation*}
\partial _{s},~S^{J}~,~r_{AB}
\end{equation*}%
where $S_{J}$ are the gradient KVs $\partial _{x}\partial _{y},\partial _{z}$
and $r_{AB}$ are the three nongradient KVs (generating $so\left( 3\right) $)
whereas the vector $\partial _{s}$ counts for the gauge freedom in the
affine parametrization of the geodesics. Therefore they confirm our Theorem %
\ref{TheoremGE2} that the Noether vectors coincide with the KVs and the HV
of the metric. Furthermore their claim that `\textit{...the conformally
transformed Friedman model admits additional conservation laws not given by
the Killing or conformal Killing vectors}' \ is not correct.

In the following lines, we compute all the Noether point symmetries of the
FRW\ spacetimes. To do that we have to have the homothetic algebra of these
models \cite{MM86}. There are two cases to consider, the conformally flat
models $(K=0)$ and the non conformally flat models ($K\neq 0)$.

In the following we need the conformal algebra of the flat metric, which in
Cartesian coordinates (See example \ref{ExCAflat}{}) is given in Table \ref%
{conflat}.

\begin{table}[tbp] \centering%
\caption{The conformal algebra of a flat 3d metric}%
\begin{tabular}{ccccc}
\hline\hline
\textbf{CKV} & \textbf{Components} & \textbf{\#} & $\mathbf{\psi (\xi )}$ &
\textbf{Comment} \\ \hline
$P_{I}$ & $\partial _{I}$ & $3$ & $0$ & gradient KV \\
$r_{AB}$ & $2\delta _{\lbrack A}^{d}x_{B]}\partial _{d}$ & $3$ & $0$ &
nongradient KV \\
$H$ & $x^{a}\partial _{a}$ & $1$ & $1$ & gradient HV \\
$K_{I}$ & $\left[ 2x_{I}x^{d}-\delta _{I}^{d}(x_{a}x^{a})\right] \partial
_{d}$ & $3$ & $2x_{I}$ & nongradient SCKV \\ \hline\hline
\end{tabular}%
\label{conflat}%
\end{table}%

Case A\textbf{:} $K\neq 0$

The metric is
\begin{equation}
ds=R^{2}\left( \tau \right) \left[ -d\tau ^{2}+\frac{1}{\left( 1+\frac{1}{4}%
Kx^{i}x_{i}\right) ^{2}}\left( dx^{2}+dy^{2}+dz^{2}\right) \right] .
\label{RW.3}
\end{equation}%
For a general $R\left( \tau \right) $ this metric admits the nongradient KVs
$~\mathbf{P}_{I},~\mathbf{r}_{\mu \nu }$ (see Table 4) and does not admit a
HV. Therefore the Noether symmetries of the geodesic Lagrangian
\begin{equation*}
L=-\frac{1}{2}R^{2}\left( \tau \right) t^{\prime 2}+\frac{1}{2}\frac{%
R^{2}\left( \tau \right) }{\left( 1+\frac{1}{4}Kx^{i}x_{i}\right) ^{2}}%
\left( x^{\prime 2}+y^{\prime 2}+z^{\prime 2}\right)
\end{equation*}%
of the FRW metric (\ref{RW.3}) are:%
\begin{equation*}
\partial _{s}~,~\mathbf{P}_{I}~,~\mathbf{r}_{\mu \nu }
\end{equation*}%
with Noether integrals
\begin{equation}
\phi _{s}=\frac{1}{2}g_{ij}x^{\prime i}x^{\prime j}~,~\phi _{I}=\mathbf{P}%
_{i}^{I}x^{\prime i}~,~\phi _{r}=~\mathbf{r}_{\left( AB\right)
_{i}}x^{\prime j}.
\end{equation}%
Concerning the Lie symmetries we note that the FRW spacetimes do not admit
ACs \textbf{\ }\cite{MaartensAC} and furthermore does not admit gradient
KVs. Therefore they do not admit special PCs. The Lie symmetries of these
spacetimes are then%
\begin{equation*}
\partial _{s},~s\partial _{s},~\mathbf{P}_{I}~,~\mathbf{r}_{\mu \nu }.
\end{equation*}

For special functions $R(\tau )$ it is possible to have more KVs and HV. In
Table \ref{SpecialK} we give the special forms of the scale factor $R(t)$
and the corresponding extra KVs and HV for $K=\pm 1$.
\begin{table}[tbp] \centering%
\caption{The special forms of the scale factor for |K|=1}$%
\begin{tabular}[t]{cccccc}
\hline\hline
$\mathbf{K}$ & Proper CKV & \# & Conformal Factor & $R(\tau )$\ for KVs~ & $%
R(\tau )$\ for HV \\ \hline
$\pm 1$ & $\mathbf{P}_{{\tau }}=\partial _{{\tau }}$ & $1$ & $(\ln R\left(
\tau \right) ),_{\tau }$ & $c~$ & $\exp \left( \tau \right) $ \\
$1$ & $\mathbf{H}_{1}^{+}$ & $1$ & -$\frac{\psi _{+}(\mathbf{H})}{R\left(
\tau \right) }\left( R\left( \tau \right) \cos \tau \right) ,_{\tau }~$ & $%
\frac{c}{\cos \tau }$ & $\nexists $ \\
$1$ & $\mathbf{H}_{2}^{+}$ & $1$ & $\frac{\psi _{+}(\mathbf{H})}{R(\tau )}%
\left( R\left( \tau \right) \sin \tau \right) ,_{\tau }$ & $\frac{c}{\sin
\tau }$ & $\nexists $ \\
$1$ & $\mathbf{Q}_{\mu }^{+}$ & $3$ & -$\frac{\psi _{+}(\mathbf{C}_{\mu })}{%
R(\tau )}\left( R\left( \tau \right) \cos \tau \right) ,_{\tau }$ & $\frac{c%
}{\cos \tau }$ & $\nexists $ \\
$1$ & $\mathbf{Q}_{\mu +3}^{+}$ & $3$ & $\frac{\psi _{+}(\mathbf{C}_{\mu })}{%
R(\tau )}\left( R\left( \tau \right) \sin \tau \right) ,_{\tau }$ & $\frac{c%
}{\sin \tau }$ & $\nexists $ \\
$-1$ & $\mathbf{H}_{1}^{-}$ & $1$ & $\frac{\psi _{-}(\mathbf{H})}{R(\tau )}%
\left( R\left( \tau \right) \cosh \tau \right) ,_{\tau }$ & $\frac{c}{\cosh
\tau }$ & $\nexists $ \\
$-1$ & $\mathbf{H}_{2}^{-}$ & $1$ & $\frac{\psi _{-}(\mathbf{H})}{R(\tau )}%
\left( R\left( \tau \right) \sinh \tau \right) ,_{\tau }$ & $\frac{c}{\sinh
\tau }$ & $\nexists $ \\
$-1$ & $\mathbf{Q}_{\mu }^{-}$ & $3$ & $\frac{\psi _{-}(\mathbf{C}_{\mu })}{%
R(\tau )}\left( R\left( \tau \right) \cosh \tau \right) ,_{\tau }$ & $\frac{c%
}{\cosh \tau }$ & $\nexists $ \\
$-1$ & $\mathbf{Q}_{\mu +3}^{-}$ & $3$ & $\frac{\psi _{-}(\mathbf{C}_{\mu })%
}{R(\tau )}\left( R\left( \tau \right) \sinh \tau \right) ,_{\tau }$ & $%
\frac{c}{\sinh \tau }$ & $\nexists $ \\ \hline\hline
\end{tabular}%
$\label{SpecialK}%
\end{table}%

From Table \ref{SpecialK} we infer the following additional Noether
symmetries of the FRW-like Lagrangian for special forms of the scale factor

Case A(1): $R\left( t\right) =c=$constant, the space is the 1+3 decomposable.

Case A(2) $K=1,~R\left( t\right) =\exp \left( \tau \right) .~$In this case
the space is flat and admits as Lie point symmetries the $sl\left(
4+2,R\right) .$

Case A(3a)~$K=1,~R\left( \tau \right) =\frac{c}{\cos \tau }$. In this case
we have the additional non-gradient KVs $H_{1}^{+},~Q_{\mu }^{+}.$ Therefore
the Noether symmetries are:%
\begin{equation*}
\partial _{s}~,~\mathbf{P}_{I}~,~\mathbf{r}_{\mu \nu }~,H_{1}^{+},~Q_{\mu
}^{+}~
\end{equation*}%
with Noether Integrals
\begin{equation*}
\phi _{s}~,~\phi _{I}~,~\phi _{r}~,~\phi _{H_{1}^{+}}=\left(
H_{1}^{+}\right) _{i}x^{\prime i}~~\text{and }~\phi _{Q_{\mu }^{+}}=\left(
Q_{\mu }^{+}\right) _{i}x^{\prime i}.
\end{equation*}%
The Lie symmetries are%
\begin{equation*}
\partial _{s}~,~s\partial _{s}~,~\mathbf{P}_{I}~,~\mathbf{r}_{\mu \nu
}~,H_{1}^{+},~Q_{\mu }^{+}
\end{equation*}

Case A(3b)~$K=1,~R\left( \tau \right) =\frac{c}{\sin \tau }.$ In this case
we have the two nongradient KVs $H_{2}^{+},~Q_{\mu +3}^{+}.$ The Noether
Symmetries are%
\begin{equation*}
\partial _{s}~,~\mathbf{P}_{I}~,~\mathbf{r}_{\mu \nu }~,H_{2}^{+},~Q_{\mu
+3}^{+}~~:~f=\text{constant}~
\end{equation*}%
with Noether Integrals
\begin{equation*}
\phi _{s}~,~\phi _{I}~,~\phi _{r}~,~\phi _{H_{2}^{+}}=\left(
H_{2}^{+}\right) _{i}x^{\prime i}~~\text{and }~\phi _{Q_{\mu +3}^{+}}=\left(
Q_{\mu +3}^{+}\right) _{i}x^{\prime i}.
\end{equation*}%
The Lie symmetries are{\large \ }%
\begin{equation*}
\partial _{s}~,~s\partial _{s}~,~\mathbf{P}_{I}~,~\mathbf{r}_{\mu \nu
}~,H_{2}^{+},~Q_{\mu +3}^{+}.
\end{equation*}

\textbf{Case A(4a)} $K=-1,~R\left( \tau \right) =\frac{c}{\cosh \tau }$ . In
this case we have the two additional nongradient KVs $H_{1}^{-},~Q_{\mu
}^{\_}.$ The Noether Symmetries are%
\begin{equation*}
\partial _{s}~,~\mathbf{P}_{I}~,~\mathbf{r}_{\mu \nu }~,H_{1}^{-},~Q_{\mu
}^{\_}
\end{equation*}%
with Noether Integrals
\begin{equation*}
\phi _{s}~,~\phi _{I}~,~\phi _{r}~,~\phi _{H_{1}^{-}}=\left(
H_{1}^{-}\right) _{i}x^{\prime i}~~\text{and }~\phi _{Q_{\mu }^{-}}=\left(
Q_{\mu }^{-}\right) _{i}x^{\prime i}.
\end{equation*}%
The Lie symmetries are{\large \ }%
\begin{equation*}
\partial _{s}~,~s\partial _{s}~,~\mathbf{P}_{I}~,~\mathbf{r}_{\mu \nu
}~,H_{1}^{-},~Q_{\mu }^{\_}.
\end{equation*}

Case A(4b)~$K=-1,~R\left( \tau \right) =\frac{c}{\sin \tau },$ we have the
nongradient KV $H_{2}^{-},~Q_{\mu +3}^{\_}$.~The Noether Symmetries are%
\begin{equation*}
\partial _{s}~,~\mathbf{P}_{I}~,~\mathbf{r}_{\mu \nu }~,H_{2}^{-},~Q_{\mu
+3}^{\_}~
\end{equation*}%
with Noether Integrals
\begin{equation*}
\phi _{s}~,~\phi _{I}~,~\phi _{r}~,~\phi _{H_{1}^{-}}=\left(
H_{2}^{-}\right) _{i}x^{\prime i}~\text{and }~\phi _{Q_{\mu +3}^{-}}=\left(
Q_{\mu +3}^{-}\right) _{i}x^{\prime i}.
\end{equation*}%
The Lie symmetries are{\large \ }%
\begin{equation*}
\partial _{s}~,~s\partial _{s}~,~\mathbf{P}_{I}~,~\mathbf{r}_{\mu \nu
}~,H_{2}^{-},~Q_{\mu +3}^{\_}.
\end{equation*}

Case B: $K=0$

In this case the metric is
\begin{equation*}
ds=R^{2}\left( t\right) \left( -dt^{2}+dx^{2}+dy^{2}+dz^{2}\right)
\end{equation*}%
and admits three nongradient KVs $\mathbf{P}_{I}$ and three nongradient KVs $%
\mathbf{r}_{AB}$. Therefore the Noether symmetries are

\begin{equation*}
\partial _{s}~,\mathbf{P}_{I}~,~\mathbf{r}_{AB}~:~f=\text{constant}~
\end{equation*}%
with Noether Integrals%
\begin{equation*}
\phi _{s}=\frac{1}{2}g_{ij}x^{\prime i}x^{\prime j}~,~\phi _{P_{I}}=\mathbf{P%
}_{i}^{I}x^{\prime i}~\text{and }\phi _{F}=~\left( \mathbf{r}_{AB}\right)
_{i}x^{\prime i}.
\end{equation*}%
The Lie symmetries are%
\begin{equation*}
\partial _{s}~,~s\partial _{s}~,~\mathbf{P}_{I}~,~\mathbf{r}_{AB}.
\end{equation*}

Again for special forms of the scale factor one obtains extra KVs and HV as
shown in Table \ref{SpecialK0}.

\begin{table}[tbp] \centering%
\caption{The special forms of the scale factor for K=0
}%
\begin{tabular}{ccccc}
\hline\hline
\textbf{\#} & Proper CKV & Conformal Factor $\psi $ & $R(\tau )$ for\ KVs~ &
$R\left( \tau \right) $ for HV \\ \hline
\textit{1} & $\mathbf{P}_{{\tau }}=\partial _{{\tau }}$ & $(\ln R(\tau ))_{,{%
\tau }}$ & $c$ & $\exp \left( \tau \right) $ \\
\textit{3} & $\mathbf{M}_{{\tau \alpha }}=x_{{\alpha }}\partial _{{\tau }}+{%
\tau }\partial _{{\alpha }}$ & $x_{{\alpha }}(\ln R(\tau ))_{,{\tau }}$ & $c$
& $\nexists $ \\
\textit{1} & $\mathbf{H}=\mathbf{P}_{{\tau }}+x^{a}\partial _{a}$ & $\mathbf{%
\tau }(\ln R(\tau ))+1$ & $c/\tau $ & $\nexists $ \\
\textit{1} & $\mathbf{K}_{{\tau }}=2{\tau }\mathbf{H}+\left( x^{c}x_{c}-\tau
^{2}\right) \partial _{{\tau }}$ & $-(\ln R(\tau ))_{,\tau }\left( -\tau
^{2}+r^{2}\right) +2\epsilon \tau $ & $\nexists $ & $\nexists $ \\
\textit{3} & $\mathbf{K}_{\mu }=2x_{\mu }\mathbf{H}-\left( x^{c}x_{c}-\tau
^{2}\right) \partial _{\mu }~$ & $2x_{\mu }\left[ \tau (\ln R(\tau ))_{,\tau
}+1\right] $ & $c/\tau $ & $\nexists $ \\ \hline\hline
\end{tabular}%
\label{SpecialK0}%
\end{table}%

From Table \ref{SpecialK0} we have the following special cases.

Case B(1): $R\left( t\right) =c=$constant. Then the space is the Minkowski
space and admits as Lie symmetries the $sl\left( 4+2,R\right) .$

Case B(2): $R\left( t\right) =\exp \left( \tau \right) .$ Then $\mathbf{P}_{{%
\tau }}~$becomes a gradient HV $\left( \psi =1,~\text{gradient function}~%
\frac{1}{2}\exp \left( 2\tau \right) \right) .$ Hence the Noether symmetries
are%
\begin{equation*}
\partial _{s}~,\mathbf{P}_{I}~,~\mathbf{r}_{AB}~,~2s\partial _{s}+\mathbf{P}%
_{{\tau }}~,~~s^{2}\partial _{s}+s\mathbf{P}_{{\tau }}~
\end{equation*}%
with Noether Integrals
\begin{equation*}
\phi _{s}~,~\phi _{P_{I}}~,~\phi _{F}~,~\phi _{\mathbf{P}_{{\tau }%
}}=sg_{ij}x^{\prime i}x^{\prime j}-g_{ij}\left( \mathbf{P}_{{\tau }}\right)
^{i}x^{\prime j}~\text{and }~\phi _{\mathbf{Y+1}}=\frac{1}{2}s^{2}g_{ij}\dot{%
x}^{i}\dot{x}^{j}-s\left( \mathbf{P}_{{\tau }}\right) _{i}x^{\prime i}+%
\mathbf{P}_{{\tau }}.
\end{equation*}%
The Lie symmetries are{\large \ }%
\begin{equation*}
\partial _{s}~,~s\partial _{s}~,~~\mathbf{P}_{I}~,~\mathbf{r}_{AB}~,~\mathbf{%
P}_{{\tau }}~,s^{2}\partial _{s}+s\mathbf{P}_{{\tau }}.~
\end{equation*}

Case B(3\textbf{):} $R\left( t\right) =\tau ^{-1}~.$ Then $\ $we have four
additional nongradient KVs, the $\mathbf{H},$ and $\mathbf{K}_{\mu },$ and
the Noether symmetries are:
\begin{equation*}
\partial _{s}~,~\mathbf{P}_{I},~\mathbf{r}_{AB}~,~\mathbf{H~,~K}_{\mu }
\end{equation*}%
with Noether Integrals
\begin{equation*}
\phi _{s}~,~\phi _{P_{I}}~,~\phi _{F}~,~\phi _{~\mathbf{H}}=\left( \mathbf{H}%
\right) _{i}x^{\prime i}~~\text{and }~\phi _{\mathbf{K}_{\mu }}=\left(
\mathbf{K}_{\mu }\right) _{i}x^{\prime i}.
\end{equation*}%
The Lie symmetries are:{\large \ }%
\begin{equation*}
\partial _{s}~,~s\partial _{s}~,~\mathbf{P}_{I},~\mathbf{r}_{AB}~,~\mathbf{%
H~,~K}_{\mu }.
\end{equation*}

\section{Conclusion}

We derived the symmetry conditions for the admittance of a Lie point
symmetry by the equations of autoparallels (paths) in an affine space. The
important conclusion is that the Lie symmetry vector is an Affine
Collineation in the jet space $\{t,x^{i}\}$ ( it preserves both the
autoparallels and their parametrization) while in the space $\{x^{i}\}$ the
vectors $\eta ^{i}(t,x)\partial _{x^{i}}$ are projective collineations (they
preserve the autoparallels but not necessarily their parametrization).

The symmetry conditions are applied to the geodesics of a Riemannian space
were they are solved and the generic Lie symmetry vector is obtained in
terms of the special projective algebra (and its degeneracies KVs, HKV, ACs)
of the metric. Furthermore we derived the Noether symmetries of the geodesic
Lagrangian and it was proved that Noether symmetries are generateted from
the homothetic algebra of the metric. We applied the results to the case of
Einstein spaces and obtained the Lie symmetry vectors in terms of the KVs of
the metric, in agreement with the conjecture made in \cite{FMkvs}.

Finally, the Lie and the Noether symmetries of the geodesic equations were
computed in the G\"{o}del spacetime, the Taub spacetime and the Friedman
Robertson Walker spacetimes. In each case the Noether symmetries were
computed explicitly together with the corresponding first integrals.

\newpage%

\begin{subappendices}%

\section{The determining equations}

\label{appendixLieCon}

Below we calculate the determining equations for the system (\ref{de.1}).
Let~$X=\xi \left( t,x^{k}\right) \partial _{t}+\eta ^{i}\left(
t,x^{k}\right) \partial _{t}~$be the infinitesimal generator of a one
parameter point transformation. $X$ is a Lie symmetry of (\ref{de.1}) if the
following condition, holds%
\begin{equation}
\eta _{\left[ 2\right] }^{i}=-X^{\left[ 1\right] }\left( \Gamma _{jk}^{i}%
\dot{x}^{j}\dot{x}^{k}+\sum\limits_{m=0}^{n}P_{j_{1}...j_{m}}^{i}\dot{x}%
^{j_{1}}\ldots \dot{x}^{j_{m}}\right)  \label{symcon.0}
\end{equation}%
where
\begin{equation*}
X^{\left[ 2\right] }=X+\eta _{\left[ 1\right] }^{i}\partial _{\dot{x}%
^{i}}+\eta _{\left[ 2\right] }^{i}\partial _{\dot{x}^{i}}
\end{equation*}%
is the second prolongation of $X$ and $\eta _{\left[ 1\right] }^{i},\eta _{%
\left[ 2\right] }^{i}$ are the prolongation functions%
\begin{equation*}
\eta _{\left[ 1\right] }^{i}=\left( \eta _{,t}^{i}\right) +\left( \eta
_{,k}^{i}-\xi _{,t}\delta _{k}^{i}\right) \dot{x}^{k}-\left( \xi
_{,j}\right) \dot{x}^{i}\dot{x}^{j}
\end{equation*}%
\begin{eqnarray*}
\eta _{\left[ 2\right] }^{i} &=&\left( \eta _{,tt}^{i}\right) +\left( 2\eta
_{,tj}^{i}-\xi _{,tt}\delta _{j}^{i}\right) \dot{x}^{j}+\left( \eta
_{,jk}^{i}-2\xi _{,tj}\delta _{k}^{i}\right) \dot{x}^{j}\dot{x}^{k}+ \\
&&-\left( \xi _{,jk}\right) \dot{x}^{i}\dot{x}^{j}\dot{x}^{k}+\left( \eta
_{,j}^{i}-\xi _{,j}\dot{x}^{i}\right) \ddot{x}^{j}-2\ddot{x}^{i}\left( \xi
_{,j}\dot{x}^{j}+\xi _{,t}\right) .
\end{eqnarray*}

Replacing $\ddot{x}^{i}$ from (\ref{de.1}) we find eventually
\begin{eqnarray*}
\eta _{\left[ 2\right] }^{i} &=&\left( \eta _{,tt}^{i}-\eta _{,j}^{i}P^{j}\
+2\xi _{,t}P^{i}\right) + \\
&&+\left( 2\eta _{,tr}^{i}-\xi _{,tt}-\eta _{,j}^{i}P_{r}^{j}+\xi
_{,j}P^{j}\delta _{r}^{i}+2\xi _{,r}P^{i}+2\xi _{,t}P_{r}^{i}\right) \dot{x}%
^{r}+ \\
&&+\left(
\begin{array}{c}
\eta _{,rk}^{i}-2\xi _{,tr}\delta _{k}^{i}-\eta _{,j}^{i}P_{rk}^{j}-\eta
_{,j}^{i}\Gamma _{rk}^{j}+ \\
+\xi _{,j}P_{r}^{j}\delta _{k}^{i}+2\xi _{,r}P_{k}^{i}+2\xi _{,t}\Gamma
_{rk}^{i}+2\xi _{,t}P_{rk}^{i}%
\end{array}%
\right) \dot{x}^{r}\dot{x}^{k} \\
&&+\left(
\begin{array}{c}
-\xi _{,rk}\delta _{s}^{i}-\eta _{,j}^{i}P_{rks}^{j}+\xi
_{,j}P_{rk}^{j}\delta _{s}^{i}+ \\
+3\xi _{,r}\Gamma _{ks}^{i}+2\xi _{,t}P_{rks}^{i}+2\xi
_{,r}P_{j_{1}j_{2}}^{i}%
\end{array}%
\right) \dot{x}^{s}\dot{x}^{r}\dot{x}^{k} \\
&&+\left[ -\left( \eta _{,j}^{i}\sum\limits_{m=4}^{n}P_{j_{1}...j_{m}}^{j}%
\dot{x}^{j_{1}}\ldots \dot{x}^{j_{m}}\right) +\xi _{,j}\dot{x}%
^{i}\sum\limits_{m=3}^{n}P_{j_{1}...j_{m}}^{j}\dot{x}^{j_{1}}\ldots \dot{x}%
^{j_{m}}\right] + \\
&&+\left[ 2\xi _{,r}\dot{x}^{r}\sum\limits_{m=3}^{n}P_{j_{1}...j_{m}}^{j}%
\dot{x}^{j_{1}}\ldots \dot{x}^{j_{m}}+\left( 2\xi
_{,t}\sum\limits_{m=4}^{n}P_{j_{1}...j_{m}}^{j}\dot{x}^{j_{1}}\ldots \dot{x}%
^{j_{m}}\right) \right]
\end{eqnarray*}%
We have computed the lhs of (\ref{symcon.0}). It remains to compute the rhs%
\begin{equation*}
X^{[1]}\left( \Gamma _{jk}^{i}\dot{x}^{j}\dot{x}^{k}+\sum%
\limits_{m=0}^{n}P_{j_{1}...j_{m}}^{i}\dot{x}^{j_{1}}\ldots \dot{x}%
^{j_{m}}\right) =X^{[1]}\left( \Gamma _{jk}^{i}\dot{x}^{j}\dot{x}^{k}\right)
+X^{\left[ 1\right] }\left( \sum\limits_{m=0}^{n}P_{j_{1}...j_{m}}^{i}\dot{x}%
^{j_{1}}\ldots \dot{x}^{j_{m}}\right) .
\end{equation*}

\textbf{Analysis of the term }$X^{[1]}\left( \Gamma _{jk}^{i}\dot{x}^{j}\dot{%
x}^{k}\right) \mathbf{:}$%
\begin{eqnarray*}
X^{\left[ 1\right] }\left( \Gamma _{jk}^{i}\dot{x}^{j}\dot{x}^{k}\right)
&=&X\left( \Gamma _{jk}^{i}\right) \dot{x}^{j}\dot{x}^{k}+2\Gamma _{jk}^{i}%
\dot{x}^{j}\delta \dot{x}^{k} \\
&=&\left( \xi \Gamma _{jk,t}^{i}+\eta ^{l}\Gamma _{jk,l}^{i}\right) \dot{x}%
^{j}\dot{x}^{k}+2\Gamma _{jk}^{i}\dot{x}^{j}\eta _{\left[ 1\right] }^{k}.
\end{eqnarray*}%
hence
\begin{eqnarray*}
X^{\left[ 1\right] }\left( \Gamma _{jk}^{i}\dot{x}^{j}\dot{x}^{k}\right)
&=&2\Gamma _{rs}^{i}\eta _{,t}^{s}\dot{x}^{r}-2\Gamma _{rs}^{i}\xi _{,r}\dot{%
x}^{s}\dot{x}^{r}\dot{x}^{r}+ \\
&&+\left( \xi \Gamma _{rk,t}^{i}+\eta ^{l}\Gamma _{rk,l}^{i}+2\Gamma
_{rs}^{i}\eta _{,k}^{s}-2\Gamma _{rs}^{i}\xi _{,t}\delta _{k}^{s}\right)
\dot{x}^{r}\dot{x}^{k}.
\end{eqnarray*}

\textbf{Analysis of the term }$X^{\left[ 1\right] }\left(
\sum\limits_{m=0}^{n}P_{j_{1}...j_{m}}^{i}\dot{x}^{j_{1}}\ldots \dot{x}%
^{j_{m}}\right) :$%
\begin{equation*}
X^{\left[ 1\right] }\left( \sum\limits_{m=0}^{n}P_{j_{1}...j_{m}}^{i}\dot{x}%
^{j_{1}}\ldots \dot{x}^{j_{m}}\right) =\sum\limits_{m=0}^{n}\left(
XP_{j_{1}...j_{m}}^{i}\right) \dot{x}^{j_{1}}\ldots \dot{x}%
^{j_{m}}+\sum\limits_{m=0}^{n}mP_{j_{1}...j_{m}}^{i}\dot{x}^{j_{1}}\ldots
\left( \eta _{\left[ 1\right] }^{j_{m}}\right)
\end{equation*}%
The first term becomes
\begin{eqnarray*}
\sum\limits_{m=0}^{n}\left( XP_{j_{1}...j_{m}}^{i}\right) \dot{x}%
^{j_{1}}\ldots \dot{x}^{j_{m}} &=&\sum\limits_{m=0}^{n}\left(
P_{j_{1}...j_{m}~,t}^{i}\xi +P_{j_{1}...j_{m}~,r}^{i}\eta ^{r}\right) \dot{x}%
^{j_{1}}\ldots \dot{x}^{j_{m}} \\
&=&\left( P_{,t}^{i}\xi +P_{,r}^{i}\eta ^{r}\right) +\left( P_{k,t}^{i}\xi
+P_{k,r}^{i}\eta ^{r}\right) \dot{x}^{k}+ \\
&&+\left( P_{ks,t}^{i}\xi +P_{ks,r}^{i}\eta ^{r}\right) \dot{x}^{k}\dot{x}%
^{s}+\left( P_{ksl,t}^{i}\xi +P_{ksl,r}^{i}\eta ^{r}\right) \dot{x}^{k}\dot{x%
}^{s}\dot{x}^{l}+ \\
&&+\sum\limits_{m=4}^{n}\left( P_{j_{1}...j_{m}~,t}^{i}\xi
+P_{j_{1}...j_{m}~,r}^{i}\eta ^{r}\right) \dot{x}^{j_{1}}\ldots \dot{x}%
^{j_{m}}.
\end{eqnarray*}%
The second term
\begin{eqnarray*}
\sum\limits_{m=0}^{n}mP_{j_{1}...j_{m}}^{i}\left( \delta \dot{x}%
^{j_{1}}\right) \ldots \delta \dot{x}^{j_{m}}
&=&\sum\limits_{m=0}^{n}mP_{j_{1}...j_{m}}^{i}\eta _{,t}^{j_{1}}\ldots \dot{x%
}^{j_{m-1}}+ \\
&&+\sum\limits_{m=0}^{n}mP_{j_{1}...j_{m}}^{i}\left( \eta
_{,s}^{j_{1}}-\delta _{s}^{j_{1}}\xi _{,t}\right) \ldots \dot{x}^{j_{m}}~%
\dot{x}^{s} \\
&&-\sum\limits_{m=0}^{n}mP_{j_{1}...j_{m}}^{i}\dot{x}^{j_{1}}\ldots \dot{x}%
^{j_{m}}\dot{x}^{s}\xi _{,s}.
\end{eqnarray*}%
The term $\sum\limits_{m=0}^{n}mP_{j_{1}...j_{m}}^{i}\dot{x}^{j_{1}}\ldots
\dot{x}^{j_{m-1}}\eta _{,t}^{j_{m}}$ can be written as
\begin{eqnarray*}
\sum\limits_{m=0}^{n}mP_{j_{1}...j_{m}}^{i}\dot{x}^{j_{1}}\ldots \dot{x}%
^{j_{m-1}}\eta _{,t}^{j_{m}} &=&P_{k}^{i}\eta _{,t}^{k}+2P_{ks}^{i}\eta
_{,t}^{k}\dot{x}^{s}+3P_{ksr}^{i}\eta _{,t}^{k}\dot{x}^{s}\dot{x}^{r}+ \\
&&+4P_{rksl}^{i}\eta _{,t}^{r}\dot{x}^{s}\dot{x}^{k}\dot{x}%
^{l}+\sum\limits_{m=5}^{n}mP_{j_{1}...j_{m}}^{j}\dot{x}^{j_{1}}\ldots \dot{x}%
^{j_{m}}.
\end{eqnarray*}%
that is%
\begin{eqnarray*}
\sum\limits_{m=0}^{n}mP_{j_{1}...j_{m}}^{i}\left( \eta _{,s}^{j_{1}}-\delta
_{s}^{j_{1}}\xi _{,t}\right) \ldots \dot{x}^{j_{m}}~\dot{x}^{s}
&=&P_{k}^{i}\left( \eta _{,s}^{k}-\delta _{s}^{k}\xi _{,t}\right) \dot{x}%
^{s}+2P_{kr}^{i}\left( \eta _{,s}^{k}-\delta _{s}^{k}\xi _{,t}\right) \dot{x}%
^{s}\dot{x}^{r}+ \\
&&+3P_{krl}^{i}\left( \eta _{,s}^{r}-\delta _{s}^{r}\xi _{,t}\right) \dot{x}%
^{s}\dot{x}^{r}\dot{x}^{l}+ \\
&&+\sum\limits_{m=4}^{n}mP_{j1...j_{m}}^{i}\left( \eta _{,s}^{j_{1}}-\delta
_{s}^{j_{1}}\xi _{,t}\right) \dot{x}^{s}\ldots \dot{x}^{j_{m}}
\end{eqnarray*}%
\begin{eqnarray*}
~-\sum\limits_{m=0}^{n}mP_{j_{1}...j_{m}}^{i}\dot{x}^{j_{1}}\ldots \dot{x}%
^{j_{m}}\dot{x}^{s}\xi _{,s} &=&-\left( P_{k}^{i}\xi _{,s}\right) \dot{x}^{k}%
\dot{x}^{s}+ \\
&&-2\left( P_{kr}^{i}\xi _{,s}\right) \dot{x}^{k}\dot{x}^{r}\dot{x}%
^{s}-\left( \sum\limits_{m=3}^{n}mP_{j_{1}...j_{m}}^{i}\dot{x}^{j_{1}}\ldots
\dot{x}^{j_{m}}\dot{x}^{s}\xi _{,s}\right) .
\end{eqnarray*}%
Finally we have
\begin{eqnarray*}
\delta \left( \sum\limits_{m=0}^{n}P_{j_{1}...j_{m}}^{i}\dot{x}%
^{j_{1}}\ldots \dot{x}^{j_{m}}\right) &=&\left( \left( P_{,t}^{i}\xi
+P_{,r}^{i}\eta ^{r}\right) +P_{k}^{i}\eta _{,t}^{k}\right) +\left(
\begin{array}{c}
P_{s,t}^{i}\xi +P_{s,r}^{i}\eta ^{r}+ \\
+2P_{ks}^{i}\eta _{,t}^{k}+P_{k}^{i}\left( \eta _{,s}^{k}-\delta _{s}^{k}\xi
_{,t}\right)%
\end{array}%
\right) \dot{x}^{s}+ \\
&&+\left(
\begin{array}{c}
P_{ks,t}^{i}\xi +P_{ks,r}^{i}\eta ^{r}+3P_{ksr}^{i}\eta _{,t}^{r}+ \\
+2P_{kr}^{i}\left( \eta _{,s}^{r}-\delta _{s}^{r}\xi _{,t}\right)
-P_{k}^{i}\xi _{,s}%
\end{array}%
\right) \dot{x}^{k}\dot{x}^{s}+ \\
&&+\left(
\begin{array}{c}
P_{ksl,t}^{i}\xi +P_{ksl,r}^{i}\eta ^{r}+4P_{rksl}^{i}\eta _{,t}^{r}+ \\
+3P_{krl}^{i}\left( \eta _{,s}^{r}-\delta _{s}^{r}\xi _{,t}\right)
-2P_{kr}^{i}\xi _{,s}%
\end{array}%
\right) \dot{x}^{k}\dot{x}^{s}\dot{x}^{l}+ \\
&&+\left(
\begin{array}{c}
\sum\limits_{m=4}^{n}\left( P_{j_{1}...j_{m}~,t}^{i}\xi
+P_{j_{1}...j_{m}~,r}^{i}\eta ^{r}\right) \dot{x}^{j_{1}}\ldots \dot{x}%
^{j_{m}}+ \\
+\sum\limits_{m=5}^{n}mP_{j_{1}...j_{m}}^{i}\eta _{,t}^{j_{1}}\ldots \dot{x}%
^{j_{m}}+ \\
+\sum\limits_{m=4}^{n}mP_{j_{1}...j_{m}}^{i}\left( \eta _{,s}^{j_{1}}-\delta
_{s}^{j_{1}}\xi _{,t}\right) \dot{x}^{s}\ldots \dot{x}^{j_{m}}+ \\
-\left( \sum\limits_{m=3}^{n}mP_{j_{1}...j_{m}}^{i}\dot{x}^{j_{1}}\ldots
\dot{x}^{j_{m}}\dot{x}^{s}\xi _{,s}\right)%
\end{array}%
\right) .
\end{eqnarray*}

Collecting terms and setting the coefficient of each product of $\dot{x}%
^{j_{1}}$ equal to zero we obtain determining equations (\ref{de.02})-(\ref%
{de.06}).

$\left( \dot{x}^{i}\right) ^{0}:$%
\begin{eqnarray*}
0 &=&\left( \eta _{,tt}^{i}-\eta _{,j}^{i}P^{j}\ +2\xi _{,t}P^{i}\right)
+\left( \left( P_{,t}^{i}\xi +P_{,r}^{i}\eta ^{r}\right) +P_{k}^{i}\eta
_{,t}^{k}\right) \\
&=&L_{\eta }P^{i}+2\xi _{,t}P^{i}+P_{,t}^{i}\xi +P_{k}^{i}\eta
_{,t}^{k}+\eta _{,tt}^{i}.
\end{eqnarray*}

$\left( \dot{x}^{i}\right) ^{1}:$

\begin{eqnarray*}
0 &=&\left[ 2\eta _{,tr}^{i}+2\Gamma _{rs}^{i}\eta _{,t}^{s}-\xi
_{,tt}\delta _{r}^{i}\right] -\eta _{,j}^{i}P_{r}^{j}+\left[ \xi
_{,j}P^{j}\delta _{r}^{i}+2\xi _{,r}P^{i}+2\xi _{,t}P_{r}^{i}\right] + \\
&&+P_{s,t}^{i}\xi +P_{s,r}^{i}\eta ^{r}+2P_{ks}^{i}\eta
_{,t}^{k}+P_{k}^{i}\left( \eta _{,s}^{k}-\delta _{s}^{k}\xi _{,t}\right) \\
&=&\left( 2\eta _{,t|r}^{i}-\xi _{,tt}\delta _{r}^{i}\right) +L_{\eta
}P_{r}^{i}+2P_{ks}^{i}\eta _{,t}^{k}+\left( \xi _{,r}P^{r}\delta
_{r}^{i}+2\xi _{,r}P^{i}+P_{r,t}^{i}\xi +\xi _{,t}P_{r}^{i}\right) .
\end{eqnarray*}

$\left( \dot{x}^{i}\right) ^{2}:$

\begin{eqnarray*}
0 &=&L_{\eta }\Gamma _{jk}^{i}+\xi \Gamma _{rk,t}^{i}+\left[ -2\xi
_{,tr}\delta _{k}^{i}\right] +\left[ -\eta _{,j}^{i}P_{rk}^{j}\right] +\left[
\xi _{,j}P_{r}^{j}\delta _{k}^{i}+2\xi _{,r}P_{k}^{i}\right] + \\
&&+2\xi _{,t}P_{rk}^{i}+P_{ks,t}^{i}\xi +P_{ks,r}^{i}\eta
^{r}+3P_{ksr}^{i}\eta _{,t}^{r}+2P_{kr}^{i}\left( \eta _{,s}^{r}-\delta
_{s}^{r}\xi _{,t}\right) -P_{k}^{i}\xi _{,r} \\
&=&\left( L_{X}\Gamma _{jk}^{i}-2\xi _{,tr}\delta _{k}^{i}\right) +L_{\eta
}P_{rk}^{i}+P_{ks,t}^{i}\xi +3P_{ksr}^{i}\eta _{,t}^{r}+\left( \xi
_{,j}P_{r}^{j}\delta _{k}^{i}+\xi _{,r}P_{k}^{i}\right) .
\end{eqnarray*}

$\left( \dot{x}^{i}\right) ^{3}:$

\begin{eqnarray*}
0 &=&\left[ \xi _{,rk}\delta _{s}^{i}+\xi _{,r}\Gamma _{ks}^{i}\right] -\eta
_{,j}^{i}P_{rks}^{j}+\xi _{,j}P_{rk}^{j}\delta _{s}^{i}+2\xi
_{,t}P_{rks}^{i}+P_{ksl,t}^{i}\xi + \\
&&+P_{ksl,r}^{i}\eta ^{r}+4P_{rksl}^{i}\eta _{,t}^{r}+3P_{krl}^{i}\left(
\eta _{,s}^{r}-\delta _{s}^{r}\xi _{,t}\right) -2P_{kr}^{i}\xi _{,s}+2\xi
_{,s}P_{kr}^{i} \\
&=&\left[ -\xi _{,r|k}\delta _{s}^{i}\right] +L_{\eta
}P_{krs}^{i}-P_{krs}^{i}\xi _{,t}+\xi _{,j}P_{rk}^{j}\delta
_{s}^{i}+P_{ksl,t}^{i}\xi +4P_{rksl}^{i}\eta _{,t}^{r}.
\end{eqnarray*}%
the rest terms gives

\begin{eqnarray*}
0 &=&\sum\limits_{m=4}^{n}\left( L_{\eta }P_{j_{1}...j_{m}}^{j}\dot{x}%
^{j_{1}}\ldots \dot{x}^{j_{m}}\right) +\sum\limits_{m=4}^{n}\left(
P_{j_{1}...j_{m}~,t}^{i}\xi \right) \dot{x}^{j_{1}}\ldots \dot{x}^{j_{m}}+ \\
&&+\left( 2\xi _{,t}\sum\limits_{m=4}^{n}P_{j_{1}...j_{m}}^{i}\dot{x}%
^{j_{1}}\ldots \dot{x}^{j_{m}}-\sum\limits_{m=4}^{n}mP_{j_{1}...j_{m}}^{i}%
\dot{x}^{j_{1}}\ldots \dot{x}^{j_{m}}\right) + \\
&&+\xi _{,r}\dot{x}^{r}\left( 2\sum\limits_{K=m-1}^{n}P_{j_{1}...j_{K}}^{i}%
\dot{x}^{j_{1}}\ldots \dot{x}^{j_{K}}-\sum%
\limits_{K=m-1}^{n}KP_{j_{1}...j_{m}}^{i}\dot{x}^{j_{1}}\ldots \dot{x}%
^{j_{K}}\right) + \\
&&+\left( \sum\limits_{C=m+2}^{n}CP_{j_{1}...j_{c}}^{i}\eta
_{,t}^{j_{1}}\ldots \dot{x}^{j_{c}}\right) +\xi _{,j}\dot{x}^{i}\left(
\sum\limits_{K=m-1}^{n}P_{j_{1}...j_{K}}^{j}\dot{x}^{j_{1}}\ldots \dot{x}%
^{j_{K}}\right) .
\end{eqnarray*}

\end{subappendices}%

\chapter{Motion on a curved space\label{chapter3}}

\section{Introduction}

\label{Introduction}The study of Lie point symmetries of a given system of
ODEs consists of two steps: (a)\ the determination of the conditions, which
the components of the Lie symmetry vectors must satisfy and (b) the solution
of the system of these conditions. These conditions can be quite involved,
but today it is possible to use algebraic computing programs to derive them
(for a review see \cite{But2003}). Therefore the essential part of the work
is the second step. For a small number of equations (say up to three) one
can possibly employ again computer algebra to look for a solution of the
system. However for a large number of equations such an attempt is
prohibitive and one has to go back to traditional methods to determine the
solution.

In Chapter \ref{LieSymGECh} the Lie and the Noether point symmetries of the
geodesic equations were calculated in terms of the special projective
algebra of the space. The purpose of the present chapter is to extend the
previous results and to provide an alternative way to solve the system of
Lie symmetry conditions for the second order equations of the form%
\begin{equation}
\ddot{x}^{i}+\Gamma _{jk}^{i}\dot{x}^{j}\dot{x}^{k}=F^{i}.  \label{L2P.1}
\end{equation}%
Here $\Gamma _{jk}^{i}(x^{r})$ are general functions, a dot over a symbol
indicates derivation with respect to the parameter $t$ along the solution
curves and $F^{i}(x^{j})$ is a $C^{p}$ vector field. This type of equations
is important, because it contains the equations of motion of a dynamical
system in a Riemannian space, in which the functions $\Gamma
_{jk}^{i}(x^{r}) $ are the connection coefficients of the metric and $t$
being an affine parameter along the trajectory. In the following we assume
this identification of $\Gamma _{jk}^{i}$'s\footnote{%
Of course it is possible to look for a metric for which a given set of $%
\Gamma _{jk}^{i}$ are the connection coefficients, or, even avoid the metric
altogether. However we shall not attempt this in the present work. For such
an attempt see \cite{Mahq08}.}.

The key idea, which is proposed here, is to express the system of Lie
symmetry conditions of (\ref{L2P.1}) in a Riemannian space in terms of
collineation (i.e. symmetry) conditions of the metric. If this is achieved,
then the Lie point symmetries of (\ref{L2P.1}) will be related to the
collineations of the metric, hence their determination will be transferred
to the geometric problem of determining the generators of a specific type of
collineations of the metric. One then can use of existing results of
Differential Geometry on collineations to produce the solution of the Lie
symmetry problem.

The natural question to ask is: If the Lie symmetries of the dynamical
systems moving in a given Riemannian space are from the same set of
collineations of the space, how will one select the Lie symmetries of a
specific dynamical system?\ The answer is as follows. The left hand side of
Equation (\ref{L2P.1}) contains the metric and its derivatives and it is
\emph{common to all} dynamical systems moving in the same Riemannian space.
Therefore geometry (i.e. collineations)\ enters in the left hand side of (%
\ref{L2P.1}) only. A dynamical system is defined by the force field $F^{i},$
which enters into the right hand side of (\ref{L2P.1}) only. We conclude
that there must exist constraint conditions, which will involve the
components of the collineation vectors and the force field $F^{i}$, which
will select the appropriate Lie symmetries for a specific dynamical system.

Indeed Theorem \ref{The general conservative system} (see section \ref%
{LieSymRS}) relates the Lie symmetry generators of (\ref{L2P.1}) with the
elements of the special projective Lie algebra of the space where motion
occurs, and provides these necessary constraint conditions.

What has been said for the Lie point symmetries of (\ref{L2P.1}) applies
also to Noether symmetries. The Noether symmetries are Lie point symmetries
which satisfy the constraint%
\begin{equation}
X^{\left[ 1\right] }L+L\frac{d\xi }{dt}=\frac{df}{dt}.  \label{L2p.3}
\end{equation}%
\newline
Theorem \ref{The Noether symmetries of a conservative system} (see section %
\ref{NoetherSymC}) relates the generators of Noether symmetries of (\ref%
{L2P.1}) with the homothetic algebra of the metric and provides the required
constraint conditions.

In the following sections we apply Theorem \ref{The general conservative
system} and Theorem \ref{The Noether symmetries of a conservative system} to
determine all two dimensional (section \ref{Newtonian dynamical systems
which admit Lie symmetries}) and all three dimensional (section \ref{Lie
point symmetries of three dimensional autonomous Newtonian systems})
Newtonian dynamical systems moving under the action of a general force $%
F^{i},$ which admit Lie and\ Noether point symmetries. In section \ref%
{Motion on the two dimensional sphere} we apply the results to determine the
conservative dynamical systems which move in a two-dimensional space of
constant non-vanishing curvature and admit Noether point symmetries. The
case of a conservative force has been addressed previously for the two
dimensional case by Sen \cite{Sen} and more recently by Damianou \emph{et al}
\cite{Damianou99} and for the three dimensional by \cite{Damianou2004}$.$ As
it will be shown both treatments are incomplete. We demonstrate the use of
the results in two cases. The non-conservative Kepler - Ermakov system \cite%
{Karasu,LeachK,MoyoL} and the case of the H\`{e}non Heiles type potentials
\cite{LeachHHP,Shrauner}. In both cases we recover and complete the existing
results.

\section{Lie symmetries of a dynamical system in a Riemannian space}

\label{LieSymRS}

I section \ref{TheSymCon} the Lie symmetry conditions were calculated for a
general system of ODEs polynomial in the velocities, therefore the Lie
symmetry conditions (determining equations) for equation (\ref{L2P.1}) with $%
F^{i}=F^{i}\left( t,x^{k}\right) $ are

\begin{equation}
L_{\eta }F^{i}+2\xi ,_{t}F^{i}+\xi F_{,t}^{i}+\eta ^{i},_{tt}=0
\label{PLS.09}
\end{equation}%
\begin{equation}
\left( \xi ,_{k}\delta _{j}^{i}+2\xi ,_{j}\delta _{k}^{i}\right) F^{k}+2\eta
^{i},_{t|j}-\xi ,_{tt}\delta _{j}^{i}=0  \label{PLS.10}
\end{equation}%
\begin{equation}
L_{\eta }\Gamma _{(jk)}^{i}=2\xi ,_{t(j}\delta _{k)}^{i}  \label{PLS.11a}
\end{equation}%
\begin{equation}
\xi _{(,i|j}\delta _{r)}^{k}=0.  \label{PLS.12}
\end{equation}

Equation (\ref{PLS.11a}) means that $\eta ^{i}$ is a projective collineation
of the metric with projective function $\xi _{,t}$. Furthermore, equation (%
\ref{PLS.12}) means that $\xi _{,j}$ is a gradient KV of $g_{ij}$; that is, $%
\eta ^{i}$ is a special projective collineation of the metric. Equation (\ref%
{PLS.09}) gives\footnote{%
Recall that $L_{\eta }F_{j}=F_{,jk}\eta ^{k}+\eta ^{k},_{j}F,_{k}.$}%
\begin{equation}
\left( L_{\eta }g^{ij}\right) F_{j}+g^{ij}L_{\eta }F_{j}+2\xi
_{,t}g^{ij}F_{j}+\xi g^{ij}F_{j,t}+\eta _{,tt}^{i}=0.  \label{de.20}
\end{equation}%
This equation restricts $\eta ^{i}$ further because it relates it directly
to the metric symmetries. Finally equation (\ref{PLS.10}) gives%
\begin{equation}
-\delta _{j}^{i}\xi _{,tt}+\left( \xi _{,j}\delta _{k}^{i}+2\delta
_{j}^{i}\xi _{,k}\right) F^{k}+2\eta _{,tj}^{i}+2\Gamma _{jk}^{i}\eta
_{,t}^{k}=0.  \label{de.21d}
\end{equation}

Equations (\ref{de.20}),(\ref{de.21d}) are the constraint conditions which
relate the components $\xi ,~\eta ^{i}$ of the Lie point symmetry vector
with the vector $F^{i}$.

\begin{proposition}
The Lie point symmetries of the dynamical system (\ref{L2P.1}) where $%
F^{i}=F^{i}\left( t,x^{k}\right) ,$ are generated from the special
projective algebra of the space where the motion occurs.
\end{proposition}

In the case where the dynamical system (\ref{L2P.1}) is autonomous and
conservative, that is, $F^{i}=g^{ij}V_{,j}\left( x^{k}\right) $ and $V_{,j}$
is not a gradient KV of the metric, the solution of the determining
equations is given by the following theorem (for a proof see Appendix \ref%
{proofautonomous}).~

\begin{theorem}
\label{The general conservative system} The Lie point symmetries of the
equations of motion of an autonomous conservative system
\begin{equation}
\ddot{x}^{i}+\Gamma _{jk}^{i}\dot{x}^{j}\dot{x}^{k}=g^{ij}V_{,i}
\label{PP.01}
\end{equation}%
in a general Riemannian space with metric $g_{ij},$ are given in terms of
the generators $Y^{i}$ of the special projective Lie algebra of the metric $%
g_{ij}$ as follows. \newline
\emph{Case I}\textbf{\ \ }Lie symmetries due to the affine algebra. The
resulting Lie symmetries are%
\begin{equation}
\mathbf{X}=\left( \frac{1}{2}d_{1}a_{1}t+d_{2}\right) \partial
_{t}+a_{1}Y^{i}\partial _{i}  \label{PP.02}
\end{equation}%
where $a_{1}$ and $d_{1}$ are constants, provided the potential satisfies
the condition%
\begin{equation}
\ L_{Y}V^{,i}+d_{1}V^{,i}=0.  \label{PP.03}
\end{equation}%
\emph{Case IIa}\textbf{\ \ }The Lie symmetries are generated by the gradient
homothetic algebra and $Y^{i}\neq V^{,i}$. The Lie symmetries are
\begin{equation}
\mathbf{X}=~2\psi \int T\left( t\right) dt\partial _{t}+T\left( t\right)
Y^{i}\partial _{i}  \label{PP.04}
\end{equation}%
where the function $T\left( t\right) $ is the solution of the equation $%
~T_{,tt}=a_{1}T~~$provided the potential $V(x^{i})$ satisfies the condition
\begin{equation}
L_{\mathbf{Y}}V^{,i}+4\psi V^{,i}+a_{1}Y^{i}=0.  \label{PP.03a}
\end{equation}%
\emph{Case IIb}\textbf{\ }The \ Lie symmetries are generated by the gradient
HV $Y^{i}=\kappa V^{,i},$ where $\kappa $ is a constant. In this case the
potential is the gradient HV\ of the metric and the Lie symmetry vectors are
\begin{equation}
\mathbf{X}=\left( -c_{1}\sqrt{\psi k}\cos \left( 2\sqrt{\frac{\psi }{k}}%
t\right) +c_{2}\sqrt{\psi k}\sin \left( 2\sqrt{\frac{\psi }{k}}t\right)
\right) \partial _{t}+\left( c_{1}\sin \left( 2\sqrt{\frac{\psi }{k}}%
t\right) +c_{2}\cos \left( 2\sqrt{\frac{\psi }{k}}t\right) \right)
H^{,i}\partial _{i}.  \label{PP.08}
\end{equation}%
\emph{Case IIIa} \ The \ Lie symmetries due to the proper special projective
algebra. In this case the Lie symmetry vectors are (the index $J$ counts the
gradient KVs)%
\begin{equation}
\mathbf{X}_{J}=\left( C\left( t\right) S_{J}+D\left( t\right) \right)
\partial _{t}+T\left( t\right) Y^{i}\partial _{i}  \label{PP.10}
\end{equation}%
where the functions $C(t),T(t),D(t)$ are solutions of the system of
simultaneous equations
\begin{equation}
D_{,t}=\frac{1}{2}d_{1}T~~,~~\
T_{,tt}=a_{1}T~~,~~T_{,t}=c_{2}C~~,~~D_{,tt}=d_{c}C~~,~~C_{,t}=a_{0}T
\end{equation}%
and in addition the potential satisfies the conditions
\begin{align}
L_{Y}V^{,i}+2a_{0}SV^{,i}+d_{1}V^{,i}+a_{1}Y^{i}& =0  \label{PP.11} \\
\left( S_{,k}\delta _{j}^{i}+2S,_{j}\delta _{k}^{i}\right) V^{,k}+\left(
2Y^{i}{}_{;~j}-a_{0}S\delta _{j}^{i}\right) c_{2}-d_{c}\delta _{j}^{i}& =0.
\label{PP.12}
\end{align}%
\emph{Case IIIb}\textbf{\ \ }Lie symmetries due to the proper special
projective algebra and $Y_{J}^{i}=\lambda S_{J}V^{,i},$ where~$V^{,i}$ is a
gradient HV and $S_{J}^{,i}$ is a gradient KV. The Lie symmetry vectors are%
\begin{equation}
X_{J}=\left( C\left( t\right) S_{J}+d_{1}\right) \partial _{t}+T(t)\lambda
S_{J}V^{,i}\partial _{i}
\end{equation}%
where the functions $C\left( t\right) $ and $T(t)$ are computed from the
relations
\begin{equation}
T_{,tt}+2C_{,t}=\lambda _{1}T~~,~~T_{,t}=\lambda _{2}C~~,~~C_{,t}=a_{0}T
\label{SP.20}
\end{equation}%
and the potential satisfies the conditions%
\begin{align}
L_{\mathbf{Y}_{J}}V^{,i}+\lambda _{1}S_{J}V^{,i}& =0  \label{SP.21} \\
C\left( \lambda _{1}S_{J}\delta _{j}^{i}+2S_{J},_{j}V^{,i}\right) +\lambda
_{2}\left( 2\lambda S_{J,j}V^{.i}+\left( 2\lambda S_{J}-a_{0}S_{J}\right)
\delta _{j}^{i}\right) & =0.  \label{SP.22}
\end{align}
\end{theorem}

An immediate important application of Theorem \ref{The general conservative
system} concerns the important case of spaces of constant curvature. As we
have seen already (see also \cite{Barnes}) the special projective algebra of
a space of constant curvature consists of non-gradient KVs only. Therefore
we have the following corollary.

\begin{corollary}
\label{The general conservative system CC}The Lie point symmetries of the
equations of motion of an autonomous conservative system (\ref{PP.01})$~$in
a space of constant curvature are elements of the non-gradient KVs algebra.
\end{corollary}

This implies that \ in spaces of constant curvature \ it is enough to
consider Case I only.

If the system (\ref{L2P.1}) is autonomous but not conservative moving under
the action of the external force $F^{i}$ the previous results remain valid
except the cases IIb, IIIb which are not applicable. We emphasize that (with
appropriate adjustments) the results apply to affine spaces in which there
exists a connection $\Gamma _{jk}^{i}$ but not necessarily a metric.

\section{Noether symmetries of a dynamical system in a Riemannian space}

\label{NoetherSymC}

Consider a particle moving in a Riemannian space with metric $g_{ij}$ under
the influence of the potential $V\left( t,x^{k}\right) .$ The Lagrangian
describing the motion of the particle is
\begin{equation}
L=\frac{1}{2}g_{ij}\dot{x}^{i}\dot{x}^{j}-V\left( t,x^{k}\right) .
\label{NSCS.1}
\end{equation}%
A Lie symmetry vector $\mathbf{X}=\xi \left( t,x^{k}\right) \partial
_{t}+\eta ^{i}\left( t,x^{k}\right) \partial _{x^{i}}$ is a Noether symmetry
of the Lagrangian if it satisfies the condition
\begin{equation}
\mathbf{X}^{\left[ 1\right] }L+\frac{d\xi }{dt}L=\frac{df}{dt}
\label{NSCS.1a}
\end{equation}%
where $\mathbf{X}^{\left[ 1\right] }$ is the first prolongation of $\mathbf{X%
}.$ It can be shown that condition (\ref{NSCS.1a}) is equivalent to the
system of equations:%
\begin{align}
V_{,k}\eta ^{k}+V\xi _{,t}+\xi V_{,t}& =-f_{,t}  \label{NSCS.4} \\
\eta _{,t}^{i}g_{ij}-\xi _{,j}V& =f_{,j}  \label{NSCS.5} \\
L_{\eta }g_{ij}& =2\left( \frac{1}{2}\xi _{,t}\right) g_{ij}  \label{NSCS.6}
\\
\xi _{,k}& =0.  \label{NSCS.7}
\end{align}

Equation (\ref{NSCS.7}) implies $\xi =\xi \left( t\right) $, and then from (%
\ref{NSCS.6}) follows that $\eta ^{i}$ is a HV. Therefore we have the
following result

\begin{proposition}
The Noether symmetries of the Lagrangian (\ref{NSCS.1}) are generated, from
the homothetic algebra of the metric $g_{ij}$ of the space where motion
occurs.
\end{proposition}

In the case the potential is autonomous, that is $V\left( t,x^{k}\right)
=V\left( x^{k}\right) ,$ the solution of (\ref{NSCS.4})-(\ref{NSCS.7})
relates the Noether symmetries of (\ref{NSCS.1}) with the elements of the
homothetic algebra of the metric $g_{ij}$ as follows.

\begin{theorem}
\label{The Noether symmetries of a conservative system}\ The Noether
symmetries of an autonomous conservative dynamical system moving in a
Riemannian space with metric $g_{ij}$ described by the Lagrangian (\ref%
{NSCS.1}) are generated from the homothetic Lie algebra of the metric $%
g_{ij} $ as follows.

\emph{Case I.} The KVs and the HV\ satisfy the condition:%
\begin{equation}
V_{,k}Y^{k}+2\psi _{Y}V+c_{1}=0.  \label{NSCS.14}
\end{equation}%
The Noether symmetry vector is%
\begin{equation}
\mathbf{X}=2\psi _{Y}t\partial _{t}+Y^{i}\partial _{i},~~f=c_{1}t,
\end{equation}%
where $T\left( t\right) =a_{0}\neq 0.$

\emph{Case II.} The metric admits the gradient KVs $S_{J}$, the gradient HV $%
H^{,i}$ and the potential satisfies the condition%
\begin{equation}
V_{,k}Y^{,k}+2\psi _{Y}V=c_{2}Y+d.  \label{NSCS.17}
\end{equation}%
In this case the Noether symmetry vector and the Noether function are%
\begin{equation}
\mathbf{X}=2\psi _{Y}\int T\left( t\right) dt\partial _{t}+T\left( t\right)
S_{J}^{,i}\partial _{i}~~~,~~f\left( t,x^{k}\right) =T_{,t}S_{J}\left(
x^{k}\right) ~+d\int Tdt.  \label{NSCS.17a}
\end{equation}%
and the functions $T(t)$ and $K\left( t\right) $ ($T_{,t}\neq 0$) are
computed from the relations%
\begin{equation}
T_{,tt}=c_{2}T~,~K_{,t}=d\int Tdt+\text{constant}  \label{NSCS.18}
\end{equation}%
where $c_{2}$ is a constant.

In addition to the above there is also the standard Noether symmetry $%
\partial _{t}$.
\end{theorem}

The first integrals for the Noether symmetry vectors have as follows.

\begin{proposition}
\label{The Noether Integrals}For the Noether vector $\partial _{t}$ the
Noether integral is the Hamiltonian $\emph{E.}$ For the Noether vectors of
Case I and Case II the Noether integrals are respectively:
\begin{equation}
\phi _{I}=2\psi _{Y}t\emph{E}-g_{ij}Y^{i}\dot{x}^{j}+c_{1}t  \label{NSCS.16}
\end{equation}%
\begin{equation}
\phi _{II}=2\psi _{Y}\emph{E}\int Tdt~-Tg_{ij}H^{i}\dot{x}^{j}+T_{,t}H+d\int
Tdt.  \label{NSCS.19}
\end{equation}
\end{proposition}

For the case of motion in spaces of constant curvature we have the following
result.

\begin{proposition}
\label{The Noether symmetries of a conservative system CC}The Noether
symmetry vectors of the Lagrangian (\ref{NSCS.1}) of an autonomous
conservative dynamical system moving in a space of constant curvature, are
generated by the non-gradient KVs of the space$.$ Hence only Case I survives.
\end{proposition}

In the following sections, the autonomous non linear Newtonian systems which
admit Lie and Noether point symmetries are calculated with the use of
Theorems \ref{The general conservative system} and \ref{The Noether
symmetries of a conservative system}.

\section{2D autonomous systems which admit Lie/Noether point symmetries}

\label{Newtonian dynamical systems which admit Lie symmetries}

In this section we apply Theorems \ref{The general conservative system} and %
\ref{The Noether symmetries of a conservative system} to determine \emph{all}
Newtonian dynamical systems with two degrees of freedom which admit at least
one Lie/Noether symmetry. The reason for considering this problem is that a
Lie/Noether symmetry lead to invariants/first integrals, which can be used
in many ways in order to study a given system of differential equations e.g.
to simplify, to determine the integrability of the system etc. Because the
Newtonian systems move in $E^{2}$ we need to consider the generators of the
special projective algebra of $E^{2}$ and then use the constraint conditions
for each case to determine the functional form of the force field $F^{i}.~$

\begin{table}[tbp] \centering%
\caption{Special projective algebra of the 2d Euclidian space}%
\begin{tabular}{ccc}
\hline\hline
\textbf{Collineation} & \textbf{Gradient} & \textbf{Non-gradient} \\ \hline
Killing vectors & $\partial _{x}~,~\partial _{y}~$ & $y\partial
_{x}-x\partial _{y}$ \\
Homothetic vector & $x\partial _{x}+y\partial _{y}~$ &  \\
Affine Collineation & $x\partial _{x}~,~y\partial _{y}$ & $y\partial
_{x}~,~x\partial _{y}$ \\
Sp. Projective collineation &  & $x^{2}\partial _{x}+xy\partial
_{y}~,~xy\partial _{x}+y^{2}\partial _{y}$ \\ \hline\hline
\end{tabular}%
\label{ProjectiveFlat2d}%
\end{table}%

We consider Cartesian coordinates so that the metric of the space is
\begin{equation*}
ds^{2}=dx^{2}+dy^{2}.
\end{equation*}%
In Table \ref{ProjectiveFlat2d} we give the elements of the projective Lie
algebra of $E^{2}~$in Cartesian coordinates. We note that the special
projective algebra of the two dimensional Lorentz space
\begin{equation*}
ds^{2}=-dx^{2}+dy^{2}
\end{equation*}%
is the same with that of the space $E^{2},$ with the difference that the non
gradient Killing vector is replaced with $y\partial _{x}+x\partial _{y}.$ We
shall use this observation in later chapters where we study the Lie and the
Noether point symmetries in scalar field cosmology.

We examine the cases where the force $F^{i}$ (a) is non-conservative and (b)
is conservative. In certain cases the results are common to both cases,
however for clarity it is better to consider the two cases separately.
Finally for economy of space, easy reference and convenience we present the
results in the form of Tables.

In order to indicate how the results of the Tables are obtained we consider
Case I \& II of theorem \ref{The general conservative system}. The Lie point
symmetry vectors for Case I are given by (\ref{PP.02}) i.e.%
\begin{equation}
\mathbf{X}=\left( \frac{1}{2}d_{1}a_{1}t+d_{2}\right) \partial
_{t}+a_{1}Y^{i}\partial _{i},
\end{equation}%
where $a_{1}$ and $d_{1}$ are constants and $Y^{i}$ is a vector of the
affine algebra of $E^{2}.$ The force field must satisfy condition (\ref%
{PP.03}) i.e.%
\begin{equation*}
L_{Y}\mathbf{F}+d_{1}\mathbf{F}=0.
\end{equation*}%
Writing
\begin{equation*}
\mathbf{F}=F^{x}\partial _{x}~+F^{y}~\partial _{y}~\text{and~}\mathbf{Y}%
=Y^{x}\partial _{x}~+Y^{y}~\partial _{y}
\end{equation*}%
we obtain a system of two differential equations involving the unknown
quantities $F^{x},F^{y}$ and the known quantities $Y^{x},Y^{y}.$ For each
vector $\mathbf{Y}$ we replace $Y^{x},Y^{y}$ from Table \ref%
{ProjectiveFlat2d} and solve the system to compute $F^{x},F^{y}.$ For
example for the gradient KV $\partial _{x}$ we have $Y^{x}=1,Y^{y}=0$ and
find the solution $F^{x}\left( x,y\right) =e^{-dx}f\left( y\right) ,$ $%
F^{y}\left( x,y\right) =e^{-dx}g\left( y\right) $ where $d$ is a constant
and $f\left( y\right) ,g\left( y\right) $ are arbitrary functions of their
argument. Working similarly we determine the form of the force field for all
cases of Theorem \ref{The general conservative system}. The results are
given in Tables \ref{2dN1} and \ref{2dN2}.

\begin{table}[tbp] \centering%
\caption{Two dimensional Newtonian systems admit Lie symmetries (1/4)}%
\begin{tabular}{ccc}
\hline\hline
\textbf{Lie }$\downarrow $\textbf{\ }$~~F^{i}\rightarrow $ & $\mathbf{F}%
^{x}\left( x,y\right) \mathbf{/F}^{\theta }\left( r,\theta \right) $ & $%
\mathbf{F}^{y}\left( x,y\right) \mathbf{/F}^{\theta }\left( r,\theta \right)
$ \\ \hline
$\frac{d}{2}t\partial _{t}+\partial _{x}$ & $e^{-dx}f\left( y\right) $ & $%
e^{-dx}g\left( y\right) $ \\
$\frac{d}{2}t\partial _{t}+\partial _{y}$ & $e^{-dy}f\left( x\right) $ & $%
e^{-dy}g\left( x\right) $ \\
$\frac{d}{2}t\partial _{t}+\left( y\partial _{x}-x\partial _{y}\right) $ & $%
f\left( r\right) e^{-d\theta }$ & $g\left( r\right) e^{-d\theta }$ \\
$\frac{d}{2}t\partial _{t}+x\partial _{x}+y\partial _{y}$ & $x^{\left(
1-d\right) }f\left( \frac{y}{x}\right) $ & $x^{\left( 1-d\right) }g\left(
\frac{y}{x}\right) $ \\
$\frac{d}{2}t\partial _{t}+x\partial _{x}$ & $x^{\left( 1-d\right) }f\left(
y\right) $ & $x^{-d}g\left( y\right) $ \\
$\frac{d}{2}t\partial _{t}+y\partial _{y}$ & $y^{-d}f\left( x\right) $ & $%
y^{\left( 1-d\right) }g\left( x\right) $ \\
$\frac{d}{2}t\partial _{t}+y\partial _{x}$ & $\left( \frac{x}{y}g\left(
y\right) +f\left( y\right) \right) e^{-d\frac{x}{y}}$ & $g\left( y\right)
e^{-d\frac{x}{y}}$ \\
$\frac{d}{2}t\partial _{t}+x\partial _{y}$ & $f\left( x\right) e^{-d\frac{y}{%
x}}$ & $\left( \frac{y}{x}f\left( x\right) +g\left( x\right) \right) e^{-d%
\frac{y}{x}}$ \\ \hline\hline
\end{tabular}%
\label{2dN1}%
\end{table}%

\begin{table}[tbp] \centering%
\caption{Two dimensional Newtonian systems admit Lie symmetries (2/4)}%
\begin{tabular}{lll}
\hline\hline
\textbf{Lie }$\downarrow $\textbf{\ }$~~V\rightarrow $ & $\mathbf{F}%
^{x}\left( x,y\right) $ & $\mathbf{F}^{y}\left( x,y\right) $ \\ \hline
$T\left( t\right) \partial _{x}$ & $-mx+f\left( y\right) $ & $g\left(
y\right) $ \\
$T\left( t\right) \partial _{y}$ & $f\left( x\right) $ & $-my+g\left(
x\right) $ \\
$2\int T\left( t\right) dt~\partial _{t}+T\left( t\right) \left( x\partial
_{x}+y\partial _{y}\right) $ & $-\frac{m}{4}x+x^{-3}f\left( \frac{y}{x}%
\right) $ & $-\frac{m}{4}y+y^{-3}g\left( \frac{y}{x}\right) $ \\ \hline\hline
\end{tabular}%
\label{2dN2}%
\end{table}%

Case III: $Y^{i}$ is a special PC. There is only one dynamical system in
this case, which is the forced oscillator, acted upon the external force $%
F^{i}=\left( \omega x+a\right) \partial _{x}+\left( \omega y+b\right)
\partial _{y}$, that is the system is conservative. As it can be seen from
Table \ref{ProjectiveFlat2d} the Lie symmetry algebra of the forced
oscillator is the $sl\left( 4,R\right) .$ This result agrees with that of
\cite{Prince01}.

Except the above three cases we have to consider the Lie point symmetries
generated from linear combinations of the vectors $Y^{i}.$ It is found that
the only new cases are the ones given given in Tables \ref{2dN3} and \ref%
{2dN4}

We assume now $F^{i}$ to be conservative with potential function $V(x,y).$
In this case the results of the previous Tables differentiate. The results
of the calculations are given in Tables \ref{2dCN1}, \ref{2dCN3} and \ref%
{2dCN4}~.

\begin{table}[tbp] \centering%
\caption{Two dimensional conservative Newtonian systems admit Lie symmetries
(1/3)}%
\begin{tabular}{cccc}
\hline\hline
\textbf{Lie }$\downarrow $\textbf{\ }$~~V\rightarrow $ & $\mathbf{d=0}$ & $%
\mathbf{d\neq 0}$ & $\mathbf{d=2}$ \\ \hline
$\frac{d}{2}t\partial _{t}+\partial _{x}$ & $c_{1}x+f\left( y\right) $ & $%
f\left( y\right) e^{-dx}$ & $f\left( y\right) e^{-2x}$ \\
$\frac{d}{2}t\partial _{t}+\partial _{y}$ & $c_{1}y+f\left( x\right) $ & $%
f\left( x\right) e^{-dy}$ & $f\left( x\right) e^{-2y}$ \\
$\frac{d}{2}t\partial _{t}+\left( y\partial _{x}-x\partial _{y}\right) $ & $%
\theta +f\left( r\right) $ & $f\left( r\right) e^{-d\theta }$ & $f\left(
r\right) e^{-2\theta }$ \\
$\frac{d}{2}t\partial _{t}+\left( x\partial _{x}+y\partial _{y}\right) $ & $%
x^{2}f\left( \frac{y}{x}\right) $ & $x^{2-d}f\left( \frac{y}{x}\right) $ & $%
c_{1}\ln x~+f\left( \frac{y}{x}\right) $ \\
$\frac{d}{2}t\partial _{t}+x\partial _{x}$ & $c_{1}x^{2}+f\left( y\right) $
& $\nexists $ & $\nexists $ \\
$\frac{d}{2}t\partial _{t}+y\partial _{y}$ & $c_{1}y^{2}+f\left( x\right) $
& $\nexists $ & $\nexists $ \\
$\frac{d}{2}t\partial _{t}+y\partial _{x}$ & $x^{2}+y^{2}+c_{1}x$ & $%
\nexists $ & $\nexists $ \\
$\frac{d}{2}t\partial _{t}+x\partial _{y}$ & $x^{2}+y^{2}+c_{1}y$ & $%
\nexists $ & $\nexists $ \\
&  &  &  \\
\multicolumn{2}{c}{\textbf{Lie }$\downarrow $\textbf{\ }$~~V\rightarrow $} &
\multicolumn{2}{c}{$\mathbf{T}_{,tt}\mathbf{=mT.}$} \\ \hline
\multicolumn{2}{c}{$T\left( t\right) \partial _{x}$} & \multicolumn{2}{c}{$-%
\frac{mx^{2}}{2}+c_{1}x+f\left( y\right) $} \\
\multicolumn{2}{c}{$T\left( t\right) \partial _{y}$} & \multicolumn{2}{c}{$-%
\frac{my^{2}}{2}+c_{1}y+f\left( x\right) $} \\
\multicolumn{2}{c}{$2\int T\left( t\right) dt~\partial _{t}+T\left( t\right)
\left( x\partial _{x}+y\partial _{y}\right) $} & \multicolumn{2}{c}{$-\frac{m%
}{8}\left( x^{2}+y^{2}\right) +\frac{1}{x^{2}}f\left( \frac{y}{x}\right) $}
\\ \hline\hline
\end{tabular}%
\label{2dCN1}%
\end{table}%

As it was stated in section \ref{Introduction} the determination of all two
dimensional potentials which admit a Lie point symmetry has been addressed
previously in \cite{Sen} and \cite{Damianou99}. Our results contain the
results of both these papers and additionally some cases missing, mainly in
the linear combinations of the HV with the KVs.

\subsection{2D autonomous Newtonian systems which admit Noether symmetries}

\label{Newtonian Hamiltomian systems which admit Noether Symmetries}

Noether symmetries are associated with a Lagrangian. Therefore we consider
only the case in which the force $F^{i}$ is conservative. Furthermore
Noether symmetries are special Lie symmetries, hence we look into the two
dimensional potentials which admit a Lie point symmetry. These potentials
were determined in the previous section. We apply Theorem \ref{The Noether
symmetries of a conservative system} to these potentials and select the
potentials which admit a Noether symmetry. The calculations are similar to
the ones for the Lie point symmetries and are omitted. The results are
listed in Tables \ref{2dNS1} and \ref{2dNS2}. In the next section we apply
the same method to determine the three dimensional autonomous Newtonian
systems which admit Lie and Noether symmetries.

\bigskip
\begin{table}[ht] \centering%
\caption{Two dimensional autonomous Newtonian systems admiting Noether
symmetries (1/2)}%
\begin{tabular}{cccc}
\hline\hline
\textbf{Noether Symmetry} & $\mathbf{V}\left( x,y\right) $ & \textbf{Noether
Symmetries} & $\mathbf{V}\left( x,y\right) $ \\ \hline
$\partial _{x}$ & $cx+f\left( y\right) $ & $\partial _{x}+b\partial _{y}$ & $%
f\left( y-bx\right) -cx$ \\
$\partial _{y}$ & $cy+f\left( x\right) $ & $\left( a+y\right) \partial
_{x}+\left( b-x\right) \partial _{y}$ & $f\left( \frac{1}{2}\left(
x^{2}+y^{2}\right) +ay-bx\right) $ \\
$y\partial _{x}-x\partial _{y}$ & $c\theta +f\left( r\right) $ & $2t\partial
_{t}+\left( x+ay\right) \partial _{x}+\left( y-ax\right) \partial _{y}$ & $%
r^{-2}~f\left( \theta -a\ln r\right) $ \\
$2t\partial _{t}+x\partial _{x}+y\partial _{y}$ & $x^{-2}f\left( \frac{y}{x}%
\right) $ & $2t\partial _{t}+\left( a+x\right) \partial _{x}+\left(
b+y\right) \partial _{y}$ & $f\left( \frac{b+x}{a+x}\right) \left(
a+x\right) ^{-2}-c\left( a+x\right) ^{-2}\left( \frac{1}{2}x^{2}+ax\right) $
\\ \hline\hline
\end{tabular}%
\label{2dNS1}%
\end{table}%

\begin{table}[H] \centering%
\caption{Two dimensional autonomous Newtonian systems admiting Noether
symmetries (2/2)}%
\begin{tabular}{cc}
\hline\hline
\textbf{Noether \ }$\downarrow $\textbf{\ }$V\rightarrow $ & $\mathbf{T}%
_{,tt}\mathbf{=mT}$ \\ \hline
$\,T\left( t\right) \partial _{x}$ & $f\left( y\right) -cx-\frac{m}{2}x^{2}$
\\
$T\left( t\right) \partial _{y}$ & $f\left( x\right) -cy-\frac{m}{2}y^{2}$
\\
$2\int T\left( t\right) dt~\partial _{t}+T\left( t\right) \left( x\partial
_{x}+y\partial _{y}\right) $ & $x^{-2}f\left( \frac{y}{x}\right) -\frac{m}{8}%
\left( x^{2}+y^{2}\right) $ \\
&  \\
\textbf{Noether }$\downarrow $\textbf{\ }$~~V\rightarrow $ & $\mathbf{T}%
_{,tt}\mathbf{=mT}$ \\ \hline
$\,T\left( t\right) \partial _{x}+bT\left( t\right) \partial _{y}$ & $-\frac{%
m}{2}\left( x^{2}+y^{2}\right) -\frac{m}{2}(y-bx)^{2}+f\left( y-bx\right)
-cx $ \\
$2\int T\left( t\right) dt~\partial _{t}+T\left( t\right) \left( \left(
a+x\right) \partial _{x}+\left( b+y\right) \partial _{y}\right) $ & $f\left(
\frac{b+x}{a+x}\right) \left( a+x\right) ^{-2}-\frac{c}{2}\left( a+x\right)
^{-2}x\left( x+2a\right) $ \\
& $-\frac{xm\left( x+2a\right) }{8\left( a+x\right) ^{4}}\left\{
\begin{array}{c}
\left( \left( a+x\right) ^{2}+a^{2}\right) y\left( y+2b\right) + \\
+x\left( x+2a\right) \left( b+\left( a+x\right) \right) \left( -b+\left(
a+x\right) \right)%
\end{array}%
\right\} $ \\ \hline\hline
\end{tabular}%
\label{2dNS2}%
\end{table}%

\section{3D autonomous Newtonian systems admit Lie/Noether point symmetries}

\label{Lie point symmetries of three dimensional autonomous Newtonian
systems}

In this section we determine all Newtonian systems with \emph{\ three }%
degrees of freedom which admit at least one Lie/Noether symmetry (except the
obvious symmetry $\partial _{t}$)$.$ In order to use Theorem \ref{The
general conservative system} we need the special projective algebra of the
Euclidian 3D metric

\begin{equation}
ds_{E}^{2}=dx^{2}+dy^{2}+dz^{2}.
\end{equation}

This algebra consists of 15 vectors\footnote{%
These vectors are not all linearly independent i.e. the HV and the rotations
are linear combinations of the ACs} as follows: Six KVs $\partial _{\mu
}~,~x_{\nu }\partial _{\mu }-x_{\mu }\partial _{\nu }~$ one HV$~R\partial
_{R},~$nine ACs~$x_{\mu }\partial _{\mu }~,~x_{\nu }\partial _{\mu }$ and
three SPCs $x_{\mu }^{2}\partial _{\mu }+x_{\mu }x_{\nu }\partial _{\nu
}+x_{\mu }x_{\sigma }\partial _{\sigma },~$where\footnote{%
If $x_{\mu }=x,~$then~$\left\{ x_{\nu }=y~,~x_{\sigma }=z\right\} ~$or~$%
\left\{ x_{\nu }=z~,~x_{\sigma }=y\right\} $}$~\mu \neq \nu \neq \sigma $ , $%
r_{\left( \mu \nu \right) }^{2}=x_{\mu }^{2}+x_{\nu }^{2},~\theta _{\left(
\mu \nu \right) }=\arctan \left( \frac{x_{\nu }}{x_{\mu }}\right) $ and $%
R,\theta ,\phi $ are spherical coordinates$.$

In the computation of Lie symmetries we consider only the linearly
independent vectors of the special projective group. We do not consider
their linear combinations because the resulting Lie symmetries are too many;
on the other hand they can be computed in the standard way.

\subsection{3D autonomous Newtonian systems admit Lie point symmetries}

In Tables \ref{3DLs1} and~\ref{3DLs2} we list the Lie point symmetries and
the functional dependence of the components of the force for Case I and II
of Theorem \ref{The general conservative system}.

\begin{table}[tbp] \centering%
\caption{Three dimensional autonomous Newtonian systems admit Lie symmetries
(1/2)}$%
\begin{tabular}{cccc}
\hline\hline
\textbf{Lie symmetry} & $\mathbf{F}_{\mu }\left( x_{\mu },x_{\nu },x_{\sigma
}\right) $ & $\mathbf{F}_{\nu }\left( x_{\mu },x_{\nu },x_{\sigma }\right) $
& $\mathbf{F}_{\sigma }\left( x_{\mu },x_{\nu },x_{\sigma }\right) $ \\
\hline
$\frac{d}{2}t\partial _{t}+\partial _{\mu }$ & $e^{-dx_{\mu }}f\left( x_{\nu
},x_{\sigma }\right) $ & $e^{-dx_{\mu }}g\left( x_{\nu },x_{\sigma }\right) $
& $e^{-dx_{\mu }}h\left( x_{\nu },x_{\sigma }\right) $ \\
$\frac{d}{2}t\partial _{t}+\partial _{\theta _{\left( \mu \nu \right) }}$ & $%
e^{-d\theta _{\left( \mu \nu \right) }}f\left( r_{\left( \mu \nu \right)
},x_{\sigma }\right) $ & $e^{-d\theta _{\left( \mu \nu \right) }}g\left(
r_{\left( \mu \nu \right) },x_{\sigma }\right) $ & $e^{-d\theta _{\left( \mu
\nu \right) }}h\left( r_{\left( \mu \nu \right) },x_{\sigma }\right) $ \\
$\frac{d}{2}t\partial _{t}+R\partial _{R}$ & $x_{\mu }^{1-d}f\left( \frac{%
x_{\nu }}{x_{\mu }},\frac{x_{\sigma }}{x_{\mu }}\right) $ & $x_{\mu
}^{1-d}g\left( \frac{x_{\nu }}{x_{\mu }},\frac{x_{\sigma }}{x_{\mu }}\right)
$ & $x_{\mu }^{1-d}h\left( \frac{x_{\nu }}{x_{\mu }},\frac{x_{\sigma }}{%
x_{\mu }}\right) $ \\
$\frac{d}{2}t\partial _{t}+x_{\mu }\partial _{\mu }$ & $x_{\mu
}^{1-d}f\left( x_{\nu },x_{\sigma }\right) $ & $x_{\mu }^{1-d}g\left( x_{\nu
},x_{\sigma }\right) $ & $x_{\mu }^{1-d}h\left( x_{\nu },x_{\sigma }\right) $
\\
$\frac{d}{2}t\partial _{t}+x_{\nu }\partial _{\mu }$ & $e^{-d\frac{x_{\mu }}{%
x_{\nu }}}\left[ \frac{x_{\mu }}{x_{\nu }}g\left( x_{\nu },x_{\sigma
}\right) +f\left( x_{\nu },x_{\sigma }\right) \right] $ & $e^{-d\frac{x_{\mu
}}{x_{\nu }}}g\left( x_{\nu },x_{\sigma }\right) $ & $e^{-d\frac{x_{\mu }}{%
x_{\nu }}}h\left( x_{\nu },x_{\sigma }\right) $ \\ \hline\hline
\end{tabular}%
$\label{3DLs1}%
\end{table}%

\begin{table}[tbp] \centering%
\caption{Three dimensional autonomous Newtonian systems admit Lie symmetries
(1/2)}$%
\begin{tabular}{cccc}
\hline\hline
\textbf{Lie symmetry} & $\mathbf{F}_{\mu }\left( x_{\mu },x_{\nu },x_{\sigma
}\right) $ & $\mathbf{F}_{\nu }\left( x_{\mu },x_{\nu },x_{\sigma }\right) $
& $\mathbf{F}_{\sigma }\left( x_{\mu },x_{\nu },x_{\sigma }\right) $ \\
\hline
$t\partial _{\mu }$ & $f\left( x_{\nu },x_{\sigma }\right) $ & $g\left(
x_{\nu },x_{\sigma }\right) $ & $h\left( x_{\nu },x_{\sigma }\right) $ \\
$t^{2}\partial _{t}+tR\partial _{R}$ & $\frac{1}{x_{\mu }^{3}}f\left( \frac{%
x_{\nu }}{x_{\mu }},\frac{x_{\sigma }}{x_{\mu }}\right) $ & $\frac{1}{x_{\mu
}^{3}}g\left( \frac{x_{\nu }}{x_{\mu }},\frac{x_{\sigma }}{x_{\mu }}\right) $
& $\frac{1}{x_{\mu }^{3}}h\left( \frac{x_{\nu }}{x_{\mu }},\frac{x_{\sigma }%
}{x_{\mu }}\right) $ \\
$e^{\pm t\sqrt{m}}\partial _{\mu }$ & $-mx_{\mu }+f\left( x_{\nu },x_{\sigma
}\right) $ & $g\left( x_{\nu },x_{\sigma }\right) $ & $h\left( x_{\nu
},x_{\sigma }\right) $ \\
$\frac{1}{\sqrt{m}}e^{\pm t\sqrt{m}}\partial _{t}\pm e^{\pm t\sqrt{m}%
}R\partial _{R}$ & $-\frac{m}{4}x_{\mu }+\frac{1}{x_{\mu }^{3}}f\left( \frac{%
x_{\nu }}{x_{\mu }},\frac{x_{\sigma }}{x_{\mu }}\right) $ & $-\frac{m}{4}%
x_{\nu }+\frac{1}{x_{\mu }^{3}}g\left( \frac{x_{\nu }}{x_{\mu }},\frac{%
x_{\sigma }}{x_{\mu }}\right) $ & $-\frac{m}{4}x_{\sigma }+\frac{1}{x_{\mu
}^{3}}h\left( \frac{x_{\nu }}{x_{\mu }},\frac{x_{\sigma }}{x_{\mu }}\right) $
\\ \hline\hline
\end{tabular}%
$\label{3DLs2}%
\end{table}%

For the remaining Case III of Theorem \ref{The general conservative system},
we have that the force $F^{\mu }$ admits Lie symmetries generated from the
proper sp. PCs if the force is the isotropic oscillator, that is, ~$F^{\mu
}=\left( \omega x^{\mu }+c^{\mu }\right) \partial _{\mu }$ where $\omega
,~c^{\mu }$ are constants. From Tables \ref{3DLs1} and \ref{3DLs2} we infer%
{\LARGE \ }that the isotropic oscillator admits 24 Lie point symmetries
generating the $Sl\left( 5,R\right) $, as many as the free particle \cite%
{Prince01}.

In order to demonstrate the use of the above Tables let us require the
equations of motion of a Newtonian dynamical system which is invariant under
the $sl(2,R)$ algebra. We know \cite{Leach1991} that $sl(2,R)$ is generated
by the following Lie symmetries%
\begin{equation*}
\partial _{t}~,~2t\partial _{t}+R\partial _{R}~,t^{2}\partial
_{t}+tR\partial _{R}.
\end{equation*}

From Table \ref{3DLs1}\ and from Table \ref{3DLs2} we have that the force
must be of the form%
\begin{equation}
F=\left( \frac{1}{x_{\mu }^{3}}f\left( \frac{x_{\nu }}{x_{\mu }},\frac{%
x_{\sigma }}{x_{\mu }}\right) ,\frac{1}{x_{\mu }^{3}}g\left( \frac{x_{\nu }}{%
x_{\mu }},\frac{x_{\sigma }}{x_{\mu }}\right) ,\frac{1}{x_{\mu }^{3}}h\left(
\frac{x_{\nu }}{x_{\mu }},\frac{x_{\sigma }}{x_{\mu }}\right) \right)
\end{equation}%
therefore, the equations of motion of this system in Cartesian coordinates
are:%
\begin{equation}
(\ddot{x},\ddot{y},\ddot{z})=\left( \frac{1}{x_{\mu }^{3}}f\left( \frac{%
x_{\nu }}{x_{\mu }},\frac{x_{\sigma }}{x_{\mu }}\right) ,\frac{1}{x_{\mu
}^{3}}g\left( \frac{x_{\nu }}{x_{\mu }},\frac{x_{\sigma }}{x_{\mu }}\right) ,%
\frac{1}{x_{\mu }^{3}}h\left( \frac{x_{\nu }}{x_{\mu }},\frac{x_{\sigma }}{%
x_{\mu }}\right) \right) .
\end{equation}%
Immediately we recognize that this dynamical system is the well known and
important generalized Kepler Ermakov system (see \cite{Leach1991}). A\
different representation of $sl(2,R)$ consists of the vectors
\begin{equation*}
\partial _{t}~,~\frac{1}{\sqrt{m}}e^{\pm t\sqrt{m}}\partial _{t}\pm e^{\pm t%
\sqrt{m}}R\partial _{R}
\end{equation*}%
For this representation from Table \ref{3DLs2} we have%
\begin{equation}
F^{\prime }=-\frac{m}{4}\left( x_{\mu },x_{\nu },x_{\sigma }\right) +\left(
\frac{1}{x_{\mu }^{3}}f\left( \frac{x_{\nu }}{x_{\mu }},\frac{x_{\sigma }}{%
x_{\mu }}\right) ,\frac{1}{x_{\mu }^{3}}g\left( \frac{x_{\nu }}{x_{\mu }},%
\frac{x_{\sigma }}{x_{\mu }}\right) ,\frac{1}{x_{\mu }^{3}}h\left( \frac{%
x_{\nu }}{x_{\mu }},\frac{x_{\sigma }}{x_{\mu }}\right) \right)
\end{equation}%
which leads again to the \emph{autonomous Kepler Ermakov system}. In a
subsequent chapter, we shall apply the results obtained here to study the
integrability of the 3D Hamiltonian Kepler--Ermakov system and generalize it
in a Riemannian space.

In case the force is given by the potential~$V=V\left( x^{\mu }\right) $,
that is, that the system is conservative, we obtain the results in Table \ref%
{3DCL}.

\begin{table}[H] \centering%
\caption{Three dimensional conservative Newtonian systems admiting Lie
symmetries}$%
\begin{tabular}{cccc}
\hline\hline
{\small {\textbf{Lie /V(x,y,z)} }} & $\mathbf{d=0}$ & $\mathbf{d=2}$ & $%
\mathbf{d\neq 0,2}$ \\ \hline
$\frac{d}{2}t\partial _{t}+\partial _{\mu }$ & $c_{1}x_{\mu }+f\left( x_{\nu
},x_{\sigma }\right) $ & $e^{-2x_{\mu }}f\left( x_{\nu },x_{\sigma }\right) $
& $e^{-dx_{\mu }}f\left( x_{\nu },x_{\sigma }\right) $ \\
$\frac{d}{2}t\partial _{t}+\partial _{\theta _{\left( \mu \nu \right) }}$ & $%
\,c_{1}\theta _{\left( \mu \nu \right) }+f\left( r_{\left( \mu \nu \right)
},x_{\sigma }\right) $ & $e^{-2\theta _{\left( \mu \nu \right) }}f\left(
r_{\left( \mu \nu \right) },x_{\sigma }\right) $ & $e^{-d\theta _{\left( \mu
\nu \right) }}f\left( r_{\left( \mu \nu \right) },x_{\sigma }\right) $ \\
$\frac{d}{2}t\partial _{t}+R\partial _{R}$ & $x^{2}f\left( \frac{x_{\nu }}{%
x_{\mu }},\frac{x_{\sigma }}{x_{\mu }}\right) $ & $c_{1}\ln \left( x_{\mu
}\right) +f\left( \frac{x_{\nu }}{x_{\mu }},\frac{x_{\sigma }}{x_{\mu }}%
\right) $ & $x^{2-d}f\left( \frac{x_{\nu }}{x_{\mu }},\frac{x_{\sigma }}{%
x_{\mu }}\right) $ \\
$\frac{d}{2}t\partial _{t}+x_{\mu }\partial _{\mu }$ & $c_{1}x_{\mu
}^{2}+f\left( x_{\nu },x_{\sigma }\right) $ & $\nexists $ & $\nexists $ \\
$\frac{d}{2}t\partial _{t}+x_{\nu }\partial _{\mu }$ & $c_{1}x_{\mu
}+c_{2}\left( x_{\mu }^{2}+x_{\nu }^{2}\right) +f\left( x_{\sigma }\right) $
& $\nexists $ & $\nexists $ \\
&  &  &  \\
\multicolumn{1}{l}{\textbf{Lie}} & \multicolumn{1}{l}{$\mathbf{V}\left(
x,y,z\right) $} & \multicolumn{1}{l}{\textbf{Lie }} & \multicolumn{1}{l}{$%
\mathbf{V}\left( x,y,z\right) $} \\ \hline
\multicolumn{1}{l}{$t\partial _{\mu }$} & \multicolumn{1}{l}{$c_{1}x_{\mu
}+f\left( x_{\nu },x_{\sigma }\right) $} & \multicolumn{1}{l}{$e^{\pm t\sqrt{%
m}}\partial _{\mu }$} & \multicolumn{1}{l}{$-\frac{m}{2}x_{\mu
}^{2}+c_{1}x_{\mu }+f\left( x_{\nu },x_{\sigma }\right) $} \\
\multicolumn{1}{l}{$t^{2}\partial _{t}+tR\partial _{R}$} &
\multicolumn{1}{l}{$\frac{1}{x_{\mu }^{2}}f\left( \frac{x_{\nu }}{x_{\mu }},%
\frac{x_{\sigma }}{x_{\mu }}\right) $} & \multicolumn{1}{l}{$\frac{1}{\sqrt{m%
}}e^{\pm t\sqrt{m}}\partial _{t}+e^{\pm t\sqrt{m}}R\partial _{R}$} &
\multicolumn{1}{l}{$-\frac{m}{8}\left( x_{\mu }^{2}+x_{\nu }^{2}+x_{\sigma
}^{2}\right) +\frac{1}{x_{\mu }^{2}}f\left( \frac{x_{\nu }}{x_{\mu }},\frac{%
x_{\sigma }}{x_{\mu }}\right) $} \\ \hline\hline
\end{tabular}%
$\label{3DCL}%
\end{table}%
\bigskip

\subsection{3D autonomous Newtonian systems which admit Noether point
symmetries}

\label{Noether point symmetries1}

In this section using theorem \ref{The Noether symmetries of a conservative
system} we determine all autonomous Newtonian Hamiltonian systems with
Lagrangian%
\begin{equation}
L=\frac{1}{2}\left( \dot{x}^{2}+\dot{y}^{2}+\dot{z}^{2}\right) -V\left(
x,y,z\right)  \label{NPC.09}
\end{equation}%
which admit a non-trivial Noether point symmetry. This problem has been
considered previously in \cite{Damianou99,Damianou2004}, however as we shall
show the results in these works are not complete. We note that the Lie
symmetries of a conservative system are not necessarily Noether symmetries.
The inverse is of course true.

Before we continue we note that the homothetic algebra of the Euclidian 3d
space $E^{3}$ has dimension seven and consists of three gradient KVs $%
\partial _{\mu }~~$with gradient function $x_{\mu }$, three non-gradient KVs
$x_{\nu }\partial _{\mu }-x_{\mu }\partial _{\nu }$ generating the
rotational algebra $so\left( 3\right) ,~$and a gradient HV $H^{i}=R\partial
_{R}~$\ with gradient function $H=\frac{1}{2}R^{2}$ , where $R^{2}=x^{\mu
}x_{\mu }.$

The Noether point symmetries generated from the homothetic algebra i.e. the
non-gradient $so(3)$ elements included, are shown in Table \ref{3DNNS1}.
Moreover, the Noether symmetries generated from the gradient homothetic
algebra are listed in Table \ref{3DNNS2}.

\begin{table}[tbp] \centering%
\caption{Three dimensional conservative Newtonian systems admit Noether
symmetries (1/5)}%
\begin{tabular}{cc}
\hline\hline
\textbf{Noether Symmetry} & $\mathbf{V(x,y,z)}$ \\ \hline
$\partial _{\mu }$ & $-px_{\mu }+f\left( x^{\nu },x^{\sigma }\right) $ \\
$x_{\nu }\partial _{\mu }-x_{\mu }\partial _{\nu }$ & $-p\theta _{\left( \mu
\nu \right) }+f\left( r_{\left( \mu \nu \right) },x^{\sigma }\right) ~$ \\
$2t\partial _{t}+R\partial _{R}$ & $\frac{1}{R^{2}}f\left( \theta ,\phi
\right) ~$or $\frac{1}{x_{\mu }^{2}}f\left( \frac{x_{\nu }}{x_{\mu }},\frac{%
x_{\sigma }}{x_{\mu }}\right) $ \\ \hline\hline
\end{tabular}%
\label{3DNNS1}%
\end{table}%

\begin{table}[tbp] \centering%
\caption{Three dimensional conservative Newtonian systems admit Noether
symmetries (2/5)}%
\begin{tabular}{cc}
\hline\hline
\textbf{Noether Symmetry} & $\mathbf{V(x,y,z)~~/~~T}_{,tt}\mathbf{=mT}$ \\
\hline
$T\left( t\right) \partial _{\mu }$ & $-\frac{m}{2}x_{\mu }^{2}-px_{\mu
}+f\left( x_{\nu },x_{\sigma }\right) $ \\
$\left( 2\int T\left( t\right) dt\right) \partial _{t}+T\left( t\right)
R\partial _{R}$ & $-\frac{m}{8}R^{2}+\frac{1}{R^{2}}f\left( \theta ,\phi
\right) $ or~$-\frac{m}{8}R^{2}+\frac{1}{x_{\mu }^{2}}f\left( \frac{x_{\nu }%
}{x_{\mu }},\frac{x_{\sigma }}{x_{\mu }}\right) $ \\ \hline\hline
\end{tabular}%
\label{3DNNS2}%
\end{table}%

The corresponding Noether integrals are computed easily from proposition \ref%
{The Noether Integrals}. In Tables \ref{3DNNS3}, \ref{3DNNS4} and \ref%
{3DNNS5} (see Appendix \ref{appendixTables2}) we give a complete list of the
potentials resulting form the linear combinations of the elements of the
homothetic algebra. From the Tables we infer that the isotropic linear
forced oscillator admits 12 Noether point symmetries, as many as the free
particle.

As it has been remarked above, the determination of the Noether point
symmetries admitted by an autonomous Newtonian Hamiltonian system has been
considered previously in \cite{Damianou2004}. Our results extend the results
of \cite{Damianou2004} and coincide with them if we set the constant $p=0.~$%
For example in page 12 case 1 and page 15 case 6 of \cite{Damianou2004} the
terms $-\frac{p}{a}x_{\mu }$ and $p\arctan \left( l\left( \theta ,\phi
\right) \right) $ are missing respectively. Furthermore the potential given
in page/line 12/1, 13/2, 13/3 of \cite{Damianou2004} admits Noether
symmetries only when \ $\lambda =0~$and $b_{1,2}\left( t\right) =const.$
This is due to the fact that the vectors given in \cite{Damianou2004} are
KVs and in order to have $b_{,t}\neq 0$ they must be given by Case II\ of
theorem \ref{The Noether symmetries of a conservative system} above, that
is, the KVs must be gradient. However the KVs used are linear combinations
of translations and rotations which are non-gradient.

It is possible that there exist integrable Newtonian dynamical systems for
potentials not included in these Tables, for example systems which admit
only dynamical symmetries \cite{Sarlet,Kalotas} with integrals quadratic in
momenta~\cite{Crampin}. However these systems are not integrable via Noether
point symmetries.

We remark that from the above results we are also able to give, without any
further calculations, the Lie and the Noether point symmetries of a
dynamical system 'moving' in a three dimensional flat space whose metric has
Lorentzian signature simply by taking one of the coordinates to be complex,
for example by setting $x^{1}=ix^{1}.$

\section{Motion on the two dimensional sphere}

\label{Motion on the two dimensional sphere}\textbf{\ }\label{Noether in a
space of constant curvature}

A first application of the results of section \ref{Noether point symmetries1}
is the determination of Lie and Noether point symmetries admitted by the
equations of motion of a Newtonian particle moving in a two dimensional
space of constant non-vanishing curvature.

Before we continue it is useful to recall some facts concerning spaces of
constant curvature. Consider a $n+1$ dimensional flat space with fundamental
form%
\begin{equation*}
ds^{2}=\sum_{a}c_{a}(dz^{a})^{2}\quad a=1,2...,n+1
\end{equation*}%
where $c_{a}$ are real constants. The hypersurfaces defined by
\begin{equation*}
\sum_{a}c_{a}(dz^{a})^{2}=eR_{0}^{2}
\end{equation*}%
where $R_{0}$\ is an arbitrary constant and $e=\pm 1$ are called \textbf{%
fundamental hyperquadrics} of the space. When all coefficients $c_{a}$ are
positive the space is Euclidian and $e=+1$. In this case there is one family
of hyperquadrics which is the hyperspheres. In all other cases (excluding
the case when all $c_{a}$ 's are negative) there are two families of
hyperquadrics corresponding to the values $e=+1$ \ and $e=-1.$ It has been
shown that in all cases the hyperquadrics are spaces of constant curvature
(see \cite{Eisenhart} p202).~

Consider an autonomous dynamical system moving in the two dimensional sphere
(Euclidian $\left( \varepsilon =1\right) ~$or Hyperbolic~$\left( \varepsilon
=-1\right) $) with Lagrangian\footnote{%
We use spherical coordinates which are natural in the case of spaces of
constant curvature.~} \cite{Carinena}%
\begin{equation}
L\left( \phi ,\theta ,\dot{\phi},\dot{\theta}\right) =\frac{1}{2}\left( \dot{%
\phi}^{2}+\mathrm{Sinn}^{2}\phi ~\dot{\theta}^{2}\right) -V\left( \theta
,\phi \right)  \label{MCC.01}
\end{equation}%
where
\begin{equation*}
\mathrm{Sinn}\phi =\left\{
\begin{array}{cc}
\mathrm{\sin }\phi & \varepsilon =1 \\
\mathrm{\sinh }\phi & \varepsilon =-1%
\end{array}%
\right. ~,~\mathrm{Cosn}\phi =\left\{
\begin{array}{cc}
\cos \phi & \varepsilon =1 \\
\mathrm{\cosh }\phi & \varepsilon =-1.%
\end{array}%
\right. ~
\end{equation*}

The equations of motion are%
\begin{eqnarray}
\ddot{\phi}-\mathrm{Sinn}\phi ~\mathrm{Cosn}\phi \mathrm{~}\dot{\theta}%
^{2}+V_{,\phi } &=&0  \label{MCC.02} \\
\ddot{\theta}+2\frac{\mathrm{Cosn}\phi }{\mathrm{Sinn}\phi }~\dot{\theta}%
\dot{\phi}+\frac{1}{\mathrm{Sinn}^{2}\phi }V_{,\theta } &=&0.  \label{MCC.03}
\end{eqnarray}

For the Lagrangian (\ref{MCC.01}) proposition \ref{The Noether symmetries of
a conservative system CC} applies and we use it to find the potentials $%
V\left( \theta ,\phi \right) $ for which additional Noether point
symmetries, hence Noether integrals are admitted.

The homothetic algebra of a metric of spaces of constant curvature consists
only of non-gradient KVs (hence $\psi =0)$ as follows

(a) $\varepsilon =1~$ (Euclidian case)

\begin{equation}
CK_{e}^{1}=\sin \theta \partial _{\phi }+\cos \theta \cot \phi \partial
_{\theta },~CK_{e}^{2}=\cos \theta \partial _{\phi }-\sin \theta \cot \phi
\partial _{\theta },~CK_{e}^{3}=\partial _{\theta }  \label{MCC.03.1}
\end{equation}

(b) $\varepsilon =-1~$(Hyperbolic case)%
\begin{equation}
CK_{h}^{1}=\sin \theta \partial _{\phi }+\cos \theta ~\mathrm{\coth }\phi
\partial _{\theta },~CK^{2}=\cos \theta \partial _{\phi }-\sin \theta ~%
\mathrm{\coth }\phi \partial _{\theta },~CK^{3}=\partial _{\theta }.
\label{MCC.03.2}
\end{equation}

Therefore the Noether vectors and the Noether function are%
\begin{equation}
\mathbf{X}=CK_{e,h}^{i}\partial _{i},~~f=pt
\end{equation}%
provided the potential satisfies the condition%
\begin{equation}
\mathcal{L}_{CK}V+p=0.  \label{MCC.04}
\end{equation}%
The first integrals given by proposition \ref{The Noether Integrals} are
\begin{equation}
\phi _{II}=-g_{ij}^{i}CK_{e,h}^{i}\dot{x}^{j}+pt  \label{MCC.05}
\end{equation}%
and are time dependent if $p\neq 0$.

\subsection{Noether Symmetries}

We consider two cases, the case $V(\theta ,\phi )=$constant which concerns
the geodesics of the space, and the case $V(\theta ,\phi )\neq $constant.

For the case of geodesics it has been shown (see section \ref{GEsAppl}) that
the Noether point symmetries are the three elements of $so(3)$ with
corresponding Noether integrals%
\begin{eqnarray}
I_{CK_{e,h}^{1}} &=&\dot{\phi}\sin \theta +\dot{\theta}\cos \theta ~\mathrm{%
Sinn}\phi ~\mathrm{Cosn}\phi \\
I_{CK_{e,h}^{2}} &=&\dot{\phi}\cos \theta -\dot{\theta}\sin \theta ~\mathrm{%
Sinn}\phi ~\mathrm{Cosn}\phi \\
I_{CK_{e,h}^{3}} &=&\dot{\theta}~\mathrm{Sinn}^{2}\phi .
\end{eqnarray}%
These integrals are in involution with the Hamiltonian hence the system is
Liouville integrable.

In the case $V(\theta ,\phi )\neq $constant we find the results of Table \ref%
{CCInt}

\begin{table}[tbp] \centering%
\caption{Noether symmetries/Integrals and potentials for the Lagrangian of
the 2D sphere }$%
\begin{tabular}{ccc}
\hline\hline
\textbf{Noether~Symmetry} & $\mathbf{V}\left( \theta ,\phi \right) $ &
\textbf{Noether~Integral} \\ \hline
$CK_{e,h}^{1}$ & $F\left( \cos \theta ~\mathrm{Sinn}\phi \right) $ & $%
I_{CK_{e,h}^{1}}$ \\
$CK_{e,h}^{2}$ & $F\left( \sin \theta ~\mathrm{Sinn}\phi \right) $ & $%
I_{CK_{e,h}^{2}}$ \\
$CK_{e,h}^{3}$ & $F\left( \phi \right) $ & $I_{CK_{e,h}^{3}}$ \\
$aCK_{e,h}^{1}+bCK_{e,h}^{2}$ & $F\left( \frac{1+\tan ^{2}\theta }{\mathrm{%
Sinn}^{2}\phi ~\left( a-b\tan \theta \right) ^{2}}\right) $ & $%
aI_{CK_{e,h}^{1}}+bI_{CK_{e,h}^{2}}$ \\
$aCK_{e,h}^{1}+bCK_{e,h}^{3}$ & $F\left( a\cos \theta \mathrm{Sinn}\phi
-\varepsilon ~b~\mathrm{Cosn}\phi \right) $ & $%
aI_{CK_{e,h}^{1}}+bI_{CK_{e,h}^{3}}$ \\
$aCK_{e,h}^{2}+bCK_{e,h}^{3}$ & $F\left( a\sin \theta \mathrm{Sinn}\phi
-\varepsilon ~b~\mathrm{Cosn}\phi \right) $ & $%
aI_{CK_{e,h}^{2}}+bI_{CK_{e,h}^{3}}$ \\
$aCK_{e,h}^{1}+bCK_{e,h}^{2}+cCK_{e,h}^{3}$ & $F\left( \left( a\cos \theta
-b\sin \theta \right) ~\mathrm{Sinn}\phi -\varepsilon ~c~\mathrm{Cosn}\phi
\right) $ & \thinspace $%
aI_{CK_{e,h}^{1}}+bI_{CK_{e,h}^{2}}+cI_{CK_{e,h}^{3}} $ \\ \hline\hline
\end{tabular}%
$\label{CCInt}%
\end{table}%

The first integrals which correspond to each potential of Table \ref{CCInt}
are in involution with the Hamiltonian and independent. Hence the
corresponding systems are integrable. From Table \ref{CCInt} we infer the
following result.

\begin{proposition}
\label{prop CCs}A dynamical system with Lagrangian (\ref{MCC.01})~has one,
two or four Noether point symmetries hence Noether integrals.
\end{proposition}

\begin{proof}
For the case of the free particle we have the maximum number of four Noether
symmetries (the rotation group $so(3)$ plus the $\partial _{t}$). In the
case the potential is not constant the Noether symmetries are produced by
the non-gradient KVs with Lie algebra%
\begin{equation*}
\left[ X_{A},X_{B}\right] =C_{AB}^{C}X_{C}
\end{equation*}%
where $C_{12}^{3}=C_{31}^{2}=~C_{23}^{1}=1$ for $\varepsilon =1~$and $\bar{C}%
_{21}^{3}=\bar{C}_{23}^{1}=\bar{C}_{31}^{2}=1$ for $\varepsilon =-1.~$%
Because the Noether point symmetries form a Lie algebra and the Lie algebra
of the KVs is semisimple the system will admit either none, one or three
Noether point symmetries generated by the KVs. The case of three is when $%
V\left( \theta ,\phi \right) =V_{0}$ that is the case of geodesics,
therefore the Noether point symmetries will be (including $\partial _{t}$)
either one, two or four.
\end{proof}

We note that the two important potentials of Celestial Mechanics,\ that is $%
V_{1}=-\frac{\mathrm{Cosn}\phi }{\mathrm{Sinn}\phi }~,~V_{2}=\frac{1}{2}%
\frac{\mathrm{Sinn}^{2}\phi }{\mathrm{Cosn}^{2}\phi }~\ $which according to
Bertrand 's Theorem \cite{Carinena,Kozlov,Vozmi} \ produce closed orbits on
the sphere are included in Table \ref{CCInt}. Hence the dynamical systems
they define are Liouville integrable via Noether point symmetries $%
~CK_{e,h}^{3}$. \ The potential $V_{1}$ corresponds to the Newtonian Kepler
potential and $V_{2}$ is the analogue of the harmonic oscillator. We also
note that our results contain those of \cite{Carinena,Voy}

We emphasize that the potentials listed in Table \ref{CCInt} concern
dynamical systems with Lagrangian (\ref{MCC.01}) which are integrable via
Noether point symmetries.

\section{Applications}

\label{sectio3ap}

In this section we demonstrate the application of the results of section \ref%
{Newtonian dynamical systems which admit Lie symmetries} in two cases. The
first case is the Kepler-Ermakov system, which (in general) is not a
conservative dynamical system and the second is the H\`{e}non - Heiles type
potential.

\subsection{Lie point symmetries of the Kepler-Ermakov system.}

The Ermakov systems are time dependent dynamical systems, which contain an
arbitrary function of time (the frequency function) and two arbitrary
homogeneous functions of dynamical variables. A central feature of Ermakov
systems is their property of always having a first integral. The
Kepler-Ermakov system is an autonomous Ermakov system defined by the
equations \cite{Athorne1991}
\begin{eqnarray}
\ddot{x}+\frac{x}{r^{3}}H\left( x,y\right) -\frac{1}{x^{3}}f\left( \frac{y}{x%
}\right) &=&0  \label{KE1} \\
\ddot{y}+\frac{y}{r^{3}}H\left( x,y\right) -\frac{1}{y^{3}}g\left( \frac{y}{x%
}\right) &=&0  \label{KE2}
\end{eqnarray}%
where $H,f,g~$are arbitrary functions. In \cite{Karasu} it has been shown
that this system admits Lie point symmetries for certain forms of the
function $H\left( x,y\right) .$ Furthermore it has been shown that for
special classes of these equations there exists a Lagrangian (see also \cite%
{LeachK}).

In the following we demonstrate the use of our results by finding the Lie
point symmetries simply by reading the entries of the proper Tables. Looking
at the Tables we find that equations (\ref{KE1}), (\ref{KE2}) admit a Lie
point symmetry for the following two cases.

\textbf{Case 1. }When $H\left( x,y\right) =\frac{h\left( \frac{y}{x}\right)
}{x}.$ Then from Tables \ref{2dN1} and \ref{2dN2} $\left( m=0\right) $ we
have that the Lie point symmetries are
\begin{equation}
X=\left( c_{1}+c_{2}2t+c_{3}t^{2}\right) \partial _{t}+\left(
c_{2}x+c_{3}tx\right) \partial _{x}+\left( c_{2}y+c_{3}ty\right) \partial
_{y}.  \label{NEP.01}
\end{equation}

\textbf{Case 2. }When $H\left( x,y\right) =\omega ^{2}r^{3}+\frac{h\left(
\frac{y}{x}\right) }{x}$ where $m=-4\omega ^{2}$ and $m\neq 0$. In this case
Table \ref{2dN2} for $m\neq 0$ applies and the Lie point symmetry generator
is
\begin{equation*}
X=\left( c_{1}-\frac{c_{2}}{\omega }\cos \left( 2\omega t\right) +\frac{c_{3}%
}{\omega }\sin \left( 2\omega t\right) \right) \partial _{t}+\left(
c_{2}\sin \left( 2\omega t\right) +c_{3}\cos \left( 2\omega t\right) \right)
x\partial _{x}+\left( c_{2}\sin \left( 2\omega t\right) +c_{3}\cos \left(
2\omega t\right) \right) y\partial _{y}.
\end{equation*}%
These symmetries coincide with the ones found in \cite{Karasu}. We note that
in both cases the Lie symmetry vectors come from the HV $x\partial
_{x}+y\partial _{y}$ of the Euclidean metric.

In a subsequent publication \cite{LeachK} it was shown that the Lagrangian
considered in \cite{Karasu} was incorrect and that the correct Lagrangian
is:
\begin{equation}
L=\frac{1}{2}\left( \dot{r}^{2}+r^{2}\dot{\theta}^{2}\right) -\frac{1}{2}%
\omega ^{2}r^{2}-\frac{\mu }{2r^{2}}-\frac{C\left( \theta \right) }{2r^{2}}
\end{equation}%
where $C(\theta )=\sec ^{2}\theta f(\tan \theta )+\csc ^{2}\theta g(\tan
\theta )$ and the functions $f,g$ satisfy two compatibility conditions (see
equation (5.2) of \cite{LeachK}).

We observe that the Lie symmetries are also Noether symmetries and that the
Noether Integrals (in addition to the Hamiltonian $E$) corresponding the
these Noether symmetries are%
\begin{eqnarray}
I_{1} &=&2tE-r\dot{r} \\
I_{2} &=&t^{2}E-tr\dot{r}+\frac{1}{2}r^{2}
\end{eqnarray}%
for $\omega =0$.~When $\omega \neq 0$ the Noether integrals are
\begin{eqnarray}
I_{1}^{\prime } &=&-\frac{1}{\omega }\cos \left( 2\omega t\right) E-\sin
\left( 2\omega t\right) r\dot{r}+\omega \cos \left( 2\omega t\right) r^{2} \\
I_{2}^{^{\prime }} &=&\frac{1}{\omega }\sin \left( 2\omega t\right) E-\cos
\left( 2\omega t\right) r\dot{r}-\omega \sin \left( 2\omega t\right) r^{2}.
\end{eqnarray}

In total we have three Noether integrals. Since we do not look for
generalized symmetries, we do not expect to find the Ermakov - Lewis
invariant \cite{MoyoL}.

\subsection{Point symmetries of the H\`{e}non - Heiles potential}

The H\`{e}non - Heiles potential
\begin{equation*}
V\left( x,y\right) =\frac{1}{2}\left( x^{2}+y^{2}\right) +x^{2}y-\frac{1}{3}%
y^{2}
\end{equation*}%
has been used as a model for the galactic cluster. Computer analysis has
suggested that for sufficiently small values of the energy, there exists a
first integral independent of energy. In \cite{LeachHHP} it is proposed to
study if there exists a Lie point symmetry of the potential which could
justify such a first integral. Working in a slightly more general scenario,
in \cite{LeachHHP} considered potentials of the form
\begin{equation}
V\left( x,y\right) =\frac{1}{2}\left( x^{2}+y^{2}\right)
+Ax^{3}+Bx^{2}y+Cxy^{2}+Dy^{3}  \label{Leach1}
\end{equation}%
where \thinspace $A,B,C,D$ are real parameters. The H\`{e}non - Heiles
potential is the special case for \thinspace $A=C=0,B=1,D=-\frac{1}{3}$.

Using standard Lie analysis in \cite{LeachHHP} it is shown that only the
potentials $V_{1}\left( x,y\right) =\frac{1}{2}\left( x^{2}+y^{2}\right)
+x^{3},$ \ $V_{2}\left( x,y\right) =\frac{1}{2}\left( x^{2}+y^{2}\right)
+y^{3},$ \ $V_{3}\left( x,y\right) =\frac{1}{2}\left( x^{2}+y^{2}\right) \pm
\left( ay\pm x\right) ^{3},$ \ $V_{4}\left( x,y\right) =\frac{1}{2}\left(
x^{2}+y^{2}\right) \pm \left( ay\mp x\right) ^{3}$ admit Lie point
symmetries, hence the H\`{e}non - Heiles potential does not admit a Lie
point symmetry and the existence of a first integral it is not justified. We
apply the results of sections \ref{Newtonian dynamical systems which admit
Lie symmetries}, \ref{Newtonian Hamiltomian systems which admit Noether
Symmetries} to give the Lie point symmetries and the Noether quantities of
these potentials, simply by reading the relevant Tables.

The potential $V_{1}\left( x,y\right) $ is of the form $cy^{2}+f\left(
x\right) $. Hence from Table \ref{2dCN1} \ the~~Lie point symmetries
admitted by this potential are:
\begin{equation*}
X=c_{0}\partial _{t}+c_{1}\sin t\partial _{y}\ +c_{2}\cos t\partial
_{y}+c_{3}y\partial _{y}.
\end{equation*}%
We note that the Lie symmetry $y\partial _{y},$ which is due to the Affine
collineation, has not been found in \cite{LeachHHP}.

The potential $V_{2}\left( x,y\right) $\textbf{\ }is obtained by $%
V_{1}\left( x,y\right) $ with $x,y$ interchanged. Therefore the Lie point
symmetries admitted by the potential $V_{2}\left( x,y\right) $ are
\begin{equation*}
X=c_{0}\partial _{t}+c_{1}\sin t\partial _{x}\ +c_{2}\cos t\partial
_{x}+c_{3}x\partial _{x}
\end{equation*}%
and again in \cite{LeachHHP} the Lie point symmetry $y\partial _{y}$ is
missing.

The potential $V_{3}\left( x,y\right) $ is of the form $\frac{1}{2}\left(
x^{2}+y^{2}\right) +f\left( x-ay\right) $. Hence from Tables \ref{2dCN3} and %
\ref{2dCN4} the admitted Lie symmetries are
\begin{equation*}
X=c_{0}\partial _{t}+(c_{1}\cos t+c_{2}\sin t)\left( a\partial _{x}\pm
\partial _{y}\right) +c_{3}\left( ax+y\right) \left( a\partial _{x}+\partial
_{y}\right) .
\end{equation*}

The potential $V_{4}\left( x,y\right) $ is of the same form as $V_{3}\left(
x,y\right) $ with $x,y$ interchanged. Therefore the Lie point symmetries
are:
\begin{equation*}
X=c_{0}\partial _{t}+a(c_{1}\cos t+c_{2}\sin t)\left( a\partial _{x}\mp
\partial _{y}\right) +c_{3}\left( ax+y\right) \left( a\partial _{x}+\partial
_{y}\right).
\end{equation*}

We observe that in all four cases the Lie point symmetries depend on four
free parameters (the $c_{0},c_{1},c\,_{2},c_{3}).$ The parameter $c_{0}$
determines the vector $c_{0}\partial _{t}$ and the rest $c_{1},c\,_{2},c_{3}$
the $x-y$ part of the symmetry generators.

\begin{table}[tbp] \centering%
\caption{Noether symmetries of Hènon - Heiles potential}%
\begin{tabular}{ccc}
\hline\hline
$\mathbf{V}\left( x,y\right) $ & \textbf{Noether Symmetry} & \textbf{Noether
Integral} \\ \hline
$\frac{1}{2}\left( x^{2}+y^{2}\right) +x^{3}$ & $\sin t\partial _{y}\ $ & $%
\dot{y}\sin t-y\cos t$ \\
& $\cos t\partial _{y}$ & $\dot{y}\cos t+y\sin t$ \\
$\frac{1}{2}\left( x^{2}+y^{2}\right) +y^{3}$ & $\sin t\partial _{x}$ & $%
\dot{x}\sin t-x\cos t$ \\
& $\cos t\partial _{x}$ & $\dot{x}\cos t+x\sin t$ \\
$\frac{1}{2}\left( x^{2}+y^{2}\right) \pm \left( ay\pm x\right) ^{3}$ & $%
\sin t\left( \mp a\partial _{x}+\partial _{y}\right) $ & $\left( \mp a\dot{x}%
+\dot{y}\right) \sin t-\left( \mp ax+y\right) \cos t$ \\
& $\cos t\left( \mp a\partial _{x}+\partial _{y}\right) $ & $\left( \mp a%
\dot{x}+\dot{y}\right) \cos t+\left( \mp ax+y\right) \sin t$ \\
$\frac{1}{2}\left( x^{2}+y^{2}\right) \pm \left( ay\mp x\right) ^{3}$ & $%
\sin t\left( \pm a\partial _{x}+\partial _{y}\right) $ & $\left( \pm a\dot{x}%
+\dot{y}\right) \sin t-\left( \pm ax+y\right) \cos t$ \\
& $\cos t\left( \pm a\partial _{x}+\partial _{y}\right) $ & $\left( \pm a%
\dot{x}+\dot{y}\right) \cos t+\left( \pm ax+y\right) \sin t$ \\ \hline\hline
\end{tabular}%
\label{HHP}%
\end{table}%

The Lie point symmetries which are possibly Noether symmetries are the ones
generated by the KVs. We check that the Lie point symmetries which are due
to the gradient KVs are Noether Symmetries of the potentials (plus the $%
\partial _{\,t}$ whose Noether integral is the Hamiltonian). The Noether
integrals and the Noether functions corresponding to each of these
symmetries are given in Table \ref{HHP}. The results coincide with those of
\cite{LeachHHP,Shrauner}.

\section{Conclusion}

We have shown in Theorem \ref{The general conservative system} and Theorem %
\ref{The Noether symmetries of a conservative system} that the Lie and the
Noether point symmetries for the general class of equations of motion (\ref%
{L2P.1}) are generated from the special projective Lie algebra and the
homothetic Lie algebra respectively of the metric of the space where motion
takes place. The specific subalgebra is determined by a set of differential
conditions which involve the potential defining the dynamical system. The
results apply to both conservative and non conservative dynamical systems.
They also apply to affine spaces and they are independent of the signature
of the metric and the dimension of the space.

The essence of the above is that they reduce the problem of finding the Lie
and the Noether point symmetries of second order systems of equations of the
form (\ref{L2P.1}) to the geometric requirement of finding the special
projective algebra of the a metric (or more general of an affine) space.
Because there is a plethora of results in existing studies on the projective
algebra of Riemannian spaces, it is possible that the problem of finding the
Lie and the Noether point symmetries of an autonomous conservative dynamical
is already solved! As it has been shown, one such case is the case of spaces
of constant curvature. The power of the geometric approach is that it \emph{%
gives all }the Lie and the Noether point symmetries without the use of
computer programs.

An additional point, which could be of interest, is one to reverse the
argument and use the computational approach of the Lie symmetries of the
geodesic equations of a space to compute the projective group, which can be
a formidable task in Differential Geometry. Aminova \cite%
{Aminova00,Aminova2006,Aminova2010} has shown that if one chooses the Cartan
parametrization of the geodesic equations then the Lie symmetries generate
the projective algebra of the underline metric. Because up to now there do
not seem to exist either a general method or general theorems which allow
the computation of the projective algebra of a metric, this approach could
be valuable.

We have applied these theorems to classify all two and three dimensional
Newtonian dynamical systems which admit at least one Lie symmetry, and in
the case of conservative forces, all two and three dimensional potentials $%
V\left( x^{k}\right) $ which admit a Lie symmetry and a Noether point
symmetry. These results complete previous results \cite%
{Sen,Damianou99,Damianou2004} concerning the Noether point symmetries of the
two and three dimensional Newtonian dynamical systems. We note that, due to
the geometric derivation and the tabular presentation, the results can be
extended easily to higher dimensional flat spaces; however at the cost of
convenience because the linear combinations of the symmetry vectors increase
dramatically.

We have demonstrated the application of the results in various important
cases. We considered the Kepler-Ermakov system, which is an autonomous, but
in general not conservative dynamical system and we determined the classes
of this type of systems which admit Lie and Noether point symmetries; we
also considered the case of the H\`{e}non Heiles type potentials and
determined their Lie point symmetries and their Noether symmetries. These
results are compatible and complete previous results in the literature.

In the following chapter, we apply the results obtained here to study the
Liouville integrability of the three dimensional Hamiltonian Kepler-Ermakov
system and generalize it in a Riemannian space.

\newpage%

\begin{subappendices}%

\section{Proof of main Theorem}

\label{proofautonomous}

Below, we give the proof of Theorem \ref{The general conservative system}.

Equation (\ref{PLS.12}) gives:%
\begin{equation}
\xi \left( t,x^{i}\right) =C\left( t\right) S\left( x^{i}\right) +D\left(
t\right)  \label{PLS.13}
\end{equation}%
where $S^{,i}$ is a gradient KV. Replacing this in (\ref{PLS.11a}) we find%
\begin{equation}
L_{\eta }\Gamma _{(jk)}^{i}=2C_{,t}S_{,(j}\delta _{k)}^{i}.  \label{PLS.14}
\end{equation}%
Because $\Gamma _{(jk)}^{i}$ is a function of $x^{i}$ only and $\eta
^{i}(t,x^{i})$ we must have%
\begin{equation}
\eta ^{i}\left( t,x^{j}\right) =T\left( t\right) Y^{i}\left( x^{j}\right)
\label{PLS.15}
\end{equation}%
hence (\ref{PLS.14}) becomes:%
\begin{equation}
T\left( t\right) L_{Y}\Gamma _{jk}^{i}=2C_{,t}S_{,(j}\delta _{k)}^{i}
\label{PLS.16}
\end{equation}%
from which follows%
\begin{equation}
C_{,t}=a_{0}T  \label{PLS.17}
\end{equation}%
\begin{equation}
L_{Y}\Gamma _{jk}^{i}=2\alpha _{0}S_{,(j}\delta _{k)}^{i}  \label{PLS.17a}
\end{equation}%
where $a_{0}=0$, when $Y^{i}$ is a KV,HKV,AC and $a_{0}\neq 0$ if $Y^{i}~$is
a special PC$.$

The remaining equations (\ref{PLS.09})-(\ref{PLS.10}) are written ($%
F^{i}=g_{ij}V^{,i}(x^{i}))$%
\begin{equation}
TL_{Y}V^{,i}+2\left( C_{,t}S+D_{,t}\right) V^{,i}+T_{,tt}Y^{i}=0
\label{PLS.18}
\end{equation}%
\begin{equation}
C\left( S_{,k}\delta _{j}^{i}+2S,_{j}\delta _{k}^{i}\right)
V^{,k}+2T_{,t}Y^{i}{}_{;~j}-\left( C_{,tt}S+D_{,tt}\right) \delta _{j}^{i}=0
\label{PLS.19}
\end{equation}%
Equation (\ref{PLS.19}) is written as%
\begin{equation*}
C(t)\left( S_{,k}\delta _{j}^{i}+2S,_{j}\delta _{k}^{i}\right)
V^{,k}+2T_{,t}Y^{i}{}_{;~j}+\left( -C_{,tt}\right) S\delta _{j}^{i}+\left(
-D_{,tt}\right) \delta _{j}^{i}=0
\end{equation*}%
which due to (\ref{PLS.17}) is simplified as follows%
\begin{equation}
C(t)\left( S_{,k}\delta _{j}^{i}+2S,_{j}\delta _{k}^{i}\right) V^{,k}+\left(
2Y^{i}{}_{;~j}-a_{0}S\delta _{j}^{i}\right) T_{,t}+\left( -D_{,tt}\right)
\delta _{j}^{i}=0.  \label{PLS.21}
\end{equation}

Collecting the results we have the system of equations%
\begin{equation}
C_{,t}=a_{0}T  \label{PLS.21a}
\end{equation}%
\begin{equation}
L_{Y}\Gamma _{jk}^{i}=2S_{,(j}\delta _{k)}^{i}  \label{PLS.21b}
\end{equation}%
\begin{eqnarray}
TL_{Y}V^{,i}+2\left( C_{,t}S+D_{,t}\right) V^{,i}+T_{,tt}Y^{i} &=&0
\label{PLS.22} \\
C\left( S_{,k}\delta _{j}^{i}+2S,_{j}\delta _{k}^{i}\right) V^{,k}+\left(
2Y^{i}{}_{;~j}-a_{0}S\delta _{j}^{i}\right) T_{,t}+\left( -D_{,tt}\right)
\delta _{j}^{i} &=&0  \label{PLS.23}
\end{eqnarray}

where $a_{0}=0$, when $Y^{i}$ is a KV,HKV,AC and $a_{0}\neq 0$ if $Y^{i}~$is
a sp. PC$.~$We consider various cases.

\textbf{Case I. }$T\left( t\right) =0.$~In this case (\ref{PLS.21a}) implies
$C_{,t}=0$ and (\ref{PLS.22}),(\ref{PLS.23}) give%
\begin{equation}
D_{,t}V^{,i}=0  \label{PLS.24}
\end{equation}%
\begin{equation}
C\left( S_{,k}\delta _{j}^{i}+2S,_{j}\delta _{k}^{i}\right) V^{,k}+\left(
-D_{,tt}\right) \delta _{j}^{i}=0  \label{PLS.25a}
\end{equation}%
From (\ref{PLS.24}) follows $D_{,t}=0$ and consequently (\ref{PLS.25a})
implies%
\begin{equation}
C\left( S_{,k}\delta _{j}^{i}+2S,_{j}\delta _{k}^{i}\right) V^{,k}=0.
\label{PLS.26}
\end{equation}%
Because $\left( S_{,k}\delta _{j}^{i}+2S,_{j}\delta _{k}^{i}\right)
V^{,k}\neq 0$ it follows that $C\left( t\right) =0.$ Therefore in this case
we have the Lie Symmetry%
\begin{equation}
X=d_{1}\partial _{t}  \label{PLS.27}
\end{equation}

\textbf{Case II.} $T\left( t\right) =a_{0}~\neq 0$. ~Equation (\ref{PLS.21a}%
) implies $C_{,t}=a_{0}a_{1}$. Then (\ref{PLS.22}),(\ref{PLS.23}) give:%
\begin{eqnarray}
a_{1}L_{Y}V^{,i}+2\left( C_{,t}S+D_{,t}\right) V^{,i} &=&0  \label{PLS.28} \\
C\left( S_{,k}\delta _{j}^{i}+2S,_{j}\delta _{k}^{i}\right) V^{,k}+\left(
-D_{,tt}\right) \delta _{j}^{i} &=&0  \label{PLS.29}
\end{eqnarray}%
It follows~$2a_{0}a_{1}=c_{1}~,~\frac{2D_{,t}}{a_{1}}=d_{1}~$where
\begin{equation}
L_{Y}V^{,i}+\left( c_{1}S+d_{1}\right) V^{,i}=0  \label{PLS.30a}
\end{equation}%
Then (\ref{PLS.29}) is written%
\begin{equation}
C\left( S_{,k}\delta _{j}^{i}+2S,_{j}\delta _{k}^{i}\right) V^{,k}=0
\label{PLS.31}
\end{equation}%
from which we infer that $C\left( t\right) =0.$ This means $a_{0}=0$ hence $%
Y^{i}$ can only be KV, HV or AC and furthermore $c_{1}=0.$ We conclude that
provided the potential satisfies the condition

\begin{equation}
L_{Y}V^{,i}+d_{1}V^{,i}=0  \label{PLS.32}
\end{equation}%
the symmetry vector is%
\begin{equation}
X=\left( \frac{1}{2}d_{1}a_{1}t+d_{2}\right) \partial
_{t}+a_{1}Y^{i}\partial _{i}  \label{PLS.33}
\end{equation}%
where $Y^{i}$ is a KV, HV or AC.

\textbf{Case III.} $T_{,t}\neq 0$.~In this case we have the system of
simultaneous equations%
\begin{equation}
C_{,t}=a_{0}T  \label{PLS.34}
\end{equation}%
\begin{eqnarray}
TL_{Y}V^{,i}+2\left( C_{,t}S+D_{,t}\right) V^{,i}+T_{,tt}Y^{i} &=&0
\label{PLS.35} \\
C\left( S_{,k}\delta _{j}^{i}+2S,_{j}\delta _{k}^{i}\right) V^{,k}+\left(
2Y^{i}{}_{;~j}-a_{0}S\delta _{j}^{i}\right) T_{,t}+\left( -D_{,tt}\right)
\delta _{j}^{i} &=&0  \label{PLS.36}
\end{eqnarray}%
Suppose that $Y^{i}$ is a non-gradient KV or non-gradient HV or AC. Then $S=$%
constant so that $S^{,i}=0.$ Then (\ref{PLS.36}) becomes:%
\begin{equation}
2T_{,t}Y{}_{i;~j}+\left( -a_{0}T_{,t}S-D_{,tt}\right) g_{ij}=0
\label{PLS.36a}
\end{equation}%
and follows (by taking the antisymmetric part in the indices $i,j)$ that $%
T_{,t}=0$ contrary to our assumption. Therefore $Y^{i}$ must be a gradient
KV, gradient HV or sp.PC$.$ We consider various subcases.\newline

\textbf{Case III. a.} $~Y^{i}$ is a gradient KV/HKV and $Y^{i}\neq V^{,i}.~$%
Then we have $a_{0}=0~,~C\left( t\right) =c_{1}=$constant and equations (\ref%
{PLS.35}),(\ref{PLS.36}) are written as follows:%
\begin{eqnarray}
L_{Y}V^{,i}+2\frac{D_{,t}}{T}V^{,i}+\frac{T_{,tt}}{T}Y^{i} &=&0
\label{PLS.37} \\
c_{1}\left( S_{,k}\delta _{j}^{i}+2S,_{j}\delta _{k}^{i}\right)
V^{,k}+\left( 2\psi T_{,t}-D_{,tt}\right) \delta _{j}^{i} &=&0.
\label{PLS.38}
\end{eqnarray}%
From (\ref{PLS.37}) we infer%
\begin{equation}
2\frac{D_{,t}}{T}=d_{1}~~,~\frac{T_{,tt}}{T}=a_{1}  \label{PLS.39}
\end{equation}%
and
\begin{equation}
L_{Y}V^{,i}+d_{1}V^{,i}+a_{1}Y^{i}=0.  \label{PLS.40}
\end{equation}%
From (\ref{PLS.38}) we find
\begin{equation}
2\psi T_{,t}-D_{,tt}=m  \label{PLS.41}
\end{equation}%
that is:
\begin{equation}
c_{1}\left( S_{,k}\delta _{j}^{i}+2S,_{j}\delta _{k}^{i}\right)
V^{,k}+m\delta _{j}^{i}=0.  \label{PLS.42}
\end{equation}%
The last relation is satisfied only for $c_{1}=0~,~m=0$.\ \ Then (\ref%
{PLS.39}),(\ref{PLS.41}) give:%
\begin{equation}
\frac{T_{,tt}}{T}=a_{1}~\ ,~D\left( t\right) =\frac{1}{2}d_{1}\int T\left(
t\right) dt~,~d_{1}=4\psi .  \label{PLS.43}
\end{equation}%
We conclude that provided the potential satisfies equation (\ref{PLS.40})
and $d_{1}=4\psi $ where $\psi =0$ for a KV and $\psi =1$ for a HV, we have
the Lie symmetry vector%
\begin{equation}
X=~\frac{1}{2}d_{1}\int T\left( t\right) dt\partial _{t}+T\left( t\right)
Y^{i}\partial _{i}.  \label{PLS.44}
\end{equation}

\textbf{Case III. b.}~$Y^{i}$ is a gradient HV and $Y^{i}=\kappa V^{,i}.~$In
this case we have $a_{0}=0~,~C\left( t\right) =c_{1}=$constant and the
system of equations (\ref{PLS.35}),(\ref{PLS.36}) becomes:%
\begin{eqnarray}
L_{Y}V^{,i}+\left( 2\frac{D_{,t}}{T}+\kappa \frac{T_{,tt}}{T}\right) V^{,i}
&=&0  \label{PLS.45} \\
c_{1}\left( S_{,k}\delta _{j}^{i}+2S,_{j}\delta _{k}^{i}\right)
V^{,k}+\left( 2\psi T_{,t}-D_{,tt}\right) \delta _{j}^{i} &=&0.
\label{PLS.46}
\end{eqnarray}%
From (\ref{PLS.46}) follows $c_{1}=0$ which implies the equation%
\begin{equation}
2\psi T_{,t}-D_{,tt}=0.  \label{PLS.47}
\end{equation}%
Because $Y^{i}=V^{,i}$ the $L_{Y}V^{,i}=0$ and we have the second condition%
\begin{equation}
2D_{,t}+\kappa T_{,tt}=0.  \label{PLS.48}
\end{equation}%
We conclude that in this case the Lie symmetry vector is%
\begin{equation}
X=D\left( t\right) \partial _{t}+T\left( t\right) V^{,i}\partial _{i}
\label{PLS.50}
\end{equation}%
where the functions $T\left( t\right) ,$ $D\left( t\right) $ are solutions
of the system of equations (\ref{PLS.46}) and (\ref{PLS.46}).\newline

\textbf{Case III. c.~}$Y^{i}$ is a special PC.~In this case the system of
symmetry conditions reads%
\begin{equation}
C_{,t}=a_{0}T  \label{PLS.51}
\end{equation}%
\begin{eqnarray}
L_{Y}V^{,i}+2\left( \frac{C_{,t}}{T}S+\frac{D_{,t}}{T}\right) V^{,i}+\frac{%
T_{,tt}}{T}Y^{i} &=&0  \label{PLS.52} \\
\left( S_{,k}\delta _{j}^{i}+2S,_{j}\delta _{k}^{i}\right) V^{,k}+\left(
2Y^{i}{}_{;~j}-a_{0}S\delta _{j}^{i}\right) \frac{T_{,t}}{C}+\left( -\frac{%
D_{,tt}}{C}\right) \delta _{j}^{i} &=&0.  \label{PLS.53}
\end{eqnarray}%
Using (\ref{PLS.51}) we write (\ref{PLS.52}) as%
\begin{equation}
L_{Y}V^{,i}+2a_{0}SV^{,i}+2\frac{D_{,t}}{T}V^{,i}+\frac{T_{,tt}}{T}Y^{i}=0
\label{PLS.54}
\end{equation}%
from which follows%
\begin{equation}
\frac{D_{,t}}{T}=\frac{1}{2}d_{1}~~,~\ \frac{T_{,tt}}{T}=a_{1}
\label{PLS.55}
\end{equation}%
where
\begin{equation}
L_{Y}V^{,i}+2a_{0}SV^{,i}+d_{1}V^{,i}+a_{1}Y^{i}=0  \label{PLS.56}
\end{equation}%
Then relation (\ref{PLS.53}) implies the conditions%
\begin{equation}
\frac{T_{,t}}{C}=c_{2}~,~\frac{D_{,tt}}{C}=d_{c}  \label{PLS.57}
\end{equation}%
and%
\begin{equation}
\left( S_{,k}\delta _{j}^{i}+2S,_{j}\delta _{k}^{i}\right) V^{,k}+\left(
2Y^{i}{}_{;~j}-a_{0}S\delta _{j}^{i}\right) c_{2}-d_{c}\delta _{j}^{i}=0.
\label{PLS.58}
\end{equation}%
We conclude that in this case provided the potential function satisfies (\ref%
{PLS.56}), we have the Lie symmetry vector%
\begin{equation}
X=\left( C\left( t\right) S+D\left( t\right) \right) \partial _{t}+T\left(
t\right) Y^{i}\partial _{i}  \label{PLS.59}
\end{equation}%
where the functions $C\left( t\right) ,D\left( t\right) ,T\left( t\right) $
are computed form the equations (\ref{PLS.51}), (\ref{PLS.55}),~(\ref{PLS.57}%
).

\textbf{Case III. d.~}$Y^{i}$ is a special PC of the form $Y^{i}=\lambda
SV^{,i}~,~\lambda =$constant. and $S^{,i}$ is a gradient KV of the metric.

This case is possible only when the potential is such that the vector $%
V^{,i} $ is the gradient HV of the metric (if the metric admits one). Then
it is easy to show that due to (\ref{PLS.34}) equation (\ref{PLS.35})
becomes:%
\begin{equation*}
L_{Y}V^{,i}+2\frac{D_{,t}}{T}V^{,i}+\left( 2\frac{C_{,t}}{T}+\lambda \frac{%
T_{,tt}}{T}\right) SV^{,i}=0.
\end{equation*}%
We compute%
\begin{equation}
L_{Y}V^{,i}=\left[ \lambda SV^{,r},V^{,i}\right] \partial _{r}=-\lambda
S_{,j}V^{,j}V^{,i}  \label{PLS.61}
\end{equation}%
therefore%
\begin{equation}
-\lambda S_{,j}V^{,j}V^{,i}+2\frac{D_{,t}}{T}V^{,i}+\left( 2\frac{C_{,t}}{T}%
+\lambda \frac{T_{,tt}}{T}\right) SV^{,i}=0.  \label{PLS.62}
\end{equation}%
It follows%
\begin{eqnarray}
D_{,t} &=&0  \label{PLS.63} \\
2\frac{C_{,t}}{T}+\lambda \frac{T_{,tt}}{T} &=&\lambda _{1}  \label{PLS.64}
\end{eqnarray}%
and the condition%
\begin{equation}
-\lambda S_{,j}V^{,j}+\lambda _{1}S=0\Rightarrow \lambda
S_{,j}V^{,i}=\lambda _{1}S.  \label{PLS.65}
\end{equation}%
Condition (\ref{PLS.36}) now reads%
\begin{align*}
C\left( S_{,k}\delta _{j}^{i}+2S,_{j}\delta _{k}^{i}\right) V^{,k}+\left(
2Y^{i}{}_{;~j}-a_{0}S\delta _{j}^{i}\right) T_{,t}& =0\Rightarrow \\
C\left( \lambda _{1}S\delta _{j}^{i}+2S,_{j}V^{,i}\right) +\left( 2\lambda
S_{,j}V^{.i}+\left( 2\lambda S-a_{0}S\right) \delta _{j}^{i}\right) T_{,t}&
=0
\end{align*}%
from which follows~$\frac{T_{,t}}{C}=\lambda _{2};$~that is,%
\begin{equation*}
C\left( \lambda _{1}S_{J}\delta _{j}^{i}+2S_{J},_{j}V^{,i}\right) +\lambda
_{2}\left( 2\lambda S_{,j}V^{.i}+\left( 2\lambda S-a_{0}S\right) \delta
_{j}^{i}\right) =0.
\end{equation*}%
\newline
We conclude that in this case we have the Lie symmetry vector%
\begin{equation*}
X=C\left( t\right) S\partial _{i}+T\left( t\right) SV^{,i}\partial _{i}
\end{equation*}%
where the functions $C\left( t\right) ,T\left( t\right) $ are computed form
the solution of the system of simultaneous equations%
\begin{equation*}
C_{,t}=a_{0}T~,~T_{,t}=\lambda _{2}C
\end{equation*}%
\begin{equation*}
2C_{,t}+\lambda T_{,tt}=\lambda _{1}T.
\end{equation*}%
\newpage

\section{Tables of Newtonian systems admit Lie and Noether symmetries}

\label{appendixTables2}

\begin{table}[H] \centering%
\caption{Two dimensional Newtonian systems admiting Lie symmetries (3/4)}%
\begin{tabular}{ccc}
\hline\hline
\textbf{Lie }$\downarrow $\textbf{\ }$~~F^{i}\rightarrow $ & $\mathbf{F}%
^{x}\left( x,y\right) \mathbf{/F}^{r}\left( r,\theta \right) $ & $\mathbf{F}%
^{y}\left( x,y\right) \mathbf{/F}^{\theta }\left( r,\theta \right) $ \\
\hline
$\frac{d}{2}t\partial _{t}+\partial _{x}+b\partial _{y}$ & $f\left(
y-bx\right) e^{-dx}$ & $g\left( y-bx\right) e^{-dx}$ \\
$\frac{d}{2}t\partial _{t}+\left( a+x\right) \partial _{x}+\left( b+y\right)
\partial _{y}$ & $f\left( \frac{b+y}{a+x}\right) \left( a+x\right) ^{\left(
1-d\right) }$ & $g\left( \frac{b+y}{a+x}\right) \left( a+x\right) ^{\left(
1-d\right) }$ \\
$\frac{d}{2}t\partial _{t}+\left( a+x\right) \partial _{x}+\left(
b+hy\right) \partial _{y}$ & $f\left( \left( \frac{b}{h}+y\right) \left(
a+bx\right) ^{-\frac{h}{b}}\right) \left( a+bx\right) ^{1-\frac{d}{b}}$ & $%
g\left( \left( \frac{b}{h}+y\right) \left( a+bx\right) ^{-\frac{h}{b}%
}\right) \left( a+bx\right) ^{\frac{h-d}{b}}$ \\
$\frac{d}{2}t\partial _{t}+\left( x+y\right) \partial _{x}+\left( x+y\right)
\partial _{y}$ & $\left(
\begin{array}{c}
f\left( y-x\right) x+ \\
+g\left( y-x\right)%
\end{array}%
\right) \left( y+x\right) ^{-\frac{d}{2}}$ & $\left(
\begin{array}{c}
f\left( y-x\right) y+ \\
-g\left( y-x\right)%
\end{array}%
\right) \left( y+x\right) ^{-\frac{d}{2}}$ \\
$\frac{d}{2}t\partial _{t}+\left( a^{2}x+ay\right) \partial _{x}+~~~~~~~~\ ~$
& $a\left( ax+y\right) ^{-\frac{d}{1+a^{2}}}\times $ & $a^{2}\left(
ax+y\right) ^{-\frac{d}{1+a^{2}}}\times \text{\ }$ \\
$+\left( ax+y\right) \partial _{y}$ & $\times \left(
\begin{array}{c}
xa^{2}f\left( y-\frac{x}{a}\right) + \\
+g\left( y-\frac{x}{a}\right)%
\end{array}%
\right) $ & $\times \left(
\begin{array}{c}
af\left( y-\frac{x}{a}\right) + \\
-g\left( y-\frac{x}{a}\right)%
\end{array}%
\right) $ \\
$\frac{d}{2}t\partial _{t}+\left( -ay+x\right) \partial _{x}+\left(
ax+y\right) \partial _{y}$ & $f\left( \theta -a\ln r\right) r^{1-d}$ & $%
g\left( \theta -a\ln r\right) r^{1-d}$ \\ \hline\hline
\end{tabular}%
\label{2dN3}%
\end{table}%

\begin{table}[H] \centering%
\caption{Two dimensional Newtonian systems admiting Lie symmetries (4/4)}%
\begin{tabular}{ccc}
\hline\hline
\textbf{Lie }$\downarrow $\textbf{\ }$F^{i}\rightarrow $ & $\mathbf{F}%
^{x}\left( x,y\right) $ & $\mathbf{F}^{y}\left( x,y\right) $ \\ \hline
$T\left( t\right) \left( \partial _{x}+b\partial _{y}\right) $ & $%
-mx+f\left( y-bx\right) $ & $-mbx+g\left( y-bx\right) $ \\
$2\int T\left( t\right) dt~\partial _{t}+T\left( t\right) \left[ \left(
a+x\right) \partial _{x}+\left( b+y\right) \partial _{y}\right] $ & $-\frac{m%
}{4}\left( a+x\right) +\left( a+x\right) ^{-3}f\left( \frac{b+y}{a+x}\right)
$ & $-\frac{m}{4}\left( b+y\right) +g\left( \frac{b+y}{a+x}\right) \left(
a+x\right) ^{-3}$ \\ \hline\hline
\end{tabular}%
\label{2dN4}%
\end{table}%

\begin{table}[H] \centering%
\caption{Two dimensional conservative Newtonian systems admiting Lie symmetries
(2/3)}%
\begin{tabular}{cc}
\hline\hline
\textbf{Lie }$\downarrow $\textbf{\ }$V\rightarrow $ & $\mathbf{T}_{,tt}%
\mathbf{=mT}$ \\ \hline
$T\left( t\right) \left( a\partial _{x}+b\partial _{y}\right) $ & $-\frac{m}{%
2}(x^{2}+y^{2})+c_{1}x+f\left( ay-bx\right) $ \\
$2\int T\left( t\right) dt~\partial _{t}+T\left( t\right) \left[ \left(
a+x\right) \partial _{x}+\left( b+y\right) \partial _{y}\right] $ & $-\frac{m%
}{8}\left( x^{2}+y^{2}+2ax+2by\right) +\left( a+x\right) ^{-2}f\left( \frac{%
b+y}{a+x}\right) $ \\ \hline\hline
\end{tabular}%
\label{2dCN3}%
\end{table}%

\begin{table}[H] \centering%
\caption{Two dimensional conservative Newtonian systems admiting Lie symmetries
(3/3)}%
\begin{tabular}{ccc}
\hline\hline
\textbf{Lie }$\downarrow $\textbf{\ }$~~V\rightarrow $ & $\mathbf{d=0}$ & $%
\mathbf{d\neq 0}$ \\ \hline
$\frac{d}{2}t\partial _{t}+a\partial _{x}+b\partial _{y}$ & $f\left(
ay-bx\right) $ & $\left[ c_{1}+f\left( ay-bx\right) \right] e^{-d\frac{x}{a}%
} $ \\
$\frac{d}{2}t\partial _{t}+\left( a+x\right) \partial _{x}+\left( b+y\right)
\partial _{y}$ & $f\left( \frac{b+y}{a+x}\right) \left( a+x\right) ^{2}$ & $%
f\left( \frac{b+y}{a+x}\right) \left( a+x\right) ^{\left( 2-d\right) }$ \\
$\frac{d}{2}t\partial _{t}+\left( x+y\right) \partial _{x}+\left( x+y\right)
\partial _{y}$ & $f\left( y-x\right) +c_{1}\left( x+y\right) ^{2}$ & $\left(
x+y\right) ^{\left( 2-\frac{d}{2}\right) ~~}$ \\
$\frac{d}{2}t\partial _{t}+\left( a^{2}x+ay\right) \partial _{x}+\left(
ax+y\right) \partial _{y}$ & $c_{1}\left( x^{2}+y^{2}\right) ~+f\left(
ay-x\right) $ & $c_{1}\left( ax+y\right) ^{\left( 2-\frac{d}{1+a^{2}}\right)
}\text{\ }$ \\
$\frac{d}{2}t\partial _{t}+\left( -ay+x\right) \partial _{x}+\left(
ax+y\right) \partial _{y}$ & $f\left( \theta -a\ln r\right) r^{2}$ & $%
f\left( \theta -a\ln r\right) r^{2-d}$ \\ \hline
&  &  \\ \hline\hline
\textbf{Lie }$\downarrow $\textbf{\ }$~~V\rightarrow $ & $\mathbf{d=2}$ & $%
\mathbf{d=1}$ \\ \hline
$\frac{d}{2}t\partial _{t}+\partial _{x}+b\partial _{y}$ & $\left[
c_{1}+f\left( y-bx\right) \right] e^{-2\frac{x}{a}}$ & $\left[ c_{1}+f\left(
y-bx\right) \right] e^{-\frac{x}{a}}$ \\
$\frac{d}{2}t\partial _{t}+\left( a+x\right) \partial _{x}+\left( b+y\right)
\partial _{y}$ & $f\left( \frac{b+y}{a+x}\right) +c_{1}\ln \left( a+x\right)
$ & $f\left( \frac{b+y}{a+x}\right) \left( a+x\right) $ \\
$\frac{d}{2}t\partial _{t}+\left( x+y\right) \partial _{x}+\left( x+y\right)
\partial _{y}$ & $\left( x+y\right) $ & $\left( x+y\right) ^{\frac{3}{2}}$
\\
$\frac{d}{2}t\partial _{t}+\left( a^{2}x+ay\right) \partial _{x}+\left(
ax+y\right) \partial _{y}$ & $\ln \left( ax+y\right) \text{\ ~}\left(
d=2\left( 1+a^{2}\right) \right) $ & $\left( ax+y\right) ^{\left( \frac{%
1+2a^{2}}{1+a^{2}}\right) }\text{\ }$ \\
$\frac{d}{2}t\partial _{t}+\left( -ay+x\right) \partial _{x}+\left(
ax+y\right) \partial _{y}$ & $c_{1}\ln r+f\left( \theta -a\ln r\right) $ & $%
c_{1}r+f\left( \theta -a\ln r\right) r$ \\ \hline
\end{tabular}%
\label{2dCN4}%
\end{table}%

\begin{table}[ht] \centering%
\caption{Three dimensional conservative Newtonian systems admitingt Noether
symmetries (3/5)}%
\begin{tabular}{cc}
\hline\hline
\textbf{Noether Symmetry} & $\mathbf{V(x,y,z)}$ \\ \hline
$a\partial _{\mu }+b\partial _{\nu }$ & $-\frac{p}{a}x_{\mu }+f\left( x^{\nu
}-\frac{b}{a}x^{\mu },x^{\sigma }\right) $ \\
$a\partial _{\mu }+b\left( x_{\nu }\partial _{\mu }-x_{\mu }\partial _{\nu
}\right) $ & $-\frac{p}{\left\vert b\right\vert }\arctan \left( \frac{%
\left\vert b\right\vert x_{\mu }}{\left\vert \left( a+bx_{\nu }\right)
\right\vert }\right) +f\left( \frac{1}{2}r_{\left( \mu \nu \right) }+\frac{a%
}{b}x^{\nu },x^{\sigma }\right) $ \\
$a\partial _{\mu }+b\left( x_{\sigma }\partial _{\nu }-x_{\nu }\partial
_{\sigma }\right) $ & $-\frac{p}{\left\vert b\right\vert }\theta _{\left(
\nu \sigma \right) }+f\left( r_{\left( \nu \sigma \right) },x^{\mu }-\frac{a%
}{b}\theta _{\left( \nu \sigma \right) }\right) $ \\
$a\left( x_{\nu }\partial _{\mu }-x_{\mu }\partial _{\nu }\right) +$ & $%
\frac{p}{a}\arctan \left( \frac{ax_{\nu }+bx_{\sigma }}{x_{\mu }\sqrt{%
a^{2}+b^{2}}}\right) +$ \\
$~~~+b\left( x_{\sigma }\partial _{\mu }-x_{\mu }\partial _{\sigma }\right) $
& \thinspace $+\frac{1}{a}f\left( x_{\sigma }-\frac{a}{b}x_{\nu },x_{\nu
}^{2}\left( 1-\left( \frac{a}{b}\right) ^{2}+\frac{2b}{a}\frac{x_{\sigma }}{%
x_{\nu }}\right) +x_{\mu }^{2}\right) $ \\
$2bt\partial _{t}+a\partial _{\mu }+bR\partial _{R}$ & $-p\frac{x_{\mu
}\left( 2a+bx_{\mu }\right) }{2\left( a+bx_{\mu }^{2}\right) }+\frac{1}{%
\left( a+bx_{\mu }^{2}\right) }f\left( \frac{x_{\nu }}{a+bx_{\mu }},\frac{%
x_{\sigma }}{a+bx_{\mu }}\right) $ \\
$2bt\partial _{t}+a\theta _{\left( \mu \nu \right) }\partial _{\theta
_{\left( \mu \nu \right) }}+bR\partial _{R}$ & $\frac{1}{r_{\left( \mu \nu
\right) }^{2}}f\left( \theta _{\left( \mu \nu \right) }-\frac{a}{b}\ln
r_{\left( \mu \nu \right) },\frac{x_{\sigma }}{r_{\left( \mu \nu \right) }}%
\right) $ \\ \hline\hline
\end{tabular}%
\label{3DNNS3}%
\end{table}%

\begin{table}[H] \centering%
\caption{Three dimensional conservative Newtonian systems admiting Noether
symmetries (4/5)}%
\begin{tabular}{cc}
\hline\hline
\textbf{Noether Symmetry} & $\mathbf{V(x,y,z)}$ \\ \hline
$a\partial _{\mu }+b\partial _{\nu }+c\partial _{\sigma }$ & $-\frac{p}{a}%
x_{\mu }+f\left( x^{\nu }-\frac{b}{a}x^{\mu },x^{\sigma }-\frac{c}{a}x^{\mu
}\right) $ \\
$a\partial _{\mu }+b\partial _{\nu }+c\left( x_{\nu }\partial _{\mu }-x_{\mu
}\partial _{\nu }\right) $ & $\frac{p}{\left\vert c\right\vert }\arctan
\left( \frac{\left( b-cx_{\mu }\right) }{\left\vert \left( a+cx_{\nu
}\right) \right\vert }\right) $ \\
& ~$+f\left( \frac{c}{2}r_{\left( \mu \nu \right) }-bx_{\mu }+ax_{\nu
},x_{\sigma }\right) $ \\
$a\partial _{\mu }+b\partial _{\nu }+c\left( x_{\sigma }\partial _{\mu
}-x_{\mu }\partial _{\sigma }\right) $ & $-\frac{p}{\left\vert c\right\vert }%
\arctan \left( \frac{\left\vert c\right\vert x_{\mu }}{\left\vert
a+cx_{\sigma }\right\vert }\right) $ \\
& $~+f\left( x_{\nu }-\frac{1}{\left\vert c\right\vert }\arctan \left( \frac{%
\left\vert c\right\vert x_{\mu }}{\left\vert a+cx_{\sigma }\right\vert }%
\right) ,\frac{1}{2}r_{\left( \mu \sigma \right) }-\frac{a}{c}x_{\sigma
}\right) $ \\
$a\partial _{\mu }+b\left( x_{\nu }\partial _{\mu }-x_{\mu }\partial _{\nu
}\right) +$ & $\frac{p}{\sqrt{b^{2}+c^{2}}}\arctan \left( \frac{\left(
ab+b^{2}x_{\nu }+bcx_{\sigma }\right) }{\left\vert bx_{\mu }\right\vert
\sqrt{b^{2}+c^{2}}}\right) +$ \\
$~~+c\left( x_{\sigma }\partial _{\mu }-x_{\mu }\partial _{\sigma }\right) $
& $+f\left( x_{\mu }^{2}+x_{\nu }^{2}\left( 1-\frac{c^{2}}{b^{2}}\right)
+\left( \frac{2a}{b}+\frac{2c}{b}x_{\sigma }\right) x_{\nu },x_{\sigma }-%
\frac{c}{b}x_{\nu }\right) $ \\
$so\left( 3\right) $ linear combination & $p\arctan \left( \lambda \left(
\theta ,\phi \right) \right) +$ \\
& $+~F\left( R,b\tan \theta \sin \phi +c\cos \phi -aM_{1}\right) $ \\
$2ct\partial _{t}+a\partial _{\mu }+b\theta _{\left( \nu \sigma \right)
}\partial _{\theta _{\left( \nu \sigma \right) }}+cR\partial _{R}$ & $\frac{1%
}{r_{\left( \nu \sigma \right) }^{2}}f\left( \theta _{\left( \nu \sigma
\right) }-\frac{b}{c}\ln r_{\left( \nu \sigma \right) },\frac{a+cx_{\mu }}{%
cr_{\left( \nu \sigma \right) }}\right) $ \\
$2lt\partial _{t}+\left( a\partial _{\mu }+b\partial _{\nu }+c\partial
_{\sigma }+lR\partial _{R}\right) $ & $-\frac{px\left( 2a+cx_{\mu }\right) }{%
2\left( a+cx_{\mu }\right) ^{2}}+\frac{1}{\left( a+lx_{\mu }\right) ^{2}}%
f\left( \frac{b+lx_{\nu }}{l\left( a+lx_{\mu }\right) },\frac{c+lx_{\sigma }%
}{l\left( a+lx_{\mu }\right) }\right) $ \\ \hline\hline
\end{tabular}%
\label{3DNNS4}%
\end{table}%

\begin{table}[H] \centering%
\caption{Three dimensional conservative Newtonian systems admiting Noether
symmetries (5/5)}%
\begin{tabular}{cc}
\hline\hline
\textbf{Noether Symmetry} & $\mathbf{V(x,y,z)~/~T}_{,tt}\mathbf{=mT}$ \\
\hline
$T\left( t\right) \left( a\partial _{\mu }+b\partial _{\nu }+c\partial
_{\sigma }\right) $ & $-\frac{m}{2a}R^{2}+f\left( x^{\nu }-\frac{b}{a}x^{\mu
},x^{\sigma }-\frac{c}{a}x^{\mu }\right) $ \\
$\left( 2l\int T\left( t\right) dt\right) \partial _{t}+$ & $\frac{1}{\left(
a+lx_{\mu }\right) ^{2}}f\left( \frac{b+lx_{\nu }}{l\left( a+lx_{\mu
}\right) },\frac{c+lx_{\sigma }}{l\left( a+lx_{\mu }\right) }\right) +$ \\
$~~+T\left( t\right) \left( a\partial _{\mu }+b\partial _{\nu }+c\partial
_{\sigma }+lR\partial _{R}\right) $ & $~-\frac{m}{8}\left( R^{2}+\frac{2a}{l}%
x_{\mu }+\frac{2c}{l}x_{\nu }+\frac{2b}{l}x_{\sigma }\right) $ \\ \hline
&  \\
\multicolumn{2}{l}{Where$~~\lambda \left( \phi ,\theta \right) =\left(
\left( a^{2}+b^{2}\right) \cos \phi -bc\tan \theta \sin \phi +cM_{1}\right)
\times $} \\
\multicolumn{2}{l}{$~~~~~~~~~~~~~~~~~~~~\times \left\{ M_{2}\left[
-b^{2}M_{1}^{2}-2b\tan \theta \sin \phi M_{1}-a^{2}\sin ^{2}\phi \tan
^{2}\theta \right] \right\} ^{-\frac{1}{2}}$} \\
\multicolumn{2}{l}{~~~$~~~~~M_{1}=\frac{1}{\cos \theta }\sqrt{\sin ^{2}\phi
\left( 2\cos ^{2}\theta -1\right) }~,M_{2}=\sqrt{a^{2}+b^{2}+c^{2}}$} \\
\hline\hline
\end{tabular}%
\label{3DNNS5}%
\end{table}%

\end{subappendices}%

\chapter{The autonomous Kepler Ermakov system in a Riemannian space\label%
{chapter4}}

\section{Introduction}

The Ermakov system has its roots in the study of the one dimensional time
dependent harmonic oscillator
\begin{equation}
\ddot{x}+\omega ^{2}(t)x=0.  \label{AEP.00.1}
\end{equation}%
Ermakov \cite{Ermakov} obtained a first integral $J$ of this equation by
introducing the auxiliary equation
\begin{equation}
\ddot{\rho}+\omega ^{2}(t)\rho =\rho ^{-3}  \label{AEP.00.2}
\end{equation}%
eliminating the $\omega ^{2}(t)$ term and multiplying with the integrating
factor $\rho \dot{x}-\dot{\rho}x$%
\begin{equation}
J=\frac{1}{2}\left[ (\rho \dot{x}-\dot{\rho}x)^{2}+(x/\rho )^{2}\right] .
\label{AEP.00.3}
\end{equation}

The Ermakov system was rediscovered nearly a century after its introduction
\cite{Lewis1967} and subsequently was generalized beyond the harmonic
oscillator to a two dimensional dynamical system which admits a first
integral \cite{RayReid}. In a series of papers the Lie, the Noether and the
dynamical symmetries of this generalized system have been studied. A short
review of these studies \ and a detailed list of relevant references can be
found in \cite{LeachAndriopoulos}. Earlier reviews of the Ermakov system and
its numerous applications in divertive areas of Physics can be found in \cite%
{Rogersetall,Schiefetall}.

The general Ermakov system does not admit Lie point symmetries. The form of
the most general Ermakov system which admits Lie point symmetries has been
determined in \cite{GoedertHaas1998} and it is called the Kepler Ermakov
system \cite{Athorne1991,Leach1991}. It is well known that these Lie point
symmetries are a representation of the $sl(2,R)$ algebra.

In an attempt to generalize the Kepler Ermakov system to higher dimensions,
Leach \cite{Leach1991} used a transformation to remove the time dependent
frequency term and then demanded that the autonomous `generalized' Kepler
Ermakov system will posses two properties: (a)\ a first integral, the
Ermakov invariant and (b) $sl(2,R)$ invariance wrt to Lie symmetries. It has
been shown, that the invariance group of the Ermakov invariant is reacher
than $sl(2,R)$\ \cite{GovinderLeach1994}. The purpose of the present work is
to use Leach's proposal and generalize the autonomous Kepler Ermakov system
in two directions: (a) to higher dimensions using the $sl(2,R)$ invariance
with respect to Noether symmetries (provided the system is Hamiltonian) and
(b) in a Riemannian space which admits a gradient homothetic vector (HV).

The generalization of the autonomous Kepler Ermakov system to three
dimensions using Lie symmetries has been done in \cite{Leach1991}. In the
following sections, we use the results of Chapter \ref{chapter3} to
generalize the subset of autonomous Hamiltonian Kepler Ermakov systems to
three dimensions via Noether symmetries. We show, that there is a family of
three dimensional autonomous Hamiltonian Kepler Ermakov systems parametrized
by an arbitrary function $f$ which admits the elements of $sl(2,R)$ \ as
Noether point symmetries.\ Each member of this family admits two first
integrals, the Hamiltonian and the Ermakov invariant.

We use this result in order to determine all three dimensional Hamiltonian
Kepler Ermakov systems which are Liouville integrable via Noether point
symmetries. To do this we need to determine all members of the family, that
is, those functions $f$ for which the corresponding system admits an
additional Noether symmetry.

The results of Chapter \ref{chapter3} indicate that there are two cases to
be considered, i.e. Noether point symmetries resulting from linear
combinations of (a) translations and (b) rotations (elements of the $%
so\left( 3\right) $ algebra). In each case we\ determine the functions $f$
and the required extra time independent first integral.

The above scenario can be generalized to an $n$ dimensional Euclidian space
as Leach indicates in \cite{Leach1991}, however at the cost of major
complexity and number of cases to be considered. Indeed as it can be seen by
the results of Chapter \ref{chapter3}, the situation is complex enough even
for the three dimensional case.

We continue with the generalization of the Kepler Ermakov system in a
different and more drastic direction. We note that the Ermakov systems
considered so far are based on the Euclidian space, therefore we may call
them Euclidian Ermakov systems. Furthermore the $sl(2,R)$ symmetry algebra
of the autonomous Kepler Ermakov system is generated by the trivial symmetry
$\partial _{t}$ and the gradient HV\ of the Euclidian two dimensional space $%
E^{2}.$ Using this observation we generalize the autonomous Kepler Ermakov
system (not necessarily Hamiltonian) in an $n~$dimensional Riemannian space
which admits a gradient HV using either Lie or Noether point symmetries. The
new dynamical system we call the Riemannian Kepler Ermakov system. This
generalization makes possible the application of the autonomous Kepler
Ermakov system in General Relativity and in particular in Cosmology.

Concerning General Relativity, we determine the four dimensional autonomous
Riemannian Kepler Ermakov system and the associated Riemannian Ermakov
invariant in the spatially flat Freedman - Robertson - Walker (FRW)\
spacetime and we use previous results to calculate the extra Noether point
symmetries. The applications to cosmology concern two models for dark energy
on a locally rotational symmetric (LRS) space time. The first model involves
a scalar field with an exponential potential minimally interacting with a
perfect fluid with a stiff equation of state. The second cosmological model
is the $f(R)$ modified gravity model of $\Lambda _{bc}CDM$. It is shown,
that, in both models the gravitational field equations define an autonomous
Riemannian Kepler Ermakov system which is integrable via Noether integrals.

In section \ref{ErmakovSys}, we review the main features of the two
dimensional autonomous Euclidian Kepler Ermakov system. In section \ref%
{Generalizing the Kepler Ermakov system}, we discuss the general scheme of
generalization of the two dimensional autonomous Euclidian Kepler Ermakov
system to higher dimensions and to a Riemannian space which admits a
gradient HV. In section \ref{The 3d Euclidian Kepler Ermakov system}, we
consider the generalization to the 3D autonomous Euclidian Hamiltonian
Kepler Ermakov system by Noether point symmetries and determine all such
systems which are Liouville integrable. In section \ref{The Riemannian
Kepler Ermakov system}, we define the autonomous Riemannian Kepler Ermakov
system by the requirements that it will admit (a) a first integral (the
Ermakov invariant) and (b) posses $sl(2,R)$ invariance. In section \ \ref%
{The non conservative Riemannian Kepler Ermakov system}, we consider the
non-conservative autonomous Riemannian Kepler Ermakov system and derive the
Riemannian Ermakov invariant and in section \ref{The Hamiltonian Kepler
Ermakov system in a n-dimensional Riemannian space} we repeat the same for
the autonomous Hamiltonian Riemannian Kepler Ermakov system. In the
remaining sections we discuss the applications of the autonomous Hamiltonian
Riemannian Kepler Ermakov system in General Relativity and in Cosmology.

\section{The two dimensional autonomous Kepler Ermakov system}

\label{ErmakovSys}

In \cite{GoedertHaas1998} Hass and Goedert considered the most general 2d
Newtonian Ermakov system to be defined by the equations:%
\begin{align}
\ddot{x}+\omega ^{2}(t,x,y,\dot{x},\dot{y})x& =\frac{1}{yx^{2}}f\left( \frac{%
y}{x}\right)  \label{EqnEr.50} \\
\ddot{y}+\omega ^{2}(t,x,y,\dot{x},\dot{y})y& =\frac{1}{xy^{2}}g\left( \frac{%
y}{x}\right) .  \label{EqnEr.51}
\end{align}%
This system admits the Ermakov first integral
\begin{equation}
I=\frac{1}{2}(x\dot{y}-y\dot{x})^{2}+\int^{y/x}f\left( \tau \right) d\tau
+\int^{y/x}g\left( \tau \right) d\tau .  \label{EqnEr.52}
\end{equation}

If one considers the transformation:%
\begin{align*}
\Omega ^{2}& =\omega ^{2}-\frac{1}{xy^{3}}g\left( \frac{y}{x}\right) \\
F\left( \frac{y}{x}\right) & =f\left( \frac{y}{x}\right) -\frac{x^{2}}{y^{2}}%
g\left( \frac{y}{x}\right)
\end{align*}%
then equations (\ref{EqnEr.50})-(\ref{EqnEr.51}) take the form%
\begin{align}
\ddot{x}+\Omega ^{2}(x,y,\dot{x},\dot{y})x& =\frac{1}{x^{2}y}F\left( \frac{y%
}{x}\right)  \label{EqnEr.53} \\
\ddot{y}+\Omega ^{2}(x,y,\dot{x},\dot{y})y& =0.  \label{EqnEr.54}
\end{align}

Due to the second equation, except for special cases, the new function $%
\Omega $ is independent of $t$; it depends only on the dynamical variables $%
x,y$ and possibly on their derivative. The Ermakov first integral in the new
variables is:%
\begin{equation}
I=\frac{1}{2}(x\dot{y}-y\dot{x})^{2}+\int^{y/x}F\left( \lambda \right)
d\lambda .  \label{EqnEr.55}
\end{equation}%
The system of equations (\ref{EqnEr.53})-(\ref{EqnEr.54}) defines the most
general 2D Ermakov system and produces all its known forms for special
choices of the function $\Omega .$ For example, the weak Kepler Ermakov
system \cite{Leach1991} is defined by the equations \cite{Athorne1991}
\begin{eqnarray}
\ddot{x}+\omega ^{2}(t)x+\frac{x}{r^{3}}H\left( x,y\right) -\frac{1}{x^{3}}%
f\left( \frac{y}{x}\right) &=&0  \label{Kepler-Ermakov 01} \\
\ddot{y}+\omega ^{2}(t)y+\frac{y}{r^{3}}H\left( x,y\right) -\frac{1}{y^{3}}%
g\left( \frac{y}{x}\right) &=&0  \label{Kepler-Ermakov 02}
\end{eqnarray}%
where $H,f,g~$are arbitrary functions of their argument, \ $\Omega $ is of
the form%
\begin{equation}
\Omega ^{2}(x,y)=\omega ^{2}(t)+H\left( x,y\right) /r^{3}  \label{EqnEr.55d}
\end{equation}%
and the Ermakov first integral becomes
\begin{equation}
I=\frac{1}{2}\left( x\dot{y}-y\dot{x}\right) ^{2}+\int^{\frac{y}{x}}\left[
\lambda f\left( \lambda \right) -\lambda ^{-3}g\left( \lambda \right) \right]
d\lambda .  \label{KeplerErma.1}
\end{equation}%
The weak Kepler Ermakov system does not admit Lie point symmetries. However,
the property of having a first integral prevails. The system of equations (%
\ref{Kepler-Ermakov 01}), (\ref{Kepler-Ermakov 02}) admits the $sl\left(
2,R\right) ~$as Lie point symmetries \cite{LeachK} only for $H\left(
x,y\right) =-\mu ^{2}r^{3}+\frac{h\left( \frac{y}{x}\right) }{x}$~where $\mu
$ is either a real or a pure imaginary number. \ This is the Kepler Ermakov
system defined by the equations

\begin{eqnarray}
\ddot{x}+\left( \omega ^{2}(t)-\mu ^{2}\right) x+\frac{1}{r^{3}}h\left(
\frac{y}{x}\right) -\frac{1}{x^{3}}f\left( \frac{y}{x}\right) &=&0
\label{Kepler-Ermakov 1} \\
\ddot{y}+\left( \omega ^{2}(t)-\mu ^{2}\right) y+\frac{1}{r^{3}}\frac{y}{x}%
h\left( \frac{y}{x}\right) -\frac{1}{y^{3}}g\left( \frac{y}{x}\right) &=&0.
\label{Kepler-Ermakov 2}
\end{eqnarray}

It is well known (see \cite{LeachK}) that the oscillator term $\omega
^{2}(t)-\mu ^{2}$ in (\ref{Kepler-Ermakov 1})-(\ref{Kepler-Ermakov 2}) is
removed if one considers new variables $T,X,Y$ defined by the relations:%
\begin{equation}
T=\int \rho ^{-2}dt,X=\rho ^{-1}x~,Y=\rho ^{-1}y  \label{EqnEr.55fa}
\end{equation}%
where $\rho $ is any smooth solution of the time dependent oscillator
equation%
\begin{equation}
\ddot{\rho}+\left( \omega ^{2}(t)-\mu ^{2}\right) \rho =0.
\label{EqnEr.55fb}
\end{equation}

In \cite{LeachK} it is commented that "the effect of $\mu ^{2}$ is to shift
the time dependent frequency function". However this is true as long as $%
\omega (t)\neq 0.$ When $\ \omega (t)=0,$ one has the autonomous Kepler
Ermakov system whose Lie symmetries span the $sl(2,R)\ $algebra with
different representations for $\mu =0$ and $\mu \neq 0.$

Before we justify the need for the consideration of the two cases $\mu =0$
and $\mu \neq 0$, we note that by applying the transformation
\begin{equation}
~s=\int v^{-2}dT~,~\bar{x}=v^{-1}X~,~\bar{y}=v^{-1}Y  \label{EqnEr.55fc}
\end{equation}%
where $\nu $ satisfies the Ermakov Pinney equation
\begin{equation}
\frac{d^{2}v}{dT^{2}}+\frac{\mu ^{2}}{v^{3}}=0  \label{EqnEr.55fd}
\end{equation}%
to the transformed equations
\begin{eqnarray}
\frac{d^{2}X}{dT^{2}}+\frac{1}{R^{3}}h\left( \frac{Y}{X}\right) -\frac{1}{%
X^{3}}f\left( \frac{Y}{X}\right) &=&0  \label{EqnEr.55e} \\
\frac{d^{2}Y}{dT^{2}}+\frac{1}{R^{3}}\frac{Y}{X}h\left( \frac{Y}{X}\right) -%
\frac{1}{Y^{3}}g\left( \frac{Y}{X}\right) &=&0  \label{EqnEr.55ff}
\end{eqnarray}%
we retain the term $\mu ^{2}$ and obtain the autonomous Kepler Ermakov
system of \cite{MoyoL}
\begin{eqnarray}
\ddot{x}-\mu ^{2}x+\frac{1}{r^{3}}h\left( \frac{y}{x}\right) -\frac{1}{x^{3}}%
f\left( \frac{y}{x}\right) &=&0  \label{EqnEr.14} \\
\ddot{y}-\mu ^{2}y+\frac{1}{r^{3}}\frac{y}{x}h\left( \frac{y}{x}\right) -%
\frac{1}{y^{3}}g\left( \frac{y}{x}\right) &=&0.  \label{EqnEr.14a}
\end{eqnarray}

The above transformations show that the consideration of the autonomous
Kepler Ermakov system is not a real restriction.

We discuss now the need for the consideration of the cases $\mu =0$ and $\mu
\neq 0.$ In Section \ref{sectio3ap}, we have determined the Lie symmetries
of the autonomous 2d Kepler Ermakov system and we have found two cases. Case
I concerns the autonomous Kepler Ermakov system with $\mu =0$ and has the
Lie symmetry vectors

\begin{equation}
\mathbf{X}=\left( \bar{c}_{1}+\bar{c}_{2}2t+\bar{c}_{3}t^{2}\right) \partial
_{t}+\left( \bar{c}_{2}+\bar{c}_{3}t\right) r\partial _{r}\qquad (~\mu =0)
\label{Kepler-Ermakov 1b}
\end{equation}

The second case, Case II, concerns the same system with $\mu \neq 0$ and has
the Lie symmetry vectors%
\begin{equation}
\mathbf{X}=\left( c_{1}+c_{2}\frac{1}{\mu }e^{2\mu t}-c_{3}\frac{1}{\mu }%
e^{-2\mu t}\right) \partial _{t}+\left( c_{2}e^{2\mu t}+c_{3}e^{-2\mu
t}\right) ~r\partial _{r}~~~\ ~(~\mu \neq 0)  \label{Kepler-Ermakov 1a}
\end{equation}%
where~in both cases $r\partial _{R}=x\partial _{x}+y\partial _{y}~$is the
gradient HV of the 2D Euclidian metric. Each set of vectors in (\ref%
{Kepler-Ermakov 1a})-(\ref{Kepler-Ermakov 1b}) is a representation of the $%
sl(2,R)$ algebra and furthermore each set of vectors is constructed from the
vector $\partial _{t}$ and the gradient HV $r\partial _{r}$ of the Euclidian
two dimensional space $E^{2}$.

The essence of the difference\ between the two representations is best seen
in the corresponding first integrals. If a Kepler Ermakov system is
Hamiltonian then the Lie point symmetries are also Noether point symmetries
therefore in order to find these integrals we determine the Noether
invariants. The Noether symmetries of the Kepler Ermakov system have been
determined in Section \ref{sectio3ap}. For the convenience of the reader we
repeat the relevant material.

Equations (\ref{EqnEr.14}), (\ref{EqnEr.14a}) follow from the Lagrangian
\cite{LeachK}
\begin{equation}
L=\frac{1}{2}\left( \dot{r}^{2}+r^{2}\dot{\theta}^{2}\right) -\frac{\mu ^{2}%
}{2}r^{2}-\frac{C\left( \theta \right) }{2r^{2}}  \label{Kepler-Ermakov 3}
\end{equation}%
where $C(\theta )=c+\sec ^{2}\theta f(\tan \theta )+\csc ^{2}\theta g(\tan
\theta )$ provided the functions $f,g$ satisfy the constraint:%
\begin{equation}
\sin ^{2}\theta f^{\prime }\left( \tan \theta \right) +\cos ^{2}\theta
~g^{\prime }\left( \tan \theta \right) =0.  \label{Kepler-Ermakov 4}
\end{equation}%
The Ermakov invariant in this case is \cite{LeachK}.
\begin{equation}
J=r^{4}\dot{\theta}^{2}+2C\left( \theta \right) .  \label{EqnEr.18a}
\end{equation}

Because the system is autonomous the first Noether integral is the
Hamiltonian
\begin{equation}
E=\frac{1}{2}\left( \dot{r}^{2}+r^{2}\dot{\theta}^{2}\right) +\frac{1}{2}\mu
^{2}r^{2}+\frac{1}{r^{2}}F\left( \theta \right)
\end{equation}%
In addition to the Hamiltonian, there exist two additional time dependent
Noether integrals as follows:\newline
$\mu =0$%
\begin{eqnarray}
I_{1} &=&2tE-r\dot{r}~ \\
I_{2} &=&t^{2}E-tr\dot{r}+\frac{1}{2}r^{2}  \label{Kepler-Ermakov 6}
\end{eqnarray}%
$\mu \neq 0$
\begin{eqnarray}
I_{1}^{\prime } &=&\left( \frac{1}{\mu }E-r\dot{r}+\mu r^{2}\right) e^{2\mu
t}  \label{Kepler-Ermakov 7} \\
I_{2}^{\prime } &=&\left( \frac{1}{\mu }E+r\dot{r}+\mu r^{2}\right) e^{-2\mu
t}.  \label{Kepler-Ermakov 8}
\end{eqnarray}

We note that the Noether integrals corresponding to the representation (\ref%
{Kepler-Ermakov 1b}) are linear in $t,$ whereas the ones corresponding to
the representation (\ref{Kepler-Ermakov 1a}) are exponential. Therefore the
consideration of the cases $\mu =0$ and $\mu \neq 0$ is not spurious
otherwise we loose important information. This latter fact is best seen in
the applications of Noether symmetries to field theories where the main core
of the theory is the Lagrangian. In these cases the potential is given and,
as it has been shown in Chapter \ref{chapter3}, a given potential admits
certain Noether symmetries only; therefore one has to consider all possible
cases. We shall come to this situation in section \ref{The Riemannian Kepler
Ermakov system in cosmology} where it will be found that the potential
selects the representation (\ref{Kepler-Ermakov 1a}).

To complete this section, we mention that for a Hamiltonian Kepler Ermakov
system the Ermakov invariant (\ref{EqnEr.18a}) is constructed \cite{MoyoL}
from the Hamiltonian and the Noether invariants (\ref{Kepler-Ermakov 7}),(%
\ref{Kepler-Ermakov 8}) as follows:
\begin{equation*}
J=E^{2}-~I_{1}^{\prime }I_{2}^{\prime }.
\end{equation*}

Finally in \cite{MoyoL}, it is shown that the Ermakov invariant is generated
by a dynamical Noether symmetry of the Lagrangian (\ref{Kepler-Ermakov 3}),
a result which is also confirmed in \cite{HaasGoedert2001}.

\section{Generalizing the autonomous Kepler Ermakov system}

\label{Generalizing the Kepler Ermakov system}

We consider the generalization of the two dimensional autonomous Kepler
Ermakov system \cite%
{Leach1991,GoedertHaas1998,Reid2,AthorneC,GovinderL,GovinderL2} using a
geometric point of view. From the results presented so far we have the
following:

(i) Equations (\ref{EqnEr.50})-(\ref{EqnEr.51}) which define the Ermakov
system employ coordinates in the Euclidian two dimensional space, therefore
the system is the \emph{Euclidian} Ermakov system.

(ii) The autonomous 2D Euclidian Kepler Ermakov system is defined by
equations (\ref{EqnEr.14}) and (\ref{EqnEr.14a})

(iii) The Lie symmetries of the Kepler Ermakov system span the $sl(2,R)$
algebra. These symmetries are constructed from the vector $\partial_t$ and
the gradient HV of the space $E^{2}$.

(iv) For the autonomous Hamiltonian Kepler Ermakov system the Lie symmetries
reduce to Noether point symmetries and the Ermakov invariant follows from a
combination of the resulting three Noether integrals, two of which are time
dependent. Furthermore, the Ermakov invariant is the Noether integral of a
dynamical Noether symmetry.

The above observations imply that we may generalize the Kepler Ermakov
system in two directions:

a. Increase the number of dimensions by defining the $n~$dimensional
Euclidian Kepler Ermakov system and/or

b. Generalize the background Euclidian space to be a Riemannian space and
obtain the Riemannian Kepler Ermakov system.

Concerning the defining characteristics of the Kepler Ermakov system we
distinguish three different properties of reduced generality: The property
of having a first integral; the property of admitting Lie/Noether point
symmetries, the $sl(2,R)$ invariance and the property of being Hamiltonian
and admitting $sl(2,R)$ invariance via Noether point symmetries.

Following Leach \cite{Leach1991} we generalize the autonomous Kepler Ermakov
system to higher dimensions by the requirement: The generalized autonomous
(Euclidian) Kepler Ermakov system admits the $sl(2,R)$\ algebra as a Lie
symmetry algebra.

\section{The three dimensional autonomous Euclidian Kepler Ermakov system}

\label{The 3d Euclidian Kepler Ermakov system}

The generalization of the autonomous Euclidian Kepler Ermakov system using $%
sl(2,R)$ invariance of Lie symmetries has been done in \cite%
{Leach1991,GovinderL,GovinderL2}. In this section, using of the results of
Chapter \ref{chapter3}, we give the generalization of the autonomous
Euclidian Hamiltonian Kepler Ermakov system to three dimensions by demanding
$sl(2,R)$ invariance with respect to Noether point symmetries. The reason
for attempting this generalization is that it leads to the potentials for
which the corresponding extended systems are Liouville integrable.\textbf{\ }%
Furthermore indicates the path to the $n$ dimensional Riemannian Kepler
Ermakov system.

Depending on the value $\mu \neq 0$ or $\mu =0,$ we consider the three
dimensional Hamiltonian Kepler Ermakov systems of type I and type II.

\subsection{The 3D autonomous Hamiltonian Kepler Ermakov system of type I $(%
\protect\mu \neq 0)$}

For $\mu \neq 0,$ the admitted Noether symmetries are required to be (see (%
\ref{Kepler-Ermakov 1a}))%
\begin{equation}
X^{1}=\partial _{t},~X_{\pm }=\frac{1}{\mu }e^{\pm 2\mu t}\partial _{t}\pm
e^{\pm 2\mu t}R\partial _{R}.  \label{AEP.05.0a}
\end{equation}%
From Table \ref{3DNNS2} and $T\left( t\right) =\frac{1}{\mu }e^{\pm 2\mu t}$
of section \ref{Noether point symmetries1}, we find that for these vectors
the potential is
\begin{equation*}
V\left( R,\phi ,\theta \right) =-\frac{~\mu ^{2}}{2}R^{2}+\frac{1}{R^{2}}%
f\left( \theta ,\phi \right)
\end{equation*}%
hence, the Lagrangian is%
\begin{equation}
L=\frac{1}{2}\left( \dot{R}^{2}+R^{2}\dot{\phi}^{2}+R^{2}\sin ^{2}\phi ~\dot{%
\theta}^{2}\right) +\frac{~\mu ^{2}}{2}R^{2}-\frac{1}{R^{2}}f\left( \theta
,\phi \right) .  \label{AEP.00}
\end{equation}

The equations of motion, that is, the equations defining the generalized
dynamical system are
\begin{eqnarray}
\ddot{R}-R\dot{\phi}^{2}-R\sin ^{2}\phi ~\dot{\theta}^{2}-\mu ^{2}R-\frac{2}{%
R^{3}}f &=&0  \label{AEP.01a} \\
\ddot{\phi}+\frac{2}{R}\dot{R}\dot{\phi}-\sin \phi \cos \phi ~\dot{\theta}%
^{2}+\frac{1}{R^{4}}f_{,\phi } &=&0  \label{AEP.01b} \\
\ddot{\theta}+\frac{2}{R}\dot{R}\dot{\theta}+\cot \phi ~\dot{\theta}\dot{\phi%
}+\frac{1}{R^{4}\sin ^{2}\phi }f_{,\theta } &=&0.  \label{AEP.01c}
\end{eqnarray}%
The Noether integrals corresponding to the Noether vectors are%
\begin{eqnarray}
E &=&\frac{1}{2}\left( \dot{R}^{2}+R^{2}\dot{\phi}^{2}+R^{2}\sin ^{2}\phi ~%
\dot{\theta}^{2}\right) -\frac{~\mu ^{2}}{2}R^{2}+\frac{1}{R^{2}}f\left(
\theta ,\phi \right)  \label{AEP.02} \\
I_{+} &=&\frac{1}{\mu }e^{2\mu t}E-e^{2~\mu t}R\dot{R}+\mu e^{2\mu t}R^{2}
\label{AEP.03} \\
I_{-} &=&\frac{1}{\mu }e^{-2\mu t}E+e^{-2\mu t}R\dot{R}+\mu e^{-2\mu t}R^{2}
\label{AEP.04}
\end{eqnarray}%
where $E$ is the Hamiltonian. The Noether integrals $I_{\pm }$ are time
dependent. Following \cite{MoyoL} we define the time independent combined
first integral%
\begin{equation}
J=E^{2}-~I_{+}I_{-}=R^{4}\dot{\phi}^{2}+R^{4}\sin ^{2}\phi ~\dot{\theta}%
^{2}+2f\left( \theta ,\phi \right) .  \label{AEP.05}
\end{equation}%
Using (\ref{AEP.05}) the equation of motion (\ref{AEP.01a})\ becomes
\begin{equation}
\ddot{R}-\mu ^{2}R=\frac{J}{R^{3}}  \label{AEP.05b}
\end{equation}%
which is the autonomous Ermakov- Pinney equation \cite{Pinney}. Therefore $J$
is the Ermakov invariant \cite{Leach1991}.

An alternative way to construct the Ermakov invariant (\ref{AEP.05}) is to
use dynamical Noether symmetries \cite{Kalotas}. Indeed one can show that
the Lagrangian (\ref{AEP.00}) admits the dynamical Noether symmetry $%
X_{D}=K_{j}^{i}\dot{x}^{j}\partial _{i}$ where $K_{ij}$ is a Killing tensor
of the second rank whose non-vanishing components are $K_{\phi \phi
}=R^{4}~,~K_{\theta \theta }=R^{4}\sin ^{2}\phi .~$The dynamical Noether
symmetry vector is~$X_{D}=R^{2}\left( \dot{\phi}\partial _{\phi }+\dot{\theta%
}\partial _{\theta }\right) ~$with gauge function $2f\left( \theta ,\phi
\right) $.

\subsection{The 3D autonomous Hamiltonian Kepler Ermakov system of type II $(%
\protect\mu =0)$}

For $\mu =0,$ the Noether point symmetries are required to be (see (\ref%
{Kepler-Ermakov 1b})) \cite{Leach1991}%
\begin{equation}
X^{1}=\partial _{t}~,~X^{2}=2t\partial _{t}~+R\partial
_{R}~,~X^{3}=t^{2}\partial _{t}~+tR\partial _{R}.  \label{AEP.05c}
\end{equation}%
From Table \ref{3DNNS1} and from Table \ref{3DNNS2} for $T(t)=t$ of section %
\ref{Noether point symmetries1}, we find that the potential is
\begin{equation*}
V\left( R,\phi ,\theta \right) =\frac{1}{R^{2}}f\left( \theta ,\phi \right)
\end{equation*}%
hence, the Lagrangian is%
\begin{equation}
L^{\prime }=\frac{1}{2}\left( \dot{R}^{2}+R^{2}\dot{\phi}^{2}+R^{2}\sin
^{2}\phi \dot{\theta}^{2}\right) -\frac{1}{R^{2}}f\left( \theta ,\phi
\right) .  \label{AEP.06}
\end{equation}

The equations of motion are (\ref{AEP.01a}) - (\ref{AEP.01c}) with $\mu =0$%
.~The Noether invariants of the Lagrangian (\ref{AEP.06}) are%
\begin{eqnarray}
E &=&\frac{1}{2}\left( \dot{R}^{2}+R^{2}\dot{\phi}^{2}+R^{2}\sin ^{2}\phi
\dot{\theta}^{2}\right) +\frac{1}{R^{2}}f\left( \theta ,\phi \right)
\label{AEP.07} \\
I_{1} &=&2tE^{\prime }-R\dot{R}  \label{AEP.08} \\
I_{2} &=&t^{2}E^{\prime }-tR\dot{R}+\frac{1}{2}R^{2}.  \label{AEP.09}
\end{eqnarray}

We note that the time dependent first integrals $I_{1,2}$ are linear in $t$
whereas the corresponding integrals $I_{\pm }$ of the case $\mu \neq 0$ are
exponential. From $I_{1,2}$ we define the time independent first integral $%
J=4I_{2}E^{\prime }-I_{1}^{2}$ which is calculated to be
\begin{equation}
J=R^{4}\dot{\phi}^{2}+R^{4}\sin ^{2}\phi ~\dot{\theta}^{2}+2f\left( \theta
,\phi \right) .  \label{AEP.10}
\end{equation}%
Using (\ref{AEP.10}) the equation of motion for $R\left( t\right) $ becomes~$%
\ddot{R}-\frac{J^{\prime }}{R^{3}}=0$\ which is the one dimensional
Ermakov-Pinney equation, hence $J^{\prime }$ is the Ermakov invariant \cite%
{Leach1991}. As it was the case with the three dimensional Hamiltonian
Kepler Ermakov system of type I, the Lagrangian (\ref{AEP.06}) admits the
dynamical Noether symmetry $X_{D}=R^{2}\left( \dot{\phi}\partial _{\phi }+%
\dot{\theta}\partial _{\theta }\right) $ whose\ integral is the (\ref{AEP.10}%
).

\section{Integrability of 3D autonomous Euclidian Kepler Ermakov system}

The 3d autonomous Hamiltonian Euclidian Kepler Ermakov systems need three
independent first integrals in involution in order to be Liouville
integrable. As we have shown each system has the two Noether integrals $%
(E,J) $ , therefore, we need one more Noether symmetry. Such a symmetry
exists only for special forms of the arbitrary function $f\left( \theta
,\phi \right) $. From tables \ref{3DNNS1}, \ref{3DNNS2}, \ref{3DNNS3}, \ref%
{3DNNS4} and \ref{3DNNS5} of Chapter \ref{chapter3} we find that extra
Noether symmetries are possible only\footnote{%
The linear combination of an element of $so(3)$ with a translation does not
give a potential, hence, an additional Noether symmetry.} for linear
combinations of translations (i.e. vectors of the form $\sum%
\limits_{A=1}^{3}a^{A}\partial _{A}$ where $a^{A}$ are constants) and/or
rotations (i.e. elements of $so(3))$.

\subsection{ Noether symmetries generated from the translation group}

We determine the functions $f\left( \theta ,\phi \right) $ for which the 3D
autonomous Hamiltonian Euclidian Kepler Ermakov system admits extra Noether
point symmetries for linear combinations of the translation group.

\underline{The Lagrangian (\ref{AEP.00})}\newline

In Cartesian coordinates the Lagrangian (\ref{AEP.00}) is%
\begin{equation}
L\left( x^{j},\dot{x}^{j}\right) =\frac{1}{2}\left( \dot{x}^{2}+\dot{y}^{2}+%
\dot{z}^{2}\right) +\frac{\mu ^{2}}{2}\left( x^{2}+y^{2}+z^{2}\right) -\frac{%
1}{x^{2}}f_{I}\left( \frac{y}{x},\frac{z}{x}\right)  \label{LIEP.01}
\end{equation}%
where $f_{I}=\left( 1+\frac{y^{2}}{x^{2}}+\frac{z^{2}}{x^{2}}\right) ^{-1}.$
From Table \ref{3DNNS5} with $m=-\mu ^{2}~,~p=0$ we find that the Lagrangian
(\ref{LIEP.01}) admits Noether symmetries, which are produced from a linear
combination of translations, $\ $if the function $f_{I}\left( \frac{y}{x},%
\frac{z}{x}\right) $ has the form%
\begin{equation}
f_{I}\left( \frac{y}{x},\frac{z}{x}\right) =\frac{1}{\left( 1-\frac{a}{b}%
\frac{y}{x}\right) ^{2}}F\left( \frac{b\frac{z}{x}-c\frac{y}{x}}{\left( 1-%
\frac{a}{b}\frac{y}{x}\right) }\right) .  \label{LIEP.02}
\end{equation}

In this case, the Lagrangian (\ref{LIEP.01}) admits at least the following
two extra Noether symmetries%
\begin{equation}
X_{\pm }=e^{\pm \mu t}{\sum\limits_{A=1}^{3}}a^{A}\partial _{A}
\label{LIEP.03a}
\end{equation}%
with corresponding Noether integrals%
\begin{equation}
I_{\pm }=e^{\pm \mu t}\left( {\sum\limits_{A=1}^{3}}a^{A}\dot{x}_{A}\right)
\mp \mu e^{\pm \mu t}\left( \sum\limits_{A=1}^{3}a^{A}x_{A}\right) .
\label{LIEP.04}
\end{equation}%
We note that the first integrals $I_{\pm }$ are time dependent; however the
first integral
\begin{equation}
J_{2}=I_{+}I_{-}=\left( a\dot{x}+b\dot{y}+c\dot{z}\right) ^{2}+\mu
^{2}\left( ax+by+cz\right) ^{2}  \label{LIEP.05}
\end{equation}%
is time independent. As it was the case with the Ermakov invariant (\ref%
{AEP.05}) the integral $J_{2}$ is possible to be constructed directly from
the dynamical Noether symmetry $X_{D}^{\prime }=K_{\left( 2\right) .j}^{i}%
\dot{x}^{i}\partial _{i}$, where $K_{(2)ij}$ is a Killing tensor of the
second rank \cite{Kalotas,Crampin}, with non-vanishing components
\begin{eqnarray*}
K_{11} &=&a^{2}~,~K_{22}=b^{2}~,~K_{33}=c^{2} \\
K_{\left( 12\right) } &=&2ab~,~K_{\left( 13\right) }=2ac~,~K_{\left(
23\right) }=2bc
\end{eqnarray*}%
so that the dynamical symmetry vector is
\begin{equation}
X_{D}^{\prime }=\left( a^{2}+ab+ac\right) \dot{x}\partial _{x}+\left(
b^{2}+ab+bc\right) \dot{y}\partial _{y}+\left( c^{2}+ac+bc\right) \dot{z}%
\partial _{z}.
\end{equation}%
The Ermakov invariant $J$ (see (\ref{AEP.05}) ) in Cartesian coordinates is
\begin{equation}
J=2E\left( x^{2}+y^{2}+z^{2}\right) -\left( x\dot{x}+y\dot{y}+z\dot{z}%
\right) ^{2}.  \label{LIEP.05b}
\end{equation}

The first integrals $J,J_{2}$ are not in involution. Using the Poisson
brackets we construct new first integrals and at some stage one of them will
be in involution. These new first integrals can also be constructed form
corresponding dynamical Noether symmetries.

An example of a\ known Lagrangian of the form (\ref{LIEP.01}) is the three
body Calogero-Moser Lagrangian \cite{Wojcie,HaasJPA,Ranada}
\begin{equation}
L=\frac{1}{2}\left( \dot{x}^{2}+\dot{y}^{2}+\dot{z}^{2}\right) -\frac{\mu
^{2}}{2}\left( x^{2}+y^{2}+z^{2}\right) -\frac{1}{\left( x-y\right) ^{2}}-%
\frac{1}{\left( x-z\right) ^{2}}-\frac{1}{\left( y-z\right) ^{2}}.
\end{equation}%
The extra Noether symmetries of this Lagrangian are produced by the vector (%
\ref{LIEP.03a})$~$for $a^{A}=\left( 1,1,1\right) .$

\underline{The Lagrangian (\ref{AEP.06})}\newline

In Cartesian coordinates the Lagrangian (\ref{AEP.06}) is%
\begin{equation}
L\left( x^{j},\dot{x}^{j}\right) =\frac{1}{2}\left( \dot{x}^{2}+\dot{y}^{2}+%
\dot{z}^{2}\right) -\frac{1}{x^{2}}f_{II}\left( \frac{y}{x},\frac{z}{x}%
\right) .  \label{LIEP.07}
\end{equation}%
According to Tables \ref{3DNNS4} and \ref{3DNNS5} (with $m=0~,~p=0$), the
Lagrangian (\ref{LIEP.07}) admits extra Noether point symmetries for a
linear combination of translations if the function $f$ is of the form (\ref%
{LIEP.02}). In this case the corresponding Noether integrals are%
\begin{equation}
I_{1}^{\prime }={\sum\limits_{A=1}^{3}}a^{A}\dot{x}_{A}~,~I_{2}^{\prime }=t{%
\sum\limits_{A=1}^{3}}a^{A}\dot{x}_{A}-{\sum\limits_{A=1}^{3}}a^{A}x_{A}.
\label{LIEP.10}
\end{equation}

Example of such a Lagrangian is the Calogero-Moser Lagrangian \cite{Wojcie}
(without the oscillator term)
\begin{equation}
L=\frac{1}{2}\left( \dot{x}^{2}+\dot{y}^{2}+\dot{z}^{2}\right) -\frac{1}{%
\left( x-y\right) ^{2}}-\frac{1}{\left( x-z\right) ^{2}}-\frac{1}{\left(
y-z\right) ^{2}}.  \label{LIEP.11}
\end{equation}

For the Lagrangian (\ref{LIEP.11}), we have the first integrals $%
E,J,I_{1}^{\prime },I_{2}^{\prime }$. From the integrals $J,I_{1}^{\prime }$
we construct the integral $\Phi =\left\{ I_{1}^{\prime },\left\{
J,I_{1}^{\prime }\right\} \right\} $ . It is easy to show that the integrals
$E,I_{1}^{\prime },\Phi $ are in involution hence the dynamical system is
Liouville integrable. We remark, that, the first integrals $%
E,J,I_{1}^{\prime },I_{2}^{\prime }$ can also be computed by making use of
the Lax pair tensor \cite{Ranada}.

\subsection{Noether symmetries generated from $so\left( 3\right) $}

The elements of $so\left( 3\right) $ in spherical coordinates are the three
vectors $CK^{1,2,3}~\ $%
\begin{equation}
CK^{1}=\sin \theta \partial _{\phi }+\cos \theta \cot \phi \partial _{\theta
},~CK^{2}=\cos \theta \partial _{\phi }-\sin \theta \cot \phi \partial
_{\theta },~CK^{3}=\partial _{\theta }  \label{CCS0.1}
\end{equation}%
which are also KVs for the Euclidian sphere.

In this case, the symmetry condition becomes%
\begin{equation}
L_{CK}\left[ \frac{1}{R^{2}}f\left( \theta ,\phi \right) \right] +p=0~
\label{CCS0.2}
\end{equation}%
or, equivalently%
\begin{equation}
\frac{1}{R^{2}}\left( R^{2}g_{ij}CK^{i}f^{,j}\right) +p=0\Rightarrow
g_{ij}CK_{\left( 1,2,3\right) }^{i}f^{,j}+p=0  \label{CCS0}
\end{equation}%
where $g_{ij}$ is the metric of the Euclidian sphere, that is%
\begin{equation}
ds^{2}=d\phi ^{2}+\mathrm{\sin }^{2}\phi ~d\theta ^{2}.  \label{CCS}
\end{equation}

We infer that the problem of determining the extra Noether point symmetries
of Lagrangian (\ref{AEP.00}) generated from elements of the $so\left(
3\right) $ is equivalent to the determination of the Noether point
symmetries for motion on the 2D sphere.

It is easy to show, that, the integrals of Table \ref{CCInt} of section \ref%
{Motion on the two dimensional sphere} are in involution with the
Hamiltonian and the Ermakov invariant, therefore, the system is Liouville
integrable via Noether point symmetries.

The above results are extended to the case in which the system moves on the
hyperbolic sphere that is, it has Lagrangian
\begin{equation}
L=\frac{1}{2}\left( \dot{R}^{2}+R^{2}\dot{\phi}^{2}+R^{2}\sinh ^{2}\phi ~%
\dot{\theta}^{2}\right) +\frac{\mu ^{2}}{2}R^{2}-\frac{1}{R^{2}}g\left(
\theta ,\phi \right) .  \label{EP.SO3}
\end{equation}

We reach at the following conclusion.

\begin{proposition}
The three dimensional autonomous Hamiltonian Kepler Ermakov system with
Lagrangian (\ref{AEP.00}) is Liouville Integrable via Noether point
symmetries, which are generated from a linear combination of the three
elements of the~$so\left( 3\right) $ algebra, if and only if the equivalent
dynamical system in the fundamental hyperquadrics of the three dimensional
flat space is integrable.
\end{proposition}

We note that it is possible a three dimensional autonomous Kepler Ermakov
system to admit more Noether symmetries which are due to the rotation group
and the translation group (but not to a linear combination of elements from
the two groups). For example, the 3D Kepler Ermakov system with Lagrangian
\cite{Damianou2004}%
\begin{equation}
L=\frac{1}{2}\left( \dot{x}^{2}+\dot{y}^{2}+\dot{z}^{2}\right) -\frac{1}{%
x^{2}\left( 1-\frac{y}{x}-\frac{z}{x}\right) ^{2}}
\end{equation}%
has the following extra Noether point symmetries (in addition to the
elements of $sl\left( 2,R\right) $)%
\begin{eqnarray*}
Y^{1} &=&\partial _{x}+\partial _{y},~Y^{2}=\partial _{x}+\partial _{z} \\
Y^{3} &=&t\left( \partial _{x}+\partial _{y}\right) ,~Y^{4}=t\left( \partial
_{x}+\partial _{z}\right) \\
Y^{5} &=&\left( y-z\right) \partial _{x}-\left( x+z\right) \partial
_{y}+\left( x+y\right) \partial _{z}.
\end{eqnarray*}%
The vectors $Y^{1,2}$ and $Y^{3,4}$ follow from (\ref{LIEP.03a}) for $%
a_{1}^{A}=\left( 1,1,0\right) ~$\ and $a_{2}^{A}=\left( 1,0,1\right) \ \ $%
respectively, whereas $Y^{5}$ is a linear combination of the three elements
of $so\left( 3\right) $. The \ Noether integrals of the Noether symmetries $%
Y^{1-5}$ are respectively
\begin{eqnarray}
I_{Y_{1}} &=&\dot{x}+\dot{y}~ \\
I_{Y_{2}} &=&\dot{x}+\dot{z} \\
I_{Y_{3}} &=&t\left( \dot{x}+\dot{y}\right) -\left( x+y\right) \\
I_{Y_{4}} &=&t\left( \dot{x}+\dot{z}~\right) -\left( x+z\right) \\
I_{Y_{5}} &=&\left( y-z\right) \dot{x}-\left( x+z\right) \dot{y}+\left(
x+y\right) \dot{z}.
\end{eqnarray}

It is clear that in order to extend the Kepler Ermakov system to higher
dimensions one needs to have the type of results of Chapter \ref{chapter3};
therefore, the remark made in \cite{Leach1991}, that the `notion is easily
generalized to higher dimensions' has to be understood as referring to the
general scenario and not to the actual work.

\section{The autonomous Riemannian Kepler Ermakov system}

\label{The Riemannian Kepler Ermakov system}

As it has been noted in section \ref{ErmakovSys}, the Kepler Ermakov systems
considered so far in the literature are Newtonian Kepler Ermakov systems. In
this section we make a drastic step forward and introduce the autonomous
Riemannian Kepler Ermakov systems of dimension $n.$ The generalization we
consider is based on the following definition

\begin{definition}
The $n~$dimensional autonomous Riemannian Kepler Ermakov system is an
autonomous dynamical system which:\newline
a. It is defined on a Riemannian space which admits a gradient HV\newline
b. Admits a first integral, which we name the Riemannian Ermakov first
integral and it is characterized by the requirement that the corresponding
equation of motion takes the form of the Ermakov Pinney equation.\newline
c. It is invariant at least under the $sl(2,R)$ algebra, which is generated
by the vector $\partial _{t}$ and the gradient HV\ of the space.
\end{definition}

There are two types of $n$ dimensional autonomous Riemannian Kepler Ermakov
systems. The ones which are not Hamiltonian and admit the $sl(2,R)$ algebra
as Lie point symmetries and the ones which are conservative and admit the $%
sl(2,R)$ algebra as Noether point symmetries.

\subsection{The non Hamiltonian autonomous Riemannian Kepler Ermakov system}

\label{The non conservative Riemannian Kepler Ermakov system}

Consider an $n~$dimensional Riemannian space which admits a gradient HV. It
is well known, that the metric of this space can always be written in the
form \cite{Tupper1989,Tupper1990}
\begin{equation}
ds^{2}=du^{2}+u^{2}h_{AB}dy^{A}dy^{B}  \label{NCGE.01}
\end{equation}%
where the Latin capital indices $A,B,..$ take the values $1,\ldots ,n-1~$and
$h_{AB}=h_{AB}\left( y^{C}\right) $ is the generic $n-1$ metric. The
gradient HV of the metric is the vector $H^{i}=u\partial _{u}$ ($\psi =1$)$~$%
generated from the function $H=\frac{1}{2}u^{2}$. We note the relation%
\begin{equation}
h_{DA}\Gamma _{BC}^{A}=\frac{1}{2}h_{DB,C}  \label{NCGE.01a}
\end{equation}%
where $\Gamma _{BC}^{A}$ are the connection coefficients of the $\left(
n-1\right) $ metric $h_{AB}.$ In that space, we consider a particle moving
under the action of the force
\begin{equation*}
F^{i}=F^{u}(u,y^{C})\frac{\partial }{\partial u}+F^{A}(u,y^{C})\frac{%
\partial }{\partial y^{A}}.
\end{equation*}

The equations of motion $\frac{Dx^{i}}{Dt}=F^{i}$ when projected along the
direction of $u$ and in the $\left( n-1\right) $ space give the equations
\begin{eqnarray}
u^{\prime \prime }-uh_{AB}y^{\prime A}y^{\prime B} &=&F^{u}  \label{NCGE.02}
\\
y^{\prime \prime A}+\frac{2}{u}u^{\prime }y^{\prime A}+\Gamma
_{BC}^{A}y^{\prime B}y^{\prime C} &=&F^{A}  \label{NCGE.03}
\end{eqnarray}%
where $u^{\prime }=\frac{du}{ds}$ and $s~$is an affine parameter.

Because the system is autonomous admits the Lie point symmetry $\partial
_{t} $. Using the vector $\partial _{t}$ and the gradient HV $%
H^{i}=u\partial _{u} $ we construct two representations of $sl(2,R)$ by
means of the sets of vectors (see (\ref{AEP.05.0a})~and (\ref{AEP.05c}))%
\begin{equation}
\partial _{s},~2s\partial _{s}+u\partial _{u},~s^{2}\partial _{t}+su\partial
_{u}\qquad \text{ when }\mu =0  \label{NCGE.04}
\end{equation}%
\begin{equation}
\partial _{s},~\frac{1}{\mu }e^{\pm 2\mu s}\partial _{s}\pm e^{\pm 2\mu
s}u\partial _{u}\qquad ~~~\text{ when }\mu \neq 0  \label{NCGE.05}
\end{equation}%
and require that the vectors in each set will be Lie point symmetries of the
system of equations (\ref{NCGE.02}),(\ref{NCGE.03}). In Appendix \ref%
{apen1force} we show that the requirement of the invariance of the force
under both representations (\ref{NCGE.04}), (\ref{NCGE.05}) of $sl(2,R)$
demands that the force be of the form%
\begin{equation}
F^{i}=\left( \mu ^{2}u+\frac{1}{u^{3}}G^{u}\left( y^{C}\right) \right)
\partial _{u}+\frac{1}{u^{4}}G^{A}\left( y^{C}\right) \partial _{A}.
\label{NCGE.05a}
\end{equation}%
Replacing $F^{i}$ in the system of equations (\ref{NCGE.02}),(\ref{NCGE.03})
we find
\begin{eqnarray}
u^{\prime \prime }-uh_{AB}y^{\prime A}y^{\prime B} &=&\mu ^{2}u+\frac{1}{%
u^{3}}G^{u}  \label{NCGE.06} \\
y^{\prime \prime A}+\frac{2}{u}u^{\prime }y^{\prime A}+\Gamma
_{BC}^{A}y^{\prime B}y^{\prime C} &=&\frac{1}{u^{4}}G^{A}.  \label{NCGE.07}
\end{eqnarray}%
\ Multiplying the second equation with $2u^{4}h_{DA}y^{\prime D}$ and using (%
\ref{NCGE.01a}) we have%
\begin{equation}
u^{4}\frac{d}{ds}\left( h_{DB}y^{\prime D}y^{\prime B}\right)
+4u^{3}h_{DA}u^{\prime }y^{\prime A}y^{\prime D}=2G_{D}y^{\prime D}
\end{equation}%
from which follows%
\begin{equation}
\frac{d}{ds}\left( u^{4}h_{DB}y^{\prime D}y^{\prime B}\right)
=2G_{D}y^{\prime D}.
\end{equation}

The rhs is a perfect differential if~$G_{D}=-\Sigma _{,D}~$where $\Sigma
(y^{A})$ is a differentiable function. If this is the case we find the first
integral
\begin{equation}
J=u^{4}h_{DB}y^{\prime D}y^{\prime B}+2\Sigma \left( y^{C}\right) .
\label{NCGE.010}
\end{equation}%
We note that $J$ involves the arbitrary metric $h_{AB}$ and the function $%
\Sigma (y^{A}$). Furthermore equations (\ref{NCGE.06}), (\ref{NCGE.07}%
)~become
\begin{eqnarray}
u^{\prime \prime }-uh_{AB}y^{\prime A}y^{\prime B} &=&\mu ^{2}u+\frac{1}{%
u^{3}}G^{u}\left( y^{C}\right)  \label{NCGE.011} \\
y^{\prime \prime A}+\frac{2}{u}u^{\prime }y^{\prime A}+\Gamma
_{BC}^{A}y^{\prime B}y^{\prime C} &=&-\frac{1}{u^{4}}h^{AB}\Sigma \left(
y^{C}\right) _{,B}~.  \label{NCGE.012}
\end{eqnarray}%
These are the equations defining the $n~$dimensional autonomous Riemannian
Kepler Ermakov system.

Using the first integral (\ref{NCGE.010}), the equation of motion (\ref%
{NCGE.011}) is written as follows
\begin{equation}
u^{\prime \prime }=\mu ^{2}u+\frac{\bar{G}\left( y^{C}\right) }{u^{3}}
\end{equation}%
where $\bar{G}=J+G^{u}\left( y^{C}\right) -2\Sigma \left( y^{C}\right) .$
This is the Ermakov-Pinney equation; hence, we identify (\ref{NCGE.010}) as
the Riemannian Ermakov integral of the autonomous Riemannian Kepler Ermakov
system.

\subsection{The autonomous conservative Riemannian Kepler Ermakov system}

\label{The Hamiltonian Kepler Ermakov system in a n-dimensional Riemannian
space}

In the following we assume that the force is derived from the potential $%
V\left( u,y^{C}\right) ,~$that is, the dynamical system is conservative so
that the equations of motion follow from the Lagrangian%
\begin{equation}
L=\frac{1}{2}\left( u^{\prime 2}+u^{2}h_{AB}y^{\prime A}y^{\prime B}\right)
-V\left( u,y^{C}\right) .  \label{GERS.02}
\end{equation}%
The Hamiltonian is%
\begin{equation}
E=\frac{1}{2}\left( u^{\prime 2}+u^{2}h_{AB}y^{\prime A}y^{\prime B}\right)
+V\left( u,y^{C}\right) .  \label{GERS.03}
\end{equation}%
The equations of motion, i.e. the Euler-Lagrange equations, are%
\begin{eqnarray}
u^{\prime \prime }-uh_{AB}y^{\prime A}y^{\prime B}+V_{,u} &=&0
\label{GERS.03a} \\
y^{\prime \prime A}+\frac{2}{u}u^{\prime }y^{\prime A}+\Gamma
_{BC}^{A}y^{\prime B}y^{\prime C}+\frac{1}{u^{2}}h^{AB}V_{,B} &=&0.
\label{GERS.03b}
\end{eqnarray}%
The demand that Lagrangian (\ref{GERS.02}) admits Noether point symmetries
which are generated from the gradient HV leads to the following cases.

\textbf{Case A:} The Lagrangian (\ref{GERS.02}) admits the Noether point
symmetries (\ref{NCGE.04}) if the potential is of the form
\begin{equation}
V\left( u,y^{C}\right) =\frac{1}{u^{2}}V\left( y^{C}\right) .  \label{GERS.5}
\end{equation}%
The Noether integrals of these Noether point symmetries are
\begin{eqnarray}
E_{A} &=&\frac{1}{2}\left( u^{\prime 2}+u^{2}h_{AB}y^{\prime A}y^{\prime
B}\right) +\frac{1}{u^{2}}V\left( y^{C}\right)  \label{GERSN.1} \\
I_{1} &=&2sE-uu^{\prime }  \label{GERSN.2} \\
I_{2} &=&s^{2}E-suu^{\prime }+\frac{1}{2}u^{2}  \label{GERSN.3}
\end{eqnarray}%
where $E_{A}$ is the $\ $Hamiltonian.

\textbf{Case B: }The Lagrangian (\ref{GERS.02}) admits the Noether point
symmetries (\ref{NCGE.05}) if the potential is of the form
\begin{equation}
V\left( u,y^{c}\right) =-\frac{\mu ^{2}}{2}u^{2}+\frac{1}{u^{2}}V^{\prime
}\left( y^{C}\right) .  \label{GERS.6}
\end{equation}%
The Noether integrals of these Noether point symmetries are
\begin{eqnarray}
E_{B} &=&\frac{1}{2}\left( u^{\prime 2}+u^{2}h_{AB}y^{\prime A}y^{\prime
B}\right) -\frac{\mu ^{2}}{2}u^{2}+\frac{1}{u^{2}}V^{\prime }\left(
y^{C}\right)  \label{GERSN.4} \\
I_{+} &=&\frac{1}{\mu }e^{2\mu s}E-e^{2\mu s}uu^{\prime }+\mu e^{2\mu s}u^{2}
\label{GERSN.5} \\
I_{-} &=&\frac{1}{\mu }e^{-2\mu s}E+e^{-2\mu s}uu^{\prime }+\mu e^{-2\mu
s}u^{2}  \label{GERSN.6}
\end{eqnarray}%
where $E_{B}$ is the $\ $Hamiltonian.

Using the Noether integrals we construct the Riemannian Ermakov invariant $%
J_{G}$, which is common for both Case A and Case B, as follows%
\begin{equation}
J_{G}=u^{4}h_{DB}y^{\prime D}y^{\prime C}+2V^{\prime }\left( y^{C}\right) .
\label{GERSN.7}
\end{equation}%
This coincides with the invariant first integral defined in (\ref{NCGE.010}%
). We note that with the use of the first integral (\ref{GERSN.7}) the
Hamiltonians (\ref{GERSN.1}) and (\ref{GERSN.4}) take the form
\begin{equation}
E=\frac{1}{2}u^{\prime 2}-\frac{\mu ^{2}}{2}u^{2}+\frac{J}{2u^{2}}
\end{equation}%
which is the Hamiltonian for the Ermakov Pinney equation.

As it was the case with the Euclidian case of section \ref{ErmakovSys}, it
can be shown that the Riemannian Ermakov invariant (\ref{GERSN.7}) is due to
a dynamical Noether symmetry\cite{Kalotas}. We collect the results in the
following proposition.

\begin{proposition}
\label{prop ghv}In a Riemannian space with metric $g_{ij}$ which admits a
gradient HV, the equations of motion of a Hamiltonian system moving under
the action of the potential $(\mu \epsilon
\mathbb{C}
)$
\begin{equation}
V\left( u,y^{c}\right) =-\frac{\mu ^{2}}{2}u^{2}+\frac{1}{u^{2}}V^{\prime
}\left( y^{C}\right)  \label{GERSN.8}
\end{equation}%
admit the $sl(2,R)\ $invariance and also an invariant first integral, the
Riemannian Ermakov invariant. This latter quantity is also possible to be
identified as the Noether integral of a dynamical Noether symmetry.
\end{proposition}

Without going into details we state the following general result.

\begin{proposition}
\label{prop gkv}Consider an $n~$ dimensional Riemannian space with an $r$~
decomposable metric,\ which in the Cartesian coordinates $x_{1},\ldots
,x_{r},$ has the general form%
\begin{equation}
ds^{2}=p\eta _{\Sigma \Lambda }dz^{\Sigma }dz^{\Lambda
}+h_{ij}dx^{i}dx^{j}~\ ,~i,j=r+1,\ldots ,n~,~\Sigma =1,\ldots ,r
\label{GERSN.10}
\end{equation}%
where $\eta _{\Sigma \Lambda }$ is a flat non~degenerate metric (of
arbitrary signature). If there exists a potential, so that the vectors $%
e^{\pm \mu s}{\sum\limits_{M}}a^{M}\partial _{^{M}}$ are Noether point
symmetries, where $a_{M}$ are constants, with Noether integrals $~$%
\begin{equation}
I_{\pm }=e^{\pm \mu s}\sum\limits_{M}a^{M}z_{M}^{\prime }\mp \mu e^{\pm \mu
s}\sum\limits_{M}a^{M}z_{M}  \label{GERSN.11}
\end{equation}%
the combined first integral $I=I_{+}I_{-}$ is time independent and it is the
result of a dynamical Noether symmetry.
\end{proposition}

In the remaining sections we consider applications of the autonomous
Riemannian Kepler Ermakov system in General Relativity and in Cosmology.

\section{The autonomous Riemannian Kepler Ermakov system in General
Relativity}

\label{The RKEGR}

Below, we study the integrability of the Riemannian Kepler Ermakov system
via Noether point symmetries in a conformally flat spacetime which admits a
homothetic Lie algebra with a gradient (proper) HV.

\subsection{The Riemannian Kepler Ermakov system on a 4D FRW spacetime}

Consider the spatially flat FRW spacetime with metric
\begin{equation}
ds^{2}=du^{2}-u^{2}\left( dx^{2}+dy^{2}+dz^{2}\right) .  \label{Apl.1}
\end{equation}%
This metric admits the gradient HV $u\partial _{u}$ and six non gradient KVs
\cite{MM86,TNA} which are the KVs of $E^{3}$.

We consider the autonomous Riemannian Kepler Ermakov system defined by the
Lagrangian (see (\ref{GERSN.8}) ) ($\mu\in \mathbb{C}$ )
\begin{equation}
L=\frac{1}{2}\left( u^{\prime 2}-u^{2}\left( x^{\prime 2}+y^{\prime
2}+z^{\prime 2}\right) \right) +\frac{\mu ^{2}}{2}u^{2}-\frac{1}{u^{2}}%
V\left( x,y,z\right) .  \label{Apl.2}
\end{equation}

The Euler Lagrange equations are%
\begin{eqnarray}
u^{\prime \prime }+u\left( x^{\prime 2}+y^{\prime 2}+z^{\prime 2}\right)
-\mu ^{2}u-\frac{2V\left( x,y,z\right) }{u^{3}} &=&0 \\
x^{\sigma \prime \prime }+\frac{2}{u}u^{\prime }x^{\sigma \prime }-\frac{%
V^{,\sigma }\left( x,y,z\right) }{u^{4}} &=&0
\end{eqnarray}%
where $\sigma =1,2,3$. The Lagrangian (\ref{Apl.2}) has the form of the
Lagrangian (\ref{GERS.02}) for the potential $V\left( u,y^{C}\right) =-\frac{%
\mu ^{2}}{2}u^{2}+\frac{1}{u^{2}}V\left( x,y,z\right) $ hence according to
proposition \ref{prop ghv} possesses $sl(2,R)$ invariance under Noether
point symmetries for \emph{both} representations (\ref{NCGE.04}) and (\ref%
{NCGE.05}). The two time independent invariants are the Hamiltonian and the
Riemannian Ermakov invariant (proposition \ref{prop ghv})
\begin{eqnarray}
E &=&\frac{1}{2}\left( u^{\prime 2}-u^{2}\left( x^{\prime 2}+y^{\prime
2}+z^{\prime 2}\right) \right) -\frac{\mu ^{2}}{2}u^{2}+\frac{1}{u^{2}}%
V\left( x,y,z\right) \\
J_{G_{4}} &=&u^{4}\left( x^{\prime 2}+y^{\prime 2}+z^{\prime 2}\right)
+2V\left( x,y,z\right) .
\end{eqnarray}

Note that had we considered the representation\ (\ref{NCGE.04}) only (that
is we had set $\mu =0)$ then we would have lost all information concerning
the system defined for $\mu \neq 0!$ We emphasize that in applications to
Physics the major datum is the Lagrangian and not the equations of motion,
therefore one should not make mathematical assumptions which restrict the
physical generality.

To assure Liouville integrability we need one more Noether symmetry whose
Noether integral is in involution with $E,J_{G_{4}}.$ This is possible for
certain forms of the potential $V\left( x,y,z\right) .$ Using the general
results of section \ref{Lie point symmetries of three dimensional autonomous
Newtonian systems} where all 3D potentials are given which admit extra
Noether symmetries we find the results of Table \ref{TT4d1}.

\begin{table}[tbp] \centering%
\caption{Potentials for which the Kepler Ermakov in the 4D FRW space admit
Noether symmetries}%
\begin{tabular}{ccc}
\hline\hline
\textbf{Noether Symmetry} & $\mathbf{V(x,y,z)}$ & \textbf{Noether Integral}
\\ \hline
$a\partial _{\mu }+b\partial _{\nu }+c\partial _{\sigma }$ & $-\frac{p}{a}%
x_{\mu }+f\left( x^{\nu }-\frac{b}{a}x^{\mu },x^{\sigma }-\frac{c}{a}x^{\mu
}\right) $ & $I_{T3}=u^{2}\left( aI_{\mu }+bI_{\nu }+cI_{\sigma }\right) $
\\
$a\partial _{\mu }+b\partial _{\nu }+c\left( x_{\nu }\partial _{\mu }-x_{\mu
}\partial _{\nu }\right) $ & $\frac{p}{\left\vert c\right\vert }\arctan
\left( \frac{\left( b-cx_{\mu }\right) }{\left\vert \left( a+cx_{\nu
}\right) \right\vert }\right) $ & $I_{T2Rr}=u^{2}\left( aI_{\mu }+I_{\nu
}+I_{R_{\mu \nu }}\right) $ \\
& ~$+f\left( \frac{c}{2}r_{\left( \mu \nu \right) }-bx_{\mu }+ax_{\nu
},x_{\sigma }\right) $ &  \\
$a\partial _{\mu }+b\partial _{\nu }+c\left( x_{\sigma }\partial _{\mu
}-x_{\mu }\partial _{\sigma }\right) $ & $-\frac{p}{\left\vert c\right\vert }%
\arctan \left( \frac{\left\vert c\right\vert x_{\mu }}{\left\vert
a+cx_{\sigma }\right\vert }\right) $ & $I_{T2R2}=u^{2}\left( aI_{\mu
}+I_{\nu }+I_{R\mu \sigma }\right) $ \\
& $~+f\left( x_{\nu }-\frac{1}{\left\vert c\right\vert }\arctan \left( \frac{%
\left\vert c\right\vert x_{\mu }}{\left\vert a+cx_{\sigma }\right\vert }%
\right) ,\frac{1}{2}r_{\left( \mu \sigma \right) }-\frac{a}{c}x_{\sigma
}\right) $ &  \\
$a\partial _{\mu }+b\left( x_{\nu }\partial _{\mu }-x_{\mu }\partial _{\nu
}\right) +$ & $\frac{p}{\sqrt{b^{2}+c^{2}}}\arctan \left( \frac{\left(
ab+b^{2}x_{\nu }+bcx_{\sigma }\right) }{\left\vert bx_{\mu }\right\vert
\sqrt{b^{2}+c^{2}}}\right) +$ & $I_{T1R3}=u^{2}\left( aI_{\mu }+I_{R_{\mu
\nu }}+I_{R\mu \sigma }\right) $ \\
$~~+c\left( x_{\sigma }\partial _{\mu }-x_{\mu }\partial _{\sigma }\right) $
& $+f\left( x_{\mu }^{2}+x_{\nu }^{2}\left( 1-\frac{c^{2}}{b^{2}}\right)
+\left( \frac{2a}{b}+\frac{2c}{b}x_{\sigma }\right) x_{\nu },x_{\sigma }-%
\frac{c}{b}x_{\nu }\right) $ &  \\
$so\left( 3\right) $ linear combination & $F\left( b\tan \theta \sin \phi
+c\cos \phi -aM_{1}\right) $ & $I_{R3}=u^{2}\left( I_{R_{\mu \nu
}}+I_{R_{\mu \sigma }}+I_{R_{\nu \sigma }}\right) $ \\ \hline\hline
\end{tabular}%
\label{TT4d1}%
\end{table}%

\begin{proposition}
\label{FRW_Flat}The Lagrangian (\ref{Apl.2}) admits an extra Noether
symmetry if and only if the potential $V\left( x,y,z\right) $ has a form
given in Table \ref{TT4d1}\footnote{%
Where $I_{\mu }=\delta _{\mu \rho }\dot{x}^{\rho }~$and $I_{R\mu \nu
}=\delta _{\rho \lbrack \mu }\delta _{\nu ]\sigma }x^{\sigma }\dot{x}%
^{\prime \rho }$}.
\end{proposition}

For example if $V\left( x,y,z\right) =\left( x^{2}+y^{2}+z^{2}\right) ^{n}$
then the system admits three extra Noether symmetries which are the elements
of $so\left( 3\right) $. If $V\left( x,y,z\right) =V_{0}$ then the system
admits six extra Noether symmetries (the KVs of the three dimensional
Euclidian space).

\subsection{The Riemannian Kepler Ermakov system on a 3D FRW spacetime}

Consider the three dimensional Lorentzian metric
\begin{equation}
ds^{2}=du^{2}-u^{2}\left( dx^{2}+dy^{2}\right)  \label{AGE.00}
\end{equation}%
which admits the gradient HV $u\partial _{u}$ and the three KVs of the
Euclidian metric $E^{2}.$ In that space consider the Lagrangian
\begin{equation}
L^{\prime }=\frac{1}{2}\left( u^{\prime 2}-u^{2}\left( x^{\prime
2}+y^{\prime 2}\right) \right) +\frac{\mu ^{2}}{2}u^{2}-\frac{1}{u^{2}}%
V\left( x,y\right) .  \label{AGE.01}
\end{equation}%
According to proposition \ref{prop ghv} this Lagrangian admits as Noether
point symmetries the elements of $sl\left( 2,R\right) .~$Then from
proposition \ref{prop ghv}\ we have that the Noether\ invariants of these
symmetries are%
\begin{eqnarray}
E &=&\frac{1}{2}\left( u^{\prime 2}-u^{2}\left( x^{\prime 2}+y^{\prime
2}\right) \right) -\frac{\mu ^{2}}{2}u^{2}+\frac{1}{u^{2}}V\left( x,y\right)
\\
J_{G_{3}} &=&u^{4}\left( x^{\prime 2}+y^{\prime 2}\right) +2V\left(
x,y\right) .
\end{eqnarray}

The requirement that the Lagrangian admits an additional Noether symmetry
leads to the condition $L_{KV}V\left( x,y\right) +p=0,$ therefore in that
case we have a 2D potential and we can use the results of Chapter \ref%
{chapter3}. If we demand the new Noether integral to be time independent~$%
\left( p=0\right) $ then the potential $V\left( x,y\right) $ and the new
Noether integrals are given in Table \ref{PtFRWEE}.

Lagrangians with kinetic term $T_{K}=\frac{1}{2}\left( u^{\prime
2}-u^{2}\left( x^{\prime 2}+y^{\prime 2}\right) \right) ~$appear in
cosmological models. In the following section we discuss such applications.

\begin{table}[tbp]\centering%
\caption{Potentials for which the Kepler Ermakov in the 3D FRW space admit
Noether symmetries}%
\begin{tabular}{ccc}
\hline\hline
\textbf{Noether Symmetry} & $\mathbf{V}\left( x,y\right) $ & \textbf{Noether
Integral} \\ \hline
$\partial _{x}$ & $f\left( y\right) $ & $I_{x}=u^{2}x^{\prime }$ \\
$\partial _{y}$ & $f\left( x\right) $ & $I_{y}=u^{2}y^{\prime }$ \\
$y\partial _{x}-x\partial _{y}$ & $f\left( x^{2}+y^{2}\right) $ & $%
I_{xy}=u^{2}\left( yx^{\prime }-xy^{\prime }\right) $ \\
$\partial _{x}+b\partial _{y}$ & $f\left( y-bx\right) $ & $%
I_{xby}=u^{2}\left( x^{\prime }+by^{\prime }\right) $ \\
$\left( a+y\right) \partial _{x}+\left( b-x\right) \partial _{y}$ & $f\left(
\frac{1}{2}\left( x^{2}+y^{2}\right) +ay-bx\right) $ & $I_{abxy}=\left(
a+y\right) u^{2}x^{\prime }+\left( b-x\right) u^{2}y^{\prime }$ \\
$\partial _{x},~\partial _{y}~,~y\partial _{x}-x\partial _{y}$ & $V_{0}$ & $%
I_{x}~,I_{y}~,~I_{xy}$ \\ \hline\hline
\end{tabular}%
\label{PtFRWEE}%
\end{table}%

\begin{center}
\
\end{center}

\section{The Riemannian Kepler Ermakov system in cosmology}

\label{The Riemannian Kepler Ermakov system in cosmology}Below, we consider
two cosmological models for dark energy, a scalar field cosmology and an $%
f(R)$ cosmology in a locally rotational symmetric (LRS) spacetime.

\subsection{The case of scalar field cosmology}

\label{scalar field cosmology}

Consider the Class A LRS spacetime
\begin{equation}
ds^{2}=-N^{2}\left( t\right) dt^{2}+a^{2}\left( t\right) e^{-2\beta \left(
t\right) }dx+a^{2}\left( t\right) e^{\beta \left( t\right) }\left(
dy^{2}+dz^{2}\right)  \label{LRS.01}
\end{equation}%
which is assumed to contain a scalar field with exponential potential $%
V\left( \phi \right) =V_{0}e^{-c\phi }~,~c\neq \sqrt{6k}$ and a perfect
fluid with a stiff equation of state $p=\rho $, where $p$ is the pressure
and $\rho $ is the energy density of the fluid. The conservation equation
for the matter density gives
\begin{equation}
\dot{\rho}+6\rho \frac{\dot{a}}{a}=0\rightarrow \rho =\frac{\rho _{0}}{a^{6}}%
.  \label{LRS.02}
\end{equation}

Einstein field equations for the comoving observers $u^{a}=\frac{1}{N\left(
t\right) }\partial _{t}~,~u^{a}u_{a}=-1$~follow from the autonomous
Lagrangian \cite{Rayn,KotsakisL}
\begin{equation}
L=-3\frac{a}{N}\dot{a}^{2}+\frac{3}{4}\frac{a^{3}}{N}\dot{\beta}+\frac{k}{2}%
\frac{a^{3}}{N}\dot{\phi}^{2}-Nka^{3}e^{-c\phi }-N\frac{\rho _{0}}{a^{3}}.
\label{LRS.03}
\end{equation}%
We set $N^{2}=e^{c\phi }$ and the Lagrangian becomes
\begin{equation}
L=(-3\dot{a}^{2}+\frac{3}{4}a^{2}\dot{\beta}^{2}+\frac{k}{2}a^{2}\dot{\phi}%
^{2})ae^{-\frac{c}{2}\phi }-ka^{3}V_{0}e^{-\frac{c}{2}\phi }-\frac{\rho _{0}%
}{a^{3}e^{\frac{c}{2}\phi }}.  \label{LRS.04}
\end{equation}%
The Hamiltonian is%
\begin{equation}
E=(-3\dot{a}^{2}+\frac{3}{4}a^{2}\dot{\beta}^{2}+\frac{k}{2}a^{2}\dot{\phi}%
^{2})ae^{-\frac{c}{2}\phi }+ka^{3}V_{0}e^{-\frac{c}{2}\phi }+\frac{\rho _{0}%
}{a^{3}e^{\frac{c}{2}\phi }}=0
\end{equation}%
If we consider the transformation%
\begin{equation}
a^{3}=e^{x+y}~,~\phi =\frac{1}{3}\sqrt{\frac{6}{k}}\left( x-y\right)
\label{LRS.07}
\end{equation}%
where
\begin{equation}
x=\frac{1}{1-\bar{c}}\ln \left( \frac{\left\vert 1-\bar{c}\right\vert }{%
\sqrt{2}}ue^{z}\right) ~,~y=\frac{1}{1+\bar{c}}\ln \left( \frac{1+\bar{c}}{%
\sqrt{2}}ue^{-z}\right) ~,~\bar{c}=\frac{c}{\sqrt{6k}}\neq 1  \label{LRS.08}
\end{equation}%
the Lagrangian (\ref{LRS.04}) becomes
\begin{equation}
L=-\frac{2}{3}\dot{u}^{2}+u^{2}\left( \frac{2}{3}\dot{z}^{2}-\frac{3}{8kV_{0}%
}\dot{\beta}^{2}\right) -\frac{\mu ^{2}}{2}u^{2}+\frac{kV_{0}\rho _{0}}{\mu
^{2}}\frac{1}{u^{2}}~\ ,~\mu ^{2}=kV_{0}\left( 1-\bar{c}^{2}\right) \neq 0.
\label{LRS.09}
\end{equation}

We consider a 2D Riemannian space with metric defined by the kinematic terms
of the Lagrangian, that is%
\begin{equation}
ds^{2}=\left( -6da^{2}+\frac{3}{2}a^{2}d\beta ^{2}+ka^{2}d\phi ^{2}\right)
ae^{-\frac{c}{2}\phi }  \label{LRS.05}
\end{equation}

We show easily that this metric admits the gradient HV $H^{i}=\frac{4}{%
6k-c^{2}}\left( ka\partial _{a}+c\partial _{\phi }\right) ~~$with gradient
function~$H=\frac{8k\varepsilon }{c^{2}-6k}ae^{-\frac{c}{2}\phi }.~\ $
Therefore the Lagrangian (\ref{LRS.04}) defines an autonomous Hamiltonian
Riemannian Kepler Ermakov system with potential ($\mu \neq 0$)
\begin{equation}
V(u,y^{A})=-\frac{1}{2}\mu ^{2}u^{2}+\frac{kV_{0}\rho _{0}}{\mu ^{2}}\frac{1%
}{u^{2}}.  \label{LRS.09a}
\end{equation}%
Because $\mu \neq 0$ this Lagrangian admits $sl(2,R)$ invariance only for
the representation (\ref{NCGE.04}) (an additional result which shows the
necessity for the consideration of the cases $\mu =0$ and $\mu \neq 0!).$

Using proposition \ref{prop ghv}, we write the Ermakov invariant
\begin{equation}
J=u^{4}\left( \frac{2}{3}\dot{z}^{2}+\frac{3}{8kV_{0}}\dot{\beta}^{2}\right)
+\frac{kV_{0}\rho _{0}}{\mu ^{2}}\frac{1}{u^{2}}.  \label{LRS.11}
\end{equation}%
The second invariant is the Hamiltonian
\begin{equation}
E=-\frac{2}{3}\dot{u}^{2}+u^{2}\left( \frac{2}{3}\dot{z}^{2}+\frac{3}{8kV_{0}%
}\dot{\beta}^{2}\right) +\frac{\mu ^{2}}{2}u^{2}-\frac{kV_{0}\rho _{0}}{\mu
^{2}}\frac{1}{u^{2}}.  \label{LRS.12}
\end{equation}

We find that the Lagrangian admits three more Noether symmetries%
\begin{equation}
\partial _{\beta },~\partial _{z}~,~z\partial _{\beta }-\beta \partial _{z}
\label{LRS.13}
\end{equation}%
with corresponding integrals%
\begin{equation}
I_{1}=u^{2}\dot{\beta}~,~I_{2}=u^{2}\dot{z}~,~I_{3}=u^{2}\left( \frac{3}{%
8kV_{0}}z\dot{\beta}-\frac{2}{3}\beta \dot{z}\right) .  \label{LRS.14}
\end{equation}%
It is easy to show that three of the integrals are in involution, therefore
the system is Liouville integrable.

\subsection{The case of $f\left( R\right) $ Cosmology}

\label{frgravity}

Consider the modified Einstein-Hilbert action
\begin{equation}
S=\int d^{4}x\sqrt{-g}f\left( R\right)  \label{MLRS.01}
\end{equation}%
where $f(R)$ is a smooth function of the curvature scalar $R.$ The resulting
field equations for this action in the metric variational approach are \cite%
{Sotiriou}
\begin{equation}
f^{\prime }R_{ab}-\frac{1}{2}fg_{ab}+g_{ab}\square f^{\prime
}-f_{;ab}^{\prime }=0  \label{MLRS.02}
\end{equation}%
where $f^{\prime }=\frac{df\left( R\right) }{dR}~$and $f^{\prime \prime
}\neq 0.~$In the LRS spacetime (\ref{LRS.01}), with $N\left( t\right) =1,$
these equations for comoving observers are the Euler-Lagrange equations of
the Lagrangian
\begin{equation}
L=\left( 6af^{\prime }\dot{a}^{2}+6a^{2}f^{\prime \prime }\dot{a}\dot{R}-%
\frac{3}{2}f^{\prime }a^{3}\dot{\beta}^{2}\right) +a^{3}\left( f^{\prime
}R-f\right) .  \label{MLRS.03}
\end{equation}%
The Hamiltonian is%
\begin{equation}
E=\left( 6af^{\prime }\dot{a}^{2}+6a^{2}f^{\prime \prime }\dot{a}\dot{R}-%
\frac{3}{2}f^{\prime }a^{3}\dot{\beta}^{2}\right) -a^{3}\left( f^{\prime
}R-f\right) =0.  \label{MLRS.05}
\end{equation}%
Again we consider the 3d Riemannian space whose metric is defined by the
kinematic part of the Lagrangian (\ref{MLRS.03})
\begin{equation}
ds^{2}=12af^{\prime }da^{2}+12a^{2}f^{\prime \prime }da~dR-3a^{3}f^{\prime
}d\beta ^{2}.  \label{MLRS.04}
\end{equation}%
This metric admits the gradient HV%
\begin{equation}
H^{i}=\frac{1}{2}\left( a\partial _{a}+\frac{f^{\prime }}{f^{\prime \prime }}%
\partial _{R}\right)  \label{MLRS.06}
\end{equation}%
with gradient function $H=3a^{3}f^{\prime }.$

In order to determine the function $f(R)$ we demand the geometric condition
that Lagrangian (\ref{MLRS.03}) admits $s(2,R)$ invariance via Noether
symmetries. Then for each representation (\ref{NCGE.04}), (\ref{NCGE.05}) we
have a different function $f(R)$ hence a different physical theory.

The representation (\ref{NCGE.04}) in the present context is:%
\begin{equation}
\partial _{t}~,~2t\partial _{t}+\frac{1}{2}\left( a\partial _{a}+\frac{%
f^{\prime }}{f^{\prime \prime }}\partial _{R}\right) ~,~t^{2}\partial _{t}+%
\frac{t}{2}\left( a\partial _{a}+\frac{f^{\prime }}{f^{\prime \prime }}%
\partial _{R}\right) .  \label{SL.2R1}
\end{equation}%
The Noether conditions become
\begin{equation}
-4a^{3}f^{\prime }R+\frac{7}{2}a^{3}f+p=0.
\end{equation}%
These vectors are Noether symmetries if $p=0$ and
\begin{equation}
f\left( R\right) =R^{\frac{7}{8}}.
\end{equation}%
However power law $f\left( R\right) $ theories are not cosmologically viable
\cite{Amendola2}.

The second representation (\ref{NCGE.05}) in the present context gives the
vectors%
\begin{equation}
\partial _{t}~,~\frac{1}{\mu }e^{\pm 2\mu t}\partial _{t}\pm \frac{1}{2}%
e^{\pm 2\mu t}\left( a\partial _{a}+\frac{f^{\prime }}{f^{\prime \prime }}%
\partial _{R}\right) .  \label{SL.2R2}
\end{equation}%
The Noether conditions give
\begin{equation}
-4a^{3}f^{\prime }R+\frac{7}{2}a^{3}f+3\mu ^{2}a^{3}f^{\prime }+p=0.
\label{SL.2R2.1}
\end{equation}%
These vectors are Noether symmetries if the constant $p=0$ and \ the
function
\begin{equation}
f\left( R\right) =\left( R-2\Lambda \right) ^{\frac{7}{8}}  \label{SL.2R2.2}
\end{equation}%
where $2\Lambda =3\mu ^{2}.$ This model is the viable $\Lambda _{bc}$%
CDM-like cosmological with $b=1,c=\frac{7}{8}$. \cite{Amendola}.

We note that if we had not considered the latter representation then we
would loose this interesting result. The importance of the result is due to
the fact that it follows from a geometric assumption which is beyond and
above the physical considerations. Furthermore the assumption of Noether
symmetries provides the Noether integrals which allow for an analytic
solution of the model.

For the function (\ref{SL.2R2.2}) the Lagrangian (\ref{MLRS.03}) becomes for
both cases (if $\Lambda =0$ we have the power-law $f\left( R\right) =R^{%
\frac{7}{8}}$)
\begin{equation}
L=\frac{21}{4}a\left( R-2\Lambda \right) ^{-\frac{1}{8}}\dot{a}^{2}-\frac{21%
}{16}a^{2}\left( R-2\Lambda \right) ^{-\frac{9}{8}}\dot{a}\dot{R}-\frac{21}{8%
}a^{3}\left( R-2\Lambda \right) ^{-\frac{1}{8}}\dot{\beta}^{2}-\frac{a^{3}}{8%
}\frac{\left( R-16\Lambda \right) }{\left( R-2\Lambda \right) ^{\frac{1}{8}}}%
.  \label{MLRS.10}
\end{equation}%
Furthermore there exist a coordinate transformation for which the metric (%
\ref{MLRS.04}) is written in the form of (\ref{GERS.02}).

We introduce new variables $u,v,w$ with the relations
\begin{equation}
a=\left( \frac{21}{4}\right) ^{-\frac{1}{3}}\sqrt{ue^{v}}~,~R=2\Lambda +%
\frac{e^{12v}}{u^{4}},~\beta =\sqrt{2}w.
\end{equation}%
In the new variables the Lagrangian (\ref{MLRS.10}) takes the form
\begin{equation}
L=\frac{1}{2}\dot{u}^{2}-\frac{1}{2}u^{2}\left( \dot{v}^{2}+\dot{w}%
^{2}\right) +\frac{\mu ^{2}}{2}u^{2}-\frac{1}{42}\frac{e^{12v}}{u^{2}}.
\label{MLRS.11}
\end{equation}%
The Hamiltonian (\ref{MLRS.05}) in the new coordinates is%
\begin{equation}
E=\frac{1}{2}\dot{u}^{2}-\frac{1}{2}u^{2}\left( \dot{v}^{2}+\dot{w}%
^{2}\right) -\frac{\mu ^{2}}{2}u^{2}+\frac{1}{42}\frac{e^{12v}}{u^{2}}.
\end{equation}%
The Lagrangian (\ref{MLRS.11}) defines a Hamiltonian Riemannian Kepler
Ermakov system with potential%
\begin{equation*}
V\left( u,v\right) =-\frac{\mu ^{2}}{2}u^{2}+\frac{1}{42}\frac{e^{12v}}{u^{2}%
}
\end{equation*}%
from which follows the potential $V(v)=\frac{1}{42}e^{12v}.~\ $ In addition
to the Hamiltonian the dynamical system admits the Riemannian Ermakov
invariant%
\begin{equation}
J_{f}=u^{4}\left( \dot{v}^{2}+\dot{w}^{2}\right) +\frac{1}{21}e^{12v}.
\end{equation}

The Lagrangian (\ref{MLRS.11}) admits the extra Noether point symmetry $%
\partial _{w}$ with Noether integral $I_{w}=u^{2}\dot{w}$ ~(see Table \ref%
{PtFRWEE}). The three integrals $E,I_{w}~$and $J_{f}$ are in involution and
independent, therefore the system is integrable.

\section{Conclusion}

\label{Conclusion}

In this Chapter we have considered the generalization of the autonomous
Kepler Ermakov dynamical system in the spirit of Leach \cite{Leach1991},
that is using invariance with respect to the $sl(2,R)$ Lie and Noether
algebra. We have generalized the autonomous Newtonian Hamiltonian Kepler
Ermakov system to three dimensions using Noether rather than Lie point
symmetries and have determined all such systems which are Liouville
integrable via Noether point symmetries. We introduced the autonomous
Riemannian Kepler Ermakov system in a Riemannian space which admits a
gradient HV. This system is the generalization of the autonomous Euclidian
Kepler Ermakov system and opens new fields of applications for the
autonomous Kepler Ermakov system, especially in relativistic Physics. Indeed
we have determined the autonomous Riemannian Kepler Ermakov system in a
spatially flat FRW\ spacetime which admits a gradient HV. As a further
application we have considered two types of cosmological models, which are
described by the autonomous Riemannian Kepler Ermakov system, the scalar
field cosmology with exponential potential and $f(R)$ gravity in an LRS
spacetime.

\newpage%

\begin{subappendices}%

\section{Appendix}

\label{apen1force}

We require that the force admits two Lie symmetries which are due to the
gradient HV $H=u\partial _{u}~$(if we require the force to be invariant
under three Lie symmetries which are due to the gradient HV then it is
reduced to the isotropic oscillator). From Theorem \ref{The general
conservative system}\ of Chapter \ref{chapter3} we have the following cases.

(I) Case $\mu =0$\newline
In this case the Lie symmetries are
\begin{equation*}
\partial _{s},~2s\partial _{s}+u\partial _{u},~s^{2}\partial _{s}+su\partial
_{u}.
\end{equation*}

The condition which the force must satisfy is%
\begin{equation*}
L_{H}F^{i}+dF^{i}=0.
\end{equation*}%
Replacing components we find the equations%
\begin{eqnarray*}
\left( \frac{\partial }{\partial u}F^{u}\right) u+\left( d-1\right) F^{u}
&=&0 \\
\left( \frac{\partial }{\partial u}F^{A}\right) u+dF^{A} &=&0
\end{eqnarray*}%
from which follows
\begin{equation*}
F^{u}=\frac{1}{u^{\left( d-1\right) }}F^{u}~,~F^{A}=\frac{1}{u^{d}}F^{A}.
\end{equation*}

Because the HV\ is gradient, Case II of Theorem \ref{The general
conservative system} applies and gives the condition
\begin{equation*}
L_{H}F^{i}+4F^{i}+a_{1}H^{i}=0
\end{equation*}%
from which follows $a_{1}=0$ and \ $d=4$. Therefore
\begin{equation*}
F^{u}=\frac{1}{u^{3}}G^{u}\left( y^{C}\right) ~,~F^{A}=\frac{1}{u^{4}}%
G^{A}\left( y^{C}\right) .
\end{equation*}

(II) Case $\mu \neq 0$\newline
In this case the Lie symmetries are
\begin{equation*}
\partial _{s},~\frac{1}{\mu }e^{\pm 2\mu s}\partial _{s}\pm e^{\pm 2\mu
s}u\partial _{u}
\end{equation*}%
The condition which the force must satisfy is%
\begin{equation*}
L_{H}F^{i}+4F^{i}+a_{1}H^{i}=0
\end{equation*}%
We demand $a_{1}\neq 0$ and obtain the system of equations:%
\begin{eqnarray*}
\left( \frac{\partial }{\partial u}F^{u}\right) u+3F^{u}+a_{1}u &=&0 \\
\left( \frac{\partial }{\partial u}F^{A}\right) u+4F^{A} &=&0
\end{eqnarray*}%
whose solution is
\begin{equation*}
F^{u}=\mu ^{2}u+\frac{1}{u^{3}}G^{u}~,~F^{A}=\frac{1}{u^{4}}G^{A}
\end{equation*}%
where we have set $a_{1}=-4\mu ^{2}$.

\end{subappendices}%

\part{Symmetries of PDEs}

\chapter{Lie symmetries of a general class of PDEs \label{chapter5}}

\section{Introduction}

In the previous chapters we studied the relation between point symmetries
(Lie and Noether symmetries) of second order ordinary differential
equations. Particularly, we considered the case of geodesic equations and
the equations of motion of a particle moving in a Riemannian space. We made
clear that there exists a unique relation between the point symmetries and
the special projective Lie algebra of the underlying space in which the
motion occurs.

In subsequent sections, we will attempt to extend the relation between point
symmetries and collineations of the space to the case of second order
partial differential equations. Obviously, a global answer to this problem
is not possible. However, it will be shown that for many interesting PDEs,
the Lie point symmetries are indeed obtained from the collineations of the
metric. Pioneering work in this direction is the work of Ibragimov \cite%
{IbragB}. Recently, Bozhkov et al. \cite{Bozhkov} have studied the Lie and
the Noether point symmetries of the Poisson equation and showed that the Lie
symmetries of the Poisson PDE are generated from the conformal algebra of
the metric.

In this chapter we show that for a general class of second order PDEs, there
is a close relation between the Lie symmetries and the conformal algebra of
the underlying space. Subsequently, we apply these results to a number of
interesting PDEs and regain existing results in a unified manner.

In Section \ref{The case of the second order PDE} we examine the generic PDE
of the form
\begin{equation}
A^{ij}u_{ij}-F(x^{i},u,u_{i})=0  \label{GPE.0}
\end{equation}%
and derive the Lie symmetry conditions. Furthermore, in case $A$ is
independent on $u,$ i.e. $A_{,u}=0$, then the Lie point symmetries of (\ref%
{GPE.0}) are related to the Conformal vectors\ (CVs) of the linear
homogeneous differential geometric object $A.$ In Section \ref{The Lie
symmetry conditions for a linear function} we consider $F(x,u,u_{i})$ to be
linear in $u_{i}$ and determine the Lie point symmetry conditions in
geometric form. In sections \ref{PoissonSym} and \ref{HeatCon}, we apply the
results of section \ref{The Lie symmetry conditions for a linear function}
in order to determine the Lie point symmetries of the\ Poisson equation, the
Yamabe equation and the heat equation with flux in a $n~$dimensional
Riemannian space. It will be shown that the Lie symmetry vectors of the
Poisson and the Yamabe equation are obtained from the conformal algebra of
the geometric object $A^{ij}$ \cite{Bozhkov} whereas the Lie symmetries of
the heat equations are obtained from the homothetic algebra of the metric.
Furthermore, we determine the Lie symmetries of Laplace equation, the Yamabe
equation and the homogeneous heat equation in various Riemannian spaces.

\section{The case of the second~order PDEs}

\label{The case of the second order PDE}

Attempting to establish a general relation between the Lie symmetries of a
second order PDE of the form (\ref{GPE.0}) and the collineations of a
Riemannian space we derive the Lie symmetry conditions of (\ref{GPE.0}) and
relate them with the collineations of the coefficients $A^{ij}(x,u)$ which
we consider to be the components of a metric. According to the standard
approach \cite{StephaniB,BlumanB,IbragB,OlverB} the symmetry condition is
\begin{equation}
X^{[2]}(H)=\lambda H~~,~modH=0  \label{GPE.10}
\end{equation}%
where $\lambda (x^{i},u,u_{i})$ is a function to be determined. $X^{[2]}$ is
the second prolongation of the Lie symmetry vector
\begin{equation}
X=\xi ^{i}\left( x^{i},u\right) \frac{\partial }{\partial x^{i}}+\eta \left(
x^{i},u\right) \frac{\partial }{\partial u}  \label{GPE.10.1}
\end{equation}%
given by the expression
\begin{equation}
X^{\left[ 2\right] }=\xi ^{i}\frac{\partial }{\partial x^{i}}+\eta \frac{%
\partial }{\partial u}+\eta _{i}^{\left[ 1\right] }\frac{\partial }{\partial
u_{i}}+\eta _{ij}^{\left[ 1\right] }\frac{\partial }{\partial u_{ij}}
\label{GPE.10a}
\end{equation}%
where\footnote{%
See section \ref{inF}.}%
\begin{eqnarray*}
\eta _{i}^{\left[ 1\right] } &=&\eta _{,i}+u_{i}\eta _{u}-\xi
_{,i}^{j}u_{j}-u_{i}u_{j}\xi _{,u}^{j} \\
\eta _{ij}^{\left[ 2\right] } &=&\eta _{ij}+(\eta _{ui}u_{j}+\eta
_{uj}u_{i})-\xi _{,ij}^{k}u_{k}+\eta _{uu}u_{i}u_{j}-(\xi
_{.,ui}^{k}u_{j}+\xi _{.,uj}^{k}u_{i})u_{k} \\
&&+\eta _{u}u_{ij}-(\xi _{.,i}^{k}u_{jk}+\xi _{.,j}^{a}u_{ik})-\left(
u_{ij}u_{k}+u_{i}u_{jk}+u_{ik}u_{j}\right) \xi _{.,u}^{k}-u_{i}u_{j}u_{k}\xi
_{uu}^{k}.
\end{eqnarray*}%
The introduction of the function $\lambda (x^{i},u,u_{i})$ in (\ref{GPE.10})
causes the variables $x^{i},u,u_{i}$ to be independent\footnote{%
See Ibragimov \cite{IbragB} p. 115}.

The symmetry condition (\ref{GPE.10}) when applied to (\ref{GPE.0}) gives:%
\begin{equation}
A^{ij}\eta _{ij}^{\left[ 2\right] }+\left( XA^{ij}\right)
u_{ij}-X^{[1]}(F)=\lambda (A^{ij}u_{ij}-F)  \label{GPE.13}
\end{equation}%
from which follows%
\begin{align}
0& =A^{ij}\eta _{ij}-\eta _{,i}g^{ij}F_{,u_{j}}-X(F)+\lambda F  \notag \\
& +2A^{ij}\eta _{ui}u_{j}-A^{ij}\xi _{,ij}^{a}u_{a}-u_{i}\eta
_{u}g^{ij}F_{,u_{j}}+\xi _{,i}^{k}u_{k}g^{ij}F_{,u_{j}}  \notag \\
& +A^{ij}\eta _{uu}u_{i}u_{j}-2A^{ij}\xi _{.,uj}^{k}u_{i}u_{k}+u_{i}u_{k}\xi
_{,u}^{k}g^{ij}F_{,u_{j}}  \notag \\
& +A^{ij}\eta _{u}u_{ij}-2A^{ij}\xi _{.,i}^{k}u_{jk}+(\xi
^{k}A_{,k}^{ij}+\eta A_{,u}^{ij})u_{ij}-\lambda A^{ij}u_{ij}  \notag \\
& -A^{ij}\left( u_{ij}u_{a}+u_{i}u_{ja}+u_{ia}u_{j}\right) \xi
_{.,u}^{a}-u_{i}u_{j}u_{a}A^{ij}\xi _{uu}^{a}.  \label{Po.0}
\end{align}

We note that we cannot deduce the symmetry conditions before we select a
specific form for the function $F(x^{i},u,u_{i}).$ However, we may determine
the conditions which are due to the second derivative of $u$ because in
these terms no $F$ terms are involved. This observation significantly
reduces the complexity of the remaining symmetry conditions. Following this,
we have the condition%
\begin{align*}
0& =A^{ij}\eta _{u}u_{ij}-A^{ij}(\xi _{.,i}^{k}u_{ja}+\xi
_{.,j}^{k}u_{ik})+(\xi ^{k}A_{,k}^{ij}+\eta A_{,u}^{ij})u_{ij}-\lambda
A^{ij}u_{ij} \\
& -A^{ij}\left( u_{ij}u_{a}+u_{i}u_{ja}+u_{ia}u_{j}\right) \xi
_{.,u}^{a}-u_{i}u_{j}u_{a}A^{ij}\xi _{uu}^{a}
\end{align*}%
from which the following system of equations results%
\begin{align*}
A^{ij}\left( u_{ij}u_{k}+u_{jk}u_{i}+u_{ik}u_{j}\right) \xi _{.,u}^{k}& =0 \\
A^{ij}\eta _{u}u_{ij}-A^{ij}(\xi _{.,i}^{k}u_{jk}+\xi _{.,j}^{k}u_{ik})+(\xi
^{k}A_{,k}^{ij}+\eta A_{,u}^{ij})u_{ij}-\lambda A^{ij}u_{ij}& =0 \\
A^{ij}\xi _{,uu}^{a}& =0.
\end{align*}%
The first equation is%
\begin{equation}
A^{ij}\xi _{.,u}^{k}+A^{kj}\xi _{.,u}^{i}+A^{ik}\xi
_{.,u}^{j}=0\Leftrightarrow A^{(ij}\xi _{.,u}^{k)}=0.  \label{Po.1}
\end{equation}%
From the second equation we get
\begin{equation}
A^{ij}\eta _{u}+\eta A_{,u}^{ij}+\xi ^{k}A_{,k}^{ij}-A^{kj}\xi
_{.,k}^{i}-A^{ik}\xi _{.,k}^{j}-\lambda A^{ij}=0.  \label{Po.2}
\end{equation}%
and the last equation gives the constraint%
\begin{equation}
A^{ij}\xi _{,uu}^{k}=0.  \label{Po.2a}
\end{equation}

It can be easily shown that condition (\ref{Po.1}) implies $\xi
_{.,u}^{k}=0,~$which is a well known result\footnote{%
We give a simple proof for $n=2$ in Appendix \ref{appenA5}. A detailed and
more general proof can be found in \cite{BlumanPaper}.}. From the analysis
so far, we obtain the first result

\begin{proposition}
\label{Coefficient xi second} For the infinitesimal generator (\ref{GPE.10.1}%
) for all second~order PDEs\ of the form (\ref{GPE.0}), holds $\xi
_{.,u}^{i}=0,$ that is, $\xi ^{i}=\xi ^{i}(x^{j}).$ Furthermore, condition (%
\ref{Po.2a}) is identically satisfied.
\end{proposition}

There remains the third symmetry condition (\ref{Po.2}). We consider the
following cases.\newline
$i,j\neq 0:$\newline
We write (\ref{Po.2}) in an alternative form by considering $A^{ij}$ to be
linear a homogeneous differential geometric object as follows:%
\begin{equation}
L_{\xi }A^{ij}=\lambda A^{ij}-(\eta A^{ij})_{,u}.  \label{GPE.32}
\end{equation}%
Then it follows:

\begin{proposition}
\label{Coefficient xi third} For the infinitesimal generator (\ref{GPE.10.1}%
) for all second~order PDEs \ of the form (\ref{GPE.0}) for which $%
A^{ij}{}_{,u}=0,$ i.e. $A^{ij}=A^{ij}(x^{i}),$ the vector $\xi ^{i}$ is a
CKV of the linear homogeneous differential geometric object $A^{ij}$ with
conformal factor ($\lambda -\eta _{u})(x).$
\end{proposition}

Assuming\footnote{%
The index $t$ refers to the coordinate $x^{0}$ whenever it is involved.} $%
A^{tt}=A^{ti}=0,$ we have\newline
- for $i=j=0$ nothing\newline
- for $i,j\neq 0$ gives (\ref{GPE.32}) and\newline
- for $i=0,j\neq 0$ (\ref{GPE.32}) becomes
\begin{align}
A^{tj}\eta _{u}+\eta A_{,u}^{tj}+\xi ^{k}A_{,k}^{tj}-A^{kj}\xi
_{.,k}^{t}-A^{tk}\xi _{.,k}^{j}-\lambda A^{tj}& =0\Rightarrow  \notag \\
A^{kj}\xi _{.,k}^{t}& =0  \label{GPE.30b}
\end{align}%
which leads to the following general result.

\begin{proposition}
\label{Coefficient xi fourth} For all second~order PDEs \ of the form $%
A^{ij}u_{ij}-F(x^{i},u,u_{i})=0,$ for which $A^{ij}$ is nondegenerate the $%
\xi _{.,k}^{t}=0$, that is, $\xi ^{t}=\xi ^{t}(t)$ .
\end{proposition}

By using that $\xi _{,u}^{i}=0,$ the symmetry condition (\ref{Po.0}) is
simplified as follows
\begin{align}
0& =A^{ij}\eta _{ij}-\eta _{,i}A^{ij}F_{,u_{j}}-X(F)+\lambda F  \notag \\
& +2A^{ij}\eta _{ui}u_{j}-A^{ij}\xi _{,ij}^{a}u_{a}-u_{i}\eta
_{u}A^{ij}F_{,u_{j}}+\xi _{,i}^{k}u_{k}A^{ij}F_{,u_{j}}  \notag \\
& +A^{ij}\eta _{uu}u_{i}u_{j}+A^{ij}\eta _{u}u_{ij}-2A^{ij}\xi
_{.,i}^{k}u_{jk}  \label{GPE.30} \\
& +(\xi ^{k}A_{,k}^{ij}+\eta A_{,u}^{ij})u_{ij}-\lambda A^{ij}u_{ij}  \notag
\end{align}%
which together with the condition (\ref{GPE.32}) are the complete set of
symmetry conditions\emph{\ }for all\emph{\ }second order PDEs of the form $%
A^{ij}u_{ij}-F(x^{i},u,u_{i})=0$. This class of PDEs is quite general. This
fact makes the above result very useful..

In order to continue, we need to consider special forms for the function $%
F(x,u,u_{i}).$

\section{The Lie symmetry conditions for a linear function $F(x,u,u_{i})$}

\label{The Lie symmetry conditions for a linear function}

Consider the function $F(x,u,u_{i})$ to be linear in $u_{i}$, that is, to be
of the form%
\begin{equation}
F(x,u,u_{i})=B^{k}(x,u)u_{k}+f(x,u)  \label{GPE.30a}
\end{equation}%
where $B^{k}(x,u)~$and $~f(x,u),$ are arbitrary functions of their
arguments. In this case, the PDE (\ref{GPE.0})\ is of the form%
\begin{equation}
A^{ij}u_{ij}-B^{k}(x,u)u_{k}-f(x,u)=0.  \label{GPE.30.1}
\end{equation}

The Lie symmetries of this type of PDEs have been studied previously by
Ibragimov\cite{IbragB}. Assuming that at least one of the components of $%
\mathbf{A}$ is $A_{ij}\neq 0$ the Lie symmetry conditions are (\ref{GPE.30})
and (\ref{GPE.32}).

Replacing $F(x,u,u_{1})$ in (\ref{GPE.30}) we find\footnote{%
We ignore the terms with $u_{ij}$ because we have already used them to
obtain condition (\ref{GPE.32}). Indeed, it is clear that these terms give $%
A^{ij}\eta _{u}-2A^{ij}\xi _{.,i}^{k}+\xi ^{k}A_{,k}^{ij}+\eta
A_{,u}^{ij}-\lambda A^{ij}=0,$ which is precisely condition (\ref{GPE.32}).}%
\begin{align}
0& =A^{ij}\eta _{ij}-\eta _{,i}g^{ij}B_{j}-\xi ^{k}f_{,k}-\eta
f_{,u}+\lambda f  \notag \\
& +2A^{ij}\eta _{ui}u_{j}-A^{ij}\xi _{,ij}^{a}u_{a}-u_{i}\eta
_{u}g^{ij}B_{j}+\xi _{,i}^{k}u_{k}g^{ij}B_{j}+\lambda B^{k}u_{k}-\eta
B_{,u}^{k}u_{k}-\xi ^{l}B_{,l}^{k}u_{k}  \notag \\
& +A^{ik}\eta _{uu}u_{i}u_{k}  \label{GPE.33} \\
& +A^{ij}\eta _{u}u_{ij}-2A^{kj}\xi _{.,k}^{i}u_{ji}+(\xi
^{k}A_{,k}^{ij}+\eta A_{,u}^{ij})u_{ij}-\lambda A^{ij}u_{ij}
\end{align}%
from which the subsequent equations follow%
\begin{align}
A^{ij}\eta _{ij}-\eta _{,i}B^{i}-\xi ^{k}f_{,k}-\eta f_{,u}+\lambda f& =0
\label{GPE.34} \\
-2A^{ik}\eta _{ui}+A^{ij}\xi _{,ij}^{k}+\eta _{u}B^{k}-\xi
_{,i}^{k}B^{i}+\xi ^{i}B_{,i}^{k}-\lambda B^{k}+\eta B_{,u}^{k}& =0
\label{GPE.35} \\
A^{ik}\eta _{uu}& =0.  \label{GPE.36}
\end{align}%
Equation (\ref{GPE.36}) gives (because at least one $A^{ik}\neq 0$)
\begin{equation}
\eta =a(x^{i})u+b(x^{i}).  \label{GPE.37}
\end{equation}%
Equation (\ref{GPE.35}) gives the constraint
\begin{equation*}
-2A^{ik}a_{,i}+aB^{k}+auB_{,u}^{k}+A^{ij}\xi _{,ij}^{k}-\xi
_{,i}^{k}B^{i}+\xi ^{i}B_{,i}^{k}-\lambda B^{k}+bB_{,u}^{k}=0.
\end{equation*}

We summarize the above results as follows.

\begin{proposition}
\label{propPDE.1} The determining equations (the Lie symmetry conditions)
for the second order PDEs (\ref{GPE.30.1}), in which at least one of the $%
A_{ij}\neq 0,$ are%
\begin{equation}
A^{ij}(a_{ij}u+b_{ij})-(a_{,i}u+b_{,i})B^{i}-\xi
^{k}f_{,k}-auf_{,u}-bf_{,u}+\lambda f=0  \label{GPE.42}
\end{equation}%
\begin{equation}
A^{ij}\xi _{,ij}^{k}-2A^{ik}a_{,i}+aB^{k}+auB_{,u}^{k}-\xi
_{,i}^{k}B^{i}+\xi ^{i}B_{,i}^{k}-\lambda B^{k}+bB_{,u}^{k}=0  \label{GPE.43}
\end{equation}%
\begin{equation}
L_{\xi ^{i}\partial _{i}}A^{ij}=(\lambda -a)A^{ij}-\eta A^{ij}{}_{,u}
\label{GPE.44}
\end{equation}%
\begin{align}
\eta & =a(x^{i})u+b(x^{i})  \label{GPE.45} \\
\xi _{,u}^{k}& =0\Leftrightarrow \xi ^{k}(x^{i}).  \label{GPE.46}
\end{align}
\end{proposition}

Pay paricular attention to the fact that for all second order PDEs of the
form (\ref{GPE.30.1}), for which, \ $A_{,u}^{ij}=0$, i.e $A^{ij}=$ $%
A^{ij}(x^{i}),$ the $\xi ^{i}(x^{j})$ is a CKV of the metric $A^{ij}.$ Also,
in this case $\lambda =\lambda (x^{i}).$ This result establishes the
relation between the Lie symmetries of this type of PDEs with the
collineations of the metric defined by the coefficients $A_{ij}.$

Moreover, in case the coordinates are $t,x^{i}$ (where $i=1,...,n$) $%
A^{tt}=A^{tx^{i}}=0$ and $A^{ij}$ is a nondegenerate metric we have that%
\begin{equation}
\xi _{,i}^{t}=0\Leftrightarrow \xi ^{t}(t).  \label{GPE.46a}
\end{equation}

These symmetry relations coincide with those given in \cite{IbragB}.
Finally, note that equation (\ref{GPE.43}) can be written as%
\begin{equation}
A^{ij}\xi _{,ij}^{k}-2A^{ik}a_{,i}+[\xi ,B]^{k}+(a-\lambda
)B^{k}+(au+b)B_{,u}^{k}=0.  \label{GPE.47}
\end{equation}%
Having derived the Lie symmetry conditions for the type of PDEs (\ref%
{GPE.30.1}) we continue with the computation of the Lie symmetries of some
important PDEs of this form. Before we proceed, we state two Lemmas which
will be used later (for details, see Appendix \ref{appendixP}).

\begin{lemma}
\label{LemmaPDE.1}For the Lie derivative of the connection coefficients, the
following properties hold.

a. In flat space (in which $\Gamma _{jk}^{i}=0)$ the following identity
holds:%
\begin{equation}
L_{\xi }\Gamma _{ij}^{k}=\xi _{,ij}^{k}.
\end{equation}%
b. For a general metric $g_{ij}$ satisfying the condition $L_{\xi
^{i}\partial _{i}}g_{ij}=-(\lambda -a)g_{ij}$ the following relation holds:
\begin{equation}
g^{jk}L_{\xi }\Gamma _{.jk}^{i}=g^{jk}\xi _{,~jk}^{i}+\Gamma _{~,l}^{i}\xi
^{l}-\xi _{~,l}^{i}\Gamma ^{l}+(a-\lambda )\Gamma _{~}^{i}.
\end{equation}
\end{lemma}

\begin{lemma}
\label{LemmaPDE.2}Assume that the vector $\xi ^{i}$ is a CKV of the metric $%
g_{ij}$ with conformal factor $-(\lambda -a)$ i.e. $L_{\xi ^{i}\partial
_{i}}g_{ij}=-(\lambda -a)g_{ij}.$ Then, the following statement is true:%
\begin{equation}
g^{jk}L_{\xi }\Gamma _{.jk}^{i}=\frac{2-n}{2}(a-\lambda )^{,i}
\end{equation}%
where $n=g^{jk}g_{kj}$ is the dimension of the space.
\end{lemma}

In the following sections, we study the Lie symmetries of two second order
PDEs which are important in physics. Particularly, we examine the relation
between the Lie symmetries of the Poisson equation and the Heat conduction
equation in a Riemannian manifold with the conformal group of the space. We
will show that the Lie symmetries of the heat conduction equation relates to
the homothetic group of the underlying space whereas the Lie symmetries of
the Poisson equation are associated to the conformal group of the metric
that defines the Laplace operator.

\section{Symmetries of the Poisson equation in a Riemannian space}

\label{PoissonSym}The Lie symmetries of the Poisson equation%
\begin{equation}
\Delta u-f\left( x^{i},u\right) =0.  \label{LE.01a}
\end{equation}%
where $\Delta =\frac{1}{\sqrt{\left\vert g\right\vert }}\frac{\partial }{%
\partial x^{i}}\left( \sqrt{\left\vert g\right\vert }g^{ij}\frac{\partial }{%
\partial x^{j}}\right) ~$\ is the Laplace operator of the metric $g_{ij},$
for $f=f\left( u\right) $ have been given in \cite{IbragB,Bozhkov}. Here we
generalize this result~\footnote{%
The proof is given in Appendix \ref{appendixP}.} for $f=f\left(
x^{i},u\right) $.

\begin{theorem}
\label{Theor}The Lie symmetries of the Poisson equation (\ref{LE.01a})$\ $%
are generated from the CKVs of the metric~$g_{ij}$ defining the Laplace
operator, as follows

a) for \thinspace \thinspace \thinspace $n>2,$ the Lie symmetry vector is%
\begin{equation}
X=\xi ^{i}\left( x^{k}\right) \partial _{i}+\left( \frac{2-n}{2}\psi \left(
x^{k}\right) u+a_{0}u+b\left( x^{k}\right) \right) \partial _{u}
\end{equation}%
where $\xi ^{i}\left( x^{k}\right) $ is a CKV with conformal factor $\psi
\left( x^{k}\right) $ and the following condition holds%
\begin{equation}
\frac{2-n}{2}\Delta \psi u+g^{ij}b_{i;j}-\xi ^{k}f_{,k}-\frac{2-n}{2}\psi
uf_{,u}-\frac{2+n}{2}\psi f-bf_{,u}=0.  \label{KG.Eq0}
\end{equation}

b) for $n=2,$ the Lie symmetry vector is
\begin{equation}
X=\xi ^{i}\left( x^{k}\right) \partial _{i}+\left( a_{0}u+b\left(
x^{k}\right) \right) \partial _{u}
\end{equation}%
where $\xi ^{i}\left( x^{k}\right) $ is a CKV with conformal factor $\psi
\left( x^{k}\right) $ and the following condition holds%
\begin{equation}
g^{ij}b_{;ij}-\xi ^{k}f_{,k}-a_{0}uf_{,u}+\left( a_{0}-2\psi \right)
f-bf_{,u}=0.
\end{equation}
\end{theorem}

In the following subsections, we apply Theorem \ref{Theor} for special forms
of the function $f\left( x^{i},u\right) $.

\subsection{Lie symmetries of Laplace equation}

\label{Lie point symmetries of the Klein Gordon equation}The Laplace
equation
\begin{equation}
\Delta u=0  \label{KG.Eq1}
\end{equation}%
follows from the Poison equation (\ref{LE.01a}) if we consider $f\left(
x^{i},u\right) =0.$ Therefore, Theorem \ref{Theor} applies and we have the
following result \cite{Bozhkov}.

\begin{theorem}
\label{KG1}The Lie symmetries of Laplace equation (\ref{KG.Eq1})$\ $are
generated from the CKVs of the metric~$g_{ij}$ defining the Laplace operator
as follows

a) for $n>2,$ the Lie symmetry vector is%
\begin{equation}
X=\xi ^{i}\left( x^{k}\right) \partial _{i}+\left( \frac{2-n}{2}\psi \left(
x^{k}\right) u+a_{0}u+b\left( x^{k}\right) \right) \partial _{u}
\end{equation}%
where $\xi ^{i}$ is a CKV with conformal factor $\psi \left( x^{k}\right) $,$%
~b\left( x^{k}\right) $ is a solution of (\ref{KG.Eq1})~and the following
condition is satisfied%
\begin{equation}
\Delta \psi =0.  \label{KG.Eq2}
\end{equation}%
b) for $n=2,$ the Lie symmetry vector is
\begin{equation}
X=\xi ^{i}\left( x^{k}\right) \partial _{i}+\left( a_{0}u+b\left(
x^{k}\right) \right) \partial _{u}
\end{equation}%
where $\xi ^{i}$ is a CKV with conformal factor $\psi \left( x^{k}\right) $%
~and$~b\left( x^{k}\right) $ is a solution of (\ref{KG.Eq1}).
\end{theorem}

\subsection{Symmetries of conformal Poisson equation in a Riemannian space}

If in the Poisson equation (\ref{LE.01a}) we replace $f\left( x^{i},u\right)
$ with
\begin{equation}
f=-M_{0}Ru+\bar{f}\left( x^{i},u\right)  \label{Ct.00.1}
\end{equation}%
where $R$ is the Ricci scalar of the metric which defines the Laplace
operator $\Delta $ and
\begin{equation}
M_{0}=\frac{2-n}{4\left( n-1\right) }  \label{Ct.00.2}
\end{equation}%
then, equation (\ref{LE.01a}) becomes%
\begin{equation}
\bar{L}_{g}u-\bar{f}\left( x^{i},u\right) =0.  \label{Ct.00}
\end{equation}%
where~$\bar{L}_{g}$ is the conformal Laplace or Yamabe operator acting on
functions on $V^{n}$ defined by%
\begin{equation}
\bar{L}_{g}=\Delta +\frac{n-2}{4\left( n-1\right) }R.  \label{Ct.00.3}
\end{equation}

Equation (\ref{Ct.00})\ is called \ the conformal Poisson or Yamabe equation
and plays a central role in the study of a conformal class of metrics by
means of the Yamabe invariant (see, e.g. \cite{LieParker}). In order to
investigate the Lie symmetries of (\ref{Ct.00}), we make use of Theorem \ref%
{Theor} and find the following result \footnote{%
The proof is given in Appendix \ref{appendixP}.}.

\begin{theorem}
\label{CTheor}The Lie symmetries of the conformal Poisson equation (\ref%
{Ct.00}) are generated from the CKVs of the metric~$g_{ij}$ defining the
conformal Laplace operator, as follows%
\begin{equation}
X=\xi ^{i}\left( x^{k}\right) \partial _{i}+\left( \frac{2-n}{2}\psi \left(
x^{k}\right) u+a_{0}u+b\left( x^{k}\right) \right) \partial _{u}
\end{equation}%
where $\xi ^{i}\left( x^{k}\right) $ is a CKV with conformal factor $\psi
\left( x^{k}\right) $ and the following condition holds%
\begin{equation}
-\xi ^{k}\bar{f}_{,k}-\frac{2-n}{2}\psi u\bar{f}_{k}-\frac{2+n}{2}\psi \bar{f%
}+g^{ij}b_{i;j}-b\left( -M_{0}R+\bar{f}_{u}\right) =0.
\end{equation}
\end{theorem}

\subsection{Lie symmetries of the conformal Laplace equation}

The conformal Laplace equation
\begin{equation}
\bar{L}_{g}u=0.  \label{CKG.01A}
\end{equation}%
is the conformal Poisson equation (\ref{Ct.00}) for $\bar{f}\left(
x^{i},u\right) =0.$ Therefore, Theorem \ref{CTheor} applies and we have the
following result.

\begin{theorem}
\label{CKG.011}The Lie symmetries of the conformal Laplace equation (\ref%
{CKG.01A}) are generated from the CKVs of the metric~$g_{ij}$ of the
conformal Laplace$~$operator as follows%
\begin{equation}
X=\xi ^{i}\left( x^{k}\right) \partial _{i}+\left( \frac{2-n}{2}\psi \left(
x^{k}\right) u+a_{0}u+b\left( x^{k}\right) \right) \partial _{u}
\end{equation}%
where $\xi ^{i}$ is a CKV with conformal factor $\psi \left( x^{k}\right) ~$%
and $b\left( x^{k}\right) $ is a solution of (\ref{CKG.01A}).
\end{theorem}

In order to compare the Lie symmetries of Laplace equation (\ref{KG.Eq1})
and of the conformal Laplace equation (\ref{CKG.01A}), we apply the results
of the Theorems \ref{KG1} and \ref{CKG.011} in the case of the FRW spacetime
with the following line element%
\begin{equation}
ds^{2}=R^{2}\left( \tau \right) \left( -d\tau
^{2}+dx^{2}+dy^{2}+dz^{2}\right) .  \label{FRWLap}
\end{equation}%
The Laplace operator for the space with Line element (\ref{FRWLap}) is%
\begin{equation*}
\Delta =\frac{1}{R^{2}}\eta ^{ij}\partial _{i}\partial _{j}-2\frac{R_{,\tau }%
}{R^{3}}\delta _{\tau }^{i}\partial _{i}
\end{equation*}%
where $\eta _{ij}$ is the metric of the Minkowski spacetime.

According to Theorem \ref{KG1}, the Laplace equation
\begin{equation}
\frac{1}{R^{2}}\eta ^{ij}u_{ij}-2\frac{R_{,\tau }}{R^{3}}\delta _{\tau
}^{i}u_{i}=0  \label{FRWLap1}
\end{equation}%
admits eight Lie point symmetries; six Lie symmetries are the KVs\footnote{%
For the conformal algebra of the FRW spacetime see Chapter \ref{LieSymGECh},
\cite{MM86}.} of (\ref{FRWLap}) plus the two Lie symmetries $\left(
a_{0}u+b\left( x^{k}\right) \right) \partial _{u}$ because, for general $%
R\left( \tau \right) ,$ the conformal factors of the CKVs of (\ref{FRWLap})
do not satisfy (\ref{FRWLap1}).

On the contrary, according to Theorem (\ref{CKG.011}), the conformal Laplace
equation%
\begin{equation}
\frac{1}{R^{2}}\eta ^{ij}u_{ij}-2\frac{R_{,\tau }}{R^{3}}\delta _{\tau
}^{i}u_{i}-\frac{R_{,\tau \tau }}{R^{3}}u=0  \label{FRWLap2}
\end{equation}%
admits seventeen Lie symmetries; fifteen are the CKVs of the metric (\ref%
{FRWLap}) plus the two Lie symmetries $\left( a_{0}u+b\left( x^{k}\right)
\right) \partial _{u}.$

For special functions $R\left( \tau \right) $, the Laplace equation (\ref%
{FRWLap1}) admits extra Lie symmetries; however, the conformal Laplace
equation (\ref{FRWLap2}) does not admit extra Lie symmetries.

A direct result, which arises from Theorems \ref{KG1} and \ref{CKG.011}, is
that, if $V^{n}$ is an $n$ dimensional Riemannian space,~$n>2$, then, if the
Laplace equation (\ref{KG.Eq1}) in $V^{n}$ is invariant under a Lie group $%
G_{L},$~then, $\ G_{L}$ is a subgroup \ of $\bar{G}_{L_{C}}$, i.e.~$%
G_{L}\subseteq \bar{G}_{LC}$ where $\bar{G}_{LC}$ is a Lie group which
leaves invariant the conformal Laplace equation (\ref{CKG.01A}). The Lie
algebras $G_{L},\bar{G}_{LC}$ are identical if the $V^{n}$ does not admit
proper CKVs or if all the conformal factors of the CKVs of $V^{n}~$are
solutions of the Laplace equation (\ref{KG.Eq1}). Moreover, if $V^{n}$ is a
conformally flat spacetime then, the conformal Laplace equation (\ref%
{CKG.01A}) admits a Lie algebra of $\frac{1}{2}\left( n+1\right) \left(
n+2\right) +2$ \ dimension. For instance, the Laplace equation in the three
dimensional sphere\footnote{%
The three dimensional sphere is a conformally flat space and it is maximall
symmetric \cite{Barnes}.} admits eight Lie point symmetries \cite{Freire2010}
while, on the contrary, the Yamabe equation admits twelve Lie symmetries.

\section{The heat conduction equation with a flux in a Riemannian space}

\label{HeatCon} The heat equation with a flux in an $n~$dimensional
Riemannian space with metric $g_{ij}$ is%
\begin{equation}
H\left( u\right) =q\left( t,x^{i},u\right)  \label{HEF.01}
\end{equation}%
where
\begin{equation*}
H\left( u\right) =\Delta u-u_{t}.
\end{equation*}%
and $\Delta $ is the Laplace operator.

The term $q\left( t,x^{i},u\right) $ \ indicates that the system exchanges
energy with the environment. In this case, the Lie symmetry vector is%
\begin{equation*}
\mathbf{X}=\xi ^{i}\left( x^{j},u\right) \partial _{i}+\eta \left(
x^{j},u\right) \partial _{u}
\end{equation*}%
where $a=t,i$. \ For this equation, we have

\begin{equation*}
A^{tt}=0,\text{ }A^{ti}=0,\text{ }A^{ij}=g^{ij},B^{i}=%
\Gamma^{i}(t,x^{i}),B^{t}=1,f(x,u)=q\left( t,x^{k},u\right).
\end{equation*}

For this PDE, the symmetry conditions (\ref{GPE.42}) - (\ref{GPE.46a}) become%
\begin{equation}
\eta =a(t,x^{i})u+b(t,x^{i})  \label{HEF.01.1}
\end{equation}%
\begin{equation}
\ \xi ^{t}=\xi ^{t}(t)  \label{HEF.01.3}
\end{equation}%
\begin{equation}
g^{ij}(a_{ij}u+b_{ij})-(a_{,i}u+b_{,i})\Gamma ^{i}-\left(
a_{,t}u+b_{,t}\right) +\lambda q=\xi ^{t}q_{,t}+\xi ^{k}q_{,k}+\eta q_{,u}
\label{HEF.01.4}
\end{equation}%
\begin{equation}
g^{ij}\xi _{,ij}^{k}-2g^{ik}a_{,i}+a\Gamma ^{k}-\xi _{,i}^{k}\Gamma ^{i}+\xi
^{i}\Gamma _{,i}^{k}-\lambda \Gamma ^{k}=0  \label{HEF.01.6}
\end{equation}%
\begin{equation}
L_{\xi ^{i}\partial _{i}}g_{ij}=(a-\lambda )g_{ij}.  \label{HEF.01.7}
\end{equation}

The solution of the symmetry conditions is summarized in Theorem \ref{The
Lie of the heat equation with flux} (for an sketch proof see Appendix \ref%
{appendixP}).

\begin{theorem}
\label{The Lie of the heat equation with flux}The Lie symmetries of the heat
equation with flux i.e.
\begin{equation}
g^{ij}u_{ij}-\Gamma ^{i}u_{i}-u_{t}=q\left( t,x,u\right)  \label{HEF.18}
\end{equation}%
in a $n$~dimensional Riemannian space with metric $g_{ij}$ are constructed
from the homothetic algebra of the metric as follows:

a. $Y^{i}$ is a nongradient HV/KV.\newline
The Lie symmetry is
\begin{equation}
X=\left( 2c_{2}\psi t+c_{1}\right) \partial _{t}+c_{2}Y^{i}\partial
_{i}+\left( a\left( t\right) u+b\left( t,x\right) \right) \partial _{u}
\label{HEF.19}
\end{equation}%
where $a(t),b\left( t,x^{k}\right) ,q\left( t,x^{k},u\right) $ must satisfy
the constraint equation%
\begin{equation}
-a_{t}u+H\left( b\right) -\left( au+b\right) q_{,u}+aq-\left( 2\psi
c_{2}qt+c_{1}q\right) _{t}-c_{2}q_{,i}Y^{i}=0.  \label{HEF.20}
\end{equation}

b. $Y^{i}=S^{,i}$ is a gradient HV/KV.\newline
The Lie symmetry is
\begin{equation}
X=\left( 2\psi \int Tdt+c_{1}\right) \partial _{t}+TS^{,i}\partial
_{i}+\left( \left( -\frac{1}{2}T_{,t}S+F\left( t\right) \right) u+b\left(
t,x\right) \right) \partial _{u}  \label{HEF.21}
\end{equation}%
where $F(t),T(t),b\left( t,x^{k}\right) ,q\left( t,x^{k},u\right) $ must
satisfy the constraint equation%
\begin{align}
0& =\left( -\frac{1}{2}T_{,t}\psi +\frac{1}{2}T_{,tt}S-F_{,t}\right)
u+H\left( b\right) +  \notag \\
& -\left( \left( -\frac{1}{2}T_{,t}S+F\right) u+b\right) q_{,u}+\left( -%
\frac{1}{2}T_{,t}S+F\right) q-\left( 2\psi q\int Tdt+c_{1}q\right)
_{t}-Tq_{,i}S^{,i}.  \label{HEF.22}
\end{align}
\end{theorem}

Below, we apply Theorem \ref{The Lie of the heat equation with flux} for
special forms of the function $q\left( t,x,u\right) $.

\subsection{The homogeneous heat equation}

\label{hhe}

In the case $q\left( t,x,u\right) =0,~$i.e equation (\ref{HEF.01}) is the
homogeneous heat equation, we have the following result.

\begin{theorem}
\label{The Lie of the heat equation}The Lie symmetries of the homogeneous
heat equation in an $n-$dimensional Riemannian space
\begin{equation}
g^{ij}u_{ij}-\Gamma ^{i}u_{i}-u_{t}=0  \label{LHEC.01}
\end{equation}%
are constructed from the homothetic algebra of the metric $g_{ij}$ as follows

a. If $Y^{i}$ is a nongradient HV/KV of the metric $g_{ij},$ the Lie
symmetry is
\begin{equation}
X=\left( 2\psi c_{1}t+c_{2}\right) \partial _{t}+c_{1}Y^{i}\partial
_{i}+\left( a_{0}u+b\left( t,x^{i}\right) \right) \partial _{u}
\label{LHEC.03}
\end{equation}%
where $c_{1},c_{2},,a_{0}$ are constants and $b\left( t,x^{i}\right) $ is a
solution of the homogeneous heat equation.

b. If $Y^{i}=S^{,i},$ that is, $Y^{i}$ is a gradient HV/KV of the metric $%
g_{ij},$ the Lie symmetry is%
\begin{equation}
X=(c_{3}\psi t^{2}+c_{4}t+c_{5})\partial _{t}+(c_{3}t+c_{4})S^{i}\partial
_{i}+\left( -\frac{c_{3}}{2}S-\frac{c_{3}}{2}n\psi t+c_{5}\ \right)
u\partial _{u}+b\left( t,x^{i}\right) \partial _{u}  \label{LHEC.04+}
\end{equation}%
where $c_{3},c_{4},c_{5}$ are constants and $b\left( t,x^{i}\right) $ is a
solution of the homogeneous heat equation.
\end{theorem}

In order to compare the above result with the existing results in the
literature, we consider the heat equation in a Euclidian space of dimension $%
n.$ Then, in Cartesian coordinates $g_{ij}=\delta _{ij}~$and $\Gamma ^{i}=0$%
, therefore, the homogeneous heat equation is%
\begin{equation}
\delta ^{ij}u_{ij}-u_{t}=0.  \label{LHEC.04}
\end{equation}

The homothetic algebra of space consists of the $n$ gradient KVs $\partial
_{i}$ with generating functions $x^{i},$ the $\frac{n\left( n-1\right) }{2}$
nongradient KVs $X_{IJ},$ which are the rotations and a gradient HV $H^{i}$
with gradient function $H=\,R\partial _{R}.$ According to Theorem \ref{The
Lie of the heat equation}, the Lie symmetries of the heat equation in the
Euclidian $n$ dimensional space are (we may take $\psi =1)$%
\begin{eqnarray}
X &=&\left[ c_{3}\psi t^{2}+(c_{4}+2\psi c_{1})t+c_{5}+c_{2}\right] \partial
_{t}+\left[ c_{1}Y^{i}+(c_{3}t+c_{4})S^{i}\right] \partial _{i}+
\label{LHEC.04a} \\
&&+\left[ \left( a_{0}+\frac{c_{3}}{2}S+\frac{c_{3}}{2}n\psi t-c_{5}\right)
u+b\left( t,x^{i}\right) \right] \partial _{u}.  \notag
\end{eqnarray}%
This result is consistent with the results of \cite{StephaniB} pg. 158.

Next, we consider the de Sitter spacetime (a four dimensional space of
constant curvature and Lorentzian character) whose metric is
\begin{equation}
ds^{2}=\frac{\left( -d\tau ^{2}+dx^{2}+dy^{2}+dz^{2}\right) }{\left( 1+\frac{%
K}{4}\left( -\tau ^{2}+x^{2}+y^{2}+z^{2}\right) \right) ^{2}}  \label{spc.00}
\end{equation}%
It is known that the homothetic algebra of this space consists of the ten KVs%
\begin{align*}
X_{1}& =\left( -x\tau \right) \partial _{\tau }+\left( \frac{\left( -\tau
^{2}-x^{2}+y^{2}+z^{2}\right) }{2}-\frac{2}{K}\right) \partial _{x}+\left(
-yx\right) \partial _{y}+\left( -zx\right) \partial _{x} \\
X_{2}& =\left( y\tau \right) \partial _{\tau }+\left( yx\right) \partial
_{x}+\left( \frac{\left( -x^{2}-z^{2}+y^{2}+\tau ^{2}\right) }{2}+\frac{2}{K}%
\right) \partial _{y}+\left( yz\right) \partial _{x} \\
X_{3}& =\left( z\tau \right) \partial _{\tau }+\left( zx\right) \partial
_{x}+\left( zy\right) \partial _{y}+\left( \frac{\left(
-x^{2}-y^{2}+z^{2}+\tau ^{2}\right) }{2}+\frac{2}{K}\right) \partial _{x} \\
X_{4}& =\left( \frac{\left( x^{2}+y^{2}+z^{2}+\tau ^{2}\right) }{2}-\frac{2}{%
K}\right) \partial _{\tau }+\left( \tau x\right) \partial _{x}+\left( \tau
y\right) \partial _{y}+\left( \tau z\right) \partial _{x} \\
X_{5}& =x\partial _{\tau }+\tau \partial _{x}~,~X_{6}=y\partial _{\tau
}+\tau \partial _{y}~,~X_{7}=z\partial _{\tau }+\tau \partial
_{z}~,~X_{8}=y\partial _{x}-x\partial _{y} \\
X_{9}& =z\partial _{x}-x\partial _{z}~,~X_{10}=z\partial _{y}-y\partial _{z}
\end{align*}%
all of which are nongradient. According to Theorem \ref{The Lie of the heat
equation}, the Lie symmetries of the heat equation in de Sitter space are%
\begin{equation*}
\partial _{t}+\sum\limits_{A=1}^{10}c_{A}X_{A}+(a_{0}u+b\left( x,u\right)
)\partial _{u}.~
\end{equation*}

From Theorem \ref{The Lie of the heat equation} we have the following
additional results.

\begin{corollary}
The one dimensional homogenous heat equation admits a maximum number of
seven Lie point symmetries (module a solution of the heat equation).
\end{corollary}

\begin{proof}
The homothetic group of a one~dimensional metric $ds^{2}=g^{2}\left(
x\right) dx^{2}$ consists of one gradient KV (the $\frac{1}{g\left( x\right)
}\partial _{x})$ and one gradient HV $(\frac{1}{g\left( x\right) }\int
g\left( x\right) dx~\partial _{x})$. According to theorem \ref{The Lie of
the heat equation}, from the KV we have two Lie symmetries and from the
gradient HV\ another two Lie point symmetries. To these we need to add the
two Lie point symmetries $X=a_{0}u\partial _{u}+b\left( t,x^{i}\right)
\partial _{u}$ and the trivial Lie symmetry $\partial _{t}$ where $b\left(
t,x^{i}\right) $ is a solution of the heat equation.
\end{proof}

\begin{corollary}
The homogeneous heat equation in a space of constant curvature of dimension $%
n$ has at most \ $\left( n+3\right) +\frac{1}{2}n\left( n-1\right) ~$Lie
symmetries (modulo a solution of the heat equation).
\end{corollary}

\begin{proof}
A space of constant curvature of dimension $n$ admits $n+\frac{1}{2}n\left(
n-1\right) $ nongradient KVs. To these we need to add the Lie symmetries \ $%
X=c\partial _{t}+a_{0}u\partial _{u}+b\left( t,x^{i}\right) \partial _{u}.$
\end{proof}

\begin{corollary}
The heat conduction equation in a space of dimension $n$ admits at most $%
\frac{1}{2}n\left( n+3\right) +5$ Lie symmetries (modulo a solution of the
heat equation) and if this is the case, the space is flat.
\end{corollary}

\begin{proof}
The space with the maximum homothetic algebra is the flat space which admits
$n$ gradient KVs, $\frac{1}{2}n\left( n-1\right) $ nongradient KVs and one
gradient HV. Therefore, from Case 1 of Theorem \ref{The Lie of the heat
equation} we have $\left( n+1\right) +\frac{1}{2}n\left( n-1\right) $ Lie
symmetries. From Case 2 of Theorem \ref{The Lie of the heat equation}, we
have $\left( n+1\right) $ Lie symmetries and to these we have to add the Lie
symmetries $X=c_{1}\partial _{t}+a_{0}u\partial _{u}+b\left( t,x^{i}\right)
\partial _{u}$ where $b\left( t,x^{i}\right) $ is a solution of the heat
equation. The set of all these symmetries is $1+2n+\frac{1}{2}n\left(
n-1\right) +2+1+1=$ $\frac{1}{2}n\left( n+3\right) +5$ \cite{IbragB}.
\end{proof}

\subsection{Case $\ q\left( t,x,u\right) =q\left( u\right) $}

Let $q\left( t,x,u\right) =q\left( u\right) ,~$then the heat conduction
equation (\ref{HEF.01}) becomes%
\begin{equation}
g^{ij}u_{ij}-\Gamma ^{i}u_{i}-u_{t}=q\left( u\right) .  \label{HEF.23}
\end{equation}%
From Theorem \ref{The Lie of the heat equation with flux}, we have the
following results.

\begin{theorem}
The Lie symmetries of the heat equation (\ref{HEF.23}) in an $n-$dimensional
Riemannian space with metric $g_{ij}$ are constructed form the homothetic
algebra of the metric as follows.

a. $Y^{i}$ is a HV/KV.~The Lie symmetry is
\begin{equation}
X=\left( 2c\psi t+c_{1}\right) \partial _{t}+cY^{i}\partial _{i}+\left(
a\left( t\right) u+b\left( t,x\right) \right) \partial _{u}  \label{HEF.24}
\end{equation}%
where the functions $a\left( t\right) ,$ $b\left( t,x\right) $ and $q\left(
u\right) $ satisfy the condition
\begin{equation}
-a_{t}u+H\left( b\right) -\left( au+b\right) q_{,u}+\left( a-2\psi c\right)
q=0.  \label{HEF.25}
\end{equation}%
b. $Y^{i}=S^{,i}$ is a gradient HV/KV.~The Lie symmetry is
\begin{equation}
X=\left( 2\psi \int Tdt+c_{1}\right) \partial _{t}+TS^{,i}\partial
_{i}+\left( \left( -\frac{1}{2}T_{,t}S+F\left( t\right) \right) u+b\left(
t,x\right) \right) \partial _{u}  \label{HEF.26}
\end{equation}%
where $b\left( t,x\right) $ is a solution of the homogeneous heat equation,
the functions $T(t),$ $F\left( t\right) ~$and the flux $q\left( u\right) ~$%
satisfy the equation:%
\begin{equation}
\left( -\frac{1}{2}T_{,t}\psi +\frac{1}{2}T_{,tt}S-F_{,t}\right) u+H\left(
b\right) -\left( \left( -\frac{1}{2}T_{,t}S+F\right) u+b\right)
q_{,u}+\left( -\frac{1}{2}T_{,t}S+F\right) q-2\psi qT=0  \label{HEF.27}
\end{equation}
\end{theorem}

For various cases of $q\left( u\right) ,$ we obtain the results of Table%
\footnote{%
Where int Table \ref{tableC5} $Y^{i}$ is a HV/KV, $S^{,i}$ is a gradient
HV/KV and $K^{,i}$ is a gradient KV.} \ref{tableC5}.%
\begin{table}[tbp] \centering%
\caption{Functions of $q(u)$ where the Heat equation admits Lie symmetries}%
\begin{tabular}{cc}
\hline\hline
Function$~q\left( u\right) $ & Lie Symmetry vector \\ \hline
$q\left( u\right) =q_{0}u$ & $\left( \psi T_{0}t^{2}+2c\psi t+c_{1}\right)
\partial _{t}+\left( cY^{i}+T_{0}tS^{,i}\right) \partial _{i}+$ \\
& $~+\left( \left[ -2\psi cq_{0}t+a_{0}+T_{0}\left( -\frac{1}{2}S-\psi
q_{0}t^{2}-\frac{1}{2}t\right) \right] u+b\left( t,x\right) \right) \partial
_{u}\;\text{where $H(b)-bq_{0}=0$}$ \\
$q\left( u\right) =q_{0}u^{n}$ & $\left( 2c\psi t+c_{1}\right) \partial
_{t}+cY^{i}\partial _{i}+\left( \frac{2\psi c}{1-n}u\right) \partial _{u}$
\\
$q\left( u\right) =u\ln u$ & $c_{1}\partial _{t}+\left(
Y^{i}+T_{0}e^{-t}K^{^{,}i}\right) \partial _{i}+\left( a_{0}e^{-t}u\right)
\partial _{u}~,~K^{^{,}i}$ \\
$q\left( u\right) =e^{u}$ & $\left( 2c\psi t+c_{1}\right) \partial
_{t}+cY^{i}\partial _{i}+\left( -2\psi c\right) \partial _{u}$ \\
\hline\hline
\end{tabular}%
\label{tableC5}%
\end{table}%

\section{Conclusion}

The main result of this chapter is Proposition \ref{Coefficient xi third},
which states that the Lie symmetries of the PDEs of the form \ (\ref{GPE.10}%
) are obtained from the conformal vectors of the metric defined by the
coefficients $A_{ij},$ provided $A_{ij,u}=0.$ This result is quite general
and covers many well known and important PDEs of Physics. The geometrization
of Lie point symmetries and their association with the collineations of the
metric dissociates their determination from the dimension of the space
because the collineations of the metric depend (in general) on the type of
the metric and not on the dimensions of the space where the metric resides.
Furthermore, this association provides a wealth of results of Differential
Geometry on collineations, which is possible to be used in the determination
of Lie pont symmetries.

We have applied the above theoretical results to the Poisson equation, the
conformal Poisson (Yamabe) equation and the heat equation. We proved that
the Lie symmetries of the Poisson equations (Laplace/Yamabe) are generated
from the elements of the conformal group of the metric that defines the
Laplace/Yamabe operator.

For the heat conduction equation, we proved that the Lie symmetries are
generated from the homothetic group of the underlying metric. Furthermore,
we specialized the equation to the homogeneous heat conduction equation and
regained the existing results for the Newtonian case.

In the following chapter. we apply the theoretical results of this chapter
to study the correlation of point symmetries between Classical and Quantum
systems.

\newpage%

\begin{subappendices}%

\section{Appendix A}

\label{appenA5}

We prove the statement for $n=2$. The generalization to any $n$ is
straightforward. For a general proof, see \cite{BlumanPaper}. We consider $A$
as a matrix and assume that the inverse of this matrix exists. We denote the
inverse matrix with $B$ and we get from (\ref{Po.1})%
\begin{align}
B_{ij}A^{ij}\xi _{.,u}^{k}+B_{ij}A^{kj}\xi _{.,u}^{i}+B_{ij}A^{ik}\xi
_{.,u}^{j}& =0  \notag \\
2\xi _{.,u}^{k}+\delta _{i}^{k}\xi _{.,u}^{i}+\delta _{j}^{k}\xi _{.,u}^{j}&
=0\Rightarrow  \notag \\
\xi _{.,u}^{k}& =0.  \label{Po.3}
\end{align}

Now, assume that the tensor $A$ does not have an inverse. Then, we consider $%
n=2$ and write:
\begin{equation*}
\lbrack A]=\left[
\begin{tabular}{ll}
$A_{11}$ & $A_{12}$ \\
$A_{12}$ & $A_{22}$%
\end{tabular}%
\right] \Rightarrow \det A=A^{11}A^{22}-(A^{12})^{2}=0
\end{equation*}%
where at least one of the $A^{ij}\neq 0.$ Assume $A^{11}\neq 0.$ Then,
equation (\ref{Po.1}) for $i=j=k=1$ gives%
\begin{equation*}
3A^{11}\xi _{.,u}^{1}=0\Rightarrow \xi _{.,u}^{1}=0.
\end{equation*}%
The same equation for $i=j=k=2$ gives%
\begin{equation*}
3A^{22}\xi _{.,u}^{2}=0
\end{equation*}%
therefore, either $\xi _{.,u}^{2}=0$ or $A^{22}=0.$ If $A^{22}=0$, then,
from the condition $\det A^{ij}=0,$ we have $A^{12}=0$; hence, $A_{ij}=0,$
which we do not assume. Thus, $\xi _{.,u}^{2}=0.$

We consider now equations $i=j\neq k$ and find%
\begin{equation*}
A^{ii}\xi _{.,u}^{k}+A^{ki}\xi _{.,u}^{i}+A^{ik}\xi _{.,u}^{i}=0.
\end{equation*}%
Because $i\neq k,$ this gives $A^{ii}\xi _{.,u}^{k}=0$ and because we have
assumed $A^{11}\neq 0$ it follows $\xi _{.,u}^{2}.$ Therefore, we find $\xi
_{.,u}^{k}=0.$

\section{Appendix B}

\label{appendixP}

\begin{proof}[Proof of Lemma \protect\ref{LemmaPDE.1}]
By using the formula%
\begin{equation*}
L_{\xi }\Gamma _{.jk}^{i}=\Gamma _{.jk,l}^{i}\xi ^{l}+\xi _{,~jk}^{i}-\xi
_{~,l}^{i}\Gamma _{.jk}^{l}+\xi _{~,j}^{s}\Gamma _{.sk}^{i}+\xi
_{~,k}^{s}\Gamma _{.sj}^{i}
\end{equation*}%
we have%
\begin{align*}
g^{jk}L_{\xi }\Gamma _{.jk}^{i}& =\Gamma _{,l}^{i}\xi ^{l}+g^{jk}\xi
_{~,,jk}^{i}-g_{~~~~,l}^{jk}\xi ^{l}\Gamma _{.jk}^{i}-\xi _{~,l}^{i}\Gamma
^{l}+2g^{jk}\xi _{~,j}^{s}\Gamma _{.sk}^{i} \\
& =\Gamma _{,l}^{i}\xi ^{l}+g^{jk}\xi _{~,~jk}^{i}-\xi _{~,l}^{i}\Gamma ^{l}
\\
& -[g^{jl}\xi _{,l}^{k}+g^{kl}\xi _{,l}^{j}-(\lambda -a))g^{jk}]\Gamma
_{.jk}^{i}+2g^{jk}\xi _{~,j}^{s}\Gamma _{~sk}^{i}
\end{align*}%
that is,%
\begin{equation*}
g^{jk}L_{\xi }\Gamma _{.jk}^{i}==\Gamma _{,l}^{i}\xi ^{l}+g^{jk}\xi
_{~,~jk}^{i}-\xi _{~,l}^{i}\Gamma ^{l}+2(g^{jl}\xi _{~,j}^{k}\Gamma
_{~kl}^{i}-g^{lj}\xi _{~,j}^{k}\Gamma _{~kl}^{i})-(\lambda -a)\Gamma ^{i}
\end{equation*}%
Therefore,%
\begin{equation*}
g^{jk}L_{\xi }\Gamma _{.jk}^{i}=g^{jk}\xi _{,~jk}^{i}+\Gamma _{~,l}^{i}\xi
^{l}-\xi _{~,l}^{i}\Gamma ^{l}+(a-\lambda )\Gamma _{~}^{i}.
\end{equation*}
\end{proof}

\begin{proof}[Proof of Lemma \protect\ref{LemmaPDE.2}]
By using the identity%
\begin{equation}
L_{\xi }\Gamma _{.jk}^{i}=\frac{1}{2}g^{ir}\left[ \nabla _{k}L_{\xi
}g_{jr}+\nabla _{j}L_{\xi }g_{kr}-\nabla _{r}L_{\xi }g_{kj}\right]
\end{equation}%
and replacing $L_{\xi }g_{ij}=(a-\lambda )g_{ij}$~because $\xi ~$is a
CKV,~we find%
\begin{align*}
L_{\xi }\Gamma _{.jk}^{i}& =\frac{1}{2}g^{ir}\left[ (a-\lambda
)_{,k}g_{jr}+(a-\lambda )_{,j}g_{kr}-(a-\lambda )_{,r}g_{kj}\right] \\
& =\frac{1}{2}\left[ (a-\lambda )_{,k}\delta _{j}^{i}+(a-\lambda
)_{,j}\delta _{k}^{i}-g^{ir}(a-\lambda )_{,r}g_{kj}\right] .
\end{align*}%
By contracting with $g^{jk}$ it follows%
\begin{equation*}
g^{jk}L_{\xi }\Gamma _{.jk}^{i}=\frac{2-n}{2}(a-\lambda )^{,i}.
\end{equation*}
\end{proof}

\begin{proof}[Proof of Theorem \protect\ref{Theor}]
The Lie symmetry conditions for the Poisson equation (\ref{LE.01a}) are
\begin{equation}
g^{ij}(a_{ij}u+b_{ij})-(a_{,i}u+b_{,i})\Gamma ^{i}-\xi
^{k}f_{,k}-auf_{,u}-bf_{,u}+\lambda f=0  \label{WDW.01}
\end{equation}%
\begin{equation}
g^{ij}\xi _{,ij}^{k}-2g^{ik}a_{,i}+a\Gamma ^{k}-\xi _{,i}^{k}\Gamma ^{i}+\xi
^{i}\Gamma _{,i}^{k}-\lambda \Gamma ^{k}=0  \label{WDW.02}
\end{equation}%
\begin{equation}
L_{\xi ^{i}\partial _{i}}g_{ij}=(a-\lambda )g_{ij}  \label{WDW.03}
\end{equation}%
\begin{equation}
\eta =a(x^{i})u+b(x^{i})~,~\xi _{,u}^{k}=0
\end{equation}%
Equation (\ref{WDW.02}) becomes (see \cite{Bozhkov})
\begin{equation}
g^{jk}L_{\xi }\Gamma _{.jk}^{i}=2g^{ik}a_{,i}  \label{WDW.04}
\end{equation}%
From (\ref{WDW.03}), $\xi ^{i}$ is a CKV, then equations (\ref{WDW.04}) give
\begin{equation*}
\frac{2-n}{2}(a-\lambda )^{,i}=2a^{,i}\rightarrow \left( a-\lambda \right)
^{i}=\frac{4}{2-n}a^{,i}.
\end{equation*}

We define
\begin{equation}
\psi =\frac{2}{2-n}a+a_{0}
\end{equation}%
where $\psi =\frac{1}{2}\left( a-\lambda \right) $ is the conformal factor
of $\xi ^{i}$ i.e. $L_{\xi }g_{ij}=2\psi g_{ij}$. Furthermore, we have%
\begin{eqnarray*}
\left( 2-n\right) \lambda ^{i} &=&\left( 2-n\right) a^{i}-4a^{i} \\
\left( 2-n\right) \lambda ^{i} &=&-\left( n+2\right) a^{i}
\end{eqnarray*}%
Finally, from (\ref{WDW.01}), we have the constraint%
\begin{equation}
g^{ij}a_{i;j}u+g^{ij}b_{;ij}-\xi ^{k}f_{,k}-auf_{,u}+\lambda f-bf_{,u}=0.
\end{equation}

For $n=2,$ holds that $~g^{jk}L_{\xi }\Gamma _{.jk}^{i}=0;$ this means that $%
a_{,i}=0\rightarrow a=a_{0}.$ From (\ref{WDW.03}), $\xi ^{i}$ is a CKV with
conformal factor
\begin{equation}
2\psi =\left( a_{0}-\lambda \right)
\end{equation}%
and~$\lambda =a_{0}-2\psi .~$Finally, from (\ref{WDW.01}), we have the
constraint%
\begin{equation}
g^{ij}b_{;ij}-\xi ^{k}f_{,k}-a_{0}uf_{,u}+\left( a_{0}-2\psi \right)
f-bf_{,u}=0.
\end{equation}
\end{proof}

\begin{proof}[Proof of Theorem \protect\ref{CTheor}]
By replacing $f\left( x^{k},u\right) =-M_{0}R+\bar{f}\left( x^{i},u\right) $
in (\ref{KG.Eq0})~where $M_{0}=\frac{n-2}{4\left( n-1\right) }$, we have the
symmetry condition
\begin{equation}
\frac{2-n}{2}\Delta \psi u+g^{ij}b_{i;j}+M_{0}\xi ^{k}R_{,k}u-\xi ^{k}\bar{f}%
_{,k}+2M_{0}\psi uR-\frac{2-n}{2}\psi u\bar{f}_{k}-\frac{2+n}{2}\psi \bar{f}%
-b\left( -M_{0}R+\bar{f}_{u}\right) =0  \label{Ct.01}
\end{equation}%
But $\xi ^{i}$ is a CKV with conformal factor $\psi $. $\ $That implies \cite%
{HallSteele}
\begin{equation}
\xi ^{k}R_{,k}=-2\psi R-2\left( n-1\right) \Delta \psi  \label{Ct.02}
\end{equation}%
and this implies for the terms in (\ref{Ct.01})
\begin{equation*}
0=\left( \frac{2-n}{2}\Delta \psi +M_{0}\xi ^{k}R_{,k}+2M_{0}\psi R\right) u.
\end{equation*}%
Hence, condition become (\ref{Ct.01}) finally becomes%
\begin{equation*}
-\xi ^{k}\bar{f}_{,k}-\frac{2-n}{2}\psi u\bar{f}_{k}-\frac{2+n}{2}\psi \bar{f%
}+g^{ij}b_{i;j}-b\left( -M_{0}R+\bar{f}_{u}\right) =0.
\end{equation*}
\end{proof}

\begin{proof}[Proof of Theorem \protect\ref{The Lie of the heat equation
with flux}]
Condition (\ref{HEF.01.7}) means that $\xi ^{i}$ is a CKV of the metric $%
g_{ij}$ with conformal factor $a(t,x^{k})-\lambda (t,x^{k}).~$Condition (\ref%
{HEF.01.6}) implies $\xi ^{k}=T\left( t\right) Y^{k}\left( x^{j}\right) ,$
where $Y^{i}$ is a HV with conformal factor $\psi ,$ that is, we have:
\begin{equation*}
L_{Y^{i}}g_{ij}=2\psi g_{ij}\text{ ~,~}\psi \text{=constant.}
\end{equation*}%
and%
\begin{equation*}
\xi _{,t}^{t}=a-\lambda
\end{equation*}%
from which follows
\begin{equation}
\xi ^{t}\left( t\right) =2\psi \int Tdt.  \label{HEF.01.08}
\end{equation}%
\begin{equation}
-2g^{ik}a_{,i}+T_{,t}Y^{k}=0.  \label{HEF.01.09}
\end{equation}%
Condition (\ref{HEF.01.4}) becomes%
\begin{equation*}
H(a)u+H(b)+(a-\xi _{,t}^{t})q=\xi ^{t}q_{,t}+T(t)Y^{k}q_{,k}+\eta
q_{,u}\Rightarrow
\end{equation*}%
\begin{equation*}
H(a)u+H(b)-\left( au+b\right) q_{,u}+aq-(\xi ^{t}q)_{,t}-T(t)Y^{k}q_{,k}=0
\end{equation*}%
\begin{equation}
H\left( a\right) u+H\left( b\right) -\left( au+b\right) q_{,u}+aq-\left(
2\psi q\int Tdt\right) _{t}-Tq_{,i}Y^{i}=0.  \label{HEF.12}
\end{equation}%
We consider the following cases:

Case 1: $Y^{k}$ is a HV/KV. From (\ref{HEF.01.09}), we have that $%
T_{,t}=0\rightarrow T\left( t\right) =c_{2}$ and $a_{,i}=0\rightarrow
a\left( t,x^{k}\right) =a\left( t\right) .$ Then, (\ref{HEF.12}) becomes%
\begin{equation}
-a_{t}u+H\left( b\right) -\left( au+b\right) q_{,u}+aq-\left( 2\psi
c_{2}qt+c_{1}q\right) _{t}-c_{2}q_{,i}Y^{i}=0  \label{HEF.12.1}
\end{equation}%
Case 2: $Y^{k}$ is a gradient HV/KV, that is $Y^{k}=S^{,k}.$ From (\ref%
{HEF.01.09}), we have%
\begin{equation}
a\left( t,x^{k}\right) =-\frac{1}{2}T_{,t}S+F\left( t\right) .
\label{HEF.13}
\end{equation}%
By replacing in (\ref{HEF.12}), we find the constraint equation%
\begin{align}
0& =\left( -\frac{1}{2}T_{,t}\psi +\frac{1}{2}T_{,tt}S-F_{,t}\right)
u+H\left( b\right) +  \notag \\
& -\left( \left( -\frac{1}{2}T_{,t}S+F\right) u+b\right) q_{,u}+\left( -%
\frac{1}{2}T_{,t}S+F\right) q-\left( 2\psi q\int Tdt+c_{1}q\right)
_{t}-Tq_{,i}S^{,i}.  \label{HEF.14}
\end{align}
\end{proof}

\end{subappendices}%

\chapter{Point symmetries of Schr\"{o}dinger and the Klein Gordon equations
\label{chapter6}}

\section{Introduction}

The Schr\"{o}dinger and the Klein Gordon equations are two important
equations of Quantum Physics. Therefore, it is important that we determine
their Lie symmetries and use them either in order to find invariant
solutions\ using Lie symmetry methods \cite{StephaniB}. In order to achieve
this, we notice that the Schr\"{o}dinger equation is a special case of the
heat conduction equation and the Klein Gordon equation is a special form of
the Poisson equation. The Lie symmetries of the heat equation and of the
Poisson equation in\ a general Riemannian space were determined in Chapter %
\ref{chapter5}. Thus, we apply these results to find the Lie symmetries of
the Schr\"{o}dinger and the Klein Gordon equation in a general Riemannian
space.

An important element of the present study is the concept of conformally
related Lagrangians, that is, Lagrangians that under under a conformal
transformation of the metric and the potential lead to the same equations
but for different dynamic variables. The condition for this is that the
Hamiltonian vanishes. Because the dynamic variables of these Lagrangians are
not the standard ones in general the Hamiltonian is not relevant to the
energy of the system.

From each Lagrangian describing a dynamical system, we define a metric
called the kinematic metric, characteristic to the dynamical system
described by this Lagrangian. As it will be seen, the conformal symmetries
of this metric are in close relation to the Noether symmetries of the
equations of motion. Furthermore, the kinetic metric of the Lagrangian
defines the Laplace operator; hence, consequently the Lie symmetries of the
corresponding Poisson equation are expressed in terms of the conformal
symmetries of the kinematic metric. We extend these results to the Yamabe
operator and study the Lie symmetries of the conformal Klein Gordon equation.

In section \ref{Symmetries of Lagrangian}, we consider the classical
Lagrangian
\begin{equation}
L\left( x^{k},\dot{x}^{k}\right) =\frac{1}{2}g_{ij}\left( x^{k}\right) \dot{x%
}^{i}\dot{x}^{j}-V\left( x^{k}\right)  \label{ANS1}
\end{equation}%
in a general Riemannian space and we show that the Noether point symmetries
of two conformally related Lagrangians are generated from the conformal
algebra of the metric $g_{ij}.$

In section \ref{Symmetries of Schrodinger and the Klein Gordon equations},
we \ study the Lie point symmetries of Schr\"{o}dinger and the Klein Gordon
equation by using the resutls of Chapter \ref{chapter5}. Using the geometric
character of the Noether symmetries for the Lagrangian (\ref{ANS1}) and that
of\ the Lie symmetries of the Schr\"{o}dinger and of the Klein Gordon
equation we establish the connection between the two. More specifically, it
will be shown that if an element of the homothetic group of the kinetic
metric generates a Noether point symmetry for the classical Lagrangian, then
it also generates a Lie point symmetry for the Schr\"{o}dinger equation.
Concerning the Klein Gordon we find that the Noether symmetry of the
Lagrangian (\ref{ANS1}) must have a constant Noether gauge function in order
to be admitted.

In section \ref{Symmetries of the Lagrangian with non constant gauge
function}, we examine the case of Noether symmetries whose Noether gauge
functions are not constant. We will consider the cases the kinematic metric
admits a gradient Killing vector (KV) or a gradient homothetic vector (HV)
which produces Noether point symmetries for the Lagrangian (\ref{ANS1}) and
show that the Lie symmetry in both cases is indeed a non-local symmetry of
the Klein Gordon equation. In section \ref{Applications}, we demonstrate the
use of the previous general results to various interesting practical
situations.

To complete our analysis, in section \ref{WKBAp}, we examine the WKB
approximation. In that case, the solution of the Klein Gordon equation
satisfies the null Hamilton Jacobi equation. We derive the symmetry
condition of the null Hamilton Jacobi equation and we prove that its Lie
point symmetries are generated from the CKVs of the underlying space.
Furthermore, there exists a unique relation between the Lie point symmetries
of the Hamilton Jacobi and the Lie symmetries of the Euler-Lagrange
equations of a classical particle; in particularly, the Lie point symmetries
of the Euler-Lagrange equations which are generated from the homothetic
algebra of the Riemannian space are generating point symmetries for the
Hamilton Jacobi equation; that is, the Lie symmetry algebra of the Hamilton
Jacobi equation can be greater than the Noether algebra of the classical
Lagrangian (\ref{ANS1}).

\section{Noether symmetries of Conformal Lagrangians}

\label{Symmetries of Lagrangian}

Consider the Lagrangian of a particle moving under the action of a potential
$V(x^{k})$ in a Riemannian space with metric $g_{ij}$
\begin{equation}
L=\frac{1}{2}g_{ij}\dot{x}^{i}\dot{x}^{j}-V\left( x^{k}\right)
\label{CLN.05}
\end{equation}%
where $\dot{x}=\frac{dx}{dt}.~$The equations of motion follow from the action%
\begin{equation}
S=\int dt\left( L\left( x^{k},\dot{x}^{k}\right) \right) =\int dt\left(
\frac{1}{2}g_{ij}\dot{x}^{i}\dot{x}^{j}-V\left( x^{k}\right) \right) .
\label{CLN.06}
\end{equation}%
Consider the change of variable $t\rightarrow \tau $ defined by the
requirement%
\begin{equation}
d\tau =N^{2}\left( x^{i}\right) dt.  \label{CLN.06a}
\end{equation}%
In the new coordinates $(\tau ,x^{i}),$ the action becomes%
\begin{equation}
S=\int \frac{d\tau }{N^{2}\left( x^{k}\right) }\left( \frac{1}{2}%
g_{ij}N^{4}\left( x^{k}\right) x^{\prime i}x^{\prime j}-V\left( x^{k}\right)
\right) ~~  \label{CLN.07}
\end{equation}%
where $x^{\prime i}=\frac{dx}{d\tau }$ and the Lagrangian is transformed to
the new Lagrangian%
\begin{equation}
\bar{L}\left( x^{k},x^{\prime k}\right) =\frac{1}{2}N^{2}\left( x^{k}\right)
g_{ij}x^{\prime i}x^{\prime j}-\frac{V\left( x^{k}\right) }{N^{2}\left(
x^{k}\right) }.  \label{CLN.08}
\end{equation}%
If we consider a conformal transformation (not a coordinate transformation!)
of the metric $\bar{g}_{ij}=N^{2}\left( x^{k}\right) g_{ij}$ and a new
potential function $\bar{V}\left( x^{k}\right) =\frac{V\left( x^{k}\right) }{%
N^{2}\left( x^{k}\right) }$ then, the new Lagrangian $\bar{L}\left(
x^{k},x^{\prime k}\right) $ in the new coordinates $\tau ,x^{k}$, takes the
form,%
\begin{equation}
\bar{L}\left( x^{k},x^{\prime k}\right) =\frac{1}{2}\bar{g}_{ij}x^{\prime
i}x^{\prime j}-\bar{V}\left( x^{k}\right)  \label{CLN.09}
\end{equation}%
implying that equation (\ref{CLN.09}) is of the same form as the Lagrangian $%
L$ in equation (\ref{CLN.05}). From now on, the Lagrangian $L\left( x^{k},%
\dot{x}^{k}\right) $ of equation (\ref{CLN.05}) and the Lagrangian $\bar{L}%
\left( x^{k},x^{\prime k}\right) $ of equation (\ref{CLN.09}) will be called
conformal. In this framework, the action remains the same, i.e. it is
invariant under the change of parameter and the equations of motion in the
new variables $(\tau ,x^{i})$ will be the same with the equations of motion
for the Lagrangian $L$ in the original coordinates $(t,x^{k})$ .

In Chapter \ref{chapter3}, it was shown that the Noether point symmetries of
a Lagrangian of the form (\ref{CLN.05}) follow from the homothetic algebra
of the metric $g_{ij}~$(see Theorem \ref{The Noether symmetries of a
conservative system}). The same applies to the Lagrangian $\bar{L}\left(
x^{k},x^{\prime k}\right) $ and the metric $\bar{g}_{ij}.$~The conformal
algebra of the metrics $g_{ij},\bar{g}_{ij}$ (as a set) is the same;
however, their closed subgroups of HVs and KVs are generally different.
Hence, the following Corollary holds.

\begin{corollary}
The Noether point symmetries of the conformally related Lagrangians (\ref%
{CLN.05}), (\ref{CLN.09}) are contained in the common conformal algebra of
the metrics $g_{ij},\bar{g}_{ij}.$
\end{corollary}

Now, we formulate and prove the following Lemma;

\begin{lemma}
\label{LemaHam}The Euler-Lagrange equations for two conformal Lagrangians
transform covariantly under the conformal transformation relating the
Lagrangians if and only if the Hamiltonian vanishes.
\end{lemma}

\begin{proof}
Consider the Lagrangian $L=\frac{1}{2}g_{ij}\dot{x}^{i}\dot{x}^{j}-V\left(
x^{k}\right) $ whose Euler-Lagrange equations are:
\begin{equation}
\ddot{x}^{i}+\Gamma _{jk}^{i}\dot{x}^{i}\dot{x}^{j}+V^{,i}=0  \label{CLN1.01}
\end{equation}%
where $\Gamma _{jk}^{i}$ are the Christofell symbols. The corresponding
Hamiltonian is given by
\begin{equation}
E=\frac{1}{2}g_{ij}\dot{x}^{i}\dot{x}^{j}+V\left( x^{k}\right) \;.
\label{CLN1.02}
\end{equation}%
For the conformally related Lagrangian $\bar{L}\left( x^{k},x^{\prime
k}\right) =\left( \frac{1}{2}N^{2}\left( x^{k}\right) g_{ij}x^{\prime
i}x^{\prime j}-\frac{V\left( x^{k}\right) }{N^{2}\left( x^{k}\right) }%
\right) $ where $N_{,j}\neq 0$ the resulting Euler Lagrange equations are%
\begin{equation}
x^{\prime \prime i}+\hat{\Gamma}_{jk}^{i}x^{\prime j}x^{\prime k}+\frac{1}{%
N^{4}}V^{,i}-\frac{2V}{N^{5}}N^{,i}=0  \label{CLN1.03}
\end{equation}%
where
\begin{equation}
\hat{\Gamma}_{jk}^{i}=\Gamma _{jk}^{i}+(\ln N)_{,k}\delta _{j}^{i}+(\ln
N)_{,j}\delta _{k}^{i}-(\ln N)^{,i}g_{jk}  \label{CLN1.04}
\end{equation}%
and the corresponding Hamiltonian is%
\begin{equation}
\bar{E}=\frac{1}{2}N^{2}\left( x^{k}\right) g_{ij}\dot{x}^{i}\dot{x}^{j}+%
\frac{V\left( x^{k}\right) }{N^{2}\left( x^{k}\right) }.  \label{CLN1.05}
\end{equation}%
In order to show that the two equations of motion are conformally related we
start from equation(\ref{CLN1.03}) and apply the conformal transformation
\begin{align*}
x^{\prime i}& =\frac{dx^{i}}{d\tau }=\frac{dx^{i}}{dt}\frac{dt}{d\tau }=\dot{%
x}^{i}\frac{1}{N^{2}} \\
x^{^{\prime \prime }i}& =\ddot{x}^{i}\frac{1}{N^{4}}-2\dot{x}^{i}\dot{x}%
^{j}\left( \ln N\right) _{,j}\frac{1}{N^{4}}.
\end{align*}%
By replacing in equation(\ref{CLN1.03}). we find:%
\begin{equation*}
\ddot{x}^{i}\frac{1}{N^{4}}-2\dot{x}^{i}\dot{x}^{j}\left( \ln N\right) _{,j}%
\frac{1}{N^{4}}+\frac{1}{N^{4}}\hat{\Gamma}_{jk}^{i}\dot{x}^{j}\dot{x}^{k}+%
\frac{1}{N^{4}}V^{,i}-\frac{2V}{N^{5}}N^{,i}=0
\end{equation*}%
By replacing $\hat{\Gamma}_{jk}^{i}$ from equation(\ref{CLN1.04}), we have
\begin{align*}
& \ddot{x}^{i}-2\dot{x}^{i}\dot{x}^{j}\left( \ln N\right) _{,j}+\Gamma
_{jk}^{i}\dot{x}^{j}\dot{x}^{k}+2(\ln N)_{,j}\dot{x}^{j}\dot{x}^{i} \\
-(\ln N)^{,i}g_{jk}\dot{x}^{j}\dot{x}^{k}+V^{,i}-2V(\ln N)^{,i}& =0
\end{align*}%
from which follows%
\begin{equation*}
\ddot{x}^{i}+\Gamma _{jk}^{i}\dot{x}^{j}\dot{x}^{k}+V^{,i}-(\ln
N)^{,i}\left( g_{jk}\dot{x}^{j}\dot{x}^{k}+2V\right) =0.
\end{equation*}%
Obviously, the above Euler-Lagrange equations coincide with equations(\ref%
{CLN1.01}) if and only if $g_{jk}\dot{x}^{j}\dot{x}^{k}+2V=0$, which implies
that the Hamiltonian of equation(\ref{CLN1.02}) vanishes. The steps are
reversible; hence, the inverse is also true.
\end{proof}

The physical meaning of this result is that systems with vanishing energy
are conformally invariant at the level of the equations of motion.

\section{Lie point symmetries of Schr\"{o}dinger and the Klein Gordon
equation}

\label{Symmetries of Schrodinger and the Klein Gordon equations}

In this section we study the Lie point symmetries of Schr\"{o}dinger and the
Klein Gordon equation in \ a Riemannian manifold. To do this, we use the
results of Chapter \ref{chapter5}. Furthermore, we will study the relation
between Noether pont symmetries of classical Lagrangians and Lie point
symmetries of the Schr\"{o}dinger and the Klein Gordon equation with the
same\ "kinetic" metric and the same potential.

\subsection{Symmetries of the Schr\"{o}dinger equation}

The Schr\"{o}dinger equation\footnote{%
We have absorbed the constant $\hslash $ and the imaginary unit $i$, in the
variables $x^{k},t$ respectively}
\begin{equation}
g^{ij}u_{ij}-\Gamma ^{i}u_{i}-u_{t}=V\left( x\right) u  \label{HC.01}
\end{equation}%
is a special form of the heat conduction equation (\ref{HEF.01}) with $%
q\left( t,x,u\right) =V\left( x\right) u.$ Therefore, it is possible to
study the Lie point symmetries of the Schr\"{o}dinger equation using Theorem %
\ref{The Lie of the heat equation with flux} which, in this case, takes the
following form.

\begin{theorem}
The Lie point symmetries of the Schr\"{o}dinger equation (\ref{HC.01}) are
generated from the elements of the homothetic algebra of the metric $g_{ij}$
as follows.

a. $Y^{i}$ is a non-gradient HV/KV. \ The Lie symmetry is
\begin{equation}
X=\left( 2c\psi t+c_{1}\right) \partial _{t}+cY^{i}\partial _{i}+\left(
a_{0}u+b\left( t,x\right) \right) \partial _{u}
\end{equation}%
with constraint equations%
\begin{equation}
H\left( b\right) -bV=0~,~~cL_{Y}V+2\psi cV+a_{0}=0.
\end{equation}

b $Y^{i}=H^{,i}$ is a gradient HV/KV. The Lie symmetry is%
\begin{equation*}
X=\left( 2\psi \int Tdt+c_{1}\right) \partial _{t}+TS^{,i}\partial
_{i}+\left( \left( -\frac{1}{2}T_{,t}S+F\left( t\right) \right) u+b\left(
t,x\right) \right) \partial _{u}
\end{equation*}%
with constraint equations%
\begin{eqnarray}
H\left( b\right) -bV &=&0 \\
L_{H}V+2\psi V-\frac{1}{2}c^{2}H+d &=&0
\end{eqnarray}%
and the functions $T,F$ are computed from the relations
\begin{equation}
T_{,tt}=c^{2}T~,~\frac{1}{2}T_{,t}\psi +F_{,t}=dT.
\end{equation}
\end{theorem}

From the form of the symmetry vectors and the symmetry conditions for the
Schr\"{o}dinger and the Lagrangian of the classical particle (\ref{CLN.05})
we have the following result.

\begin{proposition}
\label{SymKGHJ11}If a KV/HV of the metric $g_{ij}$ produces a Lie point
symmetry for the Schr\"{o}dinger equation (\ref{HC.01}), then generates a
Noether point symmetry for the Lagrangian (\ref{CLN.05}) in the space with
metric $g_{ij}$ and potential $V\left( x^{k}\right) .~$The reverse is also
true.
\end{proposition}

\subsection{Symmetries of the Klein Gordon equation}

\label{Lie point symmetries of the Klein Gordon equation2}The Klein-Gordon
equation
\begin{equation}
\Delta u-V\left( x^{k}\right) u=0  \label{KG.Eq111}
\end{equation}%
follows from the Poisson equation (\ref{LE.01a}) if we take $f\left(
x^{i},u\right) =V\left( x^{i}\right) u.$ Therefore, Theorem \ref{Theor}
applies and we have the following result

\begin{theorem}
\label{KG}The Lie point symmetries of the Klein Gordon equation (\ref%
{KG.Eq111})$\ $are generated from the CKVs of the metric~$g_{ij}$ defining
the Laplace operator, as follows

a) for $n>2,$ the Lie symmetry vector is%
\begin{equation}
X=\xi ^{i}\left( x^{k}\right) \partial _{i}+\left( \frac{2-n}{2}\psi \left(
x^{k}\right) u+a_{0}u+b\left( x^{k}\right) \right) \partial _{u}
\end{equation}%
where $\xi ^{i}$ is a CKV with conformal factor $\psi \left( x^{k}\right) $,$%
~b\left( x^{k}\right) $ is a solution of (\ref{KG.Eq111})~and the following
condition is satisfied%
\begin{equation}
\xi ^{k}V_{,k}+2\psi V-\frac{2-n}{2}\Delta \psi =0.  \label{KG.Eq22}
\end{equation}%
b) for $n=2,$ the Lie symmetry vector is
\begin{equation}
X=\xi ^{i}\left( x^{k}\right) \partial _{i}+\left( a_{0}u+b\left(
x^{k}\right) \right) \partial _{u}
\end{equation}%
where $\xi ^{i}$ is a CKV with conformal factor $\psi \left( x^{k}\right) $,$%
~b\left( x^{k}\right) $ is a solution of (\ref{KG.Eq111}) and the following
condition is satisfied%
\begin{equation}
\xi ^{k}V_{,k}+2\psi V=0.
\end{equation}
\end{theorem}

By comparing the symmetry condition of the Klein Gordon equation (\ref%
{KG.Eq111}) and the classical Lagrangian (\ref{CLN.05}) and by taking into
consideration that for a special CKV/HV/KV the conformal factor satisfies
the condition$~\psi _{;ij}=0,$ we deduce the following result.

\begin{proposition}
\label{SymKGHJ}For $n>2,~$the Lie point symmetries of the Klein Gordon
equation for the metric $g_{ij}$ which defines the Laplace operator are
related to the Noether point symmetries of the classical Lagrangian for the
same metric and the same potential as follows

a) If a KV or HV of the metric $g_{ij}$ generates a Lie point symmetry of
the Klein Gordon equation (\ref{KG.Eq111}), then it also produces a Noether
point symmetry of the classical Lagrangian with gauge function a constant.

b) If a special CKV or a proper CKV satisfying the condition $\Delta \psi =0$
of the metric $g_{ij}$ generates a Lie point symmetry of the Klein Gordon
equation (\ref{KG.Eq111}), then it also generates a Noether point symmetry
of the conformally related\ Lagrangian if there exists a conformal factor $%
N\left( x^{k}\right) ,$ such that the CKV becomes a KV or a HV.
\end{proposition}

For $n=2,$ the results are different and are given below.

\subsection{Symmetries of the Conformal Klein Gordon equation}

The conformal Klein Gordon equation (or Yamabe Klein Gordon)
\begin{equation}
\bar{L}_{g}u-V\left( x^{k}\right) u=0.  \label{CKG.01A1}
\end{equation}%
is the conformal Poisson equation (\ref{Ct.00}) for $\bar{f}\left(
x^{i},u\right) =V\left( x^{k}\right) u.$ Therefore, Theorem \ref{CTheor} is
valid and we have the result.

\begin{theorem}
\label{CKG.01}The Lie point symmetries of the conformal Klein Gordon
equation (\ref{CKG.01A1}) are as follows

a) For $n>2,$ they$\ $are generated from the CKVs of the metric~$g_{ij}$ of
the conformal Laplace$~$operator, as follows%
\begin{equation}
X=\xi ^{i}\left( x^{k}\right) \partial _{i}+\left( \frac{2-n}{2}\psi \left(
x^{k}\right) u+a_{0}u+b\left( x^{k}\right) \right) \partial _{u}
\end{equation}%
where $\xi ^{i}$ is a CKV with conformal factor $\psi \left( x^{k}\right) ,~$%
and the following conditions are satisfied%
\begin{equation}
\xi ^{k}V_{,k}+2\psi V=0  \label{GGGG}
\end{equation}%
\begin{equation}
\bar{L}_{g}b-Vb=0.
\end{equation}%
b) For $n=2$, equation (\ref{CKG.01A1}) is the Laplace Klein Gordon equation
(\ref{KG.Eq111}) and the results of theorem \ref{KG} apply.
\end{theorem}

Comparing the symmetry condition of the conformal Klein Gordon equation (\ref%
{CKG.01A1}) and the classical Lagrangian (\ref{CLN.05}) we have the
following proposition.

\begin{proposition}
\label{CKGP}a) If a CKV of the metric $g_{ij}~$($\dim g_{ij}\succeq 2$),
which defines the conformal Laplace operator, produces a Lie point symmetry
of the conformal Klein Gordon equation (\ref{CKG.01A1}), then the same
vector generates a Noether point symmetry of the conformally related
Lagrangian provided there exists a conformal factor $N\left( x^{k}\right) $
such that the CKV becomes a KV/HV of $g_{ij}$.

b) If a KV/HV of the metric $g_{ij}$ generates a Lie point symmetry for the
conformal Klein Gordon equation (\ref{CKG.01A1}) then the same vector
generates a Noether point symmetry for the classical Lagrangian with gauge
function a constant.
\end{proposition}

\section{$sl\left( 2,R\right) $ and the Klein Gordon equation}

\label{Symmetries of the Lagrangian with non constant gauge function}In the
previous considerations, we have showed that the Lie point symmetries of the
Klein Gordon equation induce Noether point symmetries for the classical
Lagrangian if the gauge function is a constant. In this section, we
investigate the case when the induced Noether symmetry has a gauge function
which is not a constant. As we shall show in this case, the induced Noether
symmetry comes from a generalized Lie symmetry of the Klein Gordon equation.

It is clear that if a KV/HV produces a Noether point symmetry for the
classical Lagrangian satisfying conditions (\ref{NSCS.14}) with $d\neq 0~$or
(\ref{NSCS.17}) with $m\neq 0~$(or $d\neq 0$) of Theorem \ref{The Noether
symmetries of a conservative system}, then it does not produce a Lie
symmetry for the Klein Gordon equation. However, a gradient\ KV/HV which
generates a Noether point symmetry for the classical Lagrangian satisfies
only condition (\ref{NSCS.17}) with $d=0$ and $m\neq 0$ and leads to two
well known dynamical systems, the oscillator and the Ermakov system{\LARGE .}%
\ The Lie and the Noether point symmetries of these dynamical systems have
been considered previously in Chapter \ref{chapter4}; however, we briefly
reproduce these results in the current framework for competences.

\subsection{The oscillator}

First, we consider the case in which the metric admits a gradient KV which
generates Lie point symmetries of the classical Lagrangian, provided
condition (\ref{NSCS.17}) is satisfied. It is well known that if a metric
admits a gradient KV, then it is decomposable and can be written in the form%
\begin{equation}
ds^{2}=dx^{2}+h_{AB}dy^{A}dy^{B}  \label{CR.01}
\end{equation}%
where the gradient KV is $S^{,i}=\partial _{x}~\left( S=x\right) $ and $%
h_{AB}=h_{AB}\left( y^{C}\right) $ is the tensor projecting normal to the
KV. In these coordinates the Lagrangian takes the form%
\begin{equation}
L=\frac{1}{2}\left( \dot{x}^{2}+h_{AB}\dot{y}^{A}\dot{y}^{B}\right) -V\left(
x,y^{C}\right) .  \label{CR.02}
\end{equation}%
The Lie point symmetry condition for the gradient KV becomes~%
\begin{equation*}
V_{,x}+\mu ^{2}x=0~
\end{equation*}%
from which follows that the potential is
\begin{equation}
V\left( x,y^{C}\right) =-\frac{1}{2}\mu ^{2}x^{2}+F\left( y^{C}\right) .
\label{CR.03}
\end{equation}

The Noether point symmetries are the vectors $e^{\pm \mu t}\partial _{x}$
with respective gauge function $f\left( t,x,y^{A}\right) =\mu e^{\pm }x$.
The corresponding Noether integrals are%
\begin{equation}
I_{\pm }=e^{\pm \mu t}\dot{x}\mp \mu e^{\pm \mu t}x
\end{equation}%
It can be easily shown that the combined Noether integral $I_{0}=I_{+}I_{-}$
is time independent and equals%
\begin{equation}
I_{0}=\dot{x}^{2}-\mu ^{2}x^{2}.  \label{CR.04}
\end{equation}

The Laplace Klein Gordon equation defined by the metric (\ref{CR.01}) and
the potential (\ref{CR.03}) is
\begin{equation}
u_{xx}+h^{AB}u_{A}u_{B}-\Gamma ^{A}u_{A}-\mu ^{2}x^{2}u-F\left( y^{C}\right)
u=0.  \label{CR.05}
\end{equation}%
This equation does not admit a Lie point symmetry for general $h_{ab},$ $%
F\left( y^{C}\right) .$ However, it is separable with respect to $x$ in the
sense that the solution can be written in the form $u\left( x,y^{A}\right)
=w\left( x\right) S\left( y^{A}\right) .$ This implies that the operator $%
\hat{I}=D_{x}D_{x}-\mu ^{2}x^{2}-I_{0}$ satisfies $\hat{I}u=0,$ which means
that the Klein Gordon equation (\ref{CR.05}) possesses a Lie B\"{a}cklund
symmetry \cite{Miller,Kamr} with generating vector $\bar{X}=\left(
u_{xx}-\mu ^{2}x^{2}\right) \partial _{u}$.$~$

Concerning the conformal Klein Gordon equation
\begin{equation}
u_{xx}+h^{AB}u_{A}u_{B}-\Gamma ^{A}u_{A}+\frac{n-2}{4\left( n-1\right) }%
Ru-\mu ^{2}x^{2}u-2\bar{F}\left( y^{C}\right) u=0  \label{CR.06}
\end{equation}%
because for a KV, say $\xi ^{a},$ we have $L_{\xi }R=0$ \cite{HallSteele}
hence $R=R\left( y^{C}\right) ,$ equation (\ref{CR.06}) is written in the
form of the Laplace Klein Gordon equation with $F\left( y^{C}\right) =2\bar{F%
}\left( y^{C}\right) -\frac{n-2}{4\left( n-1\right) }R\left( y^{C}\right) $
and the previous result applies.

\subsection{The Kepler Ermakov potential with an oscillator term}

We assume now that there exists a gradient HV which produces a Noether point
symmetry for the classical Lagrangian under the constraint condition (\ref%
{NSCS.17}). It is well known \cite{Tupper1989,DD94} that if a metric admits
a gradient HV, then there exists a coordinate system in which the metric has
the form%
\begin{equation*}
ds^{2}=dr^{2}+r^{2}h_{AB}dy^{A}dy^{B}
\end{equation*}%
where the HV is $H^{,i}=r\partial _{r}~,$ $\left( \psi =1~,~H=\frac{1}{2}%
r^{2}\right) $ and $h_{AB}=h_{AB}\left( y^{C}\right) $ is the tensor
projecting normal to $H^{,i}.$s For these coordinates, the Lagrangian is%
\begin{equation}
L=\frac{1}{2}\left( \dot{r}^{2}+r^{2}h_{AB}\dot{y}^{A}\dot{y}^{B}\right)
-V\left( r,y^{C}\right)  \label{CR.07}
\end{equation}%
and\ the gradient HV generates Lie point symmetries only for the Ermakov
potential extended by the oscillator term, that is,%
\begin{equation}
V\left( r,y^{C}\right) =-\frac{1}{2}\mu ^{2}r^{2}+\frac{F\left( y^{C}\right)
}{r^{2}}.  \label{CR.08}
\end{equation}%
The admitted Noetheroint symmetries generated from the gradient HV$~H^{,i}~$%
are the vectors $X_{\pm }=\frac{1}{\mu }e^{\pm 2\mu t}\partial _{t}\pm
e^{\pm 2\mu t}r\partial _{r}$ with corresponding gauge functions $f\left(
t,r,y^{A}\right) =\mu e^{\pm 2\mu t}r^{2}$ and corresponding Noether
integrals (\ref{GERSN.5}) and (\ref{GERSN.6}). From the Noether integrals (%
\ref{GERSN.5}),(\ref{GERSN.6}) and the Hamiltonian $h$ of (\ref{CR.07}) we
construct the time independent first integral $\Phi _{0}=h^{2}-I_{+}I_{-},$
which is%
\begin{equation}
\Phi _{0}=r^{4}h_{DB}\dot{y}^{A}\dot{y}^{B}+2F\left( y^{C}\right) .
\label{Er22}
\end{equation}%
This is the well known Ermakov invariant \cite{Lewis1967,MoyoL}.

The Laplace Klein Gordon equation defined by the metric (\ref{CR.07}) and
the potential \ (\ref{CR.08}) is%
\begin{equation}
u_{rr}+\frac{1}{r^{2}}h^{AB}u_{AB}+\frac{n-1}{r}u_{r}-\frac{1}{r^{2}}\Gamma
^{A}u_{A}+\mu ^{2}r^{2}+\frac{2}{r^{2}}F\left( y^{C}\right) =0.
\label{CR.11}
\end{equation}%
This equation does not admit a Lie point symmetry. However it is separable
in the sense that $u\left( r,y^{C}\right) =w\left( r\right) S\left(
y^{C}\right) .$ Then the operator%
\begin{equation*}
\hat{\Phi}=h^{AB}D_{A}D_{B}-\Gamma ^{A}D_{A}+2F\left( y^{C}\right) -\Phi _{0}
\end{equation*}%
satisfies the equation $\hat{\Phi}u=0$ which means that (\ref{CR.11}) admits
the B\"{a}cklund \ symmetry with generator $\bar{X}=\left( \hat{\Phi}%
u\right) \partial _{u}~$\cite{Miller,Kamr}.

Concerning the conformal Klein Gordon equation, the Ricci scalar of the
metric (\ref{CR.07}) and the HV satisfy the condition~ $L_{H}R+2R=0$ \cite%
{HallSteele},~that is $R=\frac{1}{r^{2}}\bar{R}\left( y^{C}\right) .$ Then,
as in the case of the gradient KV we absorb the term $\bar{R}\left(
y^{C}\right) $ into the potential and we obtain the same results with the
Laplace Klein Gordon equation.

From the above, we conclude that, although in the two cases considered above
the Lie symmetries do not transfer from the classical to the "quantum" level
the generalized symmetries do transfer.

\section{Applications\label{Applications}}

We apply the previous general results in two practical cases. The first case
concerns the Newtonian central\ motion and the second the classification of
potentials in two and three dimensional flat spaces for which the Schr\"{o}%
dinger equation and the Klein Gordon equation admit Lie symmetries.

\subsection{Euclidean central force}

Consider the autonomous classical Lagrangian%
\begin{equation}
L=\frac{1}{2}\dot{r}^{2}+\frac{1}{2}r^{2}\dot{\theta}^{2}-Cr^{-\left(
2m+2\right) }.  \label{AP.01}
\end{equation}%
It is well known that (\ref{AP.01}) admits as Noether point symmetries (a)
the gradient KV $\partial _{t}~$( autonomous) with Noether integral the
Hamiltonian and (b) the KV $X_{N}=\partial _{\theta }$, with constant gauge
function and Noether integral the angular momentum $p_{\theta }=r^{2}\dot{%
\theta}=I_{0}$. Lagrangian (\ref{AP.01}) admits extra Noether symmetries for
the values $m=-1$ (the free particle) and $m=0$ (the Ermakov potential) \cite%
{Pinney}. In the following we assume $m\neq -1,0;$ hence, we do not expect
to find symmetries.

From the Lagrangian (\ref{AP.01}), we consider the kinematic metric and
define the Schr\"{o}dinger equation%
\begin{equation}
u_{rr}+\frac{1}{r^{2}}u_{\theta \theta }+\frac{1}{r}u_{r}+Cr^{-\left(
2m+2\right) }u-u_{t}=0  \label{AP.04}
\end{equation}%
and the Klein Gordon equation (because the dimension of the kinematic metric
is two, the Laplace and the Yamabe operators coincide)
\begin{equation}
u_{rr}+\frac{1}{r^{2}}u_{\theta \theta }+\frac{1}{r}u_{r}+Cr^{-\left(
2m+2\right) }u=0.  \label{AP.05}
\end{equation}

The application of the results of the previous sections give the following

a. The Schr\"{o}dinger equation (\ref{AP.04}) admits as Lie point symmetries
the vectors
\begin{equation*}
\partial _{t}~,~\partial _{\theta }~,~u\partial _{u}
\end{equation*}

b. The Klein Gordon equation (\ref{AP.05}) admits as Lie point symmetries
the vectors
\begin{eqnarray*}
&&\partial _{\theta }~,~u\partial _{u}~,~b\left( r,\theta \right) \partial
_{\theta } \\
X_{1} &=&r^{m+1}\cos \left( m\theta \right) \partial _{r}+r^{m}\sin \left(
m\theta \right) \partial _{\theta } \\
X_{2} &=&r^{m+1}\sin \left( m\theta \right) \partial _{r}-r^{m}\cos \left(
m\theta \right) \partial _{\theta }
\end{eqnarray*}%
It can be easily observed that $X_{1},X_{2}$ are proper CKVs of the two
dimensional flat kinematic metric.

Concerning the Schr\"{o}dinger equation, it has the same Lie point
symmetries as the Lagrangian (\ref{AP.01}) hence there is nothing more to
do. However, the Klein Gordon equation has the extra Lie symmetries $%
X_{1},X_{2};$ hence it is possible to apply the results of proposition \ref%
{CKGP} in order to find a conformally related Lagrangian which will admit
the pair of symmetries $\partial _{t},X_{1}$ or $\partial _{t},X_{2}$.

It is important to note that if we use the zero order invariants of the Lie
symmetry $\partial _{t}~$ to reduce\ the Schr\"{o}dinger equation (\ref%
{AP.04}), we find that the reduced equation is the Klein Gordon equation (%
\ref{AP.05}). Therefore, the symmetries $X_{1},X_{2}$\ are Type II hidden
symmetries \cite{AGL95,LeGovAb99,AGA06} for the Schr\"{o}dinger equation (%
\ref{AP.04}).

Let us consider the vector $X_{1}.$ It is easy to show, that the conformal
metric
\begin{equation*}
ds^{2}=N^{2}\left( r,\theta \right) \left( dr^{2}+r^{2}d\theta ^{2}\right)
\end{equation*}%
where $N\left( r,\theta \right) =r^{-\left( m+1\right) }g\left( r^{-1}\sin ^{%
\frac{1}{m}}\left( m\theta \right) \right) $ and $g$ is an arbitrary
function of $\ r^{-1}\sin ^{\frac{1}{m}}\left( m\theta \right) $ admits $%
X_{1}$ as a KV. This leads to the family of conformal Lagrangians
\begin{equation}
\bar{L}=r^{-2\left( m+1\right) }g^{2}\left( r^{-1}\sin ^{\frac{1}{m}}\left(
m\theta \right) \right) \left( \frac{1}{2}r^{\prime 2}+\frac{1}{2}%
r^{2}\theta ^{\prime 2}\right) -\frac{C}{g^{2}\left( r^{-1}\sin ^{\frac{1}{m}%
}\left( m\theta \right) \right) }  \label{AP.06}
\end{equation}%
where $"^{\prime }"$ means derivative with respect to the conformal
\textquotedblleft time" $\tau $ and the coordinate transformation is $%
dt=N^{-2}\left( r,\theta \right) d\tau $. According to proposition \ref{CKGP}%
, the vector field $X_{1}$ generates a Noether point symmetry for the
Lagrangian (\ref{AP.06}).

Working similarly for the vector $X_{2},$ we find \emph{another} family of
conformally related Lagrangians which admit $X_{2}$ as a Noether point
symmetry.

The conformally related Lagrangians are possible to admit additional Noether
point symmetries than the sets $\partial _{t},X_{1}$ or $\partial
_{t},X_{2}. $ For example in the case of $X_{1}$ we consider the conformally
related Lagrangian defined by the function $g=1,$ i.e.
\begin{equation*}
\bar{L}=r^{-2\left( m+1\right) }\left( \frac{1}{2}r^{\prime 2}+\frac{1}{2}%
r^{2}\theta ^{\prime 2}\right) -C
\end{equation*}%
which, by means of the coordinate transformation $r=R^{-\frac{1}{m}},$
becomes%
\begin{equation}
\bar{L}=\frac{1}{2}R^{\prime 2}+\frac{1}{2}R^{2}\theta ^{\prime 2}-C.
\label{AP.07}
\end{equation}%
This is the Lagrangian of the free particle moving in the 2D flat space.

\subsection{Lie symmetry classification of Schr\"{o}dinger and the Klein
Gordon equations in Euclidian space}

In Chapter \ref{chapter3}, all two and three dimensional \ potentials for
which the corresponding Newtonian dynamical systems admit Lie and/or Noether
point were determined. Using these results, we determine all Schr\"{o}dinger
and Klein Gordon equations in Euclidian 2D and 3D space which admit Lie
point symmetries.

The Schr\"{o}dinger equation in Euclidian space is%
\begin{equation}
\delta ^{ij}u_{ij}+V\left( x^{k}\right) u=u_{t}.  \label{SC1.01}
\end{equation}%
From proposition \ref{SymKGHJ11}, we have that the potentials for which the
Schr\"{o}dinger equation (\ref{SC1.01}) admits Lie symmetries are the same
with the ones admitted by the classical Lagrangian. Therefore, the results
of Chapter \ref{chapter3} apply directly and give all potentials for which
the Schr\"{o}dinger equation (\ref{SC1.01}) admits at least one Lie symmetry.

We consider the Klein Gordon equation in flat space;that is,%
\begin{equation}
\delta ^{ij}u_{ij}+V\left( x^{k}\right) u=0.  \label{LC.01}
\end{equation}%
In this case, the conditions are different and we find that equation (\ref%
{LC.01}) admits a Lie point symmetry due to a HV/KV for the following
potentials taken from the corresponding Tables of Chapter \ref{chapter3}. In
Table \ref{2dKGe} and Table \ref{3dKGe} we provide the potentials where the
2D and 3D Klein Gordon equation (\ref{LC.01}) admits Lie point symmetries
generated from the elements of the homothetic group of the Euclidian space.

\begin{table}[tbp] \centering%
\caption{The 2D Klein Gordon admitting Lie symmetries from the homothetic
group}%
\begin{tabular}{cccc}
\hline\hline
\textbf{Lie Symmetry} & $\mathbf{V}\left( x,y\right) $ & \textbf{Lie Symmetry%
} & $\mathbf{V}\left( x,y\right) $ \\ \hline
\multicolumn{1}{l}{$\partial _{x}$} & \multicolumn{1}{l}{$f\left( y\right) $}
& \multicolumn{1}{l}{$\partial _{x}+b\partial _{y}$} & \multicolumn{1}{l}{$%
f\left( y-bx\right) $} \\
\multicolumn{1}{l}{$\partial _{y}$} & \multicolumn{1}{l}{$f\left( x\right) $}
& \multicolumn{1}{l}{$\left( a+y\right) \partial _{x}+\left( b-x\right)
\partial _{y}$} & \multicolumn{1}{l}{$f\left( \frac{1}{2}\left(
x^{2}+y^{2}\right) +ay-bx\right) $} \\
\multicolumn{1}{l}{$y\partial _{x}-x\partial _{y}$} & \multicolumn{1}{l}{$%
f\left( r\right) $} & \multicolumn{1}{l}{$\left( x+ay\right) \partial
_{x}+\left( y-ax\right) \partial _{y}$} & \multicolumn{1}{l}{$r^{-2}~f\left(
\theta -a\ln r\right) $} \\
\multicolumn{1}{l}{$x\partial _{x}+y\partial _{y}$} & \multicolumn{1}{l}{$%
x^{-2}f\left( \frac{y}{x}\right) $} & \multicolumn{1}{l}{$\left( a+x\right)
\partial _{x}+\left( b+y\right) \partial _{y}$} & \multicolumn{1}{l}{$%
f\left( \frac{b+x}{a+x}\right) \left( a+x\right) ^{-2}$} \\ \hline\hline
\end{tabular}%
\label{2dKGe}%
\end{table}%

\begin{table}[tbp] \centering%
\caption{The 3D Klein Gordon admitting Lie symmetries from the homothetic
group}%
\begin{tabular}{cc}
\hline\hline
\textbf{Lie Symmetry} & $\mathbf{V(x,y,z)}$ \\ \hline
$a\partial _{\mu }+b\partial _{\nu }+c\partial _{\sigma }$ & $f\left( x^{\nu
}-\frac{b}{a}x^{\mu },x^{\sigma }-\frac{c}{a}x^{\mu }\right) $ \\
$a\partial _{\mu }+b\partial _{\nu }+c\left( x_{\nu }\partial _{\mu }-x_{\mu
}\partial _{\nu }\right) $ & ~$+f\left( \frac{c}{2}r_{\left( \mu \nu \right)
}-bx_{\mu }+ax_{\nu },x_{\sigma }\right) $ \\
$a\partial _{\mu }+b\partial _{\nu }+c\left( x_{\sigma }\partial _{\mu
}-x_{\mu }\partial _{\sigma }\right) $ & $+f\left( x_{\nu }-\frac{1}{%
\left\vert c\right\vert }\arctan \left( \frac{\left\vert c\right\vert x_{\mu
}}{\left\vert a+cx_{\sigma }\right\vert }\right) ,\frac{1}{2}r_{\left( \mu
\sigma \right) }-\frac{a}{c}x_{\sigma }\right) $ \\
$a\partial _{\mu }+b\left( x_{\nu }\partial _{\mu }-x_{\mu }\partial _{\nu
}\right) +$ &  \\
$~~+c\left( x_{\sigma }\partial _{\mu }-x_{\mu }\partial _{\sigma }\right) $
& $+f\left( x_{\mu }^{2}+x_{\nu }^{2}\left( 1-\frac{c^{2}}{b^{2}}\right)
+\left( \frac{2a}{b}+\frac{2c}{b}x_{\sigma }\right) x_{\nu },x_{\sigma }-%
\frac{c}{b}x_{\nu }\right) $ \\
$so\left( 3\right) $ linear combination & $~F\left( R,b\tan \theta \sin \phi
+c\cos \phi -aM_{1}\right) $ \\
$a\partial _{\mu }+b\theta _{\left( \nu \sigma \right) }\partial _{\theta
_{\left( \nu \sigma \right) }}+cR\partial _{R}$ & $\frac{1}{r_{\left( \nu
\sigma \right) }^{2}}f\left( \theta _{\left( \nu \sigma \right) }-\frac{b}{c}%
\ln r_{\left( \nu \sigma \right) },\frac{a+cx_{\mu }}{cr_{\left( \nu \sigma
\right) }}\right) $ \\
$\left( a\partial _{\mu }+b\partial _{\nu }+c\partial _{\sigma }+lR\partial
_{R}\right) $ & $\frac{1}{\left( a+lx_{\mu }\right) ^{2}}f\left( \frac{%
b+lx_{\nu }}{l\left( a+lx_{\mu }\right) },\frac{c+lx_{\sigma }}{l\left(
a+lx_{\mu }\right) }\right) $ \\ \hline\hline
\end{tabular}%
\label{3dKGe}%
\end{table}%

As we have seen in section \ref{Lie point symmetries of the Klein Gordon
equation2}, the Lie point symmetries of the Klein Gordon equation are
generated from the conformal group of the space; therefore, we have to
consider the admitted CKVs in addition to the HV\ and the KVs.

\subsubsection{The two dimensional case}

We recall that the conformal algebra of a two dimensional space is infinite
dimensional \cite{Bela} and in coordinates with line element $ds^{2}=2dwdz$
are given by the vectors $X=F\left( w\right) \partial _{w}+G\left( z\right)
\partial _{z},$ with conformal factor $\psi =\frac{1}{2}\left(
F_{,w}+G_{,z}\right) $. In the coordinates $\left( w,z\right) ,$ the 2D
Klein Gordon Klein Gordon equation (\ref{LC.01}) is%
\begin{equation*}
u_{wz}+V\left( w,z\right) u=0.
\end{equation*}%
The Lie symmetry condition (\ref{KG.Eq22}) becomes%
\begin{equation}
\left( FV\right) _{,w}+\left( GV\right) _{,z}=0
\end{equation}%
from which follows that there are infinite many potentials for which the 2D
Klein Gordon equation admits Lie point symmetries.

\subsubsection{The three dimensional case}

The 3D Euclidian space admits the three special CKVs%
\begin{equation*}
K_{C}^{\mu }=\frac{1}{2}\left( \left( x^{\mu }\right) ^{2}-\left( \left(
x^{\sigma }\right) ^{2}+\left( x^{\nu }\right) ^{2}\right) \right) \partial
_{i}+x^{\mu }x^{\nu }\partial _{x}+x^{\mu }x^{\nu }\partial _{z}~,~\mu
,\sigma ,\nu =1,2,3
\end{equation*}%
with corresponding conformal factor $\psi _{C}=x^{\mu }$.

From the symmetry condition (\ref{KG.Eq22}) of theorem \ref{KG}, it follows
that a special CKV generates the Lie point symmetry $X=K_{C}^{\mu }-\frac{1}{%
2}x^{\mu }u\partial _{u}$ for the 3D Klein Gordon (\ref{LC.01}) only for the
potential
\begin{equation*}
V\left( x,y,z\right) =\frac{1}{\left( x^{\sigma }\right) ^{2}}F\left( \frac{%
x^{\nu }}{x^{\sigma }},\frac{\delta _{\kappa \lambda }x^{\kappa }x^{\lambda }%
}{x^{\sigma }}\right) .
\end{equation*}

One is possible to continue with the linear combinations and deduce all
cases that the 3D Klein Gordon equation admits a Lie symmetry. These results
hold for both the Klein Gordon and the conformal Klein Gordon equation. It
particular one can show that the results remain still valid for the
conformal Klein Gordon equation provided the metric defining the conformal
Laplacian is conformally flat.

\section{The Klein Gordon equation in a spherically symmetric space-time
\label{StaticSS}}

In this section, we consider the Lie point symmetries of the Klein Gordon
equation in a non-flat space and in particular in the static spherically
symmetric empty space-time; that is the exterior Schwarzschild solution
given by the metric ($\tau $ is the radial coordinate)
\begin{equation}
ds^{2}=-a^{2}\left( \tau \right) dt^{2}+d\tau ^{2}+b^{2}\left( \tau \right)
\left( d\theta ^{2}+\sin ^{2}\theta d\phi ^{2}\right) .  \label{SC}
\end{equation}%
The Lagrangian of Einstein field equations for this space-time is \cite%
{VakB,Christ}%
\begin{equation}
L=2ab^{\prime 2}+4ba^{\prime }b^{\prime }+2a  \label{SC.01}
\end{equation}%
where $"^{\prime }"~$means derivative with respect to the radius $\tau $. If
we see the Lagrangian (\ref{SC.01}) as a dynamical system in the space of
variables $\left\{ a,b\right\} $, then this system is "autonomous"; hence,
admits the Noether symmetry $\partial _{\tau }$ with corresponding Noether
integral the "Hamiltonian" $h=$constant. ( $h$ \emph{is not} the "energy"
because the coordinate is the radial distance not the time)%
\begin{equation*}
h=2ab^{\prime 2}+4ba^{\prime }b^{\prime }-2a.
\end{equation*}%
\

It can be shown directly that $h=\frac{2}{a}G_{1}^{1},$ where $G_{1}^{1}$ is
the Einstein tensor. Because the space is empty, from Einstein's equations
follows that $h=0.$ The Euler-Lagrange equations are%
\begin{equation*}
a^{\prime \prime }-\frac{a}{2b^{2}}b^{\prime 2}+\frac{1}{b}a^{\prime
}b^{\prime }+\frac{1}{2}\frac{a}{b^{2}}=0
\end{equation*}%
\begin{equation*}
b^{\prime \prime }+\frac{1}{2b}b^{\prime 2}-\frac{1}{2b}=0.
\end{equation*}%
We end up with a system of three equations whose solution will give the
functions $a(\tau ),b(\tau ).$ It is found that the solution of the system
does not give these functions in the well known closed form. This is due to
the Lagrangian we have considered; Indeed we shall show below that it is
possible to find the solution in closed form by considering a Lagrangian
conformally related to the Lagrangian (\ref{SC.01}).

Applying Theorem \ref{The Noether symmetries of a conservative system}, we
find that the Lagrangian (\ref{SC.01}) admits the Noether point symmetry
\begin{equation*}
X_{1}=2\tau \partial _{\tau }+H
\end{equation*}%
where $H=-2a\partial _{a}+2b\partial _{b}$ is a non-gradient homothetic
vector of the two dimensional kinetic metric
\begin{equation}
d\bar{s}^{2}=2adb^{2}+4bdadb  \label{SC.01a}
\end{equation}%
defined from the Lagrangian (\ref{SC.01}).

The Klein Gordon equation defined by the metric (\ref{SC.01a}) with
potential $V\left( a,b\right) =2a$, is%
\begin{equation}
-\frac{a}{4b^{2}}u_{aa}+\frac{1}{2b}u_{ab}-\frac{1}{4b^{2}}u_{a}-2au=0
\label{SC.02}
\end{equation}%
and admits as Lie symmetries the vectors \cite{Christ}%
\begin{equation*}
u\partial _{u}~~,~b\left( a,b\right) \partial _{u}
\end{equation*}%
\begin{equation*}
H=-2a\partial _{a}+2b\partial _{b}~~,~~X_{2}=\frac{1}{ab}\partial
_{a}~~,~X_{3}=\frac{a}{2b}\partial _{a}-\partial _{b}
\end{equation*}%
where the vectors $X_{2},X_{3}$ are proper CKVs of the two dimensional
metric (\ref{SC.01a}).

It is possible to find solutions of the Klein Gordon equation (\ref{SC.02})
which are invariant with respect to one of the admitted Lie symmetries.

For example, let us consider the Lie point symmetry $H_{u}=H-2cu\partial
_{u}.$ The zero order invariants of $H_{u}$ are ~$w=ab~,~u=a^{c}\Phi \left(
w\right) .$ Replacing in the PDE we find the solution \footnote{%
We have found this by making use of the library SADE \cite{SADE} of MAPLE.}

\begin{equation*}
u\left( a,b\right) =a^{c}\left[ c_{1}I_{c}^{B}\left( 2\sqrt{2}ab\right)
+c_{2}K_{c}^{B}\left( 2\sqrt{2}ab\right) \right]
\end{equation*}%
where $I^{B},K^{B}$ are the Bessel modified functions \cite{Christ}. Working
similarly for the Lie point symmetry $H+eX_{2}-cu\partial _{u}$ we find the
solution%
\begin{equation*}
u\left( a,b\right) =\left( a^{2}-eb^{-1}\right) ^{\frac{c}{4}}\left[
c_{1}I_{-\frac{c}{2}}\left( 2\sqrt{2b\left( a^{2}b-e\right) }\right)
+c_{2}K_{\frac{c}{2}}^{B}\left( 2\sqrt{2b\left( a^{2}b-e\right) }\right) %
\right] .
\end{equation*}%
One can find more solutions using linear combinations of the Lie symmetries.

Following proposition \ref{CKGP}, we look for a conformal metric for which
one of the CKVs $X_{2},X_{3}$ becomes a KV\ and write the corresponding
conformally related Lagrangian which admits this CKV\ as a Noether point
symmetry.

We consider the vector $X_{2}$ and the conformally related metric
\begin{equation}
ds^{2}=N^{2}\left( 4adb^{2}+8bda~db\right)  \label{M1}
\end{equation}%
where $N\left( a,b\right) =g\left( b\right) \sqrt{a}$ and $g\left( b\right) $
is an arbitrary function of its argument$.$ It is easy to show that the
vector $X_{2}$ is a KV of this metric hence a Noether symmetry for the
family of conformally related Lagrangians
\begin{equation}
\bar{L}=N^{2}\left[ 2a\left( \frac{db}{dr}\right) ^{2}+4b\left( \frac{da}{dr}%
\right) \left( \frac{db}{dr}\right) \right] +\frac{2a}{N^{2}}  \label{M2}
\end{equation}%
where we have considered the coordinate transformation $d\tau =\frac{dr}{%
N^{2}\left( a,b\right) }$. The Noether function for this Noether symmetry is
the "Hamiltonian" of the Lagrangian (\ref{M2}). This constant and the two
Lagrange equations for the "generalized" coordinates $a,b$ provide a system
of differential equations which will give the functions $a(\tau ),b(\tau ).$

We consider $g\left( b\right) =g_{0}=$constant, that is from the family of
Lagrangians (\ref{M2}) we take the Lagrangian%
\begin{equation}
L=g_{0}^{2}\left[ 2a^{2}\left( \frac{db}{dr}\right) ^{2}+4ab\left( \frac{da}{%
dr}\right) \left( \frac{db}{dr}\right) \right] +\frac{2}{g_{0}^{2}}.
\label{M2A}
\end{equation}%
For this Lagrangian we have the following system of equations:\newline
a. ~The "Hamiltonian" of the Lagrangian (\ref{M2A})
\begin{equation}
a^{2}\left( \frac{db}{dr}\right) ^{2}+2b^{2}\left( \frac{da}{dr}\right)
\left( \frac{db}{dr}\right) -V_{0}=0  \label{M3}
\end{equation}%
b. $~$The Euler Lagrange equations of (\ref{M2A}) with respect to the
variables $a,b~$ \
\begin{equation}
\frac{d^{2}a}{dr^{2}}+\frac{1}{a^{2}}\left( \frac{da}{dr}\right) ^{2}+\frac{2%
}{b}\left( \frac{da}{dr}\right) \left( \frac{db}{dr}\right) =0  \label{M4}
\end{equation}

\begin{equation}
\frac{d^{2}b}{dr^{2}}=0  \label{M5}
\end{equation}%
where we have set $V_{0}=g_{0}^{-4}$. The solution of the system of
equations (\ref{M3})-(\ref{M5}) is
\begin{equation*}
b\left( r\right) =b_{1}r+b_{2}~~,~~~a^{2}\left( r\right) =\frac{%
V_{0}r+2a_{1}b_{1}}{2b_{1}\left( b_{1}r+b_{2}\right) }.
\end{equation*}%
Under the linear transformation $b_{1}r=b_{1}R-\frac{b_{2}}{b_{1}},$ and if
we set $V_{0}=2\left( b_{1}\right) ^{2}~,~a_{1}=-2m+b_{2}$, we obtain the
exterior Schwarzschild solution in the standard coordinates%
\begin{equation}
ds^{2}=-\left( 1-\frac{2m}{R}\right) dt^{2}+\left( 1-\frac{2m}{R}\right)
^{-1}dR^{2}+R^{2}\left( d\theta ^{2}+\sin ^{2}\theta d\phi ^{2}\right)
\label{M6}
\end{equation}

The choice of the function $g\left( b\right) $ is essentially a choice of
the coordinate system. Obviously the final solution must always be the
exterior Schwarzschild solution. In order to show this let us consider $%
g\left( b\right) =\sqrt{b}$ so that $d\tau =\left( a\left( r\right) b\left(
r\right) \right) ^{-1}dr.$ Then, we get the Lagrangian%
\begin{equation*}
\bar{L}=2a^{2}b\left( \frac{db}{dr}\right) ^{2}+4ab^{2}\left( \frac{da}{dr}%
\right) \left( \frac{db}{dr}\right) +\frac{2}{b}
\end{equation*}%
and the system of equations%
\begin{equation*}
\frac{d^{2}b}{dr^{2}}+\frac{1}{b}\left( \frac{db}{dr}\right) ^{2}=0
\end{equation*}%
\begin{equation*}
2a^{2}b\left( \frac{db}{dr}\right) ^{2}+4ab^{2}\left( \frac{da}{dr}\right)
\left( \frac{db}{dr}\right) -\frac{2}{b}=0
\end{equation*}%
\begin{equation*}
\frac{d^{2}a}{dr^{2}}+\frac{1}{a}\left( \frac{da}{dr}\right) +\frac{2}{b}%
\left( \frac{da}{dr}\right) \left( \frac{db}{dr}\right) -\frac{a}{2b^{2}}%
\left( \frac{db}{dr}\right) ^{2}+\frac{1}{2ab^{4}}=0
\end{equation*}%
The solution of the system is~%
\begin{equation*}
b^{2}\left( r\right) =r~,~a^{2}\left( r\right) =\left( \sqrt{r}\right)
^{-1}\left( 4\sqrt{r}+a_{1}\right)
\end{equation*}%
from which follows: $d\tau ^{2}=\left( 4r+\sqrt{r}a_{1}\right) ^{-1}dr^{2}.$
Therefore, for this choice of Lagrangian, the metric is%
\begin{equation}
ds^{2}=-\left( \frac{4\sqrt{r}+a_{1}}{\sqrt{r}}\right) dt^{2}+\frac{1}{4r+%
\sqrt{r}a_{1}}dr^{2}+r\left( d\theta +\sin ^{2}\phi d\theta ^{2}\right) .
\label{M8}
\end{equation}%
If we make the transformation $r=R^{2}$ , $dt\rightarrow \frac{1}{2}dt$ and $%
a_{1}=-8m,$ we retain the metric (\ref{M8}) in the standard form (\ref{M1}).

Working similarly, we find that $X_{3}$ becomes a KV for the conformal
metric (\ref{M1}) if $N_{3}\left( a,b\right) =f\left( a^{2}b\right) \sqrt{a}$
and generates a Noether point symmetry for the conformal Lagrangian (\ref{M2}%
) (with $N_{3}$ in the place of $N)\ $\ and continue as above.

\section{WKB approximation\label{WKBAp}}

In WKB approximation, we search for solutions of the Klein\ Gordon equation
of the form $u=A_{n}e^{iS\left( x^{k}\right) },$ where $S\left( x^{k}\right)
$ has to satisfy the null Hamilton Jacobi equation \cite{CapHam,Kim}.%
\begin{equation}
g^{ij}S_{,i}S_{,j}+\bar{V}\left( x^{k}\right) =0  \label{nHJ.01}
\end{equation}%
where $g^{ij}$ is the metric defining the Yamabe operator and $\bar{V}\left(
x\right) =\frac{n-2}{4\left( n-1\right) }R-V\left( x^{k}\right) $. We study
the symmetries of the PDE (\ref{nHJ.01}).

We search for Lie point symmetries of the form \cite{StephaniB}%
\begin{equation*}
X=\xi ^{i}\left( x^{i},S\right) \partial _{i}+\eta \left( x^{i},S\right)
\partial _{S}.
\end{equation*}%
The symmetry condition is
\begin{equation*}
X^{\left[ 1\right] }\left( H_{null}\right) =\lambda \left( x^{k}\right)
\left( H_{null}\right) ~,~mod\left[ H_{null}\right] =0~
\end{equation*}%
where $X^{\left[ 1\right] }$ is the first prolongation and $\lambda =\lambda
\left( x^{i}\right) $. Replacing the first prolongation $X^{\left[ 1\right]
} $ in the symmetry condition and collecting terms of powers of $S_{,i}$ we
find the following result.

\begin{theorem}
\label{nullHJ} The Lie point symmetries of the null Hamilton Jacobi equation
(\ref{nHJ.01}) are the vectors
\begin{equation}
X=\xi ^{i}\left( x^{k}\right) \partial _{i}+\left( a_{0}S+b_{0}\right)
\partial _{S}  \label{nHJ.02}
\end{equation}%
where $\xi ^{i}\left( x^{k}\right) $ is a CKV of the metric $g_{ij},$ $%
a_{0},b_{0}$ are constants and the following condition holds%
\begin{equation}
V_{k}\xi ^{k}+2\psi V-a_{0}V=0.  \label{nHJ.03}
\end{equation}
\end{theorem}

\begin{proof}[Proof ]
The symmetry condition is
\begin{equation*}
X^{\left[ 1\right] }\left( H_{null}\right) =\lambda \left( H_{null}\right) ~,%
\text{ \ }mod\left[ H_{null}\right] =0~
\end{equation*}%
where $X^{\left[ 1\right] }$ is the first prolongation and we set $\lambda
=\lambda \left( x^{i}\right) $. Replacing $X^{\left[ 1\right] },$ we find
\begin{equation*}
X^{\left[ 1\right] }H_{null}=\frac{1}{2}g_{,k}^{ij}\xi
^{k}S_{i}S_{j}+g^{ij}S_{j}\eta _{i}^{\left[ 1\right] }+V_{k}\xi ^{k}
\end{equation*}%
where%
\begin{equation*}
g^{ij}u_{j}\eta _{i}^{\left[ 1\right] }=g^{ij}S_{j}\eta
_{,i}+g^{ij}S_{j}S_{i}\eta _{s}-g^{rj}S_{i}S_{j}\xi
_{,r}^{j}-g^{ir}S_{r}S_{i}S_{j}\xi _{,u}^{j}.
\end{equation*}%
We can easily show that $\xi _{,u}^{i}=0$ so that the the symmetry
conditions become:%
\begin{eqnarray}
g_{,k}^{ij}\xi ^{k}-2g^{r(j}\xi _{,r}^{j)}+2g^{ij}\eta _{S} &=&\lambda g^{ij}
\label{P1.1} \\
V_{k}\xi ^{k} &=&\lambda V  \label{P1.2} \\
\eta _{,i} &=&0.  \label{P1.3}
\end{eqnarray}%
From (\ref{P1.1}) and (\ref{P1.3}) we have that $\eta =a_{0}S+b_{0}$. Then,
the conditions (\ref{P1.1}),(\ref{P1.2}) become%
\begin{eqnarray}
L_{\xi }g_{ij} &=&\left( a_{0}-\lambda \right) g_{ij}  \label{P1.4} \\
V_{k}\xi ^{k} &=&\lambda V.  \label{P1.5}
\end{eqnarray}

Setting $a_{0}-\lambda =2\psi $ or $\lambda =a_{0}-2\psi ,$ we see that
condition (\ref{P1.4}) implies that the Lie point symmetries of the null
Hamilton Jacobi equation (\ref{nHJ.01}) are generated from the CKVs of the
the kinematic metric $g_{ij}$. Finally, condition (\ref{P1.5}) becomes%
\begin{equation*}
V_{k}\xi ^{k}+2\psi V-a_{0}V=0.
\end{equation*}
\end{proof}

Comparing the symmetry condition\ (\ref{nHJ.03}) of the null Hamilton Jacobi
equation and the symmetry condition (\ref{GGGG}) of the Yamabe Klein Gordon
equation, we have

\begin{proposition}
\label{Null HJ and Yamabe KG}If a CKV of the metric which defines the Yamabe
operator generates a Lie point symmetry for the Yamabe Klein Gordon equation
(\ref{CKG.01A1}), then the same vector generates a Lie point symmetry of the
null Hamilton Jacobi equation (\ref{nHJ.01}). The reverse holds if the CKV
generating Lie symmetry for the null Hamilton Jacobi equation (\ref{nHJ.01})
satisfies condition (\ref{nHJ.03}) with $a_{0}=0$.
\end{proposition}

Proposition \ref{Null HJ and Yamabe KG} relates the Lie point symmetries of
the Hamilton Jacobi equation with the Lie point symmetries of the Yamabe
Klein Gordon when the constant $a_{0}=0$. The question which arises is what
happens to the Lie symmetries when $a_{0}\neq 0$. The answer is given in the
following proposition

\begin{proposition}
\label{PP10}If a CKV generates Lie point symmetry for the null Hamilton
Jacobi equation (\ref{nHJ.01}) satisfying condition (\ref{nHJ.03}) with $%
a_{0}\neq 0$, then this CKV produces a Lie point symmetry for the
Euler-Lagrange equations for a conformally related Lagrangian \ if there
exist a conformal factor $N\left( x^{k}\right) $ for which the CKV becomes
KV or HV.
\end{proposition}

\begin{proof}[Proof ]
The Lie point symmetries of the autonomous system
\begin{equation}
\ddot{x}^{i}+\Gamma _{jk}^{i}\dot{x}^{i}\dot{x}^{j}+V^{,i}=0  \label{nHJ.04}
\end{equation}%
are produced from the special projective algebra of the metric provided the
potential satisfies the condition
\begin{equation}
L_{\eta }V^{,i}+d_{0}V^{,i}=0
\end{equation}%
where $d_{0}$ is a constant and $L_{\eta }$ denotes Lie derivative with
respect to $\eta ^{i}$. Equation (\ref{nHJ.04}) remains the same for the
conformally related Lagrangian, that is, we have
\begin{equation*}
x^{\prime \prime i}+\bar{\Gamma}_{jk}^{i}x^{\prime i}x^{\prime j}+V^{\prime
,i}=0.
\end{equation*}%
Therefore, the symmetry group will be again the special projective group of
the conformal metric. The special projective group is not preserved under a
conformal transformation but the subgroups of KVs and the homothetic group
are preserved. Then these two subgroups are common for both metrics.
Subsequetly, if the CKV of the metric $g_{ij}$ becomes a KV/HV of the
conformal metric $\bar{g}_{ij}=N^{2}g_{ij}$ then it will be a Lie point
symmetry of the Euler Lagrange equations of the conformally related
Lagrangian $\bar{L}$. This symmetry will not be a Noether symmetry of $\bar{L%
}$ except in the case that $d_{0}=2\bar{\psi}~\ \left( \bar{\psi}=1~\text{%
for HV and }\psi =0\text{ for KV}\right) .$
\end{proof}

\section{Conclusion\label{conclusion}}

We have determined the Lie point symmetries of Schr\"{o}dinger equation and
the Klein Gordon equation in a general Riemannian space. It has been shown
that these symmetries are related to the homothetic algebra and the
conformal algebra of the metric. Furthermore, these symmetries have been
related to the Noether point symmetries of the classical Lagrangian for
which the metric $g_{ij}$ is the kinematic metric. More precisely, for the
Schr\"{o}dinger equation (\ref{HC.01}) it has been shown that if a KV/HV of
the metric $g_{ij}$ produces a Lie point symmetry of the Schr\"{o}dinger
equation, then it produces a Noether point symmetry for the Classical
Lagrangian in the space with metric $g_{ij}$ and potential $V(x^{k})$. For
the Klein Gordon equation the situation is different; the Lie point
symmetries of the Klein Gordon are generated by elements of the conformal
group of the metric $g_{ij}.$ The KVs and the HV of this group produce a
Noether symmetry of the classical Lagrangian with a constant gauge function.
However the proper CKVs produce a Noether point symmetry for the conformal
Lagrangian if there exists a conformal factor $N\left( x^{k}\right) $ such
that the CKV becomes a KV/HV of $g_{ij}$.

{We have applied these general results to three cases of practical interest:
the motion in a central potential, the classification of all potentials in
Euclidian 2D and 3D space for which the Schr\"{o}dinger equation and the
Klein Gordon equation admit a Lie point symmetry and finally we have
considered the Lie symmetries of the Klein Gordon equations in the static,
spherically symmetric empty spacetime. In the last case, we have
demonstrated the role of Lie symmetries and that of the conformal
Lagrangians in the determination of the closed form solution of Einstein
equations. }

{Furthermore, we investigated the Lie symmetries of the null Hamilton Jacobi
equation and we proved that if a CKV generates a point symmetry for the
Klein Gordon equation, then it also generates a Lie point symmetry for the
null Hamilton Jacobi equation.}

The knowledge of the Lie symmetries of the Schr\"{o}dinger equation and the
Klein Gordon equation in a general Riemannian space makes possible the
determination of solutions of these equations which are invariant under a
given Lie symmetry. In addition, they can be used in Quantum Cosmology \cite%
{CAP94,CKP,VakF,VF12,CapHam} to determine the form of solutions of the
Wheeler--DeWitt equation \cite{DeW} in a given Riemannian space.

\chapter{The geometric origin of Type II hidden symmetries\label{chapter7}}

\section{Introduction}

Lie symmetries assist us in the simplification of differential equations
(DEs) by means of reduction. As it was indicated in Chapter \ref{chapter1},
the reduction is different for ordinary differential equations (ODEs) and
partial differential equations (PDEs). In the case of ODEs, the use of a Lie
symmetry reduces the order of ODE by one while in the case of PDEs, the
reduction by a Lie symmetry reduces by one the number of independent and
dependent variables, but not the order of the PDE. A common characteristic
in the reduction of both cases is that the Lie symmetry which is used for
the reduction is not admitted as such by the reduced DE, it is "lost".

It has been found that the reduced equation is possible to admit more Lie
symmetries than the original equation. These new Lie symmetries have been
termed Type II hidden symmetries. Also,\ if one works in the reverse way and
either increases the order of an ODE\ or increases the number of independent
and dependent variables of a PDE, then, it is possible that the new (the
`augmented') DE admits new point symmetries not admitted by the original DE.
This type of Lie symmetries are called Type I\ hidden symmetries.

The Type I\ and Type II hidden symmetries have studied extensively in the
recent years by various authors (see e.g. \cite%
{AbrahamGuo1994,AGL95,LeGovAb99,AGA06,ASG,Abraham2007,Moitsheko04}). In the
following sections, we consider mainly the Type II\ hidden symmetries as
they are the ones which could be used to reduce further the reduced DE.

The origin of Type II hidden symmetries is different for the ODEs and the
PDEs, although it has been shown recently that they are nearly the same \cite%
{LeachGovinderAndriopoulos2012}. For the case of ODEs, the inheritance or
not of a Lie symmetry, the $X_{2}$ say, by the reduced ODE\ depends on the
commutator of that symmetry with the symmetry used for the reduction, the $%
X_{1}$ say. For example, if only two Lie symmetries $X_{1},X_{2}$ are
admitted by the original equation and the commutator $[X_{1},X_{2}]=cX_{2}$
where $c$ may be zero,\ then reduction by $X_{1}$ results in $X_{2}$ being a
nonlocal symmetry for the reduced ODE while reduction by $X_{2}$ results in $%
X_{1}$ being an inherited Lie symmetry of the reduced ODE. In the reduction
by $X_{1},$ the symmetry $X_{2}$ is a Type I\ hidden symmetry of the
original equation relative to the reduced equation. In the case of more than
two Lie symmetries the situation is the same, if the Lie bracket gives a
third Lie symmetry, the $X_{3}$ say. Then, the point like nature of a
symmetry is preserved only if reduction is performed using the normal
subgroup \ and $X_{3}$ has a certain expression \cite{Govinder2001}.

The above scenario is transferred to PDEs as follows. The reduced PDE loses
the symmetry used to reduce the number of variables and it may lose other
Lie symmetries depending on the structure of the associated Lie algebra
depending if the admitted subgroup is normal or not \cite{Govinder2001}.
Similarly, if $X_{1},X_{2}$ are Lie symmetries of the original PDE with
commutator $[X_{1},X_{2}]=cX_{2}$ where $c$ may be zero,\ then reduction by $%
X_{2}$ results in $X_{1}$ being a symmetry of the reduced PDE while
reduction by $X_{1}$ results in an expression which has no relevance for the
PDE \cite{Govinder2001}.

In addition to that scenario,\ B. Abraham - Shrauner and K.S. Govinder have
proposed a new potential source for the Type II hidden symmetries \cite%
{GovinderAbraham2008} based on the observation that different PDEs with the
same variables, which admit different Lie symmetry algebras, may be reduced
to the same target PDE. Based on that observation, they propose that the
target PDE inherits Lie symmetries from all reduced PDEs, which explains why
some of the new symmetries are not admitted by the specific PDE\ used for
the reduction. In this context arises the problem of identifying the set of
all PDEs which lead to the same reduced PDE after reduction by a Lie
symmetry. In a recent paper \cite{LeachGovinderAndriopoulos2012}, it has
been shown that this is also the case with the ODEs; that is, it is shown
that different differential equations which can be reduced to the same
equation provide point sources for each of the Lie symmetries of the reduced
equation even though any particular of the higher order equations may not
provide the full complement of Lie symmetries. Therefore, concerning the
ODEs the Lie symmetries of the reduced equation can be viewed as having two
sources. Firstly, the point and nonlocal symmetries of \emph{a given} higher
order equation and secondly, the point symmetries of \emph{a variety} of
higher order ODEs. Finally, in a newer paper \cite{GovinderAbraham2008}, it
has been shown by a counter example that Type II\ hidden symmetries for PDEs
can have a nonpoint origin, i.e. they arise from contact symmetries or even
nonlocal symmetries of the original equation. Other approaches may be found
in \cite{Gandarias2008}.

In the present Chapter we will study the reduction and the consequent
existence of Type II hidden symmetries of the homogeneous heat equation
\begin{equation}
\Delta u-u_{t}=0  \label{t2.1}
\end{equation}%
and the Laplace equation
\begin{equation}
\Delta u=0  \label{PE.9}
\end{equation}%
in certain classes of Riemannian spaces, where%
\begin{equation*}
\Delta =\frac{1}{\sqrt{\left\vert g\right\vert }}\frac{\partial }{\partial
x^{i}}\left( \sqrt{\left\vert g\right\vert }g^{ij}\frac{\partial }{\partial
x^{j}}\right)
\end{equation*}%
is the Laplace operator. In a general Riemannian space, the homogeneous heat
equation (\ref{t2.1}) admits three Lie point symmetries and the Laplace
equation (\ref{PE.9}) admits two Lie point symmetries, which are not useful
for reduction. This implies that if we wish to find `sound' reductions of
equations (\ref{t2.1}) and (\ref{PE.9}), we have to consider Riemannian
spaces which admit some type of symmetry(ies) of the metric (these
symmetries are not Lie symmetries and are called collineations). Indeed, as
it has been shown in Chapter \ref{chapter5}, the Lie symmetries of the
homogeneous heat and the Laplace equation in a Riemannian space are
generated from the elements of the homothetic algebra and the conformal
algebra of the space respectively. Thus, one expects that in spaces with a
nonvoid conformal algebra there will be Lie symmetry vectors which will
allow for the reduction of equations (\ref{t2.1}), (\ref{PE.9}) and the
possibility of the existence of Type II\ hidden symmetries.

In section \ref{Reduction of the Laplace in certain Riemannian spaces} we
reduce the Laplace equation with the extra Lie symmetries existing in (a) a
decomposable space - that is a Riemannian space which admits a gradient
Killing vector (KV) - (b) in a space which admits a gradient Homothetic
vector (HV) and (c) in a space which admits a special Conformal killing
vector (sp.CKV). In section \ref{ApplicationsLaplace}, we apply the results
of the previous section and find the Type II symmetries of the Laplace
equation in four and three dimensional Minkowski spacetimes. Also we fully
recover previous results \cite{AGA06}. In order to study the reduction of
Laplace equation by a non-gradient HV and a proper CKV we consider two
further examples. In section \ref{PetroFrw}, we consider the algebraically
special vacuum solution of Einstein's equations known as Petrov type III
\cite{Steele1991b} and we make the reduction using the Lie point symmetry
generated by the nongradient HV. Moreover\ we do the same in a conformally
flat spacetime where the proper CKVs generate Lie symmetries.

In section \ref{Heatred}, we reduce the homogeneous heat equation (\ref{t2.1}%
) with the extra Lie symmetries existing (i) in a space which admits a
gradient KV and (ii) in a space which admits a gradient HV. In section \ref%
{ApplicationsHeatre}, we consider the special cases of \ the previous
section, that is, a decomposable space whose nondecomposable part is a
maximally symmetric space of non-vanishing curvature and the spatially flat
Friedmann Robertson Walker (FRW) space time used in Cosmology. Finally in
section \ref{PetrovIIIHeat} we consider the reduction of the homogeneous
heat equation in the Petrov type III spacetime using the Lie symmetry which
is generated by the HV.

We emphasize that all results are derived in a purely geometric manner
without the use of a computer package. However, we have verified them with
the libraries PDEtools and SADE \cite{PDEtools,SADE} of Maple\footnote{%
www.maplesoft.com}.

\section{Lie symmetries of Laplace equation in certain Riemannian spaces}

\label{Laplace equation in certain Riemanian spaces}

In a general Riemannian space, Laplace equation (\ref{PE.9}) admits the Lie
symmetries
\begin{equation*}
~X_{u}=u\partial _{u}~,~X_{b}=b\left( t,x\right) \partial _{u}
\end{equation*}%
where $b\left( t,x\right) $ is a solution of Laplace equation. These
symmetries are too general to provide useful reductions and lead to reduced
PDEs which posses Type II\ hidden symmetries. However if we restrict our
considerations to spaces which admit a conformal algebra (proper or not)
then we will have new Lie symmetries, hence new reductions of Laplace
equation eqn (\ref{PE.9}), which might lead to Type II\ hidden symmetries.

In the following sections, we consider spaces in which the metric $g_{ij}$
can be written in a generic form. The spaces we shall consider are:~a.
Spaces which admit a gradient KV (decomposable spaces) ~b. Spaces which
admit a gradient HV and c. Spaces which admit a sp.CKV.

The generic form of the metric for each type is as follows ($A,B,\ldots
=1,2,\ldots ,n$):

\begin{enumerate}
\item[a.] If a $1+n-$dimensional Riemannian space admits a gradient KV, the $%
S^{i}=\partial _{z}~\left( S=z\right) $ say, then the space is decomposable
along $\partial _{z}$ and the metric is written as (see e.g. \cite{TNA})
\begin{equation*}
ds^{2}=dz^{2}+h_{AB}y^{A}y^{B}~,~h_{AB}=h_{AB}\left( y^{C}\right)
\end{equation*}

\item[b.] If a $1+n-$dimensional Riemannian space admits a gradient HV, the $%
H^{i}=r\partial _{r}~\left( H=\frac{1}{2}r^{2}\right) ,~\psi _{H}=1$ say,
then the metric can be written in the generic form \cite{Tupper1989}%
\begin{equation*}
ds^{2}=dr^{2}+r^{2}h_{AB}dy^{A}dy^{B}~,~~h_{AB}=h_{AB}\left( y^{C}\right)
\end{equation*}

\item[c.] If a $1+n-$dimensional Riemannian space admits a sp.CKV then
admits a gradient KV and a gradient HV~\cite{Tupper1989,HallspCKV} and the
metric can be written in the generic form%
\begin{equation*}
ds^{2}=-dz^{2}+dR^{2}+R^{2}f_{AB}\left( y^{C}\right) dy^{A}dy^{B}
\end{equation*}%
while the sp.CKV is $C_{S}=\frac{z^{2}+R^{2}}{2}\partial _{z}+zR\partial
_{R} $ with conformal factor $\psi _{C_{S}}=z$.
\end{enumerate}

The Riemannian spaces which admit non-gradient proper HV do not have a
generic form for their metric. However, the\ spaces for which the HV\ acts
simply transitively are a few and are given together with their homothetic
algebra in \cite{Steele1991b}. A\ special class of these spaces are the
algebraically special vacuum space-times known as Petrov type N, II , III,
D. In section \ref{LPPe}, we shall consider the reduction of Laplace
equation in the Petrov type III spacetime whose metric is%
\begin{equation*}
ds^{2}=2d\rho dv+\frac{3}{2}xd\rho ^{2}+\frac{v^{2}}{x^{3}}\left(
dx^{2}+dy^{2}\right)
\end{equation*}%
with the symmetry generated by the non-gradient HV~$H=v\partial _{v}+\rho
\partial _{\rho }~,~\psi _{H}=1.$ The reduction of Laplace equation in the
rest of the Petrov types is similar both in the working method and results
and there is no need to consider them explicitly.

Finally we shall consider the conformally flat space
\begin{equation*}
ds^{2}=e^{2t}\left[ dt^{2}-\delta _{AB}y^{A}y^{B}\right]
\end{equation*}%
which admits a proper CKV which produces a Lie symmetry\footnote{%
According to theorem \ref{KG1}, the condition for this is that the conformal
factor satisfies Laplace equation} and we reduce Laplace equation using this
symmetry.

In what follows all spaces are assumed to be of dimension $n>2$.

\section{Reduction of the Laplace equation in certain Riemannian spaces}

\label{Reduction of the Laplace in certain Riemannian spaces}

As we have seen in Chapter \ref{chapter5}, the Lie symmetries of Laplace
equation (\ref{PE.9}) in a Riemannian space are generated from the CKVs (not
necessarily proper) whose conformal factor satisfies Laplace equation. This
condition is satisfied trivially by the KVs ($\psi =0),$ the HV ($\psi
_{;i}=0)$ and the sp.CKVs ($\psi _{;ij}=0).$ Therefore these vectors (which
span a subalgebra of the conformal group) are among the Lie symmetries of
Laplace equation. Concerning the proper CKVs it is not necessary that their
conformal factor satisfies the Laplace equation, therefore they may not
produce Lie symmetries for Laplace equation.

\subsection{Riemannian spaces admitting a gradient KV}

\label{grKV}

Without loss of generality we assume the gradient KV to be the $\partial
_{z},$ so that the metric has the generic form%
\begin{equation}
ds^{2}=dz^{2}+h_{AB}dy^{A}dy^{B}~,~h_{AB}=h_{AB}\left( y^{C}\right)
\label{WH.03}
\end{equation}%
where $h_{AB}$ $A,B,C=1,...,n$ is the metric of the $n-$ dimensional space.
For the metric (\ref{WH.03}) Laplace equation (\ref{PE.9}) takes the form%
\begin{equation}
u_{zz}+h^{AB}u_{AB}-\Gamma ^{A}u_{B}=0.  \label{WH.04A}
\end{equation}%
and admits as \emph{extra }Lie symmetry the gradient KV $\partial _{z}.$

We reduce (\ref{WH.04A}) by using the zeroth order invariants $y^{A}~,~w=u~$
of the extra Lie symmetry $\partial _{z}$. Taking these invariants as new
coordinates, equation (\ref{WH.04A}) reduces to
\begin{equation}
_{h}\Delta w=0  \label{WH.05A}
\end{equation}%
which is Laplace equation in the~$n~$dimensional space with metric $h_{AB}.$
Now we recall the result that the conformal algebra of the $n$ metric $%
h_{AB} $ and the $1+n$ metric (\ref{WH.03}) are related as follows \cite{TNA}%
:

a. The KVs of the $n$ metric are identical with those of the $n+1$ metric

b. The $1+n$ metric admits a HV\ if and only if the $n$ metric admits one
and if $_{n}H^{A}$ is the HV of the $n$ metric then the HV\ of the $1+n$
metric is given by the expression
\begin{equation}
_{1+n}H^{\mu }=z\delta _{z}^{\mu }+_{n}H^{A}\delta _{A}^{\mu }\text{ }\mu
=x,1,...,n.  \label{WH.05.AA}
\end{equation}

d. The \thinspace $1+n$ metric admits CKVs if and only if the $n$ metric $%
h_{AB}$ admits a gradient CKV (for details see \cite{TNA}).

Therefore Type II hidden symmetries for (\ref{PE.9}) exist if the $n~$metric
$h_{AB}$ admits more symmetries. Specifically, the sp.CKVs of the $h_{AB}$
metric as well as the proper CKVs whose conformal function is a solution of
Laplace equation (\ref{WH.05}) generate Type II\ hidden symmetries.

\subsection{Riemannian spaces admitting a gradient HV.}

\label{grHV}

In Riemannian spaces which admit a gradient HV, $H$ say, there exists a
coordinate system in which the metric is written in the form \cite%
{Tupper1989}%
\begin{equation}
ds^{2}=dr^{2}+r^{2}h_{AB}\left( y^{C}\right) dy^{A}dy^{B}  \label{LEH.01}
\end{equation}%
and the gradient HV is $H=r\partial _{r}.$ In these coordinates the
Laplacian (\ref{PE.9}) takes the form
\begin{equation}
u_{rr}+\frac{1}{r^{2}}h^{AB}u_{AB}+\frac{\left( n-1\right) }{r}u_{r}-\frac{1%
}{r^{2}}\Gamma ^{A}u_{A}=0  \label{LEH.02}
\end{equation}%
and admits the extra Lie symmetry (see Theorem {\LARGE \ }\ref{KG1}) $H$. We
reduce (\ref{LEH.02}) using $H.$

The zeroth order invariants of $H$ are $y^{A}~,~w\left( y^{A}\right) ~$and
using them it follows easily that the reduced equation is%
\begin{equation}
_{h}\Delta u=0  \label{LEH.03}
\end{equation}%
that is, the Laplacian defined by the metric $h_{AB}$.

It is easy to establish the following results concerning the conformal
algebras of the metrics (\ref{LEH.01}) and $h_{AB}$.\newline
1. The KVs of $h_{AB}$ are also KVs of (\ref{LEH.01}).\newline
2. The HV of (\ref{LEH.01}) is independent from that of $h_{AB}$.\newline
3. The metric (\ref{LEH.01}) admits proper CKVs if and only if the $n$
metric $h_{AB}$ admits gradient CKVs. This is because (\ref{LEH.01}) is
conformally related with the decomposable metric
\begin{equation}
ds^{2\prime }=dr^{2}+h_{AB}\left( y^{C}\right) dy^{A}dy^{B}.
\end{equation}

The above imply, that Type II hidden symmetries we shall have from the HV of
the metric $h_{AB},$ the sp.CKVs and finally from the proper CKVs of $h_{AB}$
whose conformal factor is a solution of Laplace equation (\ref{LEH.03}).

\subsection{Riemannian spaces admitting a sp.CKV}

\label{spCKV}

It is known \cite{HallspCKV}, that if an $n=m+1$ dimensional ($n>2)$
Riemannian space admits sp.CKVs then also admits a gradient HV and as many
gradient KVs as the number of sp.CKVs. In these spaces there exists always a
coordinate system in which the metric is written in the form \cite%
{Tupper1989}
\begin{equation}
ds^{2}=-dz^{2}+dR^{2}+R^{2}f_{AB}\left( y^{C}\right) dy^{A}dy^{B}
\label{LES.01}
\end{equation}%
where $\partial _{z}$ is the gradient KV and $z\partial _{z}+R\partial _{R}$
is the gradient HV. $f_{AB}\left( y^{C}\right) \,,~A,B,C,..=1,2,,...m-1$ is
an $m-1$ \ dimensional metric. For a general $m-1~$ dimensional metric $%
f_{AB}$ the $n~$dimensional metric (\ref{LES.01}) admits the following
special Conformal group
\begin{eqnarray*}
K_{G} &=&\partial _{z}~,~H=z\partial _{z}+R\partial _{R} \\
C_{S} &=&\frac{z^{2}+R^{2}}{2}\partial _{z}+zR\partial _{R}
\end{eqnarray*}%
where $K_{G}$ is a gradient KV, $H$ is a gradient HV and $C_{S}$ is a sp.CKV
with conformal factor $\psi _{C_{S}}=z$. In these coordinates Laplace
equation (\ref{PE.9}) takes the form%
\begin{equation}
-u_{zz}+u_{RR}+\frac{1}{R^{2}}h^{AB}u_{AB}+\frac{\left( m-1\right) }{R}u_{R}-%
\frac{1}{R^{2}}\Gamma ^{A}u_{A}=0.  \label{LES.02}
\end{equation}

From Theorem \ref{KG1}, we have that the extra Lie symmetries of (\ref%
{LES.02}) are the vectors%
\begin{eqnarray*}
X^{1} &=&K_{G}~,~X^{2}=H \\
X^{3} &=&C_{S}+2pzu\partial _{u}
\end{eqnarray*}%
where~$2p=\frac{1-m}{2}~$ and the non zero commutators are%
\begin{equation*}
\left[ X^{1},X^{2}\right] =X^{1}~,~\left[ X^{2},X^{3}\right] =X^{3}
\end{equation*}%
\begin{equation*}
\left[ X^{1},X^{3}\right] =X^{2}+2pX_{u}.
\end{equation*}%
We consider the reduction of (\ref{LES.02}) with each of the extra Lie
symmetries.

\subsubsection{Reduction with the gradient KV $X^{1}$.}

The first order invariants of $X^{1}$ are $R,y^{C}$,~$w\left( R,y^{C}\right)
$ and by using them we reduce the Laplacian (\ref{LES.02}) to (\ref{LEH.02})
which admits the Lie symmetry $X^{2}$ generated by the HV. This result is
expected because $\left[ X^{1},X^{2}\right] =X^{1}$ \cite{AGA06} hence the
Lie symmetry $X^{2}$ is inherited. Therefore, in this reduction the Type II
symmetries are generated from the CKVs of the metric (\ref{LEH.01}). It is
possible to continue the reduction by the gradient HV\ $H$ and then we find
the results of section \ref{grHV}.

\subsubsection{Reduction with the gradient HV $X^{2}$.}

The reduction with a gradient HV has been studied in section \ref{grHV}. To
apply the results of section \ref{grHV} in the present case we have to bring
the metric (\ref{LES.01}) to the form (\ref{LEH.01}). For this we consider
the transformation%
\begin{equation*}
z=r\sinh \theta ~,~R=r\cosh \theta
\end{equation*}%
which brings (\ref{LES.01}) to%
\begin{equation}
ds^{2}=dr^{2}+r^{2}\left( -d\theta ^{2}+\cosh ^{2}\theta
f_{AB}y^{A}y^{B}\right)
\end{equation}%
so that the metric $h_{AB}$ of \ (\ref{LEH.01}) is%
\begin{equation}
ds_{h}^{2}=\left( -d\theta ^{2}+\cosh ^{2}\theta f_{AB}y^{A}y^{B}\right) .
\label{LES.02A}
\end{equation}

The reduced equation of (\ref{LES.02}) under the Lie symmetry generated by
the gradient HV is Laplace equation in the space (\ref{LES.02A}). For this
reduction we do not have inherited symmetries and there exist Type II hidden
symmetries as stated in section \ref{grHV}.

\subsubsection{Reduction with the sp.CKV $X^{3}$.}

\label{sprelap}

Before we reduce (\ref{LES.02}) with the symmetry generated by the sp CKV~$%
X^{3}$, it is best to write the metric (\ref{LES.01}) in new coordinates. We
introduce the new variable $x$ via the relation
\begin{equation}
z=\sqrt{\frac{R\left( xR-1\right) }{x}}.  \label{LES.03a}
\end{equation}%
In the new variables the metric (\ref{LES.01}) becomes%
\begin{equation}
ds^{2}=-\frac{R}{4x^{3}\left( xR-1\right) }dx^{2}-\frac{2xR-1}{2x^{2}\left(
xR-1\right) }dxdR-\frac{1}{4xR\left( xR-1\right) }dR^{2}+R^{2}f_{AB}\left(
y^{C}\right) dy^{A}dy^{B}  \label{LES.03}
\end{equation}%
the Laplacian (\ref{LES.02}):%
\begin{eqnarray}
0 &=&\frac{x^{2}}{R^{2}}u_{xx}-2\frac{x}{R}\left( 2xR-1\right) u_{xR}+u_{RR}+%
\frac{1}{R^{2}}f^{AB}u_{AB}+  \label{LES.04} \\
&&+\frac{\left( m-1\right) }{R}u_{R}-\frac{x}{R^{2}}\left( m-1\right) \left(
2xR-1\right) u_{x}-\frac{1}{R^{2}}\Gamma ^{A}u_{A}  \notag
\end{eqnarray}%
and the Lie symmetry $X^{3}$
\begin{equation*}
X^{3}=\sqrt{\frac{R\left( xR-1\right) }{x}}R\partial _{R}+2p\sqrt{\frac{%
R\left( xR-1\right) }{x}}u\partial _{u}.
\end{equation*}%
The zeroth order invariants of $X^{3}$ are $x,y^{A}~,~w=uR^{-2p}.$ We choose
$x,y^{A}$ to be the independent variables and $w=w\left( x,y^{A}\right) $
the dependent one. By replacing in (\ref{LES.04}) we find the reduced
equation%
\begin{equation}
x^{2}w_{xx}+f^{AB}w_{AB}-\Gamma ^{A}w_{A}-2p\left( 2p+1\right) w=0
\label{LES.07}
\end{equation}%
We consider cases.\newline
\textbf{The case }$m\succeq 4$\textbf{$.$}

If $2p+1\neq 0,~m\succeq 4~$\ then (\ref{LES.07}) becomes%
\begin{equation}
_{\left( m\succeq 4\right) }\bar{\Delta}w-2p\left( 2p+1\right) V\left(
x\right) w=0  \label{LES.11}
\end{equation}%
where~$V\left( x\right) =x^{\frac{2}{2-m}}$ and $_{\left( m\succeq 4\right) }%
\bar{\Delta}$ is the Laplace operator for the metric%
\begin{equation}
d\bar{s}_{\left( m\succeq 4\right) }^{2}=\frac{1}{V\left( x\right) }\left(
\frac{1}{x^{2}}dx^{2}+f_{AB}dy^{A}dy^{B}\right) .  \label{LES.11b}
\end{equation}

Equation (\ref{LES.11}) is the Klein Gordon equation in a space with
potential $V\left( x\right) =x^{\frac{2}{2-m}}$ and metric (\ref{LES.11b}).
Considering the new transformation $\phi =\int \sqrt{\frac{1}{xV}}dx~\ $~or $%
x=\left( m-2\right) ^{2-m}\phi ^{m-2}$ the metric (\ref{LES.11b}) is written%
\begin{equation}
d\bar{s}_{\left( m\succeq 4\right) }^{2}=d\phi ^{2}+\phi ^{2}\bar{f}%
_{AB}dy^{A}dy^{B}  \label{LES.11c}
\end{equation}%
where $\bar{f}_{AB}=\left( m-2\right) ^{-2}f_{AB}$ whereas the potential $%
V\left( \phi \right) =\frac{\left( 2-m\right) ^{2}}{\phi ^{2}}$ which is the
well known Ermakov potential \cite{LeachAndriopoulos}.

This means that the gradient HV $\phi \partial _{\phi }~,~\psi _{\phi }=1$,
is a Lie symmetry of (\ref{LES.11}) which is the Lie symmetry $X^{2}.$
Therefore, if the metric $\bar{f}_{AB}$ admits proper CKVs which satisfy the
conditions of Theorem \ref{KG}, then these vectors generate Type II hidden
symmetries of (\ref{LES.02}).\newline
\textbf{The case }$m=3.$

If $2p+1=0,~$then $~m=3$ and $f_{AB}$ is a two dimensional metric. In this
case, equation (\ref{LES.07}) becomes%
\begin{equation}
x^{2}w_{xx}+f^{AB}w_{AB}-\Gamma ^{A}w_{A}=0  \label{LES.08}
\end{equation}%
or, by multiplying with ~$x^{2}$%
\begin{equation}
_{\left( m=3\right) }\bar{\Delta}w=0  \label{LES.09}
\end{equation}%
which is the Laplacian in the three dimensional space with metric
\begin{equation}
d\bar{s}_{\left( m=3\right) }^{2}=\frac{1}{x^{4}}dx^{2}+\frac{1}{x^{2}}%
f_{AB}dy^{A}dy^{B}.  \label{LES.10}
\end{equation}

By making the new transformation $x=\frac{1}{r},\ $ (\ref{LES.09}) is of the
form (\ref{LEH.01}) and admits the gradient HV $r\partial _{r}~$which gives
an inherited symmetry. We conclude that Type II hidden symmetries of (\ref%
{LES.10}) will be generated from the proper CKVs of the metric (\ref{LES.10}%
) which satisfy the condition of Theorem \ref{KG1}.\newline
\textbf{The case }$m=2$\textbf{$.$}

For $m=2,~$ $f_{AB}$ \ is a one dimensional metric and (\ref{LES.01}) is
\begin{equation}
ds^{2}=-dz^{2}+dR^{2}+R^{2}d\theta ^{2}
\end{equation}%
which is a flat metric\footnote{%
The only three dimensional space who admits sp.CKV is the flat space,
because in that case we also have a gradient HV\ and a gradient KV.}. In
this space, Laplace equation (\ref{LES.02}) admits ten Lie point symmetries,
as many as the elements of the Conformal algebra of the flat 3D space. Six
of these vectors are KVs, one vector is a gradient HV and three are sp.CKVs
(see example \ref{ExCAflat}). We reduce the Laplace equation with the
symmetry $X^{3}$~\ and the reduced equation is (\ref{LES.07}) which for $%
f_{AB}=\delta _{\theta \theta }$ becomes
\begin{equation}
x^{2}w_{xx}+w_{\theta \theta }+\frac{1}{4}w=0.  \label{LES.072}
\end{equation}%
Equation (\ref{LES.072}) is in the form of (\ref{GPE.30a}) (see Chapter \ref%
{chapter5}) with $A^{ij}=diag\left( x^{2},1\right) $ and $B^{i}=0.~$ By
replacing in the symmetry conditions (\ref{GPE.42})-(\ref{GPE.46}) we find
the Lie symmetries%
\begin{equation*}
X=\xi ^{i}\partial _{i}+\left( a_{0}w+b\right) \partial _{w}
\end{equation*}%
where $\xi ^{i}$ are the CKVs of the two dimensional space with metric $%
A^{ij}$. In this case, all proper CKVs of the two dimensional space $A^{ij}$
generate Type II Lie symmetries. Recall that the conformal algebra of a two
dimensional space is infinite dimensional.

\section{Type II\ hidden symmetries of the 3D and the 4D wave equation}

\label{ApplicationsLaplace}

In this section we apply the general results of the previous section to
specific spaces where the metric is known.

\subsection{Laplacian in $M^{4}$}

\label{M4A}

The Laplace equation in the four dimensional Minkowski spacetime $M^{4}$
\begin{equation}
ds^{2}=-dt^{2}+dx^{2}+dy^{2}+dz^{2}  \label{Ap1.01}
\end{equation}%
is the wave equation \cite{AGA06} in $E^{3}$%
\begin{equation}
u_{tt}-u_{xx}-u_{yy}-u_{zz}=0.  \label{Ap1.02}
\end{equation}

The conformal algebra of the metric (\ref{Ap1.01}) is generated by 15
vectors (see example \ref{ExCAflat}). From theorem \ref{KG1} we have that
the extra Lie point symmetries of (\ref{Ap1.02}) are the vectors
\begin{equation}
K_{G}^{1}~,~K_{G}^{A}~,~X_{R}^{1A}~,~X_{R}^{AB}~,~H~,~X_{C}^{1}-tu\partial
_{u}~,~X_{C}^{A}-y^{A}u\partial _{u}  \label{Ap1.02A}
\end{equation}%
where $y^{A}=\left( x,y,z\right) .$

The nonzero commutators are%
\begin{eqnarray*}
\left[ K_{G}^{I},X_{R}^{IJ}\right] &=&-K_{G}^{J}~,~\left[ K_{G}^{I},~H\right]
=K_{G}^{I} \\
\left[ K_{G}^{I},X_{C}^{I}\right] &=&H-X_{u}~,~\left[ K_{G}^{I},X_{C}^{J}%
\right] =X_{R}^{IJ} \\
\left[ H,X_{C}^{I}\right] &=&X_{C}^{I}~,~\left[ X_{R}^{IJ},X_{C}^{I}\right]
=X_{C}^{J}~.
\end{eqnarray*}

\subsubsection{Reduction with a gradient KV \label{M4G}}

We choose to make reduction of (\ref{Ap1.02}) with the gradient KV $%
K_{G}^{z}=\partial _{z}$. The reduced equation is%
\begin{equation}
w_{tt}-w_{xx}-w_{yy}=0  \label{Ap1.033}
\end{equation}%
which is Laplace equation in the space $M^{3}.$ The extra Lie symmetries of (%
\ref{Ap1.033}) are
\begin{equation}
K_{G}^{1}~,~K_{G}^{x}~,~K_{G}^{y}~,~~X_{R}^{1a}~\ ,~X_{R}^{ab}~,~H
\label{Ap1.033A}
\end{equation}%
and are inherited symmetries (see also the last commutators). The Type II
symmetries are the vectors%
\begin{equation}
\bar{X}_{C}^{1}-\frac{1}{2}tu\partial _{u}~,~\bar{X}_{C}^{x}-\frac{1}{2}%
xu\partial _{u}~,~\bar{X}_{C}^{y}-\frac{1}{2}yu\partial _{u}
\label{Ap1.033B}
\end{equation}%
that is the Type II hidden symmetries are generated from the sp.CKVs of $%
M^{3}$.

\subsubsection{Reduction with the gradient HV}

In this case it is better to switch to hyperspherical coordinates $(r,\theta
,\phi ,\zeta ).$

In these coordinates the metric (\ref{Ap1.01}) is written%
\begin{equation}
ds^{2}=dr^{2}-r^{2}\left[ d\theta ^{2}+\cosh ^{2}\theta \left( d\phi
^{2}+\cosh ^{2}\phi d\zeta ^{2}\right) \right]  \label{Ap1.04}
\end{equation}%
and the wave equation (\ref{Ap1.02}) becomes%
\begin{equation}
u_{rr}-\frac{1}{r^{2}}\left( u_{\theta \theta }+\frac{u_{\phi \phi }}{\cosh
^{2}\theta }+\frac{u_{\zeta \zeta }}{\cosh ^{2}\theta \cosh \phi ^{2}}%
\right) +\frac{3}{r}u_{r}-2\frac{\tanh \theta }{r^{2}}u_{\theta }-\frac{%
\tanh \phi }{r^{2}\cosh ^{2}\theta }u_{\phi }=0.  \label{Ap1.05}
\end{equation}

According to the analysis of section \ref{grHV} the reduced equation is (\ref%
{LEH.03}) which is the Laplacian in the three dimensional space of the
variables $\left( \theta ,\phi ,\zeta \right) :$%
\begin{equation}
w_{\theta \theta }+\frac{w_{\phi \phi }}{\cosh ^{2}\theta }+\frac{w_{\zeta
\zeta }}{\cosh ^{2}\theta \cosh \phi ^{2}}+2\frac{\tanh \theta }{r^{2}}%
w_{\theta }+\frac{\tanh \phi }{r^{2}\cosh ^{2}\theta }w_{\phi }=0.
\label{Ap1.06}
\end{equation}%
This space is a space of constant curvature. The conformal algebra\footnote{%
All spaces of constant curvature are conformally flat, hence, they admit the
same conformal algebra with the flat space but not the same subalgebras,
i.e. the same conformal factors.} of a 3D space of constant non-vanishing
curvature consists of 6 non-gradient KVs and 4 proper CKVs\footnote{%
The rotations and the sp.CKVs of the flat space are KVs for the space of
constant curvature, the rest are proper gradient CKVs} \cite{Barnes}. The
conformal factors of the CKVs do not satisfy the condition $_{h}\Delta \psi
=0$ (see Theorem \ref{KG1}); hence, they do not generate Lie point
symmetries for the reduced equation (\ref{Ap1.06}) whereas for the same
reason the KVs are Lie symmetries of (\ref{Ap1.06}). Therefore, all point
Lie point symmetries are inherited and we do not have Type II\ hidden
symmetries.

We note that the proper CKVs in a space of constant non-vanishing curvature
are gradient and their conformal factor satisfies the relation \cite{Barnes}
\begin{equation*}
\psi _{;ab}=C\psi g_{ab}\rightarrow g^{ab}\psi _{;ab}=nC\psi \rightarrow
_{h}\Delta \psi =nC\psi
\end{equation*}%
which implies that they are Lie symmetries of the conformally invariant
operator but not of the Laplace equation (\ref{Ap1.06}).

\subsubsection{Reduction with a sp.CKV}

Following the steps of section \ref{spCKV}, we consider the transformation
to axi-symmetric coordinates $\left( t,R,\theta ,\phi \right) $ in which (%
\ref{Ap1.02}) takes the form
\begin{equation}
u_{tt}-u_{RR}-\frac{1}{R^{2}}\left( u_{\theta \theta }+\frac{u_{\phi \phi }}{%
\cosh ^{2}\theta }\right) -\frac{2}{R}u_{r}-\frac{\tanh \theta }{R^{2}}%
u_{\theta }=0.  \label{Ap1.07}
\end{equation}%
Applying the transformation (\ref{LES.03a}) $t=\sqrt{\frac{R\left( \tau
R-1\right) }{\tau }}$ we find (note that this is the case $m=3$) that (\ref%
{Ap1.07}) is written as (\ref{LES.04}) and the reduced equation is the
Laplacian $_{\left( m=3\right) }\Delta w~\ \ $for the 3D metric
\begin{equation}
ds^{2}=\frac{1}{\tau ^{4}}d\tau ^{2}-\frac{1}{\tau ^{2}}\left( d\theta
^{2}+\cosh ^{2}\theta d\phi ^{2}\right) .  \label{Ap.1.08}
\end{equation}%
The metric (\ref{Ap.1.08}) under the coordinate transformation $\tau =\frac{1%
}{T}$ is written%
\begin{equation}
ds^{2}=dT^{2}-T^{2}\left( d\theta ^{2}+\cosh ^{2}\theta d\phi ^{2}\right)
\label{AP1.08}
\end{equation}%
which is the flat 3D Lorentzian metric, which does not admit proper CKVs.
This implies that the Lie symmetries of the reduced equation are generated
from the KVs/HV/sp.CKVs of the flat $M^{3}$ metric and all are inherited.
Therefore we do not have Type II\ hidden symmetries for the reduction with a
sp.CKV.

As we shall show in the next section this is not the case for the reduction
of Laplace equation in $M^{3}.$

\subsection{Laplacian in $M^{3}$}

We consider the reduction of Laplace equation in the 3d Minkowski $M^{3}~$%
spacetime \cite{AGA06}, i.e. the wave equation in $E^{2}$%
\begin{equation}
u_{tt}-u_{xx}-u_{yy}=0.  \label{Ap2.02}
\end{equation}%
As we have seen in section \ref{M4G}\ the extra Lie point symmetries of (\ref%
{Ap2.02}) are the ten vectors (\ref{Ap1.033A}) and (\ref{Ap1.033B}).

\subsubsection{Reduction with a gradient KV\label{M3G}}

We choose the vector $\partial _{y}$ and the reduction gives the reduced
equation%
\begin{equation}
w_{tt}-w_{xx}=0  \label{Ap2.03}
\end{equation}%
which is the one dimensional wave equation. The 2d space $(t,x)$ has an
infinite number of CKVs therefore (\ref{Ap2.03}) has infinite Lie point
symmetries \cite{StephaniB}. From these symmetries the KVs and the HV are
inherited symmetries and the CKVs are Type II\ symmetries.

\subsubsection{Reduction with the gradient HV}

In order to do the reduction with the gradient HV we introduce spherical
coordinates $(r,\theta ,\phi )$ and find that (\ref{Ap2.02}) becomes%
\begin{equation}
u_{rr}-\frac{1}{r^{2}}\left( u_{\theta \theta }+\frac{u_{\phi \phi }}{\cosh
^{2}\theta }\right) +\frac{2}{r}u_{r}-\frac{\tanh \theta }{r^{2}}u_{\theta
}=0.
\end{equation}%
According to the results of section \ref{grHV} the reduced equation is
equation (\ref{LEH.03}) which is the Laplace equation in the 2d space of the
variables $\left( \phi ,\theta \right)$
\begin{equation}
w_{\theta \theta }+\frac{w_{\phi \phi }}{\cosh ^{2}\theta }+\tanh \theta
w_{\theta }=0.  \label{Ap2.05}
\end{equation}%
By making the transformation $\theta =\ln \left( \tan \frac{\rho }{2}\right)
$ equation (\ref{Ap2.05}) becomes
\begin{equation}
\sin \left( x\right) ^{2}\left( w_{\rho \rho }+w_{\phi \phi }\right) =0
\label{Ap2.06}
\end{equation}%
which is the wave equation (\ref{Ap2.03}) with $t=\rho ~,~x=-i\phi $.~

We obtain the results, concerning the Lie symmetries of (\ref{Ap2.06}) from
section \ref{M3G} with the difference that the Lie symmetry due to the HV of
the two dimensional metric is not inherited but it is a Type II hidden
symmetry.

The reduction of the wave equation in the 4D and in 3D Minkowski space has
been done previously by Abraham-Shrauner et. all \cite{AGA06} and our
results coincide with theirs. For example in the 3d case equation (17) of
\cite{AGA06} is our equation {(\ref{Ap2.05})} in other variables. However
there are two differences (a) in the case of the 2D space they do not obtain
that the Lie symmetries are infinite and (b) they use algebraic computing
programs to find the Lie symmetry generators whereas our approach is
geometric and general and does not need algebraic computing programs to find
the complete answer.

The reduction with a sp.CKV has been considered in section \ref{sprelap}.

\section{The Laplace equation in the Petrov type III and in the FRW-like
spacetime}

\label{PetroFrw}

To complete our analysis, we have to reduce the Laplace equation using a non
gradient homothetic vector and a proper (i.e. non special)\ CKV. In order to
do this, we consider the reduction of Laplace equation in the Petrov type
III spacetime and in the FRW-like spacetime.

\subsection{The Laplace equation in the Petrov Type III spacetime}

\label{LPPe}

In this section we consider the reduction of Laplace equation in spaces
which do not admit gradient KVs or a gradient HV. As it has been mentioned
in section \ref{Laplace equation in certain Riemanian spaces} we shall
consider the algebraically special solutions of Einstein equations, that is
the Petrov type D,N,II and III. In fact we restrict our discussion to Petrov
Type III because both the method of work and the results are the same for
all Petrov types.

The metric of the Petrov type III space-time is
\begin{equation}
ds^{2}=2d\rho dv+\frac{3}{2}xd\rho ^{2}+\frac{v^{2}}{x^{3}}\left(
dx^{2}+dy^{2}\right)  \label{P3.00}
\end{equation}%
with conformal algebra%
\begin{eqnarray*}
K^{1} &=&\partial _{\rho }~,~K^{2}=\partial _{y}~,~K^{3}=v\partial _{v}-\rho
\partial _{\rho }+2x\partial _{x}+2y\partial _{y} \\
H &=&v\partial _{v}+\rho \partial _{\rho }~,~\psi =1
\end{eqnarray*}%
where $K^{1-4}$ are KVs and $H$ is a non-gradient HV. (The space does not
admit proper CKVs).

In this space-time Laplace equation (\ref{PE.9}) takes the form%
\begin{equation}
-\frac{3}{2}xu_{vv}+2u_{v\rho }+\frac{x^{3}}{v^{2}}\left(
u_{xx}+u_{yy}\right) -3\frac{x}{v}u_{v}+\frac{2}{v}u_{\rho }=0.
\label{P3.01}
\end{equation}%
From{\LARGE \ }Theorem \ref{KG1} we have{\LARGE \ }that the extra Lie point
symmetries are the vectors
\begin{equation*}
X_{1-3}=K_{1-3}~,~~X_{4}=H
\end{equation*}%
with nonzero commutators:%
\begin{equation*}
\left[ X_{2,},X_{3}\right] =2X_{2}
\end{equation*}%
\begin{equation*}
\left[ X_{3},X_{1}\right] =X_{1}~,~~\left[ X_{1},X_{4}\right] =X_{1}.~
\end{equation*}%
We use $X_{4}~$to reduce the PDE because this is the Lie symmetry generated
by the nongradient HV{\LARGE .}

The zero order invariants of $X_{4}$ are $\sigma =\frac{\rho }{v},x,y,w$. We
choose $\sigma ,x,y$ as the independent variables and $w=w\left( \sigma
,x,y\right) $ as the dependent variable and we find the reduced equation
\begin{equation}
-\sigma \left( \frac{3}{2}x\sigma +2\right) w_{\sigma \sigma }+x^{3}\left(
w_{xx}+w_{yy}\right) =0.  \label{P3.02}
\end{equation}%
Equation (\ref{P3.02}) can be written
\begin{equation}
_{III}\Delta ^{\ast }w-\left( \frac{3x\sigma }{2}+1\right) w_{\sigma }-\frac{%
3x^{3}\sigma }{2\left( 3x\sigma +4\right) }w_{x}=0  \label{P3.03}
\end{equation}%
where $_{III}\Delta ^{\ast }$ is the Laplacian for the metric%
\begin{equation}
ds^{2}=-\frac{1}{\sigma \left( \frac{3}{2}x\sigma +2\right) }d\sigma ^{2}+%
\frac{1}{x^{3}}\left( dx^{2}+dy^{2}\right) .  \label{P3.04}
\end{equation}%
The Lie symmetries of (\ref{P3.03}) will be generated from the conformal
algebra of (\ref{P3.04}) with some extra conditions (see equations (\ref%
{GPE.42})-(\ref{GPE.46})). Finally, we find that equation (\ref{P3.03})
admits as Lie point symmetries the vectors~$\partial _{y}~,~x\partial
_{x}+y\partial _{x}-\sigma \partial _{\sigma }$ which are inherited
symmetries. Therefore we do not have Type II hidden symmetries.

\subsection{The Laplace equation in the $n~$dimensional FRW-like spacetime.}

\label{FRWCKV}

We consider the $n$ dimensional FRW-like space $\left( n>2\right) $ with
metric
\begin{equation}
ds^{2}=e^{2t}\left[ dt^{2}-\left( \delta _{AB}dy^{A}dy^{B}\right) \right]
\label{EC.01}
\end{equation}%
where $\delta _{AB}$ is the $n-1$ dimensional Euclidian metric. The
reduction of Laplace equation in this space (for $n=4$) has been studied
previously in \cite{Kara}. The metric (\ref{EC.01}) is conformally flat
hence admits the same CKVs with the flat space but with different
subalgebras. More precisely the space admits

a. $\left( n-1\right) +\frac{\left( n-2\right) \left( n-3\right) }{2}~$KVs
the $K_{G}^{A},~X_{R}^{AB}~$

b. $1$ gradient HV$~$the $K_{G}^{1}=\partial _{t}~~$\newline
the rest vectors being proper\ CKVs \cite{MM86}. In this space Laplace
equation (\ref{PE.9}) becomes%
\begin{equation}
e^{-2t}\left[ u_{tt}-\left( \delta ^{AB}u_{AB}\right) +\left( n-2\right)
u_{t}\right] =0  \label{EC.02}
\end{equation}%
and the extra Lie symmetries are
\begin{equation*}
K_{G}^{A},~X_{R}^{AB}~,~K_{G}^{1}~,~X_{R}^{1A}-2pY^{a}u\partial _{u}
\end{equation*}%
where $2p=\frac{2-n}{2}$. \ The algebra of the Lie point symmetries is the
same with that of\ section \ref{M4}. We consider the reduction with a proper
CKV.

\subsubsection{Reduction with a proper CKV}

We may take any of the vectors $X_{R}^{1A}$ (because as one can see in the
Appendix there is a symmetry between the coordinates $y^{A}$). We choose the
vector
\begin{equation*}
X_{R}^{1x}=x\partial _{t}+t\partial _{x}+2pxu\partial _{u}.
\end{equation*}%
whose zero order invariants are
\begin{equation*}
R=t^{2}-x^{2}~,~y^{C}~,~w=e^{-2pt}u.
\end{equation*}%
We take the dependent variable to be the $w=w\left( R,y^{C}\right) $ and
find the reduced equation
\begin{equation}
4Rw_{RR}-\delta ^{ab}w_{ab}+4w_{R}-4p^{2}w=0  \label{EC.03}
\end{equation}%
where $a=1,\ldots ,n-2$. We consider cases.\newline
{\textbf{Case }}$n>3.$

For $n>3$ equation (\ref{EC.03}) is
\begin{equation}
_{C}\Delta w-4p^{2}f\left( R\right) w=0  \label{EC.04}
\end{equation}%
where $_{C}\Delta $ is the Laplace operator for the $\left( n-1\right) $
dimensional metric%
\begin{equation}
ds_{C}^{2}=\frac{1}{f\left( R\right) }\left( \frac{1}{4R}dR^{2}-\delta
_{ab}dy^{a}dy^{b}\right)  \label{EC.05}
\end{equation}%
and $f\left( R\right) =R^{-\frac{1}{n-3}}$. The metric (\ref{EC.05}) is
conformally flat hence we know its conformal algebra. Application of theorem %
\ref{KG} gives that the Lie point symmetries of (\ref{EC.04}) are the
vectors
\begin{eqnarray*}
X_{u} &=&u\partial _{u}~,~X_{b}=b\partial _{u} \\
X_{K}^{a} &=&\partial _{y^{a}}~,~X_{R}^{ab}~=y^{b}\partial
_{a}-y^{a}\partial _{b}.
\end{eqnarray*}%
These are inherited symmetries (this result agrees with the commutators). We
conclude that for this reduction we do not have Type II hidden symmetries.

{\textbf{Case }}$n=3${\textbf{$.$}}

For $n=3$ the reduced equation is a two dimensional equation $\left( \text{%
that is }\delta _{AB}=\delta _{yy}\right) $%
\begin{equation}
4Rw_{RR}-w_{yy}+4w_{R}-\frac{1}{4}w=0
\end{equation}%
and admits as Lie point symmetry the KV $\partial _{y}$ which is an
inherited symmetry. Hence, we do not have Type II hidden symmetries.

We conclude that the reduction of Laplace equation in an $n~$dimensional FRW
like space with the proper CKV does not produce Type II hidden symmetries
and in fact the inherited symmetries of the reduced equation are the KVs of
the flat metric.

\subsubsection{Reduction with the gradient HV}

The gradient HV $K_{G}^{1}=\partial _{t}$ is a Lie symmetry of the Laplacian
(\ref{EC.02}) hence we consider the reduction by this vector. The zero order
invariants are $y^{A},w$ and lead to the reduced equation
\begin{equation}
\delta ^{AB}u_{AB}=0  \label{EC.07}
\end{equation}%
which is Laplace equation in the flat space $E^{n-1}$. We consider again
cases.\newline
{\textbf{Case }}$n>3$

In this case the Lie symmetries of (\ref{EC.07}) are given by the vectors
\begin{equation}
K_{G}^{A}~,~X_{R}^{AB}~,~_{n-1}H~,X_{C}^{A}-y^{A}u\partial _{u}.
\end{equation}%
From these the $K_{G}^{A},X_{R}^{AB}$ are inherited symmetries and the rest
- which are produced by the HV and the sp.CKVs of the space $E^{n-1}$ - are
Type II hidden symmetries.\newline

If $n=3,$ the reduced equation (\ref{EC.07}) is the Laplacian in $E^{2},$
hence, admits infinite Lie symmetries. Type II hidden symmetries are
generated from the HV and the CKVs of $E^{2}$.

In the following sections, we study the reduction of the homogeneous heat
equation (\ref{t2.1}) in certain Riemannian spaces.

\section{Reduction of the homogeneous heat equation in certain Riemannian
spaces}

\label{Heatred}

In a general Riemannian space with metric $g_{ij}$ the heat conduction
equation with flux is%
\begin{equation}
\Delta u-u_{t}=q  \label{WH.0}
\end{equation}%
where $\Delta $ is the Laplace operator $\Delta =\frac{1}{\sqrt{g}}\frac{%
\partial }{\partial x^{i}}\left( \sqrt{g}g^{ij}\frac{\partial }{\partial
x^{j}}\right) $ and $q=q(t,x,u)$. Equation (\ref{WH.0}) can also be written
\begin{equation}
g^{ij}u_{ij}-\Gamma ^{i}u_{i}-u_{t}=q  \label{WH.001}
\end{equation}%
where $\Gamma ^{i}=\Gamma _{jk}^{i}g^{jk}$ and $\Gamma _{jk}^{i}$ are the
Christofell Symbols of the metric $g_{ij}$.

For $q=0$, equation (\ref{WH.001}) admits the Lie point symmetries
\begin{equation}
X_{t}=\partial _{t}~\ ,~X_{u}=u\partial _{u}~,~X_{b}=b\left( t,x\right)
\partial _{u}  \label{WH.001a}
\end{equation}%
where $b\left( t,x\right) $ is a solution of the heat equation. These
symmetries are too general to provide sound reductions and consequently
reduced PDEs which can give Type II\ hidden symmetries. However, in Chapter %
\ref{chapter5} it has been shown that there is a close relation between the
Lie point symmetries of the heat equation and the collineations of the
metric. Specifically it has been shown that the Lie point symmetries of the
heat equation are generated from the HV and the KVs of $g_{ij}.$ This
implies that if we want to have new Lie point \ symmetries which will allow
for sound reductions of the heat equation eqn (\ref{WH.001}) we have to
restrict our considerations to spaces which admit a homothetic algebra. Our
intention is to keep the discussion as general as possible therefore we
consider spaces in which the metric $g_{ij}$ can be written in generic form.
The spaces we shall consider are:

a. Spaces which admit a gradient KV (\ref{WH.03}).

b. Spaces which admit a gradient HV (\ref{LEH.01}).

c. Space which admits a nongradient HV acting simply transitive, i.e. Petrov
Type III (\ref{P3.00}).

In what follows all spaces are of dimension $n\succeq 2$. The case $n=1$
although relatively trivial for our approach in general it is not so and has
been studied for example in \cite{Clarkson93,Ivanova2008}.

\subsection{The heat equation in a space which admits a gradient KV}

\label{The heat equation in an 1+n decomposable space}

In the $1+n$ decomposable space with line element (\ref{WH.03}) the heat
equation (\ref{t2.1}) takes the form%
\begin{equation}
u_{zz}+h^{AB}u_{AB}-\Gamma ^{A}u_{B}-u_{t}=0.  \label{WH.04}
\end{equation}%
Application of Theorem \ref{The Lie of the heat equation} gives that (\ref%
{WH.04}) admits the following \emph{extra }Lie point symmetries generated by
the gradient KV $\partial _{x}:$
\begin{equation*}
X_{1}=\partial _{z}~,~X_{2}=t\partial _{z}-\frac{1}{2}zu\partial _{u}
\end{equation*}%
with nonvanishing commutators
\begin{equation}
\left[ X_{t},X_{2}\right] =X_{1}~,~\left[ X_{2},X_{1}\right] =\frac{1}{2}%
X_{u}.  \label{WH.04.a}
\end{equation}

We reduce (\ref{WH.04}) using the zero order invariants of the extra Lie
point \ symmetries $X_{1},X_{2}$.

\subsubsection{Reduction by $X_{1}$}

\label{gradientKV1}

The zero order invariants of $X_{1}$ are%
\begin{equation*}
\tau =t~,~y^{A}~,~w=u.
\end{equation*}%
Taking these invariants as new coordinates eqn (\ref{WH.04}) reduces to
\begin{equation}
_{h}\Delta w-w_{t}=0  \label{WH.05}
\end{equation}%
where $_{h}\Delta $ is the Laplace operator in the~$n-$dimensional space
with metric $h_{AB}:$
\begin{equation}
_{h}\Delta w=h^{AB}w_{AB}-\Gamma ^{A}w_{B}.  \label{WH.05.a}
\end{equation}%
Equation (\ref{WH.05}) is the homogeneous heat eqn (\ref{t2.1}) in the $n~$\
dimensional space with metric $h_{AB}$. \ According to the Theorem \ref{The
Lie of the heat equation} (see Chapter \ref{chapter5}), the Lie symmetries
of this equation are the homothetic algebra of $h_{AB}$.{\LARGE \ }As it has
been already mentioned the homothetic algebras of the $n$ and the $1+n$
metrics are related as follows \cite{TNA}:

a. The KVs of the $n-$ metric are identical with those of the $1+n$ metric.

b. The $1+n$ metric admits a HV\ if the $n$ metric admits one and if $%
_{n}H^{A}$ is the HV of the $n$ - metric then the HV\ of the $1+n$ metric is
given by the expression
\begin{equation}
_{1+n}H^{\mu }=x\delta _{z}^{\mu }+_{n}H^{A}\delta _{A}^{\mu }\qquad \quad
\mu =z,1,...,n.  \label{WH.05.b}
\end{equation}

The above imply that equation (\ref{WH.05}) inherits all symmetries which
are generated from the KVs/HV of the $n-$metric $h_{AB}.$ Hence we do not
have Type II symmetries in this reduction.

\subsubsection{Reduction by $X_{2}$}

\label{gradientKV2}

The zero order invariants of $X_{2}$ are%
\begin{equation*}
\tau =t~,~y^{A}~,~w=ue^{\frac{z^{2}}{4t}}.
\end{equation*}%
Taking these invariants as new coordinates eqn (\ref{WH.04}) reduces to%
\begin{equation}
h^{AB}w_{AB}-\Gamma ^{A}w_{B}-w_{\tau }-\frac{1}{2\tau }w=0  \label{WH.06}
\end{equation}%
or%
\begin{equation*}
_{h}\Delta w-w_{\tau }=\frac{1}{2\tau }w.
\end{equation*}%
This is the nonhomogeneous heat equation with flux $q\left( \tau
,y^{A},w\right) =\frac{1}{2\tau }w$. Application of Theorem \ref{The Lie of
the heat equation with flux} gives the following result\footnote{%
For details see Appendix \ref{apTypeII}.}.

\begin{proposition}
\label{Cor1} The Lie point symmetries of the heat equation (\ref{WH.06}) in
an $n-$dimensional Riemannian space with metric $h_{AB}$ are constructed
form the homothetic algebra of the metric as follows:\newline

a. $Y^{i}$ is a HV/KV.\newline
The Lie symmetry is
\begin{equation}
X=\left( 2c_{2}\psi \tau +c_{1}\right) \partial _{\tau }+c_{2}Y^{i}\partial
_{i}+\left[ \left( -\frac{c_{1}}{2\tau }+a_{0}\right) w+b\left( \tau
,x\right) \right] \partial _{w}
\end{equation}

b. $Y^{i}=S_{J}^{,i}$ is a gradient HV/KV (the index $J$ counts gradient
KVs).\newline
The Lie symmetry is
\begin{equation}
X=\left( \psi T_{0}\tau ^{2}\right) \partial _{\tau }+T_{0}\tau
S_{J}^{,i}\partial _{i}-\left( \frac{1}{2}T_{0}S_{J}+T_{0}\psi \tau \right)
w\partial _{w}
\end{equation}%
where $b\left( \tau ,x\right) $ is a solution of the heat equation (\ref%
{WH.06}).
\end{proposition}

We infer that for this reduction we have the Type II hidden symmetry ~$%
\partial _{\tau }-\frac{1}{2t}w\partial _{w}.$ The rest of the\ Lie point
symmetries are inherited.

\subsection{The heat equation in a space which admits a gradient homothetic
vector}

\label{The heat equation in a (n+1) space}

For the spacetime with line element
\begin{equation*}
ds^{2}=dr^{2}+r^{2}h_{AB}dy^{A}dy^{B}
\end{equation*}%
the homogeneous heat equation becomes%
\begin{equation}
u_{rr}+\frac{1}{r^{2}}h^{AB}u_{AB}+\frac{\left( n-1\right) }{r}u_{r}-\frac{1%
}{r^{2}}\Gamma ^{A}u_{A}-u_{t}=0  \label{WH.11}
\end{equation}%
where $\Gamma ^{A}=\Gamma _{BC}^{A}h^{BC}$ and $\Gamma _{BC}^{A}$ are the
connection coefficients of the Riemannian metric $h_{AB}$ ($%
A,B,C=1,2,...,n). $ Application of Theorem \ref{The Lie of the heat equation}
gives that the heat equation (\ref{WH.11}) admits the following \emph{extra }%
Lie point symmetries generated by the gradient homothetic vector{\LARGE \ }%
\begin{equation}
~\bar{X}_{1}=2t\partial _{t}+r\partial _{r}~\ ,~~\bar{X}_{2}=t^{2}\partial
_{t}+tr\partial _{r}-\left( \frac{1}{4}r^{2}+\frac{n}{2}t\right) u\partial
_{u}~\,  \label{WH.11b}
\end{equation}%
with nonzero commutators
\begin{equation*}
\left[ X_{t},\bar{X}_{1}\right] =2X_{t}~~,~~\left[ \bar{X}_{1},\bar{X}_{2}%
\right] =2X_{t}
\end{equation*}%
\begin{equation*}
\left[ X_{t},\bar{X}_{2}\right] =\bar{X}_{1}-\frac{n}{2}X_{u}.
\end{equation*}%
We consider again the reduction of (\ref{WH.11}) using the zero order
invariants of these extra Lie point symmetries.

\subsubsection{Reduction by $\bar{X}_{1}$}

\label{HV1}

The zero order invariants of $\bar{X}_{1}$ are
\begin{equation*}
\phi =\frac{r}{\sqrt{t}}~,~w=u,~y^{A}.
\end{equation*}%
We choose $w=w\left( \phi ,y^{A}\right) $ as the dependent variable.

Replacing in (\ref{WH.11}) we find the reduced PDE
\begin{equation}
w_{\phi \phi }+\frac{1}{\phi ^{2}}h^{AB}w_{AB}+\frac{\left( n-1\right) }{%
\phi }w_{\phi }+\frac{\phi }{2}w_{\phi }-\frac{1}{\phi ^{2}}\Gamma
^{A}w_{A}=0.  \label{WH.12}
\end{equation}
Consider a nonvanishing function $N^{2}\left( \phi \right) $ and divide (\ref%
{WH.12}) with $N^{2}\left( \phi \right) $ \ to get:%
\begin{equation}
\frac{1}{N^{2}}w_{\phi \phi }+\frac{1}{\phi^{2} N^{2}}h^{AB}w_{AB}+\frac{%
\left( n-1\right) }{{\phi}N^{2}}w_{\phi }+\frac{\phi }{2N^{2}}w_{\phi }-%
\frac{1}{\phi ^{2}N^{2}}\Gamma ^{A}w_{A}=0  \label{WH.12C}
\end{equation}
It follows that \ (for $n>2$) equation (\ref{WH.12C})\ can be written as%
\begin{equation}
_{\bar{g}}\Delta w=0  \label{WH.12d}
\end{equation}%
where $_{\bar{g}}\Delta $ $\ $\ is the Laplace operator if \ $N^{2}\left(
\phi \right) =\exp \left( \frac{\phi ^{2}}{2\left( n-2\right) }\right) $ and
$\bar{g}_{ij}$ is the conformally related metric
\begin{equation}
d\bar{s}^{2}=\exp \left( \frac{\phi ^{2}}{2\left( n-2\right) }\right) \left(
d\phi ^{2}+\phi ^{2}h_{AB}dy^{A}dy^{B}\right).  \label{WH.12e}
\end{equation}

According to Theorem \ref{KG1}\ the Lie point symmetries of (\ref{WH.12d})
are the CKVs of the metric (\ref{WH.12e}) \ whose conformal factor satisfies
the condition $_{\bar{g}}\Delta \psi =0.$ Therefore, Type II hidden
symmetries will be generated from the proper CKVs. The existence and the
number of these vectors depends mainly on the $n$ metric $h_{AB}.$

\subsubsection{Reduction by $\bar{X}_{2}$\label{HV2}}

For $\bar{X}_{2}$ the zero order invariants are
\begin{equation*}
\phi =\frac{r}{t}~,~w=ut^{\frac{n}{2}}e^{\frac{r^{2}}{4t}},~y^{A}.
\end{equation*}%
We choose $w=w\left( \phi ,y^{A}\right) $ as the dependent variable and we
have the reduced equation%
\begin{equation}
_{g}\Delta w=0  \label{WH.13}
\end{equation}%
where%
\begin{equation}
_{g}\Delta w=w_{\phi \phi }+\frac{\left( n-1\right) }{\phi }w_{\phi }+\frac{1%
}{\phi ^{2}}h^{AB}w_{AB}-\frac{1}{\phi ^{2}}\Gamma ^{A}w_{A}.
\end{equation}%
Equation (\ref{WH.13}) is the Laplace equation in the space $\left( \phi
,y^{A}\right) $ with metric%
\begin{equation}
ds^{2}=d\phi ^{2}+\phi ^{2}h_{AB}dy^{A}dy^{B}.  \label{WH.14}
\end{equation}

The Lie point symmetries of Laplace equation (\ref{WH.13}) are given in
Theorem \ref{KG1}. As in the last case the existence and the number of these
vectors depends mainly on the $n$ metric $h_{AB}.$

We note that both vectors $\bar{X}_{1},\bar{X}_{2}$ \ are generated form the
gradient HV and in both cases the heat equation is reduced to Laplace
equation. This gives the following

\begin{proposition}
\label{Prop}The reduction of the heat equation (\ref{t2.1}) in a space with
metric (\ref{LEH.01}) $\left( n>2\right) ~$by means of the Lie symmetries
generated by the gradient HV leads to Laplace equation $\Delta u=0$, where
\thinspace $\Delta $ is the Laplace operator for the metric (\ref{WH.12e})
if the reduction is done by $\bar{X}_{1}$ and for the metric (\ref{WH.14})
if the reduction is done by $\bar{X}_{2}$.
\end{proposition}

\section{Applications of the reduction of the homogeneous heat equation}

\label{ApplicationsHeatre}

In this section, we consider applications of the general results of section %
\ref{Heatred} in various spacetimes.

\subsection{The heat equation in a $1+n$ decomposable space}

Consider the $1+n$ decomposable space%
\begin{equation}
ds^{2}=dx^{2}+N^{-2}\left( y^{C}\right) \delta _{AB}y^{A}y^{B}  \label{WH.17}
\end{equation}%
where $N\left( y^{C}\right) =\left( 1+\frac{K}{4}y^{C}y_{C}\right) ,$ that
is, the $n$ space is a space of constant non vanishing ($K\neq 0)$
curvature. The space (\ref{WH.17}) \ does not admit proper HV, however,
admits $\frac{n\left( n-1\right) }{2}+n~$nongradient KVs and $1$ gradient KV
as follows \cite{TNA}%
\begin{eqnarray*}
1~\text{gradient KV}\text{:~} &&\partial _{x} \\
n~\text{nongradient KVs} &\text{:}&\text{\ }K_{V}=\frac{1}{N}\left[ \left(
2N-1\right) \delta _{I}^{i}+\frac{K}{2}Nx_{I}x^{i}\right] \partial _{i} \\
\frac{n\left( n-1\right) }{2}\text{nongradient KVs} &\text{:}&\text{\ }%
X_{IJ}=\delta _{\lbrack I}^{j}\delta _{J]}^{i}\partial _{i}
\end{eqnarray*}%
In a space with metric (\ref{WH.17}), the homogeneous heat equation takes
the form%
\begin{equation}
u_{xx}+N^{2}\left( y^{C}\right) \delta ^{AB}u_{AB}-\frac{N}{2}%
Ky^{A}u_{A}-u_{t}=0  \label{WH.18}
\end{equation}%
Applying Theorem \ref{The Lie of the heat equation} we find that equation (%
\ref{WH.18}) admits the extra Lie point \ symmetries
\begin{equation}
\partial _{x}~,~t\partial _{x}-\frac{1}{2}xu\partial _{x}~,K_{V}~,X_{IJ}.
\end{equation}%
The Lie point \ symmetries which are generated by the gradient KV are%
\footnote{%
Here the algebra is the one given in section \ref{The heat equation in an
1+n decomposable space} and a separate algebra is the algebra of the KVs of
the space of constant curvature. More specifically the KVs $K_{V}~,X_{IJ}$
commute with all other symmetries but not between themselves} $\partial
_{x}~,~t\partial _{x}-\frac{1}{2}xu\partial _{x}$.\

Reduction of (\ref{WH.18}) by means of the gradient KV $\partial _{x}$
results in the special form of equation (\ref{WH.05})
\begin{equation}
\frac{1}{N^{2}\left( y^{C}\right) }\delta ^{AB}u_{AB}-\frac{N}{2}%
Ky^{A}u_{A}-u_{t}=0.  \label{WH.19}
\end{equation}%
This is the homogeneous heat equation in an $n$- dimensional space of
constant curvature. The Lie point \ symmetries of this equation have been
determined in section \ref{hhe} and are inherited symmetries. Hence, in this
case, we do not have Type II\ hidden symmetries.

Reduction of (\ref{WH.18}) with the Lie symmetry $~t\partial _{x}-\frac{1}{2}%
xu\partial _{x}$\ gives that the reduced equation (\ref{WH.06}) is
\begin{equation}
N^{2}\left( y^{C}\right) \delta ^{AB}w_{AB}-\frac{N}{2}Ky^{A}w_{A}-w_{\tau }=%
\frac{1}{2\tau }w~,~w=ue^{\frac{x^{2}}{4t}}  \label{WH.20}
\end{equation}%
which is the heat equation with flux. By Proposition \ref{Cor1}, the Lie
point \ symmetries of (\ref{WH.20}) are:%
\begin{equation}
X=c_{1}\partial _{\tau }+\left( K_{V}~+X_{IJ}\right) +\left[ \left( -\frac{%
c_{1}}{2\tau }+a_{0}\right) w+b\left( \tau ,y^{C}\right) \right] \partial
_{w}  \label{WH.21}
\end{equation}%
where $c_{1},a_{0}$ are constants. From section \ref{gradientKV2} we have
that Type II hidden symmetry is the one defined by the constant $c_{1}$.

\subsection{FRW space-time with a gradient HV}

Consider the spatially flat FRW\ metric%
\begin{equation}
ds^{2}=d\sigma ^{2}-\sigma ^{2}\left( dx^{2}+dy^{2}+dz^{2}\right)
\label{WH.22a}
\end{equation}%
This metric admits the gradient HV \cite{MM86}%
\begin{equation*}
H=\sigma \partial _{\sigma }~~\left( \psi _{H}=1\right)
\end{equation*}%
and six nongradient KVs%
\begin{equation*}
X_{1-3}=\partial _{y^{A}}~~,~~X_{4-6}=y^{B}\partial _{A}-y^{A}\partial _{B}.
\end{equation*}%
where $y^{A}=\left( x,y,z\right) .~$

In this space the heat equation takes the form%
\begin{equation}
u_{\sigma \sigma }-\frac{1}{\sigma ^{2}}\left( u_{xx}+u_{yy}+u_{zz}\right) +%
\frac{3}{\sigma }u_{\sigma }-u_{t}=0.  \label{WH.22}
\end{equation}%
The Lie point symmetries of (\ref{WH.22}) are%
\begin{eqnarray*}
&&\partial _{t}~,~u\partial _{u}~,~b\left( \tau ,y^{A}\right) \partial
_{u}~,~X_{1-3}~,~X_{4-6}~, \\
H_{1} &=&2t\partial _{t}+\sigma \partial _{\sigma }~,~H_{2}=t^{2}\partial
_{t}+t\sigma \partial _{\sigma }-\left( \frac{1}{4}\sigma ^{2}+2t\right)
u\partial _{u}~.
\end{eqnarray*}%
The Lie point \ symmetries $H_{1},H_{2}$ are produced by the gradient HV
therefore we use them to reduce (\ref{WH.22}). We note that this case is a
special case of the one we considered in section \ref{The heat equation in a
(n+1) space}\ for $h_{AB}=\delta _{AB}$.

Reduction by $H_{1}$ gives that (\ref{WH.22})~becomes:
\begin{equation}
w_{\phi \phi }-\frac{1}{\phi ^{2}}\left( w_{xx}+w_{yy}+w_{zz}\right) +\left(
\frac{3}{\phi }+\frac{\phi }{2}\right) w_{\phi }=0  \label{WH.23}
\end{equation}%
where $\phi =\frac{r}{\sqrt{\sigma }}~,~w=u.$ This is a special form of (\ref%
{WH.12}).

Dividing with $N^{2}\left( \phi \right) =\exp \left( \frac{\phi ^{2}}{4}%
\right) $ we find that (\ref{WH.23}) is written as
\begin{equation}
_{\bar{g}}\Delta w=0  \label{WH.23a}
\end{equation}%
where the metric $\bar{g}_{ij}$ is the conformally related metric of \ (\ref%
{WH.22a}):%
\begin{equation}
d\bar{s}^{2}=e^{\frac{\phi ^{2}}{4}}\left( d\phi ^{2}-\phi ^{2}\left(
dx^{2}+dy^{2}+dz^{2}\right) \right)  \label{WH.23aa}
\end{equation}

We have that the Lie symmetries of (\ref{WH.23a}) are generated from
elements of the conformal algebra of the space whose conformal factors
satisfy condition $_{\bar{g}}\Delta \psi =0.$ The metric (\ref{WH.23a}) is
conformally flat therefore its conformal group is the same with that of the
flat space \cite{TNA,MM86}, however with different subgroups. We find that
these vectors (i.e. the Lie symmetries)\ are the vectors%
\begin{equation}
~X_{1-3}~,~X_{4-6}~,\partial _{t},\text{ }w\partial _{w}~,~b_{0}\left( \phi
,y^{A}\right) \partial _{w}.~  \label{WH.24}
\end{equation}%
We conclude that there are no Type II symmetries for this reduction.

Using reduction by $H_{2}$ we find that (\ref{WH.22})~reduces to :%
\begin{equation}
w_{\phi \phi }-\frac{1}{\phi ^{2}}\left( w_{xx}+w_{yy}+w_{zz}\right) +\frac{3%
}{\phi }w_{\phi }=0  \label{WH.25}
\end{equation}%
where $\phi =\frac{\tau }{\tau }~,~w=ut^{2}e^{\frac{\tau ^{2}}{4t}}.~$\ This
is a special form of (\ref{WH.13}) which is the Laplace equation. In this
case the results of \cite{Bozhkov} apply and we infer that the Lie point \
symmetries of (\ref{WH.25}) are: \
\begin{eqnarray}
&&X_{1-3}~,~X_{4-6}~,~~w\partial _{w}~,~b_{1}\left( \phi ,y^{A}\right)
\partial _{w}  \notag \\
X_{7} &=&\phi \partial _{\phi }~,~X_{8-10}=\phi y^{A}\partial _{\phi }+\ln
\phi \partial _{A}-y^{A}w\partial _{w}.
\end{eqnarray}%
The vector $X_{7}$ is the proper HV of the metric and the vectors $X_{8-10}$
the proper CKVs which are not special CKVs, therefore these vectors are Type
II\ hidden symmetries. A further analysis of (\ref{WH.25}) can be found in
\cite{Kara}.

\section{The Heat equation in Petrov type III spacetime}

\label{PetrovIIIHeat} In this section we consider the special class of
Petrov type III spacetime which admits a nongradient HV which acts simply
transitively.

The metric of the Petrov type III space-time is
\begin{equation}
ds^{2}=2d\rho dv+\frac{3}{2}xd\rho ^{2}+\frac{v^{2}}{x^{3}}\left(
dx^{2}+dy^{2}\right)  \label{WH.III}
\end{equation}%
with Homothetic algebra%
\begin{eqnarray*}
K^{1} &=&\partial _{\rho }~,~K^{2}=\partial _{y}~,~K^{3}=v\partial _{v}-\rho
\partial _{\rho }+2x\partial _{x}+2y\partial _{y} \\
H &=&v\partial _{v}+\rho \partial _{\rho }~\left( \psi _{H}=1\right)
\end{eqnarray*}%
where $K^{1-4}$ are KVs and $H$ is a nongradient HV.

In this space-time equation (\ref{t2.1}) takes the form:%
\begin{equation}
-\frac{3}{2}xu_{vv}+2u_{v\rho }+\frac{x^{3}}{v^{2}}\left(
u_{xx}+u_{yy}\right) -3\frac{x}{v}u_{v}+\frac{2}{v}u_{\rho }-u_{t}=0.
\end{equation}%
From\ Theorem\ \ref{The Lie of the heat equation} we have\ that the extra
Lie point \ symmetries are the vectors
\begin{equation*}
X_{1-3}=K_{1-3}~,~~X_{4}=2t\partial _{t}+H
\end{equation*}%
with nonzero commutators:%
\begin{equation*}
\left[ X_{t},X_{4}\right] =2X_{t}~,~\left[ X_{2,},X_{3}\right] =2X_{2}
\end{equation*}%
\begin{equation*}
\left[ X_{3},X_{1}\right] =X_{1}~,~~\left[ X_{1},X_{4}\right] =X_{1}~.
\end{equation*}%
We use $X_{4}$ for reduction because this is the Lie point symmetry
generated by the HV.

The zero order invariants of $X_{4}$ are%
\begin{equation*}
\alpha =\frac{v}{\sqrt{t}}~,~\beta =\frac{\rho }{\sqrt{t}}~,~\gamma
=x~,~\delta =y~,~w=u.
\end{equation*}%
We choose $w=w\left( \alpha ,\beta ,\gamma ,\delta \right) $ as the
dependent variable and we find that the reduced PDE is%
\begin{equation}
_{III}\Delta w+\frac{\alpha }{2}w_{\alpha }+\frac{\beta }{2}w_{\beta }=0
\label{WH.III1}
\end{equation}%
where $_{III}\Delta $ is the Laplace operator for metric (\ref{WH.III}).

It is clear that the second order PDE\ (\ref{WH.III1}) admits only the Lie
symmetries ~$X_{2},X_{3}.$ Therefore,\ we do not have Type II hidden
symmetries and the symmetries \ $X_{1},X_{4}~$ are Type I\ hidden symmetries.

\section{Conclusion}

Up to now in the literature the study of Type II\ hidden symmetries has been
done by counter examples or by considering very special PDEs and in low
dimensional flat spaces. In this chapter we have improved this scenario and
have studied the problem of Type II\ hidden symmetries of second order PDEs
from a geometric of view in $n~$ dimensional Riemannian spaces. We have
considered the reduction of the Laplace and of the homogeneous heat equation
and the consequent possibility of existence of Type II hidden symmetries in
some general classes of spaces which admit some kind of symmetry; hence,
they admit nontrivial Lie symmetries .

For the Laplace equation, the conclusion of this study is that the Type II\
hidden symmetries are directly related to the transition of the CKVs from
the space where the original equation is defined to the space where the
reduced equation resides. In this sense, we related the Lie point symmetries
of PDEs with the basic collineations of the metric i.e. the CKVs. \
Concerning the general results of the reduction of Laplace equation we can
summarize them as follows:

\begin{itemize}
\item If we reduce the Laplace equation with a gradient KV the reduced
equation is a Laplace equation in the non-decomposable space. In this case,
the Type II hidden symmetries are generated from the special and the proper
CKVs of the non-decomposable space.

\item If we reduce the Laplace equation with a gradient HV the reduced
equation is a Laplace equation for an appropriate metric. In this case, the
Type II hidden symmetries are generated from the HV and the special/proper
CKVs.

\item If we reduce the Laplace equation with the symmetry generated by a
sp.CKV, the reduced equation is the Klein Gordon equation for an appropriate
metric that inherits the Lie point \ symmetry generated by the gradient HV.
In this case, the Type II hidden symmetries are generated from the proper
CKVs.
\end{itemize}

We also considered the reduction of Laplace equation in some spaces of
interest in which the metric does not admit the symmetries of the previous
cases. In this context, we showed that the reduction with the non-gradient
HV in the Petrov type III does not give any Type II\ hidden symmetries.
Also, it is of interest the reduction of Laplace equation (i.e. the wave
equation) in Minkowski spaces $M^{4}$ \ and $M^{3}$ where we recover the
results of \cite{AGA06} in a straightforward manner without the need of \
computer software. Finally we considered an $n-$dimensional flat FRW like
space and showed that the reduction with the proper CKV does not produce
Type II hidden symmetries.

Moreover, we applied the zero order invariants of the Lie symmetries in
order to reduce the number of independent variables of the homogeneous heat
equation in certain general classes of Riemannian spaces, which admit some
type of basic symmetry. For each reduction, we determined the origin of Type
II\ hidden symmetries. The spaces we considered are the spaces which admit a
gradient KV, a gradient HV and finally spacetimes which admits a HV\ which
acts simply and transitively. For the reduction of the homogeneous heat
equation and the existence of Type II hidden symmetries, we found the
following general geometric results:

\begin{itemize}
\item If we reduce the homogeneous heat equation via the symmetries which
are generated by a gradient KV~$\left( S^{,i}\right) $ the reduced equation
is a heat equation in the nondecomposable space. In this case we have the
Type II hidden symmetry $\partial _{t}-\frac{1}{2t}w\partial _{w}$ provided
we reduce the heat equation with the symmetry~$tS^{,i}-\frac{1}{2}Su\partial
_{u}$.

\item If we reduce the homogeneous heat equation via the symmetries which
are generated by a gradient HV the reduced equation is Laplace equation for
an appropriate metric. In this case the Type II hidden symmetries are
generated from the proper CKVs.

\item In Petrov type III spacetime, the reduction of the homogeneous heat
equation via the symmetry generated from the nongradient HV gives a PDE
which inherits the Lie point \ symmetries, hence no Type II hidden
symmetries are admitted.
\end{itemize}

The above results can be used in many important space-times and help
facilitate the solution of the heat equation in these space-times.

\newpage%

\begin{subappendices}%

\section{Proof of Corollary}

\label{apTypeII}

\begin{proof}[Proof of Corollary \protect\ref{Cor1}]
Using Theorem \ref{The Lie of the heat equation with flux} and replacing $q$
we have

For case a)
\begin{eqnarray}
-a_{\tau }w+H\left( b\right) -\frac{1}{2\tau }\left( aw+b\right) +\frac{a}{%
2\tau }w-\left( \psi c_{2}w+\frac{1}{2\tau }wc_{1}\right) _{\tau } &=&0. \\
-a_{\tau }w+H\left( b\right) -\frac{1}{2\tau }b+\frac{1}{2\tau ^{2}}wc_{1}
&=&0 \\
\left[ -a_{\tau }+\frac{c_{1}}{2\tau ^{2}}\right] w+\left[ H\left( b\right) -%
\frac{1}{2\tau }b\right] &=&0
\end{eqnarray}%
that is%
\begin{equation}
a=-\frac{c_{1}}{2\tau }+a_{0}~,~~H\left( b\right) -\frac{1}{2\tau }b=0
\end{equation}

For case b)%
\begin{equation}
0=\left( -\frac{1}{2}T_{,\tau }\psi +\frac{1}{2}T_{,\tau \tau }S-F_{,\tau
}\right) w-\left( 2\psi q\int Td\tau \right) _{\tau }-Tq_{,i}S^{,i}.  \notag
\end{equation}%
then\qquad
\begin{equation}
0=\left( -\frac{1}{2}T_{,\tau }\psi +\frac{1}{2}T_{,\tau \tau }S-F_{,\tau
}\right) w+\frac{\psi }{\tau ^{2}}\int Td\tau ~w-\frac{\psi }{\tau }Tw
\end{equation}%
from here we have%
\begin{equation}
T_{,\tau \tau }=0\rightarrow ~T=T_{0}\tau +T_{1}
\end{equation}%
and%
\begin{equation*}
F=-T_{0}\psi \tau .
\end{equation*}
\end{proof}

\section{The homogeneous heat equation in the Petrov spacetimes}

\label{T2Petrov}

In the following subsections, we study the reduction of the homogeneous heat
equation in the Petrov spacetimes of type N,D and II.

\subsection{Petrov type N}

The metric of the Petrov type N space-time is
\begin{equation}
ds^{2}=dx^{2}+x^{2}dy^{2}+2d\rho dv+\ln x^{2}d\rho ^{2}  \label{WH.07b}
\end{equation}%
and has the homothetic algebra \cite{Steele1991b}%
\begin{eqnarray*}
K^{1} &=&\partial _{\rho }~,~K^{2}=\partial _{v}~,~K^{3}=\partial _{y} \\
H &=&x\partial _{x}+\rho \partial _{\rho }+\left( v-2\rho \right) \partial
_{v}~~\left( \psi _{H}=1\right)
\end{eqnarray*}%
where $K^{1-3}$ are KVs and $H$ is a nongradient HV.

The heat equation (\ref{t2.1}) in this space-time is
\begin{equation}
u_{xx}+\frac{1}{x^{2}}u_{yy}+2u_{\rho v}-2\ln x^{2}~u_{vv}+\frac{1}{x}%
u_{x}-u_{t}=0.  \label{WH.07}
\end{equation}%
Application of Theorem \ref{The Lie of the heat equation} gives that the
extra Lie point symmetries of (\ref{WH.07}) are
\begin{equation*}
X_{1-3}=K_{1-3}~\ ,~X_{4}=2t\partial _{t}+H
\end{equation*}%
with nonzero commutators
\begin{equation*}
\left[ X_{t},X_{4}\right] =2X_{t}
\end{equation*}%
\begin{equation*}
\left[ X_{1},X_{4}\right] =X_{1}-2X_{2}~,~\left[ X_{2},X_{4}\right] =X_{2}.~
\end{equation*}

We use $X_{4}$ to reduce the PDE because this is the Lie symmetry generated
by the HV. The zero order invariants of $X_{4}$ are\qquad\
\begin{equation}
\alpha =\frac{x}{\sqrt{t}}~,~\beta =\frac{\rho }{\sqrt{t}}~,~\gamma =\frac{%
v+\rho \ln \left( t\right) }{\sqrt{t}}~,~\delta =y~,~w=u.  \label{WW.0a}
\end{equation}%
Choosing $\alpha ,\beta ,\gamma ,\delta $ as the independent variables and $%
w=w\left( \alpha ,\beta ,\gamma ,\delta \right) $ as the dependent variable
we find that the reduced PDE is%
\begin{equation}
_{N}\Delta w+\left( \frac{1}{2}\alpha w_{\alpha }+\frac{1}{2}\beta w_{\beta
}+\left( \frac{1}{2}\gamma -\beta \right) w_{\gamma }\right) =0.
\label{WH.08}
\end{equation}%
where $_{N}\Delta $ is the Laplace operator for the metric (\ref{WH.07b}).

Equation (\ref{WH.08}) is of the form (\ref{GPE.30.1}) with
\begin{equation*}
A_{ij}=g_{ij}\left( x^{k}\right) ,B^{i}\left( x^{k}\right) =\Gamma ^{i}+%
\frac{1}{2}\alpha \delta _{\alpha }^{i}+\frac{1}{2}\beta \delta _{\beta
}^{i}+\left( \frac{1}{2}\gamma -\beta \right) \delta _{\gamma }^{i},~f\left(
x^{k},u\right) =0
\end{equation*}%
where $g_{ij}$ is the metric (\ref{WH.07b}). Replacing in equations (\ref%
{GPE.42})-(\ref{GPE.46}) we obtain the Lie symmetry conditions for (\ref%
{WH.08}). \ Because $A_{ij,u}=0$ it follows from equation (\ref{GPE.44})
that the Lie point \ symmetries are generated from the CKVs of the metric (%
\ref{WH.07b}). However taking into consideration the rest of the symmetry
conditions we find that the only Lie symmetry which remains is the one of
the KV\ $X_{3}.$ We conclude that in this reduction we do not have Type II\
hidden symmetries.

\subsection{Petrov type D}

The metric of the Petrov type D space-time is
\begin{equation}
ds^{2}=-dx^{2}+x^{-\frac{2}{3}}dy^{2}-x^{\frac{4}{3}}\left( d\rho
^{2}+dz^{2}\right)  \label{WH.09b}
\end{equation}%
with Homothetic algebra%
\begin{eqnarray*}
K^{1} &=&\partial _{\rho }~,~K^{2}=\partial _{z}~,~K^{3}=\partial
_{y}~,~K^{4}=z\partial _{\rho }-\rho \partial _{z} \\
H &=&x\partial _{x}+\frac{4}{3}y\partial _{y}+\frac{z}{3}\partial _{z}+\frac{%
\rho }{3}\partial _{\rho }~~\left( \psi _{H}=1\right)
\end{eqnarray*}%
where $K^{1-4}$ are KVs and $H$ is a nongradient HV.

In this space-time the heat equation (\ref{t2.1}) takes the form:
\begin{equation}
-u_{xx}+x^{\frac{2}{3}}u_{yy}-x^{-\frac{4}{3}}\left( u_{\rho \rho
}+u_{zz}\right) -\frac{1}{x}u_{x}-u_{t}=0.  \label{WH.09}
\end{equation}%
From\ Theorem\ \ref{The Lie of the heat equation} we have\ that the extra
Lie point \ symmetries are the vectors
\begin{equation*}
X_{1-4}=K_{1-4}~,~~X_{5}=2t\partial _{t}+H.
\end{equation*}%
with nonzero commutators:%
\begin{eqnarray*}
\left[ X_{t},X_{5}\right] &=&2X_{t} \\
\left[ X_{1},X_{5}\right] &=&\frac{1}{3}X_{1}~,~\left[ X_{4},X_{1}\right]
=-X_{2} \\
\left[ X_{2},X_{4}\right] &=&X_{1}~,~\left[ X_{2},X_{5}\right] =\frac{1}{3}%
X_{2} \\
\left[ X_{3},X_{5}\right] &=&\frac{4}{3}X_{3}.
\end{eqnarray*}

We use $X_{5}$ to reduce the PDE because this is the Lie symmetry generated
by the HV. The zero order invariants of $X_{5}$ are%
\begin{equation*}
\alpha =\frac{x}{t^{\frac{1}{2}}}~,~\beta =\frac{y}{t^{\frac{2}{3}}}%
~,~\gamma =\frac{\rho }{t^{\frac{1}{6}}}~,~\delta =\frac{z}{t^{\frac{1}{6}}}%
~,~w=u.
\end{equation*}%
We choose $\alpha ,\beta ,\gamma ,\delta $ as the independent variables and $%
w=w\left( \alpha ,\beta ,\gamma ,\delta \right) $ as the dependent variable
and we find that the reduced PDE is%
\begin{equation}
_{D}\Delta w+\left( \frac{1}{2}aw_{\alpha }+\frac{2}{3}\beta w_{\beta }+%
\frac{1}{6}\gamma w_{\gamma }+\frac{1}{6}\delta w_{\delta }\right) =0
\label{WH.10a}
\end{equation}%
where~$_{D}\Delta $ is the Laplace operator with metric (\ref{WH.09b}).

Working again with the Lie symmetry conditions (\ref{GPE.42})-(\ref{GPE.46})
we find that equation (\ref{WH.10a}) admits as Lie point symmetry the vector
$X_{4}$ only which is an inherited symmetry. Hence we do not have Type II\
hidden symmetries.\ Obviously the Lie point \ symmetries $X_{1-4}$\ are Type
I\ hidden symmetries for equation (\ref{WH.10a})\ for the reduction by $%
X_{5}.$

\subsection{Petrov type II}

The metric of the Petrov type II space-time is
\begin{equation}
ds^{2}=\rho ^{-\frac{1}{2}}\left( d\rho ^{2}+dz^{2}\right) -2\rho dxdy+\rho
\ln \rho ~dy^{2}  \label{WH.II1}
\end{equation}%
with homothetic algebra%
\begin{eqnarray*}
K^{1} &=&\partial _{x}~,~K^{2}=\partial _{y}~,~K^{3}=\partial _{z}~ \\
H &=&\frac{1}{3}\left( x+2y\right) \partial _{x}+\frac{1}{3}y\partial _{y}+%
\frac{4}{3}z\partial _{z}+\frac{4}{3}\rho \partial _{\rho }~\left( \psi
_{H}=1\right)
\end{eqnarray*}%
where $K^{1-4}$ are KVs and $H$ is a nongradient HV.

In this space-time equation (\ref{t2.1}) takes the form:%
\begin{equation}
\rho ^{\frac{1}{2}}\left( u_{\rho \rho }+u_{zz}\right) -\frac{1}{\rho }%
\varepsilon \ln \rho u_{xx}-\frac{2}{\rho }u_{xy}+\rho ^{-\frac{1}{2}%
}u_{\rho }-u_{t}=0.
\end{equation}%
From Theorem \ref{The Lie of the heat equation} we have that the extra Lie
point \ symmetries are the vectors
\begin{equation*}
X_{1-3}=K_{1-3}~,~~X_{4}=2t\partial _{t}+H
\end{equation*}%
with nonzero commutators:%
\begin{equation*}
\left[ X_{t},X_{4}\right] =2X_{t}
\end{equation*}%
\begin{eqnarray*}
\left[ X_{1},X_{4}\right] &=&\frac{1}{3}X_{1}~,~\left[ X_{3},X_{4}\right] =%
\frac{4}{3}X_{3} \\
\left[ X_{2},X_{4}\right] &=&\frac{2}{3}X_{1}+\frac{1}{3}X_{2}.
\end{eqnarray*}%
We use $X_{4}$ to reduce the PDE because this is the Lie symmetry generated
by the HV. The zero order invariants of $X_{4}$ are%
\begin{equation*}
\alpha =\frac{\rho }{t^{\frac{2}{3}}}~,~\beta =\frac{z}{t^{\frac{2}{3}}}%
~,~\gamma =\frac{x-\frac{1}{3}y\ln \left( t\right) }{t^{\frac{1}{6}}}%
~,~\delta =\frac{y}{t^{\frac{1}{6}}}~~,~w=u
\end{equation*}%
We choose $\alpha ,\beta ,\gamma ,\delta $ as the independent variables and $%
w=w\left( \alpha ,\beta ,\gamma ,\delta \right) $ as the dependent variable
and we find that the reduced PDE is%
\begin{equation}
_{II}\Delta w+\frac{2}{3}aw_{\alpha }+\frac{2}{3}\beta w_{\beta }+\left(
\frac{1}{3}\delta -\frac{1}{6}\gamma \right) w_{\gamma }+\frac{1}{6}\delta
w_{\delta }=0  \label{WH.II}
\end{equation}%
where $_{II}\Delta $ is the Laplace operator for metric (\ref{WH.II1}).

From the Lie symmetry conditions (\ref{GPE.42})-(\ref{GPE.46}) it follows
that (\ref{WH.II}) does not admit any Lie point \ symmetries. Hence, we do
not have Type II\ hidden symmetries in this case.

\end{subappendices}%

\part{Noether symmetries and theories of gravity}

\chapter{Noether symmetries in Scalar field Cosmology\label{chapter8}}

\section{Introduction}

The detailed analysis of the cosmological data indicate that the Universe is
spatially flat and has suffered two acceleration phases. An early
acceleration phase (inflation), which occurred prior to the
radiation-dominated era, and a recently initiated accelerated expansion \cite%
{Lima00,Tegmark04,Spergel06,Davis07,Kowalski08,Hicken09,Komatsu09}. The
source for the late time cosmic acceleration has been attributed to an
unidentified type of matter, the dark energy. Dark energy contrary to the
ordinary baryonic matter has a negative pressure, i.e. a negative equation
of state parameter, which counteracts the gravitational force and leads to
the observed accelerated expansion.

The simplest dark energy probe is the cosmological constant leading to the $%
\Lambda $CDM cosmology \cite{Weinberg89,Peebles03,Padmanabhan03}. However,
it has been shown that $\Lambda $CDM cosmology suffers from two major
drawbacks known as the fine tuning problem and the coincidence problem \cite%
{Periv08}. Besides $\Lambda $CDM cosmology, many other candidates have been
proposed in the literature. Most are based either on the existence of new
fields (i.e. a scalar field) or in some modification of the Einstein Hilbert
action (see \cite{AmeBook}).

In addition to dark energy and the ordinary baryonic matter, it is believed
that the Universe contains a third type of matter, the dark matter. This
type of matter is assumed to be pressureless (non-relativistic) and
interacts very weakly with the standard baryonic matter. Therefore, its
presence is mainly inferred from gravitational effects on visible matter.

In the following, we consider scalar field cosmology (minimally coupled
scalar field and non minimally coupled scalar field) and we propose a
geometric principle (`selection rule') for specifying the potential $V(\phi
) $ and the coupling function $F\left( \phi \right) $ of the scalar field in
order to solve analytically the system of the resulting field equations. We
propose that the scalar field model should be selected by the geometric
requirement that the dynamical system of the field equations admits Noether
point symmetries. The main reason for the consideration of this hypothesis
is that the Noether symmetries provide first integrals, which assist the
integrability of the system. Furthermore, as we saw in chapter \ref{chapter3}
the Noether symmetries are generated from the elements of the homothetic
algebra of the kinetic metric of the Lagrangian of the theory. Therefore,
with this assumption we let the theory to select the potential, i.e. the
dark energy model.

The idea to use Noether symmetries in cosmology, either on scalar field
models and on modified theories of gravity is not new and indeed a lot of
attention has been paid in the literature \cite%
{deRiti90,Cap93deR,CAP94,Cap97M,Cap02,Cap07Nes,CapP09,CapHam,KotsakisL,VakF,VakB,Vak07,Sanyal05,Sanyal10,Jamil2011315,Wei2012298,Motavali08,Kucuk13}%
. However our approach is geometric and more fundamental.

The structure of the present chapter is as follows. In sections \ref%
{Conformally equivalent Lagrangians and scalar field Cosmology} and \ref%
{Generalization to dimension n} we discuss the conformal equivalence of
Lagrangians for scalar fields in a Riemannian space of dimension $4$ and $n$
respectively.\ In particular we enunciate a theorem which proves that the
field equations for a non-minimally coupled scalar field are the same at the
conformal level with the field equations of the minimally coupled scalar
field. The necessity to preserve Einstein's equations in the context of
Friedmann--Robertson--Walker (FRW) space-time leads us to apply, in section %
\ref{Conformal Lagrangians in scalar field FRW cosmology}, the current
general analysis to the scalar field (quintessence or phantom) spatially
flat FRWcosmologies.

In section \ref{Noether symmetries and exact} we apply the Noether symmetry
approach in non minimally coupled scalar field in a spatially flat FRW
spacetime and by using the Noether invariants we determine analytical
solutions for the field equations. Furthermore in sections \ref{MCsfcosm}
and \ref{MCsfcosmBianchi} we apply the same procedure for a minimally
coupled scalar field in a spatially flat FRW spacetime and in Bianchi Class
A homogeneous spacetimes.

\section{Conformally equivalent Lagrangians and scalar field Cosmology}

\label{Conformally equivalent Lagrangians and scalar field Cosmology}

In this section we discuss the conformal equivalence of Lagrangians for
scalar fields in a general $V^{4}$ Riemannian space. The field equations in
the scalar field cosmology are derived from two different variational
principles. In the first case the scalar field $\phi $ and the gravitational
field are minimally coupled and the equations of motion follow form the
action

\begin{equation}
S_{M}=\int d\tau dx^{3}\sqrt{-g}\left[ R+\frac{1}{2}g_{ij}\phi ^{;i}\phi
^{;j}-V\left( \phi \right) \right] .  \label{CLN.11}
\end{equation}%
In the second case the scalar field $\psi $ (which is different from the
minimally coupled scalar field $\phi )$ interacts with the gravitational
field (non minimal coupling) and the field equations follow from the action%
\begin{equation}
S_{NM}=\int d\tau dx^{3}\sqrt{-g}\left[ F\left( \psi \right) R+\frac{1}{2}%
\bar{g}_{ij}\psi ^{;i}\psi ^{;j}-\bar{V}\left( \psi \right) \right]
\label{CLN.12}
\end{equation}%
where $F(\psi )$ is the coupling function between the gravitational and the
scalar field $\psi $. Below we state the following proposition.

\begin{proposition}
The field equations for a non minimally coupled scalar field $\psi $ with
Lagrangian $\bar{L}\left( \tau ,x^{k},\dot{x}^{k}\right) $ and coupling
function $F(\psi )$ in the gravitational field $\bar{g}_{ij}$ are the same
with the field equations of the minimally coupled scalar field $\Psi $ for a
conformal Lagrangian $L\left( \tau ,x^{k},\dot{x}^{k}\right) $ in the
conformal metric $g_{ij}$ $=N^{-2}\bar{g}_{ij},$ where the conformal
function $N=\frac{1}{\sqrt{-2F\left( \psi \right) }}$ with $F\left( \psi
\right) <0.$ The inverse is also true, that is, to a minimally coupled
scalar field it can be associated a unique non minimally coupled scalar
field in a conformal metric and with a different potential function.
\end{proposition}

\begin{proof}
The action for the non minimally coupled Lagrangian $\bar{L}\left( \tau
,x^{k},x^{\prime k}\right) $ is:%
\begin{equation}
S_{NM}=\int d\tau dx^{3}\sqrt{-\bar{g}}\left[ F\left( \psi \right) \bar{R}+%
\frac{\varepsilon }{2}\bar{g}^{ij}\psi _{;i}\psi _{;j}-\bar{V}\left( \psi
\right) \right]  \label{CL.12.0}
\end{equation}%
where $\varepsilon =1$ for a real field and $\varepsilon =-1$ for phantom
field. Let $~g_{ij}$ be the conformally related metric (this is not a
coordinate transformation!):
\begin{equation*}
g_{ij}=N^{-2}\bar{g}_{ij}.
\end{equation*}%
Then the action (\ref{CL.12.0}) becomes\footnote{%
We use the result tha if $A=(a_{ij})$ is a $4x4$ matrix the
\begin{equation*}
\det A=\varepsilon ^{ijkl}a_{ij}a_{kl}
\end{equation*}%
hence
\begin{equation}
\bar{g}=\varepsilon ^{ijkl}\bar{g}_{ij}\bar{g}_{kl}=N^{4}g.  \label{CLN.12.1}
\end{equation}%
}:%
\begin{equation*}
S_{NM}=\int d\tau dx^{3}N^{4}\sqrt{-g}\left[ F\left( \psi \right) \bar{R}+%
\frac{\varepsilon }{2}N^{-2}g^{ij}\psi _{;i}\psi _{;j}-\bar{V}\left( \psi
\right) \right] .
\end{equation*}%
Replacing \cite{HawkingB}%
\begin{equation*}
\bar{R}=N^{-2}R-2(n-1)N^{-3}\Delta _{2}N
\end{equation*}%
where $\Delta _{2}N=g_{ij}N^{;ij}$ we find:%
\begin{eqnarray*}
S_{NM} &=&\int d\tau dx^{3}N^{4}\sqrt{-g}\left[ F\left( \psi \right) \left(
N^{-2}R-6N^{-3}\Delta _{2}N\right) +\frac{\varepsilon }{2}N^{-2}\Delta
_{1}\psi -\bar{V}\left( \psi \right) \right] \\
&=&\int d\tau dx^{3}N^{4}\sqrt{-g}\left[ F\left( \psi \right)
N^{-2}R-6F\left( \psi \right) N^{-3}\Delta _{2}N+\frac{\varepsilon }{2}%
N^{-2}\Delta _{1}\psi -\bar{V}\left( \psi \right) \right] \\
&=&\int d\tau dx^{3}\sqrt{-g}\left[ F\left( \psi \right) N^{2}R-6F\left(
\psi \right) N\Delta _{2}N+\frac{\varepsilon }{2}N^{2}\Delta _{1}\psi -N^{4}%
\bar{V}\left( \psi \right) \right] .
\end{eqnarray*}%
Define the conformal function in terms of the coupling function $F(\psi )$
by the requirement ($F\left( \psi \right) <0)$:%
\begin{equation}
N=\frac{1}{\sqrt{-2F\left( \psi \right) }}~.  \label{CLN.12.2}
\end{equation}%
We compute
\begin{equation}
N_{;i}=\frac{F_{\psi }\psi _{;i}}{\left( -2F\right) ^{\frac{3}{2}}}.
\label{CLN.12.3}
\end{equation}%
Then the first term in the integral becomes:%
\begin{equation*}
\int d\tau dx^{3}\sqrt{-g}F\left( \psi \right) N^{2}R=\int d\tau dx^{3}\sqrt{%
-g}\left( -\frac{R}{2}\right) .
\end{equation*}%
The second term gives after integration by parts:%
\begin{eqnarray*}
\int d\tau dx^{3}\sqrt{-g}\left( -6F\left( \psi \right) N\Delta _{2}N\right)
&=&\int d\tau dx^{3}\sqrt{-g}\left( -6\frac{F}{\sqrt{-2F}}%
N_{;ij}g^{ij}\right) \\
&=&\int d\tau dx^{3}\sqrt{-g}\left( -6\frac{F}{\sqrt{-2F}}\frac{1}{\sqrt{-g}}%
\left( \sqrt{-g}g^{ij}N_{,k}\right) _{,j}\right) \\
&=&\int d\tau dx^{3}\left( -6\frac{F}{\sqrt{-2F}}\left( \sqrt{-g}%
g^{ij}N_{,k}\right) _{,j}\right) \\
&=&\int d\tau dx^{3}\sqrt{-g}\left( 3\frac{F_{\psi }}{\sqrt{-2F}}\psi
_{;j}N_{;i}g^{ij}\right) .
\end{eqnarray*}%
Replacing $N_{;i}$ from (\ref{CLN.12.3}) we find:%
\begin{equation*}
\int d\tau dx^{3}\sqrt{-g}\left( -6F\left( \psi \right) N\Delta _{2}N\right)
=\int d\tau dx^{3}\sqrt{-g}\left( 3\frac{F_{\psi }^{2}}{\left( -2F\right)
^{2}}\psi _{;i}\psi _{;j}g^{ij}\right) .
\end{equation*}%
The third term gives:%
\begin{equation*}
\frac{\varepsilon }{2}N^{2}\Delta _{1}\psi =\frac{\varepsilon }{4F}\psi
_{;i}\psi _{;j}g^{ij}
\end{equation*}%
Collecting all the results we find:
\begin{eqnarray*}
S_{MN} &=&\int d\tau dx^{3}\sqrt{-g}\left[ -\frac{R}{2}+3\frac{F_{\psi }^{2}%
}{\left( -2F\right) ^{2}}\psi _{;i}\psi _{;j}g^{ij}-\frac{\varepsilon }{4F}%
\psi _{;i}\psi _{;j}g^{ij}-\frac{\bar{V}\left( \psi \right) }{4F^{2}}\right]
\\
&=&\int d\tau dx^{3}\sqrt{-g}\left[ -\frac{R}{2}+3\frac{F_{\psi }^{2}}{4F^{2}%
}\psi _{;i}\psi _{;j}g^{ij}-\frac{\varepsilon }{4F}\psi _{;i}\psi
_{;j}g^{ij}-\frac{\bar{V}\left( \psi \right) }{4F^{2}}\right]
\end{eqnarray*}%
or%
\begin{equation}
S_{MN}=\int d\tau dx^{3}\sqrt{-g}\left[ -\frac{R}{2}+\frac{\varepsilon }{2}%
\left( \frac{3\varepsilon F_{\psi }^{2}-F}{2F^{2}}\right) \psi _{;i}\psi
_{;j}g^{ij}-\frac{\bar{V}\left( \psi \right) }{4F^{2}}\right] .
\label{CLN.12.5}
\end{equation}%
We introduce the scalar field $\Psi $ with the requirement:%
\begin{equation}
d\Psi =\sqrt{\left( \frac{3\varepsilon F_{\psi }^{2}-F}{2F^{2}}\right) }d\psi
\label{CLN.12.6}
\end{equation}%
and the action becomes
\begin{equation}
S_{MN}=\int d\tau dx^{3}\sqrt{-g}\left[ -\frac{R}{2}+\frac{\varepsilon }{2}%
\Psi _{;i}\Psi _{;j}g^{ij}-\frac{\bar{V}\left( \Psi \right) }{4F\left( \Psi
\right) ^{2}}\right] .  \label{CLN.12.7}
\end{equation}
\end{proof}

We conclude that the scalar field $\Psi $ is minimally coupled to the
gravitational field. Therefore we have proved that to every non minimally
coupled scalar field we may associate a unique minimally coupled scalar
field in a conformally related space and an appropriate potential. Since all
considerations are reversible, the result is reversible.

In the following section we extend the above proposition to a general $V^{n}$
Riemannian space.

\section{Generalization to dimension $n$}

\label{Generalization to dimension n}

Consider the non minimally coupled scalar field $\psi $ whose field
equations are obtained from the action:

\begin{eqnarray*}
S_{NM} &=&\int dx^{n}N^{n}\sqrt{-g}\left[ F\left( \psi \right) \bar{R}+\frac{%
\varepsilon }{2}N^{-2}g^{ij}\psi _{,i}\psi _{,j}-\bar{V}\left( \psi \right) %
\right] \\
&=&\int dx^{n}\sqrt{-g}\left[
\begin{array}{c}
F\left( \psi \right) N^{n-2}R-2(n-1)F\left( \psi \right) N^{n-3}\Delta _{2}N+
\\
-F\left( \psi \right) N^{n}(n-1)(n-4)\Delta _{1}N+\frac{\varepsilon }{2}%
N^{n-2}g^{ij}\psi _{,i}\psi _{,j}-N^{n}\bar{V}\left( \psi \right)%
\end{array}%
\right]
\end{eqnarray*}%
where we have substituted \cite{HawkingB}
\begin{equation*}
\bar{R}=N^{-2}R-2(n-1)N^{-3}\Delta _{2}N-(n-1)(n-4)\Delta _{1}N
\end{equation*}%
and%
\begin{eqnarray*}
\Delta _{1}N &=&g_{ij}N^{,i}N^{,j} \\
\Delta _{2}N &=&g_{ij}N^{;ij}.
\end{eqnarray*}%
\qquad

Define the function $N(x^{i})$ in terms of the coupling function $F(\psi )$
by the requirement:%
\begin{equation*}
N^{n-2}=\frac{1}{-2F}.~,~F=-\frac{N^{2-n}}{2}.
\end{equation*}

For each term of the action $S_{NM}$ we have the following:\newline
The first term gives:%
\begin{equation*}
\int dx^{n}\sqrt{-g}\left( F\left( \psi \right) N^{n-2}R\right) =\int dx^{n}%
\sqrt{-g}\left( -\frac{R}{2}\right) .
\end{equation*}%
The second term gives:
\begin{eqnarray*}
\int dx^{n}\sqrt{-g}\left( -2(n-1)F\left( \psi \right)
N^{n-3}N_{;ij}g^{ij}\right) &=&\int dx^{n}\sqrt{-g}\left(
(n-1)N^{2-n}N^{n-3}N_{;ij}g^{ij}\right) \\
&=&\int dx^{n}\sqrt{-g}\left( (n-1)N^{-1}\frac{1}{\sqrt{-g}}\left( \sqrt{-g}%
g^{ij}N_{,k}\right) _{,j}\right) \\
&=&\int dx^{n}\sqrt{-g}\left( (n-1)N^{-1}N_{;ij}g^{ij}\right) \\
&=&\int dx^{n}\sqrt{-g}\left( -(n-1)\left( N^{-1}\right)
_{;j}N_{;i}g^{ij}\right) .
\end{eqnarray*}%
Furthermore we compute%
\begin{equation*}
N=\frac{1}{\left( -2F\right) ^{\frac{1}{n-2}}}\rightarrow N^{-1}=\left(
-2F\right) ^{\frac{1}{n-2}}
\end{equation*}%
\begin{equation*}
N_{;i}N_{;j}^{-1}=-\frac{F_{\psi }^{2}}{\left( n-2\right) ^{2}F^{2}}\psi
_{;i}\psi _{;j}.
\end{equation*}%
Replacing we find for the second term:%
\begin{equation*}
\int dx^{n}\sqrt{-g}\left( -2(n-1)F\left( \psi \right)
N^{n-3}N_{;ij}g^{ij}\right) =\int dx^{n}\sqrt{-g}\left( \frac{(n-1)}{\left(
n-2\right) ^{2}}\frac{F_{\psi }^{2}}{F^{2}}\psi _{;i}\psi _{;j}g^{ij}\right)
.
\end{equation*}%
(note that this is the same with the previous expression for $n=4$).\newline
The third term gives:%
\begin{eqnarray*}
\int dx^{n}\sqrt{-g}\left( F\left( \psi \right) N^{n}(n-1)(n-4)\Delta
_{1}N\right) &=&\int dx^{n}\sqrt{-g}\left( -\frac{N^{2-n}}{2}%
N^{n}(n-1)(n-4)\Delta _{1}N\right) \\
&=&\int dx^{n}\sqrt{-g}\left( -\frac{1}{2}N^{2}(n-1)(n-4)\Delta _{1}N\right)
\\
&=&\int dx^{n}\sqrt{-g}\left( -\frac{1}{2}\frac{(n-1)(n-4)}{\left(
n-2\right) ^{2}}\frac{F_{\psi }^{2}}{\left( -2F\right) ^{\frac{4}{2-n}}F^{2}}%
\psi _{;i}\psi _{;j}g^{ij}\right) .
\end{eqnarray*}%
Finally the the fourth term gives:%
\begin{equation*}
\int dx^{n}\sqrt{-g}\left( \frac{\varepsilon }{2}N^{n-2}g^{ij}\psi _{,i}\psi
_{,j}\right) =\int dx^{n}\sqrt{-g}\left( -\frac{\varepsilon }{4}\frac{1}{F}%
g^{ij}\psi _{;i}\psi _{;j}\right) .
\end{equation*}%
Collecting the results for the last three terms we find%
\begin{eqnarray*}
&&\int dx^{n}\sqrt{-g}\left( \frac{(n-1)}{\left( n-2\right) ^{2}}\frac{%
F_{\psi }^{2}}{F^{2}}\psi _{;i}\psi _{;j}g^{ij}-\frac{1}{2}\frac{(n-1)(n-4)}{%
\left( n-2\right) ^{2}}\frac{F_{\psi }^{2}}{\left( -2F\right) ^{\frac{4}{2-n}%
}F^{2}}\psi _{;i}\psi _{;j}g^{ij}-\varepsilon \frac{1}{4F}g^{ij}\psi
_{;i}\psi _{;j}\right) \\
&=&\int dx^{n}\sqrt{-g}\left( \frac{(n-1)}{\left( n-2\right) ^{2}}\frac{%
F_{\psi }^{2}}{F^{2}}-\frac{1}{2}\frac{(n-1)(n-4)}{\left( n-2\right) ^{2}}%
\frac{F_{\psi }^{2}}{\left( -2F\right) ^{\frac{4}{2-n}}F^{2}}-\varepsilon
\frac{1}{4F}\right) \psi _{;i}\psi _{;j}g^{ij} \\
&=&\int dx^{n}\sqrt{-g}\frac{\varepsilon }{2}\left( \frac{2\varepsilon (n-1)%
}{\left( n-2\right) ^{2}}\frac{F_{\psi }^{2}}{F^{2}}-\varepsilon \frac{%
(n-1)(n-4)}{\left( n-2\right) ^{2}}\frac{F_{\psi }^{2}}{\left( -2F\right) ^{%
\frac{4}{2-n}}F^{2}}-\frac{1}{2F}\right) \psi _{,i}\psi _{,j}g^{ij}.
\end{eqnarray*}%
We define the new scalar field $\Psi $ with the requirement%
\begin{equation*}
d\Psi =\left( \frac{2\varepsilon (n-1)}{\left( n-2\right) ^{2}}\frac{F_{\psi
}^{2}}{F^{2}}-\varepsilon \frac{(n-1)(n-4)}{\left( n-2\right) ^{2}}\frac{%
F_{\psi }^{2}}{\left( -2F\right) ^{\frac{4}{2-n}}F^{2}}-\frac{1}{2F}\right)
^{\frac{1}{2}}d\psi .
\end{equation*}%
In terms of $\Psi \ $the action becomes%
\begin{eqnarray*}
&&\int dx^{n}\sqrt{-g}\frac{\varepsilon }{2}\left( \frac{2\varepsilon (n-1)}{%
\left( n-2\right) ^{2}}\frac{F_{\psi }^{2}}{F^{2}}-\varepsilon \frac{%
(n-1)(n-4)}{\left( n-2\right) ^{2}}\frac{F_{\psi }^{2}}{\left( -2F\right) ^{%
\frac{4}{2-n}}F^{2}}-\frac{1}{2F}\right) \psi _{,i}\psi _{,j}g^{ij}. \\
&=&\int dx^{n}\sqrt{-g}\left( \frac{\varepsilon }{2}\Psi _{;i}\Psi
_{;j}g^{ij}\right) .
\end{eqnarray*}

Collecting the above we find the action $S_{M}$ of a minimally coupled
scalar field%
\begin{equation*}
S_{M}=\int dx^{n}\sqrt{-g}\left( -\frac{R}{2}+\frac{\varepsilon }{2}\Psi
_{;i}\Psi _{;j}g^{ij}-\frac{\bar{V}\left( \Psi \right) }{\left( -2F\right) ^{%
\frac{n}{n-2}}}\right) .
\end{equation*}%
We note that the new scalar field $\Psi $ is minimally coupled to the
gravitational field $g_{ij}$ and that the potential of $\Psi $ is $\frac{%
\bar{V}\left( \Psi \right) }{\left( -2F\right) ^{\frac{n}{n-2}}}$.

The above proof agrees with the one given in the paper of \cite{Kaiser}.
However it is obviously simpler, more direct and clear.

\section{Conformal Lagrangians in scalar field cosmology}

\label{Conformal Lagrangians in scalar field FRW cosmology}

In this section we apply the conformal transformation in the Lagrangian of
the field equations in a FRW spatially flat spacetime.

We consider the flat FRW~$\left( K=0\right) $ spacetime whose metric is%
\begin{equation}
ds^{2}=-dt^{2}+a^{2}\left( t\right) \delta _{ij}dx^{i}dx^{j}  \label{CLN.13}
\end{equation}%
where $\delta _{ij}$ is the 3-space metric in Cartesian coordinates. The
Lagrangian of a scalar field $\phi $ minimally coupled to gravity in these
coordinates is

\begin{equation}
L_{M}\left( a,\dot{a},\phi ,\dot{\phi}\right) =-3a\dot{a}^{2}+\frac{%
\varepsilon }{2}a^{2}\dot{\phi}^{2}-a^{3}V\left( \phi \right) .
\label{CLN.14}
\end{equation}%
The Lagrangian for a non minimally coupled scalar field $\psi $ is
\begin{equation}
L_{NM}\left( a,\dot{a},\psi ,\dot{\psi}\right) =6F\left( \psi \right) a\dot{a%
}^{2}+6F_{\psi }\left( \psi \right) a^{2}\dot{a}\dot{\psi}+\frac{\varepsilon
}{2}a^{3}\dot{\psi}^{2}-a^{3}V\left( \psi \right)  \label{CLN.15}
\end{equation}%
where $\ F(\psi )<0$ is the coupling function and $\varepsilon =\pm 1$ where
$\varepsilon =1$ for real scalar field and $\varepsilon =-1$ for phantom
field. \ The Hamiltonian for the Lagrangian (\ref{CLN.15}) is
\begin{equation}
E=6F\left( \psi \right) a\dot{a}^{2}+6F_{\psi }\left( \psi \right) a^{2}\dot{%
a}\dot{\psi}+\frac{\varepsilon }{2}a^{3}\dot{\psi}^{2}+a^{3}V\left( \psi
\right) .  \label{CLN.15e}
\end{equation}

We construct a conformal Lagrangian which corresponds to a minimally coupled
scalar field.

To do that we consider first a change in the scale factor from $a(t)$ $%
\rightarrow A(t)$ defined by the formula (see \cite{Cap01,Cap02})%
\begin{equation}
A\left( t\right) =\sqrt{-2F}a\left( t\right) \text{ }  \label{CLN.15a1}
\end{equation}%
Then the Lagrangian (\ref{CLN.15}) takes the form:%
\begin{equation}
L_{NM}\left( A,\dot{A},\psi ,\dot{\psi}\right) =\frac{1}{\sqrt{-2F}}\left[
-3A\dot{A}^{2}+\frac{\varepsilon }{2}\left( \frac{3\varepsilon F_{\psi
}^{2}-F}{2F^{2}}\right) A^{3}~\dot{\psi}^{2}\right] -\frac{A^{3}}{\left(
-2F\right) ^{\frac{3}{2}}}V\left( \psi \right)  \label{CLN.16}
\end{equation}%
that is, the cross term $\dot{a}\dot{\psi}$ disappears.

Next we consider the coordinate transformation:%
\begin{equation}
d\Psi =\sqrt{\left( \frac{3\varepsilon F_{\psi }^{2}-F}{2F^{2}}\right) }%
~d\psi ~  \label{CLN.16aa}
\end{equation}%
and Lagrangian (\ref{CLN.16}) becomes:%
\begin{equation}
L_{NM}\left( A,A^{\prime },\Psi ,\Psi ^{\prime }\right) =\frac{1}{\sqrt{-2F}}%
\left[ -3A\dot{A}^{2}+\frac{\varepsilon }{2}A^{3}~\dot{\Psi}^{2}\right] ~,~%
\bar{V}\left( \psi \right) =\frac{V\left( \Psi \right) }{\left( -2F\right) ^{%
\frac{3}{2}}}  \label{CLN.18}
\end{equation}

The form of the Lagrangian (\ref{CLN.18}) is (\ref{CLN.14}) hence the
previous result applies and under the conformal transformation
\begin{equation}
d\tau =\sqrt{-2F\left( \psi \right) }dt  \label{CLN.17a}
\end{equation}
the Lagrangian (\ref{CLN.18}) becomes:
\begin{equation}
L_{M}\left( A,\dot{A},\Psi ,\Psi ^{\prime }\right) =-3AA^{\prime 2}+\frac{%
\varepsilon }{2}A^{3}\Psi ^{\prime 2}-\frac{A^{3}}{\left( -2F\right) ^{\frac{%
3}{2}}}V\left( \Psi \right)  \label{CLN.17}
\end{equation}%
where a prime $^{\prime }$ indicates derivative wrt the new coordinate $\tau
.$

We note that if in the new coordinates $\tau ,x^{i}$ we consider the metric%
\begin{equation}
d\bar{s}^{2}=-d\tau ^{2}+A^{2}\left( \tau \right) \delta _{ij}dx^{i}dx^{j}
\label{CLN.19}
\end{equation}%
then the term $3AA^{\prime 2}$ equals the Ricci scalar $\bar{R}$ of the
conformally flat metric $d\bar{s}^{2}.$ Therefore the Lagrangian (\ref%
{CLN.17}) can be seen as the Lagrangian of a scalar field $\Psi $ of
potential $\bar{V}\left( \Psi \right) $ minimally coupled to the
gravitational filed $\bar{g}_{ij}$ in the space with metric $d\bar{s}^{2}.$
Replacing the coordinate $\tau $ and the quantity $A\left( \tau \right) $
from (\ref{CLN.15a1}), (\ref{CLN.17a}), we find:%
\begin{equation}
d\bar{s}^{2}=\sqrt{-2F}(-dt^{2}+a^{2}\left( t\right) \delta
_{ij}dx^{i}dx^{j})=\sqrt{-2F}ds^{2}  \label{CLN.20}
\end{equation}%
that is, the metric $d\bar{s}^{2}$ is conformally related to the metric $%
ds^{2}$ with conformal function $\sqrt{-2F}.$ This means that the
non-minimally coupled scalar field in the gravitational field $ds^{2}$ is
equivalent to a minimally coupled scalar field - with appropriate potential
defined in terms of the conformal function - in the gravitational field $d%
\bar{s},$ the resutl is reversible. Equivalently the Lagrangians $L_{M},$ $%
L_{NM}$ are conformally related and the field equations (the Euler-Lagrange
equations) are invariant under the conformal transformation if the
Hamiltonian constrains $H_{M},~H_{NM}$ vanish (see Lemma \ref{LemaHam}).

In the following sections we apply the Noether symmetry approach as a
geometric selection rule in order to determine the dark energy models; that
is, we search for dark energy models by requiring the field equations to
admit Noether point symmetries.

\section{Noether point symmetries of a non minimally coupled Scalar field.}

\label{Noether symmetries and exact}

Consider a non-minimally coupled scalar field with action%
\begin{equation*}
S_{NM}=\int d\tau dx^{3}\sqrt{-g}\left[ F\left( \psi \right) R+\frac{%
\varepsilon }{2}\bar{g}_{ij}\psi ^{;i}\psi ^{;j}-\bar{V}\left( \psi \right) %
\right] +\int L_{m}d\tau dx^{3}
\end{equation*}%
in a flat FRW space-time, whose metric in Cartesian coordinates is%
\begin{equation*}
ds^{2}=-dt^{2}+a^{2}\left( t\right) \delta _{ij}dx^{i}dx^{j}
\end{equation*}%
and $L_{m}$ is the Lagrangian of dust matter of density $\rho _{D}$ (for
commoving observers).

The Lagrangian of the field equations is (\ref{CLN.15}) are the Hamiltonian
(total energy density) (\ref{CLN.15e}) and the Euler-Lagrange equations%
\begin{equation*}
\frac{d}{dt}\frac{\partial L_{NM}}{\partial \left( \dot{a},\dot{\psi}\right)
}-\frac{\partial L}{\partial \left( a,\psi \right) }=0.
\end{equation*}

If the Hamiltonian (\ref{CLN.15e})~is $E\neq 0$ then the space admits dust
which, however does not interact with the scalar field and has energy
density $\rho _{D}=\frac{\left\vert E\right\vert }{a^{3}}.$If $E\,$%
\thinspace $=0$ the space does not admit dust. In the following we determine
all potentials $V\left( \psi \right) $ for which this dynamical system
admits Noether point symmetries beyond the trivial one $\partial _{t}$.
Subsequently we use the resulting Noether integrals to find analytical
solutions for the field equations for each of these potentials.

In order to determine the Noether point symmetries of the Lagrangian (\ref%
{CLN.15}) we shall follow the results of chapter \ref{chapter3}. That is we
brake the Lagrangian in the kinematic part which defines the kinematic
metric and the remaining part which we consider to be the potential. Then we
apply theorem \ref{The Noether symmetries of a conservative system} which
states that the Noether point symmetries of the Lagrangian follow form the
homothetic algebra of the kinematic metric. The kinematic metric admits a
non-trivial homothetic (not necessarily proper homothetic) algebra if a
given condition is satisfied which involves the symmetry vector and the
potential. The solution of this relation provides all the potentials for
which extra Noether symmetries are admitted. We use the Noether integrals of
these extra Noether symmetries to find an analytic solution for each of the
corresponding potentials..

From the Lagrangian (\ref{CLN.15}) we define the kinematic metric:%
\begin{equation}
ds_{KM}^{2}=12F\left( \psi \right) a\dot{a}^{2}+12F_{\psi }\left( \psi
\right) a^{2}\dot{a}\dot{\psi}+\varepsilon a^{3}\dot{\psi}^{2}.
\label{CLN.21}
\end{equation}

This is a two dimensional metric in the space $(a,\psi ).$ Because the
homothetic algebra of a 2 dimensional metric is different for a flat and a
non-flat (but conformally flat because all two dimensional spaces are
conformally flat) we consider the case the metric (\ref{CLN.21}) is
maximally symmetric\footnote{%
The Ricci scalar is $R=K~$where $K$ is a constant}.

The Ricci scalar of the metric (\ref{CLN.21}) is computed to be:%
\begin{equation}
R_{\left( KM\right) }=\frac{\varepsilon }{4a^{3}}\frac{\left( 2F_{\psi \psi
}F-F_{\psi }^{2}\right) }{\left( \varepsilon F-3F_{\psi }^{2}\right) ^{2}}
\label{CLN.22}
\end{equation}%
Hence if the metric (\ref{CLN.21}) is maximally symmetric then it follows
that $R_{\left( KM\right) }=0$, that is, the kinetic metric (\ref{CLN.21})
must be flat.

\subsection{Case A. $R_{KM}=0$}

Condition $R_{KM}=0$ and (\ref{CLN.22}) give:%
\begin{equation}
2F_{\psi \psi }F-F_{\psi }^{2}=0  \label{CLN.23}
\end{equation}%
provided
\begin{equation}
\varepsilon F-3F_{\psi }^{2}\neq 0.  \label{CLN.24}
\end{equation}%
The solution of (\ref{CLN.23}) is%
\begin{equation}
F\left( \psi \right) =-\frac{F_{0}\varepsilon }{12}\left( \psi +\psi
_{0}\right) ^{2}  \label{CLN.25}
\end{equation}%
where $\varepsilon =\pm 1$ and $F_{0}\varepsilon >0$. We note that this $%
F\left( \psi \right) $ satisfies condition (\ref{CLN.24}) therefore it is
acceptable.

In order to determine the homothetic algebra of the kinematic metric (\ref%
{CLN.21}) we write it in a more familiar form. We introduce the coordinates $%
A,\Psi $ by the relations%
\begin{eqnarray}
A &=&\sqrt{-2F}a  \label{CLN.15a} \\
d\Psi &=&\sqrt{\left( \frac{3\varepsilon F_{\psi }^{2}-F}{2F^{2}}\right) }%
~d\psi .  \label{CLN.16a}
\end{eqnarray}%
In the new coordinates $A,\Psi $ the metric (\ref{CLN.21}), the Lagrangian (%
\ref{CLN.15}) and the non minimal coupling function $F\left( \Psi \right) $
take the following form

For \thinspace $\varepsilon =1$ and $F_{0}>0~$
\begin{eqnarray}
ds_{KM}^{2} &=&\frac{1}{\sqrt{-2F\left( \Psi \right) }}\left[ -3A\dot{A}^{2}+%
\frac{1}{2}A^{3}\dot{\Psi}^{2}\right]  \label{CLN.26.1} \\
L_{M} &=&\frac{1}{\sqrt{-2F\left( \Psi \right) }}\left[ -3A\dot{A}^{2}+\frac{%
1}{2}A^{3}\dot{\Psi}^{2}\right] -\frac{A^{3}}{\left( -2F\left( \Psi \right)
\right) ^{\frac{3}{2}}}V\left( \Psi \right) ~.  \label{CLN.26.2} \\
F\left( \Psi \right) &=&-\frac{F_{0}}{12}\exp \left( \pm \frac{\sqrt{6}}{3}%
\sqrt{\frac{F_{0}}{F_{0}+1}}\Psi \right) ~,~F_{0}>0  \label{CLN.26.3}
\end{eqnarray}

For $\varepsilon =-1~\,$\ and $-1<F_{0}<0~$%
\begin{eqnarray}
ds_{KM}^{2} &=&\frac{1}{\sqrt{-2F\left( \Psi ^{\ast }\right) }}\left[ -3A%
\dot{A}^{2}-\frac{1}{2}A^{3}\dot{\Psi}^{2}\right]  \label{CLN.26.4} \\
L_{NM} &=&\frac{1}{\sqrt{-2F\left( \Psi ^{\ast }\right) }}\left[ -3A\dot{A}%
^{2}-\frac{1}{2}A^{3}\dot{\Psi}^{2}\right] -\frac{A^{3}}{\left( -2F\left(
\Psi \right) \right) ^{\frac{3}{2}}}V\left( \Psi \right)  \label{CLN.26.5} \\
F\left( \Psi \right) &=&-\frac{\left\vert F_{0}\right\vert }{12}\exp \left(
\pm \frac{\sqrt{6}}{3}\sqrt{\frac{\left\vert F_{0}\right\vert }{1-\left\vert
F_{0}\right\vert }}\Psi \right) ~,~-1<F_{0}<0  \label{CLN.26.6}
\end{eqnarray}%
For $\varepsilon =-1~$and $F_{0}<-1$ we introduce a new real field $\Psi
^{\ast }=i\Psi $ and get the real field $F\left( \Psi ^{\ast }\right) $

\begin{eqnarray}
ds_{KM}^{2} &=&\frac{1}{\sqrt{-2F\left( \Psi ^{\ast }\right) }}\left[ -3A%
\dot{A}^{2}+\frac{1}{2}A^{3}\dot{\Psi}^{\ast 2}\right]  \label{CLN.26.7} \\
L_{NM} &=&\frac{1}{\sqrt{-2F\left( \Psi ^{\ast }\right) }}\left[ -3A\dot{A}%
^{2}+\frac{1}{2}A^{3}\dot{\Psi}^{\ast 2}\right] -\frac{A^{3}}{\left(
-2F\left( \Psi ^{\ast }\right) \right) ^{\frac{3}{2}}}V\left( \Psi ^{\ast
}\right)  \label{CLN.26.8} \\
F\left( \Psi ^{\ast }\right) &=&-\frac{\left\vert F_{0}\right\vert }{12}\exp
\left( \mp \frac{\sqrt{6}}{3}\sqrt{\frac{\left\vert F_{0}\right\vert }{%
\left\vert F_{0}\right\vert -1}}~\Psi ^{\ast }\right) ~,~F_{0}<-1.
\label{CLN.26.9}
\end{eqnarray}

We see that in the region $-\infty <F_{0}<-1$ starting from a phantom field $%
\Psi $ we end up with a real scalar field $\Psi ^{\ast }.$

In order to consider the three Lagrangians (\ref{CLN.26.2}),(\ref{CLN.26.5}%
),(\ref{CLN.26.8}) at the same time we consider the Lagrangian%
\begin{equation}
L=N^{2}\left( \Psi \right) \left[ -3A\dot{A}^{2}+\frac{\varepsilon _{k}}{2}%
A^{3}\dot{\Psi}^{2}\right] -A^{3}\bar{V}\left( \Psi \right)  \label{CLN.27}
\end{equation}%
where $N^{2}\left( \Psi \right) =\frac{1}{\sqrt{-2F\left( \Psi \right) }}~$%
and $\bar{V}\left( \Psi \right) =\frac{V\left( \Psi \right) }{\left(
-2F\left( \Psi \right) \right) ^{\frac{3}{2}}}$ where $F\left( \Psi \right)
<0$. The constant $\varepsilon _{k}=\pm 1$ where the value $+1$ is for $%
\varepsilon =1,~F_{0}>0~$\ and $\varepsilon =-1~,~F_{0}<-1$ and the value $%
-1 $ is for $\varepsilon =-1~,~-1<F_{0}<0.$ \ Then the new kinematic metric
is written as
\begin{equation}
ds_{KM}^{2}=N^{2}\left( \Psi \right) \left[ -3A\dot{A}^{2}+\frac{\varepsilon
_{k}}{2}A^{3}\dot{\Psi}^{2}\right] .  \label{CLN.28}
\end{equation}

We simplify this metric by introducing new coordinates $r,\theta $ defined
by the transformation:%
\begin{equation}
r=\sqrt{\frac{8}{3}}A^{\frac{3}{2}}~,~\theta =\sqrt{\frac{3\varepsilon _{k}}{%
8}}~\Psi  \label{CLN.29}
\end{equation}%
This step is necessary in order to deduce the homothetic algebra of the
metric from well known previous results. In the new coordinates the metric (%
\ref{CLN.28}) takes the simple form:%
\begin{equation}
ds_{KM}^{2}=N^{2}\left( \theta \right) \left( -dr^{2}+r^{2}d\theta
^{2}\right) .  \label{CLN.30}
\end{equation}%
that is, it is directly related to the flat 2d Lorentzian space with metric
\begin{equation*}
ds^{2}=-dr^{2}+r^{2}d\theta ^{2}
\end{equation*}%
with conformal factor $N^{2}\left( \theta \right) .$ In the new coordinates
the curvature scalar is%
\begin{equation}
R_{KM}=-\frac{2}{r^{4}N^{3}\left( \theta \right) }\left( N_{,\theta \theta }-%
\frac{1}{N\left( \theta \right) }\left( N_{,\theta }\right) ^{2}\right)
\label{CLN.31}
\end{equation}%
hence the condition $R_{KM}=0$ gives the function $N\left( \theta \right) $:%
\begin{equation}
N\left( \theta \right) =N_{0}e^{k\theta }~,~k\in
\mathbb{C}
~,N_{0}\in
\mathbb{R}
\label{CLN.32}
\end{equation}%
where $k$ is a new constant.

Taking this into account we have that in the $r,\theta $ coordinates the
Lagrangian (\ref{CLN.27})\ becomes:%
\begin{equation}
L=N_{0}^{2}e^{2k\theta }\left( -\frac{1}{2}\dot{r}^{2}+\frac{1}{2}r^{2}\dot{%
\theta}^{2}\right) -r^{2}\bar{V}\left( \theta \right) .  \label{CLN.33}
\end{equation}%
The constant $k$ is related to the previous constant $F_{0}$ via the
function $N(\theta )$. Using (\ref{CLN.29}) we express $N\left( \theta
\right) $ in terms of $\Psi $:%
\begin{equation}
N\left( \Psi \right) =N_{0}e^{k\sqrt{\frac{3\varepsilon _{k}}{8}}\Psi }.
\label{CLN.34}
\end{equation}%
Comparing with $F\left( \Psi \right) $ and eliminating $N\left( \Psi \right)
$ we find:
\begin{equation}
F\left( \Psi \right) =-\frac{1}{2N_{0}^{4}}e^{-4k~\sqrt{\frac{3\varepsilon }{%
8}}~\Psi }.  \label{CLN.36}
\end{equation}%
This is a second expression of $F\left( \Psi \right) $ in terms of the
constant $k.$ Comparing with the previous expression (\ref{CLN.26.3}),(\ref%
{CLN.26.6}) and (\ref{CLN.26.9}), which expresses $F\left( \Psi \right) $ in
terms of $F_{0}$ (and holds for all ranges of values of $F_{0}!),$ we find:%
\begin{equation}
-\frac{F_{0}\varepsilon }{12}\exp \left( \pm \frac{\sqrt{6}}{3}\sqrt{\frac{%
\varepsilon F_{0}}{F_{0}+1}}\Psi \right) =-\frac{1}{2N_{0}^{4}}e^{-4k~\sqrt{%
\frac{3\varepsilon _{k}}{8}}~\Psi }.  \label{CLN.37}
\end{equation}%
This relation must hold identically which leads to the conditions%
\begin{equation}
~N_{0}=\left( \frac{6}{\varepsilon F_{0}}\right) ^{1/4}  \label{CLN.38}
\end{equation}%
and%
\begin{equation}
~\left\vert k\right\vert =\frac{1}{3}\sqrt{\frac{\varepsilon F_{0}}{F_{0}+1}}%
.  \label{CLN.39}
\end{equation}%
In Appendix \ref{appendix1} \ we give the relation between the various
ranges of the constant $F_{0}$\ and the corresponding ranges of the constant
$\left\vert k\right\vert .$

\subsubsection{Case $\left\vert k\right\vert \neq 1$. \newline
}

For $\left\vert k\right\vert \neq 1$ the homothetic algebra consists of the
gradient KVs vectors

\begin{eqnarray}
K^{1} &=&\frac{e^{\left( 1-k\right) \theta }r^{k}}{N_{0}^{2}}\left(
-\partial _{r}+\frac{1}{r}\partial _{\theta }\right)  \label{KV.1} \\
K^{2} &=&\frac{e^{-\left( 1+k\right) \theta }r^{-k}}{N_{0}^{2}}\left(
\partial _{r}+\frac{1}{r}\partial _{\theta }\right)  \label{KV.2}
\end{eqnarray}%
\qquad the non gradient KV
\begin{equation}
K^{3}=r\partial _{r}-\frac{1}{k}\partial _{\theta }  \label{KV.3}
\end{equation}%
and the gradient HV%
\begin{equation}
H^{i}=\frac{1}{N_{0}^{2}\left( k^{2}-1\right) }\left( -r\partial
_{r}+k\partial _{\theta }\right) ~,~H\left( r,\theta \right) =\frac{1}{2}%
\frac{r^{2}e^{2k\theta }}{k^{2}-1}.  \label{KV.4}
\end{equation}%
Using the results of chapter \ref{chapter3} we find the following results

\begin{enumerate}
\item The gradient KV\ $K^{1}$ produces Noether symmetries for the following
potentials

a) For $V\left( \theta \right) =V_{0}e^{2\theta }$ we have the Noether
symmetries $K^{1},~tK^{1}$ with Noether integrals%
\begin{equation}
I_{1}=\frac{d}{dt}\left( \frac{r^{1+k}e^{\left( 1+k\right) \theta }}{\left(
k+1\right) }\right) ~,~I_{2}=t\frac{d}{dt}\left( \frac{r^{1+k}e^{\left(
1+k\right) \theta }}{\left( k+1\right) }\right) -\left( \frac{%
r^{1+k}e^{\left( 1+k\right) \theta }}{\left( k+1\right) }\right)
\label{CLN.75a}
\end{equation}

b) For $V\left( \theta \right) =V_{0}e^{2\theta }-\frac{mN_{0}^{2}}{2\left(
k^{2}-1\right) }e^{2k\theta }$ we have the Noether symmetries $e^{\pm \sqrt{m%
}t}K^{1}$ $m=$constant, with Noether integrals%
\begin{equation}
I_{\pm }^{\prime }=e^{\pm \sqrt{m}t}\left[ \frac{d}{dt}\left( \frac{%
r^{1+k}e^{\left( 1+k\right) \theta }}{\left( k+1\right) }\right) \mp \sqrt{m}%
\left( \frac{r^{1+k}e^{\left( 1+k\right) \theta }}{\left( k+1\right) }%
\right) \right]  \label{CLN.99}
\end{equation}%
From the Noether integrals we construct the time independent first integral $%
I_{K^{1}}=I_{+}I_{-}.$

\item The gradient KV $K^{2}$ produces the following Noether symmetries for
the following potentials

a) For $V\left( \theta \right) =V_{0}e^{-2\theta }$,we have the Noether
symmetries $K^{1},~tK^{1}$ with Noether integrals%
\begin{equation}
J_{1}=\frac{d}{dt}\left( \frac{r^{1-k}e^{-\left( 1-k\right) \theta }}{k-1}%
\right) ~,~J_{2}=t\frac{d}{dt}\left( \frac{r^{1-k}e^{-\left( 1-k\right)
\theta }}{k-1}\right) -\frac{r^{1-k}e^{-\left( 1-k\right) \theta }}{k-1}
\label{CLN.100}
\end{equation}

b) For $V\left( \theta \right) =V_{0}e^{-2\theta }-\frac{mN_{0}^{2}}{2\left(
k^{2}-1\right) }e^{2k\theta }$, we have the Noether symmetries $e^{\pm \sqrt{%
m}t}K^{2}$ $m=$constant, with Noether integrals%
\begin{equation}
J_{\pm }^{^{\prime }}=e^{\pm \sqrt{m}t}\left[ \frac{d}{dt}\left( \frac{%
r^{1-k}e^{-\left( 1-k\right) \theta }}{k-1}\right) \mp \sqrt{m}\frac{%
r^{1-k}e^{-\left( 1-k\right) \theta }}{k-1}\right]  \label{CLN.100.a}
\end{equation}%
From the Noether integrals we construct the time independent first integral $%
J_{K^{2}}=J_{+}^{\prime }J_{-}^{\prime }.$

\item The non gradient KV $K^{3}$ produces a Noether symmetry for the
potential $V\left( \theta \right) =V_{0}e^{2k\theta }$ with Noether integral
\begin{equation}
I_{3}=\frac{re^{2k\theta }}{k}\left( k\dot{r}+r\dot{\theta}\right) .
\label{CLN.100b}
\end{equation}

\item The gradient HV produces the following Noether symmetries for the
following potentials

a) For $V\left( \theta \right) =V_{0}e^{-2\frac{\left( k^{2}-2\right) }{k}%
\theta }$ , $k^{2}-2\neq 0$ we have the Noether symmetries $2t\partial
_{t}+H^{i}~,~t^{2}\partial _{t}+tH^{i}$ with Noether integrals%
\begin{equation}
I_{H_{1}}=2tE-\frac{d}{dt}\left( \frac{1}{2}\frac{r^{2}e^{2k\theta }}{k^{2}-1%
}\right) ~,~I_{H_{2}}=t^{2}E-t\frac{d}{dt}\left( \frac{1}{2}\frac{%
r^{2}e^{2k\theta }}{k^{2}-1}\right) +\frac{1}{2}\frac{r^{2}e^{2k\theta }}{%
k^{2}-1}.  \label{CLN.100c}
\end{equation}%
We note that in this case the system is the Ermakov-Pinney dynamical system
(because it admits the Noether symmetry algebra the $sl(2,R),\ $hence the
Lie symmetry algebra is at least $sl(2,R))$ .

b) For $V\left( \theta \right) =V_{0}e^{-2\frac{\left( k^{2}-2\right) }{k}%
\theta }-\frac{N_{0}^{2}m}{k^{2}-1}e^{2k\theta }~$ , $k^{2}-2\neq 0$ we have
the Noether symmetries $\frac{2}{\sqrt{m}}e^{\pm \sqrt{m}t}\partial _{t}\pm
e^{\pm \sqrt{m}t}H^{i}$ , $m=$constant with Noether integrals%
\begin{equation}
I_{+,-}=e^{\pm 2\sqrt{m}t}\left( \frac{1}{\sqrt{m}}E\mp \frac{d}{dt}\left(
\frac{1}{2}\frac{r^{2}e^{2k\theta }}{k^{2}-1}\right) +2\sqrt{m}\left( \frac{1%
}{2}\frac{r^{2}e^{2k\theta }}{k^{2}-1}\right) \right)  \label{CLN.105}
\end{equation}%
For this potential the Noether symmetries form the $sl\left( 2,R\right) $
Lie algebra, i.e the dynamical system is the two dimensional Kepler-Ermakov
system Therefore it admits the Ermakov - Pinney invariant which we may
construct with the use of the Noether symmetries or with the use of the
corresponding Killing Tensor (see Proposition \ref{prop ghv})

\item The case $V\left( \theta \right) =0$ corresponds the free particle
(see chapter \ref{chapter1}).
\end{enumerate}

\subsubsection{Case $\left\vert k\right\vert =1$}

We have to consider two cases i.e. $k=1$ and $k=-1.$

\subparagraph{Case $k=1$}

The KVs of the kinematic metric are the $K_{k=1}^{1,2}$ of (\ref{KV.1},\ref%
{KV.2}) and the vector
\begin{equation}
K_{k=1}^{3}=-r\left( \ln \left( re^{-\theta }\right) -1\right) \partial
_{r}+\ln \left( re^{-\theta }\right) \partial _{\theta }.  \label{CLN.68}
\end{equation}%
The vectors $K_{k=1}^{1,2}$ are gradient and $K^{3}$ is non-gradient. The
HV\ is gradient and it is given by%
\begin{equation}
H^{i}=\frac{1}{4}r\left( 2\ln \left( re^{-\theta }\right) +3\right) \partial
_{r}-\frac{1}{2}\left( \ln \left( re^{-\theta }\right) +\frac{1}{2}\right)
\partial _{\theta }  \label{CLN.69}
\end{equation}

\subparagraph{Case $k=-1$}

The KVs of the kinematic metric are $K_{k=-1}^{1,2}$ of (\ref{KV.1},\ref%
{KV.2}) and the vector
\begin{equation}
\bar{K}^{3}=r\left( \ln \left( re^{\theta }\right) -1\right) \partial
_{r}+\ln \left( re^{\theta }\right) \partial _{\theta }.  \label{CLN.70}
\end{equation}%
The vectors $\bar{K}^{1,2}$ are gradient and $\bar{K}^{3}$ is non-gradient.
The gradient HV is given by%
\begin{equation}
\bar{H}^{i}=\frac{1}{4}r\left( 2\ln \left( re^{\theta }\right) +3\right)
\partial _{r}+\frac{1}{2}\left( \ln \left( re^{\theta }\right) +\frac{1}{2}%
\right) \partial _{\theta }  \label{CLN.71}
\end{equation}

In the following we consider only the case $k=1$. The results for the case $%
k=-1$ are found if we make the change $\theta _{\left( k=-1\right) }=-\bar{%
\theta}$.

Using theorem \ref{The Noether symmetries of a conservative system} and
making simple calculations we find the following results

\begin{enumerate}
\item Noether symmetries generated by the KV $K^{1}$.

a) If $V\left( \theta \right) =V_{0}e^{2\theta }~$then we have the extra
Noether symmetries $K^{1}~,~tK^{1}$ with Noether integrals the (\ref{CLN.75a}%
) with $k=1$.

b) If $V\left( \theta \right) =\left( V_{0}e^{2\theta }-\frac{m}{2}\theta
e^{2\theta }\right) $, then we have the Noether symmetries $e^{\pm \sqrt{m}%
t}K^{1}$ with Noether integrals (\ref{CLN.99}) with $k=1$.

\item Noether symmetries generated by the KV $K^{2}$.

a) If $V\left( \theta \right) =V_{0}e^{-2\theta }$ then we have the Noether
symmetries $K^{2}~,~tK^{2}$ with Noether integrals%
\begin{equation}
I_{2}^{\prime }=\frac{d}{dt}\left( \theta -\ln r\right) ~,~I_{2}^{\prime }=t%
\left[ \frac{d}{dt}\left( \theta -\ln r\right) \right] -\left( \theta -\ln
r\right)  \label{CLN.72}
\end{equation}

b) If $V\left( \theta \right) =V_{0}e^{-2\theta }-\frac{1}{4}pe^{2\theta }\,$
then we have the Noether symmetries $K^{2}~,~tK^{2}$ with Noether integrals%
\begin{equation}
I_{1}^{\prime }=\frac{d}{dt}\left( \theta -\ln r\right) -pt~,~I_{2}^{\prime
}=t\left[ \frac{d}{dt}\left( \theta -\ln r\right) \right] -\left( \theta
-\ln r\right) -\frac{1}{2}pt^{2}  \label{CLN.73}
\end{equation}

\item If $V\left( \theta \right) =0$ then the system becomes the free
particle and admits seven extra Noether symmetries.
\end{enumerate}

As we have remarked the results for $k=-1$ are obtained directly from those
for $k=1$ if we make the substitution $\theta _{\left( k=-1\right) }=-\bar{%
\theta}$. Therefore there is no need to state them explicitly.

In the next section using the Noether symmetries for the potentials we have
found, we determine the analytic solution for each case.

\subsection{Case A: Analytic solutions for $k=1$}

We introduce new coordinates $u,v$ by the relations%
\begin{equation}
u=\left( \theta -\ln r\right) ~,~v=\frac{1}{2}e^{2\theta }r^{2}
\label{CLN.76}
\end{equation}%
with inverse relations:%
\begin{equation}
r=\left( 2ve^{-2u}\right) ^{\frac{1}{4}}~,~~\theta =\frac{1}{4}\ln \left(
2ve^{2u}\right) .  \label{CLN.77}
\end{equation}

In the new coordinates $u,v$ the Lagrangian (\ref{CLN.33}) of the field
equations becomes%
\begin{equation}
L\left( u,v,\dot{u},\dot{v}\right) =\frac{N_{0}^{2}}{2}\left( \dot{u}\dot{v}%
\right) -U\left( u,v\right)
\end{equation}%
and the field equations are
\begin{equation*}
E=\frac{N_{0}^{2}}{2}\left( \dot{u}\dot{v}\right) +U\left( u,v\right)
\end{equation*}%
\begin{eqnarray*}
\frac{N_{0}^{2}}{2}\ddot{u}+U_{,v} &=&0 \\
\frac{N_{0}^{2}}{2}\ddot{v}+U_{,u} &=&0
\end{eqnarray*}%
where the potential $U\left( u,v\right) $ is one of the potentials we have
found in the last section. In the new coordinates we have ($p,m$ are
constants; recall that in general $U(r,\theta )=r^{2}V\left( \theta \right) $%
):

\begin{itemize}
\item
\begin{equation}
U(r,\theta )=V_{0}r^{2}e^{-2\theta }\rightarrow U\left( u,v\right)
=V_{0}e^{-2u}  \label{CLN.79}
\end{equation}

\item
\begin{equation}
U(r,\theta )=V_{0}r^{2}e^{2\theta }\rightarrow U\left( u,v\right) =2V_{0}v
\label{CLN.80}
\end{equation}

\item
\begin{equation}
U(r,\theta )=r^{2}\left( V_{0}e^{-2\theta }-\frac{1}{4}pe^{2\theta }\right)
\rightarrow U\left( u,v\right) =V_{0}e^{-2u}-\frac{p}{2}v  \label{CLN.81}
\end{equation}

\item
\begin{equation}
U(r,\theta )=r^{2}\left( V_{0}e^{2\theta }-\frac{m}{2}\theta e^{2\theta
}\right) \rightarrow U\left( u,v\right) =2\left( V_{0}-m\frac{\ln 2}{8}%
\right) v-\frac{m}{4}v\ln v-\frac{m}{2}uv  \label{CLN.82}
\end{equation}

\item
\begin{equation}
U(r,\theta )=0\rightarrow U\left( u,v\right) =0  \label{CLN.83}
\end{equation}%
The Hamiltonian equals%
\begin{equation}
E=\frac{N_{0}^{2}}{2}\left( \dot{u}\dot{v}\right) +U\left( u,v\right) .
\label{CLN.84}
\end{equation}
\end{itemize}

We write Lagrange equations for each potential and solve them taking into
consideration the first integrals for each Noether symmetry we have found
and the constraint imposed by the Hamiltonian. \ For each of the potentials
above we find the corresponding analytic solution given below.

\begin{itemize}
\item $U\left( u,v\right) =V_{0}e^{-2u}$%
\begin{equation}
u\left( t\right) =u_{1}t+u_{2},\text{ }v\left( t\right) =\frac{e^{-2u_{2}}}{%
u_{1}^{2}}\frac{V_{0}}{N_{0}^{2}}e^{-2u_{1}t}+v_{1}t+v_{2}  \label{CLN.85}
\end{equation}%
with Hamiltonian constraint%
\begin{equation}
E=\frac{u_{1}v_{1}}{2N_{0}^{2}}.  \label{CLN.86}
\end{equation}

\item $U\left( u,v\right) =2V_{0}v$%
\begin{equation}
u\left( t\right) =-\frac{2V_{0}}{N_{0}^{2}}t^{2}+u_{1}t+u_{2},v\left(
t\right) =v_{1}t+v_{2}  \label{CLN.87}
\end{equation}%
with Hamiltonian constrain
\begin{equation}
E=2V_{0}v_{2}+N_{0}^{2}\frac{u_{1}v_{1}}{2}.  \label{CLN.88}
\end{equation}

\item $U\left( u,v\right) =V_{0}e^{-2u}-\frac{p}{2}v$

where
\begin{equation}
v_{1}=\frac{2V_{0}\sqrt{\pi }}{p^{\frac{3}{2}}N_{0}}\exp \left( \frac{%
u_{1}^{2}N_{0}^{2}}{p}-2u_{2}\right) ~,~v_{2}=\frac{2V_{0}}{p}e^{-2u_{2}}
\label{CLN.89}
\end{equation}%
with Hamiltonian constraint%
\begin{equation}
E=N_{0}^{2}\frac{v_{3}u_{1}}{2}-\frac{v_{4}p}{2}.  \label{CLN.90}
\end{equation}

\item $U\left( u,v\right) =2\left( V_{0}-m\frac{\ln 2}{8}\right) v-\frac{m}{4%
}v\ln v-\frac{m}{2}uv$%
\begin{eqnarray}
u\left( t\right) &=&u_{1}e^{\frac{1}{N_{0}}\sqrt{m}t}+u_{2}e^{-\frac{1}{N_{0}%
}\sqrt{m}t}-\frac{\sqrt{m}}{2N_{0}}t+\frac{4V_{0}}{m}-\ln \left( \sqrt{2}%
\right) -\frac{1}{2}~,~  \label{CLN.91} \\
v\left( t\right) &=&e^{\frac{1}{N_{0}}\sqrt{m}t}~~,~E=-u_{2}m \\
u\left( t\right) &=&u_{1}e^{\frac{1}{N_{0}}\sqrt{m}t}+u_{2}e^{-\frac{1}{N_{0}%
}\sqrt{m}t}+\frac{\sqrt{m}}{2N_{0}}t+\frac{4V_{0}}{m}-\ln \left( \sqrt{2}%
\right) -\frac{1}{2}~,~  \label{CLN.92} \\
v\left( t\right) &=&e^{-\frac{1}{N_{0}}\sqrt{m}t}~,~E=-u_{1}m.
\end{eqnarray}

\item $U\left( u,v\right) =0$ (the free particle)%
\begin{equation}
u\left( t\right) =u_{1}t+u_{2}~,~v\left( t\right) =v_{1}t+v_{2}
\label{CLN.93}
\end{equation}%
with Hamiltonian constraint%
\begin{equation}
E=\frac{N_{0}^{2}}{2}u_{1}v_{1}.  \label{CLN.94}
\end{equation}
\end{itemize}

\subsection{Case A: Analytic solutions for $\left\vert k\right\vert \neq 1$}

\label{AnNMSFA}

When \ $\left\vert k\right\vert \neq 1$ we have to consider two cases $%
k^{2}>1$ and $k^{2}<1.$ Both cases it is convenient to be discussed if we
use as variables the functions $S_{1}\left( r,\theta \right) ,$ $S_{2}\left(
r,\theta \right) $ which generate the Killing vectors (i.e. $%
K^{I,i}=S_{I}^{,i}$ $I=1,2).$ A standard calculation gives (see appendix B
for details) :

\begin{eqnarray}
K^{1} &=&\frac{e^{\left( 1-k\right) \theta }r^{k}}{N_{0}^{2}}\left(
-\partial _{r}+\frac{1}{r}\partial _{\theta }\right) ~,~S_{1}\left( r,\theta
\right) =\frac{r^{1+k}e^{\left( 1+k\right) \theta }}{\left( k+1\right) }
\label{CLN.95} \\
K^{2} &=&\frac{e^{-\left( 1+k\right) \theta }r^{-k}}{N_{0}^{2}}\left(
\partial _{r}+\frac{1}{r}\partial _{\theta }\right) ,~S_{2}\left( r,\theta
\right) =\frac{r^{1-k}e^{-\left( 1-k\right) \theta }}{k-1}.  \label{CLN.96}
\end{eqnarray}%
Because the metric is flat the new variables $S_{1}\left( r,\theta \right) ,$
$S_{2}\left( r,\theta \right) $ are Cartesian and will be denoted with $x,y.$
So we write:
\begin{equation}
x=\frac{r^{1+k}e^{-\left( 1+k\right) \theta }}{\left( k+1\right) }~,~y=\frac{%
r^{1-k}e^{\left( -1+k\right) \theta }}{k-1}.  \label{CLN.97}
\end{equation}

The inverse transformation is:

for $k^{2}>1$%
\begin{eqnarray}
\theta &=&\frac{1}{2\left( 1-k^{2}\right) }\ln \left( \frac{\left(
k^{2}-1\right) ^{1-k}}{\left( k-1\right) ^{2}}\frac{x^{1-k}}{y^{1+k}}\right)
\label{CLN.98} \\
r &=&\sqrt{\left( k^{2}-1\right) xy}\left( \frac{\left( k^{2}-1\right) ^{1-k}%
}{\left( k-1\right) ^{2}}\frac{x^{1-k}}{y^{1+k}}\right) ^{\frac{k}{2\left(
k^{2}-1\right) }}  \label{CLN.98a}
\end{eqnarray}
and

for $k^{2}<1$
\begin{eqnarray}
\theta &=&\frac{1}{2\left( 1-k^{2}\right) }\ln \left( \frac{\left(
1-k^{2}\right) ^{1-k}}{\left( 1-k\right) ^{2}}\frac{\bar{x}^{1-k}}{\bar{y}%
^{1+k}}\right)  \label{CLN.110} \\
r &=&\sqrt{\left( 1-k^{2}\right) \bar{x}\bar{y}}\left( \frac{\left(
1-k^{2}\right) ^{1-k}}{\left( 1-k\right) ^{2}}\frac{\bar{x}^{1-k}}{\bar{y}%
^{1+k}}\right) ^{\frac{k}{2\left( k^{2}-1\right) }}  \label{CLN.111}
\end{eqnarray}%
Note that in the second case we have written $\bar{x},\bar{y}$ while we have
kept the $x,y$ notation for the case $k^{2}>1.$

\subsubsection{The case $k^{2}>1$}

\qquad Before we look for analytic solutions we transform the Lagrangian in
the canonical coordinates $x,y.$ Using the transformation relations (\ref%
{CLN.98}), (\ref{CLN.98a}) we find that in the coordinates $x,y$ the
Lagrangian (\ref{CLN.31}) takes the form%
\begin{equation}
L\left( x,y,\dot{x},\dot{y}\right) =\frac{N_{0}^{2}}{2}\dot{x}\dot{y}%
-U\left( x,y\right)  \label{CLN.115}
\end{equation}%
where $U\left( x,y\right) =r^{2}V\left( \theta \right) $ where $V\left(
\theta \right) $ is one of the potentials computed above. In the coordinates
$x,y$ these potentials are as follows:

\begin{itemize}
\item
\begin{equation}
U_{1}\left( x,y\right) =V_{0}r^{2}e^{2k\theta }=V_{0}\left( k^{2}-1\right) xy
\label{CLN.116}
\end{equation}

\begin{equation}
U_{2}\left( x,y\right) =V_{0}r^{2}e^{+2\theta }=V_{0}\left( k+1\right) ^{%
\frac{2}{1+k}}~x^{\frac{^{2}}{1+k}}  \label{CLN.117}
\end{equation}%
\begin{equation}
U_{3}\left( x,y\right) =V_{0}r^{2}e^{-2\theta }=V_{0}\left( k-1\right) ^{%
\frac{2}{1-k}}y^{\frac{2}{1-k}}  \label{CLN.118}
\end{equation}

\item
\begin{equation*}
U_{4}\left( x,y\right) =V_{0}r^{2}e^{-2\frac{\left( k^{2}-2\right) }{k}%
\theta }=V_{0}\frac{\left( k^{2}-1\right) ^{\frac{2}{k}-1}}{\left(
k-1\right) ^{\frac{4}{k}}}\frac{1}{y^{2}}\left( \frac{x}{y}\right) ^{\frac{2%
}{k}-1}.
\end{equation*}
\end{itemize}

To the potentials $U_{2},U_{3},$ $U_{4}$ we have to add three more which we
obtain by the addition of the potential of the harmonic oscillator.
Therefore finally we have 7 potentials. The extra potentials are

\begin{itemize}
\item
\begin{eqnarray}
U_{5}\left( x,y\right) &=&V_{0}r^{2}e^{+2\theta }+mr^{2}e^{2k\theta }=\bar{V}%
_{0+}~x^{\frac{^{2}}{1+k}}+\bar{m}xy  \label{CLN.120} \\
U_{6}\left( x,y\right) &=&r^{2}e^{-2\theta }+mr^{2}e^{2k\theta }=\bar{V}%
_{0-}y^{\frac{2}{1-k}}+\bar{m}xy  \label{CLN.121}
\end{eqnarray}%
where $\bar{V}_{0+}=V_{0}\left( k+1\right) ^{\frac{2}{1+k}}~,~\bar{V}%
_{0-}=\left( k-1\right) ^{\frac{2}{1-k}}~,~\bar{m}=m\left( k^{2}-1\right) $

\item
\begin{equation}
U_{7}\left( x,y\right) =V_{0}r^{2}e^{-2\frac{\left( k^{2}-2\right) }{k}%
\theta }+mr^{2}e^{2k\theta }=\bar{V}_{0}\frac{1}{y^{2}}\left( \frac{x}{y}%
\right) ^{\frac{2}{k}-1}+\bar{m}xy  \label{CLN.122}
\end{equation}%
where $\bar{V}_{0}=V_{0}\frac{\left( k^{2}-1\right) ^{\frac{2}{k}-1}}{\left(
k-1\right) ^{\frac{4}{k}}}$.

\item And the free particle potential
\begin{equation}
U_{8}\left( x,y\right) =0
\end{equation}
\end{itemize}

These expressions allow us to write for each potential the Lagrangian and
the Hamiltonian constraint in the coordinates $x,y$. This means that we
obtain the corresponding field equations in the coordinates $x,y$ to
determine their solution.

\paragraph{Analytic solutions}

The solution of the field equations for each potential is a formal and
lengthy operation which adds nothing but unnecessary material to the matter.
What is interesting is of course the final answer for each case and this is
what we give below for each of the potentials above.

\begin{itemize}
\item $U_{1}\left( x,y\right) =V_{0}\left( k^{2}-1\right) xy$%
\begin{eqnarray}
x\left( t\right) &=&x_{1}\sin \left( \omega t\right) +x_{2}\cos \left(
\omega t\right)  \label{CLN.123} \\
y\left( t\right) &=&y_{1}\sin \left( \omega t\right) +y_{2}\cos \left(
\omega t\right)  \label{CLN.124}
\end{eqnarray}%
where $\omega ^{2}=\frac{2V_{0}\left( k^{2}-1\right) }{N_{0}^{2}}$ and the
Hamiltonian is $\ $%
\begin{equation*}
E=V_{0}\left( k^{2}-1\right) \left( x_{1}y_{1}+x_{2}y_{2}\right)
\end{equation*}

\item $U_{2}\left( x,y\right) =V_{0}\left( k+1\right) ^{\frac{2}{1+k}}~x^{%
\frac{^{2}}{1+k}}$

When $k\neq -3$%
\begin{eqnarray}
x\left( t\right) &=&x_{1}t+x_{2}  \label{CLN.125} \\
y\left( t\right) &=&-\frac{2\bar{V}\left( k+1\right) \left(
x_{1}t+x_{2}\right) ^{\left( 1+\frac{2}{1+k}\right) }}{x_{1}^{2}\left(
3+k\right) N_{0}^{2}}+y_{1}t+y_{2}  \label{CLN.126}
\end{eqnarray}%
where $\bar{V}=V_{0}\left( k+1\right) ^{\frac{2}{1+k}},$ and the Hamiltonian
is ~%
\begin{equation*}
E=\frac{y_{1}x_{1}N_{0}^{2}}{2}.
\end{equation*}

~When $k=-3$
\begin{equation}
y\left( t\right) =-2\frac{\bar{V}}{N_{0}^{2}x_{1}^{2}}\ln \left(
x_{1}t+x_{2}\right) +y_{1}t+y_{2}.  \label{CLN.127}
\end{equation}%
$x(t),$ $E$\ being the same.

\item $U_{3}\left( x,y\right) =V_{0}\left( k-1\right) ^{\frac{2}{1-k}}y^{%
\frac{2}{1-k}}$

When $k\neq 3$%
\begin{eqnarray}
x\left( t\right) &=&\frac{2\bar{V}\left( k-1\right) \left(
y_{1}t+y_{2}\right) ^{1+\frac{2}{k-1}}}{y_{1}^{2}\left( k-3\right) N_{0}^{2}}%
+x_{1}t+x_{2}  \label{CLN.128} \\
y\left( t\right) &=&y_{1}t+y_{2}  \label{CLN.129}
\end{eqnarray}%
where $\bar{V}=V_{0}\left( k-1\right) ^{\frac{2}{1-k}}$ , $k\neq 3$and the
Hamiltonian is ~~%
\begin{equation*}
E=\frac{y_{1}x_{1}N_{0}^{2}}{2}.
\end{equation*}

When $k=3~;$\
\begin{equation}
x\left( t\right) =-\frac{2\bar{V}}{N_{0}^{2}y_{1}^{2}}\ln \left(
y_{1}t+y_{2}\right) +x_{1}t+x_{2}.  \label{CLN.130}
\end{equation}%
$y(t),E$ being the same

\item $U_{5}\left( x,y\right) =\bar{V}_{0+}~x^{\frac{^{2}}{1+k}}+\bar{m}xy$%
\begin{eqnarray}
x\left( t\right) &=&x_{1}\sin \left( \omega t+\omega _{0}\right) \\
y\left( t\right) &=&\cos \left( \omega t+\omega _{0}\right) \left( y_{1}+2%
\frac{\omega }{\bar{m}}\int \frac{y_{2}-x_{1}\bar{V}_{0+}\sin \left( \omega
t+\omega _{0}\right) ^{\frac{2}{1+k}}}{x_{1}\left( \cos \left( \omega
t+\omega _{0}\right) +1\right) }dt\right)
\end{eqnarray}%
where $\omega ^{2}=\frac{2\bar{m}}{N_{0}^{2}}$ and $E=y_{2}.$

\item $U_{6}\left( x,y\right) =r^{2}e^{-2\theta }+mr^{2}e^{2k\theta }=\bar{V}%
_{0-}y^{\frac{2}{1-k}}+\bar{m}xy$%
\begin{eqnarray}
x\left( t\right) &=&\cos \left( \omega t+\omega _{0}\right) \left( x_{1}+2%
\frac{\omega }{\bar{m}}\int \frac{x_{2}-y_{1}\bar{V}_{0-}\sin \left( \omega
t+\omega _{0}\right) ^{\frac{2}{1-k}}}{y_{1}\left( \cos \left( \omega
t+\omega _{0}\right) +1\right) }dt\right) \\
y\left( t\right) &=&y_{1}\sin \left( \omega t+\omega _{0}\right)
\end{eqnarray}%
where $\omega ^{2}=\frac{2\bar{m}}{N_{0}^{2}}$and $E=x_{2}.$

\item $U_{4,7}\left( x,y\right) =\bar{V}_{0}\frac{1}{y^{2}}\left( \frac{x}{y}%
\right) ^{\frac{2}{k}-1}+\bar{m}xy$.~$\left( U_{4}\text{ is for }\bar{m}%
=0\right) $

This is the Ermakov Pinney system. To solve it, it is convenient to go to
spherical coordinates. We consider the coordinate transformation%
\begin{equation}
x=ze^{w}~,~y=ze^{-w}~  \label{CLN.131}
\end{equation}%
and the Lagrangian takes the form%
\begin{equation}
L\left( z,w,\dot{z},\dot{w}\right) =\frac{N_{0}^{2}}{2}\left( \dot{z}%
^{2}-z^{2}\dot{w}^{2}\right) -\frac{\bar{V}_{0}}{z^{2}}e^{\frac{4}{k}w}-\bar{%
m}z^{2}  \label{CLN.132}
\end{equation}%
whereas the Hamiltonian becomes%
\begin{equation}
E=\frac{N_{0}^{2}}{2}\left( \dot{z}^{2}-z^{2}\dot{w}^{2}\right) +\frac{\bar{V%
}_{0}}{z^{2}}e^{\frac{4}{k}w}+\bar{m}z^{2}.  \label{CLN.133}
\end{equation}%
This system admits the Ermakov-Lewis invariant, which is%
\begin{equation}
J_{EL}=z^{4}\dot{w}^{2}-2\frac{\bar{V}_{0}}{N_{0}^{2}}e^{\frac{4}{k}w}.
\label{CLN.134}
\end{equation}%
Using the Ermakov invariant the Hamiltonian becomes%
\begin{equation}
E=\frac{N_{0}^{2}}{2}\dot{z}^{2}-N_{0}^{2}\frac{J_{EL}}{2z^{2}}+\bar{m}z^{2}.
\label{CLN.135}
\end{equation}%
This is the Hamiltonian of the Ermakov Pinney equation:%
\begin{equation}
\ddot{z}+2\bar{m}z+N_{0}^{2}\frac{J_{EL}}{z^{3}}=0  \label{CLN.136}
\end{equation}%
whose solution is
\begin{eqnarray}
z\left( t\right) &=&\left( l_{0}z_{1}\left( t\right) +l_{1}z_{2}\left(
t\right) +l_{3}\right) ^{\frac{1}{2}}  \label{CLN.137} \\
e^{\frac{4}{k}w\left( t\right) } &=&\frac{N_{0}^{2}J_{EL}}{2\bar{V}_{0}}%
\left[ \tanh ^{2}\left( \frac{2\sqrt{J_{EL}}}{k}\left( \int \frac{dt}{%
z^{2}\left( t\right) }+l_{4}\right) \right) -1\right]  \label{CLN.138}
\end{eqnarray}%
where $z_{1,2}\left( t\right) $ are solutions of the differential equation ~$%
\ddot{z}+2\bar{m}z=0$ and $l_{0-4}$ are constants.

\item $U_{8}(x,y)=0$

This is the free particle whose solution is%
\begin{eqnarray}
x\left( t\right) &=&x_{1}t+x_{2}~,~y\left( t\right) =y_{1}t+y_{2}
\label{CLN.139} \\
E &=&\frac{N_{0}^{2}}{2}x_{1}y_{1}  \label{CLN.140}
\end{eqnarray}
\end{itemize}

\subsubsection{The case $k^{2}<1$}

In this case the canonical coordinates are the $\bar{x},\bar{y}.$ Using the
transformation equations (\ref{CLN.110}), (\ref{CLN.111}) we write the
Lagrangian (\ref{CLN.31}) as follows:%
\begin{equation}
L\left( \bar{x},\bar{y},\overset{\cdot }{\bar{x},}\overset{\cdot }{\bar{y}}%
\right) =-\frac{N_{0}^{2}}{2}\overset{\cdot }{\bar{x}}\overset{\cdot }{\bar{y%
}}-\bar{U}\left( \bar{x},\bar{y}\right)  \label{CLN.141}
\end{equation}%
where $U\left( \bar{x},\bar{y}\right) =r^{2}V\left( \theta \right) $ the
potentials $V\left( \theta \right) $ being as above. In the coordinates $%
\bar{x},\bar{y}$ the potentials $U\left( \bar{x},\bar{y}\right) $ are:

\begin{itemize}
\item
\begin{equation}
\bar{U}_{1}\left( \bar{x},\bar{y}\right) =V_{0}r^{2}e^{2k\theta
}=V_{0}\left( 1-k^{2}\right) \bar{x}\bar{y}  \label{CLN.142}
\end{equation}

\begin{equation}
\bar{U}_{2}\left( \bar{x},\bar{y}\right) =V_{0}r^{2}e^{+2\theta
}=V_{0}\left( k+1\right) ^{\frac{2}{1+k}}~\bar{x}^{\frac{^{2}}{1+k}}
\label{CLN.143}
\end{equation}%
\begin{equation}
\bar{U}_{3}\left( \bar{x},\bar{y}\right) =V_{0}r^{2}e^{-2\theta
}=V_{0}\left( 1-k\right) ^{\frac{2}{1-k}}\bar{y}^{\frac{2}{1-k}}
\label{CLN.144}
\end{equation}

\item
\begin{equation}
\bar{U}_{4}\left( \bar{x},\bar{y}\right) =V_{0}r^{2}e^{-2\frac{\left(
k^{2}-2\right) }{k}\theta }=V_{0}\frac{\left( 1-k^{2}\right) ^{\frac{2}{k}-1}%
}{\left( 1-k\right) ^{\frac{4}{k}}}\frac{1}{\bar{y}^{2}}\left( \frac{\bar{x}%
}{\bar{y}}\right) ^{\frac{2}{k}-1}  \label{CLN.145}
\end{equation}
\end{itemize}

As before we have three more potentials

\begin{itemize}
\item
\begin{eqnarray}
\bar{U}_{5}\left( \bar{x},\bar{y}\right) &=&V_{0}r^{2}e^{+2\theta }=\bar{V}%
_{0+}~\bar{x}^{\frac{^{2}}{1+k}}+\bar{m}\bar{x}\bar{y}  \label{CLN.146} \\
\bar{U}_{6}\left( \bar{x},\bar{y}\right) &=&V_{0}r^{2}e^{-2\theta }=\bar{V}%
_{0-}\bar{y}^{\frac{2}{1-k}}+\bar{m}\bar{x}\bar{y}  \label{CLN.147}
\end{eqnarray}%
where $\bar{V}_{0+}=V_{0}\left( k+1\right) ^{\frac{2}{1+k}}~,~\bar{V}%
_{0-}=\left( 1-k\right) ^{\frac{2}{1-k}}~,~\bar{m}=m\left( 1-k^{2}\right) .$

\item
\begin{equation}
\bar{U}_{7}\left( \bar{x},\bar{y}\right) =V_{0}r^{2}e^{-2\frac{\left(
k^{2}-2\right) }{k}\theta }+mr^{2}e^{2k\theta }=\bar{V}_{0}\frac{1}{\bar{y}%
^{2}}\left( \frac{\bar{x}}{\bar{y}}\right) ^{\frac{2}{k}-1}+\bar{m}\bar{x}%
\bar{y}  \label{CLN.147.1}
\end{equation}%
where $\bar{V}_{0}=V_{0}\frac{\left( 1-k^{2}\right) ^{\frac{2}{k}-1}}{\left(
1-k\right) ^{\frac{4}{k}}}$.

\item and the free particle potential.
\begin{equation}
\bar{U}_{8}\left( \bar{x},\bar{y}\right) =0
\end{equation}
\end{itemize}

\paragraph{Analytic solutions}

Working as in the case $k^{2}>1$ we find the following analytic solutions
and the associated Hamiltonian constraint for each of the potentials above.

\begin{itemize}
\item $\bar{U}_{1}\left( \bar{x},\bar{y}\right) =V_{0}\left( 1-k^{2}\right)
\bar{x}\bar{y}.$
\begin{eqnarray}
\bar{x}\left( t\right) &=&x_{1}\cosh \left( \omega t\right) +x_{2}\sinh
\left( \omega t\right)  \label{CLN.147.2} \\
\bar{y}\left( t\right) &=&y_{1}\cosh \left( \omega t\right) +y_{2}\sinh
\left( \omega t\right)  \label{CLN.147.3}
\end{eqnarray}%
and the Hamiltonian $E=V_{0}m\left( x_{1}y_{1}-x_{2}y_{2}\right) $ where $%
\omega ^{2}=\frac{2V_{0}\left( 1-k^{2}\right) }{N_{0}^{2}}$.

\item $\bar{U}_{2}\left( \bar{x},\bar{y}\right) =V_{0}\left( k+1\right) ^{%
\frac{2}{1+k}}~\bar{x}^{\frac{^{2}}{1+k}}$%
\begin{eqnarray}
\bar{x}\left( t\right) &=&x_{1}t+x_{2}  \label{CLN.147.4} \\
\bar{y}\left( t\right) &=&\frac{2\bar{V}\left( k+1\right) \left(
x_{1}t+x_{2}\right) ^{\left( 1+\frac{2}{1+k}\right) }}{x_{1}^{2}\left(
3+k\right) N_{0}^{2}}+y_{1}t+y_{2}  \label{CLN.147.5}
\end{eqnarray}%
where $\bar{V}=V_{0}\left( k+1\right) ^{\frac{2}{1+k}}$and the Hamiltonian $%
E=-\frac{N_{0}^{2}}{2}x_{1}y_{1}$.

\item $\bar{U}_{3}\left( \bar{x},\bar{y}\right) =V_{0}\left( 1-k\right) ^{%
\frac{2}{1-k}}\bar{y}^{\frac{2}{1-k}}$%
\begin{eqnarray}
\bar{x}\left( t\right) &=&\frac{2\bar{V}\left( 1-k\right) \left(
y_{1}t+y_{2}\right) ^{1+\frac{2}{k-1}}}{y_{1}^{2}\left( k-3\right) N_{0}^{2}}%
+x_{1}t+x_{2}  \label{CLN.147.6} \\
\bar{y}\left( t\right) &=&y_{1}t+y_{2}  \label{CLN.147.7}
\end{eqnarray}%
where $\bar{V}=V_{0}\left( k+1\right) ^{\frac{2}{1+k}}$and the Hamiltonian
constrain $E=-\frac{N_{0}^{2}}{2}x_{1}y_{1}$.

\item $\bar{U}_{5}\left( \bar{x},\bar{y}\right) =\bar{V}_{0+}~\bar{x}^{\frac{%
^{2}}{1+k}}+\bar{m}\bar{x}\bar{y}$%
\begin{eqnarray*}
\bar{x}\left( t\right) &=&\bar{x}_{1}\sinh \left( \omega t+\omega _{0}\right)
\\
\bar{y}\left( t\right) &=&\cosh \left( \omega t+\omega _{0}\right) \left(
\bar{y}_{1}-\frac{2\omega }{\bar{m}}\int \frac{E-\bar{x}_{1}\bar{V}%
_{0+}\sinh \left( \omega t+\omega _{0}\right) ^{\frac{2}{1+k}}}{\bar{x}%
_{1}\left( \cosh \left( \omega t+\omega _{0}\right) +1\right) }\right)
\end{eqnarray*}%
where $\omega ^{2}=\frac{2\bar{m}}{N_{0}^{2}}$

\item $\bar{U}_{6}\left( \bar{x},\bar{y}\right) =\bar{V}_{0-}\bar{y}^{\frac{2%
}{1-k}}+\bar{m}\bar{x}\bar{y}$%
\begin{eqnarray*}
\bar{x}\left( t\right) &=&\cosh \left( \omega t+\omega _{0}\right) \left(
\bar{x}_{1}-\frac{2\omega }{\bar{m}}\int \frac{E-\bar{y}_{1}\bar{V}%
_{0-}\sinh \left( \omega t+\omega _{0}\right) ^{\frac{2}{1-k}}}{\bar{y}%
_{1}\left( \cosh \left( \omega t+\omega _{0}\right) +1\right) }\right) \\
\bar{y}\left( t\right) &=&\bar{y}_{1}\sinh \left( \omega t+\omega _{0}\right)
\end{eqnarray*}%
where $\omega ^{2}=\frac{2\bar{m}}{N_{0}^{2}}$

\item $\bar{U}_{4,7}\left( \bar{x},\bar{y}\right) =\bar{V}_{0}\frac{1}{\bar{y%
}^{2}}\left( \frac{\bar{x}}{\bar{y}}\right) ^{\frac{2}{k}-1}+\bar{m}\bar{x}%
\bar{y}$.~$\left( \bar{U}_{4}\text{ is for }\bar{m}=0\right) $.

This is again the Ermakov Piney potential. To solve it we go to spherical
coordinates%
\begin{equation}
\bar{x}=ze^{w}~\ ,~\bar{y}=ze^{-w}~  \label{CLN.147.8}
\end{equation}%
in which the Lagrangian takes the form%
\begin{equation}
L\left( z,w,\dot{z},\dot{w}\right) =\frac{N_{0}^{2}}{2}\left( -\dot{z}%
^{2}+z^{2}\dot{w}^{2}\right) -\frac{\bar{V}_{0}}{z^{2}}e^{\frac{4}{k}w}-\bar{%
m}z^{2}  \label{CLN.147.9}
\end{equation}%
and the Hamiltonian%
\begin{equation}
E=\frac{N_{0}^{2}}{2}\left( -\dot{z}^{2}+z^{2}\dot{w}^{2}\right) +\frac{\bar{%
V}_{0}}{z^{2}}e^{\frac{4}{k}w}+\bar{m}z^{2}.  \label{CLN.147.10}
\end{equation}%
\qquad The Ermakov-Lewis invariant is%
\begin{equation}
J_{EL}=z^{4}\dot{w}^{2}+2\frac{\bar{V}_{0}}{N_{0}^{2}}e^{\frac{4}{k}w}
\label{CLN.147.11}
\end{equation}%
which when replaced in the Hamiltonian gives%
\begin{equation}
E=-\frac{N_{0}^{2}}{2}\dot{z}^{2}+N_{0}^{2}\frac{J_{EL}}{2z^{2}}+\bar{m}%
z^{2}.  \label{CLN.147.12}
\end{equation}%
This is the Hamiltonian of the Ermakov Pinney equation:%
\begin{equation}
\ddot{z}-2\bar{m}z-N_{0}^{2}\frac{J_{EL}}{z^{3}}=0  \label{CLN.147.13}
\end{equation}%
whose solution is
\begin{eqnarray}
z\left( t\right) &=&\left( l_{0}z_{1}\left( t\right) +l_{1}z_{2}\left(
t\right) +l_{3}\right) ^{\frac{1}{2}}  \label{CLN.147.14} \\
e^{\frac{4}{k}w\left( t\right) } &=&\frac{N_{0}^{2}J_{EL}}{2\bar{V}_{0}}%
\left[ 1-\tanh ^{2}\left( \frac{2\sqrt{J_{EL}}}{k}\left( \int \frac{dt}{%
z^{2}\left( t\right) }+l_{4}\right) \right) \right]  \label{CLN.147.15}
\end{eqnarray}%
where $z_{1,2}\left( t\right) $ are solutions of the ode~$\ddot{z}-2\bar{m}%
z=0$ and $l_{0-4}$ are constants.

\item $\bar{U}_{8}\left( \bar{x},\bar{y}\right) =0$

\item This is the free particle whose solution is
\begin{eqnarray}
\bar{x}\left( t\right) &=&x_{1}t+x_{2}~,~\bar{y}\left( t\right) =y_{1}t+y_{2}
\label{CLN.147.16} \\
E &=&-\frac{N_{0}^{2}}{2}x_{1}y_{1}.  \label{CLN.147.17}
\end{eqnarray}
\end{itemize}

\subsection{Case B: The 2d metric is conformally flat}

\label{AnNMSFB}

In this case it is preferable to work with (\ref{CLN.28}) which under the
coordinate transformation (\ref{CLN.29}) becomes:%
\begin{equation}
L=N^{2}\left( \theta \right) \left[ -\frac{1}{2}\dot{r}^{2}+\frac{1}{2}r^{2}%
\dot{\theta}^{2}\right] -r^{2}V\left( \theta \right) .  \label{CLN.150a}
\end{equation}%
The kinetic metric in this case is not flat (i.e. $R_{\left( 2\right) }\neq
0)$ but of course it is conformally flat being a two dimensional metric. Its
conformal algebra is infinity dimensional however it has a closed subalgebra
consisting of the following vectors (this is the special conformal algebra
of $M^{2}$):%
\begin{eqnarray}
X^{1} &=&\cosh \theta \partial _{r}-\frac{1}{r}\sinh \theta \partial
_{\theta }~~,~X^{2}=\sinh \theta \partial _{r}-\frac{1}{r}\cosh \theta
\partial _{\theta }  \notag \\
X^{3} &=&\partial _{\theta }~~\ ,~X^{4}=r\partial _{r}~~,~X^{5}=\frac{1}{2}%
r^{2}\cosh \theta \partial _{r}+\frac{1}{2}r\sinh \theta \partial _{\theta }
\notag \\
X^{6} &=&\frac{1}{2}r^{2}\sinh \theta \partial _{r}+\frac{1}{2}r\cosh \theta
\partial _{\theta }  \label{CLN.150}
\end{eqnarray}

Writing ~$L_{X^{I}}g_{ij}=2C_{I}\left( r,\theta \right) g_{ij}$ we find the
conformal factors of the CKVs $X^{I}$ $I=1,...6$ above in terms of the the
conformal function. The result is:%
\begin{eqnarray}
C_{1}\left( r,\theta \right) &=&-\frac{1}{r}\sinh \left( \theta \right)
\frac{N_{,\theta }}{N}~,~C_{2}\left( r,\theta \right) =-\frac{1}{r}\cosh
\left( \theta \right) \frac{N_{,\theta }}{N}  \notag \\
C_{3}\left( r,\theta \right) &=&\frac{N_{,\theta }}{N}~~,~C_{4}\left(
r,\theta \right) =1~,~~C_{5}\left( r,\theta \right) =\frac{r}{2}\left( \frac{%
2N\cosh \theta +\sinh \theta N_{\theta }}{N}\right)  \notag \\
C_{6}\left( r,\theta \right) &=&\frac{r}{2}\left( \frac{2N\sinh \theta
+\cosh \theta N_{\theta }}{N}\right)  \label{CLN.153}
\end{eqnarray}

The $N\left( \theta \right) \neq e^{c\theta }$, otherwise the kinetic metric
of the Lagrangian (\ref{CLN.150a}) is flat (the Ricci scalar vanishes) and
we return to the previous case A. This means that the vectors $X^{I}$ $%
I=1,...6$ except the $I=4$ are proper CKVs therefore they do not give (in
general) a Noether symmetry. The vector $X_{4}$ is a non-gradient HV which
does not also produce a Noether symmetry. Therefore according to theorem \ref%
{The Noether symmetries of a conservative system} only Killing vectors are
possible to serve as Noether symmetries. KVs do not exist but for special
forms of the conformal function $N(\theta )$. Each such form of $N(\theta )$
results in a potential $V(\theta )\ $hence to a scalar field potential which
admits Noether point symmetries. In the following we determine the possible $%
N(\theta )$ which lead to a KV and give the corresponding Noether symmetry
and the corresponding Noether integral which will be used for the solution
of the field equations.

\begin{enumerate}
\item If $N\left( \theta \right) =\frac{N_{0}}{\cosh \left( 2\theta \right)
-1}$ then $X^{5}$ is a non gradient KV and a Noether symmetry for the
Lagrangian (\ref{CLN.150a}) for the potential
\begin{equation}
V\left( \theta \right) =\frac{V_{0}}{\cosh \left( 2\theta \right) -1}~\text{%
or }V\left( \theta \right) =0  \label{CLN.155}
\end{equation}%
The corresponding Noether integral is
\begin{equation}
I_{X^{5}}=\frac{N_{0}^{2}r^{2}}{\left( \cosh \left( 2\theta \right)
-1\right) ^{2}}\left( r\dot{\theta}\sinh \theta -\dot{r}\cosh \theta \right)
.  \label{CLN.156}
\end{equation}

\item If $N\left( \theta \right) =\frac{N_{0}}{\cosh \left( 2\theta \right)
+1}$ then $X^{6}$ is a non gradient KV, $X^{6}$ and a Noether symmetry for
the Lagrangian (\ref{CLN.150a}) if
\begin{equation}
V\left( \theta \right) =\frac{V_{0}}{\cosh \left( 2\theta \right) +1}~\text{%
or }V\left( \theta \right) =0  \label{CLN.157}
\end{equation}%
The corresponding Noether integral is
\begin{equation}
I_{X^{6}}=\frac{N_{0}^{2}r^{2}}{\left( \cosh \left( 2\theta \right)
+1\right) ^{2}}\left( r\dot{\theta}\cosh \theta -\dot{r}\sinh \theta \right)
\label{CLN.158}
\end{equation}

\item If $N\left( \theta \right) =\frac{N_{0}}{\cosh ^{2}\left( \theta
+\theta _{0}\right) }$ then the linear combination $%
X^{56}=c_{1}X^{5}+c_{2}X^{6}$ where $c_{1}=\sinh \left( \theta _{0}\right) $
and $c_{2}=\cosh \left( \theta _{0}\right) $. $X^{56}$ is a Noether symmetry
for the Lagrangian (\ref{CLN.150a}) if
\begin{equation}
V\left( \theta \right) =\frac{V_{0}}{\cosh ^{2}\left( \theta +\theta
_{0}\right) }~\text{or }V\left( \theta \right) =0  \label{CLN.159}
\end{equation}%
with Noether integral
\begin{equation*}
I_{X^{56}}=\frac{N_{0}^{2}r^{2}}{\cosh ^{4}\left( \theta +\theta _{0}\right)
}\left( r\dot{\theta}\cosh \left( \theta +\theta _{0}\right) -\dot{r}\sinh
\left( \theta +\theta _{0}\right) \right)
\end{equation*}%
Obviously case 3 is the most general and contains cases 1 and 2 (and the
trivial case)\ as special cases. Therefore in the following we look for
analytic solutions for the vector $X^{56}$only.
\end{enumerate}

We recall that $\frac{1}{\sqrt{-2F\left( \theta \right) }}=N^{2}\left(
\theta \right) $ from which follows:%
\begin{equation}
~F\left( \theta \right) =-\frac{1}{2N_{0}^{4}}\cosh ^{8}\left( \theta
+\theta _{0}\right) ~,N_{0}\in
\mathbb{R}
.  \label{CLN.161}
\end{equation}

We may consider $\theta _{0}=0$ (e.g. by introducing the new variable $%
\Theta =\theta +\theta _{0}).$

For the potential (\ref{CLN.159}) the Lagrangian (\ref{CLN.150a}) becomes
\begin{equation}
L=\frac{N_{0}^{2}}{\cosh ^{4}\theta }\left[ -\frac{1}{2}\dot{r}^{2}+\frac{1}{%
2}r^{2}\dot{\theta}^{2}\right] -r^{2}\frac{V_{0}}{\cosh ^{2}\theta }
\label{CLN.162}
\end{equation}%
and the Hamiltonian
\begin{equation}
E=\frac{N_{0}^{2}}{\cosh ^{4}\theta }\left[ -\frac{1}{2}\dot{r}^{2}+\frac{1}{%
2}r^{2}\dot{\theta}^{2}\right] +r^{2}\frac{V_{0}}{\cosh ^{2}\theta }.
\label{CLN.163}
\end{equation}

The equations of motion are:%
\begin{eqnarray}
\ddot{r}+r\dot{\theta}^{2}-4\tanh \theta ~\dot{r}\dot{\theta}-2\frac{V_{0}}{%
N_{0}^{2}}r\cosh ^{2}\theta &=&0  \label{CLN.164} \\
\ddot{\theta}-2\tanh \theta ~\left( \frac{1}{r^{2}}\dot{r}^{2}+\dot{\theta}%
^{2}\right) +\frac{2}{r}\dot{r}\dot{\theta}-2\frac{V_{0}}{N_{0}^{2}}\cosh
\theta \sinh \theta &=&0  \label{CLN.165}
\end{eqnarray}%
and the Noether integral $I_{X^{56}}$ for $\theta _{0}=0$ becomes:
\begin{equation}
I_{X^{56}}=\frac{N_{0}^{2}r^{2}}{\cosh ^{4}\left( \theta +\theta _{0}\right)
}\left( r\dot{\theta}\cosh \left( \theta \right) -\dot{r}\sinh \left( \theta
\right) \right) .  \label{CLN.166}
\end{equation}

In order to proceed with the solution of the system of equations (\ref%
{CLN.164}), (\ref{CLN.165}) we change to the coordinates $x,y~$which we
define by the relations%
\begin{equation}
r=\frac{x}{\sqrt{1-x^{2}y^{2}}}~,~\theta =\arctan h\left( xy\right) .
\label{CLN.168}
\end{equation}%
In the coordinates $x,y$ the Lagrangian and the Hamiltonian become:%
\begin{equation}
L=\frac{N_{0}^{2}}{2}\left( -\dot{x}^{2}+x^{4}\dot{y}\right) -V_{0}x^{2}
\label{CLN.169}
\end{equation}%
\begin{equation}
E=\frac{N_{0}^{2}}{2}\left( -\dot{x}^{2}+x^{4}\dot{y}^{2}\right) +V_{0}x^{2}
\label{CLN.170}
\end{equation}%
and the Noether integral (we write $I$ for $I_{X^{56}})$%
\begin{equation}
I=x^{4}\dot{y}.  \label{CLN.171}
\end{equation}%
Let us assume that $I\neq 0.$In the new variables the Euler - Lagrange
equations read:%
\begin{eqnarray}
\ddot{x}+2x^{3}\dot{y}^{2}-2V_{0}x &=&0  \label{CLN.172} \\
\ddot{y}+\frac{4}{x}\dot{x}\dot{y} &=&0.  \label{CLN.173}
\end{eqnarray}%
\qquad From the Noether integral we have $\dot{y}=\frac{I}{x^{4}}~$which
upon substitution in the field equations gives the equation%
\begin{eqnarray}
\ddot{x}-2V_{0}x+\frac{2I^{2}}{x^{5}} &=&0  \label{CLN.175} \\
\frac{1}{2}\left( -\dot{x}^{2}+\frac{I^{2}}{x^{4}}\right) +V_{0}x^{2} &=&E.
\label{CLN.176}
\end{eqnarray}%
from which we compute%
\begin{equation}
\dot{x}=\sqrt{\frac{I^{2}}{x^{4}}+2V_{0}x^{2}-2E}  \label{CLN.179}
\end{equation}%
and finally the analytic solution
\begin{equation}
\int \frac{dx}{\sqrt{\frac{I^{2}}{x^{4}}+2V_{0}x^{2}-2E}}=t-t_{0}.
\label{CLN.180}
\end{equation}%
From the Noether integral we find%
\begin{equation}
y\left( t\right) -y_{0}=\int \frac{I}{x^{4}}dt.  \label{CLN.181}
\end{equation}

If $I=0$ then the analytic solution is%
\begin{equation}
x=x_{0}\sinh \left( \sqrt{2V_{0}}t+x_{1}\right) ,~y=y_{0}
\end{equation}%
with Hamiltonian constrain $E=-x_{0}^{2}V_{0}.$

If \thinspace $V_{0}=0~$(i.e. free particle) and $I=0$ the analytic solution
is
\begin{equation}
x=x_{0}t+x_{1}~,~y=y_{0}
\end{equation}%
with Hamiltonian constrain \ $E=-\frac{1}{2}x_{0}^{2}.$

\section{Noether point symmetries of a minimally coupled Scalar field.}

\label{MCsfcosm}

In this section we study the Noether point symmetries of a minimally coupled
scalar field in a spatially flat FRW spacetime. The action of the field
equations is%
\begin{equation}
S_{M}=\int d\tau dx^{3}\sqrt{-g}\left[ R+\frac{1}{2}g_{ij}\phi ^{;i}\phi
^{;j}-V\left( \phi \right) \right] +\int L_{m}d\tau dx^{3}.
\end{equation}%
where $L_{m}$ is the Lagrangian of the dust matter fluid. For a spatially
flat FRW\ spacetime the Ricciscalar is
\begin{equation*}
R=6\left( \frac{\ddot{a}}{a}+\frac{\dot{a}^{2}}{a^{2}}\right)
\end{equation*}%
hence, the Lagrangian of the field equation is (\ref{CLN.14}) and the field
equations are
\begin{equation}
E=-3a\dot{a}^{2}+\frac{\varepsilon }{2}a^{2}\dot{\phi}^{2}+a^{3}V\left( \phi
\right)  \label{SF.60e}
\end{equation}%
\begin{equation}
\ddot{a}+\frac{1}{2a}\dot{a}^{2}+\frac{\varepsilon }{4}\dot{\phi}^{2}-\frac{1%
}{2}aV=0
\end{equation}%
\begin{equation}
\ddot{\phi}+\frac{3}{a}\dot{a}\dot{\phi}+\varepsilon V_{,\phi }=0
\end{equation}

From the kinetic term of (\ref{CLN.14}) we define the two dimensional metric%
\begin{equation}
ds^{2}=-6a~da^{2}+\varepsilon a^{3}d\phi ^{2}.  \label{SF.52a}
\end{equation}

We find that the curvature of the $\{a,\phi \}$ space is $R=0$ implying
flatness (since all $2$ dimensional spaces are Einstein spaces hence $\hat{R}%
=0$ implies that the space is flat). Also, the signature of the metric eq.(%
\ref{SF.52a}) is $0$, hence the space is the 2-d Minkowski space. In order
to simplify the field equations we apply the following coordinate
transformation
\begin{equation}
r=\sqrt{\frac{8}{3}}a^{3/2}\;\;\;\;\;\theta =\sqrt{\frac{3k\epsilon }{8}}%
\phi \;,  \label{tran1A}
\end{equation}%
in the new coordinates the two dimensional metric (\ref{SF.52a}) is given by
\begin{equation}
d{\hat{s}}^{2}=-dr^{2}+r^{2}d\theta ^{2}  \label{SF.56}
\end{equation}%
that is, $(r,\theta )$ are hyperbolic spherical coordinates in the two
dimensional Minkowski space $\{a,\phi \}$. Next we introduce the new
coordinates $(x,y)$ with the transformation:%
\begin{align}
x& =r\cosh \;(\theta )  \notag  \label{trans} \\
y& =r\sinh \;(\theta )
\end{align}%
which implies that the metric (\ref{SF.56}) becomes $d\hat{s}%
^{2}=-dx^{2}+dy^{2}$. We also point here that
\begin{equation}
r^{2}=x^{2}-y^{2}\;\;\;\;\;\theta =\mathrm{arctanh}(y/x)\;.  \label{tran1}
\end{equation}%
The scale factor ($a(t)>0$) is now given by:
\begin{equation}
a=\left[ \frac{3(x^{2}-y^{2})}{8}\right] ^{1/3}  \label{alcon}
\end{equation}%
which means that the new variables have to satisfy the following inequality:
$x\geq |y|$.

In the new coordinate system $(x,y)$ the Lagrangian (\ref{CLN.14}) and the
Hamiltonian (\ref{SF.60e}) are written:
\begin{equation}
L=\frac{1}{2}\left( \dot{y}^{2}-\dot{x}^{2}\right) -V_{eff}(x,y)
\label{SF.60}
\end{equation}%
\begin{equation}
E=\frac{1}{2}\left( \dot{y}^{2}-\dot{x}^{2}\right) +V_{eff}(x,y)
\label{SF.60a}
\end{equation}%
where
\begin{equation}
V_{eff}(x,y)=\left( x^{2}-y^{2}\right) \tilde{V}\left( \frac{y}{x}\right) .
\label{SF.60aa}
\end{equation}%
Note that we have used
\begin{equation}
\tilde{V}\left( \theta \right) =\frac{3k}{8}V(\theta )\;.  \label{modpot}
\end{equation}

We now proceed in an attempt to provide the Noether point symmetries of the
current dynamical problem using the results of chapter \ref{chapter3}.

Since the Lagrangian (\ref{SF.60}) is autonomous admits the Noether point
symmetry $\partial _{t}$ with Noether integral the Hamiltonian (\ref{SF.60a}%
). Lagrangian (\ref{SF.60}) admits extra Noether point symmetries in the
following cases.

\subsection{Hyperbolic - UDM\ Potential}

\textbf{Hyperbolic - UDM Potential:} Generically, we use the following
potential:
\begin{equation}
\tilde{V}(\theta )=\frac{\omega _{1}\cosh ^{2}\left( \theta \right) -\omega
_{2}\sinh ^{2}\left( \theta \right) }{2}  \label{hype1}
\end{equation}%
or
\begin{equation}
V_{eff}(x,y)=r^{2}\tilde{V}(\theta )=\frac{\omega _{1}x^{2}-\omega _{2}y^{2}%
}{2}  \label{hype2}
\end{equation}%
The corresponding Noether symmetries, $X_{n}$, are known (see for example
\cite{Leach80a}). These are:
\begin{align*}
X_{n_{1}}& =\partial _{t}~,~X_{n_{2}}=\sinh \left( \sqrt{\omega _{1}}%
t\right) \partial _{x}~,~X_{n_{3}}=\cosh \left( \sqrt{\omega _{1}}t\right)
\partial _{x} \\
X_{n_{4}}& =\sinh \left( \sqrt{\omega _{2}}t\right) \partial
_{y}~,~X_{n_{5}}=\cosh \left( \sqrt{\omega _{2}}t\right) \partial _{y}
\end{align*}%
The Noether integrals are the Hamiltonian and the quantities:%
\begin{align*}
I_{n_{2}}& =\sinh \left( \sqrt{\omega _{1}}t\right) \dot{x}-\sqrt{\omega _{1}%
}\cosh \left( \sqrt{\omega _{1}}t\right) x \\
I_{n_{3}}& =\cosh \left( \sqrt{\omega _{1}}t\right) \dot{x}-\sqrt{\omega _{1}%
}\sinh \left( \sqrt{\omega _{1}}t\right) x \\
I_{n_{4}}& =\sinh \left( \sqrt{\omega _{2}}t\right) \dot{y}-\sqrt{\omega _{2}%
}\cosh \left( \sqrt{\omega _{2}}t\right) y \\
I_{n_{5}}& =\cosh \left( \sqrt{\omega _{2}}t\right) \dot{y}-\sqrt{\omega _{2}%
}\sinh \left( \sqrt{\omega _{2}}t\right) y
\end{align*}%
Obviously the UDM potential is a particular case of the current general
hyperbolic potential. Indeed for $\omega _{1}=2\omega _{2}$ and with the aid
of eqs.(\ref{tran1A}), (\ref{modpot}) we fully recover the UDM potential
\cite{Bertacca07,Gorini05,BasilL08,CapP09}
\begin{equation}
V(\phi )=V_{0}\left[ 1+\cosh ^{2}\left( \frac{3k\epsilon }{8}\phi \right) %
\right]  \label{pott1}
\end{equation}%
where $V_{0}=\frac{4\omega _{2}}{3k}$ modulus a constant.

\subsection{Exponential Potential}

\textbf{Exponential Potential:} The exponential potential
\begin{equation*}
V_{eff}(r,\theta )=r^{2}\tilde{V}(\theta )=r^{2}e^{-d\theta }\;.
\end{equation*}%
admits the extra Noether symmetry%
\begin{equation}
X_{n}=2t\partial _{t}+\left( x+\frac{4}{d}y\right) \partial _{x}+\left( y+%
\frac{4}{d}x\right) \partial _{y}\;.  \label{SF.64}
\end{equation}%
In general the Noether integral for the vector $X_{n}=2t\partial _{t}+\eta
^{i}\partial _{i}$ is
\begin{equation}
I=2tE+\left( x+\frac{4}{d}y\right) \dot{x}-\left( y+\frac{4}{d}x\right) \dot{%
y}  \label{SF.68}
\end{equation}%
where $E$ is the Hamiltonian. Using ${\tilde{V}}=e^{-d\theta }$ together
with eq.(\ref{tran1A}) and eq.(\ref{modpot}) we write the potential to its
nominal form \cite{Sievers03} which is
\begin{equation}
V(\phi )=V_{0}\mathrm{exp}\left( -d\;\sqrt{\frac{3k\epsilon }{8}}\phi \right)
\label{pott}
\end{equation}%
where $V_{0}=\frac{8}{3k}$ modulus a constant\footnote{%
In the special case of $d=2$, the system admits an additional Lie symmetry $%
\partial _{x}+\partial _{y}$, with Noether integral $I=\dot{x}-\dot{y}.$}.

We note that the analytic solutions of sections \ref{AnNMSFA} and \ref%
{AnNMSFB} are also solutions for the minimally coupled Lagrangian (\ref%
{CLN.14}) if and only if the space does not admit dust.

\section{The Lie and Noether symmetries of Bianchi class A homogeneous
cosmologies with a scalar field.}

\label{MCsfcosmBianchi}

The class of Bianchi spatially homogeneous cosmologies contains many
important cosmological models, including the standard FRW\ model. In these
models the spacetime manifold is foliated along the time axis, with three
dimensional homogeneous hypersurfaces. Bianchi has classified all three
dimensional real Lie algebras and has shown that there are nine of them.
This results in nine types (two of them being families of spacetimes) of
Bianchi spatially homogeneous spacetimes. The principal advantage of Bianchi
cosmological models is that, in these models the physical variables depend
on the time only, reducing the Einstein and other governing equations to
ordinary differential equations.

The Bianchi models are studied in the well known ADM decomposition (\cite%
{Rayanbook,Misner69})\ according to which the metric is written%
\begin{equation}
ds^{2}=-N^{2}(t)dt^{2}+g_{\mu \nu }\omega ^{\mu }\otimes \omega ^{\nu }
\label{BCA.1}
\end{equation}%
where $N(t)$ is the lapse function and $\{\omega ^{a}\}$ is the canonical
basis of 1-forms which satisfy the Lie algebra%
\begin{equation}
d\omega ^{i}=C_{jk}^{i}\omega ^{j}\wedge \omega ^{k}.
\end{equation}%
$C_{jk}^{i}$ are the structure constants of the algebra. The spatial metric $%
g_{\mu \nu }$ splits so that
\begin{equation}
g_{\mu \nu }=\exp (2\lambda )\exp (-2\beta )_{\mu \nu }
\end{equation}%
where $\exp (2\lambda )$ is the scale factor of the universe and $\beta
_{\mu \nu }$ is a $3\times 3$ symmetric, traceless matrix, which can be
written in a diagonal form with two independent quantities, the so called
anisotropy parameters $\beta _{+},\beta _{-}$, as follows:%
\begin{equation}
\beta _{\mu \nu }=diag\left( \beta _{+},-\frac{1}{2}\beta _{+}+\frac{\sqrt{3}%
}{2}\beta _{-},-\frac{1}{2}\beta _{+}-\frac{\sqrt{3}}{2}\beta _{-}\right) .
\end{equation}

The Bianchi models are grouped in classes A and B by means of a vector $%
a^{\mu }$ and a symmetric 3$\times 3$ metric $n_{\mu \nu }$ which are
constrainted by the condition $n_{\mu \nu }a^{\nu }=0.$ Class A is defined
by $a^{\mu }=0$ and Class B by $a^{\mu }\neq 0.$ Each Class is divided into
several types according to the rank and (the modulus of the) signature of $%
n^{\mu \nu }.$ Because of the difficulty in formulating the class B Bianchi
models in the ADM\ formalism, it is usually the case that one confines
attention to the class A\ models. Furthermore it is well known that for the
class A models there is a Lagrangian \cite{MacCallumn1979} whereas for the
class B models, to the author's knowledge, no such Lagrangian seems to
exist. Details on the structure and the Physics of the Bianchi models can be
found e.g. in \cite{Rayanbook,MacCallumn1979}.

Research in Physics on inflationary models has shown the importance of
scalar fields in cosmology \cite{MadsenColes}. This has raised interest in
the dynamics of Bianchi spacetimes filled with a scalar field, with an
arbitrary self interaction potential, minimally coupled to the gravitational
field \cite{Dem92}. The Lagrangian leading to the full Bianchi scalar field
dynamics is
\begin{equation}
L=e^{3\lambda }\left[ R^{\ast }+6\lambda -\frac{3}{2}(\dot{\beta}_{1}^{2}+%
\dot{\beta}_{2}^{2})-\dot{\phi}^{2}+V(\phi )\right]  \label{BCA.5}
\end{equation}%
where $R^{\ast }$ is the Ricci scalar of the $3$ dimensional spatial
hypersurfaces given by the expression:%
\begin{align*}
R^{\ast }& =-\frac{1}{2}e^{-2\lambda }\left[ N_{1}^{2}e^{4\beta
_{1}}+e^{-2\beta _{1}}\left( N_{2}e^{\sqrt{3}\beta _{2}}-N_{3}e^{-\sqrt{3}%
\beta _{2}}\right) ^{2}-2N_{1}e^{\beta _{1}}\left( N_{2}e^{\sqrt{3}\beta
_{2}}-N_{3}e^{-\sqrt{3}\beta _{2}}\right) \right] \\
& +\frac{1}{2}N_{1}N_{2}N_{3}(1+N_{1}N_{2}N_{3}).
\end{align*}%
The constants $N_{1},N_{2},$ and $N_{3}$ are the components of the
classification vector $n^{\mu }$ and $\beta _{1}=-\frac{1}{2}\beta _{+}+%
\frac{\sqrt{3}}{2}\beta _{-},$ $\beta _{2}=-\frac{1}{2}\beta _{+}-\frac{%
\sqrt{3}}{2}\beta _{-}$. It is important to note that the curvature scalar $%
R^{\ast }$ does not depend on the derivatives of the anisotropy parameters $%
\beta _{+},$ $\beta _{-}$ , equivalently on $\beta _{1},\beta _{2}.$

The Euler Lagrange equations due to the Lagrangian (\ref{BCA.5}) are \cite%
{KotsakisL}:%
\begin{align*}
\ddot{\lambda}+\frac{3}{2}\dot{\lambda}^{2}+\frac{3}{8}(\dot{\beta}_{1}^{2}+%
\dot{\beta}_{2}^{2})+\frac{1}{4}\dot{\phi}^{2}-\frac{1}{12}e^{-3\lambda }%
\frac{\partial }{\partial \lambda }\left( e^{3\lambda }R^{\ast }\right) -%
\frac{1}{2}V(\phi )& =0 \\
\ddot{\beta}_{1}+3\dot{\lambda}\dot{\beta}_{1}+\frac{1}{3}\frac{\partial
R^{\ast }}{\partial \beta _{1}}& =0 \\
\ddot{\beta}_{2}+3\dot{\lambda}\dot{\beta}_{2}+\frac{1}{3}\frac{\partial
R^{\ast }}{\partial \beta _{2}}& =0 \\
\ddot{\phi}+3\dot{\phi}\dot{\lambda}+\frac{\partial V}{\partial \phi }& =0
\end{align*}%
where a dot over a symbol indicates derivative with respect to $t.$

In the following we apply Theorem \ref{The general conservative system} and
Theorem \ref{The Noether symmetries of a conservative system} and compute
the Lie and the Noether point symmetries of class A Bianchi models. The Lie
and the Noether point symmetries of Bianchi class A models with a scalar
field have also been computed in \cite{KotsakisL,Vak07,Cap97M}. However, as
it will be shown, these studies are not complete, in the sense that they
have not found all Noether symmetries. Furthermore our approach is entirely
different than the classical Lie approach employed in these works. Finally
it is general and can be applied without difficulty to Class B spacetimes.

We consider the four dimensional Riemannian space with coordinates $%
x^{i}=\left( \lambda ,\beta _{1},\beta _{2},\phi \right) $ and metric%
\begin{equation}
ds^{2}=e^{3\lambda }\left( 12d\lambda ^{2}-3d\beta _{1}^{2}-3d\beta
_{2}^{2}-2d\phi ^{2}\right) .  \label{BCA.10}
\end{equation}%
The metric (\ref{BCA.10}) is the conformally flat FRW spacetime whose
special projective algebra consists of the non gradient KVs \cite%
{MaartensAC,MM86}%
\begin{align*}
Y^{1}& =\partial _{\beta _{1}},~Y^{2}=\partial _{\beta _{2}},~Y^{3}=\partial
_{\phi },~Y^{4}=\beta _{2}\partial _{\beta _{1}}-\beta _{1}\partial _{\beta
_{2}} \\
~Y^{5}& =\phi \partial _{\beta _{1}}-\frac{3}{2}\beta _{1}\partial _{\phi
},~Y^{6}=\phi \partial _{\beta _{2}}-\frac{3}{2}\beta _{2}\partial _{\phi }
\end{align*}%
and the gradient HV
\begin{equation*}
H^{i}=\frac{2}{3}\partial _{\lambda }~,~\psi =1.
\end{equation*}

The Lagrangian (\ref{BCA.5}) is written:%
\begin{equation*}
L=T-U(x^{i})
\end{equation*}%
where $T=\frac{1}{2}g_{ij}\dot{x}^{i}\dot{x}^{i}$ is the geodesic
Lagrangian, $U(x^{i})$ is the potential function%
\begin{equation}
U(x^{i})=-e^{3\lambda }\left( V\left( \phi \right) +R^{\ast }\right)
\label{BCA.11}
\end{equation}%
and we have used the fact that the curvature scalar does not depend on the
derivatives of the coordinates $\beta _{1},\beta _{2}.$ Now we apply Theorem %
\ref{The general conservative system} and Theorem \ref{The Noether
symmetries of a conservative system} to determine the Lie and the point
Noether symmetries of the dynamical system with Lagrangian (\ref{BCA.5}).

In order to compute the potential $U(x^{i})$ we need the expression of $%
R^{\ast}$ for each Bianchi type. We find for the Class A models

Bianchi I: $R^{\ast}=0$

Bianchi II: $R^{\ast}=-e^{\left( 2\beta_{1}-\lambda\right) }$

Bianchi VI$_{0}$: Class A. , $R^{\ast}=-\frac{1}{2}e^{-2\lambda}\left(
e^{4\beta_{1}}+e^{-2\left( \beta_{1}-\sqrt{3}\beta_{2}\right)
}+2e^{\beta_{1}+\sqrt{3}\beta_{2}}\right) $

Bianchi VII$_{0}$: Class A., $R^{\ast}=-\frac{1}{2}e^{-2\lambda}\left(
e^{4\beta_{1}}+e^{-2\left( \beta_{1}-\sqrt{3}\beta_{2}\right)
}-2e^{\beta_{1}+\sqrt{3}\beta_{2}}\right) $

Bianchi VIII: $R^{\ast}=-\frac{1}{2}e^{-2\lambda}\left(
e^{4\beta_{1}}+e^{-2\beta_{1}}\left( e^{\sqrt{3}\beta_{2}}+e^{-\sqrt{3}%
\beta_{2}}\right) ^{2}-2e^{\beta_{1}}\left( e^{\sqrt{3}\beta_{2}}-e^{-\sqrt{3%
}\beta_{2}}x\right) \right) $

Bianchi IX: \ $R^{\ast}=-\frac{1}{2}e^{-2\lambda}\left(
e^{4\beta_{1}}+e^{-2\beta_{1}}\left( e^{\sqrt{3}\beta_{2}}+e^{-\sqrt{3}%
\beta_{2}}\right) ^{2}-2e^{\beta_{1}}\left( e^{\sqrt{3}\beta_{2}}-e^{-\sqrt{3%
}\beta_{2}}x\right) \right) +1$.

We determine the Lie and the Noether point symmetries in the following cases:

Case 1. Vacuum. In this case $\phi=$constant and the metric (\ref{BCA.10})
reduces to the three dimensional FRW\ metric.

Case 2. Zero potential $V\left( \phi\right) =0,$ $\dot{\phi}\neq0$

Case 3. Constant Potential $V\left( \phi\right) =$constant, $\dot{\phi}\neq0$

Case 4. Arbitrary Potential $V\left( \phi\right) ,~\dot{\phi}\neq0.$

\subsection{Bianchi I}

\textbf{Case 1. }

In this case $\dot{\phi}=0$, $V\left( \phi \right) =0$ and the Lagrangian
becomes $~L=e^{3\lambda }\left[ 6\dot{\lambda}^{2}-\frac{3}{2}\left( \dot{%
\beta}_{1}^{2}+\dot{\beta}_{2}^{2}\right) \right] ~$hence the potential $%
U(x^{\mu })=0$ where $x^{\mu }=(\lambda ,\beta _{1},\beta _{2}).$ The
auxiliary metric is $ds^{2}=e^{3\lambda }\left( 12d\lambda ^{2}-3d\beta
_{1}^{2}-3d\beta _{2}^{2}\right) .$ The special PCs of this metric are the
non gradient KVs~$Y^{1},~Y^{2},~Y^{4}~$and the gradient HV~$H^{i}$ . \newline
From Theorem \ref{The general conservative system} we find that the Lie
point symmetries are the vectors
\begin{equation*}
\partial _{t},~t\partial _{t},~Y^{1},~Y^{2},~Y^{4},H^{i},~t^{2}\partial
_{t}+tH^{i}
\end{equation*}%
which coincide with those found in \cite{KotsakisL}. From Theorem \ref{The
Noether symmetries of a conservative system} we find that the Noether point
\ symmetries are
\begin{equation*}
\partial _{t},~Y^{1},~Y^{2},~Y^{4},~2t\partial _{t}+H^{i},t^{2}\partial
_{t}+tH^{i}
\end{equation*}%
i.e. we find two more Noether symmetries than \cite{KotsakisL}.

\textbf{Case 2}. $V\left( \phi\right) =0\,,$ $\dot{\phi}\neq0$

In this case the Lagrangian is~$~L=e^{3\lambda }\left[ 6\dot{\lambda}^{2}-%
\frac{3}{2}\left( \dot{\beta}_{1}^{2}+\dot{\beta}_{2}^{2}\right) -\dot{\phi}%
^{2}\right] ~$and the potential function $U(x^{i})=0.$ The auxiliary metric
is (\ref{BCA.10}).~From Theorem \ref{The general conservative system} we
find that the Lie point \ symmetries are
\begin{equation*}
\partial _{t},~~t\partial
_{t},~Y^{1},~Y^{2},~Y^{3},~Y^{4},~Y^{5},~Y^{6},~H^{i},~t^{2}\partial
_{t}+tH^{i}~
\end{equation*}%
and coincide with those found in \cite{KotsakisL}. Application of Theorem %
\ref{The Noether symmetries of a conservative system} gives that the Noether
point \ symmetries are

\begin{equation*}
\partial _{t},~Y^{1},~Y^{2},~Y^{3},~Y^{4},~Y^{5},~Y^{6},~2t\partial
_{t}+H^{i},t^{2}\partial _{t}+tH^{i}~
\end{equation*}%
i.e. two more from the ones found in \cite{KotsakisL}.

\bigskip\textbf{Case 3.} $V\left( \phi\right) =C_{0},~\dot{\phi}\neq0$

The Lagrangian is$~L=e^{3\lambda }\left[ 6\dot{\lambda}^{2}-\frac{3}{2}%
\left( \dot{\beta}_{1}^{2}+\dot{\beta}_{2}^{2}\right) -\dot{\phi}^{2}+C_{0}%
\right] ~$hence the potential~$U(x^{i})=-C_{0}e^{3\lambda }.~$Using Theorem %
\ref{The general conservative system} we find that the Lie point \
symmetries are%
\begin{equation*}
\partial _{t},~Y^{1},~Y^{2},~Y^{3},~Y^{4},~Y^{5},~Y^{6},~H^{i},~\frac{1}{C}%
e^{\pm Ct}\partial _{t}\pm e^{\pm Ct}H^{i}
\end{equation*}%
where $C=\frac{\sqrt{6C_{0}}}{2}$ , and coincide with those found in \cite%
{KotsakisL}. Application of Theorem \ref{The Noether symmetries of a
conservative system} gives the Noether point \ symmetries
\begin{equation*}
\partial _{t},~Y^{1},~Y^{2},~Y^{3},~Y^{4},~Y^{5},~Y^{6},~\frac{1}{C}e^{\pm
Ct}\partial _{t}\pm e^{\pm Ct}H^{i}.~
\end{equation*}%
\ Again we find two more Noether symmetries than \cite{KotsakisL}.\newline

\textbf{Case 4.} $V\left( \phi\right) =$arbitrary $\dot{\phi}\neq0$

In this case the Lagrangian is~$L=e^{3\lambda }\left[ 6\dot{\lambda}^{2}-%
\frac{3}{2}\left( \dot{\beta}_{1}^{2}+\dot{\beta}_{2}^{2}\right) -\dot{\phi}%
^{2}+V\left( \phi \right) \right] ~$and the potential~$U(x^{i})=-e^{3\lambda
}V\left( \phi \right) $. Application of Theorem \ref{The general
conservative system} gives the Lie point \ symmetries~$~\partial
_{t}~,~Y^{1}~,~Y^{2}~,~Y^{4}~~,~H^{i}~\ $and application of Theorem \ref{The
Noether symmetries of a conservative system} gives the Noether symmetries~$%
\partial _{t}~,~Y^{1}~,~Y^{2}~,~Y^{4}.$

Working in a similar manner we compute the Lie and the Noether point \
symmetries of all Bianchi class A homogenous spacetimes. The results of the
calculations are collected in the following Tables.

\begin{table}[tbp] \centering%
\caption{Lie and Noether Symmetries of Bianchi I scalar field}%
\begin{tabular}{lll}
\hline\hline
\textbf{Bianchi I} & Noether Symmetries & Lie Symmetries \\ \hline
Case 1 & $\partial _{t},~Y^{1},~Y^{2},~Y^{4}$ & $\partial _{t},~t\partial
_{t},~Y^{1},~Y^{2},~Y^{4},~H^{i}$ \\
& $2t\partial _{t}+H^{i},~t^{2}\partial _{t}+tH^{i}$ & $t^{2}\partial
_{t}+tH^{i}$ \\
Case 2 & $\partial _{t},~Y^{1},~Y^{2},~Y^{3},~Y^{4},~Y^{5},~Y^{6}$ & $%
\partial _{t},~t\partial _{t},~Y^{1},~Y^{2},~Y^{3},~Y^{4},Y^{5},~Y^{6}~$ \\
& $~2t\partial _{t}+H^{i},~t^{2}\partial _{t}+tH^{i}$ & $H^{i},~t^{2}%
\partial _{t}+tH^{i}$ \\
Case 3 & $\partial _{t},~Y^{1},~Y^{2},~Y^{3},~Y^{4},~Y~,~Y^{6}$ & $\partial
_{t},~Y^{1},~Y^{2},~Y^{3},~Y^{4},~Y^{5},~Y^{6},~H^{i}$ \\
& $\frac{1}{C}e^{\pm Ct}\partial _{t}\pm e^{\pm Ct}H^{i}$ & $\frac{1}{C}%
e^{\pm Ct}\partial _{t}\pm e^{\pm Ct}H^{i}$ \\
Case 4 & $\partial _{t},~Y^{1},~Y^{2},~Y^{4}$ & $\partial
_{t},~Y^{1},~Y^{2},~Y^{4}~,~H^{i}$ \\ \hline\hline
\end{tabular}%
\label{8B1}%
\end{table}%

\begin{table}[tbp] \centering%
\caption{Lie and Noether Symmetries of Bianchi II scalar field}%
\begin{tabular}{lll}
\hline\hline
\textbf{Bianchi II} & Noether Symmetries & Lie Symmetries \\ \hline
Case 1 & $\partial _{t},~Y^{2},~6t\partial _{t}+3H^{i}-5Y^{1}$ & $\partial
_{t},~Y^{2},~\frac{1}{3}t\partial _{t}+H^{i},~t\partial _{t}-Y^{1}~$ \\
Case 2 & $\partial _{t},~Y^{2},~Y^{3},~Y^{6},~6t\partial _{t}+3H^{i}-5Y^{1}$
& $\partial _{t},~Y^{2},~~Y^{3},~Y^{6},~\frac{1}{3}t\partial
_{t}+H^{i},~t\partial _{t}-Y^{1}$ \\
Case 3 & $\partial _{t},~Y^{2},~Y^{3},~Y^{6}$ & $\partial
_{t},~Y^{2},~Y^{3},~Y^{6},~3H^{i}+Y^{1}$ \\
Case 4 & $\partial _{t},~Y^{2}$ & $\partial _{t},~Y^{2},~3H^{i}+Y^{1}$ \\
\hline\hline
\end{tabular}%
\label{TableKey copy(1)}%
\end{table}%

\begin{table}[tbp] \centering%
\caption{Lie and Noether Symmetries of Bianchi VI/VII scalar field}%
\begin{tabular}{lll}
\hline\hline
\textbf{Bianchi VI}$_{0}~/$\textbf{\ VII}$_{0}$ & Noether Symmetries & Lie
Symmetries \\ \hline
Case 1 & $\partial _{t},~~6t\partial _{t}+3H^{i}-2Y^{1}-2\sqrt{3}Y^{2}$ & $%
\partial _{t},~H^{i}+\frac{1}{3}Y^{1}+\frac{\sqrt{3}}{3}Y^{2},~2t\partial
_{t}-Y^{1}-\sqrt{3}Y^{2}$ \\
Case 2 & $\partial _{t},~Y^{3},~6t\partial _{t}+3H^{i}-2Y^{1}-2\sqrt{3}Y^{2}$
& $\partial _{t},~Y^{3},~H^{i}+\frac{1}{3}Y^{1}+\frac{\sqrt{3}}{3}%
Y^{2},~2t\partial _{t}-Y^{1}-\sqrt{3}Y^{2}$ \\
Case 3 & $\partial _{t},~Y^{3}$ & $\partial _{t},~Y^{3},~H^{i}+\frac{1}{3}%
Y^{1}+\frac{\sqrt{3}}{3}Y^{2}$ \\
Case 4 & $\partial _{t}~$ & $\partial _{t},~H^{i}+\frac{1}{3}Y^{1}+\frac{%
\sqrt{3}}{3}Y^{2}$ \\ \hline\hline
\end{tabular}%
\label{TableKey copy(2)}%
\end{table}%

\begin{table}[tbp] \centering%
\caption{Lie and Noether Symmetries of Bianchi VIII/IX scalar field}%
\begin{tabular}{lll}
\hline\hline
\textbf{Bianchi VIII} & Noether Symmetries & Lie Symmetries \\ \hline
Case 1 & $\partial _{t}$ & $\partial _{t},~\frac{2}{3}t\partial _{t}~+H^{i}$
\\
Case 2 & $\partial _{t},~Y^{3}$ & $\partial _{t},~Y^{3},~\frac{2}{3}%
t\partial _{t}~+H^{i}~$ \\
Case 3 & $\partial _{t},~Y^{3}$ & $\partial _{t},~Y^{3}$ \\
Case 4 & $\partial _{t}$ & $\partial _{t}$ \\ \hline
&  &  \\ \hline\hline
\textbf{Bianchi IX} & Noether Symmetries & Lie Symmetries \\ \hline
Case 1 & $\partial _{t}$ & $\partial _{t}$ \\
Case 2 & $\partial _{t},~Y^{3}$ & $\partial _{t},~Y^{3}$ \\
Case 3 & $\partial _{t},~Y^{3}$ & $\partial _{t},~Y^{3}$ \\
Case 4 & $\partial _{t}$ & $\partial _{t}$ \\ \hline\hline
\end{tabular}%
\label{8B9}%
\end{table}%

From the above tables we infer that the Lie point \ symmetries we found
coincide with those of \cite{KotsakisL}. Some differences which appear are
due to linear combinations of symmetries from the other set. The same does
not apply to the Noether symmetries, for which we found a larger number than
in \cite{KotsakisL}.

We note that in Case 1 the Noether symmetry $2t+H$ is a combination of two
Lie point \ symmetries, which is peculiar since the Noether point \
symmetries are considered to be a direct subset of Lie symmetries. This is
explained as follows. The addition of a Killing vector to a homothetic
vector retains a homothetic vector. Therefore a Lie symmetry due to a
Killing vector and one due to a homothetic vector is possible to give a Lie
point symmetry due to a homothetic vector. Concerning \cite{Vak07} from the
examination of the Tables \ref{8B1}-\ref{8B9} and the results they present
it can be seen that they loose the Noether symmetries which have a component
along $\partial _{t}~$direction.

\subsubsection{The Bianchi I\ model with scalar field and exponential
potential}

\label{Bianchi I scalar field}

We consider a scalar field described by an exponential potential $V\left(
\phi \right) =e^{-d\phi }$ in a Bianchi class A spacetime. For this
potential all the models admit the extra Lie symmetry $t\partial _{t}+\frac{2%
}{d}Y^{3}.$ Concerning the Noether symmetries we have an extra Noether
symmetry only for the types I. II, VI$_{0},$ VII $_{0}$ as follows:

Type I%
\begin{equation*}
t\partial_{t}+\frac{1}{2}H^{i}+\frac{2}{d}Y^{3}
\end{equation*}

Type II%
\begin{equation*}
t\partial_{t}+\frac{1}{2}H^{i}-\frac{5}{6}Y^{1}+\frac{2}{d}Y^{3}
\end{equation*}

Type VI$_{0}~/$\ VII$_{0}$%
\begin{equation*}
6t\partial_{t}+3H^{i}-2Y^{1}-2\sqrt{3}Y^{2}+\frac{6}{d}Y^{3}.
\end{equation*}

In the following we concentrate on the Bianchi I\ model and make use of the
extra Noether integral to define a transformation \cite{Russo} which allows
the determination of the analytic form of the metric.

The Lagrangian describing a scalar field with exponential potential in an
empty Bianchi I spacetime is~%
\begin{equation}
L=e^{3\lambda}\left[ 6\dot{\lambda}^{2}-\frac{3}{2}\left( \dot{\beta}%
_{1}^{2}+\dot{\beta}_{2}^{2}\right) -\dot{\phi}^{2}+V_{0}e^{-d\phi}\right]
\end{equation}
and the corresponding Hamiltonian vanishes. The metric defined by this
Lagrangian is (\ref{BCA.10}).

Using the transformation $\lambda=\frac{1}{3}\ln\left( \frac{a^{3}}{2}%
\right) $ the Lagrangian becomes~%
\begin{equation*}
L=3a\dot{a}^{2}-\frac{1}{2}a^{3}\dot{\phi}-\frac{3}{4}a^{3}\left( \dot{\beta
}_{1}^{2}+\dot{\beta}_{2}^{2}\right) +\frac{1}{2}V_{0}a^{3}e^{-d\phi}.
\end{equation*}
We change variables by means of the transformation
\begin{equation*}
u=\frac{\sqrt{6}}{4}\phi+\frac{1}{2}\ln\left( a^{3}\right) ~,~v=-\frac {%
\sqrt{6}}{4}\phi+\frac{1}{2}\ln\left( a^{3}\right) ~
\end{equation*}
and the Lagrangian takes the form%
\begin{equation*}
L\left( u,v,\dot{u},\dot{v}\right) =e^{\left( u+v\right) }\left( \frac {8}{3}%
\dot{u}\dot{v}-\frac{3}{2}\left( \dot{\beta}_{1}^{2}+\dot{\beta}%
_{2}^{2}\right) +V_{0}e^{-2K\left( u-v\right) }\right) .
\end{equation*}
Next we change the time coordinate as follows
\begin{equation*}
\frac{d\tau}{dt}=\sqrt{\frac{3V_{0}}{8}}e^{-K\left( u-v\right) }.
\end{equation*}
We make one more change $\beta_{1}=\sqrt{\frac{9}{16}}B_{1}~,~\beta_{2}=%
\sqrt{\frac{9}{16}}B_{2}~$and in the coordinates $\tau,u,v,B_{1},B_{2}$ the
Lagrangian is:
\begin{equation}
L\left( \tau,u,v,B_{1},B_{2}\right) =e^{\left( u+v\right) }e^{-K\left(
u-v\right) }\left( u^{\prime}v^{\prime}-\left(
B_{1}^{\prime2}+B_{2}^{\prime2}\right) +1\right)
\end{equation}
where $u^{\prime}=\frac{du}{d\tau},$ $v^{\prime}=\frac{dv}{d\tau}%
,~B_{1,2}^{\prime}=\frac{dB_{1,2}}{d\tau}.$ The equations of motion are
\begin{align}
u^{\prime\prime}+\left( 1-K\right) u^{\prime2}+\left( 1+K\right)
B_{1}^{^{\prime}2}+\left( 1+K\right) B_{2}^{^{\prime}2}-\left( 1+K\right) &
=0 \\
v^{\prime\prime}+\left( 1+K\right) v^{\prime2}+\left( 1-K\right)
B_{1}^{^{\prime}2}+\left( 1-K\right) B_{2}^{^{\prime}2}-\left( 1-K\right) &
=0 \\
B_{1}^{^{\prime\prime}}+\left( 1-K\right) B_{1}^{\prime}u^{\prime}+\left(
1+K\right) B_{1}^{^{\prime}}v^{\prime} & =0 \\
B_{2}^{^{\prime\prime}}+\left( 1-K\right) B_{2}^{^{\prime}}u^{\prime
}+\left( 1+K\right) B_{2}^{\prime}v^{\prime} & =0
\end{align}
with constrain (the zero Hamiltonian) $u^{\prime}v^{\prime}-\left(
B_{1}^{\prime2}+B_{2}^{\prime2}\right) -1=0.~$ These expressions are
symmetric in $B_{1},B_{2}$ therefore we set $B_{1}=B_{2.}=\frac{1}{2}B$ and
the system of equations of motion becomes%
\begin{align}
u^{\prime\prime}+\left( 1-K\right) u^{\prime2}+\left( 1+K\right)
B^{^{\prime}2}-\left( 1+K\right) & =0 \\
v^{\prime\prime}+\left( 1+K\right) v^{\prime2}+\left( 1-K\right)
B^{^{\prime}2}-\left( 1-K\right) & =0 \\
B^{^{\prime\prime}}+\left( 1-K\right) B^{\prime}u^{\prime}+\left( 1+K\right)
B^{^{\prime}}v^{\prime} & =0
\end{align}
with constraint
\begin{equation}
u^{\prime}v^{\prime}-B^{\prime2}-1=0.
\end{equation}

We consider two cases $K=1$ and $K\neq1$.

For $K=1$ the metric is~$ds^{2}=2e^{2v}\left( dudv-dB^{2}\right) .~$The
potential is the gradient KV $V\left( u,v\right) =-e^{2v}$ and the solution
of the system is%
\begin{align}
u\left( \tau\right) & =\tau^{2}+2c_{1}^{2}\ln\left( \tau+c_{2}\right)
+2tc_{2} \\
v\left( \tau\right) & =\frac{1}{2}\ln\left( 2\tau+c_{2}\right) \\
B\left( \tau\right) & =c_{1}\ln\left( \tau+c_{2}\right) .
\end{align}

For $K\neq1$ we make use of the extra Noether integral $e^{\left( u+v\right)
}e^{-K\left( u-v\right) }\dot{B}=C$ and solve the system. We find that $K=%
\frac{1-C^{2}}{1+C^{2}}$ and that the system has two solutions. The first is:%
\begin{align}
u\left( \tau\right) & =\frac{1}{2\left( 1-K\right) }\ln\left( \frac{K-1}{1+K}%
\sin\left( 2\sqrt{K^{2}-1}\tau\right) \right) \\
v\left( \tau\right) & =\frac{1}{2\left( 1+K\right) }\ln\left( \sin\left( 2%
\sqrt{K^{2}-1}\tau\right) \right) \\
B\left( \tau\right) & =\frac{i}{2\sqrt{K^{2}-1}}\text{{\normalsize arctanh}}%
\left( \cos\left( 2\sqrt{K^{2}-1}\tau\right) \right)
\end{align}
and the second:%
\begin{align}
u\left( \tau\right) & =\frac{1}{2\left( 1-K\right) }\ln\left( \frac{K-1}{1+K}%
\cos\left( 2\sqrt{K^{2}-1}\tau\right) \right) \\
v\left( \tau\right) & =\frac{1}{2\left( 1+K\right) }\ln\left( \cos\left( 2%
\sqrt{K^{2}-1}\tau\right) \right) \\
B\left( \tau\right) & =\frac{i}{2\sqrt{K^{2}-1}}\text{{\normalsize arctanh}}%
\left( \frac{1}{\sin\left( 2\sqrt{K^{2}-1}\tau\right) }\right) .
\end{align}

These solutions complement the results of \cite{ChrietalScalar}

\section{Conclusion}

In this chapter we have studied conformally related metrics and Lagrangians
in the context of scalar--tensor cosmology. We have found that to every
non-minimally coupled scalar field we can associate a unique minimally
coupled scalar field in a conformally related space with an appropriate
potential. The existence of such a connection can be used in order to study
the dynamical properties of the various cosmological models, since the field
equations of a non-minimally coupled scalar field are the same, at the
conformal level, of the field equations of the minimally coupled scalar
field. The above propositions can be extended to general Riemannian spaces
in n-dimensions. Furthermore, we have identified the Noether point
symmetries and the analytic solutions of the equations of motion in the
context of a minimally coupled and a non minimally coupled scalar field in a
FRW spacetime and we have classified the Noether symmetries of the field
equations in Bianchi class A models with a minimally coupled scalar field.
We found that there is a rather large class of hyperbolic and exponential
potentials which admit extra (beyond the $\partial _{t}$) Noether pont
symmetries which lead to integrals of motions.

In general, the Noether point symmetries play an important role in physics
because they can be used to simplify a given system of differential
equations as well as to determine the integrability of the system. The
latter will provide the necessary platform in order to solve the equations
of motion analytically and thus to obtain the evolution of the physical
quantities. In cosmology, such a method is extremely relevant in order to
compare cosmographic parameters, such as scale factor, Hubble expansion
rate, deceleration parameter, density parameters with observational
constrains.

However, since the Noether point symmetries are generated from the kinetic
metric of the Lagrangian, they are not only a criterion for the
integrability of the system and a method to determine analytical solutions
of the field equations, but they are also a geometric criterion since by
demanding the existence of Noether symmetries we let the geometry to select
the dynamics, i.e. the dark energy model.

In the following chapters we study the Noether symmetries of the $f\left(
R\right) $ and the $f\left( T\right) $ theories of gravity.

\newpage%

\begin{subappendices}%

\section{Relating the ranges of the constants $F_{0}$ and $\left\vert
k\right\vert $}

\label{appendix1}

We consider the following ranges for the constants $F_{0}~$\ and $%
~\left\vert k\right\vert $

a. $F_{0}>0.~$In this case we have $\left\vert k\right\vert =\frac{1}{3}%
\sqrt{\frac{F_{0}}{F_{0}+1}}$from which follows%
\begin{equation*}
\left\vert k\right\vert ^{2}<1~,~F_{0}>0
\end{equation*}

b. $-1<F_{0}<0.$ In this case we have$\left\vert k\right\vert =\frac{1}{3}%
\sqrt{\frac{\left\vert F_{0}\right\vert }{F_{0}+1}}$from which follows%
{\LARGE .}%
\begin{equation*}
\left\vert k\right\vert ^{2}=1~,~F_{0}=-\frac{9}{10}~;~\left\vert
k\right\vert ^{2}<1,~F_{0}>-\frac{9}{10};~\left\vert k\right\vert
^{2}>1~,~F_{0}\in \left( -1,-\frac{9}{10}\right)
\end{equation*}

c. $F_{0}<1$ $.$ Then $\left\vert k\right\vert =\frac{1}{3}\sqrt{\frac{%
\left\vert F_{0}\right\vert }{\left\vert F_{0}\right\vert -1}}$from which
follows%
\begin{equation*}
\left\vert k\right\vert ^{2}=1~,~F_{0}=-\frac{9}{8}~;~\left\vert
k\right\vert ^{2}<1~,~F_{0}<-\frac{9}{8};\left\vert k\right\vert
^{2}>1~,~F_{0}\in \left( -\frac{9}{8},-1\right) .
\end{equation*}%
The ranges of $\left\vert k\right\vert $ are needed because they select
different groups of Killing vectors of the metric (\ref{CLN.30}).

\section{Computation of the gradient functions $S_{1}\left( r,\protect\theta %
\right) ,$ $S_{2}\left( r,\protect\theta \right) $}

\label{appendix2}

The functions $S_{1}\left( r,\theta \right) ,$ $S_{2}\left( r,\theta \right)
$ are the canonical coordinates $x,y$ for the KVs $K^{1},K^{2}.$ The
canonical coordinates are defined with the requirement $K^{1}=\frac{\partial
}{\partial x},K^{2}=\frac{\partial }{\partial y}$ and are computed as
follows. We have the system of differential equations:%
\begin{eqnarray*}
\frac{\partial }{\partial y} &=&\frac{e^{\left( 1-k\right) \theta }r^{k}}{%
N_{0}^{2}}\left( -\partial _{r}+\frac{1}{r}\partial _{\theta }\right) \\
\frac{\partial }{\partial x} &=&\frac{e^{-\left( 1+k\right) \theta }r^{-k}}{%
N_{0}^{2}}\left( \partial _{r}+\frac{1}{r}\partial _{\theta }\right) .
\end{eqnarray*}%
To solve it we consider the associated Lagrange system and write:%
\begin{equation*}
\frac{dy}{1}=-N_{0}^{2}\frac{dr}{r^{k}e^{\left( 1-k\right) \theta }}%
=N_{0}^{2}\frac{d\theta }{r^{k-1}e^{\left( 1-k\right) \theta }}
\end{equation*}%
The first equation gives:%
\begin{equation*}
y=-N_{0}^{2}e^{\left( -1+k\right) \theta }\int \frac{dr}{r^{k}}%
=-N_{0}^{2}e^{\left( -1+k\right) \theta }\frac{r^{1-k}}{1-k}+\Phi (\theta )
\end{equation*}%
The second equation gives:%
\begin{equation*}
y=N_{0}^{2}r^{-k+1}\int \frac{d\theta }{e^{\left( 1-k\right) \theta }}%
=N_{0}^{2}\frac{r^{1-k}}{-1+k}e^{\left( 1-k\right) \theta }+\Phi _{1}(r)
\end{equation*}%
hence we have (this is the $S_{2}\left( r,\theta \right) $):%
\begin{equation*}
y=N_{0}^{2}\frac{r^{1-k}}{k-1}e^{-\left( 1-k\right) \theta }.
\end{equation*}%
For the other coordinate we have:%
\begin{equation*}
\frac{dx}{1}=N_{0}^{2}\frac{dr}{r^{-k}e^{-\left( 1+k\right) \theta }}%
=N_{0}^{2}\frac{d\theta }{r^{-k-1}e^{-\left( 1+k\right) \theta }}
\end{equation*}%
The first equation gives:%
\begin{equation*}
x=N_{0}^{2}e^{\left( 1+k\right) \theta }\int \frac{dr}{r^{-k}}%
=N_{0}^{2}e^{\left( 1+k\right) \theta }\frac{1}{1+k}r^{1+k}+\Phi (\theta )
\end{equation*}%
and the second equation gives:%
\begin{equation*}
x=N_{0}^{2}r^{1+k}\int \frac{d\theta }{e^{-\left( 1+k\right) \theta }}=\frac{%
1}{1+k}N_{0}^{2}r^{1+k}e^{\left( 1+k\right) \theta }+\Phi _{1}(r)
\end{equation*}%
hence (this is the $S_{1}\left( r,\theta \right) $):
\begin{equation*}
x=\frac{1}{1+k}N_{0}^{2}r^{1+k}e^{\left( 1+k\right) \theta }.
\end{equation*}%
Therefore, we have the canonical coordinates%
\begin{equation*}
x=\frac{r^{1+k}e^{-\left( 1+k\right) \theta }}{k+1}~,~y=\frac{%
r^{1-k}e^{\left( -1+k\right) \theta }}{k-1}
\end{equation*}

\end{subappendices}%

\chapter{Using Noether point symmetries to specify $f(R)$ gravity\label%
{chapter9}}

\section{Introduction}

In chapter \ref{chapter8} we used the Noether point symmetries of the scalar
tensor theories in order to constrain the dark energy models. Except the
sclarar field cosmology there are other possibilities to explain the present
accelerating stage. For instance, one may consider that the dynamical
effects attributed to dark energy can be resembled by the effects of a
nonstandard gravity theory. In other words, the present accelerating stage
of the universe can be driven only by cold dark matter, under a modification
of the nature of gravity. Such a reduction of the so-called dark sector is
naturally obtained in the $f(R)$ gravity theories \cite{Sotiriou}. In the
original nonstandard gravity models, one modifies the Einstein-Hilbert
action with a general function $f(R)$ of the Ricci scalar $R$. The $f(R)$
approach is a relatively simple but still a fundamental tool used to explain
the accelerated expansion of the universe. A pioneering fundamental approach
was proposed long ago with $f(R)=R+mR^{2}$\thinspace\ \cite{Star80}. Later
on, the $f(R)$ models were further explored from different points of view in
\cite{Amendola,Amendola2,Carrol} and indeed a large number of functional
forms of $f(R)$ gravity is currently available in the literature \cite%
{Cap02,Hu07,Starobinsky007,Tsuj,AmeTsuj}.

In the following, we will use the Lie and the Noether point symmetries in
order to specify the $f\left( R\right) $ gravity in a FRW spacetime and use
the first integrals of these models to determine analytic solutions of their
field equations.

The structure of this chapter is as follows. The basic theoretical elements
of the problem are presented in section \ref{CMG}, where we also introduce
the basic FRW cosmological equations in the framework of $f(R)$ models. The
Noether point symmetries and their relevance to the $f(R)$ models are
discussed in section \ref{LieN}. In section \ref{Analsol} we provide
analytical solutions for those $f(R)$ models which are Liouville integrable
via Noether point symmetries. In section \ref{Nonf} \ we study the Noether
point symmetries in spatially non-flat $f(R)$ cosmological models. Finally,
we draw our main conclusions in section \ref{Conc}.

\section{Cosmology with a modified gravity}

\label{CMG}

Consider the modified Einstein-Hilbert action:
\begin{equation}
S=\int d^{4}x\sqrt{-g}\left[ \frac{1}{2k^{2}}f\left( R\right) +\mathcal{L}%
_{m}\right]  \label{action1}
\end{equation}%
where $\mathcal{L}_{m}$ is the Lagrangian of dust-like ($p_{m}=0$) matter
and $k^{2}=8\pi G$. Varying the action with respect to the metric\footnote{%
We use the metric i.e. the Hilbert variational approach.} we arrive at
\begin{equation}
(1+f^{\prime })G_{\nu }^{\mu }\,-\,g^{\mu \alpha }f_{_{R},\,\alpha \,;\,\nu
}+\left[ \frac{2\Box f^{\prime }-(f-Rf^{\prime })}{2}\right] \delta _{\;\nu
}^{\mu }=k^{2}\,T_{\nu }^{\mu }  \label{EE}
\end{equation}%
where the prime denotes derivative with respect to $R$, $G_{\nu }^{\mu }$ is
the Einstein tensor and $T_{\nu }^{\mu }$ is the ordinary energy-momentum
tensor of matter. Based on the matter era we treat the expanding universe as
a dust fluid which includes only cold dark matter with comoving observers $%
U^{\mu }=\delta _{0}^{\mu }$. Thus the energy momentum tensor becomes $%
T_{\mu \nu }=\rho _{m}U_{\mu }U_{\nu }$, where $\rho _{m}$ is the energy
density of the cosmic fluid.

Now, in the context of a flat FRW model the metric is
\begin{equation}
ds^{2}=-dt^{2}+a^{2}(t)(dx^{2}+dy^{2}+dz^{2}).  \label{SF.1}
\end{equation}%
The components of the Einstein tensor are computed to be:
\begin{equation}
G_{0}^{0}=-3H^{2},\;G_{b}^{a}=-\delta _{b}^{a}\left( 2\dot{H}+3H^{2}\right)
\;.  \label{EIN.1}
\end{equation}%
Inserting (\ref{EIN.1}) into the modified Einstein's field equations (\ref%
{EE}), for comoving observers, we derive the modified Friedman's equation
\begin{equation}
3f^{^{\prime }}H^{2}=k^{2}\rho _{m}+\frac{f^{^{\prime }}R-f}{2}%
-3Hf^{^{\prime \prime }}\dot{R}  \label{motion1}
\end{equation}

\begin{equation}
2f^{^{\prime }}\dot{H}+3f^{^{\prime }}H^{2}=-2Hf^{^{\prime \prime }}\dot{R}%
-\left( f^{^{\prime \prime \prime }}\dot{R}^{2}+f^{^{\prime \prime }}\ddot{R}%
\right) -\frac{f-Rf^{^{\prime }}}{2}.  \label{motion2}
\end{equation}%
The contraction of the Ricci tensor provides the Ricci scalar
\begin{equation}
R=g^{\mu \nu }R_{\mu \nu }=6\left( \frac{\ddot{a}}{a}+\frac{\dot{a}^{2}}{%
a^{2}}\right) =6(2H^{2}+\dot{H})\;.  \label{SF.3b}
\end{equation}%
The Bianchi identity $\bigtriangledown ^{\mu }\,{T}_{\mu \nu }=0$ leads to
the matter conservation law:%
\begin{equation}
\dot{\rho}_{m}+3H\rho _{m}=0\,  \label{frie3}
\end{equation}%
whose solution is
\begin{equation}
\rho _{m}=\rho _{m0}a^{-3}.
\end{equation}%
Note that the over-dot denotes derivative with respect to the cosmic time $t$
and $H\equiv \dot{a}/a$ is the Hubble parameter.

If we consider $f(R)=R$ then the field equations (\ref{EE}) boil down to the
Einstein's equations. On the other hand, the concordance $\Lambda $
cosmology is fully recovered for $f(R)=R-2\Lambda $.

From the current analysis it becomes clear that unlike the standard Friedman
equations in Einstein's GR, the modified equations of motion (\ref{motion1})
and (\ref{motion2}) are complicated and thus it is difficult to solve them
analytically.

We would like to stress here that within the context of the metric formalism
the above $f(R)$ cosmological models must obey simultaneously some strong
conditions \cite{AmeBook}. These are: (i) $f^{^{\prime }}>0$ for $R\geq
R_{0}>0$, where $R_{0}$ is the Ricci scalar at the present time. If the
final attractor is a de Sitter point we need to have $f^{^{\prime }}>0$ for $%
R\geq R_{1}>0$, where $R_{1}$ is the Ricci scalar at the de Sitter point,
(ii) $f^{^{\prime \prime }}>0$ for $R\geq R_{0}>0$, (iii) $f(R)\approx
R-2\Lambda $ for $R\gg R_{0}$ and finally (iv) $0<\frac{Rf^{^{\prime \prime
}}}{f^{^{\prime }}}(r)<1$ at $r=-\frac{Rf^{^{\prime }}}{f}=-2$

\section{Modified gravity versus symmetries}

\label{MGvsSy}

In the last decade a large number of experiments have been proposed in order
to constrain dark energy and study its evolution. Naturally, in order to
establish the evolution of the dark energy (DE) (\textquotedblright
geometrical\textquotedblright\ in the current work) equation of state
parameter a realistic form of $H(a)$ is required while the included free
parameters must be constrained through a combination of independent DE
probes (for example SNIa, BAOs, CMB etc). However, a weak point here is the
fact that the majority of the $f(R)$ models appeared in the literature are
plagued with no clear physical basis and/or many free parameters. Due to the
large number of free parameters many such models could fit the data. The
proposed additional criterion of Noether point symmetry requirement is a
physically meaning-full geometric ansatz.

According to the theory of general relativity, the space-time Killing and
homothetic symmetries via the Einstein's field equations, are also
symmetries of the energy momentum tensor. Due to the fact that the $f(R)$
models provide a natural generalization of GR one would expect that the
theories of modified gravity must inherit the symmetries of the space-time
as the usual gravity (GR) does.

Furthermore, besides the geometric symmetries we have to consider the
dynamical symmetries, which are the symmetries of the field equations (Lie
symmetries). If the field equations are derived from a Lagrangian then there
is the special class of Lie symmetries, the Noether symmetries, which lead
to conserved currents or, equivalently, to first integrals of the equations
of motion. The Noether integrals are used to reduce the order of the field
equations or even to solve them. Therefore a sound requirement, which is
possible to be made in Lagrangian theories, is that they admit extra Noether
symmetries. This assumption is model independent, because it is imposed
after the field equations have been derived, therefore it does not lead to
conflict with the geometric symmetries while, at the same time, serves the
original purpose of a selection rule. Of course, it is possible that a
different method could be assumed and select another subset of viable
models. However, symmetry has always played a dominant role in Physics and
this gives an aesthetic and a physical priority to our proposal.

In the Lagrangian context, the main field equations (\ref{motion1}) and (\ref%
{motion2}), described in section \ref{CMG}, can be produced by the following
Lagrangian:
\begin{equation}
L\left( a,\dot{a},R,\dot{R}\right) =6af^{^{\prime }}~\dot{a}%
^{2}+6a^{2}f^{^{\prime \prime }}~\dot{a}\dot{R}+a^{3}\left( f^{^{\prime
}}R-f\right) \qquad  \label{SF.50}
\end{equation}%
in the space of the variables $\{a,R\}$. Using eq.(\ref{SF.50}) we obtain
the Hamiltonian of the current dynamical system
\begin{equation}
E=6af^{^{\prime }}~\dot{a}^{2}+6a^{2}f^{^{\prime \prime }}~\dot{a}\dot{R}%
-a^{3}\left( f^{^{\prime }}R-f\right)  \label{SF.60e1}
\end{equation}%
or
\begin{equation}
E=6a^{3}\left[ f^{^{\prime }}H^{2}-\frac{1}{6f^{^{\prime }}}\left( \left(
f^{^{\prime }}R-f\right) -6\dot{R}Hf^{^{\prime \prime }}\right) \right] \;.
\label{SF.61e}
\end{equation}%
Combining the first equation of motion (\ref{motion1}) with eq.(\ref{SF.61e}%
) we find
\begin{equation}
\rho _{m}=\frac{E}{2k^{2}}\;a^{-3}\;.  \label{Smm}
\end{equation}%
The latter equation together with $\rho _{m}=\rho _{m0}a^{-3}$ implies that
\begin{equation}
\rho _{m0}=\frac{E}{2k^{2}}\Rightarrow \Omega _{m}\rho _{cr,0}=\frac{E}{%
2k^{2}}\Rightarrow E=6\Omega _{m}H_{0}^{2}
\end{equation}%
where $\Omega _{m}=\rho _{m0}/\rho _{cr,0}$, $\rho _{cr,0}=3H_{0}^{2}/k^{2}$
is the critical density at the present time and $H_{0}$ is the Hubble
constant.

We note that the current Lagrangian eq.(\ref{SF.50}) is time independent
implying that the dynamical system is autonomous hence the Hamiltonian $E$
is conserved.

\section{Noether point symmetries of $f\left( R\right) $ gravity}

\label{LieN}

The Noether condition for the Lagrangian (\ref{SF.50}) is equivalent with
the following system of eight equations
\begin{eqnarray}
\xi _{,a} &=&0~  \label{NC.01} \\
~\xi _{,R} &=&0  \label{NC.02}
\end{eqnarray}%
\begin{equation}
a^{2}f^{\prime \prime }\eta _{,R}^{\left( 1\right) }=0  \label{NC.03}
\end{equation}%
\begin{equation}
f^{\prime }\eta ^{\left( 1\right) }+af^{\prime \prime }\eta ^{\left(
2\right) }+2af^{\prime }\eta _{,a}^{\left( 1\right) }+a^{2}f^{\prime \prime
}\eta _{,a}^{\left( 2\right) }-\frac{1}{2}af^{\prime }\xi _{,t}=0
\label{NC.04}
\end{equation}%
\begin{equation}
2af^{\prime \prime }\eta ^{\left( 1\right) }+a^{2}f^{\prime \prime \prime
}\eta ^{\left( 2\right) }+a^{2}f^{\prime \prime }\eta _{,a}^{\left( 1\right)
}+2af^{\prime }\eta _{,R}^{\left( 1\right) }+a^{2}f^{\prime \prime }\eta
_{,R}^{\left( 2\right) }-\frac{1}{2}a^{2}f^{\prime \prime }\xi _{,t}=0
\label{NC.05}
\end{equation}%
\begin{equation}
-3a^{2}Rf^{\prime }\eta ^{\left( 1\right) }+3a^{2}f\eta ^{\left( 1\right)
}-a^{3}Rf^{\prime \prime }\eta ^{\left( 2\right) }+a^{3}\left( f-f^{^{\prime
}}R\right) \xi _{,t}+g_{,t}=0  \label{NC.06}
\end{equation}

\begin{equation}
12af^{\prime }\eta _{,t}^{\left( 1\right) }+6a^{2}f^{\prime \prime }\eta
_{,t}^{\left( 2\right) }+a^{3}\left( f^{^{\prime }}R-f\right) \xi
_{,a}-g_{,a}=0  \label{NC.07}
\end{equation}%
\begin{equation}
6a^{2}f^{\prime \prime }\eta _{,t}^{\left( 1\right) }+a^{3}\left(
f^{^{\prime }}R-f\right) \xi _{,R}-g_{,R}=0  \label{NC.08}
\end{equation}%
The solution of the system (\ref{NC.01})-(\ref{NC.08}) determines the
Noether symmetries.

Since the Lagrangian (\ref{SF.50}) is in the form $L=T\left( a,\dot{a},R,%
\dot{R}\right) -V\left( a,R\right) $, the results of chapter \ref{chapter3}
can be used \footnote{%
Where $T$ is the "kinetic" term and $V$ is the "potential}. The kinematic
term defines a two dimensional metric in the space of $\{a,R\}$ with line
element%
\begin{equation}
d\hat{s}^{2}=12af^{\prime }da^{2}+12a^{2}f^{\prime \prime }da~dR
\label{FR.03}
\end{equation}%
while the \textquotedblright potential\textquotedblright\ is
\begin{equation}
V(a,R)=-a^{3}(f^{^{\prime }}R-f)\;.  \label{pot}
\end{equation}

The Ricci scalar of the two dimensional metric (\ref{FR.03}) is computed to
be $\hat{R}=0,$ therefore the space is a flat space\footnote{%
All two dimensional Riemannian spaces are Einstein spaces implying that if $%
\hat{R}=const$ the space is maximally symmetric \cite{Barnes} and if $\hat{R}%
=0,$ the space admit gradient homothetic vector, i.e. is flat.} with a
maximum homothetic algebra. \ The homothetic algebra of the metric (\ref%
{FR.03}) consists of the vectors
\begin{align*}
\mathbf{K}^{1}& =a\partial _{a}-3\frac{f^{\prime }}{f^{\prime \prime }}%
\partial _{R}~,~\mathbf{K}^{2}=\frac{1}{a}\partial _{a}-\frac{1}{a^{2}}\frac{%
f^{\prime }}{f^{\prime \prime }}\partial _{R}~ \\
\mathbf{K}^{3}& =\frac{1}{a}\frac{1}{f^{\prime \prime }}\partial _{R}~,~%
\mathbf{H}=\frac{a}{2}~\partial _{a}+\frac{1}{2}\frac{f^{\prime }}{f^{\prime
\prime }}\partial _{R}
\end{align*}%
where $\mathbf{K}$ are Killing vectors ($\mathbf{K}^{2,3}$ are gradient) and
$\mathbf{H}$ is a gradient Homothetic vector.

Therefore applying theorem \ref{The Noether symmetries of a conservative
system} we have the following cases:

\textbf{Case 1:} If $f\left( R\right) $ is arbitrary the dynamical system
admits as Noether symmetry~the~$X^{1}=\partial _{t}~$with Noether integral
the Hamiltonian~$E$.

\textbf{Case 2:} If $f\left( R\right) =R^{\frac{3}{2}}~$the dynamical system
admits the extra Noether symmetries%
\begin{equation}
X^{2}=\mathbf{K}^{2},~X^{3}=t\mathbf{K}^{2}  \label{NS.01}
\end{equation}%
\begin{equation}
X^{4}=2t\partial _{t}+\mathbf{H+}\frac{5}{6}\mathbf{K}^{1}.  \label{NS.02}
\end{equation}%
with corresponding Noether Integrals
\begin{equation}
I_{2}=\frac{d}{dt}\left( a\sqrt{R}\right)  \label{NI.02}
\end{equation}%
\begin{equation}
I_{3}=t\frac{d}{dt}\left( a\sqrt{R}\right) -a\sqrt{R}  \label{NI.03}
\end{equation}%
\begin{equation}
I_{4}=2tE-6a^{2}\dot{a}\sqrt{R}-6\frac{a^{3}}{\sqrt{R}}\dot{R}.
\label{NI.04}
\end{equation}%
the non vanishing commutators of the Noether algebra being%
\begin{equation*}
\left[ X^{1},X^{3}\right] =X^{2}~~~~\left[ X^{1},X^{4}\right] =2X^{1}
\end{equation*}%
\begin{equation*}
\left[ X^{2},X^{4}\right] =\frac{8}{3}X^{2}~~~~\left[ X^{3},X^{4}\right] =%
\frac{2}{3}X^{3}
\end{equation*}

\textbf{Case 3:} If $f\left( R\right) =R^{\frac{7}{8}}$ the dynamical system
admits the extra Noether symmetries%
\begin{equation}
X^{5}=2t\partial _{t}+\mathbf{H}~,~X^{6}=t^{2}\partial _{t}+t\mathbf{H}
\label{NS.03}
\end{equation}%
with corresponding Noether Integrals
\begin{equation}
I_{5}=2tE-\frac{21}{8}\frac{d}{dt}\left( a^{3}R^{-\frac{1}{8}}\right)
\label{NI.05}
\end{equation}%
\begin{equation}
I_{6}=t^{2}E-\frac{21}{8}t\frac{d}{dt}\left( a^{3}R^{-\frac{1}{8}}\right) +%
\frac{21}{8}a^{3}R^{-\frac{1}{8}}.  \label{NI.06}
\end{equation}%
\qquad with non vanishing commutators%
\begin{equation*}
\left[ X^{1},X^{5}\right] =2X^{1}~~~~\left[ X^{1},X^{6}\right] =X^{5}~~~~%
\left[ X^{5},X^{6}\right] =2X^{6}
\end{equation*}%
From the time dependent integrals (\ref{NI.05}),(\ref{NI.06}) and the
Hamiltonian we construct the Ermakov-Lewis invariant (see chapter \ref%
{chapter4})
\begin{equation}
\Sigma =4I_{6}E-I_{5}^{2}  \label{NI.06b}
\end{equation}

\textbf{Case 4}: If $f\left( R\right) =\left( R-2\Lambda \right) ^{\frac{3}{2%
}}$ the dynamical system admits the extra Noether symmetries
\begin{equation}
\bar{X}^{2}=e^{\sqrt{m}t}\mathbf{K}^{2}~,~\bar{X}^{3}=e^{-\sqrt{m}t}\mathbf{K%
}^{2}  \label{NS.b2}
\end{equation}%
with corresponding Noether Integrals
\begin{equation}
\bar{I}_{2}=e^{\sqrt{m}t}\left( \frac{d}{dt}\left( a\sqrt{R-2L}\right) -9%
\sqrt{m}a\sqrt{R-2\Lambda }\right)  \label{NI.b2}
\end{equation}%
\begin{equation}
\bar{I}_{3}=e^{-\sqrt{m}t}\left( \frac{d}{dt}\left( a\sqrt{R-2L}\right) +9%
\sqrt{m}a\sqrt{R-2\Lambda }\right)  \label{NI.b3}
\end{equation}%
where $m=\frac{2}{3}\Lambda .$ The non vanishing commutators of the Noether
algebra are%
\begin{equation*}
\left[ X^{1},\bar{X}^{2}\right] =\sqrt{m}\bar{X}^{2}~~~~\left[ \bar{X}%
^{3},X^{1}\right] =\sqrt{m}\bar{X}^{3}~
\end{equation*}%
From the time dependent integrals (\ref{NI.b2}),(\ref{NI.b3}) we construct
the time independent integral $\bar{I}_{23}=\bar{I}_{2}\bar{I}_{3}.$

\textbf{Case 5}: If $f\left( R\right) =\left( R-2\Lambda \right) ^{\frac{7}{8%
}}$ the dynamical system admits the extra Noether symmetries%
\begin{eqnarray}
\bar{X}^{5} &=&\frac{1}{\sqrt{m}}e^{2\sqrt{m}t}\partial _{t}+e^{2\sqrt{m}t}~%
\mathbf{H}  \label{NS.b5} \\
\bar{X}^{6} &=&-\frac{1}{\sqrt{m}}e^{-2\sqrt{m}t}\partial _{t}+e^{-2\sqrt{m}%
t}~\mathbf{H}  \label{NS.b6}
\end{eqnarray}%
with corresponding Noether Integrals
\begin{equation}
\bar{I}_{5}=e^{2\sqrt{m}t}\left[ \frac{1}{\sqrt{m}}E-\frac{21}{8}\frac{d}{dt}%
\left( a^{3}\left( R-2\Lambda \right) ^{-\frac{1}{8}}\right) +\frac{21}{4}%
\sqrt{m}a^{3}\left( R-2\Lambda \right) ^{-\frac{1}{8}}\right]  \label{NI.b5}
\end{equation}%
\begin{equation}
\bar{I}_{6}=e^{-2\sqrt{m}t}\left[ \frac{1}{\sqrt{m}}E+\frac{21}{8}\frac{d}{dt%
}\left( a^{3}\left( R-2\Lambda \right) ^{-\frac{1}{8}}\right) +\frac{21}{4}%
\sqrt{m}a^{3}\left( R-2\Lambda \right) ^{-\frac{1}{8}}\right]  \label{NI.b6}
\end{equation}%
and the non vanishing commutators of the Noether algebra are%
\begin{equation*}
\left[ X^{1},\bar{X}^{5}\right] =2\sqrt{m}\bar{X}^{5}~~~~\left[ \bar{X}%
^{6},X^{1}\right] =2\sqrt{m}\bar{X}^{6}
\end{equation*}%
\begin{equation*}
\left[ \bar{X}^{5},\bar{X}^{6}\right] =\frac{4}{\sqrt{m}}X^{1}
\end{equation*}%
From the time dependent integrals (\ref{NI.05}),(\ref{NI.06}) and the
Hamiltonian we construct the Ermakov-Lewis invariant.
\begin{equation}
\phi =E^{2}-\bar{I}_{5}\bar{I}_{6}  \label{NI.b6b}
\end{equation}

\textbf{Case 6:~}If $f\left( R\right) =R^{n}$ (with $n\neq 0,1,\frac{3}{2},%
\frac{7}{8}$) the dynamical system admits the extra Noether symmetry%
\begin{equation}
X^{7}=2t\partial _{t}+\mathbf{H+}\left( \frac{4n}{3}-\frac{7}{6}\right)
\mathbf{K}^{1}  \label{NS.07}
\end{equation}%
with corresponding Noether Integral
\begin{equation}
I_{7}=2tE-8na^{2}R^{n-1}\dot{a}\left( 2-n\right) -4na^{3}R^{n-2}\dot{R}%
\left( 2n-1\right) \left( n-1\right) .  \label{NI.07}
\end{equation}%
and the commutator of the Noether algebra is $\left[ X^{1},X^{7}\right]
=2X^{1}.$

We note that the Noether subalgebra of case 2, $\left\{
X^{1},X^{2},X^{3}\right\} $ and the algebra of case 4 $\left\{ X^{1},\bar{X}%
^{2},\bar{X}^{3}\right\} $ is the same Lie algebra but in different
representation. The same observation applies to the subalgebra of case 3 $%
\left\{ X^{1},X^{5},X^{6}\right\} $ and the algebra of case 5 $\left\{ X^{1},%
\bar{X}^{5},\bar{X}^{6}\right\} $. This connection between the Lie groups is
useful because it reveals common features in the dynamic systems, as is the
common transformation to the normal coordinates of the systems.

For the cosmological viability of the models see \cite{AmeTsuj,AmeBook}

\section{Analytic Solutions}

\label{Analsol}

Using the Noether symmetries and the associated Noether integrals we solve
analytically the differential eqs.(\ref{motion1}), (\ref{motion2}) and (\ref%
{SF.3b}) for the cases where the dynamical system is Liouville integrable,
that is for cases 2-5. Case 6$~$(i.e. $f\left( R\right) =R^{n}$) is not
Liouville integrable via Noether point symmetries, since the Noether
integral (\ref{NI.07}) is time dependent\footnote{%
In the appendix \ref{AppendixAAA} we present special solutions for the $%
f\left( R\right) =R^{n}~$model, using the zero order invariants.}.

\subsection{Power law model $R^{\protect\mu }$ with $\protect\mu =\frac{3}{2}
$}

\label{subs1}

In this case the Lagrangian eq.(\ref{SF.50}) of the $f(R)=R^{\frac{3}{2}}$
model is written as
\begin{equation}
L=9a\sqrt{R}\dot{a}^{2}+\frac{9a^{2}}{2\sqrt{R}}\dot{a}\dot{R}+\frac{a^{3}}{2%
}R^{\frac{3}{2}}  \label{FR.32}
\end{equation}%
Changing the variables from $(a,R)$ to $(z,w)$ via the relations:
\begin{equation}
a=\left( \frac{9}{2}\right) ^{-\frac{1}{3}}\sqrt{z}\;\;\;\;R=\frac{w^{2}}{z}
\end{equation}%
the Lagrangian (\ref{FR.32}) and the Hamiltonian (\ref{SF.60e}) become
\begin{equation}
L=\dot{z}\dot{w}+V_{0}w^{3}
\end{equation}%
\begin{equation}
E=\dot{z}\dot{w}-V_{0}w^{3}
\end{equation}%
where $V_{0}=\frac{1}{9}.~$The equations of motion in the new coordinate
system are
\begin{eqnarray}
\ddot{w} &=&0 \\
\ddot{z}-3V_{0}w^{2} &=&0
\end{eqnarray}%
The Noether integrals (\ref{NI.02}),(\ref{NI.03}) in the coordinate system $%
\left\{ z,y\right\} $ are
\begin{equation}
I_{1}^{^{\prime }}=\dot{w}~\ ,~~I_{2}^{^{\prime }}=t\dot{w}-w
\end{equation}%
The general solution of the system is:
\begin{equation}
y\left( t\right) =I_{1}^{\prime }t-I_{2}^{^{\prime }}
\end{equation}%
\begin{equation}
z\left( t\right) =\frac{1}{36\left( I_{1}^{^{\prime }}\right) ^{2}}\left(
I_{1}^{^{\prime }}t-I_{2}^{^{\prime }}\right) ^{4}+z_{1}t+z_{0}
\end{equation}%
The Hamiltonian constraint gives~$E=z_{1}I_{1}^{^{\prime }}$ where $z_{0,1}~$%
are constants and the singularity condition results in the constraint%
\begin{equation}
\frac{1}{36\left( I_{1}^{^{\prime }}\right) ^{2}}\left( I_{2}^{^{\prime
}}\right) ^{4}+z_{0}=0.
\end{equation}

\subsection{Power law model $R^{\protect\mu }$ with $\protect\mu =\frac{7}{8}
$}

\label{subsb1}

In this case the Lagrangian eq.(\ref{SF.50}) is written as
\begin{equation}
L=\frac{21a}{4R^{\frac{1}{8}}}\dot{a}^{2}-\frac{21}{16}\frac{a^{2}}{R^{\frac{%
9}{8}}}\dot{a}\dot{R}-\frac{1}{8}a^{3}R^{\frac{7}{8}}.  \label{FR.78}
\end{equation}%
Changing now the variables from $(a,R)$ to $(\rho ,\sigma )$ via the
relations:
\begin{equation}
a=\left( \frac{21}{4}\right) ^{-\frac{1}{3}}\sqrt{\rho e^{\sigma }}\;\;\;\;R=%
\frac{e^{12\sigma }}{\rho ^{4}}.
\end{equation}%
The Lagrangian (\ref{FR.16}) and the Hamiltonian (\ref{SF.60e}) become%
\begin{equation}
L=\frac{1}{2}\dot{\rho}^{2}-\frac{1}{2}\rho ^{2}\dot{\sigma}^{2}+V_{0}\frac{%
e^{12\sigma }}{\rho ^{2}}  \label{FR.78a}
\end{equation}%
\begin{equation}
E=\frac{1}{2}\dot{\rho}^{2}-\frac{1}{2}\rho ^{2}\dot{\sigma}^{2}-V_{0}\frac{%
e^{12\sigma }}{\rho ^{2}}.  \label{FR.78b}
\end{equation}%
where$~V_{0}=-\frac{1}{42}.$ \ The Euler-Lagrange equations provide the
following equations of motion:%
\begin{align}
\ddot{\rho}+\rho \dot{\sigma}^{2}+2V_{0}\frac{e^{12\sigma }}{\rho ^{3}}& =0
\label{FR.78c} \\
\ddot{\sigma}+\frac{2}{\rho }\dot{\sigma}\dot{\rho}+12V_{0}\frac{e^{12\sigma
}}{\rho ^{2}}& =0.  \label{FR.78d}
\end{align}%
The Noether integrals (\ref{NI.05}), (\ref{NI.06}) and the Ermakov-Lewis
invariant \ref{NI.06b} in the coordinate system $\left\{ u,v\right\} $ are
\begin{align}
I_{5}^{^{\prime }}& =2tE-\rho \dot{\rho} \\
I_{6}^{^{\prime }}& =t^{2}E-t\rho \dot{\rho}+\frac{1}{2}\rho ^{2}.
\end{align}%
\begin{equation}
\Sigma =\rho ^{4}\dot{\sigma}^{2}+4V_{0}e^{12\sigma }.
\end{equation}%
Using the Ermakov-Lewis Invariant, the Hamiltonian (\ref{FR.78a}) and
equation (\ref{FR.78c}) are written:%
\begin{align}
\frac{1}{2}\dot{\rho}^{2}-\frac{1}{2}\frac{\Sigma }{\rho ^{2}}& =E \\
\ddot{\rho}+\frac{\Sigma }{\rho ^{3}}& =0.
\end{align}%
And the analytical solution of the system is
\begin{equation}
\rho \left( t\right) =\left( \rho _{2}t^{2}+\rho _{1}t+\frac{\left( \left(
\rho _{1}\right) ^{2}-4\Sigma \right) }{4\rho _{2}}\right) ^{\frac{1}{2}}
\end{equation}%
\begin{equation}
\exp \left( \sigma \left( t\right) \right) =\left\{ \frac{21}{2}\Sigma \left[
\left( \tanh \left[ \sigma _{0}\rho _{2}\sqrt{\Sigma }-6\arctan h\left(
\frac{2\rho _{2}t+\rho _{1}}{2\sqrt{\Sigma }}\right) \right] \right) ^{2}-1%
\right] \right\} ^{\frac{1}{12}}
\end{equation}%
where $B\left( t\right) =\left( \frac{1}{2}\frac{2\rho _{2}t+\rho _{1}}{%
\sqrt{\Sigma }}\right) $ and $\rho _{1,2}~$,$~\sigma _{0}$ are constants
with Hamiltonian constrain $E=\frac{1}{2}\rho _{2}~$. The singularity
constraint gives $\left( \rho _{1}\right) ^{2}=4\Sigma $

In the case $\Sigma =0$ the analytical solution is%
\begin{equation}
\rho \left( t\right) =\left( \rho _{2}t^{2}+\rho _{1}t+\frac{1}{2}\frac{%
\left( \rho _{1}\right) ^{2}}{\rho _{2}}\right) ^{\frac{1}{2}}
\end{equation}%
\begin{equation}
\exp \sigma \left( t\right) =\left[ \frac{1}{24\sqrt{V_{0}}}\frac{\left(
2\rho _{2}t+\rho _{1}\right) }{\left( 4\sigma _{0}\rho _{2}^{2}t+2\sigma
_{0}\rho _{2}\rho _{1}-1\right) }\right] ^{\frac{1}{6}}
\end{equation}%
The singularity constraint gives $\rho _{1}=0$, then the solution is%
\begin{equation}
a\left( t\right) =\frac{a_{0}t^{\frac{7}{6}}}{\left( a_{2}t-1\right) ^{\frac{%
1}{6}}}
\end{equation}

In contrast with the claim of \cite{Nayem12} this model is analytically
solvable and there exists models which admit Noether integrals with time
dependent gauge functions.

\subsection{$\Lambda _{bc}$CDM model with $(b,c)=(1,\frac{3}{2})$}

\label{subs}

Inserting $f(R)=(R-2\Lambda )^{3/2}$ into eq.(\ref{SF.50}) we obtain
\begin{equation}
L=9a\sqrt{R-2\Lambda }\dot{a}^{2}+\frac{9a^{2}}{2\sqrt{R-2\Lambda }}\dot{a}%
\dot{R}+\frac{a^{3}}{2}\left( R+4\Lambda \right) \sqrt{R-2\Lambda }
\label{FR.13a}
\end{equation}%
Changing the variables from $(a,R)$ to $(x,y)$ via the relations:
\begin{equation}
a=\left( \frac{9}{2}\right) ^{-\frac{1}{3}}\sqrt{x}\;\;\;\;R=2\Lambda +\frac{%
y^{2}}{x}
\end{equation}%
the Lagrangian (\ref{FR.13a}) and the Hamiltonian (\ref{SF.60e}) become
\begin{equation}
L=\dot{x}\dot{y}+V_{0}\left( y^{3}+\bar{m}xy\right)
\end{equation}%
\begin{equation}
E=\dot{x}\dot{y}-V_{0}\left( y^{3}+{\bar{m}}xy\right)  \label{hamm}
\end{equation}%
where $V_{0}=\frac{1}{9}$ and $\bar{m}=6\Lambda $. \newline
The equations of motion, using the Euler-Lagrange equations, in the new
coordinate system are
\begin{equation}
\ddot{x}-3V_{0}y^{2}-{\bar{m}V}_{0}x=0  \label{FR.13b}
\end{equation}%
\begin{equation}
\ddot{y}-{\bar{m}}V_{0}y=0.  \label{FR.14}
\end{equation}%
The Noether integrals (\ref{NI.b2}),(\ref{NI.b3}) in the coordinate system $%
\left\{ x,y\right\} $ are
\begin{align}
\bar{I}_{1}^{\prime }& =e^{\omega t}\dot{y}-\omega e^{\omega t}y
\label{II.14} \\
\bar{I}_{2}^{\prime }& =e^{-\omega t}\dot{y}+\omega e^{-\omega t}y.
\label{II1.14}
\end{align}%
where $\omega =\sqrt{2\Lambda /3}$. From these we construct the time
independent first integral%
\begin{equation}
\Phi =I_{1}I_{2}=\dot{y}^{2}-\omega ^{2}y^{2}.  \label{II2.14}
\end{equation}%
The constants of integration are further constrained by the condition that
at the singularity ($t=0$), the scale factor has to be exactly zero, that
is, $x(0)=0$.\newline
The general solution of the system (\ref{FR.13b})-(\ref{FR.14}) is:
\begin{equation}
y\left( t\right) =\frac{I_{2}}{2\omega }e^{\omega t}-\frac{I_{1}}{2\omega }%
e^{-\omega t}
\end{equation}%
\begin{equation}
x\left( t\right) =x_{1G}e^{\omega t}+x_{2G}e^{-\omega t}+\frac{1}{4\bar{m}%
\omega ^{2}}\left( I_{2}e^{\omega t}+I_{1}e^{-\omega t}\right) ^{2}+\frac{%
\Phi }{\bar{m}\omega ^{2}}.
\end{equation}%
The Hamiltonian constrain gives~$E=\omega \left(
x_{1G}I_{1}-x_{2G}I_{2}\right) $ where $x_{1G,2G}~$are constants and the
singularity condition results in the constrain%
\begin{equation}
x_{1G}+x_{2G}+\frac{1}{4\bar{m}\omega ^{2}}\left( I_{1}+I_{2}\right) ^{2}+%
\frac{\Phi }{\bar{m}\omega ^{2}}=0.
\end{equation}%
At late enough times the solution becomes $a^{2}(t)\propto e^{2\omega t}$

\subsection{$\Lambda _{bc}$CDM model with $(b,c)=(1,\frac{7}{8})$}

\label{subsb2}

In this case the Lagrangian eq.(\ref{SF.50}) of the $f(R)=(R-2\Lambda
)^{7/8} $ model is written as
\begin{equation}
L=\frac{21a}{4\left( R-2\Lambda \right) ^{\frac{1}{8}}}\dot{a}^{2}-\frac{21}{%
16}\frac{a^{2}}{\left( R-2\Lambda \right) ^{\frac{9}{8}}}\dot{a}\dot{R}-%
\frac{1}{8}a^{3}\frac{\left( R-16\Lambda \right) }{\left( R-2\Lambda \right)
^{\frac{1}{8}}}.  \label{FR.16}
\end{equation}%
Changing the variables from $(a,R)$ to $(u,v)$ via the relations:
\begin{equation}
a=\left( \frac{21}{4}\right) ^{-\frac{1}{3}}\sqrt{ue^{v}}\;\;\;\;R=2\Lambda +%
\frac{e^{12v}}{u^{4}}.  \label{FR.16a}
\end{equation}%
the Lagrangian (\ref{FR.16}) and the Hamiltonian (\ref{SF.60e}) become%
\begin{equation}
L=\frac{1}{2}\dot{u}^{2}-\frac{1}{2}u^{2}\dot{v}^{2}+V_{0}\frac{\bar{m}}{4}%
u^{2}+V_{0}\frac{e^{12v}}{u^{2}}  \label{FR.20}
\end{equation}%
\begin{equation}
E=\frac{1}{2}\dot{u}^{2}-\frac{1}{2}u^{2}\dot{v}^{2}-V_{0}\frac{\bar{m}}{4}%
u^{2}-V_{0}\frac{e^{12v}}{u^{2}}.  \label{FR.21}
\end{equation}%
where $\bar{m}=-28\Lambda ~,~V_{0}=-\frac{1}{42}.$ \newline
The Euler-Lagrange equations provide the following equations of motion:%
\begin{align}
\ddot{u}+u\dot{v}^{2}-\frac{V_{0}\bar{m}}{2}u+2V_{0}\frac{e^{12v}}{u^{3}}& =0
\label{FR.22} \\
\ddot{v}+\frac{2}{u}\dot{u}\dot{v}+12V_{0}\frac{e^{12v}}{u^{4}}& =0.
\label{FR.23}
\end{align}%
The Noether integrals (\ref{NI.b5}),(\ref{NI.b6}) and the Ermakov-Lewis
invariant (\ref{NI.b6b}) in the coordinate system $\left\{ u,v\right\} $ are
\begin{align}
I_{+}& =\frac{1}{\lambda }e^{2\lambda t}E-e^{2\lambda t}u\dot{u}+\lambda
e^{2\lambda t}u^{2}  \label{FR.24} \\
I_{-}& =\frac{1}{\lambda }e^{-2\lambda t}E-e^{-2\lambda t}u\dot{u}+\lambda
e^{-2\lambda t}u^{2}.  \label{FR.25}
\end{align}%
\begin{equation}
\phi =u^{4}\dot{v}^{2}+4V_{0}e^{12v}.  \label{FR.261}
\end{equation}%
where $\lambda =\frac{1}{2}\sqrt{\frac{2}{3}\Lambda }.$

Using the Ermakov-Lewis Invariant (\ref{FR.261}), the Hamiltonian (\ref%
{FR.21}) and equation (\ref{FR.22}) are written:%
\begin{align}
\frac{1}{2}\dot{u}^{2}-V_{0}\frac{m}{8}u^{2}-\frac{1}{2}\frac{\phi }{u^{2}}&
=E  \label{FR.27} \\
\ddot{u}-\frac{V_{0}m}{4}u+\frac{\phi }{u^{3}}& =0.  \label{FR.28}
\end{align}%
The solution of (\ref{FR.28}) has been given by Pinney \cite{Pinney} and it
is the following:%
\begin{equation}
u\left( t\right) =\left( u_{1}e^{2\lambda t}+u_{2}e^{-2\lambda
t}+2u_{3}\right) ^{\frac{1}{2}}  \label{FR.29a}
\end{equation}%
where~$u_{1-3}$. From the Hamiltonian constraint (\ref{FR.27}) and the
Noether Integrals (\ref{FR.24}),(\ref{FR.25}) we find%
\begin{equation*}
E=-2\lambda u_{3}~,~I_{+}=2\lambda u_{2}~,~I_{-}=2\lambda u_{1}.
\end{equation*}%
Replacing (\ref{FR.29a}) in the Ermakov-Lewis Invariant (\ref{FR.261}) and
assuming $\phi \neq 0~$we find:%
\begin{equation}
\exp \left( v\left( t\right) \right) =2^{\frac{1}{6}}\phi ^{\frac{1}{12}%
}e^{-A\left( t\right) }\left( 4V_{0}+e^{-12A\left( t\right) }\right) ^{-%
\frac{1}{6}}  \label{FR.29}
\end{equation}%
where
\begin{equation}
A\left( t\right) =\arctan \left[ \frac{2\lambda }{\sqrt{\phi }}\left(
u_{1}e^{2\lambda t}+u_{3}\right) \right] +4\lambda ^{2}u_{1}\sqrt{\phi }.
\label{FR.31}
\end{equation}%
Then the solution is
\begin{equation}
a^{2}\left( t\right) =2^{-\frac{1}{3}}\phi ^{\frac{1}{12}}e^{-A\left(
t\right) }\left( 4V_{0}+e^{-12A\left( t\right) }\right) ^{-\frac{1}{6}%
}\left( u_{1}e^{2\lambda t}+u_{2}e^{-2\lambda t}+2u_{3}\right) ^{\frac{1}{2}}
\label{FR.32A}
\end{equation}%
where from the singularity condition $x\left( 0\right) =0~$we have the
constrain~$u_{1}+u_{2}+2u_{3}=0$ , or
\begin{equation}
2E-\left( I_{+}+I_{-}\right) =0.
\end{equation}%
At late enough time we find $A\left( t\right) \simeq A_{0}$, which implies $%
a^{2}(t)\propto e^{\lambda t}.$

In the case where $\phi =0$ equations (\ref{FR.27}),(\ref{FR.28}) describe
the hyperbolic oscillator and the solution is%
\begin{equation}
u\left( t\right) =\sinh \lambda t~,~2E=\lambda ^{2}.
\end{equation}%
From the Ermakov-Lewis Invariant we have%
\begin{equation}
\exp \left( v\left( t\right) \right) =\left( \frac{\lambda \sinh \lambda t}{%
\lambda v_{1}\sinh \lambda t-12\sqrt{\left\vert V_{0}\right\vert }%
e^{-2\lambda t}}\right) ^{\frac{1}{6}}
\end{equation}%
where $v_{1}$ is a constant. The analytic solution is%
\begin{equation}
a^{2}\left( t\right) =\left( \frac{\lambda \sinh ^{7}\lambda t}{\lambda
v_{1}\sinh \lambda t-12\sqrt{\left\vert V_{0}\right\vert }e^{-2\lambda t}}%
\right) ^{\frac{1}{6}}
\end{equation}

\section{Noether point symmetries in spatially non-flat $f(R)$ models}

\label{Nonf}

In this section we study further the Noether point symmetries in non flat $%
f(R)$ cosmological models. In the context of a FRW spacetime the Lagrangian
of the overall dynamical problem and the Ricci scalar are
\begin{equation}
L=6f^{\prime }a\dot{a}^{2}+6f^{\prime \prime }\dot{R}a^{2}\dot{a}%
+a^{3}\left( f^{\prime }R-f\right) -6Kaf^{\prime }  \label{NF.01}
\end{equation}

\begin{equation}
R=6\left( \frac{\ddot{a}}{a}+\frac{\dot{a}^{2}+K}{a^{2}}\right)
\end{equation}%
where $K$ is the spatial curvature. Note that the two dimensional metric is
given by eq.(\ref{FR.03}) while the \textquotedblright
potential\textquotedblright\ in the Lagrangian takes the form
\begin{equation}
V_{K}(a,R)=-a^{3}(f^{\prime }R-f)+Kaf^{\prime }.
\end{equation}
Based on the above equations and using the theoretical formulation presented
in section \ref{LieN}, we find that the $f(R)$ models which admit non
trivial Noether symmetries are the $f(R)=(R-2\Lambda )^{3/2}$, $%
~f(R)=R^{3/2} $ and $f(R)=R^{2}$. The Noether symmetries can be found in
section \ref{LieN}.

In particular, inserting $f(R)=(R-2\Lambda )^{3/2}$ into the Lagrangian (\ref%
{NF.01})~and changing the variables from $(a,R)$ to $(x,y)$ [see section \ref%
{subs}] we find
\begin{equation}
L=\dot{x}\dot{y}+V_{0}\left( y^{3}+\bar{m}xy\right) -\bar{K}y
\end{equation}%
\begin{equation}
E=\dot{x}\dot{y}-V_{0}\left( y^{3}+{\bar{m}}xy\right) +\bar{K}y
\end{equation}%
where $\bar{K}=3(6^{1/3}K)$. Therefore, the equations of motion are
\begin{eqnarray*}
\ddot{x}-3V_{0}y^{2}-\bar{m}V_{0}x+\bar{K} &=&0 \\
\ddot{y}-\bar{m}V_{0}y &=&0\;.
\end{eqnarray*}%
The constant term $\bar{K}$ appearing in the first equation of motion is not
expected to affect the Noether symmetries (or the integrals of motion).
Indeed we find that the corresponding Noether symmetries coincide with those
of the spatially flat $f(R)=(R-2\Lambda )^{3/2}$ model. However, in the case
of $K\neq 0$ (or $\bar{K}\neq 0$) the analytic solution for the $x$-variable
is written as
\begin{equation}
x_{K}(t)\equiv x(t)+\frac{{\bar{K}}}{\omega ^{2}}
\end{equation}%
where $x(t)$ is the solution of the flat model $K=0$ (see section \ref{subs}%
). Note that the solution of the $y$-variable remains unaltered.

Similarly, for the $f(R)=R^{3/2}~$model the analytic solution is
\begin{equation}
z_{K}\left( t\right) =z\left( t\right) +\bar{K}
\end{equation}%
where $z\left( t\right) $ is the solution of the spatially flat model (see
section \ref{subs1}).

\section{Conclusion}

\label{Conc}

In the literature the functional forms of $f(R)$ of the modified $f(R)$
gravity models are mainly defined on a phenomenological basis. In this
article we use the Noether symmetry approach to constrain these models with
the aim to utilize the existence of non-trivial Noether symmetries as a
selection criterion that can distinguish the $f(R)$ models on a more
fundamental level. Furthermore the resulting Noether integrals can be used
to provide analytic solutions.

In the context of $f(R)$ models, the system of the modified field equations
is equivalent to a two dimensional dynamical system moving in $M^{2}$ (mini
superspace) under the constraint $\bar{E}=$constant. Following the general
methodology of chapter \ref{chapter3}, we require that the two dimensional
system admits extra Noether point symmetries. This requirement fixes the $%
f\left( R\right) $ function and the corresponding analytic solutions are
computed. It is interesting that two well known dynamical systems appear in
cosmology: the anharmonic oscillator and the Ermakov-Pinney system. We
recall that the field equations of the $\Lambda -$cosmology is equivalent
with that of the hyperbolic oscillator.

\newpage%

\begin{subappendices}%

\section{Special solutions for the power law model $R^{n}$}

\label{AppendixAAA}

The case $f\left( R\right) =R^{n}$ is not Liouville integrable via Noether
point symmetries. However the zero order invariant will be used in order to
find special solutions. Inserting $f(R)=R^{n}~~\left( n\neq 0,1,\frac{3}{2},%
\frac{7}{8}\right) ~$into eq.(\ref{SF.50}) we obtain%
\begin{equation}
L\left( a,\dot{a},R,\dot{R}\right) =6naR^{n-1}\dot{a}^{2}+6n\left(
n-1\right) a^{2}R^{n-2}\dot{a}\dot{R}+\left( n-1\right) a^{3}R^{n}
\end{equation}%
and this leads to the modified field equations%
\begin{equation}
\ddot{a}+\frac{1}{a}\dot{a}^{2}-\frac{1}{6}aR=0  \label{RN.01}
\end{equation}%
\begin{equation}
\ddot{R}+\frac{n-2}{R}\dot{R}^{2}-\frac{1}{n-1}\frac{R}{a^{2}}\dot{a}^{2}+%
\frac{2}{a}\dot{a}\dot{R}-\frac{\left( n-3\right) }{6n\left( n-1\right) }%
R^{2}=0  \label{RN.02}
\end{equation}%
\begin{equation}
E=6naR^{n-1}\dot{a}^{2}+6n\left( n-1\right) a^{2}R^{n-2}\dot{a}\dot{R}%
-\left( n-1\right) a^{3}R^{n}.  \label{RN.03}
\end{equation}

The Noether point symmetry (\ref{NS.07}) is also and a Lie symmetry, hence
we have the zero order invariants%
\begin{equation}
a_{0}=at^{-N}~,~R_{0}=Rt^{-2}.
\end{equation}

Applying the zero order invariants in the field equations (\ref{RN.01})-(\ref%
{RN.03}) and in the Noether integral (\ref{NI.07}) we have the following
results.

The dynamical system admits a special solution of the form
\begin{equation}
a=a_{0}t^{N}~,~R=6N\left( 2N-1\right) t^{-2}
\end{equation}%
where the constants $N,$ $E$ $\ $and $I_{7}$ are%
\begin{equation*}
N=\frac{1}{2}~,~E=0~,~I_{7}=0
\end{equation*}%
or
\begin{equation*}
N=-\frac{\left( 2n-1\right) \left( n-1\right) }{n-2}~,~E=0~,~I_{7}=0
\end{equation*}%
or%
\begin{equation*}
N=\frac{2}{3}n~,~E=\left( \frac{12n}{9}\right) ^{n}\left( 4n-3\right)
^{n-1}\left( 13n-8n^{2}-3\right) a_{0}^{3}~,~I_{7}=0.
\end{equation*}

Another special solution is the deSitter solution for $n=2$%
\begin{equation}
a=a_{0}e^{H_{0}t}~,~R=12H_{0}^{2}
\end{equation}%
where $I_{7}=0~$\ and the spacetime is empty i.e. $E=0$.

\end{subappendices}%

\chapter{Noether symmetries in $f(T)$ gravity\label{chapter10}}

\section{Introduction}

In this chapter we continue our analysis on the application of Noether point
symmetries in alternative theories of gravity and specifically the $f\left(
T\right) $ modified theory of gravity. $f\left( T\right) $ gravity it is
based on the old formulation of Teleparallel Equivalent of General
Relativity (TEGR) \cite{Ein1,Ein2,Hayashi79,Maluf94} which instead of the
torsion-less Levi-Civita connection uses the curvatureless Weitzenbock
connection \cite{Weitz01} in which the corresponding dynamical fields are
the four linearly independent vierbeins. Therefore, all the information
concerning the gravitational field is included in the Weitzenbock tensor.
Within this framework, considering invariance under general coordinate
transformations, global Lorentz-parity transformations, and requiring up to
second order terms of the torsion tensor, one can write down the
corresponding Lagrangian density $T$ by using some suitable contractions.

Furthermore $f\left( T\right) $ gravity which is based on the fact that we
allow the Lagrangian to be a function of $T$ \ \cite%
{Ferraro07,Linder10,Chen11}, inspired by the extension of $f\left( R\right)
~ $Einstein-Hilbert action. However, $f\left( T\right) $ gravity does not
coincide with $f\left( R\right) $ extension, but it rather consists of a
different class of modified gravity models. It is interesting to mention
that the torsion tensor includes only products of first derivatives of the
vierbeins, giving rise to second-order field differential equations in
contrast to the $f\left( R\right) $ gravity that provides fourth-order
equations, which potentially may lead to some problems, for example in the
well position and well formulation of the Cauchy problem \cite{Cap09aa}.
Moreover, as we showed in chapter \ref{chapter9} the Lagrangian of the field
equations in $f\left( R\right) $ gravity described a regular dynamical
system; however, in $f\left( T\right) $ the dynamical system is singular and
the $T$ variable can be seen as a Lagrange multiplier.

In section \ref{UnhFrame} we discuss the role of unholonomic frames in the
context of teleparallel gravity and its straightforward extension. In
section \ref{modelfttt}, we briefly present $f(T)$ gravity, while in section %
\ref{Lagrformul} we construct the corresponding generalized Lagrangian
formulation. In section \ref{Noethercond}, we analyze the main properties of
the Noether Symmetry Approach for $f(T)$ gravity. Then, in section \ref%
{frwftt} and \ref{spergem}, we apply these results in the FRW and the static
spherically symmetric spaces.

\section{Unholonomic frames and Connection Coefficient}

\label{UnhFrame}

In an $n-$dimensional manifold $M$ consider a coordinate neighborhood $U$
with a coordinate system $\{x^{\mu }\}.$ At each point $P\epsilon U$ we have
the resulting holonomic frame $\{\partial _{\mu }\}.$ We define in $U\ \ $a
new frame $\{e_{a}(x^{\mu })\}$ which is related to the holonomic frame $%
\{\partial _{a}\}$ as follows:%
\begin{equation}
e_{a}(x^{\mu })=h_{a}^{\mu }\partial _{\mu }~\ \ a,\mu =1,2,...,n
\label{H.1}
\end{equation}%
where the quantities $h_{a}^{\mu }(x)$ are in general functions of the
coordinates (i.e. depend on the point $P$). Notice that Latin indexes count
vectors, while Greek indexes are tensor indices. We assume that $\det
h_{a}^{\mu }\neq 0$ which guaranties that the vectors $\{e_{a}(x^{\mu })\}$
form a set of linearly independent vectors. We define the \textquotedblright
inverse\textquotedblright\ quantities $h_{a}^{\mu }$ by means of the
following \textquotedblright orthogonality\textquotedblright\ relations:
\begin{equation}
h_{a}^{\mu }h_{\nu }^{a}=\delta _{\nu }^{\mu },h_{b}^{\mu }h_{\mu
}^{c}=\delta _{b}^{c}.  \label{H.1a}
\end{equation}%
The commutators of the vectors $\{e_{a}\}$ are not in general all zero. If
they are, then there exists a new coordinate system in $U$, the $\{y^{b}\}$
say so that $e_{b}=\frac{\partial }{\partial y^{b}}$ i.e. the new frame is
holonomic. If there are commutators $[e_{a},e_{b}]\neq 0$ then the new frame
$\{e_{a}\}$ is called unholonomic and the vectors $e_{a}$ cannot be written
in the form $e_{b}=\partial _{b}.$ The quantities which characterize an
unholonomic frame are the objects of unholonomicity or Ricci rotation
coefficients $\Omega _{\text{ }bc}^{a}$ defined by the relation%
\begin{equation}
\lbrack e_{a},e_{b}]=\Omega _{\text{ }ab}^{c}e_{c}.  \label{H.3}
\end{equation}

\bigskip We compute:%
\begin{equation*}
\lbrack e_{a},e_{b}]=[h_{a}^{\mu }\partial _{\mu },h_{b}^{\nu }\partial
_{\nu }]=\left[ h_{a}^{\mu }h_{b,\mu }^{\nu }h_{\nu }^{c}-h_{b}^{\nu
}h_{a,\nu }^{\mu }h_{\mu }^{c}\right] e_{c}
\end{equation*}%
from which follows that the Ricci rotation coefficients of the frame $%
\{e_{a}\}$ are:
\begin{equation}
\Omega _{\text{ }bc}^{a}=2h_{[b}^{\mu }h_{c],\mu }^{\nu }h_{\nu }^{a}.
\label{H.2}
\end{equation}%
The condition for $\{e_{a}\}$ to be a holonomic basis is $\Omega _{\text{ }%
bc}^{a}=0$ at all points $P\epsilon U.$ This is a set of linear partial
differential equations whose solution defines all holonomic frames and all
coordinate systems in $U.$ One obvious solution is $h_{b}^{c}=\delta
_{b}^{c} $. The set of all coordinate systems in $U\ $equipped with the
operation of composition of transformations has the structure of an infinite
dimensional Lie group which is called the Manifold Mapping Group.

We consider now the special unholonomic frames which satisfy Jacobi's
identity:%
\begin{equation}
\lbrack \lbrack
e_{a},e_{b}],e_{c}]+[[e_{b},e_{c}],e_{a}]+[[e_{c},e_{a}],e_{b}]=0.
\label{H.4}
\end{equation}%
These frames are the generators of a Lie Algebra, therefore they have an
extra role to play. Replacing the commutator in terms of the unholonomicity
objects we find the following identity:%
\begin{equation}
\Omega _{\text{ }ab,c}^{d}+\Omega _{\text{ }ba,a}^{d}+\Omega _{\text{ }%
ca,b}^{d}-\Omega _{\text{ }ab}^{l}\Omega _{cl}^{d}-\Omega _{\text{ }%
bc}^{l}\Omega _{al}^{d}-\Omega _{\text{ }ca}^{l}\Omega _{bl}^{d}=0.
\end{equation}

Using the definition of the covariant derivative we write:%
\begin{equation}
\nabla _{e_{i}}e_{j}=\Gamma _{ij}^{k}e_{k}  \label{H.20}
\end{equation}%
where $\Gamma _{ij}^{k}$ are the connection coefficients in the frame $%
\{e_{i}\}.$ Let us compute these $\Gamma _{ij}^{k}$ assuming that
\begin{equation*}
\lbrack e_{i},e_{j}]=C_{.ij}^{k}e_{k}
\end{equation*}%
from which follows%
\begin{equation*}
C_{.ij}^{k}=\Omega _{.jk}^{k}.
\end{equation*}

Consider three vector fields $X,Y,Z$ and the covariant derivative of the
metric wrt $X.$ Then we have:

\begin{equation}
\nabla _{X}g(Y,Z)=X(g(Y,Z))-g(\nabla _{X}Y,Z)-g(Y,\nabla _{Y}Z)  \label{H.21}
\end{equation}%
and by interchanging the role of $X,Y,Z:$%
\begin{equation}
\nabla _{Y}g(Z,X)=Y(g(Z,X))-g(\nabla _{Y}Z,X)-g(Z,\nabla _{Z}X)  \label{H.22}
\end{equation}

\begin{equation}
\nabla _{Z}g(X,Y)=Z(g(X,Y))-g(\nabla _{Z}X,Y)-g(X,\nabla _{X}Y).
\label{H.23}
\end{equation}

Adding (\ref{H.21}), (\ref{H.22}) and subtracting (\ref{H.23}) we obtain:

\begin{eqnarray*}
\nabla _{X}g(Y,Z)+\nabla _{Y}g(Z,X)-\nabla _{Z}g(X,Y)
&=&X(g(Y,Z))+Y(g(Z,X))-Z(g(X,Y))+ \\
&&-\left[ g(\nabla _{X}Y,Z)+g(\nabla _{Y}Z,X)-g(\nabla _{Z}X,Y)\right] + \\
&&-\left[ g(Y,\nabla _{X}Z)+g(Z,\nabla _{Y}X)-g(X,\nabla _{Z}Y)\right]
\end{eqnarray*}%
then%
\begin{eqnarray*}
\nabla _{X}g(Y,Z)+\nabla _{Y}g(Z,X)-\nabla _{Z}g(X,Y)
&=&X(g(Y,Z))+Y(g(Z,X))-Z(g(X,Y))+ \\
&&-\left[ g(\nabla _{X}Y,Z)+g(Z,\nabla _{Y}X)\right] + \\
&&-\left[ g(\nabla _{Y}Z,X)-g(X,\nabla _{Z}Y)\right] + \\
&&-\left[ g(Y,\nabla _{X}Z)-g(\nabla _{Z}X,Y)\right] /
\end{eqnarray*}%
that is%
\begin{eqnarray*}
\nabla _{X}g(Y,Z)+\nabla _{Y}g(Z,X)-\nabla _{Z}g(X,Y)
&=&X(g(Y,Z))+Y(g(Z,X))-Z(g(X,Y))+ \\
&&-\left[ g(Z,\nabla _{X}Y+\nabla _{Y}X)+g(X,\nabla _{Y}Z-\nabla
_{Z}Y)+g(Y,\nabla _{X}Z-\nabla _{Z}X)\right] .
\end{eqnarray*}%
The term%
\begin{equation*}
g(Z,\nabla _{X}Y+\nabla _{Y}X)=2g\left( Z,\nabla _{X}Y\right) +g\left(
Z,\nabla _{Y}X-\nabla _{X}Y\right) .
\end{equation*}%
Replacing in the last relation and solving for $2g\left( Z,\nabla
_{X}Y\right) $ we find%
\begin{eqnarray*}
2g\left( Z,\nabla _{X}Y\right) &=&\left[ X(g(Y,Z))+Y(g(Z,X))-Z(g(X,Y))\right]
+ \\
&&-\left[ \nabla _{X}g(Y,Z)+\nabla _{Y}g(Z,X)-\nabla _{Z}g(X,Y)\right] + \\
&&-\left[ g\left( Z,\nabla _{Y}X-\nabla _{X}Y\right) +g(X,\nabla
_{Y}Z-\nabla _{Z}Y)+g(Y,\nabla _{X}Z-\nabla _{Z}X)\right]
\end{eqnarray*}%
or%
\begin{eqnarray*}
2g\left( Z,\nabla _{X}Y\right) &=&\left[ X(g(Y,Z))+Y(g(Z,X))-Z(g(X,Y))\right]
+ \\
&&-\left[ \nabla _{X}g(Y,Z)+\nabla _{Y}g(Z,X)-\nabla _{Z}g(X,Y)\right] + \\
&&-\left[ g\left( Z,\nabla _{Y}X-\nabla _{X}Y-\left[ Y,X\right] \right)
+g(X,\nabla _{Y}Z-\nabla _{Z}Y-\left[ Y,Z\right] )+g(Y,\nabla _{X}Z-\nabla
_{Z}X-\left[ X,Z\right] )\right] + \\
&&-\left[ g\left( Z,\left[ Y,X\right] \right) +g\left( X,\left[ Y,Z\right]
\right) +g\left( Y,\left[ X,Z\right] \right) \right] .
\end{eqnarray*}%
Define the quantities%
\begin{eqnarray*}
T_{\nabla }(X,Y) &=&\nabla _{X}Y-\nabla _{Y}X-[X,Y] \\
A_{\nabla }(X,Y,Z) &=&\nabla _{X}g(Y,Z)
\end{eqnarray*}%
The tensors $T_{\nabla }$ and $A_{\nabla }$ are called the \textit{torsion}
and the \textit{metricity} of the connection $\nabla $ respectively. Last
relation in terms of the fields $T_{\nabla }$ and $A_{\nabla }$ is written
as follows$:$%
\begin{eqnarray}
2g\left( Z,\nabla _{X}Y\right) &=&\left[ X(g(Y,Z))+Y(g(Z,X))-Z(g(X,Y))\right]
+  \notag \\
&&-\left[ A_{\nabla }\left( X,Y,Z\right) +A_{\nabla }\left( Y,Z,X\right)
-A_{\nabla }\left( Z,X,Y\right) \right] +  \notag \\
&&-\left[ g\left( Z,T_{\nabla }\left( Y,X\right) \right) +g\left(
X,T_{\nabla }\left( Y,Z\right) \right) +g\left( Y,T_{\nabla }\left(
X,Z\right) \right) \right]  \notag \\
&&-\left[ g\left( Z,\left[ Y,X\right] \right) +g\left( X,\left[ Y,Z\right]
\right) +g\left( Y,\left[ X,Z\right] \right) \right] .  \label{H.24}
\end{eqnarray}

Let $X=e_{l}~,~Y=e_{j}$ and $Z=e_{k}.$ Contracting with $\frac{1}{2}g^{il}$
we have%
\begin{equation*}
2g\left( Z,\nabla _{X}Y\right) \rightarrow \Gamma _{jk}^{i}
\end{equation*}%
\begin{equation*}
\left[ X(g(Y,Z))+Y(g(Z,X))-Z(g(X,Y))\right] \rightarrow \left\{
_{jk}^{i}\right\}
\end{equation*}%
\begin{equation*}
g\left( X,T_{\nabla }\left( Y,Z\right) \right) \rightarrow Q_{.kj}^{i}
\end{equation*}%
\begin{equation*}
g\left( Z,T_{\nabla }\left( Y,X\right) \right) +g\left( Y,T_{\nabla }\left(
X,Z\right) \right) \rightarrow g^{il}(g_{tj}Q_{kl}^{t}+g_{tk}Q_{jl}^{t})=-%
\bar{S}_{.kj}^{i}
\end{equation*}%
\begin{equation*}
g\left( X,\left[ Y,Z\right] \right) \rightarrow \frac{1}{2}C_{.jk}^{i}
\end{equation*}%
\begin{equation*}
g\left( Z,\left[ Y,X\right] \right) +g\left( Y,\left[ X,Z\right] \right) =%
\frac{1}{2}g^{il}(g_{tj}C_{lk}^{t}+g_{tk}C_{jl}^{t})=-S_{.kj}^{i}
\end{equation*}%
and%
\begin{equation*}
A_{\nabla }\left( X,Y,Z\right) +A_{\nabla }\left( Y,Z,X\right) -A_{\nabla
}\left( Z,X,Y\right) \rightarrow \frac{1}{2}g^{il}\Delta _{jkl}
\end{equation*}%
Replacing in (\ref{H.24}) we find the connection coefficients in the frame $%
\{e_{i}\}$%
\begin{equation}
\Gamma _{jk}^{i}=\left\{ _{jk}^{i}\right\} +\bar{S}_{.kj}^{i}+S_{.kj}^{i}-%
\frac{1}{2}g^{il}\Delta _{jkl}+Q_{jk}^{i}-\frac{1}{2}C_{.jk}^{i}  \label{H25}
\end{equation}%
where $\left\{ _{jk}^{i}\right\} $ is the standard Levi Civita connection
coefficients (Christofell symbols). \ This is the most general expression
for the connection coefficients in terms of the fields $\left\{
_{jk}^{i}\right\} ,$ $T_{\nabla },$ $A_{\nabla }$ and $C_{jk}^{i}$.

Concerning the symmetric and antisymmetric part we have:%
\begin{eqnarray}
\Gamma _{.(jk)}^{i} &=&\left\{ _{jk}^{i}\right\} +\bar{S}%
_{.jk}^{i}+S_{.jk}^{i}-\frac{1}{2}g^{il}\Delta _{jkl}  \label{H.26} \\
\Gamma _{.[jk]}^{i} &=&Q_{.jk}^{i}-\frac{1}{2}C_{.jk}^{i}.  \label{H.27}
\end{eqnarray}

From the above we draw the following conclusions:

1. The connection coefficients in a frame $\{e_{i}\}$ are determined from
the metric, the torsion, the metricity and the unholonomicity objects
(equivalently the commutators) of the frame vectors

2. The symmetric part $\Gamma _{.(jk)}^{i}$ of $\Gamma _{jk}^{i}$ depends on
all fields. This means that the geodesics and the autoparallels in a given
frame depend on the geometric properties of the space (fields $g_{ij},$ $%
Q_{.kj}^{i},g_{ij|k})$ and the unholonomicity of the frame (field $%
C_{.jk}^{i})$

3. The antisymmetric part $\Gamma _{.[jk]}^{i}$ of $\Gamma _{jk}^{i}$
depends only on the fields $Q_{.kj}^{i}$ and $C_{.jk}^{i}.$

4. The objects of unholonomicity $C_{.jk}^{i}$ behave in the same way as the
components of the torsion. This means that even in a Riemannian space where $%
Q_{.kj}^{i}=0,g_{ij|k}=0$ in an unholonomic basis the antisymmetric part $%
\Gamma _{.[jk]}^{i}=-\frac{1}{2}C_{.jk}^{i}\neq 0.$ This result has lead to
the misunderstanding that when one works in an unholonomic frame then one
has introduced torsion, which is not correct! This is the case with the
TEGR. This misunderstanding has important consequences because the effects
one will observe in an unholonomic frame will be frame dependent effects and
not covariant effects. Therefore all conclusions made in a specific
unholonomic frame are restricted to that frame only.

\section{$f(T)$-gravity}

\label{modelfttt}

Teleparallelism uses as dynamical objects special unholonomic frames in
spacetime, called vierbeins, which are defined by the requirement
\begin{equation}
g(e_{i},e_{j})=e_{i}.e_{j}=\eta _{ij}  \label{metrdef}
\end{equation}%
where $\eta _{ij}=\mathrm{diag}(1,-1,-1,-1)$ is the Lorentz metric in
canonical form. Obviously $g_{\mu \nu }(x)=\eta _{ij}h_{\mu }^{i}(x)h_{\nu
}^{j}(x)$ where $e^{i}(x)=h_{\mu }^{i}(x)dx^{i}$ is the dual basis.
Differing from GR, which uses the torsionless Levi-Civita connection,
Teleparallelism utilizes the curvatureless Weitzenb\"{o}ck connection \cite%
{Weitz01}, where Weitzenb\"{o}ck non-null torsion is
\begin{equation}
T_{\mu \nu }^{\beta }=\hat{\Gamma}_{\nu \mu }^{\beta }-\hat{\Gamma}_{\mu \nu
}^{\beta }=h_{i}^{\beta }(\partial _{\mu }h_{\nu }^{i}-\partial _{\nu
}h_{\mu }^{i})\;.  \label{Wein}
\end{equation}%
Notice that we assume that the Ricci rotation coefficients obey $\Omega
_{jk}^{i}=T_{jk}^{i}$ and encompass all the information concerning the
gravitational field. The TEGR Lagrangian for the gravitational field
equations (Einstein equations) is taken to be:%
\begin{equation}
T=S_{\beta }^{\mu \nu }T_{\mu \nu }^{\beta }  \label{lagrangian}
\end{equation}%
where
\begin{equation}
S_{\beta }^{\mu \nu }=\frac{1}{2}(K_{\beta }^{\mu \nu }+\delta _{\beta
}^{\mu }T_{\theta }^{\theta \nu }-\delta _{\beta }^{\nu }T_{\theta }^{\theta
\mu })  \label{s}
\end{equation}%
and $K_{\beta }^{\mu \nu }$ is the tensor
\begin{equation}
K_{\beta }^{\mu \nu }=-\frac{1}{2}(T_{\beta }^{\mu \nu }-T_{\beta }^{\nu \mu
}-T_{\beta }^{\mu \nu }),  \label{contorsion}
\end{equation}%
which equals the difference of the Levi Civita connection in the holonomic
and the unholonomic frame.

In this work the gravitational field will be driven by a Lagrangian density
which is a function of $T$. Therefore, the corresponding action of $f(T)$
gravity reads as
\begin{equation}
\mathcal{A}=\frac{1}{16\pi G}\int {d^{4}xef(T)}  \label{action1122}
\end{equation}%
where $e=det(e_{\mu }^{i}\cdot e_{\nu }^{i})=\sqrt{-g}$ and $G$ is Newton's
constant. TEGR and thus General Relativity is restored when $f(T)=T$. First
of all, in order to construct a realistic cosmology we have to incorporate
the matter and radiation sectors too. Therefore, the total action is written
as
\begin{equation}
S_{\mathrm{tot}}=\mathcal{A}+\frac{1}{16\pi G}\int d^{4}xe\left(
L_{m}+L_{r}\right) ,  \label{action11}
\end{equation}%
If matter couples to the metric in the standard form then the variation of
the action (\ref{action11}) with respect to the vierbein leads to the
equations \cite{Ferraro}
\begin{eqnarray}
&&e^{-1}\partial _{\mu }(e{S}_{i}^{\mu \nu })f^{\prime }(T)-h_{i}^{\lambda
}T_{\mu \lambda }^{\beta }S_{\beta }^{\nu \mu }f^{\prime }(T)  \notag \\
&&+S_{i}^{\mu \nu }\partial _{\mu }(T)f^{\prime \prime }(T)+\frac{1}{4}%
h_{i}^{\nu }f(T)=4\pi Gh_{i}^{\beta }T_{\beta }^{\nu }  \label{equations}
\end{eqnarray}%
where a prime denotes differentiation with respect to $T$, ${S_{i}}^{\mu \nu
}={h_{i}}^{\beta }S_{\beta }^{\mu \nu }$ and $T_{\mu \nu }$ is the matter
energy-momentum tensor.

\section{Generalized Lagrangian formulation of $f(T)$ gravity}

\label{Lagrformul}

In this section, we provide a generalized Lagrangian formulation in order to
construct a theory of $f(T)$ gravity. Specifically, the gravitational field
is driven by the Lagrangian density $f(T)$ in (\ref{action1122}), which can
be generalized through the use of a Lagrange multiplier. In particular, we
can write it as
\begin{equation}
L\left( x^{k},x^{\prime k},T\right) =2f_{T}\bar{\gamma}_{ij}\left(
x^{k}\right) x^{\prime i}x^{\prime j}+M\left( x^{k}\right) \left(
f-Tf_{T}\right)  \label{L1}
\end{equation}%
where $x^{\prime }=\frac{dx}{d\tau }$, $M(x^{k})$ is the Lagrange multiplier
and $\bar{\gamma}_{ij}$ is a second rank tensor which is related to the
frame [one can use $eT(x^{k},x^{\prime k})$] of the background spacetime. In
the same lines, the Hamiltonian of the system is written as
\begin{equation}
H\left( x^{k},x^{\prime k},T\right) =2f_{T}\bar{\gamma}_{ij}\left(
x^{k}\right) x^{\prime i}x^{\prime j}-M\left( x^{k}\right) \left(
f-Tf_{T}\right) =0\;.  \label{L2}
\end{equation}

In this case, the system is autonomous hence $\partial _{\tau }$ is a
Noether symmetry with corresponding Noether integral the Hamiltonian $H$. \

In this framework, considering $\{x^{k},T\}$ as the canonical variables of
the configuration space, we can derive, after some algebra, the general
field equations of $f(T)$ gravity. Indeed, starting from the Lagrangian (\ref%
{L1}), the Euler-Lagrange equations
\begin{equation}
\frac{\partial L}{\partial T}=0,\;\;\;\;\frac{d}{d\tau }\left( \frac{%
\partial L}{\partial x^{\prime k}}\right) -\frac{\partial L}{\partial x^{k}}%
=0\,  \label{Lf.01}
\end{equation}%
give rise to
\begin{equation}
f_{TT}\left( 2\bar{\gamma}_{ij}x^{\prime i}x^{\prime j}-MT\right) =0,
\label{Lf.04}
\end{equation}%
\begin{equation}
x^{i\prime \prime }+{\bar{\Gamma}}_{jk}^{i}x^{j\prime }x^{k\prime }+\frac{%
f_{TT}}{f_{T}}x^{i\prime }T^{\prime }-M^{,i}\frac{\left( f-Tf_{T}\right) }{%
4f_{T}}=B_{m}^{i}\;.  \label{Lf.06}
\end{equation}

The functions ${\bar{\Gamma}}_{jk}^{i}$ are considered to be the Christoffel
symbols for the metric $\bar{\gamma}_{ij}$. Therefore, the system is
determined by the two independent differential equations (\ref{Lf.04}),(\ref%
{Lf.06}) and the Hamiltonian constraint $H=C_{m}$ where $H$ is given by
equation (\ref{L2}) and $C_{m},$ $B_{m}^{i}$ are the components of the
energy momentum tensor $T_{\mu \nu .}$.

The point-like Lagragian (\ref{L1}) determines completely the related
dynamical system in the minisuperspace $\{x^{k},T \}$, implying that one can
easily recover some well known cases of cosmological interest. In brief,
these are:

\begin{itemize}
\item The static spherically symmetric spacetime:
\begin{equation}
ds^{2}=-a^{2}\left( \tau \right) dt^{2}+\frac{1}{N^{2}\left( a\left( \tau
\right) ,b\left( \tau \right) \right) }d\tau ^{2}\,+b^{2}\left( \tau \right)
\left( d\theta ^{2}+\sin ^{2}\theta d\phi ^{2}\right)  \label{SSA.1}
\end{equation}%
arising from the diagonal vierbein \footnote{%
Note that, in general, one can choose a non-diagonal vierbein, giving rise
to the same metric through (\ref{metrdef}).}
\begin{equation}
e_{i}^{A}=\left( a\left( \tau \right) ,\frac{1}{N\left( a\left( \tau \right)
,b\left( \tau \right) \right) },b\left( \tau \right) ,b\left( \tau \right)
\sin \theta \right) \;  \label{SSA}
\end{equation}%
where $a(\tau )$ and $b(\tau )$ are functions which need to be determined.
Therefore, the line element of $\bar{\gamma}_{ij}$ and $M\left( x^{k}\right)
$ are given by
\begin{equation}
ds_{\bar{\gamma}}^{2}=N\left( 2b~da~db+a~db^{2}\right) ~,~~M(a,b)=\frac{%
ab^{2}}{N}.  \label{SSA1}
\end{equation}

\item The flat FRW spacetime with Cartesian coordinates:
\begin{equation}
ds^{2}=-dt^{2}+a^{2}\left( t\right) \left( dx^{2}+dy^{2}+dz^{2}\right)
\end{equation}%
arising from the vierbein
\begin{equation}
e_{i}^{A}=\left( 1,a\left( t\right) ,a\left( t\right) ,a\left( t\right)
\right) \;
\end{equation}%
where $t$ is the cosmic time and $a(t)$ is the scale factor of the universe.
In this case we have
\begin{equation}
ds_{\bar{\gamma}}^{2}=3a~da^{2}~,~M(a)=a^{3}(t).
\end{equation}

\item The Bianchi type I spacetime:
\begin{equation}
\!ds^{2}=-\frac{1}{N^{2}\left( a\left( t\right) ,\beta \left( t\right)
\right) }dt^{2}+a^{2}\left( t\right) \left[ e^{-2\beta \left( t\right)
}dx^{2}+e^{\beta \left( t\right) }\left( dy^{2}+dz^{2}\right) \right] \
\end{equation}%
arising from the vierbein
\begin{equation}
e_{i}^{A}=\left( \frac{1}{N\left( a\left( t\right) ,\beta \left( t\right)
\right) },a(t)e^{-\beta (t)},a(t)^{\frac{\beta (t)}{2}},a(t)^{\frac{\beta (t)%
}{2}}\right) \;.
\end{equation}%
In this case, we obtain
\begin{equation}
ds_{\bar{\gamma}}^{2}=N\left( -4ada^{2}+a^{3}d\beta ^{2}\right)
~~,~~M(a,\beta )=\frac{a^{3}(t)}{N}.
\end{equation}%
\qquad
\end{itemize}

In the present work we will focus on the static spherically-symmetric metric
deriving new spherically symmetric solutions for $f(T)$ gravity. In
particular, we look for Noether symmetries in order to reveal the existence
of analytical solutions.

\section{The Noether Symmetry Approach for $f(T)$ gravity}

\label{Noethercond}

The aim is now to apply the Noether Symmetry Approach to a general class of $%
f(T)$ gravity models where the corresponding Lagrangian of the field
equations is given by equation (\ref{L1}). First of all, we perform the
analysis for arbitrary spacetimes, and then we focus on the spatially flat
FRW spacetime and on the static spherically-symmetric spacetime.

\subsection{Searching for Noether point symmetries in general spacetimes}

The Noether symmetry condition for Lagrangian (\ref{L1}) is given by
\begin{equation}
X^{\left[ 1\right] }L+L\xi ^{\prime }=g^{\prime }~,~g=g\left( \tau
,x^{i}\right) .  \label{LL2}
\end{equation}%
Notice that the Lagrangian (\ref{H.52}) is a singular Lagrangian (the
Hessian vanishes), hence the jet space is $\bar{B}_{M}=\left\{ \tau ,x^{i},T,%
\dot{x}^{i}\right\} $ and thus the first prolongation of $X$ in the jet
space $\bar{B}_{M}~$is \cite{Christ,Havelkova, Ziping}
\begin{eqnarray}
&&X^{\left[ 1\right] }=\xi \left( \tau ,x^{k},T\right) \partial _{\tau
}+\eta ^{k}\left( \tau ,x^{k},T\right) \partial _{i}  \notag \\
&&\ \ \ \ \ \ \ \ \ +\mu \left( \tau ,x^{k},T\right) \partial _{T}+\left(
\eta ^{\prime i}-\xi ^{\prime }x^{\prime i}\right) \partial _{x^{\prime i}}.
\end{eqnarray}

For each term of the Noether condition (\ref{LL2}) for the Lagrangian (\ref%
{L1}) we obtain
\begin{eqnarray*}
X^{\left[ 1\right] }L &=&2f_{T}\bar{g}_{ij,k}\eta ^{k}x^{\prime i}x^{\prime
j}+M_{,k}\eta ^{k}\left( f-Tf_{T}\right) \\
&&+2f_{TT}\mu \bar{g}_{ij}x^{\prime i}x^{\prime j}-Mf_{TT}\mu \\
&&+4f_{T}\bar{g}_{ij}x^{\prime i}\left( \eta _{,\tau }^{j}+\eta
_{,k}^{j}x^{\prime k}+\eta _{,T}^{j}T^{\prime }\right. \\
&&\left. -\xi _{,\tau }x^{\prime j}-\xi _{,k}x^{\prime j}x^{\prime k}-\xi
_{,T}x^{\prime j}T^{\prime }\right) ,
\end{eqnarray*}%
\begin{equation*}
L\xi ^{\prime }=\left[ 2f_{T}\bar{g}_{ij}x^{\prime i}x^{\prime j}+M\left(
x^{i}\right) \left( f-Tf_{T}\right) \right] \left( \xi _{,\tau }+\xi
_{,k}x^{\prime k}+\xi _{,T}T^{\prime }\right) ,
\end{equation*}%
\begin{equation*}
g^{\prime }=g_{,\tau }+g_{,k}x^{\prime k}+g_{,T}T^{\prime }\;.
\end{equation*}%
Inserting these expressions into (\ref{LL2}) we find the Noether symmetry
conditions
\begin{equation}
\xi _{,k}=0~,~\xi _{,T}=0~,~g_{,T}=0~,~\eta _{,T}=0,
\end{equation}%
\begin{equation}
4f_{T}\bar{\gamma}_{ij}\eta _{,\tau }^{k}=g_{,k},  \label{LL3}
\end{equation}%
{\small {\
\begin{equation}
M_{,k}\eta ^{k}\left( f-Tf_{T}\right) -MTf_{TT}\mu +\xi _{,\tau }M\left(
f-Tf_{T}\right) -g_{,\tau }=0,  \label{LL4}
\end{equation}%
}}
\begin{equation}
2f_{T}\bar{\gamma}_{ij,k}\eta ^{k}+2f_{TT}\mu \bar{\gamma}_{ij}+4f_{T}\bar{%
\gamma}_{ij}\eta _{,k}^{j}-2f_{T}\bar{\gamma}_{ij}\xi _{,\tau }=0\;.
\label{LL5}
\end{equation}%
Conditions $\eta _{,T}=g_{,T}=0$ imply, through equation (\ref{LL3}), that $%
\eta _{,\tau }^{k}=g_{,k}=0$. Also, equation (\ref{LL5}) takes the form
\begin{equation}
L_{\eta }\bar{\gamma}_{ij}=\left( \xi _{,\tau }-\frac{f_{TT}}{f_{T}}\mu
\right) \bar{\gamma}_{ij},  \label{LL6}
\end{equation}%
where $L_{\eta }\bar{\gamma}_{ij}$ is the Lie derivative with respect to the
vector field $\eta ^{i}(x^{k})$. Furthermore, from (\ref{LL6}) we deduce
that $\eta ^{i}$ is a CKV of the metric $\bar{\gamma}_{ij}$, with conformal
factor
\begin{equation}
2\bar{\psi}\left( x^{k}\right) =\xi _{,\tau }-\frac{f_{TT}}{f_{T}}\mu =\xi
_{,\tau }-S(\tau ,x^{k})\;.  \label{LL7}
\end{equation}%
Finally, utilizing simultaneously equations (\ref{LL4}), (\ref{LL6}), (\ref%
{LL7}) and the condition $g_{,\tau }=0$, we rewrite (\ref{LL4}) as
\begin{equation}
M_{,k}\eta ^{k}+\left[ 2\bar{\psi}+\left( 1-\frac{Tf_{T}}{f-Tf_{T}}\right) S%
\right] M=0\;.  \label{LL9}
\end{equation}%
Considering that $S=S(x^{k})$ and using the condition $g_{,\tau }=0$, we
acquire $\xi _{,\tau }=2\bar{\psi}_{0},\bar{\psi}_{0}\in \mathbb{R}$ with $%
S=2(\bar{\psi}_{0}-\bar{\psi})$. At this point, we have to deal with the
following two situations:

Case 1. In the case of $S=0$, the symmetry conditions are
\begin{eqnarray}
&&L_{\eta }\bar{\gamma}_{ij}=2\bar{\psi}_{0}\bar{\gamma}_{ij}  \notag \\
&&M_{,k}\eta ^{k}+2\bar{\psi}_{0}M=0
\end{eqnarray}%
implying that the vector $\eta ^{i}(x^{k})$ is a homothetic Vector of the
metric $~\bar{\gamma}_{ij}$. The latter means that for arbitrary $f\left(
T\right) \neq T^{n}$ functional forms, the dynamical system could possibly
admit extra (time independent) Noether symmetries.

Case 2. If $S\neq 0$ then equation (\ref{LL9}) leads to the differential
equation
\begin{equation}
\frac{Tf_{T}}{f-Tf_{T}}=C
\end{equation}%
which has the solution
\begin{equation}
f(T)=T^{n},\;\;\;\;C\equiv \frac{n}{1-n}\;.
\end{equation}%
In this context, $\eta ^{i}(x^{k})$ is a CKV of $\bar{\gamma}_{ij}$, and the
symmetry conditions become
\begin{eqnarray}
&&L_{\eta }\bar{\gamma}_{ij}=2\bar{\psi}\bar{\gamma}_{ij},  \notag \\
&&M_{,k}\eta ^{k}+\left[ 2\bar{\psi}+\left( 1-C\right) S\right] =0,
\end{eqnarray}%
with $S=2(\bar{\psi}_{0}-\bar{\psi})$.

Collecting the above results we have the following result

\begin{lemma}
\label{Lemmaft}The general autonomous Lagrangian%
\begin{equation*}
L\left( x^{k},x^{\prime k},T\right) =2f_{T}\bar{\gamma}_{ij}\left(
x^{k}\right) x^{\prime i}x^{\prime j}+M\left( x^{k}\right) \left(
f-Tf_{T}\right)
\end{equation*}%
admits extra Noether point symmetries as follows

a) If $f\left( T\right) $ is an arbitrary function of $T$, then the symmetry
vector is written as%
\begin{equation*}
X=\left( 2\psi _{0}\tau +c_{1}\right) \partial _{\tau }+\eta ^{i}\left(
x^{k}\right) \partial _{i}
\end{equation*}%
where $\eta ^{i}\left( x^{k}\right) $ is a HV/KV of the metric $\bar{\gamma}%
_{ij}$ and the following condition holds%
\begin{equation*}
M_{,k}\eta ^{k}+2\bar{\psi}_{0}M=0.
\end{equation*}

b) If $f\left( T\right) $ is a power law, i.e. $f\left( T\right) =T^{n}$,
then we have the extra symmetry vector%
\begin{equation*}
X=\left( 2\bar{\psi}_{0}\tau \right) \partial _{\tau }+\eta ^{i}\left(
x^{k}\right) \partial _{i}+\frac{2\bar{\psi}_{0}-2\bar{\psi}\left(
x^{k}\right) }{C}T\partial _{T}
\end{equation*}%
where $C=\frac{n}{1-n}$, $\eta ^{i}\left( x^{k}\right) $ is a CKV of the
metric $\bar{\gamma}_{ij}$ with conformal factor $\bar{\psi}\left(
x^{k}\right) $ and the following condition holds%
\begin{equation*}
M_{,k}\eta ^{k}+\left[ 2\bar{\psi}+\left( 1-C\right) S\right] M=0
\end{equation*}%
where $S=2\left( \bar{\psi}_{0}-\bar{\psi}\right) $.

In both cases the corresponding gauge function is a constant.
\end{lemma}

In the following we apply the above Lemma in the case of FRW cosmology and
static spherical symmetric spacetimes.

\section{Spatially flat FRW}

\label{frwftt}

The FRW\ in the holonomic (commoving) frame $\{\partial t,\partial
x,\partial y,\partial z\}$ has the form
\begin{equation*}
ds^{2}=-dt^{2}+a^{2}(t)(dx^{2}+dy^{2}+dz^{2})
\end{equation*}%
where $a(t)$ is the cosmological scale factor. In this spacetime we define
the vierbein (unholonomic frame) $\{e_{i}\}$ with the requirement:
\begin{equation}
h_{\mu }^{i}(t)=\mathrm{diag}(-1,a(t),a(t),a(t)).  \label{metric}
\end{equation}%
In order to derive the cosmological equations in a FRW metric, we need to
deduce a point-like Lagrangian from the action (\ref{action1122}). As a
consequence, the infinite number of degrees of freedom of the original field
theory will be reduced to a finite number. In this framework considering $%
(a,T)$ as canonical variables the corresponding $f(T)$ Lagrangian becomes:%
\begin{equation}
\mathcal{L}=a^{3}\left[ f(T)-Tf^{\prime }(T)\right] -6\dot{a}^{2}af^{\prime
}(T)\,.  \label{H.52}
\end{equation}

Therefore the field equations are
\begin{equation}
T=-6\left( \frac{\dot{a}^{2}}{a^{2}}\right) =-6H^{2}  \label{H.50}
\end{equation}%
\begin{equation}
12H^{2}f^{\prime }(T)+f(T)=16\pi G\rho \;  \label{friedmann}
\end{equation}%
\begin{equation}
48H^{2}f^{\prime \prime }(T)\dot{H}-f^{\prime 2}+4\dot{H}-f(T)=16\pi Gp
\label{acceleration}
\end{equation}%
where $H$ is the Hubble parameter, $\rho =\rho _{m}+\rho _{r}$ and $%
p=p_{m}+p_{r}$ are the total energy density and (isotropic) pressure
respectively, which are measured in the unholonomic frame. It is interesting
to mention that using the conservation equation $\dot{\rho}+3H(\rho +p)=0$
one can rewrite equations (\ref{friedmann}) and (\ref{acceleration}) in the
usual form
\begin{equation}
H^{2}=\frac{8\pi G}{3}(\rho +\rho _{T})  \label{modfri}
\end{equation}%
\begin{equation}
2\dot{H}+3H^{2}=-\frac{8\pi G}{3}(p+p_{T})  \label{modacce}
\end{equation}%
where
\begin{equation}
\rho _{T}=\frac{1}{16\pi G}[2Tf^{\prime }(T)-f(T)-T]  \label{rhoT}
\end{equation}%
\begin{equation}
p_{T}=\frac{1}{16\pi G}[2\dot{H}(4Tf^{\prime \prime }(T)+2f^{\prime
}(T)-1)]-\rho _{T}.  \label{pT}
\end{equation}%
are the unholonomicity contributions to the energy density and pressure.
Finally, a basic question here is the following: under which circumstances $%
f(T)$ gravity can resemble that of the scalar field dark energy? In order to
address this crucial question we need to calculate the effective
equation-of-state parameter $w(a)$ for the $f(T)$ cosmology. Indeed,
utilizing equations (\ref{rhoT}) and (\ref{pT}), we can easily obtain the
effective unholonomicity equation of state as
\begin{equation}
\omega _{T}\equiv \frac{p_{T}}{\rho _{T}}=-1+\frac{4\dot{H}(4Tf^{\prime
\prime }(T)+2f^{\prime }(T)-1)}{4Tf^{\prime }(T)-2f(T)-T}\;.
\label{omegaeff}
\end{equation}

\subsection{Noether symmetries}

From Lemma \ref{Lemmaft} for the Lagrangian (\ref{H.52}) we find that for

\begin{itemize}
\item For arbitrary $f\left( T\right) $ the Lagrangian (\ref{H.52}) admits
only the Noether symmetry $\partial _{t}$

\item For $f\left( T\right) =f_{0}T^{n}~$where $f_{0}$ is the integration
constant we have the following extra Noether symmetries:\newline
- For $n\neq \frac{1}{2},\frac{3}{2}$ the Noether point symmetry vector is
\begin{eqnarray*}
X_{1} &=&\left( \frac{3C}{2n-1}t\right) \partial _{t}+\left( Ca+c_{3}a^{1-%
\frac{3}{2n}}\right) \partial _{a}+ \\
&&+\left[ \frac{1}{n}\left( \left( c-m\right) n+3c_{3}a^{-\frac{3}{2n}%
}\right) +\frac{3C}{2n-1}+c\right] T\partial _{T}
\end{eqnarray*}%
with corresponding Noether integral
\begin{equation*}
I_{1}=\left( \frac{3C}{2n-1}t\right) \mathcal{H}-12f_{0}n\left(
Ca^{2}+c_{3}a^{2-\frac{3}{2n}}\right) T^{n-1}\dot{a}
\end{equation*}%
where $C=\frac{m\left( 1-n\right) +nc}{3}$.\newline
- For $n=\frac{3}{2},$ the Noether point symmetry is
\begin{eqnarray}
X_{2} &=&\frac{1}{5}\left( 3c-2m\right) t\partial _{t}+\left[ \left( \frac{c%
}{2}-\frac{m}{6}\right) a+c_{4}\right] \partial _{a}+  \notag \\
&&+\left[ \left( m+11c\right) -\frac{c_{4}}{a}+\frac{2}{5}\left(
8c-2m\right) \right] T\partial _{T}
\end{eqnarray}%
with corresponding Noether integral
\begin{equation*}
I_{2}=\frac{1}{5}\left( 3c-2m\right) t\mathcal{H}-18f_{0}\left[ \left( \frac{%
c}{2}-\frac{m}{6}\right) a^{2}+c_{4}a\right] T^{\frac{1}{2}}\dot{a}\;.
\end{equation*}%
- For $n=\frac{1}{2}$ the Noether point symmetry becomes
\begin{equation*}
X_{3}=c_{1}t\partial _{t}+\left( -2c_{1}+c_{3}a^{\frac{1}{4}}\right)
\partial _{a}+\left( 4c_{1}+c_{2}+\frac{3c_{3}}{2}a^{-\frac{3}{4}}\right)
T\partial _{T}
\end{equation*}%
with Noether integral
\begin{equation*}
I_{3}=c_{1}t\mathcal{H}-6f_{0}\left( -2c_{1}a+c_{3}a^{\frac{3}{4}}\right)
T^{-\frac{1}{2}}\dot{a}.
\end{equation*}
\end{itemize}

We would like to stress that our results are in agreement with those of \cite%
{Wei2012298} but they are richer because we have considered the term $\xi
\partial _{t}$ in the generator which is not done in \cite{Wei2012298}. To
this end it becomes evident that $f\left( T\right) =f_{0}T^{n}$ is the only
form that admits extra Noether point symmetries implying the existence of
exact analytical solutions.

\subsection{Exact cosmological solutions}

In this section we proceed in an attempt to analytically solve the basic
cosmological equations of the $f\left( T\right) =f_{0}T^{n}$ gravity model.
In particular from the Lagrangian (\ref{H.52}) we obtain the main field
equation
\begin{equation}
\ddot{a}+\frac{1}{2a}\dot{a}^{2}+\frac{f^{\prime \prime }}{f^{\prime }}\dot{a%
}\dot{T}-\frac{1}{4}a\frac{f^{\prime }T-f}{f^{\prime }}=0\;.  \label{FF}
\end{equation}%
Also differentiating equation (\ref{H.50}) we find
\begin{equation}
\dot{T}=12\left[ \left( \frac{\dot{a}}{a}\right) ^{3}-\frac{\dot{a}\ddot{a}}{%
a^{2}}\right] \;.  \label{FFa}
\end{equation}%
Finally, inserting $f\left( T\right) =f_{0}T^{n}$, $H=\dot{a}/a$, equation (%
\ref{H.50}) and equation(\ref{FFa}) into equation(\ref{FF}) we derive after
some algebra that
\begin{equation}
\left( 2n-1\right) \left[ \ddot{a}-\frac{\dot{a}^{2}}{2a}\frac{\left(
2n-3\right) }{n}\right] =0
\end{equation}%
a solution of which is
\begin{equation}
a(t)=a_{0}t^{2n/3}\;\;\;\;H(t)=\frac{2n}{3t}  \label{aHE}
\end{equation}%
or
\begin{equation}
H=H_{0}a^{-3/2n}=H_{0}(1+z)^{3/2n}  \label{HE}
\end{equation}%
where $n\in \mathcal{R}_{+}^{\star }-\{\frac{1}{2}\}$, $a(z)=(1+z)^{-1}$ and
$H_{0}$ is the Hubble parameter. We note that the above analytic solution
confirms that of \cite{Wei2012298}.

From equation(\ref{aHE}) it is evident that this cosmological models have no
inflection point (that is the deceleration parameter does not change sign).
Therefore, the main drawback of the $f(T)=f_{0}T^{n}$ gravity model is that
the deceleration parameter preserves sign, and therefore the universe always
accelerates or always decelerates depending on the value of $n$. Indeed, if
we consider $n=1$ (TEGR) then the above solution boils down to the Einstein
de Sitter model as it should. On the other hand, the accelerated expansion
of the universe ($q<0$) is recovered for $n>\frac{3}{2}$. The latter means
that even if we admit $n>\frac{3}{2}$ as a mere phenomenological
possibility, we would be also admitting that the universe has been
accelerating forever, which is of course difficult to accept.

\subsection{Cosmological analogue to other models}

In this section (assuming flatness) we present the cosmological equivalence
at the background level between the current $f(T)$ gravity with $f(R)$
modified gravity and dark energy, through a specific reconstruction of the $%
f(R)$ and vacuum energy density namely, $f(R)=R^{n}$ and $\Lambda
(H)=3\gamma H^{2}$. In particular, in the case of $f(R)=R^{n}$ it has been
found (see chapter \ref{chapter9}) that the corresponding scale factor obeys
equation (\ref{aHE}), where $n\in \mathcal{R}_{+}^{\star }-\{2,\frac{3}{2},%
\frac{7}{8}\}$.\footnote{%
The Lagrangian here is $\mathcal{L}_{R}=6naR^{n-1}\dot{a}%
^{2}+6n(n-1)a^{2}R^{n-2}\dot{a}\dot{R}+(n-1)a^{3}R^{n}$, where $R$ is the
Ricci scalar. For $n=1$ the solution of the Euler-Lagrange equations is the
Einstein de-Sitter model [$a(t)\propto t^{2/3}$] as it should. Note, that
for $n=2$ one can find a de-Sitter solution [$a(t)\propto e^{H_{0}t}$, see
chapter \ref{chapter9}].}

On the other hand, considering a spatially flat FRW metric and in the
context of GR the combination of the Friedmann equations with the total
(matter+vacuum) energy conservation in the matter dominated era, provides
(for more details see \cite{BPS1,BPS2})
\begin{equation}
{\dot{H}}+\frac{3}{2}H^{2}=\frac{\Lambda }{2}\;.  \label{LLAA}
\end{equation}%
Solving equation (\ref{LLAA}) for $\Lambda (H)=3\gamma H^{2}$ (see \cite%
{FreeseET87,CarvalhoET92,ArcuriWaga94}) we end up with
\begin{equation}
H=H_{0}a^{-3(1-\gamma )/2}=H_{0}(1+z)^{3(1-\gamma )/2}\;.  \label{HE11}
\end{equation}%
Now, comparing equations (\ref{HE}), (\ref{HE11}) and connecting the above
coefficients as $n^{-1}=1-\gamma $, we find that the $f(T)=f_{0}T^{n}$ and
the flat $\Lambda (H)=3\gamma H^{2}$ models can be viewed as equivalent
cosmologies as far as the Hubble expansion is concerned, despite the fact
that the time varying vacuum model is inside GR. However, when the $\Lambda
(H)=3\gamma H^{2}$ cosmological model is confronted with the current
observations it provides a poor fit\thinspace \cite{BPS1,BPS2}. Because the
current time varying vacuum model shares exactly the same Hubble parameter
with the $f(T)=f_{0}T^{n}$ gravity model, it follows that the latter is also
under observational pressure when compared against the background
cosmological data. The same observational situation holds also for the $%
f(R)=R^{n}$ modified gravity.

\section{Static spherically symmetric spacetimes}

\label{spergem}

We apply now the results of the general Noether analysis of the previous
subsection, to the specific case of static spherically-symmetric geometry
given by the \ metric (\ref{SSA.1}), that is the vierbein (\ref{SSA}). Armed
with the general expressions provided above, we can deduce the Noether
algebra of the metric (\ref{SSA.1}).

In this metric the Lagrangian (\ref{L1}) and the Hamiltonian (\ref{L2})
become
\begin{equation}
L=2f_{T}N\left( 2ba^{\prime }b^{\prime }+ab^{\prime 2}\right) +M(a,b)\left(
f-f_{T}T\right)
\end{equation}%
\begin{equation}
H=2f_{T}N\left( 2ba^{\prime }b^{\prime }+ab^{\prime 2}\right) -M(a,b)\left(
f-f_{T}T\right) \equiv 0\;  \label{HLf.06}
\end{equation}%
where $M(a,b)$ is given by (\ref{SSA1}). As one can immediately deduce, TEGR
and thus General Relativity is restored when $f(T)=T$, while if $N=1$, $\tau
=r$ and $ab=1$ we fully recover the standard Schwarzschild solution.

Applying Lemma \ref{Lemmaft} in the case of static spherically-symmetric
geometry, we determine all the functional forms of $f(T)$ for which the
above dynamical system admits Noether point symmetries beyond the trivial
one $\partial _{\tau }.$ We summarize the results in Tables \ref{Tft1}, \ref%
{Tft2} and \ref{Tft3}. Furthermore, we can use the obtained Noether
integrals in order to classify the analytic solutions for each case.

\begin{table}[tbp] \centering%
\caption{Noether Symmetries and Noether Integrals for arbitary f(T)}%
\begin{tabular}{lll}
\hline\hline
$\mathbf{N}\left( a,b\right) $ & \textbf{Noether Symmetry} & \textbf{Noether
Integral} \\ \hline
$\frac{1}{a^{3}}N_{1}\left( a^{2}b\right) $ & $X_{1}=-\frac{a}{2b^{3}}%
\partial _{a}+\frac{1}{b^{2}}\partial _{b}$ & $I_{1}=\frac{N_{1}\left(
a^{2}b\right) }{2a^{3}b^{2}}\left( 2ba^{\prime }+ab^{\prime }\right) f_{T}$
\\
$N_{2}\left( b\sqrt{a}\right) $ & $X_{2}=-2a\partial _{a}+b\partial _{b}$ & $%
I_{2}=N_{2}\left( b\sqrt{a}\right) \left( b^{2}a^{\prime }-abb^{\prime
}\right) f_{T}$ \\
$aN_{3}\left( b\right) $ & $X_{3}=\frac{1}{ab}\partial _{a}~$ & $%
I_{3}=N_{3}\left( b\right) b^{\prime }f_{T}$ \\ \hline\hline
\end{tabular}%
\label{Tft1}%
\end{table}%

In the case of $f\left( T\right) =T^{n}$ we have the additional extra
Noether symmetries of Table \ref{Tft1}

\begin{table}[tbp] \centering%
\caption{Extra Noether Symmetries and Noether Integrals for $f(T)=T^n$}%
\begin{tabular}{lll}
\hline\hline
$\mathbf{N}\left( a,b\right) $ & \textbf{Noether Symmetry\footnote{%
Where $\bar{\psi}_{5-7}$ are the conformal factors of the generators of the
symmetry vectors $\eta ^{i}$, i.e. $\psi =\frac{1}{\dim \gamma _{ij}}\eta
_{;i}^{i}$.}} & \textbf{Noether Integral} \\ \hline
arbitrary & $X_{4}=2\bar{\psi}_{0}\tau +\frac{2\bar{\psi}_{0}\left(
C-1\right) }{2C+1}a\partial _{a}+\frac{2\bar{\psi}_{0}-2\bar{\psi}_{4}}{C}%
T\partial _{T}$ & $I_{4}=2\psi _{0}n\frac{C-1}{1+2C}abN\left( a,b\right)
T^{n-1}b^{\prime }~$ \\
arbitrary & $X_{5}=-2a\partial _{a}+b\partial _{b}-\frac{2\bar{\psi}_{5}}{C}%
T\partial _{T}$ & $I_{5}=nN\left( a,b\right) T^{n-1}\left( b^{2}a^{\prime
}-abb^{\prime }\right) $ \\
arbitrary & $X_{6}=-\frac{a}{2}b^{-\frac{3\left( 1+2C\right) }{4C}}\partial
_{a}+b^{-\frac{3+2C}{4C}}\partial _{b}-\frac{2\bar{\psi}_{6}}{C}T\partial
_{T}$ & $I_{6}=\frac{n}{2}N\left( a,b\right) T^{n-1}\left( 2b^{\frac{2C-3}{4C%
}}a^{\prime }+ab^{-\frac{3+2C}{4C}}b^{\prime }\right) $ \\
arbitrary & $X_{7}=a^{-\frac{1}{2C}}b^{-\frac{1+2C}{4C}}\partial _{a}-\frac{2%
\bar{\psi}_{7}}{C}T\partial _{T}~$ & $I_{7}=N\left( a,b\right) na^{-\frac{1}{%
2C}}b^{-\frac{1+2C}{4C}}T^{n-1}b^{\prime }$ \\ \hline\hline
\end{tabular}%
\label{Tft2}%
\end{table}%

\begin{table}[tbp] \centering%
\caption{Extra Noether Symmetries and Noether Integrals for
$f(T)=T^{\frac{1}{2}}$}%
\begin{tabular}{lll}
\hline\hline
$\mathbf{N}\left( a,b\right) $ & \textbf{Noether Symmetry} & \textbf{Noether
Integral} \\ \hline
arbitrary & $\bar{X}_{4}=2\bar{\psi}_{0}\tau +\frac{3\bar{\psi}_{0}}{2}a\ln
\left( a^{2}b\right) \partial _{a}+\frac{2\bar{\psi}_{0}-2\bar{\psi}%
_{4}^{\prime }}{C}T\partial _{T}$ & $\bar{I}_{4}=\frac{3}{2}\psi _{0}N\left(
a,B\right) T^{-\frac{1}{2}}ab\ln \left( a^{2}b\right) b^{\prime }$ \\
arbitrary & $\bar{X}_{5}=b\partial _{b}-\frac{2\bar{\psi}_{5}}{C}T\partial
_{T}$ & $\bar{I}_{5}=\frac{1}{2}N\left( a,b\right) T^{-\frac{1}{2}}\left(
b^{2}a^{\prime }+abb^{\prime }\right) $ \\
arbitrary & $\bar{X}_{6}=-a\ln \left( ab\right) \partial _{a}+b\ln b\partial
_{b}-\frac{2\bar{\psi}_{6}}{C}T\partial _{T}$ & $\bar{I}_{2}=\frac{1}{2}%
N\left( a,B\right) T^{-\frac{1}{2}}b\left( b\ln b~a^{\prime }-a\ln
a~b^{\prime }\right) $ \\
arbitrary & $\bar{X}_{7}=a\partial _{a}-\frac{2\bar{\psi}_{7}}{C}T\partial
_{T}~$ & $\bar{I}_{3}=\frac{1}{2}N\left( a,b\right) T^{-\frac{1}{2}%
}ab~b^{\prime }$ \\ \hline\hline
\end{tabular}%
\label{Tft3}%
\end{table}%


\subsection{Exact Solutions}

\label{Analysol}

Using the Noether symmetries and the corresponding integral of motions
obtained in the previous section, we can extract all the static
spherically-symmetric solutions of $f(T)$ gravity. Without loss of
generality, we choose the conformal factor $N(a,b)$ as $N(a,b)=ab^{2}$ [or
equivalently\footnote{%
Since the space is empty, the field equations are conformally invariant,
therefore the resutls are similar for an arbitrary function $N(a,b)~$(see
chapter \ref{chapter8})} $M(a,b)=1$]. In order to simplify the current
dynamical problem, we consider the coordinate transformation
\begin{equation}
b=\left( 3y\right) ^{\frac{1}{3}}~~\;\;a=\sqrt{\frac{2x}{\left( 3y\right) ^{%
\frac{1}{3}}}}\;.  \label{L4S001}
\end{equation}%
Substituting the above variables into the field equations (\ref{Lf.04}), (%
\ref{Lf.06}), (\ref{HLf.06}) we immediately obtain
\begin{eqnarray}
&&x^{\prime \prime }+\frac{f_{TT}}{f_{T}}x^{\prime }T^{\prime }=0
\label{L4S03} \\
&&y^{\prime \prime }+\frac{f_{TT}}{f_{T}}y^{\prime }T^{\prime }=0
\label{L4S04} \\
&&H=4f_{T}x^{\prime }y^{\prime }-\left( f-Tf_{T}\right)  \label{L4S02}
\end{eqnarray}%
while the torsion scalar is given by
\begin{equation}
T=4x^{\prime }y^{\prime }\;.  \label{L4S01}
\end{equation}%
Finally, the generalized Lagrangian (\ref{L1}) acquires the simple form
\begin{equation}
L=4f_{T}x^{\prime }y^{\prime }+\left( f-Tf_{T}\right) \;.
\end{equation}%
Since the analysis of the previous subsection revealed two classes of
Noether symmetries, namely for arbitrary $f(T)$, and $f(T)=T^{n}$, in the
following subsections we investigate them separately.

\subsubsection{Arbitrary $f(T)$}

In the case where $f(T)$ is arbitrary, a special solution of the system (\ref%
{L4S03})-(\ref{L4S01}) is
\begin{eqnarray}
x\left( \tau \right) &=&c_{1}\tau +c_{2}  \label{L4S5} \\
y\left( \tau \right) &=&c_{3}\tau +c_{4}  \label{L4S6}
\end{eqnarray}%
and the Hamiltonian constraint ($H=0$) reads
\begin{equation}
4c_{1}c_{3}\frac{df}{dT}|_{T=4c_{1}c_{3}}-f+T\frac{df}{dT}|_{T=4c_{1}c_{3}}=0
\label{yLL}
\end{equation}%
where $T=4c_{1}c_{3}$, and $c_{1-4}$ are integration constants. Utilizing (%
\ref{L4S001}), (\ref{L4S5}) and (\ref{L4S6}), we get
\begin{eqnarray}
&&b\left( \tau \right) =3^{\frac{1}{3}}\left( c_{3}\tau +c_{4}\right) ^{%
\frac{1}{3}}  \notag \\
&&a\left( \tau \right) =\frac{\sqrt{6}}{3^{\frac{2}{3}}}\left( c_{1}\tau
+c_{4}\right) ^{\frac{1}{2}}\left( c_{3}\tau +c_{4}\right) ^{\frac{1}{6}}\;.
\end{eqnarray}

For convenience, we can change variables from $b\left( \tau \right) $ to $r$
according to the transformation $b\left( \tau \right) =r$, where $r$ denotes
the radial variable. Inserting this into the above equations, we conclude
that the spacetime (\ref{SSA.1}) in the coordinates $(t,r,\theta ,\phi )$
can be written as
\begin{equation}
ds^{2}=-A\left( r\right) dt^{2}+\frac{1}{c_{3}^{2}}\frac{1}{A\left( r\right)
}dr^{2}+r^{2}\left( d\theta ^{2}+\sin ^{2}\theta d\phi ^{2}\right) ,
\label{yLLa}
\end{equation}%
with
\begin{equation}
A\left( r\right) =\frac{2c_{1}}{3c_{3}}r^{2}-\frac{2c_{\mu }}{c_{3}r}%
=\lambda _{A}(1-\frac{r_{\star }}{r})R(r),  \label{ALLa}
\end{equation}%
and
\begin{equation}
R(r)=\left( \frac{r}{r_{\star }}\right) ^{2}+\frac{r}{r_{\star }}+1.
\label{RLLa}
\end{equation}%
In these expressions we have $c_{\mu }=c_{1}c_{4}-c_{2}c_{3}$, $\lambda
_{A}=\left( \frac{8c_{1}c_{\mu }^{2}}{3c_{3}^{3}}\right) ^{1/3}$ and $%
r_{\star }=(\frac{3c_{\mu }}{c_{1}})^{1/3}=(\frac{3c_{3}\lambda _{A}}{2c_{1}}%
)^{1/2}$ is a characteristic radius with the restriction $c_{\mu }c_{1}>0$.

We observe that if we select the constant $c_{3}\equiv 1$ then we retain the
Schwarzschild-like metric. On the other hand, the function $R(r)$ can be
viewed as a distortion factor which quantifies the smooth deviation from the
pure Schwarzschild solution. Thus, the $f(T)$ gravity on small spherical
scales ($r\rightarrow r_{\star }^{+}$) tends to create a Schwarzschild
solution. In particular, the $f(T)$ spherical solution admits singularity
only at $r=r_{\star }$, which is also the case with the usual Schwarzschild
solution.

Within this framework, for commoving observers, $u^{i}u_{i}=A^{2}(r)$, it is
easy to show that the Einstein's tensor becomes
\begin{equation*}
G_{j}^{i}=\mathrm{diag}\left( 2c_{1}c_{3}-\frac{1}{r^{2}},2c_{1}c_{3}-\frac{1%
}{r^{2}},2c_{1}c_{3},2c_{1}c_{3}\right) .
\end{equation*}%
Therefore, from the 1+3 decomposition of $G_{ij}$we define the fluid
physical parameters
\begin{eqnarray}
&&\rho _{T}=\frac{1}{\left( u^{i}u_{i}\right) }G_{ij}u^{i}u^{j}=2c_{1}c_{3}-%
\frac{1}{r^{2}}  \label{TE.1} \\
&&p_{T}=\frac{1}{3}h_{ij}G_{ij}=2c_{1}c_{3}-\frac{1}{3r^{2}}  \label{TE.2} \\
&&q^{i}=h^{ij}G_{jk}u^{k}=0  \label{TE.3} \\
\pi _{\theta }^{\theta } &=&\pi _{\phi }^{\phi }=-\frac{\pi _{r}^{r}}{2}=%
\frac{1}{3r^{2}}
\end{eqnarray}%
where
\begin{equation}
\pi _{ij}=(h_{i}^{r}h_{j}^{s}-\frac{1}{3}h_{ij}h^{rs})G_{rs}  \label{TE.4}
\end{equation}%
is the anisotropic stress tensor and $h_{ij}$ is the $u^{i}$ projection
tensor defined by
\begin{equation}
h_{ij}=g^{ij}-\frac{1}{\left( u^{i}u_{i}\right) }u^{i}u^{j}.
\end{equation}

Furthermore the fluid is also anisotropic $(\pi _{ij}\neq 0$ ) but not heat
conducting $(q^{i}=0).$

In order to apply the above considerations for specific $f(T)$ forms, we
consider the following viable $f(T)$ models, motivated by cosmology\footnote{%
The $f(T)$ models of Refs.\cite{Ben09,Linn1} are consistent with the
cosmological data.}

\begin{itemize}
\item Exponential $f(T)$ gravity \cite{Linn1}:
\begin{equation*}
f(T)=T+f_{0}e^{-f_{1}T},
\end{equation*}
where $f_{0}$ and $f_{1}$ are the two model parameters which are connected
via (\ref{yLL})
\begin{equation*}
f_{0}=\frac{4c_{1}c_{3}}{8f_{1}c_{1}c_{3}+1}\exp \left(
4f_{1}c_{1}c_{3}\right) .
\end{equation*}

\item A sum of two different power law $f(T)$ gravity:
\begin{equation*}
f(T)=T^{m}+f_{0}T^{n}
\end{equation*}%
where from (\ref{yLL}) we have
\begin{equation*}
f_{0}=\frac{1-2m}{2n-1}\left( 4c_{1}c_{3}\right) ^{m-n}\;.
\end{equation*}%
Note that in the case of $m=1$ we recover the $f(T)$ model by Bengochea \&
Ferraro \cite{Ben09}.
\end{itemize}

\subsubsection{$f(T)=T^{n}$}

In the $f(T)=T^{n}$ case, the field equations (\ref{Lf.04}), (\ref{Lf.06}), (%
\ref{HLf.06}) and the torsion scalar (\ref{L4S01}) give rise to the
following dynamical system:
\begin{eqnarray}
&&T=4x^{\prime }y^{\prime },  \label{L4S1} \\
&&4nT^{n-1}x^{\prime }y^{\prime }-\left( 1-n\right) T^{n}=0,  \label{L4S2} \\
&&x^{\prime \prime }+\left( n-1\right) x^{\prime }T^{-1}T^{\prime }=0,
\label{L4S3} \\
&&y^{\prime \prime }+\left( n-1\right) y^{\prime }T^{-1}T^{\prime }=0\;.
\label{L4S4}
\end{eqnarray}%
It is easy to show that combining equation (\ref{L4S1}) with the Hamiltonian
(\ref{L4S2}), we can impose constraints on the value of $n$, namely $n=1/2$.
Under this condition, solving the system of equations (\ref{L4S3}) and (\ref%
{L4S4}) we arrive at the solutions
\begin{eqnarray}
&&x(\tau )=\frac{\sigma (\tau )^{3}}{3}+c_{\sigma } \\
&&y(\tau )=\frac{\sigma (\tau )^{3}}{3}  \label{L4S4b}
\end{eqnarray}%
where $c_{\sigma }$ is the integration constant. Now using (\ref{L4S001}) we
derive $a$,$b$ as
\begin{eqnarray}
&&b(\tau )=\sigma (\tau )  \label{L4S44} \\
&&a(\tau )=\sqrt{\frac{2\left[ \sigma ^{3}(\tau )+3c_{\sigma }\right] }{%
3\sigma (\tau )}}.
\end{eqnarray}%
Using the coordinate transformation $\sigma (\tau )=r$, which implies $\tau
=F(r)$ [with $F(\sigma (\tau ))=\tau $], and using simultaneously (\ref%
{L4S44}), the spherical metric (\ref{SSA.1}) can be written as
\begin{equation}
ds^{2}=-A(r)dt^{2}+B(r)dr^{2}+r^{2}\left( d\theta +\sin ^{2}\theta d\phi
^{2}\right) \;  \label{MM44}
\end{equation}%
where
\begin{equation}
A(r)=\frac{2}{3}r^{2}+\frac{2c}{r}  \label{AA44}
\end{equation}%
and
\begin{equation}
B(r)=\frac{F_{,r}^{2}}{A(r)r^{4}}.  \label{BB44}
\end{equation}

Furthermore, considering the commoving observers $\left( u^{i}u_{i}\right) =-%
\frac{2\left( r^{3}+3c_{\sigma }\right) }{3r}$, we can write the Einstein
tensor components as
\begin{eqnarray*}
G_{t}^{t} &=&-\frac{r}{3F_{r}^{3}}\left[ 4rF_{,rr}\left( r^{3}+3c_{\sigma
}\right) -2F_{,r}\left( 7r^{3}+12c_{\sigma }\right) +\frac{3}{r^{3}}%
F_{,r}^{3}\right] \\
G_{r}^{r} &=&\frac{2r^{4}}{F_{,r}^{2}}-\frac{1}{r^{2}} \\
G_{\theta }^{\theta } &=&G_{\phi }^{\phi }=-\frac{1}{3}\frac{r}{F_{,r}^{3}}%
\left[ rF_{,rr}\left( 4r^{3}+3c_{\sigma }\right) -F_{,r}\left(
14r^{3}+6c_{\sigma }\right) \right]
\end{eqnarray*}%
where $F_{,r}=dF/dr$ and $F_{,rr}=d^{2}F/dr^{2}$.

Similarly, based on the equalities (\ref{TE.1})-(\ref{TE.4}), we compute the
following fluid parameters%
\begin{equation}
\rho _{T}=\frac{4r^{2}F_{,rr}}{F_{,r}^{3}}\left( \frac{1}{3}r^{3}+c_{\sigma
}\right) -\frac{2r}{F_{,r}^{2}}\left( \frac{7}{3}r^{3}-4c_{\sigma }\right) +%
\frac{1}{r^{2}}  \label{TEE.1}
\end{equation}%
\begin{equation}
p_{T}=-\frac{2}{3}\frac{r^{2}F_{,rr}}{F_{,r}^{3}}\left( \frac{4}{3}%
r^{3}+c_{\sigma }\right) +\frac{2}{3}\frac{r}{F_{,r}^{2}}\left( \frac{17}{3}%
r^{3}+2c_{\sigma }\right) -\frac{1}{3r^{2}}  \label{TEE.2}
\end{equation}%
\begin{equation*}
\pi _{,r}^{r}=\frac{2}{3}\frac{r^{3}F_{,rr}}{F_{,r}^{2}}\left( \frac{4}{3}%
r^{3}+c_{\sigma }\right) -\frac{2}{3}\frac{1}{F_{,r}^{2}}\left( \frac{8}{3}%
r^{3}+2c_{\sigma }\right) -\frac{2}{3r^{2}}
\end{equation*}%
\begin{equation*}
\pi _{\theta }^{\theta }=\pi _{\phi }^{\phi }=-\frac{1}{2}\pi _{r}^{r}\;
\end{equation*}%
\begin{equation*}
q^{i}=0
\end{equation*}

\section{Conclusion}

In this chapter we studied the Noether symmetries of $f\left( T\right) $
gravity. \ We proved that for some diagonal frames the Lagrangian of the
field equations admits Noether symmetries for arbitrary $f\left( T\right) $
function. However, in the case of power law $f\left( T\right) $, i.e. $%
f\left( T\right) =T^{n}$ it is possible the Lagrangian to admits extra
Noether symmetries. We applied this results in order to classify the Noether
symmetries of the field equations in a spatially flat FRW spacetime and in a
static spherical symmetric spacetime. For each background spacetime we found
analytical solutions of the field equations.

\chapter{Discussion\label{Discussion}}

\section{Discussion}

In this thesis we study the Lie point symmetries and the Noether point
symmetries of second order differential equations usinng a geometric
approach and we apply the results to systems which are relevant to
relativistic physics.

In particular, we have studied the point symmetries of the equations of
motion of dynamical systems in a Riemannian space with Lagrangian%
\begin{equation}
L\left( x^{i},\dot{x}^{k}\right) =\frac{1}{2}g_{ij}\dot{x}^{i}\dot{x}%
^{j}-V\left( x^{k}\right)  \label{Dis.01}
\end{equation}%
where $g_{ij}=g_{ij}\left( x^{k}\right) $ is the metric of the space \ and
we proved that the Lie point symmetries of the Euler-Lagrange
equations,~i.e. $E^{i}\left( L\right) =0,$ of Lagrangian (\ref{Dis.01}) are
generated from the elements of the special Projective algebra of the
Riemannian manifold with metric $g_{ij}$ whereas the Noether point
symmetries are generated from the homothetic algebra of the space with
metric $g_{ij}$. Therefore we have transfer the problem of determination of
the Lie/Noether symmetries of differential equations to the determination of
the collineations of the underlying manifold; hence, we are able to use the
plethora of existing results of differential geometry.

We have applied this geometric approach in many directions. In particular,
we have classified the Lie and the Noether symmetries of the geodesic
Lagrangian for some important spacetimes, such as the FRW spacetime, the G%
\"{o}del space, the Taub space and the 1+3 decomposable spacetimes. Moreover
we proved that for Einstein spaces the point symmetries of the geodesic
equations are generated from the elements of the Killing algebra of the
metric.

Furthermore we have determined all the two and the three dimensional
Newtonian systems which admit Lie and Noether point symmetries. We note
that, due to the geometric derivation and the tabular presentation, the
results can be extended easily to higher dimensional flat spaces. We applied
these results in the study of the symmetries of the H\`{e}non - Heiles
potential and of the Kepler-Ermakov potential in a two dimensional space.
Moreover, we determined the potentials which admit Noether symmetries in a
two dimensional sphere $S^{2}$.

We proved that a dynamical system admits as Lie symmetries the $sl\left(
2,R\right) $ Lie algebra if and only if the underlying manifold admits a
gradient Homothetic vector. The Newtonian system which is invariant under
the Lie group $sl\left( 2,R\right) $ is the well known Kepler-Ermakov
system. Therefore, the requirement for a dynamical system with Lagrangian of
the form of (\ref{Dis.01}) to admit as Lie and Noether symmetries the
generators of the $sl\left( 2,R\right) $ Lie algebra leads us to the
generalization of the Newotian Kepler-Ermakov system in a Riemannian
manifold; that is, we found that the general autonomous Kepler-Ermakov
system follows from the Lagrangian%
\begin{equation}
L\left( u,\dot{u},y^{A},\dot{y}^{A}\right) =\frac{1}{2}\left( \dot{u}^{2}+%
\dot{u}^{2}h_{AB}\dot{y}^{A}\dot{y}^{B}\right) +\frac{\mu ^{2}}{2}u^{2}+%
\frac{1}{u^{2}}V\left( y^{C}\right)  \label{Dis.02}
\end{equation}%
In additionally, we studied the Liouville integrability of the three
dimensional Newtonian Kepler-Ermakov via Noether point symmetries and we
investigated the application of Lagrangian (\ref{Dis.02}) in dynamical
systems emerging from alternative theories of gravity. In particular we
showed that the field equations in a Bianchi I spacetime for an exponential
scalar field and for $f\left( R\right) $ gravity when $f\left( R\right)
=\left( R-2\Lambda \right) ^{\frac{7}{8}}$ follow from Lagrangians of the
form of (\ref{Dis.02}) and for these models we proved that the field
equations are Liouville integrable.

Concerning the second order partial differential equations we considered
equations of the generic form%
\begin{equation}
A^{ij}\left( x^{k}\right) u_{ij}-B^{i}(x^{k},u)u_{i}-f(x^{k},u)=0
\label{Dis.03}
\end{equation}%
and we proved a theorem which relates the Lie pont symmetries of equation (%
\ref{Dis.03}) with the elements of the Conformal Killing vectors of the
second order tensor $A_{ij}\left( x^{k}\right) $ (considered to be a
metric). We have applied this result in order to study the Lie point
symmetries of the Heat equation and the Poisson equation. It has been shown
that the Lie symmetries of the Heat equation follow from the Killing and the
homothetic algebras of $A_{ij}\left( x^{k}\right) ,$ whereas the Lie
symmetries of the Poisson equation follow from the Killing, the homothetic
and the conformal algebras of $A_{ij}\left( x^{k}\right) .$ In each case we
have determoned the form of the Lie symmetry vectors.

Furthermore, we have determined the Lie symmetries of the Schr\"{o}dinger
equation
\begin{equation}
\Delta u-u_{,t}=V\left( x^{k}\right) u  \label{Dis.04}
\end{equation}%
and the Klein Gordon equation%
\begin{equation}
\Delta u=V\left( x^{k}\right) u  \label{Dis.05}
\end{equation}%
in a general Riemannian space. It has been shown that these symmetries are
related to the Noether symmetries of the classical Lagrangian for which the
metric $g_{ij}$ is the kinematic metric. More precisely, for the Schr\"{o}%
dinger equation (\ref{Dis.04}) it has been shown that if a KV or a HV of the
metric $g_{ij}$ produces a Lie symmetry for the Schr\"{o}dinger equation,
then it produces a Noether symmetry for the Classical Lagrangian in the
space with metric $g_{ij}$ and potential $V(x^{k})$. For the Klein Gordon
equation (\ref{Dis.05}) the situation is different; the Lie symmetries of
the Klein Gordon are generated by the elements of the conformal group of the
metric $g_{ij}.$ The KVs and the HV of this group produce a Noether symmetry
of the classical Lagrangian with a constant gauge function. However the
proper CKVs produce a Noether symmetry for the conformal Lagrangian if there
exists a conformal factor $N\left( x^{k}\right) $ such that the CKV becomes
a KV/HV of $g_{ij}$.

{We have applied these results to three cases of practical interest: the
motion in a central potential, the classification of all potentials in two
and three dimensional Euclidian spaces for which the Schr\"{o}dinger
equation and the Klein Gordon equation admit a Lie symmetry and finally we
have considered the Lie symmetries of the Klein Gordon equation in the
static, spherically symmetric empty spacetime. In the last case, we have
demonstrated the role of the Lie symmetries and that of the conformal
Lagrangians in the determination of the closed form solution of Einstein
equations. Furthermore, we investigated the Lie point symmetries of the null
Hamilton Jacobi equation and we proved that if a CKV generates a point
symmetry for the Klein Gordon equation, then it also generates a point
symmetry for the null Hamilton Jacobi equation.}

We also studied the problem of Type II\ hidden symmetries of second order
partial differential equations in $n~$ dimensional Riemannian spaces from a
geometric of view. We have considered the reduction of the Laplace and of
the homogeneous heat equation and the consequent possibility of existence of
Type II hidden symmetries in some general classes of spaces which admit some
kind of symmetry; hence, they admit nontrivial Lie symmetries.

The Type II\ hidden symmetries of Laplace equation are directly related to
the transition of the CKVs from the space where the original equation is
defined to the space where the reduced equation resides. In this sense, we
related the Lie symmetries of PDEs with the basic collineations of the
metric i.e. the CKVs. \

Concerning the Type II hidden symmetries of the homogeneous heat equation we
considered the problem in the spaces which admit a gradient KV or a gradient
HV and finally spacetime which admits a HV\ which acts simply and
transitively. For the reduction of the homogeneous heat equation and the
existence of Type II hidden symmetries, we found the following general
geometric results: (a) If we reduce the homogeneous heat equation via the
symmetries which are generated by a gradient KV $\left( S^{,i}\right) $~the
reduced equation is a heat equation in the nondecomposable space. In this
case we have the Type II hidden symmetry $\partial _{t}-\frac{1}{2t}%
w\partial _{w}$ provided if we reduce the heat equation with the symmetry~$%
tS^{,i}-\frac{1}{2}Su\partial _{u}$. (b) If we reduce the homogeneous heat
equation via the symmetries which are generated by a gradient HV the reduced
equation is Laplace equation for an appropriate metric. In this case the
Type II hidden symmetries are generated from the proper CKVs and (c) in
Petrov type III spacetime, the reduction of the homogeneous heat equation
via the symmetry generated from the nongradient HV gives PDE that inherit
the Lie symmetries, hence no Type II hidden symmetries are admitted.

Finally, we applied the point symmetries and especially the Noether point
symmetries in modified theories of gravity in order to probe the nature of
dark energy. We used the Noether symmetries as a geometric criterion or
"selection rule" in order to select the scalar field potential in
scalar-tensor theories, and the functions $f\left( R\right) $ and $f\left(
T\right) $ in the corresponding alternative theories of gravity.

In the context of scalar-tensor cosmology we have found that to every
non-minimally coupled scalar field, we can associate a unique minimally
coupled scalar field in a conformally related space with an appropriate
potential. This result can be used in order to study the dynamical
properties of the various cosmological models, since the field equations of
a non-minimally coupled scalar field are the same, at the conformal level,
with the field equations of the minimally coupled scalar field. Furthermore,
we have identified the Noether point symmetries and the analytic solutions
of the equations of motion in the context of a minimally coupled and a non
minimally coupled scalar field in a FRW spacetime and we have classified the
Noether point symmetries of the field equations in Bianchi class A models
with a minimally coupled scalar field. We find that there is a rather large
class of hyperbolic and exponential potentials which admit extra (beyond the
$\partial _{t}$) Noether symmetries which lead to integral of motions. For
these potentials we used the corresponding Noether integrals in order to
solve analytically the field equations and find the functional form of the
scalar factor.

Concerning the $f\left( R\right) $ models we applied the Noether point
symmetries with the aim to utilize the existence of non-trivial Noether
symmetries as a selection criterion that can distinguish the $f(R)$ models
on a more fundamental level. We proved that in a spatially flat FRW
background the $f\left( R\right) $ theories which admit Noether point
symmetries are the$~R^{n},~R^{\frac{7}{8}},~\left( R-2\Lambda \right) ^{%
\frac{7}{8}},~R^{\frac{3}{2}}~$and$~\left( R-2\Lambda \right) ^{\frac{3}{2}}$%
. The last two functional forms of the $f\left( R\right) $ function admit
Noether point symmetries also in the case of a non spatially flat FRW
background. It is interesting to note that the $\frac{7}{8}$ models are
equivalent with the Newtonian Kepler-Ermakov system whereas the $\frac{3}{2}$
models are equivalent with the anisotropic hyperbolic oscillator. For these
functional forms we use the Noether integrals in order to find exact
solutions of the modified field equations.

Furthermore, in $f\left( T\right) $ gravity we have proved a Lemma that for
diagonal frames the only functional form of $f\left( T\right) $ which admits
extra Noether symmetries is the $f\left( T\right) =T^{n}$. \ We applied this
functional form in a spatially flat FRW background and we determined the
analytic solution of the field equations for each case. Finally, we studied
the field equations for the $f\left( T\right) =T^{n}$ model in a static
spherically symmetric spacetime and we determined a family of analytic
solutions.

The geometric approach is a new method for studying the symmetries of
differential equations and has shown that gives directly results by using
only the results (usually existing) of differential geometry without the need
to use of a computer library in order to determine the Lie or the Noether
point symmetries. This approach implies a better understanding of the nature
of symmetries and of the conservation laws and can be used in order to find
analogues of classical Newtonian systems in relativistic physics. It is of
interest that this method would be extendented in other classes of
differential equations and in other transformations which are not necessary
point transformations. Concerning the applications in Cosmology, it is of
interest the classification of modified theories of gravity with geometric
selection rules. In this thesis we studied some of the basic modified
theories of gravity; however there are other more such theories which could
be studied further either by means of point symmetries or by some new
geometric criteria. 

\bibliographystyle{science}
\bibliography{Refer}

\end{document}